%% file: ms.tex
\documentclass[conference,a4paper,10pt,times,romanappendices]{IEEEtran}
\pdfoutput=1

\newif\ifdraft
 \draftfalse

\newif\iffull
\fulltrue

\usepackage{lipsum}
\usepackage[normalem]{ulem}
\usepackage{mathtools}
\usepackage{mathpartir}
\usepackage{stmaryrd}
\usepackage{amsthm}
\usepackage{float}
\usepackage{wrapfig}
\usepackage{framed}
\usepackage{xcolor}
\usepackage{textcomp}
\usepackage{xspace}
\usepackage{amsmath,amssymb,amsfonts}
\usepackage[T1]{fontenc}
\usepackage[utf8]{inputenc}
\usepackage{cite}
\usepackage{algorithmic}
\usepackage{graphicx}
\usepackage{placeins}

\def\BibTeX{{\rm B\kern-.05em{\sc i\kern-.025em b}\kern-.08em
    T\kern-.1667em\lower.7ex\hbox{E}\kern-.125emX}}

\usepackage{url}
\usepackage{tikz}
\usepackage{proof}
\usepackage{hyperref}
\usepackage{enumerate}

\input{macros.sty}

\usetikzlibrary{backgrounds}
\usetikzlibrary{intersections}
\usetikzlibrary{scopes}
\usetikzlibrary{calc}
\usetikzlibrary{decorations.pathreplacing}
\usetikzlibrary{decorations.pathmorphing,shapes}
\usetikzlibrary{backgrounds}
\usetikzlibrary{shapes.misc}

\pgfdeclarelayer{bg1}    
\pgfdeclarelayer{bg0}    
\pgfsetlayers{background,bg1,bg0,main} 

\newtheorem{assumption}{Assumption}
\newtheorem*{theorem*}{Theorem}
\newtheorem*{lemma*}{Lemma}
\newtheorem*{proposition*}{Proposition}

\newtheorem{definition}{Definition}
\newtheorem{proposition}{Proposition}
\newtheorem{theorem}{Theorem}
\newtheorem{lemma}{Lemma}

\theoremstyle{remark}
\newtheorem{remark}{Remark}
\newtheorem{example}{Example}

\begin{document}
\iffull
\title{The 5G-AKA Authentication Protocol Privacy
  \\(Technical Report)}
\fi

\author{\IEEEauthorblockN{Adrien Koutsos}
  \IEEEauthorblockA{\textit{LSV, CNRS, ENS Paris-Saclay,
      Université Paris-Saclay} \\
    Cachan, France\\
    adrien.koutsos@lsv.fr}}

\maketitle
\ifdraft
\thispagestyle{plain}
\pagestyle{plain}
\fi

\begin{abstract}
  We study the 5G-AKA authentication protocol described in the 5G mobile communication standards. This version of AKA tries to achieve a better privacy than the 3G and 4G versions through the use of asymmetric randomized encryption. Nonetheless, we show that except for the IMSI-catcher attack, all known attacks against 5G-AKA privacy still apply.
  
  Next, we modify the 5G-AKA protocol to prevent these attacks, while \ch{satisfying 5G-AKA efficiency constraints as much as possible}{satisfying the cost and efficiency constraints of the 5G-AKA protocol}. We then formally prove that our protocol is $\sigma$-unlinkable. This is a new security notion, which allows for a fine-grained quantification of a protocol privacy. Our security proof is carried out in the Bana-Comon indistinguishability logic. We also prove mutual authentication as a secondary result.
\end{abstract}


\input{introduction}

\input{standards}

\input{attack}

\input{protocol}

\input{ss-unlinkability}

\input{modeling}

\input{axioms-body}

\input{security}

\input{conclusion}

\section*{Acknowledgment}
This research has been partially funded by the French National Research Agency (ANR) under the project TECAP (ANR-17-CE39-0004-01).

\FloatBarrier
\bibliographystyle{IEEEtran}
\bibliography{IEEEabrv,biblio}

\iffull
\onecolumn
\appendices

\tableofcontents
\newpage

\input{axioms}

\newpage
\input{symbolic}

\newpage
\input{acceptance}

\newpage
\input{unlinkability}
\fi

\end{document}

%% file: introduction.tex
\section{Introduction}

Mobile communication technologies are widely used for voice, text and Internet access. These technologies allow a subscriber's device, typically a mobile phone, to connect wirelessly to an antenna, and from there to its service provider. The two most recent generations of mobile communication standards, the $\threeg$ and $\fourg$ standards, have been designed by the $\threegpp$ consortium. The \emph{fifth generation} ($\fiveg$) of mobile communication standards is being finalized, and drafts are now available~\cite{ts33501}. These standards describe protocols that aim at providing security guarantees to the subscribers and service providers. One of the most important such protocol is the \emph{Authentication and Key Agreement} ($\aka$) protocol, which allows a subscriber and its service provider to establish a shared secret key in an authenticated fashion. There are different variants of the $\aka$ protocol, one for each generation.

In the $\threeg$ and $\fouraka$ protocols, the subscriber and its service provider share a long term secret key. The subscriber stores this key in a cryptographic chip, the \emph{Universal Subscriber Identity Module} ($\tusim$), which also performs all the cryptographic computations. Because of the $\tusim$ limited computational power, the protocols only use symmetric key cryptography without any pseudo-random number generation on the subscriber side. Therefore
the subscriber does not use a random challenge to prevent replay attacks, but instead relies on a sequence number $\sqn$. Since the sequence number has to be tracked by the subscriber and its service provider, the $\aka$ protocols are stateful.

Because a user could be easily tracked through its mobile phone, it is important that the $\aka$ protocols provide privacy guarantees. The $\threeg$ and $\fouraka$ protocols try to do that using temporary identities. While this provides some privacy against a \emph{passive adversary}, this is not enough against an \emph{active adversary}. Indeed, these protocols allow an antenna to ask for a user permanent identity when it does not know its temporary identity (this naturally happens in roaming situations). This mechanism is abused by $\imsi$-catchers~\cite{imsi-catcher} to collect the permanent identities of all mobile devices in range.

The $\imsi$-catcher attack is not the only known attack against the privacy of the $\aka$ protocols. In~\cite{BH-BH17}, the authors show how an attacker can obtain the least significant bits of a subscriber's sequence number, which allows the attacker to monitor the user's activity. The authors of~\cite{DBLP:conf/ccs/ArapinisMRRGRB12} describe a linkability attack against the $\threeaka$ protocol. This attack is similar to the attack on the French e-passport~\cite{DBLP:conf/csfw/ArapinisCRR10}, and relies on the fact that $\threeaka$ protocol uses different error messages if the authentication failed because of a bad $\macsym$ or because a de-synchronization occurred.

The $\fiveg$ standards include changes to the $\aka$ protocol to improve its privacy guarantees. In $\fiveaka$, a user never sends its permanent identity in plain-text. Instead, it encrypts it using a \emph{randomized asymmetric encryption} with its service provider public key. While this prevents the $\imsi$-catcher attack, this is not sufficient to get unlinkability. Indeed, the attacks from~\cite{DBLP:conf/ccs/ArapinisMRRGRB12,BH-BH17} against the $\threeg$ and $\fouraka$ protocols still apply. Moreover, the authors of~\cite{DBLP:journals/popets/FouqueOR16} proposed an attack against a variant of the $\aka$ protocol introduced in~\cite{DBLP:conf/ccs/ArapinisMRRGRB12}, which uses the fact that an encrypted identity can be replayed. It turns out that their attack also applies to $\fiveaka$.

\paragraph{Objectives}
Our goal is to improve the privacy of $\fiveaka$ while satisfying its design and efficiency constraints. In particular, our protocol should be as efficient as the $\fiveaka$ protocol, have \ch{a similar}{the same} communication complexity and rely on the same cryptographic primitives. Moreover, we want formal guarantees on the privacy provided by our~protocol.

\paragraph{Formal Methods}
Formal methods are the best way to get a strong confidence in the security provided by a protocol. They have been successfully applied to prove the security of crucial protocols, such as Signal~\cite{DBLP:conf/eurosp/KobeissiBB17} and TLS~\cite{DBLP:journals/ieeesp/BhargavanFK16, DBLP:conf/ccs/CremersHHSM17}. There exist several approaches to formally prove a protocol security.

In the \emph{symbolic} or \emph{Dolev-Yao} (DY) model, protocols are modeled as members of a formal process algebra~\cite{DBLP:journals/jacm/AbadiBF18}. In this model, the attacker controls the network: he reads all messages and he can forge new messages using capabilities granted to him through a fixed set of rules. While security in this model can be automated (e.g.~\cite{proverif, DBLP:conf/sp/ChevalKR18, Meier:2013:TPS:2526861.2526920, DBLP:conf/ccs/CortierGLM17}), it offers limited guarantees: we only prove security against an attacker that has the designated capabilities.

The \emph{computational model} is more realistic. The attacker also controls the network, but is not limited by a fixed set of rules. Instead, the attacker is any Probabilistic Polynomial-time Turing Machine (PPTM for short). Security proofs in this model are typically sequences of game transformations~\cite{DBLP:journals/iacr/Shoup04} between a game stating the protocol security and cryptographic hypotheses. This model offers strong security guarantees, but proof automation is much harder. For instance, \textsc{CryptoVerif}~\cite{DBLP:journals/tdsc/Blanchet08} cannot prove the security of stateful cryptographic protocols (such as the $\aka$ protocols).


There is a third model, the \emph{Bana-Comon} (BC) model~\cite{DBLP:conf/post/BanaC12,Bana:2014:CCS:2660267.2660276}. In this model, messages are terms and the security property is a first-order formula. Instead of granting the attacker capabilities through rules, as in the symbolic approach, we state what the adversary \emph{cannot} do. This model has several advantages. First, since security in the BC model entails computational security, it offers strong security guarantees. Then, there is no ambiguity: the adversary can do anything which is not explicitly forbidden. Finally, this approach is well-suited to model stateful protocols.

\paragraph{Related Work}
There are several formal analysis of $\aka$ protocols in the symbolic models. In~\cite{DBLP:conf/sp/ChevalKR18}, the authors use the \textsc{Deepsec} tool to prove unlinkability of the protocol for three sessions.  In~\cite{DBLP:conf/ccs/ArapinisMRRGRB12} and \cite{DBLP:conf/ccs/BroekVR15}, the authors use \textsc{Proverif} to prove unlinkability of $\aka$ variants for, respectively, three sessions and an unbounded number of sessions. In these three works, the authors abstracted away several key features of the protocol. Because \textsc{Deepsec} and \textsc{Proverif} do not support the xor operator, they replaced it with a symmetric encryption. Moreover, sequence numbers are modeled by nonces in~\cite{DBLP:conf/ccs/ArapinisMRRGRB12} and~\cite{DBLP:conf/sp/ChevalKR18}. While \cite{DBLP:conf/ccs/BroekVR15} models the sequence number update, they assume it is always incremented by one, which is incorrect. Finally, none of these works modeled the re-synchronization or the temporary identity mechanisms. Because \ch{of these inaccuracies in their models}{they have inaccurate models}, they all miss attacks.

In~\cite{Basin2018AFA}, the authors use the \textsc{Tamarin} prover to analyse multiple properties of $\fiveaka$. For each property, they either find a proof, or exhibit an attack. To our knowledge, this is the most precise symbolic analysis of an $\aka$ protocol. For example, they correctly model the xor and the re-synchronization mechanisms, and they represent sequence numbers as integers (which makes their model stateful). Still, they decided not to include the temporary identity mechanism. Using this model, they successfully rediscover the linkability attack from~\cite{DBLP:conf/ccs/ArapinisMRRGRB12}.

We are aware of two analysis of $\aka$ protocols in the computational model. In~\cite{DBLP:journals/popets/FouqueOR16}, the authors present
\ch
{a significantly modified version of $\aka$,}
{a variant of $\aka$,}
called $\privaka$, and claim it is unlinkable. However, we discovered a linkability attack against the protocol, which falsifies the authors claim. In~\cite{DBLP:journals/ijisec/LeeSWW14}, the authors study the $\fouraka$ protocol \emph{without its first message}. They show that this reduced protocol satisfies a form of anonymity (which is weaker than unlinkability). Because they consider a weak privacy property for a reduced protocol, \ch{they fail to}{their model does not} capture the linkability attacks from the literature.

To summarize, there is currently no computational security proof of a complete version of an $\aka$ protocol.

\paragraph{Contributions}
Our contributions are:
\begin{itemize}
\item We study the privacy of the $\fiveaka$ protocol described in the $\threegpp$ draft~\cite{ts33501}. Thanks to the introduction of asymmetric encryption, the $\fiveg$ version of $\aka$ is not vulnerable to the $\imsi$-catcher attack. However, we show that the linkability attacks from~\cite{DBLP:journals/popets/FouqueOR16, DBLP:conf/ccs/ArapinisMRRGRB12, BH-BH17} against older versions of $\aka$ still apply to $\fiveaka$.
\item We present a new linkability attack against $\privaka$,
  \ch
  {a significantly modified version of}
  {a variant of}
  the $\aka$ protocol introduced and claimed unlinkable in~\cite{DBLP:journals/popets/FouqueOR16}. This attack exploits the fact that%
  \ch
  {, in $\privaka$,}
  {}
  a message can be delayed to yield a state update later in the execution of the protocol, where it can be detected.
\item We propose the $\faka$ protocol, which is a modified version of $\fiveaka$ with better privacy guarantees and satisfying the same design and efficiency constraints.
\item We introduce a new privacy property, called $\sigma$-unlinkability, inspired from \cite{DBLP:conf/esorics/HermansPVP11} and Vaudenay's Privacy~\cite{DBLP:conf/asiacrypt/Vaudenay07}. Our property is parametric and allows us to have a fine-grained quantification of a protocol privacy.
\item We formally prove that $\faka$ satisfies the $\sigma$-unlinkability property in the computational model. Our proof is carried out in the BC model, and holds for any number of agents and sessions that are not related to the security parameter. We also show that $\faka$ provides mutual authentication.
\end{itemize}

\paragraph{Outline}
In Section~\ref{section:standards} and~\ref{section:attack-unlink} we describe the $\fiveaka$ protocol and the known linkability attacks against it. We present the $\faka$ protocol in Section~\ref{section:faka-description}, and we define the $\sigma$-unlinkability property in Section~\ref{section:unlink-body}. Finally, we show how we model the $\faka$ protocol using the BC logic in Section~\ref{section:modeling-body}, and we state and sketch the proofs of the mutual authentication and $\sigma$-unlinkability of $\faka$ in Section~\ref{section:security-proofs-body}.
\iffull
The full proofs are in Appendix.
\fi


%% file: standards.tex
\section{The \fiveaka Protocol}
\label{section:standards}

We present the $\fiveaka$ protocol described in the $\threegpp$ standards~\cite{ts33501}. This is a three-party authentication protocol between:
\begin{itemize}
\item The \emph{User Equipment} ($\tue$). This is the subscriber's physical device  using the mobile communication network (e.g. a mobile phone). Each $\tue$ contains a cryptographic chip, the \emph{Universal Subscriber Identity Module} ($\tusim$), which stores the user confidential material (such as secret keys).
\item The \emph{Home Network} ($\thn$), which is the subscriber's service provider. It maintains a database with the necessary data to authenticate its subscribers.
\item The \emph{Serving Network} ($\tsn$). It controls the base station (the antenna) the $\tue$ is communicating with through a wireless channel.
\end{itemize}
If the $\thn$ has a base station nearby the $\tue$, then the $\thn$ and the $\tsn$ are the same entity. But this is not always the case (e.g. in roaming situations). When no base station from the user's $\thn$ are in range, the $\tue$ uses another network's base station.

The $\tue$ and its corresponding $\thn$ share some confidential key material and the \emph{Subscription Permanent Identifier} ($\supi$), which uniquely identifies the $\tue$. The $\tsn$ does not have access to the secret key material. It follows that all cryptographic computations are performed by the $\thn$, and sent to the $\tsn$ through a secure channel. The $\tsn$ also forwards all the information it gets from the $\tue$ to the $\thn$. But the $\tue$ permanent identity is not kept hidden from the $\tsn$: after a successful authentication, the $\thn$ sends the $\supi$ to the $\tsn$. This is not technically needed, but is done for legal reasons. Indeed, the $\tsn$ needs to know whom it is serving to be able to answer to \emph{Lawful Interception} requests.

Therefore, privacy requires to trust both the $\thn$ and the $\tsn$. Since, in addition, they communicate through a secure channel, we decided to model them as a single entity and we include the $\tsn$ inside the $\thn$. A description of the protocol with three distinct parties can be found in~\cite{Basin2018AFA}.


\subsection{Description of the Protocol}
The $\fiveg$ standard proposes two authentication protocols, $\eapakap$ and $\fiveaka$. Since their differences are not relevant for privacy, we only describe the $\fiveaka$ protocol.

\paragraph{Cryptographic Primitives}
As in the $\threeg$ and $\fourg$ variants, the $\fiveaka$ protocol uses several keyed cryptographic one-way functions: $\owsym^1$, $\owsym^2$, $\owsym^5$, $\owsym^{1,\mstar}$ and $\owsym^{5,\mstar}$. These functions are used both for integrity and confidentiality, and take as input a long term secret key $\key$ (which is different for each subscriber).

A major novelty in $\fiveaka$ is the introduction of an asymmetric randomized encryption $\enc{\cdot}{\pk}{\enonce}$. Here $\pk$ is the public key, and $\enonce$ is the encryption randomness. Previous versions of $\aka$ did not use asymmetric encryption because the $\tusim$, which is a cryptographic micro-processor, \ch{had no randomness generation capabilities.}{was not powerful enough.} The asymmetric encryption is used to conceal the identity of the $\tue$, by sending $\enc{\supi}{\pk}{\enonce}$ instead of transmitting the $\supi$ in clear (as in $\threeg$ and $\fouraka$). 


\paragraph{Temporary Identities}
After a successful run of the protocol, the $\thn$ may issue a temporary identity, a \emph{Globally Unique Temporary Identity} ($\guti$), to the $\tue$. Each $\guti$ can be used in \emph{at most one session} to replace the encrypted identity $\enc{\supi}{\pk}{\enonce}$. It is renewed after each use. Using a $\guti$ allows to avoid one asymmetric encryption. This saves a pseudo-random number generation and the expensive computation of an asymmetric encryption.

\paragraph{Sequence Numbers}
The $\fiveaka$ protocol prevents replay attacks using a sequence number $\sqn$ instead of a random challenge. This sequence number is included in the messages, incremented after each successful run of the protocol, and must be tracked and updated by the $\tue$ and the $\thn$. As it may get de-synchronized (e.g. because a message is lost), there are two versions of it: the $\tue$ sequence number $\sqn_\ue$, and the $\thn$ sequence number~$\sqn_\hn$.

\paragraph{State}
The $\tue$ and $\thn$ share the $\tue$ identity $\supi$, a long-term symmetric secret key $\key$, \ch{a sequence number $\sqn_\ue$}{} and the $\thn$ public key $\pk_\hn$. The $\tue$ also stores in $\guti$ the value of the last temporary identity assigned to it (if there is one). Finally, the $\thn$ stores the secret key $\sk_\hn$ corresponding to $\pk_\hn$, \ch{its version $\sqn_\hn$ of every $\tue$'s sequence number}{} and a mapping between the $\guti$s and the $\supi$s.

\begin{figure}[tb]
  \begin{center}
    \begin{tikzpicture}[>=latex]
      \tikzset{every node/.style={font=\footnotesize}};
      \tikzset{en/.style={minimum height=0.2cm,minimum width=0.15cm,fill=black}};
      \tikzset{mb/.style={solid,thick,draw=gray!90}};
      \tikzset{mp/.style={inner sep=0,outer sep=0,draw=none}};
      \tikzset{sb/.style={double}};

      \path (0,0)
      node[above,yshift=0.1cm,anchor=south] (ue) {$\tue$}
      node[draw,name path=suen,double=gray!80,double distance=0.035cm,below,
      anchor=north west,xshift=-0.5cm,align=left] (sue)
      {$\supi,\guti,\key,\pk_\hn,\sqn_\ue$};

      \path[name path=aline] (ue) -- ++(0,-10.6);

      \path[draw,very thick,name intersections={of=suen and aline}]
      (intersection-2) node[inner sep=0] (tue) {}
      -- ++(0,-10.3) node[en]{};

      \path (7,0)
      node[above,yshift=0.1cm,anchor=south] (hn) {$$\thn$$}
      node[double=gray!80,draw,name path=shnn,double distance=0.035cm,below,
      anchor=north east,xshift=0.5cm,align=left] (shn)
      {$\supi,\guti,\key,\sk_\hn,\sqn_\hn$};

      \path[name path=bline] (hn) -- ++(0,-10.7);

      \path[draw,very thick,name intersections={of=shnn and bline}]
      (intersection-2) node[inner sep=0] (thn) {}
      -- ++(0,-10.3) node[en]{};

      \path (tue) -- ++(0,-0.8)
      node[mp] (ntue0) {}
      [draw,->] -- ++(7,0)
      node[midway,above]{$\guti$ \textbf{or} $\enc{\supi}{\pk_\hn}{\enonce}$};

      \path (ntue0) -- ++(0,-0.3)
      node[sb,fill=white,name path=nptue0,anchor=north west,
      xshift=-0.5cm,draw,align=left] {
        \textbf{if $\guti$ was used:}
        $\;\guti \la \unset$};

      \path[name intersections={of=aline and nptue0}]
      (intersection-2) node[mp] (nptue1){};

      \path let \p1=(nptue1) in
      let \p2=(hn) in
      (\x2,\y1) node[mp] (pthn1) {};

      \path (pthn1) -- ++(0,-0.8) node[mp] (t){};
      \path[draw,->] (t) -- ++(-7,0)
      node[midway,above]{$\triplet
        {\nonce}
        {\sqn_\hn \oplus \xow{\nonce}{\key}{5}}
        {\xow{\spair{\sqn_\hn}{\nonce}}{\key}{1}}$}
      node (ue1) {};

      \path (ue1) -- ++(0,-0.3)
      node[mp] (inputn) {}
      node[sb,fill=white,name path=pue1,anchor=north west,
      xshift=-0.5cm,draw,align=left] (c) {
        \textbf{\underline{Input $\sfx$:}}\\
        $\nonce_\sfr,\, \sqn_\sfr \la
        \pi_1(\sfx),\, \pi_2(\sfx) \oplus \xow{\nonce_\sfr}{\key}{5}$\\
        $\sfb_\lmacsym \la
        \xow{\spair{\sqn_\sfr}{\nonce_\sfr}}{\key}{1} =
        \pi_3(\sfx)$\\
        $\sfb_\sqn \la \range{\sqn_\ue}{\sqn_\sfr}$};

      \path (inputn) -- ++(7,0)
      node[sb,fill=white,anchor=north east,
      xshift=0.5cm,draw,align=left] (c) {
        $\sqn_\hn \la \sqn_\hn + 1$};

      \path[name intersections={of=aline and pue1}]
      (intersection-2) -- ++(0,-0.95)
      node[sb,fill=white,name path=pue2,anchor=north west,
      xshift=-0.5cm,draw] (ue2p) {
        $\sqn_\ue \la \sqn_\sfr$};

      \path[name intersections={of=pue2 and aline}]
      (intersection-2) -- ++(0,-0.3) node (ue3) {}
      [draw,->] -- ++(7,0)
      node[midway,above]{$
        {\xow{\nonce_\sfr}{\key}{2}}
        $};

      \path (ue3) -- ++(-0.85,1.5) node[mp] (ue3l){};
      \path (ue3) -- ++(7,0) -- ++(0.7,-0.2) node[mp] (ue3r){};
      \path (ue3l) node[outer sep=0,draw,fill=white,anchor=north west] (a)
      {\textbf{$\sfb_\lmacsym \wedge \sfb_\sqn$}};
      \path (a.north east) -- ++ (0,-0.15) node[mp] (x){}
      (a.south west) -- ++ (0.15,0) node[mp] (y){};
      \draw[mb] (x.base) -| (ue3r.base) -| (y.base);

      \path (ue3) -- ++(0,-1.1) node (ue2) {}
      [draw,->] -- ++(7,0)
      node[midway,above] {``Auth-Failure''};

      \path (ue2) -- ++(-0.85,0.7) node[mp] (ue2l){};
      \path (ue2) -- ++(7,0) -- ++(0.7,-0.2) node[mp] (ue2r){};
      \path (ue2l) node[outer sep=0,draw,fill=white,anchor=north west] (a)
      {\textbf{$\neg \sfb_\lmacsym$}};
      \path (a.north east) -- ++ (0,-0.15) node[mp] (x){}
      (a.south west) -- ++ (0.15,0) node[mp] (y){};
      \draw[mb] (x.base) -| (ue2r.base) -| (y.base);

      \path (ue2) -- ++(0,-1.2) node[mp] (ue40) {}
      [draw,->] -- ++(7,0) node[mp] (ue4) {}
      node[midway,above]{$
        \lpair
        {\sqn_\ue \oplus \sxow{\nonce_\sfr}{\key}{5}}
        {\sxow{\spair{\sqn_\ue}{\nonce_\sfr}}{\key}{1}}$};
      \path (ue4) -- ++(0,-0.3)
      node[sb,fill=white,name path=ue4path,anchor=north east,align=left,
      xshift=0.5cm,draw] (ue2p) {
        \textbf{\underline{Input $\sfy$:}}\\
        $\sqn^\mstar_\sfr \la \pi_1(\sfy) \oplus \sxow{\nonce}{\key}{5}$\\
        $\textsf{if }
        \sxow{\spair{\sqn^\mstar_\sfr}{\nonce}}{\key}{1} =
        \pi_2(\sfy)$
        \textsf{ then}
        $\;\sqn_\hn \la \sqn^\mstar_\sfr + 1$};

      \path[name intersections={of=ue4path and bline}] (intersection-2)
      node[mp](ue4i){};
      \path (ue40) -- ++(-0.85,0.8) node[mp] (ue4l){};
      \path (ue4i) -- ++(0.7,-0.2) node[mp] (ue4r){};
      \path (ue4l) node[outer sep=0,draw,fill=white,anchor=north west] (a)
      {\textbf{$\sfb_\lmacsym \wedge \neg \sfb_\sqn$}};
      \path (a.north east) -- ++ (0,-0.15) node[mp] (x){}
      (a.south west) -- ++ (0.15,0) node[mp] (y){};
      \draw[mb] (x.base) -| (ue4r.base) -| (y.base);
    \end{tikzpicture}
  \end{center}

  \ch{\textbf{Conventions:} $\la$ is used for assignments, and has a lower priority than the equality comparison operator $=$.}{}
  \caption{The $\fiveaka$ Protocol}
  \label{fig:fiveaka}
\end{figure}

\paragraph{Authentication Protocol}
The $\fiveaka$ protocol is represented in Fig.~\ref{fig:fiveaka}. We now describe an honest execution of the protocol. The $\tue$ initiates the protocol by identifying itself to the $\thn$, which it can do in two different ways:
\begin{itemize}
\item It can send a temporary identity $\guti$, if one was assigned to it. After sending the $\guti$, the $\tue$ sets it to $\unset$ to ensure that it will not be used more than once. Otherwise, it would allow an adversary to link sessions together.
\item It can send its concealed permanent identity $\enc{\supi}{\pk_\hn}{\enonce}$, using the $\thn$ public key $\pk_\hn$ and a fresh randomness $\enonce$.
\end{itemize}
Upon reception of an identifying message, the $\thn$ retrieves the permanent identity $\supi$: if it received a temporary identity $\guti$, this is done through a database look-up; and if a concealed permanent identity was used, it uses $\sk_\hn$ to decrypt it. It can then recover $\sqn_\hn$ and the key $\key$ associated to the identity $\supi$ from its memory. The $\thn$ then generates a fresh nonce $\nonce$. It masks the sequence number $\sqn_\hn$ by xoring it with $\xow{\nonce}{\key}{5}$, and $\lmacsym$ the message by computing $\xow{\spair{\sqn_\hn}{\nonce}}{\key}{1}$ (we use $\langle\dots\rangle$ for tuples). It then sends the message $\striplet{\nonce}{\sqn_\hn \oplus \xow{\nonce}{\key}{5}}{\,\xow{\spair{\sqn_\hn}{\nonce}}{\key}{1}}$.

When receiving this message, the $\tue$ computes $\xow{\nonce}{\key}{5}$. With it, it unmasks $\sqn_\hn$ and checks the authenticity of the message by re-computing $\xow{\spair{\sqn_\hn}{\nonce}}{\key}{1}$ and verifying that it is equal to the third component of the message. It also checks whether $\sqn_\hn$ and $\sqn_\ue$ are in range\footnote{The specification is loose here: it only requires that $\sqn_\ue < \sqn_\hn \le \sqn_\ue + C$, where $C$ is some constant chosen by the $\thn$.}. If both checks succeed, the $\tue$ sets $\sqn_\ue$ to $\sqn_\hn$, which prevents this message from being accepted again. It then sends
$\xow{\nonce}{\key}{2}$ to prove to $\thn$ the knowledge of $\key$.
If the authenticity check fails, an ``Auth-Failure'' message is sent. Finally, if the authenticity check succeeds but the range check fails, $\tue$ starts the re-synchronization sub-protocol, which we describe below.

\paragraph{Re-synchronization}
The re-synchronization protocol allows the $\thn$ to obtain the current value of $\sqn_\ue$. First, the $\tue$ masks $\sqn_\ue$ by xoring it with $\sxow{\nonce}{\key}{5}$, $\lmacsym$ the message using $\sxow{\spair{\sqn_\ue}{\nonce}}{\key}{1}$ and sends the pair $\spair{\sqn_\ue \oplus \sxow{\nonce}{\key}{5}}{\sxow{\spair{\sqn_\ue}{\nonce}}{\key}{1}}$. When receiving this message, the $\thn$ unmasks $\sqn_\ue$ and checks the $\lmacsym$. If the authentication test is successful, $\thn$ sets the value of $\sqn_\hn$ to $\sqn_\ue + 1$. This ensures that $\thn$ first message in the next session of the protocol is in the correct range.

\paragraph{$\guti$ Assignment}
\ch{There is a final component of the protocol which is not described in Fig.~\ref{fig:fiveaka} (as it is not used in the privacy attacks we present later).}{} After a successful run of the protocol, the $\thn$ generates a new temporary identity $\guti$ and links it to the $\tue$'s permanent identity in its database. Then, it sends the \ch{masked}{} fresh $\guti$ to the $\tue$.




%% file: attack.tex
\section{Unlinkability Attacks Against $\fiveaka$}
\label{section:attack-unlink}

We present in this section several attacks against $\aka$ that appeared in the literature. All these attacks but one (the $\imsi$-catcher attack) carry over to $\fiveaka$. Moreover, several fixes of the $\threeg$ and $\fourg$ versions of $\aka$ have been proposed. We discuss the two most relevant fixes, the first by Arapinis et al.~\cite{DBLP:conf/ccs/ArapinisMRRGRB12}, and the second by Fouque et al.~\cite{DBLP:journals/popets/FouqueOR16}.

None of these fixes are satisfactory. The modified $\aka$ protocol given in~\cite{DBLP:conf/ccs/ArapinisMRRGRB12} has been shown flawed in~\cite{DBLP:journals/popets/FouqueOR16}. The authors of~\cite{DBLP:journals/popets/FouqueOR16} then propose
\ch
{their own protocol, called $\privaka$,}
{their own version, called $\privaka$,}
and claim it is unlinkable (they only provide a proof sketch). While analyzing the $\privaka$ protocol, we discovered an attack allowing to permanently de-synchronize the $\tue$ and the $\thn$. Since a de-synchronized $\tue$ can be easily tracked (after being de-synchronized, the $\tue$ rejects all further messages), our attack is also an unlinkability attack. This is in direct contradiction with the security property claimed in~\cite{DBLP:journals/popets/FouqueOR16}. This is a novel attack that never appeared in the literature.

\subsection{$\imsi$-Catcher Attack}
\begin{figure}[tb]
  \begin{center}
    \begin{tikzpicture}[>=latex]
      \tikzset{every node/.style={font=\footnotesize}};
      \tikzset{en/.style={minimum height=0.2cm,minimum width=0.15cm,fill=black}};
      \tikzset{mb/.style={solid,thick,draw=gray!90}};
      \tikzset{mp/.style={inner sep=0,outer sep=0,draw=none}};
      \tikzset{sb/.style={draw,fill=white,double,align=left}};

      \path (0,0)
      node[above,yshift=0.1cm,anchor=south] (ue) {$\tue$};

      \path[name path=aline] (ue) -- ++(0,-2.4) node[mp] (xend) {};

      \path (7.7,0)
      node[above,yshift=0.1cm,anchor=south] (hn) {\text{Attacker}};

      \path[name path=bline] (hn) -- ++(0,-2.5);

      \path[name path=xendp] (xend) -- ++(8,0);
      \path[name intersections={of=xendp and bline}]
      (intersection-1) node[mp] (xendr){};
      \path[draw,very thick]
      (hn) -- (xendr) node[en]{};

      \path[draw,very thick]
      (ue) -- (xend) node[en]{};

      \path (ue) -- ++(0,-0.4)
      [draw,->] -- ++(7.7,0)
      node[pos=0.5,above]{$\tmsi$ or $\imsi$}
      node[mp] (hn1) {};

      \path (hn1) -- ++(0,-1)
      [draw,->] -- ++(-7.7,0)
      node[pos=0.5,above]{``Permanent-ID-Request''}
      node[mp] (ue2) {};

      \path (hn1) -- ++(0.5,-0.2) node[mp] (hn1l){};
      \path (hn1l) -- ++(-8,0) -- ++(-0.5,-1.5) node[mp] (hn1r){};
      \path (hn1l) node[outer sep=0,draw,fill=white,anchor=north east] (a)
      {\textbf{If $\tmsi$ received}};
      \path (a.north west) -- ++ (0,-0.15) node[mp] (x){}
      (a.south east) -- ++ (-0.15,0) node[mp] (y){};
      \draw[mb] (x.base) -| (hn1r.base) -| (y.base);

      \path (ue2) -- ++(0,-0.5)
      [draw,->] -- ++(7.7,0)
      node[pos=0.5,above]{$\imsi$}
      node[mp] (ue2) {};

    \end{tikzpicture}
  \end{center}

  \caption{An $\imsi$-Catcher Attack}
  \label{fig:imsi-catcher-attack}
\end{figure}
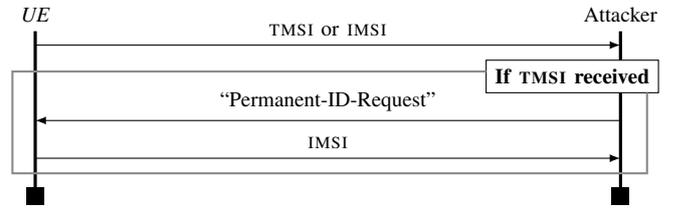
All the older versions of $\aka$ ($\fourg$ and earlier) are vulnerable to the  $\imsi$-catcher attack~\cite{imsi-catcher}. This attack simply relies on the fact that, in these versions of $\aka$, the permanent identity (called the \emph{International Mobile Subscriber Identity} or $\imsi$ in the $\fourg$ specifications) is not encrypted but sent in plain-text. Moreover, even if a temporary identity is used (a \emph{Temporary Mobile Subscriber Identity} or $\tmsi$), an attacker can simply send a \text{Permanent-ID-Request} message to obtain the $\tue$'s permanent identity. The attack is depicted in Fig.~\ref{fig:imsi-catcher-attack}.

This necessitates an active attacker with its own base station. At the time, this required specialized hardware, and was believed to be too expensive. This is no longer the case, and can be done for a few hundreds dollars (see~\cite{DBLP:conf/ndss/ShaikSBAN16}).

\subsection{The Failure Message Attack}
\label{subsec:desync-attack}
\begin{figure}[tb]
  \begin{center}
    \begin{tikzpicture}[>=latex]
      \tikzset{every node/.style={font=\footnotesize}};
      \tikzset{en/.style={minimum height=0.2cm,minimum width=0.15cm,fill=black}};
      \tikzset{mb/.style={solid,thick,draw=gray!90}};
      \tikzset{mp/.style={inner sep=0,outer sep=0,draw=none}};
      \tikzset{sb/.style={draw,fill=white,double,align=left}};

      \path (0,0)
      node[above,yshift=0.1cm,anchor=south] (ue) {$\tue_{\imsi_{\textsf{t}}}$};

      \path[name path=aline] (ue) -- ++(0,-1.6);

      \path[draw,very thick] (ue)  -- ++(0,-1.6) node[mp](xend){};

      \path (7.7,0)
      node[above,yshift=0.1cm,anchor=south] (hn) {$\thn$};

      \path[name path=bline] (hn) -- ++(0,-1.6);

      \path[name path=xendp] (xend) -- ++(8,0);
      \path[name intersections={of=xendp and bline}]
      (intersection-1) node[mp] (xendr){};

      \path[draw,very thick] (hn) -- (xendr) node[mp]{};

      \path (hn) -- ++(0,-0.5)
      [draw,->] -- ++(-7.7,0)
      node[midway,above]{$
        t_{\textsf{auth}} \equiv
        \triplet
        {\nonce}
        {\sqn_\hn \oplus \xow{\nonce}{\key}{5}}
        {\xow{\spair{\sqn_\hn}{\nonce}}{\key}{1}}$}
      node[mp] (ue1) {};

      \path (ue1) -- ++(0,-0.8)
      [draw,->] -- ++(7.7,0)
      node[mp] (hn1) {}
      node[midway,above]{$
        {\xow{\nonce}{\key}{2}}
        $};

      \path (hn1) -- ++(-7.7,-0.7)
      node (ue) {$\tue_{\imsi'}$};

      \path[name path=aline] (ue) -- ++(0,-3.1);

      \path[draw,very thick] (ue)  -- ++(0,-3.1) node[en](xend){};

      \path (ue) -- ++((7.7,0) node (hn) {Attacker};

      \path[name path=bline] (hn) -- ++(0,-3.1);

      \path[name path=xendp] (xend) -- ++(8,0);
      \path[name intersections={of=xendp and bline}]
      (intersection-1) node[en] (xendr){};

      \path[draw,very thick] (hn) -- (xendr) node[mp]{};

      \path (hn) -- ++(0,-0.5)
      [draw,->] -- ++(-7.7,0)
      node[midway,above]{$t_{\textsf{auth}}$}
      node[mp] (ue1) {};

      \path (ue1) -- ++(0,-0.9)
      node[mp] (ue2) {}
      [draw,->] -- ++(7.7,0)
      node[mp] (hn1) {}
      node[pos=0.55,above]{``Auth-Failure''};

      \path (ue2) -- ++(-0.5,0.7) node[mp] (ue2l){};
      \path (ue2) -- ++(7.5,0) -- ++(0.5,-0.2) node[mp] (ue2r){};
      \path (ue2l) node[outer sep=0,draw,fill=white,anchor=north west] (a)
      {\textbf{$\textsf{If }\imsi' \ne \imsi_{\textsf{t}}$}};
      \path (a.north east) -- ++ (0,-0.15) node[mp] (x){}
      (a.south west) -- ++ (0.15,0) node[mp] (y){};
      \draw[name path=p1,mb] (x.base) -| (ue2r.base) -| (y.base);

      \path[name intersections={of=p1 and aline}]
      (intersection-1) node[mp] (ue3){};

      \path (ue3) -- ++(0,-1)
      node[mp] (ue4) {}
      [draw,->] -- ++(7.7,0)
      node[mp] (hn2) {}
      node[pos=0.55,above]{$
        \lpair
        {\sqn_\ue \oplus \sxow{\nonce_\sfr}{\key}{5}}
        {\sxow{\spair{\sqn_\ue}{\nonce_\sfr}}{\key}{1}}$};

      \path (ue4) -- ++(-0.5,0.8) node[mp] (ue4l){};
      \path (ue4) -- ++(7.5,0) -- ++(0.5,-0.2) node[mp] (ue4r){};
      \path (ue4l) node[outer sep=0,draw,fill=white,anchor=north west] (a)
      {\textbf{$\textsf{If }\imsi' = \imsi_{\textsf{t}}$}};
      \path (a.north east) -- ++ (0,-0.15) node[mp] (x){}
      (a.south west) -- ++ (0.15,0) node[mp] (y){};
      \draw[mb] (x.base) -| (ue4r.base) -| (y.base);

    \end{tikzpicture}
  \end{center}

  \caption{The Failure Message Attack by~\cite{DBLP:conf/ccs/ArapinisMRRGRB12}}
  \label{fig:desync-attack}
\end{figure}
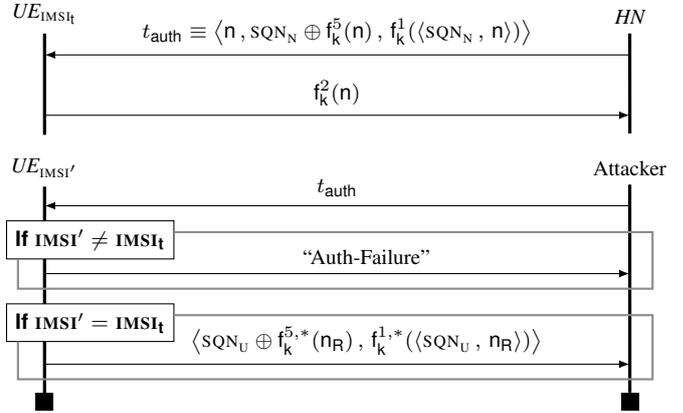

In \cite{DBLP:conf/ccs/ArapinisMRRGRB12}, Arapinis et al. propose to use an asymmetric encryption to protect against the $\imsi$-catcher attack: each $\tue$ carries the public-key of its corresponding $\thn$, and uses it to encrypt its permanent identity. This is basically the solution that was adopted by $\threegpp$ for the $\fiveg$ version of $\aka$. Interestingly, they show that this is not enough to ensure privacy, and give a linkability attack that does not rely on the identification message sent by $\tue$. While their attack is against the $\threeaka$ protocol, it is applicable to the $\fiveaka$ protocol. 

\paragraph{The Attack}
The attack is depicted in Fig.~\ref{fig:desync-attack}, and works in two phases. First, the adversary eavesdrops a successful run of the protocol between the $\thn$ and the target $\tue$ with identity $\imsi_{\textsf{t}}$, and stores the authentication message $t_{\textsf{auth}}$ sent by $\thn$. In a second phase, the attacker $\mathcal{A}$ tries to determine whether a $\tue$ with identity $\imsi'$ is the initial $\tue$ (i.e. whether $\imsi' = \imsi_{\textsf{t}}$). To do this, $\mathcal{A}$ initiates a new session of the protocol and replays the message $t_{\textsf{auth}}$. If $\imsi' \ne \imsi_{\textsf{t}}$, then the $\lmacsym$ test fails, and $\tue_{\imsi'}$ answers ``Auth-Failure''. If $\imsi' = \imsi_{\textsf{t}}$, then the $\lmacsym$ test succeeds but the $\textsf{range}$ test fails, and $\tue_{\imsi'}$ sends a re-synchronization message. 

The adversary can distinguish between the two messages, and therefore knows if it is interacting with the original or a different $\tue$. Moreover, the second phase of the attack can be repeated every time the adversary wants to check for the presence of the tracked user $\imsi_{\textsf{t}}$ in its vicinity.

\paragraph{Proposed Fix}
To protect against the failure message attack, the authors of~\cite{DBLP:conf/ccs/ArapinisMRRGRB12} propose that the $\tue$ encrypts \ch{both}{} error messages using the public key $\pk_\hn$ of the $\thn$\ch{, making them indistinguishable. To the adversary, there is no distinctions between}{}
an authentication and a de-synchronization failure. The fixed $\aka$ protocol, \emph{without the identifying message $\enc{\imsi}{\pk_\hn}{\enonce}$}, was formally checked in the symbolic model using the \textsc{proverif} tool. Because this message was omitted in the model, an attack was missed. We present this attack next. 


\subsection{The Encrypted $\imsi$ Replay Attack}
\label{subsec:encrpt-imsi-replay}
In~\cite{DBLP:journals/popets/FouqueOR16}, Fouque et al. give an attack against the fixed $\aka$ proposed by Arapinis et al. in~\cite{DBLP:conf/ccs/ArapinisMRRGRB12}. Their attack, described in Fig.~\ref{fig:arapinis-fixed-aka-attack}, uses the fact the identifying message $\enc{\imsi_{\textsf{t}}}{\pk_\hn}{\enonce}$ in the proposed $\aka$ protocol by Arapinis et al. can be replayed. 

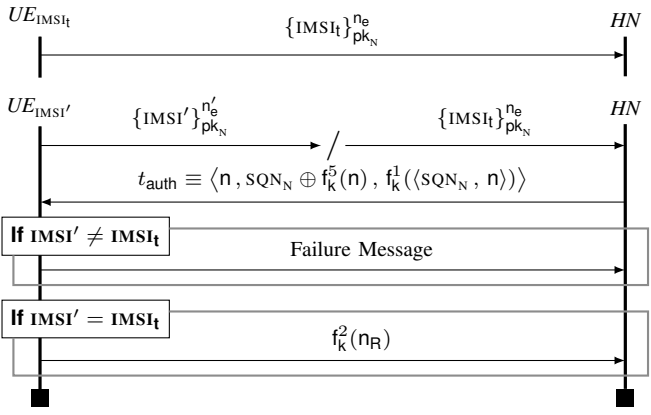
\begin{figure}[tb]
  \begin{center}
    \begin{tikzpicture}[>=latex]
      \tikzset{every node/.style={font=\footnotesize}};
      \tikzset{en/.style={minimum height=0.2cm,minimum width=0.15cm,fill=black}};
      \tikzset{mb/.style={solid,thick,draw=gray!90}};
      \tikzset{mp/.style={inner sep=0,outer sep=0,draw=none}};
      \tikzset{sb/.style={draw,fill=white,double,align=left}};

      \path (0,0)
      node[above,yshift=0.1cm,anchor=south] (ue) {$\tue_{\imsi_{\textsf{t}}}$};

      \path[name path=aline] (ue) -- ++(0,-0.8);

      \path[draw,very thick] (ue)  -- ++(0,-0.8) node[mp](xend){};

      \path (7.7,0)
      node[above,yshift=0.1cm,anchor=south] (hn) {$\thn$};

      \path[name path=bline] (hn) -- ++(0,-0.8);

      \path[name path=xendp] (xend) -- ++(8,0);
      \path[name intersections={of=xendp and bline}]
      (intersection-1) node[mp] (xendr){};

      \path[draw,very thick] (hn) -- (xendr) node[mp]{};

      \path (ue) -- ++(0,-0.5)
      node[mp] (ue1) {}
      [draw,->] -- ++(7.7,0)
      node[midway,above]{$\enc{\imsi_{\textsf{t}}}{\pk_\hn}{\enonce}$};

      \path (ue1) -- ++(0,-0.7)
      node (ue) {$\tue_{\imsi'}$};

      \path (ue) -- ++(7.7,0)
      node (hn) {$\thn$};

      \path[name path=aline] (ue) -- ++(0,-3.85);

      \path[draw,very thick] (ue)  -- ++(0,-3.85) node[en](xend){};

      \path[name path=bline] (hn) -- ++(0,-3.85);

      \path[name path=xendp] (xend) -- ++(8,0);
      \path[name intersections={of=xendp and bline}]
      (intersection-1) node[en] (xendr){};

      \path[draw,very thick] (hn) -- (xendr) node[mp]{};

      \path (ue) -- ++(0,-0.5)
      node[mp] (ue1) {}
      [draw,->] -- ++(3.7,0)
      node[midway,above]{$\enc{\imsi'}{\pk_\hn}{\enonce'}$}
      node[mp] (a){};
      \path[draw=none] (a) -- ++(0.3,0)
      node[midway,transform shape,scale=1.4] (c) {$/$}
      node[mp](b){};
      \path[draw,->] (b) -- ++(3.7,0)
      node[midway,above]{$\enc{\imsi_{\textsf{t}}}{\pk_\hn}{\enonce}$}
      node[mp] (hn1) {};

      \path (hn1) -- ++(0,-0.75)
      [draw,->] -- ++(-7.7,0)
      node[midway,above]{$
        t_{\textsf{auth}} \equiv
        \triplet
        {\nonce}
        {\sqn_\hn \oplus \xow{\nonce}{\key}{5}}
        {\xow{\spair{\sqn_\hn}{\nonce}}{\key}{1}}$}
      node[mp] (ue2) {};

      \path (ue2) -- ++(0,-0.9)
      node[mp] (ue3) {}
      [draw,->] -- ++(7.7,0)
      node[mp] (hn2) {}
      node[pos=0.55,above]{Failure Message};

      \path (ue3) -- ++(-0.5,0.7) node[mp] (ue3l){};
      \path (ue3) -- ++(7.5,0) -- ++(0.5,-0.2) node[mp] (ue3r){};
      \path (ue3l) node[outer sep=0,draw,fill=white,anchor=north west] (a)
      {\textbf{$\textsf{If }\imsi' \ne \imsi_{\textsf{t}}$}};
      \path (a.north east) -- ++ (0,-0.15) node[mp] (x){}
      (a.south west) -- ++ (0.15,0) node[mp] (y){};
      \draw[name path=p1,mb] (x.base) -| (ue3r.base) -| (y.base);

      \path[name intersections={of=p1 and aline}]
      (intersection-1) node[mp] (ue4){};

      \path (ue4) -- ++(0,-1)
      node[mp] (ue5) {}
      [draw,->] -- ++(7.7,0)
      node[mp] (hn3) {}
      node[pos=0.55,above]{$
        {\xow{\nonce_\sfr}{\key}{2}}
        $};

      \path (ue5) -- ++(-0.5,0.8) node[mp] (ue5l){};
      \path (ue5) -- ++(7.5,0) -- ++(0.5,-0.2) node[mp] (ue5r){};
      \path (ue5l) node[outer sep=0,draw,fill=white,anchor=north west] (a)
      {\textbf{$\textsf{If }\imsi' = \imsi_{\textsf{t}}$}};
      \path (a.north east) -- ++ (0,-0.15) node[mp] (x){}
      (a.south west) -- ++ (0.15,0) node[mp] (y){};
      \draw[mb] (x.base) -| (ue5r.base) -| (y.base);

    \end{tikzpicture}
  \end{center}

  \caption{The Encrypted $\imsi$ Replay Attack by~\cite{DBLP:journals/popets/FouqueOR16}}
  \label{fig:arapinis-fixed-aka-attack}
\end{figure}

In a first phase, the attacker $\mathcal{A}$ eavesdrops and stores the identifying message $\enc{\imsi_{\textsf{t}}}{\pk_\hn}{\enonce}$ of an honest session between the user $\tue_{\imsi_{\textsf{t}}}$ it wants to track and the $\thn$. Then, every time $\mathcal{A}$ wants to determine whether some user $\tue_{\imsi'}$ is the tracked user $\tue_{\imsi_{\textsf{t}}}$, it intercepts the identifying message $\enc{\imsi'}{\pk_\hn}{\enonce'}$ sent by $\tue_{\imsi'}$, and replaces it with the stored message $\enc{\imsi_{\textsf{t}}}{\pk_\hn}{\enonce}$. Finally, $\mathcal{A}$ lets the protocol continue without further tampering. We have two possible outcomes:
\begin{itemize}
\item If $\imsi' \ne \imsi_{\textsf{t}}$ then the message $t_{\textsf{auth}}$ sent by $\thn$ is $\lmacsym$-ed using the wrong key, and the $\tue$ rejects the message. Hence the attacker observes a failure message.
\item If $\imsi' = \imsi_{\textsf{t}}$ then $t_{\textsf{auth}}$ is accepted by $\tue_{\imsi'}$, and the attacker observes a success message.
\end{itemize}
Therefore the attacker can deduce whether it is interacting with $\tue_{\imsi_{\textsf{t}}}$ or not, which breaks unlinkability.

\subsection{Attack Against The $\privaka$ Protocol}
\label{subsection:perm-desync-attack}

The authors of~\cite{DBLP:journals/popets/FouqueOR16} then propose
\ch
{the $\privaka$ protocol, which is a significantly modified version of $\aka$.}
{their own version of the $\aka$ protocol, which they call $\privaka$.}
The authors claim that their protocol achieves authentication and client unlinkability. But we discovered a de-synchronization attack: it is possible to permanently de-synchronize the $\tue$ and the $\thn$. Our attack uses the fact that in $\privaka$, the $\thn$ sequence number is incremented only upon reception of the confirmation message from the $\tue$. Therefore, by intercepting the last message from the $\tue$, we can prevent the $\thn$ from incrementing its sequence number. We now describe the attack.

We run a session of the protocol, but we intercept the last message and store it for later use. Note that the $\thn$'s session is not closed. At that point, the $\tue$ and the $\thn$ are de-synchronized by one. We re-synchronize them by running a full session of the protocol. We then re-iterate the steps described above: we run a session of the protocol, prevent the last message from arriving at the $\thn$, and then run a full session of the protocol to re-synchronize the $\thn$ and the $\tue$. Now the $\tue$ and the $\thn$ are synchronized, and we have two stored messages, one for each uncompleted session. We then send the two messages to the corresponding $\thn$ sessions, which accept them and increment the sequence number. In the end, it is incremented by~\emph{two}.

The problem is that the $\tue$ and the $\thn$ cannot recover from a de-synchronization by two. We believe that this was missed by the authors of \cite{DBLP:journals/popets/FouqueOR16}\footnote{``the two sequence numbers may become desynchronized by one step [...]. Further desynchronization is prevented [...]'' (p. 266 \cite{DBLP:journals/popets/FouqueOR16})}. Remark that this attack is also an unlinkability attack. To attack some user $\tue_{\imsi}$'s privacy, we permanently de-synchronize it. Then each time $\tue_{\imsi}$ tries to run the $\privaka$ protocol, it will abort, which allows the adversary to track it.

\begin{remark}
  \ch{Our attack requires that the $\thn$ does not close the first session when we execute the second session. At the end of the attack, before sending the two stored messages, there are two $\thn$ sessions simultaneously opened for the same $\tue$. If the $\thn$ closes any un-finished sessions when starting a new session with the same $\tue$, our attack does not work.}{}

  \ch{But this make another unlinkability attack possible. Indeed, closing a session because of some later session between the $\thn$ and the same $\tue$ reveals a link between the two sessions. We describe the attack. First, we start a session $i$ between a user $\tue_{\agent{A}}$ and the $\thn$, but we intercept and store the last message $t_{\agent{A}}$ from the user. Then, we let the $\thn$ run a full session with some user $\tue_{\agent{X}}$. Finally, we complete the initial session $i$ by sending the stored message $t_{\agent{A}}$ to the $\thn$. Here, we have two cases. If $\agent{X} = \agent{A}$, then the $\thn$ closed the first session when it completed the second. Hence it rejects $t_{\agent{A}}$. If $\agent{X} \ne \agent{A}$, then the first session is still opened, and it accepts~$t_{\agent{A}}$.}{}

  \ch{Closing a session may leak information to the adversary. Protocols which aim at providing unlinkability must explicit when sessions can safely be closed. By default, we assume a session stays open. In a real implementation, a timeout \emph{tied to the session} (and not the user identity) could be used to avoid keeping sessions opened forever.}{}
\end{remark}
\subsection{Sequence Numbers and Unlinkability}
\label{subsection:conj-unlink}
We conjecture that it is not possible to achieve functionality (i.e. honest sessions eventually succeed), authentication and unlinkability at the same time when using a sequence number based protocol%
\ch{\ with no random number generation capabilities in the $\tue$ side.}
{.}
We briefly explain our intuition.

In any sequence number based protocol, the agents may become de-synchronized because they cannot know if their last message has been received.
Furthermore, the attacker can cause de-synchronization by blocking messages. The problem is that we have contradictory requirements. On the one hand, to ensure authentication, an agent must reject a  replayed message. On the other hand, in order to guarantee unlinkability, an honest agent has to behave the same way when receiving a message from a synchronized agent or from a de-synchronized agent. Since functionality requires that a message from a synchronized agent is accepted, it follows that a message from a de-synchronized agent must be accepted. Intuitively, it seems to us that an honest agent cannot distinguish between a protocol message which is being replayed and an honest protocol message from a de-synchronized agent. It follows that a replayed message should be both rejected and accepted, which is a contradiction.

This is only a conjecture. We do not have a formal statement, or a proof. Actually, it is unclear how to formally define the set of protocols that rely on sequence numbers to achieve authentication. Note however that all requirements can be satisfied simultaneously if we allow \emph{both} parties to generate random challenges in each session (in $\aka$, only $\thn$ uses a random challenge). Examples of challenge based unlinkable authentication protocols can be found in~\cite{DBLP:conf/sp/HirschiBD16}.


%% file: protocol.tex
\section{The $\faka$ Protocol}
\label{section:faka-description}

We now describe our principal contribution, which is the design of the $\faka$ protocol. This is a fixed version of the $\fiveaka$ protocol offering some form of privacy against an \emph{active} attacker. First, we explicit the efficiency and design constraints. We then describe the $\faka$ protocol, and explain how we designed this protocol from $\fiveaka$ by fixing all the previously described attacks. As we mentioned before, we think unlinkability cannot be achieved under these constraints. Nonetheless, our protocol satisfies some weaker notion of unlinkability that we call $\sigma$-unlinkability. This is a new security property that we introduce. Finally, we will show a subtle attack, and explain how we fine-tuned $\faka$ to prevent it.

\subsection{Efficiency and Design Constraints}
We now explicit the protocol design constraints. These constraints are necessary for an efficient, in-expensive to implement and backward compatible protocol. Observe that, in a mobile setting, it is very important to avoid expensive computations as they quickly drain the $\tue$'s battery.

\paragraph{Communication Complexity}
In $\fiveaka$, authentication is achieved using only three messages: two messages are sent by the $\tue$, and one by the $\thn$. \ch{We want our protocol to have a similar communication complexity. While we did not manage to use only three messages in all scenarios, our protocol achieves authentication in less than four messages.}
{Therefore, we want our protocol to have at most three messages per session.}

\paragraph{Cryptographic primitives}
We recall that all cryptographic primitives are computed in the $\tusim$, where they are implemented in hardware. It follows that using more primitives in the $\tue$ would make the $\tusim$ more voluminous and expensive. Hence we restrict $\faka$ to the cryptographic primitives used in $\fiveaka$: we use only symmetric keyed one-way functions and asymmetric encryption. Notice that the $\tusim$ cannot do asymmetric \emph{decryption}. As in $\fiveaka$, we use some in-expensive functions, e.g. xor, pairs, by-one increments and boolean tests. We believe that relying on the same cryptographic primitives helps ensuring backward compatibility, and would simplify the protocol deployment.

\paragraph{Random Number Generation}
In $\fiveaka$, the $\tue$ generates at most one nonce per session, which is used to randomize the asymmetric encryption. Moreover, if the $\tue$ was assigned a $\guti$ in the previous session then there is no random number generation.
Remark that when the $\tue$ and the $\thn$ are de-synchronized, the authentication fails and the $\tue$ sends a re-synchronization message. 
Since the session fails, no fresh $\guti$ is assigned to the $\tue$. Hence, the next session of the protocol has to conceal the $\supi$ using $\enc{\supi}{\pk_\hn}{\enonce}$, which requires a random number generation. Therefore, we constrain our protocol to use at most one random number generation by the $\tue$ per session, and only if no $\guti$ has been assigned or if the $\tue$ and the $\thn$ have been de-synchronized.

\paragraph{Summary}
We summarize the constraints for~$\faka$:
\begin{itemize}
\item \ch{It must use at most four messages per sessions.}{It must have at most three messages per session.}
\item The $\tue$ may use only keyed one-way functions and asymmetric \emph{encryption}. The $\thn$ may use these functions, plus asymmetric \emph{decryption}.
\item The $\tue$ may generate at most one random number per session, and only if no $\guti$ is available, or if re-synchronization with the $\thn$ is necessary.
\end{itemize}

\subsection{Key Ideas}
In this section, we present the two key ideas used in the design of the $\faka$ protocol.

\paragraph{Postponed Re-Synchronization Message}
We recall that whenever the $\tue$ and the $\thn$ are de-synchronized, the authentication fails and the $\tue$ sends a re-synchronization message. The problem is that this message can be distinguished from a $\lmacsym$ failure message, which allows the attack presented in Section~\ref{subsec:desync-attack}. Since the session fails, no $\guti$ is assigned to the $\tue$, and  the next session will use the asymmetric encryption to conceal the $\supi$. The first key idea is to piggy-back on the randomized encryption of the \emph{next session} to send a concealed re-synchronization message. More precisely, we replace the message $\enc{\supi}{\pk_\hn}{\enonce}$ by $\enc{\spair{\supi}{\sqn_\ue}}{\pk_\hn}{\enonce}$. This has several advantages:
\begin{itemize}
\item We can remove the re-synchronization message that lead to the unlinkability attack presented in Section~\ref{subsec:desync-attack}. In $\faka$, whenever the $\lmacsym$ check or the range check fails, the same failure message is sent.
\item This does not require more random number generation by the $\tue$, since a random number is already being generated to conceal the $\supi$ in the next session.
\end{itemize}
The $\threegpp$ technical specification (see~\cite{ts33501}, Annex C) requires that the asymmetric encryption used in the $\fiveaka$ protocol is the \textsc{ecies} encryption scheme, which is an hybrid encryption scheme. Hybrid encryption schemes use a randomized asymmetric encryption to conceal a temporary key. This key is then used to encrypt the message using a symmetric encryption, which is in-expensive. Hence encrypting the pair $\spair{\supi}{\sqn_\ue}$ is almost as fast as encrypting only $\supi$, and requires the $\tue$ to generate the same amount of randomness.


\paragraph{$\thn$ Challenge Before Identification}
To prevent the Encrypted $\imsi$ Replay Attack of Section~\ref{subsec:encrpt-imsi-replay}, we add a random challenge $\nonce$ from the $\thn$.
\ch
{The $\tue$ initiates the protocol by requesting a challenge without identifying itself. When requested, the $\thn$ generates and sends a fresh challenge $\nonce$ to the $\tue$, which includes it in its response}
{The challenge is sent before the $\tue$ identifies itself, and is included in the $\tue$'s response}
by $\lmacsym$-ing it with the $\supi$ using a symmetric one-way function $\macsym^1$ with key $\mkey^\ID$. The $\tue$ response is~now:
\[
  \lpair
  {\enc{\pair{\agent{\supi}}{\sqn_\ue}}
    {\pk_\hn}{\enonce}}
  {\mac{\spair
      {\enc{\pair{\agent{\supi}}{\sqn_\ue}}
        {\pk_\hn}{\enonce}}
      {\nonce}}
    {\mkey^\ID}{1}}
\]
This challenge is only needed when the encrypted permanent identity is used. If the $\tue$ uses a temporary identity $\guti$, then we do not need to use a random challenge. Indeed, temporary identities can only be used once before being discarded, and are therefore not subject to replay attacks. By consequence we split the protocol in two sub-protocols:
\begin{itemize}
\item The $\supi$ sub-protocol
  \ch
  {uses a random challenge from the $\thn$, encrypts the}
  {is initiated by the $\thn$ with a random challenge, uses the encrypted}
  permanent identity and allows to re-synchronize the $\tue$ and the $\thn$.
\item The $\guti$ sub-protocol is initiated by the $\tue$ using a temporary identity.
\end{itemize}
In the $\supi$ sub-protocol, the $\tue$'s answer includes the challenge. We use this to save one message: the last confirmation step from the $\tue$ is not needed, and is removed.
\ch
{The resulting sub-protocol has four messages. Observe that the $\guti$ sub-protocol is faster, since it uses only three messages.}
{As wanted, the resulting sub-protocol has only three messages.}

\subsection{Architecture and States}
\input{architecture-fig}

Instead of a monolithic protocol, we have three sub-protocols: the $\supi$ and $\guti$ sub-protocols, which handle authentication; and the $\refresh$ sub-protocol, which is run after authentication has been achieved and assigns a fresh temporary identity to the $\tue$. A full session of the $\faka$ protocol comprises a session of the $\supi$ or $\guti$ sub-protocols, followed by a session of the $\refresh$ sub-protocol. This is graphically depicted in Fig.~\ref{fig:faka-full-p-stack}.

\ch
{Since the $\guti$ sub-protocol uses only three messages and does not require the $\tue$ to generate a random number or compute an asymmetric encryption,}
{Since the $\guti$ sub-protocol does not require the $\tue$ to compute asymmetric encryption or to do random number generation,}
it is faster than the $\supi$ sub-protocol. By consequence, the $\tue$ should always use the $\guti$ sub-protocol if it has a temporary identity available.

The $\thn$ runs concurrently an arbitrary number of sessions, but a subscriber cannot run more than one session at the same time. Of course, sessions from \emph{different} subscribers may be concurrently running. We associate a unique integer, the session number, to every session, and we use $\thn(j)$ and $\tue_\ID(j)$ to refer to the $j$-th session of, respectively, the $\thn$ and the $\tue$ with identity $\ID$.

\paragraph{One-Way Functions}
We separate functions that are used only for confidentiality from functions that are also used for integrity. We have two confidentiality functions $\owsym$ and $\rowsym$, which use the key $\key$, and five integrity functions $\macsym^1$--$\,\macsym^5$, which use the key $\mkey$. We require that $\owsym$ and $\rowsym$ (resp. $\macsym^1$--$\,\macsym^5$) satisfy jointly the $\prfass$ assumption.

\ch{This is a new assumption, which requires that these functions are \emph{simultaneously} computationally indistinguishable from random functions.}{}
\ch{\begin{definition}[Jointly $\prfass$ Functions]
  Let $H_1(\cdot,\cdot),\dots,H_n(\cdot,\cdot)$ be a finite family of keyed hash functions from $\{0,1\}^*\times\{0,1\}^\eta$ to $\{0,1\}^\eta$. The functions $H_1,\dots,H_n$ are \emph{Jointly Pseudo Random Functions} if, for any PPTM adversary $\cala$ with access to oracles ${\mathcal{O}}_{f_1},\dots,{\mathcal{O}}_{f_n}$:
  \begin{multline*}
    |\Pr(k: \; {\cala}^{{\calo}_{H_1(\cdot,k)},\dots,{\calo}_{H_n(\cdot,k)}}(1^\eta)=1)
    -\\
    \Pr(g_1,\dots,g_n:\;
    {\cala}^{{\calo}_{g_1(\cdot)},\dots,{\calo}_{g_n(\cdot)}}(1^\eta)=1)|
  \end{multline*}
  is negligible, where:
  \begin{itemize}
  \item $k$ is drawn uniformly in $\{0,1\}^\eta$.
  \item $g_1,\dots,g_n$ are drawn uniformly in the set of all functions from $\{0,1\}^*$ to $\{0,1\}^\eta$.
  \end{itemize}
\end{definition}}{}
\ch{Observe that if $H_1,\dots,H_n$ are jointly $\prfass$ then, in particular, every individual $H_i$ is a $\prfass$.}{}
\ch{\begin{remark}
While this is a non-usual assumption, it is simple to build a set of functions $H_1,\dots,H_n$ which are jointly $\prfass$ from a single $\prfass$ $H$. For example, let $\ttag_1,\dots,\ttag_n$ be non-ambiguous tags, and let $H_i(m,\key) = H(\ttag_i(m),\key)$. Then, $H_1,\dots,H_n$ are jointly $\prfass$ whenever $H$ is a $\prfass$ (see %
\iffull
Appendix~\ref{app:subsection-joints-prf}).
\fi
\end{remark}}{}
\paragraph{$\tue$ Persistent State}
Each $\tue_\ID$ with identity $\ID$ has a state $\textsf{state}_\ue^\ID$ persistent across sessions. It contains the following immutable values: the permanent identity $\supi = \ID$, the confidentiality key $\key^\ID$, the integrity key $\mkey^\ID$ and the $\thn$'s public key $\pk_\hn$. The states also contain mutable values: the sequence number $\sqn_\ue$, the temporary identity $\guti_\ue$ and the boolean $\success_\ue$. We have $\success_\ue = \false$ whenever no valid temporary identity is assigned to the $\tue$. Finally, there are mutable values that are not persistent across sessions. E.g. $\bauth_\ue$ stores $\thn$'s random challenge, and $\eauth_\ue$ stores $\thn$'s random challenge \emph{when the authentication is successful}.

\paragraph{$\thn$ Persistent State}
The $\thn$ state $\textsf{state}_\hn$ contains the secret key $\sk_\hn$ corresponding to the public key $\pk_\hn$. Also, for every subscriber with identity $\ID$, it stores the keys $\key^\ID$ and $\mkey^\ID$, the permanent identity $\supi = \ID$, the $\thn$ version of the sequence number $\sqn_\hn^\ID$ and the temporary identity $\guti_\hn^\ID$. It stores in $\tsuccess_\hn^\ID$ the random challenge of the last session that was either a successful $\supi$ session which modified the sequence number, or a $\guti$ session which authenticated $\ID$. This is used to detect and prevent some subtle attacks, which we present later. Finally, every session $\thn(j)$ stores in $\bauth_\hn^j$ the identity claimed by the $\tue$, and in $\eauth_\hn^j$ the identity of the $\tue$ it authenticated.

\input{figure-supi}

\subsection{The $\supi$, $\guti$ and $\refresh$ Sub-Protocols}
We describe honest executions of the three sub-protocols of the $\faka$ protocol. An honest execution is an execution where the adversary dutifully forwards the messages without tampering. Each execution is between a $\tue$ and $\thn(j)$.

\paragraph{The $\supi$ Sub-Protocol}
This protocol uses the $\tue$'s permanent identity, re-synchronizes the $\tue$ and the $\thn$ and is expensive to run. The protocol is sketched in Fig.~\ref{fig:faka-supi}. 

\ch
{The $\tue$ initiates the protocol by requesting a challenge from the network. When asked, $\thn(j)$ sends a fresh random challenge $\nonce^j$. After}
{$\thn(j)$ initiates the protocol by sending a random challenge $\nonce^j$. After}
receiving $\nonce^j$, the $\tue$ stores it in $\bauth_\ue$, and answers with the encryption of its permanent identity together with the current value of its sequence number, using the $\thn$ public key $\pk_\hn$. It also includes the $\lmacsym$ of this encryption and of the\ch{}{ random} challenge, which yields the message:
\[
  \lpair
  {\enc{\pair{\agent{\supi}}{\sqn_\ue}}
    {\pk_\hn}{\enonce}}
  {\mac{\spair
      {\enc{\pair{\agent{\supi}}{\sqn_\ue}}
        {\pk_\hn}{\enonce}}
      {\nonce^j}}
    {\mkey^\ID}{1}}
\]
Then the $\tue$ increments its sequence number by one. When it gets this message, the $\thn$ retrieves the pair $\spair{\supi}{\sqn_\ue}$ by decrypting the encryption using its secret key $\sk_\hn$. For every identity $\ID$, it checks if $\supi = \ID$ and if the $\lmacsym$ is correct. If this is the case, $\thn$ authenticated $\ID$, and it stores $\ID$ in $\bauth_\hn^j$ and $\eauth_\hn^j$. After having authenticated $\ID$, $\thn$ checks whether the sequence number $\sqn_\ue$ it received is greater than or equal to $\sqn_\hn^\ID$. If this holds, it sets $\sqn_\hn^\ID$ to $\sqn_\ue + 1$, stores $\nonce^j$ in $\tsuccess_\hn^\ID$, generates a fresh temporary identity $\guti^j$ and stores it into $\guti_\hn^\ID$. This additional check ensures that the $\thn$ sequence number is always increasing, which is a crucial property of the protocol.

If the $\thn$ authenticated $\ID$, it sends a confirmation message $\mac{\spair{\nonce^j}{\sqn_\ue + 1}}{\mkey^\ID}{2}$ to the $\tue$. This message is sent even if the received sequence number $\sqn_\ue$ is smaller than $\sqn_\hn^\ID$. When receiving the confirmation message, if the $\lmacsym$ is valid then the $\tue$ authenticated the $\thn$, and it stores in $\eauth_\ue$ the initial random challenge (which it keeps in $\bauth_\ue$). If the $\lmacsym$ test fails, it stores in $\eauth_\ue$ the special value $\fail$.


\input{figure-guti}

\paragraph{The $\guti$ Sub-Protocol}
This protocol uses the $\tue$'s temporary identity, requires synchronization to succeed and is inexpensive.  The protocol is sketched in Fig.~\ref{fig:faka-guti}. 

When $\success_\ue$ is true, the $\tue$ can initiate the protocol by sending its temporary identity $\guti_\ue$. The $\tue$ then sets $\success_\ue$ to $\false$ to guarantee that this temporary identity is not used again. When receiving a temporary identity $\sfx$, $\thn$ looks if there is an $\ID$ such that $\guti_\hn^\ID$ is equal to $\sfx$ and is not $\unset$. If the temporary identity belongs to $\ID$, it sets $\guti_\hn^\ID$ to $\unset$ and stores $\ID$ in $\bauth_\hn^j$. Then it generates a random challenge $\nonce^j$, stores it in $\tsuccess_\hn^\ID$, and sends it to the $\tue$, together with the xor of the sequence number $\sqn_\hn^\ID$ with $\ow{\nonce^j}{\key^\ID}$, and a $\lmacsym$:
\[
  \ltriplet
  {\nonce^j}
  {\sqn^\ID_\hn \oplus \ow{\nonce^j}{\key^\ID}}
  {\mac{\striplet
      {\nonce^j}{\sqn^\ID_\hn}{\guti_\hn^\ID}}
    {\mkey^\ID}{3}}
\]
When it receives this message, the $\tue$ retrieves the challenge $\nonce^j$ at the beginning of the message, computes $\ow{\nonce^j}{\key^\ID}$ and uses this value to unconceal the sequence number $\sqn_\hn^\ID$. It then computes $\mac{\striplet{\nonce^j}{\sqn_\hn^\ID}{\guti_\ue}}{\mkey^\ID}{3}$ and compares it to the $\lmacsym$ received from the network. If the $\lmacsym$s are not equal, or if the range check $\range{\sqn_\ue}{\sqn_\hn^\ID}$ fails, it puts $\fail$ into $\bauth_\ue$ and $\eauth_\ue$ to record that the authentication was not successful. If both tests succeed, it stores in $\bauth_\ue$ and $\eauth_\ue$ the random challenge, increments $\sqn_\ue$ by one and sends the confirmation message $\mac{\nonce^j}{\mkey^\ID}{4}$. When receiving this message, the $\thn$ verifies that the $\lmacsym$ is correct. If this is the case then the $\thn$ authenticated the $\tue$, and stores $\ID$ into $\eauth_\hn^\ID$. Then, $\thn$ checks whether $\tsuccess_\hn^\ID$ is still equal to the challenge $\nonce^j$ stored in it at the beginning of the session. If this is true, the $\thn$ increments $\sqn_\hn^\ID$ by one, generates a fresh temporary identity $\guti^j$ and stores it into $\guti_\hn^\ID$.

\paragraph{The $\refresh$ Sub-Protocol}
The $\refresh$ sub-protocol is run after a successful authentication, regardless of the authentication sub-protocol used. It assigns a fresh temporary identity to the $\tue$ to allow the next $\faka$ session to run the faster $\guti$ sub-protocol. It is depicted in Fig.~\ref{fig:faka-refresh}.

The $\thn$ conceals the temporary identity $\guti^j$ generated by the authentication sub-protocol by xoring it with $\row{\nonce^j}{\key^\ID}$, and $\lmacsym$s it.
When receiving this message, $\tue$ unconceals the temporary identity $\guti_\hn^\ID$ by xoring its first component with $\row{\eauth_\ue}{\mkey^\ID}$ (since $\eauth_\ue$ contains the $\thn$'s challenge after authentication). Then $\tue$ checks that the $\lmacsym$ is correct and that the authentication was successful. If it is the case, it stores $\guti_\hn^\ID$ in $\guti_\ue$ and sets $\success_\ue$ to $\true$.

\input{figure-refresh}


%% file: architecture-fig.tex
\begin{figure}[tb]
  \begin{center}
    \begin{tikzpicture}[>=latex, every node/.style={font={\footnotesize}}]
      \tikzset{mn/.style={draw=gray!80,line width=1.6pt,rounded corners, minimum width=2.5cm, minimum height=0.7cm}}

      \node[mn] (a) at (0,0) {
        {$\supi$ Sub-Protocol}
      };

      \node[mn] (b) at (3.6,0) {
        {$\guti$ Sub-Protocol}
      };

      \node[mn] (c) at (1.8,-1.3) {
        {$\refresh$ Sub-Protocol}};

      \draw[->] (a) -- (c);
      \draw[->] (b) -- (c);
    \end{tikzpicture}
  \end{center}
  \caption{General Architecture of the $\faka$ Protocol}
  \label{fig:faka-full-p-stack}
\end{figure}
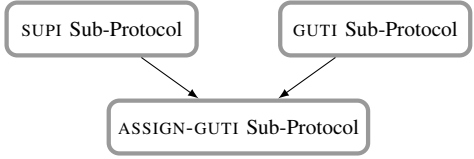


%% file: figure-supi.tex
\begin{figure}[tb]
  \begin{center}
    \begin{tikzpicture}[>=latex]
      \tikzset{every node/.style={font=\footnotesize}};
      \tikzset{en/.style={minimum height=0.2cm,minimum width=0.15cm,fill=black}};
      \tikzset{mb/.style={solid,thick,draw=gray!90}};
      \tikzset{mp/.style={inner sep=0,outer sep=0,draw=none}};
      \tikzset{sb/.style={draw,fill=white,double,align=left}};

      \path (0,0)
      node[above,yshift=0.1cm,anchor=south] (ue) {$\tue$}
      node[draw,name path=suen,double=gray!80,double distance=0.035cm,below,
      anchor=north west,xshift=-0.5cm,align=left] (sue)
      {$\textsf{state}_\ue^\ID$};

      \path[name path=aline] (ue) -- ++(0,-10.6);

      \path[draw,very thick,name intersections={of=suen and aline}]
      (intersection-2) node[inner sep=0] (tue) {}
      -- ++(0,-9.65) node[en](xend){};

      \path (7.7,0)
      node[above,yshift=0.1cm,anchor=south] (hn) {$\thn(j)$}
      node[double=gray!80,draw,name path=shnn,double distance=0.035cm,below,
      anchor=north east,xshift=0.5cm,align=left] (shn)
      {$\textsf{state}_\hn$};

      \path[name path=bline] (hn) -- ++(0,-10.6);

      \path[name path=xendp] (xend) -- ++(8,0);
      \path[name intersections={of=xendp and bline}]
      (intersection-1) node[mp] (xendr){};
      \path[draw,very thick,name intersections={of=shnn and bline}]
      (intersection-2) node[inner sep=0] (thn) {}
      -- (xendr) node[en]{};

      \path[black] (tue) -- ++(0,-0.3)
      [draw,->
      ] -- ++(7.7,0)
      node[midway,above]{$\ch{\rchallenge}{}$}
      node[mp] (hnpre0) {};

      \path (hnpre0) -- ++(0,-0.6)
      [draw,->] -- ++(-7.7,0)
      node[midway,above]{$\nonce^j$}
      node[mp] (ue1) {};

      \path (ue1) -- ++(0,-0.3)
      node[sb,name path=pue1,anchor=north west,xshift=-0.5cm] (c) {
        \textbf{\underline{Input $\nonce_\sfr$}:} 
        $\bauth_\ue \la \nonce_\sfr$};

      \path[name intersections={of=aline and pue1}]
      (intersection-2) node[mp] (ue2){};

      \path (ue2) -- ++(0,-0.8)
      node[mp] (ue2bs){}
      [draw,->] -- ++(7.7,0)
      node[midway,above]{$
        \lpair
        {\enc{\pair{\agent{\supi}}{\sqn_\ue}}
          {\pk_\hn}{\enonce}}
        {\mac{\spair
            {\enc{\pair{\agent{\supi}}{\sqn_\ue}}
              {\pk_\hn}{\enonce}}
            {\nonce_\sfr}}
          {\mkey^\ID}{1}}$}
      node[mp] (hn3) {};

      \path (ue2bs) -- ++(0,-0.3)
      node[sb,name path=pue2bs,anchor=north west,xshift=-0.5cm] (c) {
        $\sqn_\ue \la \sqn_\ue + 1$};

      \path (hn3) -- ++(0,-0.3)
      node[sb,name path=phn3,anchor=north east,xshift=0.5cm] (c) {
        \textbf{\underline{Input $\sfy$:}}\\[0.15em]
        $\pair{\ID_\sfr}{\sqn_\sfr} \la
        \dec(\pi_1(\sfy),\sk_\hn)$\\[0.25em]
        $\sfb^\ID_\macsym \la
        \begin{alignedat}[t]{1}
          &
          \pi_2(\sfy) =
          \mac{\spair
              {\pi_1(\sfy)}
              {\nonce^j}}
            {\mkey^\ID}{1}\\[-0.4em]
          &\wedge \ID_\sfr = \ID
        \end{alignedat}$\\[0.25em]
        $\sfbinc^\ID \la
        \sfb_\macsym^\ID \wedge \sqn_\sfr \ge \sqn^\ID_\hn$\\[0.25em]
        $\!
        \begin{alignedat}{4}
          &\textsf{if }\sfb^\ID_\macsym
          &&\textsf{ then }\;&&\bauth^j_\hn,\eauth^j_\hn &\;\la\;& \ID
        \end{alignedat}$\\[0.3em]
        $\!
        \begin{alignedat}{4}
          &\textsf{if }\sfbinc^\ID
          &&\textsf{ then }\;&&\sqn^\ID_\hn &\;\la\;& \sqn_\sfr + 1\\[-0.3em]
          &&&&&\tsuccess_\hn^\ID &\;\la\;& \nonce^j\\[-0.3em]
          &&&&&\guti^\ID_\hn &\;\la\;& \guti^j
        \end{alignedat}$};

      \path[name intersections={of=bline and phn3}]
      (intersection-2) node[mp] (hn4){};

      \path (hn4) -- ++(0,-1.1)
      node[mp] (hn5) {}
      [draw,->] -- ++(-7.7,0)
      node[mp] (ue4) {}
      node[midway,above]{$
        \mac{\spair{\nonce^j}{\sqn_\sfr + 1}}{\mkey^\ID}{2}$};

      \path (hn5) -- ++(0.5,0.9) node[mp] (hn5l){};
      \path (hn5) -- ++(-7.5,0) -- ++(-0.5,-0.2) node[mp] (hn5r){};
      \path (hn5l) node[outer sep=0,draw,fill=white,anchor=north east] (a)
      {\textbf{$\sfb_\macsym$}};
      \path (a.north west) -- ++ (0,-0.15) node[mp] (x){}
      (a.south east) -- ++ (-0.15,0) node[mp] (y){};
      \draw[mb] (x.base) -| (hn5r.base) -| (y.base);

      \path (ue4) -- ++(0,-0.4)
      node[sb,name path=pue4,anchor=north west,xshift=-0.5cm] (c) {
        \textbf{\underline{Input $\sfz$:}}\\[0.25em]
        $\sfb_{\textsf{ok}} \la
        \sfz = \mac{\spair{\bauth_\ue}{\sqn_\ue}}{\mkey^\ID}{2}$\\[0.25em]
        $\eauth_\ue \la \ite{\sfb_{\textsf{ok}}}{\bauth_\ue}{\fail}$};
    \end{tikzpicture}
  \end{center}
  \caption{The $\supi$ Sub-Protocol of the $\faka$ Protocol}
  \label{fig:faka-supi}
\end{figure}


%% file: figure-guti.tex
\begin{figure}[tb]
  \begin{center}
    \begin{tikzpicture}[>=latex]
      \tikzset{every node/.style={font=\footnotesize}};
      \tikzset{en/.style={minimum height=0.2cm,minimum width=0.15cm,fill=black}};
      \tikzset{mb/.style={solid,thick,draw=gray!90}};
      \tikzset{mp/.style={inner sep=0,outer sep=0,draw=none}};
      \tikzset{sb/.style={draw,fill=white,double,align=left}};

      \path (0,0)
      node[above,yshift=0.1cm,anchor=south] (ue) {$\tue$}
      node[draw,name path=suen,double=gray!80,double distance=0.035cm,below,
      anchor=north west,xshift=-0.5cm,align=left] (sue)
      {$\textsf{state}_\ue^\ID$};

      \path[name path=aline] (ue) -- ++(0,-12.9);

      \path[draw,very thick,name intersections={of=suen and aline}]
      (intersection-2) node[inner sep=0] (tue) {}
      -- ++(0,-12.6) node[en](xend){};

      \path (7.7,0)
      node[above,yshift=0.1cm,anchor=south] (hn) {$\thn(j)$}
      node[double=gray!80,draw,name path=shnn,double distance=0.035cm,below,
      anchor=north east,xshift=0.5cm,align=left] (shn)
      {$\textsf{state}_\hn$};

      \path[name path=bline] (hn) -- ++(0,-13.6);

      \path[name path=xendp] (xend) -- ++(8,0);
      \path[name intersections={of=xendp and bline}]
      (intersection-1) node[mp] (xendr){};
      \path[draw,very thick,name intersections={of=shnn and bline}]
      (intersection-2) node[inner sep=0] (thn) {}
      -- (xendr) node[en]{};

      \path (tue) -- ++(0,-0.9)
      node[mp] (uetemp) {}
      [draw,->] -- ++(7.7,0)
      node[midway,above]{$\guti_\ue$}
      node[mp] (hn1) {};

      \path (uetemp) -- ++(-0.5,0.7) node[mp] (uetempl){};
      \path (uetemp) -- ++(7.5,0) -- ++(0.5,-0.2) node[mp] (uetempr){};
      \path (uetempl) node[outer sep=0,draw,fill=white,anchor=north west] (a)
      {\textbf{$\success_\ue$}};
      \path (a.north east) -- ++ (0,-0.15) node[mp] (x){}
      (a.south west) -- ++ (0.15,0) node[mp] (y){};
      \draw[mb] (x.base) -| (uetempr.base) -| (y.base);

      \path (uetemp) -- ++(0,-0.5)
      node[sb,name path=puetemp,anchor=north west,xshift=-0.5cm] (c)
      {$\success_\ue \la \false$};

      \path (hn1) -- ++(0,-0.5)
      node[sb,name path=phn1,anchor=north east,xshift=0.5cm] (c) {
        \textbf{\underline{Input $\sfx$}:}\\[0.15em]
        $\sfb_\ID \la \guti_\hn^\ID = \sfx \wedge
        \guti_\hn^\ID \ne \unset$\\[0.3em]
        $\!
        \begin{alignedat}{4}
          &\textsf{if }\sfb_\ID&&\textsf{ then }\;&&
          \guti_\hn^\ID &\;\la\;& \unset\\[-0.4em]
          &&&&&\bauth^j_\hn &\;\la\;& \ID\\[-0.4em]
          &&&&&\tsuccess_\hn^\ID &\;\la\;& \nonce^j
        \end{alignedat}$};

      \path[name intersections={of=bline and phn1}]
      (intersection-2) node[mp] (hn2){};

      \path (hn2) -- ++(0,-1.2)
      node[mp] (hn3) {}
      [draw,->] -- ++(-7.7,0)
      node[mp] (ue2) {}
      node[midway,above]{$
        \triplet
        {\nonce^j}
        {\sqn^\ID_\hn \oplus \ow{\nonce^j}{\key^\ID}}
        {\mac{\striplet
            {\nonce^j}{\sqn^\ID_\hn}{\guti_\hn^\ID}}
          {\mkey^\ID}{3}}$};

      \path (hn3) -- ++(0.5,1) node[mp] (hn3l){};
      \path (hn3) -- ++(-7.5,0) -- ++(-0.5,-0.2) node[mp] (hn3r){};
      \path (hn3l) node[outer sep=0,draw,fill=white,anchor=north east] (a)
      {\textbf{$\sfb_\ID$}};
      \path (a.north west) -- ++ (0,-0.15) node[mp] (x){}
      (a.south east) -- ++ (-0.15,0) node[mp] (y){};
      \draw[mb] (x.base) -| (hn3r.base) -| (y.base);

      \path (ue2) -- ++(0,-0.5)
      node[sb,name path=pue2,anchor=north west,xshift=-0.5cm] (c) {
        \textbf{\underline{Input $\sfy$}:}\\[0.15em]
        $\nonce_\sfr,\,\sqn_\sfr \la
        \pi_1(\sfy),\,\pi_2(\sfy) \oplus \ow{\nonce_\sfr}{\key^\ID}$\\[0.15em]
        $\sfb_\acc \la
        \begin{alignedat}[t]{1}
          &\pi_3(\sfy) =
          \mac{\striplet
            {\nonce_\sfr}
            {\sqn_\sfr}
            {\guti_\ue}}{\mkey^\ID}{3})\\[-0.3em]
          & \wedge \range{\sqn_\ue}{\sqn_\sfr}
        \end{alignedat}$\\[0.3em]
        $\!
        \begin{alignedat}{4}
          &\textsf{if }\sfb_\acc&&\textsf{ then }\;&&
          \bauth_\ue,\eauth_\ue &\;\la\;& \nonce_\sfr\\[-0.3em]
          &&&&&\sqn_\ue &\;\la\;& \sqn_\ue + 1
        \end{alignedat}$\\[0.2em]
        $\!\begin{alignedat}{4}
          &\textsf{if }\neg \sfb_\acc&&\textsf{ then }\;&&
          \bauth_\ue,\eauth_\ue &\;\la\;& \fail
        \end{alignedat}$};

      \path[name intersections={of=aline and pue2}]
      (intersection-2) node[mp] (ue3){};

      \path (ue3) -- ++(0,-1.2)
      node[mp] (ue4) {}
      [draw,->] -- ++(7.7,0)
      node[mp] (hn4) {}
      node[midway,above]{$
        \mac{\nonce_\sfr}{\mkey^\ID}{4}$};

      \path (ue4) -- ++(-0.5,1) node[mp] (ue4l){};
      \path (ue4) -- ++(7.5,0) -- ++(0.5,-0.2) node[mp] (ue4r){};
      \path (ue4l) node[outer sep=0,draw,fill=white,anchor=north west] (a)
      {\textbf{$\sfb_\acc$}};
      \path (a.north east) -- ++ (0,-0.15) node[mp] (x){}
      (a.south west) -- ++ (0.15,0) node[mp] (y){};
      \draw[mb] (x.base) -| (ue4r.base) -| (y.base);

      \path (hn4) -- ++(0,-0.5)
      node[sb,name path=pue2,anchor=north east,xshift=0.5cm] (c) {
        \textbf{\underline{Input $\sfz$}:}\\[0.25em]
        $\sfb^\ID_\macsym \la (\bauth_\hn^j = \ID) \wedge
        (\sfz =
        \mac{\nonce^j}{\mkey^\ID}{4})$\\[0.25em]
        $\sfbinc^\ID \la
        \sfb^\ID_\macsym \wedge \tsuccess_\hn^\ID = \nonce^j$\\[0.25em]
        $\!
        \begin{alignedat}{4}
          &\textsf{if }\sfb^\ID_\macsym &&\textsf{ then }\;&&
          \eauth^j_\hn &\;\la\;& \ID\\
          &\textsf{if }\sfbinc^\ID&&\textsf{ then }\;&&
          \sqn^\ID_\hn &\;\la\;& \sqn^\ID_\hn + 1\\[-0.3em]
          &&&&& \guti^\ID_\hn &\;\la\;& \guti^j
        \end{alignedat}$};
    \end{tikzpicture}
  \end{center}
  \caption{The $\guti$ Sub-Protocol of the $\faka$ Protocol}
  \label{fig:faka-guti}
\end{figure}


%% file: figure-refresh.tex
\begin{figure}[tb]
  \begin{center}
    \begin{tikzpicture}[>=latex]
      \tikzset{every node/.style={font=\footnotesize}};
      \tikzset{en/.style={minimum height=0.2cm,minimum width=0.15cm,fill=black}};
      \tikzset{mb/.style={solid,thick,draw=gray!90}};
      \tikzset{mp/.style={inner sep=0,outer sep=0,draw=none}};
      \tikzset{sb/.style={draw,fill=white,double,align=left}};
      
      \path (0,0)
      node[above,yshift=0.1cm,anchor=south] (ue) {$\tue$}
      node[draw,name path=suen,double=gray!80,double distance=0.035cm,below,
      anchor=north west,xshift=-0.5cm,align=left] (sue)
      {$\textsf{state}_\ue^\ID$};

      \path[name path=aline] (ue) -- ++(0,-4.85);
      
      \path[draw,very thick,name intersections={of=suen and aline}]
      (intersection-2) node[inner sep=0] (tue) {}
      -- ++(0,-4.45) node[en](xend){};
      
      \path (7.7,0)
      node[above,yshift=0.1cm,anchor=south] (hn) {$\thn(j)$}
      node[double=gray!80,draw,name path=shnn,double distance=0.035cm,below,
      anchor=north east,xshift=0.5cm,align=left] (shn)
      {$\textsf{state}_\hn$};

      \path[name path=bline] (hn) -- ++(0,-5.35);

      \path[name path=xendp] (xend) -- ++(8,0);
      \path[name intersections={of=xendp and bline}]
      (intersection-1) node[mp] (xendr){};
      \path[draw,very thick,name intersections={of=shnn and bline}]
      (intersection-2) node[inner sep=0] (thn) {}
      -- (xendr) node[en]{};
      
      \path (thn) -- ++(0,-1.2)
      node[mp] (hn1) {}
      [draw,->] -- ++(-7.7,0)
      node[pos=0.6,above]{$
        \spair
        {\guti^j \oplus \row{\nonce^j}{\key^\ID}}
        {\mac{\pair{\guti^j}{\nonce^j}}{\mkey^\ID}{5}}$}
      node[mp] (ue1) {};

      \path (hn1) -- ++(0.5,0.9) node[mp] (hn1l){};
      \path (hn1) -- ++(-7.5,0) -- ++(-0.5,-0.2) node[mp] (hn1r){};
      \path (hn1l) node[outer sep=0,draw,fill=white,anchor=north east] (a)
      {\textbf{$\eauth_\hn^\ID = \ID$}};
      \path (a.north west) -- ++ (0,-0.15) node[mp] (x){}
      (a.south east) -- ++ (-0.15,0) node[mp] (y){};
      \draw[mb] (x.base) -| (hn1r.base) -| (y.base);

      \path (ue1) -- ++(0,-0.5)
      node[sb,name path=pue1,anchor=north west,xshift=-0.5cm] (c) {
        \textbf{\underline{Input $\sfx$}:}\\[0.15em]
        $\guti_\sfr \la \pi_1(\sfx) \oplus \row{\eauth_\ue}{\mkey^\ID}$\\[0.25em]
        $\sfb_\acc \la
        \begin{alignedat}[t]{1}
          &\big(\pi_2(\sfx) =
          \mac{\spair{\guti_\sfr}{\eauth_\ue}}{\mkey^\ID}{5}\big) \\[-0.2em]
          &\wedge
          (\eauth_\ue \ne \fail)
        \end{alignedat}$\\[0.3em]
        $\guti_\ue \la \ite{\sfb_\acc}{\guti_\sfr}{\unset}$\\
        $\success_\ue \la \sfb_\acc$};

    \end{tikzpicture}
  \end{center}
  \caption{The $\refresh$ Sub-Protocol of the $\faka$ Protocol}
  \label{fig:faka-refresh}
\end{figure}


%% file: ss-unlinkability.tex
\section{Unlinkability}
\label{section:unlink-body}
We now define the unlinkability property we use, which is inspired from \cite{DBLP:conf/esorics/HermansPVP11} and Vaudenay's privacy~\cite{DBLP:conf/asiacrypt/Vaudenay07}.

\paragraph{Definition}
The property is defined by a game in which an adversary tries to link together some subscriber's sessions. The adversary is a PPTM which interacts, through oracles, with $N$ different subscribers with identities $\ID_1,\dots,\ID_N$, and with the $\thn$. The adversary cannot use a subscriber's permanent identity to refer to it, as it may not know it. Instead, we associate a virtual handler $\vh$ to any subscriber currently running a session of the protocol. We maintain a list $\freeuelist$ of all subscribers that are ready to start a session. We now describe the oracles $\mathcal{O}_b$:
\begin{itemize}
\item $\startsession()$: starts a new $\thn$ session and returns its session number $j$.
\item $\sendhn(m,j)$ (resp. $\sendue(m,\vh)$): sends the message $m$ to $\thn(j)$ (resp. the $\tue$ associated with $\vh$), and returns $\thn(j)$ (resp. $\vh$) answer.
\item $\resulthn(j)$ (resp. $\resultue(\vh)$): returns $\true$ if $\thn(j)$ (resp. the $\tue$ associated with $\vh$) has made a successful authentication.
\item $\drawue(\ID_{i_0},\ID_{i_1})$: checks that $\ID_{i_0}$ and $\ID_{i_1}$ are both in $\freeuelist$. If that is the case, returns a new virtual handler pointing to $\ID_{i_b}$, depending on an internal secret bit $b$. Then, it removes $\ID_{i_0}$ and $\ID_{i_1}$ from $\freeuelist$.
\item $\freeue(\vh)$: makes the virtual handler $\vh$ no longer valid, and adds back to $\freeuelist$ the two identities that were removed when the virtual handler was created.
\end{itemize}

We recall that a function is negligible if and only if it is asymptotically smaller than the inverse of any polynomial. An adversary $\mathcal{A}$ interacting with $\mathcal{O}_b$ is winning the $q$-unlinkability game if: $\mathcal{A}$ makes less than $q$ calls to the oracles; and it can guess the value of the internal bit $b$ with a probability better than $1/2$ by a non-negligible margin, i.e. if the following quantity is non negligible in $\eta$:
\[
  \left| 2 \times \Pr \left(b: \mathcal{A}^{\mathcal{O}_b}(1^\eta) = b \right) - 1 \right|
\]
Finally, a protocol is $q$-unlinkable if and only if there are no winning adversaries against the $q$-unlinkability game.

\ch{\paragraph{Corruption}
In~\cite{DBLP:conf/esorics/HermansPVP11,DBLP:conf/asiacrypt/Vaudenay07}, the adversary is allowed to corrupt some tags using a $\corrupt$ oracle. Several classes of adversary are defined by restricting its access to the corruption oracle. A \emph{strong} adversary has unrestricted access, a \emph{destructive} adversary can no longer use a tag after corrupting it (it is destroyed), a \emph{forward} adversary can only follow a $\corrupt$ call by further $\corrupt$ calls, and finally a \emph{weak} adversary cannot use $\corrupt$ at all. A protocol is $\calc$ unlinkable if no adversary in $\calc$ can win the unlinkability game. Clearly, we have the following relations:
\[
  \emph{strong} \;\Ra\;
  \emph{destructive} \;\Ra\;
  \emph{forward} \;\Ra\;
  \emph{weak}
\]
The $\fiveaka$ protocol does not provide forward secrecy: indeed, obtaining the long-term secret of a $\tue$ allows to decrypt all its past messages. By consequence, the best we can hope for is \emph{weak} unlinkability. Since such adversaries cannot call $\corrupt$, we removed the oracle from our definition.}{}

\paragraph{Wide Adversary}
Note that the adversary knows if the protocol was successful or not using the $\resultue$ and $\resulthn$ oracles (such an adversary is called \emph{wide} in Vaudenay's terminology~\cite{DBLP:conf/asiacrypt/Vaudenay07}). Indeed, in an authenticated key agreement protocol, this information is always available to the adversary: if the key exchange succeeds then it is followed by another protocol using the newly established key; while if it fails then either a new key-exchange session is initiated, or no message is sent. Hence the adversary knows if the key exchange was successful by passive monitoring.

\subsection{\texorpdfstring{$\sigma$}{sigma}-Unlinkability}

\input{fig-attack-guti}

\ch
{In accord with our conjecture in Section~\ref{subsection:conj-unlink}, the $\faka$ protocol is not unlinkable. Indeed, an adversary $\cala$ can easily win the linkability game. First, $\cala$ ensures that $\ID_{\agent{A}}$ and $\ID_{\agent{B}}$ have a valid temporary identity assigned: $\cala$ calls $\drawue(\ID_{\agent{A}},\ID_{\agent{A}})$ to obtain a virtual handler for $\ID_{\agent{A}}$, and runs a $\supi$ and $\refresh$ sessions between $\ID_{\agent{A}}$ and the $\thn$ with no interruptions. This assigns a temporary identity to $\ID_{\agent{A}}$. We use the same procedure for $\ID_{\agent{B}}$.}{}

\ch
{Then, $\cala$ executes the attack described in Fig.~\ref{fig:attack-guti-link}. It starts a $\guti$ session with $\ID_{\agent{A}}$, and intercepts the last message. At that point, $\ID_{\agent{A}}$ no longer has a temporary identity, while $\ID_{\agent{B}}$ still does. Then, it calls $\drawue(\ID_{\agent{A}},\ID_{\agent{B}})$, which returns a virtual handler $\vh$ to $\ID_{\agent{A}}$ or $\ID_{\agent{B}}$. The attacker then start a new $\guti$ session with $\vh$. If $\vh$ is a handler for $\ID_{\agent{A}}$, the $\tue$ returns $\nosuci$. If $\vh$ aliases $\ID_{\agent{B}}$, the $\tue$ returns the temporary identity $\guti_{\agent{A}}$. The adversary $\cala$ can distinguish between these two cases, and therefore wins the game.}
{As discussed in Section~\ref{subsection:conj-unlink}, we believe that unlinkability cannot be achieved for $\faka$. Indeed, an adversary can easily win the linkability game. First, it calls $\drawue(\ID_{\agent{A}},\ID_{\agent{B}})$ on a de-synchronized user $\ID_{\agent{A}}$ and a synchronized user $\ID_{\agent{B}}$. The call to $\drawue$ returns a virtual handler $\vh$ to $\ID_{\agent{A}}$ or $\ID_{\agent{B}}$. The attacker then runs a full session of the protocol with $\vh$. If the session fails, $\vh$ is a handler for $\ID_{\agent{A}}$, otherwise $\vh$ aliases $\ID_{\agent{B}}$. In both cases, the adversary guesses the correct value of the internal bit $b$.}

\paragraph{$\sigma$-unlinkability}
To prevent this, we want to forbid $\drawue$ to be called on de-synchronized subscribers. We do this by modifying the state of the user chosen by $\drawue$. We let $\sigma$ be an update on the state of the subscribers. We then define the oracle $\drawue_\sigma(\ID_{i_0},\ID_{i_1})$: it checks that $\ID_{\agent{A}}$ and $\ID_{\agent{B}}$ are both free, then \emph{applies the update} $\sigma$ to $\ID_{i_b}$'s state, and returns a new virtual handler pointing to $\ID_{i_b}$. The $(q,\sigma)$-unlinkability game is the $q$-unlinkability game in which we replace $\drawue$ with $\drawue_\sigma$. A protocol is $(q,\sigma)$-unlinkable if and only if there is no winning adversary against the $(q,\sigma)$-unlinkability game. Finally, a protocol is $\sigma$-unlinkable if it is $(q,\sigma)$-unlinkable for any $q$.

\paragraph{Application to $\faka$}
The privacy guarantees given by the $\sigma$-unlinkability depend on the choice of $\sigma$. The idea is to choose a $\sigma$ that allows to establish privacy in \emph{some scenarios} of the standard unlinkability game\footnote{Remark that when $\sigma$ is the empty state update, the $\sigma$-unlinkability and unlinkability properties coincide.}. 

We illustrate this on the $\faka$ protocol. Let $\sigma_\sunlink = \success_\ue \mapsto \false$ be the function that makes the $\tue$'s temporary identity not valid. This simulates the fact that the $\guti$ has been used and is no longer available. If the $\tue$'s temporary identity is not valid, then it can only run the $\supi$ sub-protocol. Hence, if the $\faka$ protocol is $\sigma_\sunlink$-unlinkable, then no adversary can distinguish between a normal execution and an execution where we change the identity of a subscriber each time it runs the $\supi$ sub-protocol. We give in Fig.~\ref{fig:ssunlink-example} an example of such a scenario. We now state our main result:
\input{ss-fig}
\begin{theorem}
  \label{thm:main-unlink}
  The $\faka$ protocol is $\sigma_\sunlink$-unlinkable for an arbitrary number of agents and sessions when the asymmetric encryption $\enc{\_}{\_}{\_}$ is \textsc{ind-cca1} secure and $\owsym$ and $\rowsym$ (resp. $\macsym^1$--$\,\macsym^5$) satisfy jointly the $\prfass$ assumption.
\end{theorem}
This result is shown later in the paper. Still, the intuition is that no adversary can distinguish between two sessions of the $\supi$ protocol. Moreover, the $\supi$ protocol has two important properties. First, it re-synchronizes the user with the $\thn$, which prevents the attacker from using any prior de-synchronization. Second, the $\faka$ protocol is designed in such a way that no message sent by the $\tue$ before a successful $\supi$ session can modify the $\thn$'s state after the $\supi$ session. Therefore, any time the $\supi$ protocol is run, we get a ``clean slate'' and we can change the subscriber's identity. Note that we have a trade-off between efficiency and privacy: the $\supi$ protocol is more expensive to run, but provides more privacy. 

\subsection{A Subtle Attack}
\begin{figure}[tb]
  \begin{center}
    \begin{tikzpicture}[>=latex]
      \tikzset{every node/.style={font=\footnotesize}};
      \tikzset{en/.style={minimum height=0.2cm,minimum width=0.15cm,fill=black}};
      \tikzset{mb/.style={solid,thick,draw=gray!90}};
      \tikzset{mp/.style={inner sep=0,outer sep=0,draw=none}};
      \tikzset{sb/.style={draw,fill=white,double,align=left}};
      
      \path (0,0)
      node[mp,yshift=0.1cm,anchor=south] (ue) {}
      node[xshift=-0.5cm,yshift=0.1cm,anchor=south west]
      (uet) {$\tue_{\ID_{\agent{A}}}$};

      \path[name path=aline] (ue) -- ++(0,-1.4) node[mp] (xend) {};
            
      \path (7.7,0)
      node[above,yshift=0.1cm,anchor=south] (hn) {$\thn$};

      \path[name path=bline] (hn) -- ++(0,-2.8);

      \path[name path=uebeg] (ue.south) -- ++(8,0);
      \path[name intersections={of=uebeg and bline}]
      (intersection-1) node[mp] (hnbeg){};
      
      \path[name path=xendp] (xend) -- ++(8,0);
      \path[name intersections={of=xendp and bline}]
      (intersection-1) node[mp] (xendr){};
      \path[draw,very thick]
      (hnbeg) -- (xendr) node[]{};

      \path[draw,very thick]
      (ue) -- (xend) node[]{};
      
      \path (ue) -- ++(0,-0.2)
      [draw,->] -- ++(7.7,0)
      node[pos=0.5,above]{$\guti$}
      node[mp] (hn1) {};

      \path (hn1) -- ++(0,-0.5)
      [draw,->] -- ++(-7.7,0)
      node[pos=0.5,above]{$\dots$}
      node[mp] (ue2) {};

      \path (ue2) -- ++(0,-0.5)
      node[mp] (ue3) {}
      [draw,->] -- ++(3.7,0)
      node[midway,above]{$t_{\textsf{auth}}$}
      node[mp] (a){};
      \path[draw=none] (a) -- ++(0.3,0)
      node[midway,transform shape,scale=1.4] (c) {$/$}
      node[mp](b){};

      
      \path (xend) -- ++(0,-0.8)
      node[mp,yshift=0.1cm,anchor=south] (ue) {}
      node[xshift=-0.5cm,yshift=0.1cm,anchor=south west]
      (uet) {$\tue_{\ID_{\agent{A}}}$ or $\tue_{\ID_{\agent{B}}}$};

      \path[name path=aline] (ue) -- ++(0,-2.7) node[mp] (xend) {};
            
      \path (ue) -- ++(7.7,0)
      node[above,yshift=0.1cm,anchor=south] (hn) {$\thn$};

      \path[name path=bline] (hn) -- ++(0,-3.1);

      \path[name path=uebeg] (ue.south) -- ++(8,0);
      \path[name intersections={of=uebeg and bline}]
      (intersection-1) node[mp] (hnbeg){};

      \path[name path=xendp] (xend) -- ++(8,0);
      \path[name intersections={of=xendp and bline}]
      (intersection-1) node[mp] (xendr){};
      \path[draw,very thick]
      (hnbeg) -- (xendr) node[en]{};

      \path[draw,very thick]
      (ue) -- (xend) node[en]{};
      
      \path (hn) -- ++(0,-0.85)
      node[mp] (hn1) {}
      -- ++(-7.7,0)
      node[pos=0.5]{$\supi$ Session};

      \path (hn1) -- ++(0.5,0.3) node[mp] (hn1l){};
      \path (hn1) -- ++(-7.7,0) -- ++(-0.5,-0.3) node[mp] (hn1r){};
      \draw[mb] (hn1l.center) -| (hn1r.center) -| (hn1l.center);
     
      \path (hn1r) --  ++(0.5,0) -- ++(0,-0.55)
      node[mp] (ue2) {}
      -- ++(3.7,0)
      node[mp] (a){};
      \path[draw=none] (a) -- ++(0.3,0)
      node[midway,transform shape,scale=1.4] (c) {$/$}
      node[mp](b){};
      \draw[->] (b)
      -- ++(3.7,0)
      node[midway,above]{$t_{\textsf{auth}}$}
      node[mp] (hn2){};

      \path (hn2) -- ++(0,-0.75)
      node[mp] (hn3) {}
      -- ++(-7.7,0)
      node[pos=0.5]{$\guti$ Session};

      \path (hn3) -- ++(0.5,0.3) node[mp] (hn3l){};
      \path (hn3) -- ++(-7.7,0) -- ++(-0.5,-0.3) node[mp] (hn3r){};
      \draw[mb] (hn3l.center) -| (hn3r.center) -| (hn3l.center);
    \end{tikzpicture}
  \end{center}
  \caption{A Subtle Attack Against The $\fakam$ Protocol}
  \label{fig:subtle-attack-fakam}
\end{figure}

We now explain what is the role of $\tsuccess_\hn^\ID$, and how it prevents a subtle attack against the $\sigma_\sunlink$-unlinkability of $\faka$. We let $\fakam$ be the $\faka$ protocol where we modify the $\guti$ sub-protocol we described in Fig.~\ref{fig:faka-guti}: in the state update of the $\thn$'s last input, we remove the check $\tsuccess_\hn^\ID = \nonce^j$ (i.e. $\sfbinc^\ID = \sfb_{\macsym}^\ID$). The attack is described in Fig.~\ref{fig:subtle-attack-fakam}.

First, we run a session of the $\guti$ sub-protocol between $\tue_{\ID_{\agent{A}}}$ and the $\thn$, but we do not forward the last message $t_{\textsf{auth}}$ to the $\thn$. We then call $\drawue_{\sigma_\sunlink}(\ID_{\agent{A}},\ID_{\agent{B}})$, which returns a virtual handler $\vh$ to $\ID_{\agent{A}}$ or $\ID_{\agent{B}}$.
We run a full session using the $\supi$ sub-protocol with $\vh$, and then send the message $t_{\textsf{auth}}$ to the $\thn$. We can check that, because we removed the condition $\tsuccess_\hn^\ID = \nonce^j$ from $\sfbinc^\ID$, this message causes the $\thn$ to increment $\sqn_\hn^{\ID_{\agent{A}}}$ by one. At that point, $\tue_{\ID_{\agent{A}}}$ is de-synchronized but $\tue_{\ID_{\agent{B}}}$ is synchronized. Finally, we run a session of the $\guti$ sub-protocol. The session has two possible outcomes: if $\vh$ aliases to ${\agent{A}}$ then it fails, while if $\vh$ aliases to ${\agent{B}}$, it succeeds. This leads to an attack.

When we removed the condition $\tsuccess_\hn^\ID = \nonce^j$, we broke the ``clean slate'' property of the $\supi$ sub-protocol: we can use a message from a session that started \emph{before} the $\supi$ session to modify the state \emph{after} the $\supi$ session. $\tsuccess_\hn^\ID$ allows to detect whether another session has been executed since the current session started, and to prevent the update of the sequence number when this is the case. 


%% file: fig-attack-guti.tex
\begin{figure}[t]
  \begin{center}
    \begin{tikzpicture}[>=latex
      ]
      \tikzset{every node/.style={font=\footnotesize}};
      \tikzset{en/.style={minimum height=0.2cm,minimum width=0.15cm,fill}};
      \tikzset{mb/.style={solid,thick,draw=gray!90}};
      \tikzset{mp/.style={inner sep=0,outer sep=0,minimum size=0,draw=none}};
      \tikzset{sb/.style={draw,fill=white,double,align=left}};

      \pgfmathsetmacro{\width}{7.7}
      
      \path (0,0)
      node[mp,yshift=0.1cm,anchor=south] (ue) {}
      node[xshift=-0.5cm,yshift=0.1cm,anchor=south west]
      (uet) {$\tue_{\ID_{\agent{A}}}$};

      \path[name path=aline] (ue) -- ++(0,-1.4) node[mp] (xend) {};
            
      \path (\width,0)
      node[above,yshift=0.1cm,anchor=south] (hn) {$\thn$};

      \path[name path=bline] (hn) -- ++(0,-2.8);

      \path[name path=uebeg] (ue.south) -- ++(8,0);
      \path[name intersections={of=uebeg and bline}]
      (intersection-1) node[mp] (hnbeg){};
      
      \path[name path=xendp] (xend) -- ++(8,0);
      \path[name intersections={of=xendp and bline}]
      (intersection-1) node[mp] (xendr){};
      \path[draw,very thick]
      (hnbeg) -- (xendr) node[]{};

      \path[draw,very thick]
      (ue) -- (xend) node[]{};
      
      \path (ue) -- ++(0,-0.2)
      [draw,->] -- ++(\width,0)
      node[pos=0.5,above]{$\guti_{\agent{A}}$}
      node[mp] (hn1) {};

      \path (hn1) -- ++(0,-0.5)
      [draw,->] -- ++(-\width,0)
      node[pos=0.5,above]{$\dots$}
      node[mp] (ue2) {};

      \path (ue2) -- ++(0,-0.5)
      node[mp] (ue3) {}
      [draw,->] -- ++(\width/2-0.15,0)
      node[midway,above]{$\dots$}
      node[mp] (a){};
      \path[draw=none] (a) -- ++(0.3,0)
      node[midway,transform shape,scale=1.4] (c) {$/$}
      node[mp](b){};

      
      \path (xend) -- ++(0,-0.8)
      node[mp,yshift=0.1cm,anchor=south] (ue) {}
      node[xshift=-0.5cm,yshift=0.5,anchor=south west]
      (uet) {$\tue_{\ID_{\agent{X}}}$ where $\ID_{\agent{X}} = \ID_{\agent{A}}$ or $\ID_{\agent{B}}$};

      \path[name path=aline,overlay] (ue) -- ++(0,-20);
            
      \path (ue) -- ++(\width,0)
      node[above,anchor=south] (hn) {$\thn$};

      \path[name path=bline,overlay] (hn) -- ++(0,-20);

      \path[name path=uebeg] (ue.south) -- ++(8,0);
      \path[name intersections={of=uebeg and bline}]
      (intersection-1) node[mp] (hnbeg){};
      
      \path (ue) -- ++(0,-0.95)
      node[mp] (ue2){}
      [draw,->] -- ++(\width,0)
      node[pos=0.5,above]{$\nosuci$};

      \path (ue2) -- ++(-0.5,0.75) node[mp] (ue2l){};
      \path (ue2) -- ++(\width,0) -- ++(0.3,-0.2) node[mp] (ue2r){};
      \path (ue2l) node[outer sep=0,draw,fill=white,anchor=north west] (a)
      {\textbf{$\ID_{\agent{X}} = \ID_{\agent{A}}$}};
      \path (a.north east) -- ++ (0,-0.2) node[mp] (x){}
      (a.south west) -- ++ (0.2,0) node[mp] (y){};
      \draw[mb] (x.base) -| (ue2r.base) -| (y.base);

      \path (ue2) -- ++(0,-1.1)
      node[mp] (ue3){}
      [draw,->] -- ++(\width,0)
      node[pos=0.5,above]{$\guti_{\agent{B}}$};

      \path (ue3) -- ++(-0.5,0.75) node[mp] (ue3l){};
      \path (ue3) -- ++(\width,0) -- ++(0.3,-0.2) node[mp] (ue3r){};
      \path (ue3l) node[outer sep=0,draw,fill=white,anchor=north west] (a)
      {\textbf{$\ID_{\agent{X}} = \ID_{\agent{B}}$}};
      \path (a.north east) -- ++ (0,-0.2) node[mp] (x){}
      (a.south west) -- ++ (0.2,0) node[mp] (y){};
      \draw[mb] (x.base) -| (ue3r.base) -| (y.base);

      \begin{pgfonlayer}{bg0}
        \draw[very thick] (ue) -- (ue3)
        -- ++ (0,-0.5) node[en]{};
        \draw[very thick]
        let \p1 = (hnbeg),
        \p2 = (ue3) in
        (hnbeg) -- (\x1,\y2) 
        -- ++ (0,-0.5) node[en]{};
      \end{pgfonlayer}
    \end{tikzpicture}
  \end{center}
  \caption{Consecutive $\guti$ Sessions of $\faka$ Are Not Unlinkable.}
  \label{fig:attack-guti-link}
\end{figure}


%% file: ss-fig.tex
\begin{figure}[t]
  \begin{center}
    \begin{tikzpicture}
      [>=latex, every node/.style={font={\footnotesize}},
      scale=0.9,transform shape]

      \tikzset{mb/.style={draw=gray!80,line width=1.6pt,
          minimum width=1cm, minimum height=0.7cm}}
      \tikzset{mp/.style={mb}}
      \tikzset{ms/.style={mb,rounded corners = 8pt}}

      \path (0,0)
      node[mp] (a) {$\agent{A}$}
      -- +(0,-1.55)
      node[mp] (a0) {$\agent{A}$}
      
      -- ++(1.6,0)
      node[mp] (b) {$\agent{B}$}
      -- +(0,-1.55)
      node[mp] (b0) {$\agent{B}$}
      
      -- ++(1.6,0)
      node[ms] (c) {$\agent{A}$}
      -- +(0,-1.55)
      node[ms] (c0) {$\agent{A}$}
      
      -- ++(1.6,0)
      node[mp] (d) {$\agent{B}$}
      -- +(0,-1.55)
      node[mp] (d0) {$\agent{C}$}
      
      -- ++(1.6,0)
      node[ms] (e) {$\agent{B}$}
      -- +(0,-1.55)
      node[ms] (e0) {$\agent{C}$}
      
      -- ++(1.6,0)
      node[ms] (f) {$\agent{B}$}
      -- +(0,-1.55)
      node[ms] (f0) {$\agent{C}$};

      \draw[->] (a) -- (b) -- (c) -- (d) -- (e) -- (f)
      (f.east) -- ++(0.4,0);
      \draw[->] (a0) -- (b0) -- (c0) -- (d0) -- (e0) -- (f0)
      (f0.east) -- ++(0.4,0);

      \path[name path=p1] (c) -- (d0);
      \path[name path=p2] (c0) -- (d);
      \path[name intersections={of=p1 and p2}]
      (intersection-1) node[scale=1.8] {$\sim$};

    \end{tikzpicture}    
  \end{center}
  \caption{Two indistinguishable executions. Square (resp. round) nodes are executions of the $\supi$ (resp. $\guti$) sub-protocol. Each time the $\supi$ sub-protocol is used, we can change the subscriber's identity.}
  \label{fig:ssunlink-example}
\end{figure}


%% file: modeling.tex
\section{Modeling in The Bana-Comon Logic}
\label{section:modeling-body}

We prove Theorem~\ref{thm:main-unlink} using the Bana-Comon model introduced in~\cite{Bana:2014:CCS:2660267.2660276}. This is a first order logic, in which protocol messages are represented by terms using special function symbols for the adversary's inputs. It has only one predicate, $\sim$, which represents computational indistinguishability. To use this model, we first build a set of axioms $\axioms$ specifying what the adversary \emph{cannot} do. This set of axiom comprises computationally valid properties, cryptographic hypotheses and implementation assumptions. Then, given a protocol and a security property, we compute a formula $\phi$ expressing the protocol security. Finally, we show that the security property $\phi$ can be deduced from the axioms $\axioms$. If this is the case, this entails \emph{computational security}.

\subsection{Syntax and Semantics}
We quickly recall the syntax and semantics of the logic.
\paragraph{Terms}
Terms are built using function symbols in $\sig$, names in $\Nonce$ (representing random samplings) and variables in~$\mathcal{X}$. The set $\sig$ of function symbols contains a countable set of \emph{adversarial} function symbols $\mathcal{G}$, which represent the adversary inputs, and protocol function symbols. The protocol function symbols are the functions used in the protocol, e.g. the pair $\pair{\_}{\_}$, the $i$-th projection $\pi_i$, encryption $\enc{\_}{\_}{\_}$, decryption $\dec(\_,\_)$, $\symite$, $\true$, $\false$, equality $\eq{\_}{\_}$, integer greater or equal $\Geq{\_}{\_}$ and length $\length(\_)$.

\paragraph{Formulas}
For every integer $n$, we have one predicate symbol $\sim_{n}$ of arity $2n$, which represents equivalence between two vectors of terms of length $n$. We use an infix notation for $\sim_n$, and omit $n$ when not relevant. Formulas are built using the usual Boolean connectives and first-order quantifiers.

\paragraph{Semantics}
We use the classical semantics of first-order logic. Given an interpretation domain, we interpret terms, function symbols and predicates as, respectively, elements, functions and relations of this domain.

We focus on a particular class of models, called the \emph{computational models} (see~\cite{Bana:2014:CCS:2660267.2660276} for a formal definition). In a computational model $\cmodel$, terms are interpreted in the set of PPTMs equipped with a working tape and two random tapes $\rho_1,\rho_2$. The tape $\rho_1$ is used for the protocol random values, while $\rho_2$ is for the adversary's random samplings. The adversary cannot access directly the random tape $\rho_1$, although it may obtain part of $\rho_1$ through the protocol messages. A key feature is to let the interpretation of an adversarial function $g$ be \emph{any} PPTM, which soundly models an attacker \emph{arbitrary probabilistic polynomial time computation}. Moreover, the predicates $\sim_n$ are interpreted using \emph{computational indistinguishability} $\approx$. Two families of distributions of bit-string sequences $(m_\eta)_{\eta}$ and $(m'_\eta)_{\eta}$, indexed by $\eta$, are indistinguishable iff for every PPTM $\mathcal{A}$ with random tape $\rho_2$, the following quantity is negligible in $\eta$:
\begin{multline*}
  \big|\,\Pr(\rho_1,\rho_2: \; \mathcal{A}(m_\eta(\rho_1,\rho_2),\rho_2)=1) \;-\\
  \Pr(\rho_1,\rho_2: \; \mathcal{A}(m'_\eta(\rho_1,\rho_2),\rho_2)=1) \,\big|
\end{multline*}

\subsection{Modeling of the $\faka$ Protocol States and Messages}

We now use the Bana-Comon logic to model the $\sigma_{\sunlink}$-unlinkability of the $\faka$ protocol. We consider a setting with $N$ identities $\ID_1,\dots,\ID_N$, and we let $\iddom$ be the set of all identities. To improve readability, protocol descriptions often omit some details. For example, in Section~\ref{section:faka-description} we sometimes omitted the description of the error messages. The failure message attack of~\cite{DBLP:conf/ccs/ArapinisMRRGRB12} demonstrates that such details may be crucial for security. An advantage of the Bana-Comon model is that it requires us to fully formalize the protocol, and to make all assumptions explicit.


\paragraph{Symbolic State}
For every identity $\ID \in \iddom$, we use several variables to represent $\tue_{\ID}$'s state. E.g., $\sqn_\ue^\ID$ and $\guti_\ue^\ID$ store, respectively, $\tue_{\ID}$'s sequence number and temporary identity. Similarly, we have variables for $\thn$'s state, e.g. $\sqn_\hn^\ID$. We let $\vardom$ be the set of variables used in $\faka$:
\begin{equation*}
  \small
  \bigcup_{\substack{j \in \mathbb{N},\textsc{a} \in \{\ue,\hn\}\\ \ID \in \iddom\\}}
  \left\{
    \begin{array}[c]{l}
      \small
      \sqn_{\textsc{a}}^\ID, \guti_{\textsc{a}}^\ID,\eauth_\ue^\ID, \bauth_\ue^\ID,
      \eauth_\hn^j\\
      \small
      \bauth_\hn^j, \uetsuccess^\ID, \success_{\ue}^\ID, \tsuccess_\hn^\ID
    \end{array}
  \right\}
\end{equation*}
A symbolic state $\cstate$ is a mapping from $\vardom$ to terms. Intuitively, $\cstate(\sfx)$ is a term representing (the distribution of) the value of~$\sfx$.
\begin{example}
  To avoid confusion with the \emph{semantic} equality $=$, we use $\equiv$ to denote \emph{syntactic} equality. Then, we can express the fact that $\guti_\ue^\ID$ is unset in a symbolic state $\cstate$ by having $\cstate(\guti_\ue^\ID) \equiv \unset$. Also, given a state $\cstate$, we \ch{can state}{} that $\cstate'$ is the state $\cstate$ in which we incremented  $\sqn_\ue^\ID$ by having $\cstate'(\sfx)$ be the term $\cstate(\sqn_\ue^\ID) + 1$ if $\sfx$ is $\sqn_\ue^\ID$, and $\cstate(\sfx)$ otherwise.
\end{example}

\paragraph{Symbolic Traces}
We explain how to express $(q,\sigma_\sunlink)$-unlinkability in the BC model. In the $(q,\sigma_\sunlink)$-unlinkability game, the adversary chooses dynamically which oracle it wants to call. This is not convenient to use in proofs, as we do not know statically the $i$-th action of the adversary. We prefer an alternative point-of-view, in which the trace of oracle calls is fixed (w.l.o.g., as shown later in Proposition~\ref{prop:main-unlink-bc}). Then, there are no winning adversaries against the $\sigma_\sunlink$-unlinkability game with a fixed trace of oracle calls if the adversary's interactions with the oracles when $b = 0$ are indistinguishable from the interactions with the oracles when~$b = 1$.

We use the following action identifiers to represent symbolic calls to the oracle of the $(q,\sigma_\sunlink)$-unlinkability game:
\begin{itemize}
\item $\ns_{\ID}(j)$ represents a call to $\drawue_{\sigma_\sunlink}(\ID,\_)$ when $b = 0$ or $\drawue_{\sigma_\sunlink}(\_,\ID)$ when $b = 1$.
\item $\npuai{i}{\ID}{j}$ (resp. $\cuai_\ID(j,i)$) is the $i$-th user message in the session $\tue_{\ID}(j)$ of the $\supi$ (resp. $\guti$) sub-protocol.
\item $\fuai_\ID(j)$ is the only user message in the session $\tue_{\ID}(j)$ of the $\refresh$ sub-protocol.
\item $\pnai(j,i)$ (resp. $\cnai(j,i)$) is the $i$-th network message in the session $\thn(j)$ of the $\supi$ (resp. $\guti$) sub-protocol.
\item $\fnai(j)$ is the only network message in the session $\thn(j)$ of the $\refresh$ sub-protocol.
\end{itemize}
The remaining oracle calls either have no outputs and do not modify the state (e.g. $\startsession$), or can be simulated using the oracles above. E.g., since the $\thn$ sends an error message whenever the protocol is not successful, the output of $\resulthn$ can be deduced from the protocol messages.

A \emph{symbolic trace} $\tau$ is a finite sequence of action identifiers. We associate, to any execution of the $(q,\sigma_{\sunlink})$-unlinkability game with a fixed trace of oracle calls, a pair of symbolic traces $(\tau_l,\tau_r)$, which corresponds to the adversary's interactions with the oracles when $b$ is, respectively, $0$ and $1$. We let $\runlink$ be the set of such pairs of traces.

\begin{example}
  We give the symbolic traces corresponding to the honest execution of $\faka$ between $\tue_{\ID}(i)$ and $\thn(j)$. If the $\supi$ protocol is used, we have the trace $\tau^{i,j}_\supi(\ID)$:
  {\small
    \begin{gather*}
      \ch{\npuai{0}{\ID}{i},\,}{}\pnai(j,0), \npuai{1}{\ID}{i}, \pnai(j,1), \npuai{2}{\ID}{i}, \fnai(j), \fuai_\ID(i)
    \end{gather*}}
  And if the $\guti$ sub-protocol is used, the trace $\tau^{i,j}_\guti(\ID)$:
  {\small
    \begin{gather*}
      \cuai_\ID(i,0), \cnai(j,0), \cuai_\ID(i,1), \cnai(j,1), \fnai(j), \fuai_\ID(i)
    \end{gather*}}
  Which such notations, the left trace $\tau_l$ of the attack described in Fig.~\ref{fig:subtle-attack-fakam}, in which the adversary only interacts with $\agent{A}$, is:
  {\small
    \[
      \cuai_{\agent{A}}(0,0), \cnai(0,0), \cuai_{\agent{A}}(0,1),
      \tau^{1,1}_\supi(\agent{A}),\cnai(0,1), \tau^{2,2}_\guti(\agent{A})
    \]}
  Similarly, we can give the right trace $\tau_r$ in which the adversary interacts with $\agent{A}$ and $\agent{B}$:
  {\small
    \[
      \cuai_{\agent{A}}(0,0), \cnai(0,0), \cuai_{\agent{A}}(0,1),
      \tau^{0,1}_\supi(\agent{B}),\cnai(0,1), \tau^{1,2}_\guti(\agent{B})
    \]}
\end{example}

\paragraph{Symbolic Messages}
We define, for every action identifier $\ai$, the term representing the output observed by the adversary when $\ai$ is executed. Since the protocol is stateful, this term is a function of the prefix trace of actions executed since the beginning. We define by mutual induction, for any symbolic trace $\tau = \tauo,\ai$ whose last action is $\ai$:
\begin{itemize}
\item The term $t_\tau$ representing the last message observed during the execution of $\tau$.
\item The symbolic state $\cstate_\tau$ representing the state after the execution of $\tau$.
\item The frame $\cframe_\tau$ representing the sequence of all messages observed during the execution of $\tau$.
\end{itemize}
Some syntactic sugar: we let $\instate_\tau = \cstate_\tauo$ be the symbolic state before the execution of the last action; and $\inframe_\tau = \cframe_\tauo$ be the sequence of all messages observed during the execution of $\tau$, except for the last message.

The frame $\cframe_\tau$ is simply the frame $\inframe_\tau$ extended with $t_\tau$, i.e. $\cframe_\tau \equiv \inframe_\tau,t_\tau$. Moreover the initial frame contains only $\pk_\hn$, i.e. $\cframe_\epsilon \equiv \pk_\hn$. When executing an action $\ai$, only a subset of the symbolic state is modified. For example, if the adversary interacts with $\tue_{\ID}$ then the state of the $\thn$ and of all the other users is unchanged. Therefore instead of defining $\cstate_\tau$, we define the \emph{symbolic state update} $\upstate_\tau$, which is a \emph{partial} function from $\vardom$ to terms. Then $\cstate_\tau$ is the function:
\[
  \cstate_\tau(\sfx) \equiv
  \begin{dcases}
    \instate_\tau(\sfx) & \text{ if $\sfx \not \in \dom(\upstate_\tau)$}\\
    \upstate_\tau(\sfx) & \text{ if $\sfx \in \dom(\upstate_\tau)$}
  \end{dcases}
\]
where $\dom$ gives the domain of a function. Now, for every action $\ai$, we define $t_\tau$ and $\upstate_\tau$ using $\inframe_\tau$ and $\instate_\tau$. As an example, we describe the second message and state update of the session $\tue_{\ID}(j)$ for the $\supi$ sub-protocol, which corresponds to the action $\npuai{1}{\ID}{j}$. We recall the relevant part of Fig.~\ref{fig:faka-supi}:
\begin{center}
  \begin{tikzpicture}[>=latex]
    \tikzset{every node/.style={font=\footnotesize}};
    \tikzset{en/.style={minimum height=0.2cm,minimum width=0.15cm,fill=black}};
    \tikzset{mb/.style={solid,thick,draw=gray!90}};
    \tikzset{mp/.style={inner sep=0,outer sep=0,draw=none}};
    \tikzset{sb/.style={draw,fill=white,double,align=left}};

    \path (0,0)
    node[above,yshift=0.1cm,anchor=south] (ue) {$\tue$};
    \path[name path=aline] (ue) -- ++(0,-2.9);

    \path[draw,very thick]
    (ue) -- ++(0,-3) ;

    \path (ue) -- ++(0,-0.4)
    -- ++(4,0) node[mp] (a) {};
    \path[draw,->] (a)
    -- ++(-4,0)
    node[mp] (ue1) {};

    \path (ue1) -- ++(0,-0.25)
    node[sb,name path=pue1,anchor=north west,xshift=-0.5cm] (c) {
      \textbf{\underline{Input $\nonce_\sfr$}:} $\bauth_\ue \la \nonce_\sfr$};

    \path[name intersections={of=aline and pue1}]
    (intersection-2) node[mp] (ue2){};

    \path (ue2) -- ++(0,-0.8)
    node[mp] (ue2bs){}
    [draw,->] -- ++(7,0)
    node[midway,above]{$
      \lpair
      {\enc{\pair{\agent{\ID}}{\sqn_\ue}}
        {\pk_\hn}{\enonce}}
      {\mac{\spair
          {\enc{\pair{\agent{\ID}}{\sqn_\ue}}
            {\pk_\hn}{\enonce}}
          {\nonce_\sfr}}
        {\mkey}{1}}$}
    node[mp] (hn3) {};

    \path (ue2bs) -- ++(0,-0.25)
    node[sb,name path=pue2bs,anchor=north west,xshift=-0.5cm] (c) {
      $\sqn_\ue \la \sqn_\ue + 1$};
  \end{tikzpicture}
\end{center}
First, we need a term representing the value inputted by $\tue_{\ID}$ from the network. As we have an active adversary, this value can be anything that the adversary can compute using the knowledge it accumulated since the beginning of the protocol. The knowledge of the adversary, or the frame, is the sequence of all messages observed during the execution of $\tau$, except for the last message. This is exactly $\inframe_\tau$. Finally, we use a special function symbol $g \in \mathcal{G}$ to represent the arbitrary polynomial time computation done by the adversary. This yields the term $g(\inframe_\tau)$, which symbolically represents the input.

We now need to build a term representing the asymmetric encryption of the pair containing the $\tue$'s permanent identity $\ID$ and its sequence number. The permanent identity $\ID$ is simply represented using a constant function symbol $\ID$ (we omit the parenthesis $()$). $\tue_{\ID}$'s sequence number is stored in the variable $\sqn_\ue^\ID$. To retrieve its value, we just do a look-up in the symbolic state $\instate_\tau$, which yields $\instate_\tau(\sqn_\ue^\ID)$. Finally, we use the asymmetric encryption function symbol to build the term $t^{\textsf{enc}}_\tau \equiv \enc{\spair{\ID}{\instate_\tau(\sqn_\ue^\ID)}}{\pk_\hn}{\enonce^j}$. Notice that the encryption is randomized using a nonce $\enonce^j$, and that the freshness of the randomness is guaranteed by indexing the nonce with the session number $j$. Finally, we can give $t_\tau$ and $\upstate_\tau$:
{  \small
  \begin{alignat*}{2}
    t_\tau &\;\equiv\;&&
    \;\;\lpair{t^{\textsf{enc}}_\tau}
    {\mac{\spair
        {t^{\textsf{enc}}_\tau}
        {g(\inframe_\tau)}}
      {\mkey^\ID}{1}}\\
    \upstate_\tau &\;\equiv\;&&
    \left\{
      \begin{array}{lcl}
        \sqn_\ue^\ID \mapsto \sqnsuc(\instate_\tau(\sqn_\ue^\ID))&\;&
        \eauth_\ue^\ID \mapsto \fail\\
        \bauth_\ue^\ID \mapsto g(\inframe_\tau)&&
        \suci_\ue^\ID \mapsto \unset\\
        \success_\ue^\ID \mapsto \false&&
      \end{array}\right.
  \end{alignat*}}
Remark that we omitted some state updates in the description of the protocol in Fig.~\ref{fig:faka-supi}. For example, $\tue_{\ID}$ temporary identity $\guti_\ue^\ID$ is reset when starting the $\supi$ sub-protocol. In the BC model, these details are made explicit.

The description of $t_\tau$ and $\upstate_\tau$ for the other actions can be found in Fig.~\ref{fig:protocol-term-supi} and Fig.~\ref{fig:protocol-guti-refresh}. Observe that we describe
\ch{one more message for}
{four messages for}
the $\supi$ and $\guti$ protocols
\ch
{than in}
{, while in}
Section~\ref{section:faka-description}
\ch{.}{they only have three messages.} This is because we add one message ($\npuai{2}{\ID}{j}$ for $\supi$ and $\cnai(j,1)$ for $\guti$) for proof purposes, to simulate the $\resultue$ and $\resulthn$ oracles. Also, notice that in the $\guti$ protocol, when $\thn$ receives a $\guti$ that is not assigned to anybody, it sends a decoy message to a special dummy identity $\IDdum$.

\input{symbolic-terms}

The following soundness theorem states that security in the BC model implies computationally security:
\begin{proposition}
  \label{prop:main-unlink-bc}
  The $\faka$ protocol is $\sigma_\sunlink$-unlinkable in any computational model satisfying the axioms $\axioms$ if, for every $(\tau_l,\tau_r) \in \runlink$, we can derive $\cframe_{\tau_l} \sim \cframe_{\tau_r}$ using $\axioms$.
\end{proposition}
The proof of this result is basically the proof that Fixed Trace Privacy implies Bounded Session Privacy in~\cite{DBLP:conf/csfw/ComonK17}. We omit the details.


%% file: symbolic-terms.tex
\begin{figure}[tb]
  \footnotesize
  \ch{\underline{\textbf{Case} $\ai = {\npuai{0}{\ID}{j}}$.}
  $t_\tau \;\equiv\; \rchallenge$}{}\\[0.4em]
  \underline{\textbf{Case} $\ai = {\pnai(j,0)}$.}
  $t_\tau \;\equiv\; \nonce^j$

  \underline{\textbf{Case} $\ai = {\npuai{1}{\ID}{j}}$.}
  Let $t^{\textsf{enc}}_\tau \equiv \enc{\spair{\ID}{\instate_\tau(\sqn_\ue^\ID)}}{\pk_\hn}{\enonce^j}$, then:
  \begin{alignat*}{2}
    t_\tau &\;\equiv\;&&
    \;\;\lpair{t^{\textsf{enc}}_\tau}
    {\mac{\spair
        {t^{\textsf{enc}}_\tau}
        {g(\inframe_\tau)}}
      {\mkey^\ID}{1}}\\
    \upstate_\tau &\;\equiv\;&&
    \left\{
      \begin{array}{lcl}
        \sqn_\ue^\ID \mapsto \sqnsuc(\instate_\tau(\sqn_\ue^\ID))&\quad&
        \eauth_\ue^\ID \mapsto \fail\\
        \bauth_\ue^\ID \mapsto g(\inframe_\tau)&&
        \suci_\ue^\ID \mapsto \unset\\
        \success_\ue^\ID \mapsto \false&&
      \end{array}\right.
  \end{alignat*}

  \underline{\textbf{Case} $\ai = {\pnai(j,1)}$.}
  Let $t_{\textsf{dec}} \equiv \dec(\pi_1(g(\inframe_\tau)),\sk_\hn)$, and let:
  \begin{alignat*}{2}
    \accept_\tau^{\ID_i} &\;\equiv\;\;&&
    \begin{alignedat}[t]{1}
      &\eq{\pi_2(g(\inframe_\tau))}
      {\mac{\spair
          {\pi_1(g(\inframe_\tau))}
          {\nonce^j}}
        {\mkey^{\ID_i}}{1}}\\
      &\wedge\;
      \eq{\pi_1(t_{\textsf{dec}})}{\ID_i}
    \end{alignedat}\\
    \incaccept_\tau^{\ID_i}  &\;\equiv\;\;&&
    \accept_\tau^{\ID_i} \wedge
    \Geq{\pi_2(t_{\textsf{dec}})}
    {\instate_\tau(\sqn^{\ID_i}_\hn)}\qquad\qquad\null
  \end{alignat*}
  \vspace{-2em}
  \begin{alignat*}{2}
    t_\tau &\;\equiv\; &&
    \begin{alignedat}[t]{2}
      &\ite{\accept_\tau^{\ID_1}}
      {\mac{\spair
          {\nonce^j}
          {\sqnsuc(\pi_2(t_{\textsf{dec}}))}}
        {\mkey^{\ID_1}}{2}\\ &\!\!}
      {\ite{\accept_\tau^{\ID_2}}
        {\mac{\spair
            {\nonce^j}
            {\sqnsuc(\pi_2(t_{\textsf{dec}}))}}
          {\mkey^{\ID_2}}{2}\\[-0.4em] & \qquad \cdots\\[-0.3em] &\!\!}
        {\unknownid}}
    \end{alignedat}\\
    \upstate_\tau &\;\equiv\;&&
    \begin{dcases}
      \tsuccess_\hn^{\ID_i} \mapsto
      \begin{alignedat}[t]{2}
        \ite{\incaccept_\tau^{\ID_i}&}{\nonce^j}
        {\instate_\tau(\tsuccess_\hn^{\ID_i})}
      \end{alignedat}\\
      \suci_\hn^{\ID_i} \mapsto
      \begin{alignedat}[t]{2}
        \ite{\incaccept_\tau^\ID &}
        {\suci^j }
        {\instate_\tau(\suci_\hn^{\ID_i})}
      \end{alignedat}\\
      \sqn_\hn^{\ID_i} \mapsto
      \begin{alignedat}[t]{2}
        \ite{\incaccept_\tau^{\ID_i}&}
        {\sqnsuc(\pi_2(t_{\textsf{dec}})) }
        {\instate_\tau(\sqn_\hn^{\ID_i})}
      \end{alignedat}\\
      \bauth_\hn^{j},\eauth_\hn^{j} \mapsto
      \begin{alignedat}[t]{2}
        &\ite{\accept_\tau^{\ID_1}}
        {\ID_1\\ & \!\!}
        {\ite{\accept_\tau^{\ID_2}}
          { \ID_2 \\[-0.3em] & \qquad \cdots \\[-0.3em] &\!\! }
          { \unknownid}}
      \end{alignedat}
    \end{dcases}
  \end{alignat*}
  \underline{\textbf{Case} $\ai = {\npuai{2}{\ID}{j}}$.}
  \begin{alignat*}{2}
    \accept_\tau^\ID  &\;\equiv\;\;&&
    \eq{g(\inframe_\tau)}
    {\mac{\spair
        {\instate_\tau(\bauth_\ue^\ID)}
        {\instate_\tau(\sqn_\ue^\ID)}}
      {\mkey^\ID}{2}}\\
    t_\tau  &\;\equiv\;\;&&
    \ite{\accept_\tau^\ID}
    {\textsf{ok}}
    {\textsf{error}}\\
    \upstate_\tau &\;\equiv\;&&
    \eauth_\ue^\ID \mapsto
    \ite{\accept_\tau^\ID}
    {\instate_\tau(\bauth_\ue^\ID)}{\fail}
  \end{alignat*}
  \caption{The Symbolic Terms and State Updates for the $\supi$ Sub-Protocol.}
  \label{fig:protocol-term-supi}
\end{figure}

\begin{figure}[tbp]
  \footnotesize
  \underline{\textbf{Case} $\ai = {\newsession_\ID(j)}$.}
  $\upstate_\tau \;\equiv\; \success_\ue^\ID \mapsto \false$
  \vspace{0.3em}

  \underline{\textbf{Case} $\ai = {\cuai_\ID(j,0)}$.}
  \begin{alignat*}{2}
    t_\tau &\;\equiv\;\;&&
    \ite{\instate_\tau(\success_\ue^\ID)}
    {\instate_\tau(\suci_\ue^\ID)}
    {\nosuci}\\
    \upstate_\tau &\;\equiv\;&&
    \left\{
      \begin{array}{lcl}
        \success_\ue^\ID \mapsto \false &\quad&
        \eauth_\ue^\ID \mapsto \fail\\
        \uetsuccess^\ID \mapsto \instate_\tau(\success_\ue^\ID) &\quad&
        \bauth_\ue^\ID \mapsto \fail
      \end{array}
    \right.
  \end{alignat*}

  \underline{\textbf{Case} $\ai = {\cnai(j,0)}$.}
  Let $t^{\ID_i}_{\oplus} \equiv \instate_\tau(\sqn_\hn^{\ID_i}) \oplus \ow{\nonce^j}{\key^{\ID_i}} $, then:
  \begin{alignat*}{2}
    \msg_\tau^{\ID_i} &\;\equiv\;\;&&
    \striplet{\nonce^j}
    {t^{\ID_i}_{\oplus}}
    {\mac{\striplet
        {\nonce^j}
        {\instate_\tau(\sqn_\hn^{\ID_i})}
        {\instate_\tau(\suci_\hn^{\ID_i})}}
      {\mkey^{\ID_i}}{3}}\\
    \accept_\tau^{\ID_i} &\;\equiv\;\;&&
    \eq{\instate_\tau(\suci_\hn^{\ID_i})}{g(\inframe_\tau)}
    \wedge
    \neg \eq{\instate_\tau(\suci_\hn^{\ID_i})}{\unset}\\
    t_\tau &\;\equiv\;\; &&
    \begin{alignedat}[t]{2}
      &\ite{\accept_\tau^{\ID_1}}
      {\msg_\tau^{\ID_1}\\
        & \!\!}
      {\ite{\accept_\tau^{\ID_2}}
        {\msg_\tau^{\ID_2} \\[-0.3em]
          & \qquad \cdots \\[-0.3em] &\!\! }
        {\msg_{\tau}^{\IDdum}}}
    \end{alignedat}\\
    \upstate_\tau &\;\equiv&&
    \begin{dcases}
      \suci_\hn^{\ID_i} \mapsto
      \begin{alignedat}[t]{2}
        &\ite{\accept_\tau^{\ID_i}}
        {\unset}
        {\instate_\tau(\suci_\hn^{\ID_i})}
      \end{alignedat}\\
      \tsuccess_\hn^{\ID_i} \mapsto
      \begin{alignedat}[t]{2}
        &\ite{\accept_\tau^{\ID_i}}
        {\nonce^j}
        {\instate_\tau(\tsuccess_\hn^{\ID_i})}
      \end{alignedat}\\[0.3em]
      \bauth_\hn^{j} \mapsto
      \begin{alignedat}[t]{2}
        &\ite{\accept_\tau^{\ID_1}}
        {\ID_1\\ & \!\!}
        {\ite{\accept_\tau^{\ID_2}}
          {\ID_2 \\[-0.3em] & \qquad \cdots \\[-0.3em] &\!\! }
          {\unknownid}}
      \end{alignedat}
    \end{dcases}
  \end{alignat*}

  \underline{\textbf{Case} $\ai = {\cuai_\ID(j,1)}$.}
  Let $t_\sqn \equiv
  \pi_2(g(\inframe_\tau)) \xor \ow{\pi_1(g(\inframe_\tau))}{\key^\ID}$, then:
  \begin{alignat*}{2}
    \accept_\tau^\ID &\; \equiv \;\;&&
    \begin{alignedat}[t]{1}
      &\eq
      {\pi_3(g(\inframe_\tau))}
      {\mac{\striplet
          {\pi_1(g(\inframe_\tau))}
          {t_{\sqn}}
          {\instate_\tau(\suci_\ue^\ID)}}
        {\mkey^\ID}{3}}\\
      &\wedge \;
      \instate_\tau(\uetsuccess^{\ID}) \;\wedge \;
      \range{\instate_\tau(\sqn_\ue^\ID)}
      {t_\sqn}
    \end{alignedat}\\
    t_\tau &\; \equiv \;\;&&
    \ite{\accept_\tau^\ID}
    {\mac{\pi_1(g(\inframe_\tau))}{\mkey^\ID}{4}}
    {\textsf{error}}\\
    \upstate_\tau &\;\equiv&&\!\!\!\!
    \begin{dcases}
      \bauth_\ue^\ID,\eauth_\ue^\ID \mapsto
      \begin{alignedat}[t]{2}
        \ite{\accept_\tau^\ID}
        {&\pi_1(g(\inframe_\tau))}{ \fail}
      \end{alignedat}\\
      \sqn_\ue^\ID \mapsto
      \begin{alignedat}[t]{2}
        \ite{\accept_\tau^\ID}
        {&\sqnsuc(\instate_\tau(\sqn_\ue^\ID))}{ \instate_\tau(\sqn_\ue^\ID)}
      \end{alignedat}
    \end{dcases}
  \end{alignat*}

  \underline{\textbf{Case} $\ai = {\cnai(j,1)}$.}
  \begin{alignat*}{2}
    \accept_\tau^{\ID_i}
    &\; \equiv \;\;&&
    \eq{g(\inframe_\tau)}{\mac{\nonce^j}{\mkey^{\ID_i}}{4}}
    \wedge
    \eq{\instate_\tau(\bauth_\hn^j)}{\ID_i}\\
    \incaccept_\tau^{\ID_i}
    &\; \equiv \;\;&&
    \accept_\tau^{\ID_i} \wedge
    \eq{\instate_\tau(\tsuccess_\hn^{\ID_i})}{\nonce^j}\\
    t_\tau &\; \equiv \;\;&&
    \ite{\textstyle\bigvee_{i} \accept_\tau^{\ID_i}}{\textsf{ok}}{\textsf{error}}
  \end{alignat*}
  \vspace{-2em}
  \begin{alignat*}{2}
    \upstate_\tau &\;\equiv&&
    \begin{dcases}
      \sqn_\hn^{\ID_i} \mapsto
      \begin{alignedat}[t]{2}
        \ite{\incaccept_\tau^{\ID_i}&}
        {\sqnsuc(\instate_\tau(\sqn_\hn^{\ID_i}))\\ &}
        {\instate_\tau(\sqn_\hn^{\ID_i})}
      \end{alignedat}\\
      \suci_\hn^{\ID_i} \mapsto
      \begin{alignedat}[t]{2}
        \ite{\incaccept_\tau^{\ID_i} &}
        {\suci^j }
        {\instate_\tau(\suci_\hn^{\ID_i})}
      \end{alignedat}\\[0em]
      \eauth_\hn^{j} \mapsto
      \begin{alignedat}[t]{2}
        &\ite{\accept_\tau^{\ID_1}}
        {\ID_1\\ & \!\!}
        {\ite{\accept_\tau^{\ID_2}}
          {\ID_2 \\[-0.3em] & \qquad \cdots \\[-0.3em] &\!\! }
          {\unknownid}}
      \end{alignedat}
    \end{dcases}
  \end{alignat*}

  \underline{Case $\ai = {\fnai(j)}$.}
  \begin{alignat*}{2}
    \msg_\tau^{\ID_i} &\;\equiv\;\;&&
    \spair
    {\suci^j \oplus \row{\nonce^j}{\key^{\ID_i}}}
    {\mac{\spair
        {\suci^j}
        {\nonce^j}}
      {\mkey^{\ID_i}}{5}}\\
    t_\tau &\; \equiv \;\;&&
    \begin{alignedat}[t]{2}
      &\ite{\eq{\instate_\tau(\eauth_\hn^j)}{\ID_1}}
      {\msg_\tau^{\ID_1}\\ & \!\!}
      {\ite{\eq{\instate_\tau(\eauth_\hn^j)}{\ID_2}}
        {\msg_\tau^{\ID_2} \\[-0.3em] & \qquad \cdots \\[-0.3em] &\!\! }
        {\unknownid}}
    \end{alignedat}
  \end{alignat*}

  \underline{\textbf{Case} $\ai = {\fuai_\ID(j)}$.}
  Let $t_\guti \equiv \pi_1(g(\inframe_\tau)) \xor \row{\instate_\tau(\eauth_\ue^{\ID})}{\key^\ID}$, then:
  \begin{alignat*}{2}
    \accept_\tau^\ID  &\; \equiv \;\;&&
    \begin{alignedat}[t]{1}
      &\eq
      {\pi_2(g(\inframe_\tau))}
      {\mac
        {\spair
          {t_{\guti}}
          {\instate_\tau(\eauth_\ue^{\ID})}}
        {\mkey^\ID}{5}}\\
      &\wedge\;
      \neg \eq{\instate_\tau(\eauth_\ue^{\ID})}{\fail}
      \;\wedge\;
      \neg \eq{\instate_\tau(\eauth_\ue^{\ID})}{\bot}
    \end{alignedat}\\
    t_\tau &\; \equiv \;\;&&
    \begin{alignedat}[t]{2}
      &\ite
      {\accept_\tau^\ID}
      {\textsf{ok}}
      { \textsf{error}}
    \end{alignedat}\\
    \upstate_\tau &\;\equiv\;&&
    \begin{dcases}
      \success_\ue^{\ID} \mapsto \accept_\tau^\ID\\
      \suci_\ue^{\ID} \mapsto
      \begin{alignedat}[t]{2}
        &\ite
        {\accept_\tau^\ID}
        {t_\guti}
        {\unset}
      \end{alignedat}
    \end{dcases}
  \end{alignat*}
  \caption{The Symbolic Terms and State Updates for $\ns_\ID(j)$ and the $\guti$ and $\refresh$ Sub-Protocols.}
  \label{fig:protocol-guti-refresh}
\end{figure}


%% file: axioms-body.tex
\subsection{Axioms}
Using Proposition~\ref{prop:main-unlink-bc}, we know that to prove Theorem~\ref{thm:main-unlink} we need to derive $\cframe_{\tau_l} \sim \cframe_{\tau_r}$, for every $(\tau_l,\tau_r) \in \runlink$, using a set of inference rules $\axioms$. Moreover, we need the axioms $\axioms$ to be valid in any computational model where the asymmetric encryption $\enc{\_}{\_}{\_}$ is \textsc{ind-cca1} secure and $\owsym$ and $\rowsym$ (resp. $\macsym^1$--$\,\macsym^5$) satisfy jointly the $\prfass$ assumption.

Remark that the $\faka$ protocol described in Section~\ref{section:faka-description} is under-specified. E.g., we never specified how the $\pair{\_}{\_}$ function should be implemented. Instead of giving a complex specification of the protocol, we are going to put requirements on $\faka$ implementations through the set of axioms $\axioms$. Then, if we can derive $\cframe_{\tau_l} \sim \cframe_{\tau_r}$ using $\axioms$ for every $(\tau_l,\tau_r) \in \runlink$, we know that any implementation of $\faka$ satisfying the axioms $\axioms$ is secure.

Our axioms are of two kinds. First, we have \emph{structural axioms}, which are properties that are valid in any computational model. For example, we have axioms stating that $\sim$ is an equivalence relation. Second, we have \emph{implementation axioms}, which reflect implementation assumptions on the protocol functions. For example, we can declare that different identity symbols are never equal by having an axiom $\eq{\ID_1}{\ID_2} \sim \false$ for every $\ID_1 \not \equiv \ID_2$. 
For space reasons, we only describe a few of them here (the full set of axioms $\axioms$ is given in %
\iffull
Appendix~\ref{section:app-axioms}). 
\fi

\paragraph{Equality Axioms}
If $\eq{s}{t} \sim \true$ holds in any computational model then we know that the interpretations of $s$ and $t$ are always equal except for a negligible number of samplings. Let $s \peq t$ be a shorthand for $\eq{s}{t} \sim \true$. We use $\peq$ to specify functional correctness properties of the protocol function symbols. For example, the following rules state that the $i$-th projection of a pair is the $i$-th element of the pair, and that the decryption with the correct key of a cipher-text is equal to the message in plain-text:
\begin{mathpar}
  \small
  \unary{\pi_i(\pair{x_1}{x_2}) \peq x_i}
  \text{ for } i \in \{1,2\}
  
  \unary{\dec(\enc{x}{\pk(y)}{z},\sk(y)) \peq x}
\end{mathpar}

\paragraph{Structural Axioms}
Structural axioms are axioms which are valid in any computational model, e.g.:
\begin{mathpar}
  \infer[\fa]{
    f(\vec{u}_1),\vec{v}_1 \sim f(\vec{u}_2), \vec{v}_2
  }{
    \vec{u}_1,\vec{v}_1 \sim \vec{u}_2, \vec{v}_2
  }

  \infer[R]{
    \vec u, s \sim \vec v
  }{
    \vec u, t \sim \vec v
    \;&\;
    s \peq t
  }    
\end{mathpar}
The axiom $\fa$ states that to show that two function applications are indistinguishable, it is sufficient to show that their arguments are indistinguishable. The axiom $R$ states that if $s \peq t$ holds then we can safely replace $s$ by $t$.

\paragraph{Cryptographic Assumptions}
We now explain how cryptographic assumptions are translated into axioms. We illustrate this on the unforgeability property of the functions $\macsym^1$--$\,\macsym^5$. Recall that $\tue_{\ID}$ uses the same secret key $\mkey^\ID$ for these five functions. Therefore, instead of the standard $\prfass$ assumption, we assume that these functions are \emph{jointly} $\prfass$\ch{, i.e.}{. Basically, we assume that} $\macsym^1$--$\,\macsym^5$ are \emph{simultaneously} computationally indistinguishable from random functions.
\ch{}{While this is a non-usual assumption, it is simple to build a set of functions $H_1,\dots,H_l$ which are jointly $\prfass$ from a single $\prfass$ $H$. For example, let $\ttag_1,\dots,\ttag_l$ be non-ambiguous tags, and let $H_i(m,\key) = H(\ttag_i(m),\key)$. Then, $H_1,\dots,H_l$ are jointly $\prfass$ whenever $H$ is a $\prfass$ (see %
  \iffull
  Appendix~\ref{app:subsection-joints-prf}).
  \fi}

It is well-known that if $H$ is a $\prfass$ then $H$ is unforgeable against an adversary with oracle access to $H(\cdot,\mkey)$. Similarly, we can show that if $H,H_1,\dots,H_l$ are jointly $\prfass$, then no adversary can forge a mac of $H(\cdot,\mkey)$, even if it has oracle access to $H(\cdot,\mkey),H_1(\cdot,\mkey),\dots,H_l(\cdot,\mkey)$. We translate this property as follows: let $s,m$ be ground terms where $\mkey$ appears only in subterms of the form $\mac{\_}{\mkey}{\,\_}$, then for every $1 \le j \le 5$, if $S$ is the set of subterms of $s,m$ of the form $\mac{\_}{\mkey}{j}$ then we have an instance of $\textsc{euf-mac}^j$:
\[
  \unary{s = \mac{m}{\mkey}{j} \ra
    \bigvee_{u \in S} s = \mac{u}{\mkey}{j}}
  \tag{$\textsc{euf-mac}^j$}
\]
where $u = v$ denotes the term $\eq{u}{v}$. \ch{Basically, if $s$ is a valid $\macsym$ then $s$ must have been honestly generated.}{} Similarly, we can build a set of axioms reflecting the fact that some functions are jointly collision-resistant. Indeed, if $H,H_1,\dots,H_l$ are jointly $\prfass$, then no adversary can build a collision for $H(\cdot,\mkey)$, even if it has oracle access to $H(\cdot,\mkey),H_1(\cdot,\mkey),\dots,H_l(\cdot,\mkey)$. This translates as follows: let $m_1,m_2$ be ground terms, if $\mkey$ appears only in subterms of the form $\mac{\_}{\mkey}{\,\_}$ then we have an instance of $\textsc{cr}^j$:
\[
  \begin{array}[c]{c}
    \unary{\mac{m_1}{\mkey}{j} = \mac{m_2}{\mkey}{j}
      \ra m_1 = m_2}
  \end{array}
  \tag{$\textsc{cr}^j$}
\]
These axioms are sound (the proof is given in~%
\iffull
Appendix~\ref{app:subsection-joints-prf}).
\fi
\begin{proposition}
  \label{prop:euf-mac-valid}
  For every $1 \le j \le 5$, the $\textsc{euf-mac}^j$ and $\textsc{cr}^j$ axioms are valid in any computational model where the $(\macsym^i)_i$ functions are interpreted as jointly $\prfass$ functions.
\end{proposition}


%% file: security.tex
\section{Security Proofs}
\label{section:security-proofs-body}
We now state the authentication and $\sigma_{\sunlink}$-unlinkability lemmas. For space reasons, we only sketch the proofs (the full proofs are given in %
\iffull
Appendix~\ref{section:acc-charac-app-first} and \ref{section:app-unlink}).
\fi

\subsection{Mutual Authentication of the $\faka$ Protocol}
Authentication is modeled by a correspondence property~\cite{287633} of the form ``in any execution, if event $A$ occurs, then event $B$ occurred''. This can be translated in the BC logic.

\paragraph{Authentication of the User by the Network}
$\faka$ guarantees authentication of the user by the network if in any execution, if $\thn(j)$ believes it authenticated $\tue_{\ID}$, then $\tue_{\ID}$ stated earlier that it had initiated the protocol with~$\thn(j)$.

We recall that $\eauth_\hn^j$ stores the identity of the $\tue$ authenticated by $\thn(j)$, and that $\tue_{\ID}$ stores in $\bauth_\ue^\ID$ the random challenge it received. Moreover, the session $\thn(j)$ is uniquely identified by its random challenge $\nonce^j$. Therefore, authentication of the user by the network is modeled by stating that, for any symbolic trace $\tau \in \support(\runlink)$
, if $\instate_\tau(\eauth_\hn^j) = \ID$ then there exists some prefix $\tau'$ of $\tau$ such that $\instate_{\tau'}(\bauth_\ue^\ID) = \nonce^j$. Let $\popreleq$ be the prefix ordering on symbolic traces, then:
\begin{lemma}
  \label{lem:auth-serv-net-body}
  For every $\tau \in \support(\runlink)$, $\ID \in \iddom$ and $j \in \mathbb{N}$, there is derivation using $\axioms$ of:
  \[
    \textstyle \instate_\tau(\eauth_\hn^j) = \ID \;\ra\;
    \bigvee_{\tau' \popreleq \tau}\;
    \instate_{\tau'}(\bauth_\ue^\ID) = \nonce^j
  \]
\end{lemma}
The key ingredients to show this lemma are \emph{necessary conditions} for a message to be accepted by the network. Basically, a message can be accepted only if it was honestly generated by a subscriber. These necessary conditions rely on the unforgeability and collision-resistance of $(\macsym^j)_{1 \le j \le 5}$.

\paragraph{Necessary Acceptance Conditions}
Using the $\textsc{euf-mac}^j$ and $\textsc{cr}^j$ axioms, we can find necessary conditions for a message to be accepted by a user. We illustrate this on the $\thn$'s second message in the $\supi$ sub-protocol. We depict a part of the execution between session $\tue_{\ID}(i)$ and session $\thn(j)$ below:
\begin{center}
  \begin{tikzpicture}[>=latex,scale=0.94,transform shape]
    \tikzset{every node/.style={font=\footnotesize}};
    \tikzset{en/.style={minimum height=0.2cm,minimum width=0.15cm,fill=black}};
    \tikzset{mb/.style={solid,thick,draw=gray!90}};
    \tikzset{mp/.style={inner sep=0,outer sep=0,draw=none}};
    \tikzset{sb/.style={draw,fill=white,double,align=left}};

    \path (0,0)
    node[above,yshift=0.1cm,anchor=south] (ue) {$\tue_{\ID}(i)$};

    \path[name path=aline] (ue) -- ++(0,-1.6);

    \path[draw,very thick]
    (ue) -- ++(0,-1.6) node(xend){};

    \path (6.8,0)
    node[above,yshift=0.1cm,anchor=south] (hn) {$\thn(j)$};

    \path[name path=bline] (hn) -- ++(0,-1.7);

    \path[name path=xendp] (xend) -- ++(8,0);
    \path[name intersections={of=xendp and bline}]
    (intersection-1) node[mp] (xendr){};
    \path[draw,very thick]
    (hn) -- (xendr) node{};

    \path (hn) -- ++(0,-0.5)
    node[right]{$\pnai(j,0)$}
    [draw,->] -- ++(-6.8,0)
    node[midway,above]{$\nonce^j$}
    node[mp] (ue1) {};

    \path (ue1) -- ++(0,-0.9)
    node[left]{$\npuai{1}{\ID}{i}$}
    [draw,->] -- ++(6.8,0)
    node[midway,above]{$
      \lpair
      {\enc{\pair{\agent{\ID}}{\sqn_\ue}}
        {\pk_\hn}{\enonce^i}}
      {\mac{\spair
          {\enc{\pair{\agent{\ID}}{\sqn_\ue}}
            {\pk_\hn}{\enonce^i}}
          {\nonce^j}}
        {\mkey}{1}}$}
    node[right]{$\pnai(j,1)$};
  \end{tikzpicture}
\end{center}
We then prove that if a message is accepted by $\thn(j)$ as coming from $\tue_{\ID}$, then the first component of this message must have been honestly generated by a session of $\tue_{\ID}$. Moreover, we know that this session received the challenge~$\nonce^j$.
\begin{lemma}
  \label{lem:acc-cond-body}
  Let $\ID \in \iddom$ and $\tau \in \support(\runlink)$ be a trace ending with $\pnai(j,1)$. There is a derivation using $\axioms$~of:
  \[
    \accept_\tau^{\ID} \ra
    \quad\;\bigvee_{
      \mathclap{\taut = \_,\npuai{1}{\ID}{\_} \popreleq \tau}}\quad\;
    \left(
      \pi_1(g(\inframe_\tau)) =
      t^{\textsf{enc}}_\taut
      \wedge
      g(\inframe_{\taut}) = \nonce^j
    \right)
  \]
\end{lemma}
\begin{proof}[Proof sktech]
  Let $t_{\textsf{dec}}$ be the term $\dec(\pi_1(g(\inframe_\tau)),\sk_\hn)$. Then $\thn(j)$ accepts the last message iff the following test succeeds:
\[
  \dashuline{
    \pi_2(g(\inframe_\tau))
    =
    \mac{\spair
      {\pi_1(g(\inframe_\tau))}
      {\nonce^j}}
    {\mkey^{\ID}}{1}}
  \wedge
  \pi_1(t_{\textsf{dec}}) = \ID
\]
By applying $\textsc{euf-mac}^1$ to the underlined part above, we know that if the test holds then $\pi_2(g(\inframe_\tau))$ is equal to one of the honest $\macsym_{\mkey^\ID}^1$ subterms of $\pi_2(g(\inframe(\tau)))$, which are the terms:
\begin{gather}
  \left(
    \mac{\spair
      {t^{\textsf{enc}}_\taut}
      {g(\inframe_{\taut})}}
    {\mkey^\ID}{1}
  \right)_{\substack{\taut = \_, \npuai{1}{{\ID}}{\_} \popre \tau}}\\
  \left(
    {\mac
      {\spair
        {\pi_1(g(\inframe_{\taut}))}
        {\nonce^{j_1}}}
      {\mkey^\ID}{1}}
  \right)_{\taut = \_, \pnai(j_1,1) \popre \tau}
  \label{eq:case-1-body}
\end{gather}
Where $\popre$ is the strict version of $\popreleq$. We know that $\pnai(j,1)$ cannot appear twice in $\tau$. Hence for every $\taut = \_, \pnai(j_1,1) \popre \tau$, we know that $j_1 \ne j$. Using the fact that two distinct nonces are never equal except for a negligible number of samplings, we can derive that $\eq{\nonce^{j_1}}{\nonce^j} = \false$. Using an axiom stating that the pair is injective and the $\textsc{cr}^1$ axiom, we can show that $\pi_2(g(\inframe_\tau))$ cannot by equal to one of the terms in~\eqref{eq:case-1-body}.

Finally, for every $\taut = \_, \npuai{1}{{\ID}}{\_} \popre \tau$, using the $\textsc{cr}^1$ and the pair injectivity axioms we can derive that:
\begin{multline*}
  \small
  \mac{\spair
    {\pi_1(g(\inframe_\tau))}
    {\nonce^j}}
  {\mkey^{\ID}}{1}
  =
  \mac{\spair
    {t^{\textsf{enc}}_\taut}
    {g(\inframe_{\taut})}}
  {\mkey^\ID}{1}\\
  \ra
  \pi_1(g(\inframe_\tau))
  =
  t^{\textsf{enc}}_\taut
  \wedge
  \nonce^j
  =
  g(\inframe_{\taut})
  \qedhere
\end{multline*}
\end{proof}
We prove a similar lemma for $\cnai(j,1)$. The proof of Lemma~\ref{lem:auth-serv-net-body} is straightforward using these two properties.

\paragraph{Authentication of the Network by the User}
The $\faka$ protocol also provides authentication of the network by the user. That is, in any execution, if $\tue_{\ID}$ believes it authenticated session $\thn(j)$ then $\thn(j)$ stated that it had initiated the protocol with $\tue_{\ID}$. Formally:
\begin{lemma}
  \label{lem:auth-net-body}
  For every $\tau \in \support(\runlink)$, $\ID \in \iddom$ and $j \in \mathbb{N}$, there is derivation using $\axioms$ of:
  \[
    \textstyle
    \instate_\tau(\eauth_\ue^\ID) = \nonce^j
    \;\ra\;
    \bigvee_{\tau' \popreleq \tau}\;
    \instate_{\tau'}(\bauth_\hn^j) = \ID
  \]
\end{lemma}
This is shown using the same techniques than for Lemma~\ref{lem:auth-serv-net-body}.

\subsection{\texorpdfstring{$\sigma$}{Sigma}-Unlinkability of the $\faka$ Protocol}

Lemma~\ref{lem:acc-cond-body} gives a necessary condition for a message to be accepted by $\pnai(j,1)$ as coming from $\ID$. We can actually go further, and show that a message is accepted by $\pnai(j,1)$ as coming from $\ID$ \emph{if and only if} it was honestly generated by a session of $\tue_{\ID}$ which received the challenge $\nonce^j$.
\begin{lemma}
  \label{lem:acc-equ-cond-body}
  Let $\ID \in \iddom$ and $\tau \in \support(\runlink)$ be a trace ending with $\pnai(j,1)$. There is a derivation using $\axioms$~of:
  \[
    \accept_\tau^{\ID} \lra
    \quad\;\bigvee_{
      \mathclap{\taut = \_,\npuai{1}{\ID}{\_} \popreleq \tau}}\quad\;
    \left(\,
      g(\inframe_\tau) =
      t_\taut
      \wedge
      g(\inframe_{\taut}) = \nonce^j
    \,\right)
  \]
\end{lemma}
We prove similar lemmas for most actions of the $\faka$ protocol. Basically, these lemmas state that a message is accepted if and only if it is part of an honest execution of the protocol between $\tue_{\ID}$ and $\thn$. This allow us to replace each acceptance conditional $\accept_\tau^\ID$ by a disjunction over all possible honest partial transcripts of the protocol.

We now state the $\sigma_\sunlink$-unlinkability lemma:
\begin{lemma}
  \label{lem:main-unlink-proof}
  For every $(\tau_l,\tau_r) \in \runlink$, there is a derivation using $\axioms$ of the formula $\cframe_{\tau_l}  \sim \cframe_{\tau_r}$.
\end{lemma}
The full proof is long and technical. It is shown by induction over $\tau$. Let $(\tau_l,\tau_r) \in \runlink$, we assume by induction that there is a derivation of $\inframe_{\tau_l} \sim \inframe_{\tau_r}$. We want to build a derivation of $\inframe_{\tau_l},t_{\tau_l} \sim \inframe_{\tau_r},t_{\tau_r}$ using the inference rules in $\axioms$.

First, we rewrite $t_{\tau_l}$ using the acceptance characterization lemmas such as Lemma~\ref{lem:acc-equ-cond-body}. This replaces each $\accept_{\tau_l}^\ID$ by a case disjunction over all honest executions \emph{on the left side}. Similarly, we rewrite $t_{\tau_r}$ as a case disjunction over honest executions \emph{on the right side}. Our goal is then to find a matching between left and right transcripts such that matched transcripts are indistinguishable. If a left and right transcript correspond to the same trace of oracle calls, this is easy. But since the left and right traces of oracle calls may differ, this is not always possible. E.g., some left transcript may not have a corresponding right transcript. When this happens, we have two possibilities: instead of a one-to-one match we build a many-to-one match, e.g. matching a left transcript to several right transcripts; or we show that some transcripts always result in a failure of the protocol. Showing the latter is complicated, as it requires to precisely track the possible values of $\sqn_\ue^\ID$ and $\sqn_\hn^\ID$ across multiple sessions of the protocol to prove that some transcripts always yield a de-synchronization between $\tue_{\ID}$ and $\thn$.


%% file: conclusion.tex
\section{Conclusion}

We studied the privacy provided by the $\fiveaka$ authentication protocol. While this protocol is not vulnerable to $\imsi$ catchers, we showed that several privacy attacks from the literature apply to it. We also discovered a novel desynchronization attack against $\privaka$, a
\ch
{modified version of $\aka$,}
{variant of $\aka$,}
even though it had been claimed secure.

We then proposed the $\faka$ protocol. This is a fixed version of $\fiveaka$, which is both efficient and has improved privacy guarantees. To study $\faka$'s privacy, we defined the $\sigma$-unlinkability property. This is a new parametric privacy property, which requires the prover to establish privacy only for a subset of the standard unlinkability game scenarios. Finally, we formally proved that $\faka$ provides mutual authentication and $\sigma_{\sunlink}$-unlinkability for any number of agents and sessions. Our proof is carried out in the Bana-Comon model, which is well-suited to the formal analysis of stateful protocols.


%% file: axioms.tex
\section{Axioms}
\label{section:app-axioms}
\FloatBarrier

In this section, we define the set of axioms $\axioms$. We split our set of axioms in three parts, $\axioms = \axioms_{\textsf{struct}} \cup \axioms_{\textsf{impl}} \cup \axioms_{\textsf{crypto}}$, where $\axioms_{\textsf{struct}}$ is the set of structural axioms, $\axioms_{\textsf{impl}}$ is the set of implementation axioms and $\axioms_{\textsf{crypto}}$ is the set of cryptographic~axioms.

\paragraph*{Definitions}
We give some definitions used to define the cryptographic axioms.

\begin{definition}
  For any subset $\mathcal{S}$ of $\sig,\Nonce$ and $\mathcal{X}$, we let $\mathcal{T}(\mathcal{S})$ be the set of terms built upon~$\mathcal{S}$.
\end{definition}

\begin{definition}
  A \emph{position} is a word in $\mathbb{N}^*$. The value of a term $t$ at a position $p$, denoted by $(t)_{|p}$, is the partial function defined inductively as follows:
  \[
    \begin{array}{lcl}
      (t)_{|\epsilon} &=& t\\
      (f(u_0,\dots,u_{n-1}))_{|i.p} &=&
      \begin{cases}
        (u_i)_{|p} & \text{ if } i < n\\
        \text{undefined} & \text{ otherwise}
      \end{cases}
    \end{array}
  \]
  We say that a position in \emph{valid} is $t$ if $(t)_{|p}$ is defined. The set of positions $\pos(t)$ of a term is the set of positions which are valid in $t$. Moreover, given two position $p,p'$, we have $p \le p'$ if and only if $p$ is a prefix of $p'$.
\end{definition}


\begin{definition}
  A context $D[]_{\vec x}$ (sometimes written $D$ when there is no confusion) is a term in $\mathcal{T}(\sig,\Nonce,\{[]_y \mid y \in \vec x\})$ where $\vec x$ are distinct special variables called holes.

  For all contexts $D[]_{\vec x},C_0,\dots,C_{n-1}$ with $|\vec x| = n$, we let $D[(C_i)_{i<n}]$ be the context $D[]_{\vec x}$ in which we substitute, for every $0 \le i < n$, all occurrences of the hole $[]_{x_i}$ by $C_i$.

  A one-holed context is a context with one hole (in which case we write $D[]$ where $[]$ is the only variable).
\end{definition}

\begin{definition}
 Given a term $t$, we let $\st(t)$ be the set of subterms of $t$. We extend this to sequences of terms by having $\st(u_1,\dots,u_n) = \st(u_1) \cup \dots \cup \st(u_n)$.
\end{definition}

\begin{definition}
  For every terms $b,t$, we let $\cond{b}{t}$ be the term $\ite{b}{t}{\bot}$.
\end{definition}

\begin{definition}
  Let $s,\vec u$ be ground terms and $C_{\vec x, \cdot}$ be a context with one distinguished hole variable $\cdot$, and we require that the hole variable $\cdot$ appears exactly once in $C_{\vec x, \cdot}$. Then we let $s \tpos_{C_{\vec x, \cdot}} \vec u$ holds whenever $s$ appears in $\vec u$ only in subterms of the form $C[\vec w,s]$. Formally:
  \[
    \forall u \in \vec u, \forall p \in \pos(u), u_{|p} \equiv s \ra
    \exists \vec w \in \mathcal{T}(\sig,\Nonce),
    \exists q \in \pos(u) \text{ s.t. }
    q \le p \wedge
    u_{|q} \equiv C[\vec w,s]
  \]
  Given $n$ contexts $C_1,\dots,C_n$, we let $s \tpos_{C_1,\dots,C_n} \vec u$ if and only if for all $1 \le i \le n$, $s \tpos_{C_i} \vec u$.
\end{definition}

\begin{example}
  For example, $\nonce \tpos_{\pk(\cdot),\sk(\cdot)} \vec u$ states that the nonce $\nonce$ appears only in terms of the form $\pk(\nonce)$ or $\sk(\nonce)$ in $\vec u$.

  Similarly, $\sk(\nonce) \tpos_{\dec(\_,\cdot)} \vec u$ states that the secret key $\sk(\nonce)$ appears only in decryption position in $\vec u$.
\end{example}

\subsection{$\ccao$ Axioms}
\label{subsection:app-ccao}

\paragraph{The $\ccao_s$ Axioms}
To prove that the $\faka$ protocol is $\sigma_{\sunlink}$-unlinkable, we need the encryption scheme to be $\textsc{ind-cca1}$ secure. We define first set of axioms $\ccao_s$:
\begin{definition}
  We let $\ccao_s$ be the set of axioms:
  \[
    \begin{array}[c]{c}
      \infer[\ccao_s]
      {\vec u, \enc{s}{\pk(\nonce)}{\enonce}
        \sim
        \vec u, \enc{t}{\pk(\nonce)}{\enonce}}
      { \length(s) \peq \length(t)}
    \end{array}
    \qquad \qquad
    \text{ when }
    \begin{dcases}
      \fresh{\enonce}{\vec u,s,t}\\
      \nonce \tpos_{\pk(\cdot),\sk(\cdot)} \vec u,s,t
      \;\wedge\; \sk(\nonce) \tpos_{\dec(\_,\cdot)} \vec u,s,t
    \end{dcases}
  \]
\end{definition}
This set of axioms $\ccao_s$ is very similar to the one used in \cite{Bana:2014:CCS:2660267.2660276}. The only difference is that in \cite{Bana:2014:CCS:2660267.2660276}, the length equality requirement is not a premise of the axiom. Instead, if the length are not equal they return a error message. We found our version of the axiom simpler to use.

We have the following soundness property:
\begin{proposition}
  \label{prop:ccaos-valid}
  The $\ccao_s$ axioms are valid in any computational model where $(\enc{\_}{\_}{\_},\dec(\_,\_),\pk(\_),\sk(\_))$ is interpreted as a $\textsc{ind-cca1}$ secure encryption scheme.
\end{proposition}

\begin{proof}
  The proof is by contradiction, and is sketched below:

  We assume that there is a computational model $\cmodel$ where the encryption scheme is $\textsc{ind-cca1}$ secure, and such that there is an instance $\vec u, \enc{s}{\pk(\nonce)}{\enonce} \sim \vec v, \enc{t}{\pk(\nonce)}{\enonce}$  of the axioms $\ccao_s$ which is not valid. We deduce that there exists an attacker $\mathcal{A}$ that can distinguish between the left and right terms, i.e. the following quantity is non-negligible:
  \[
    \left|
      \Pr\left(
        \vec w \rla \sem{\vec u, \enc{s}{\pk(\nonce)}{\enonce}}_\cmodel :
        \mathcal{A}(1^\eta,\vec w) = 1
      \right)
      -
      \Pr\left(
        \vec w \rla \sem{\vec u, \enc{t}{\pk(\nonce)}{\enonce}}_\cmodel :
        \mathcal{A}(1^\eta,\vec w) = 1
      \right)
    \right|
    \numberthis\label{eq:qeeqwiofhaifna}
  \]
  Where $\rla$ denotes a uniform random sampling. Using $\mathcal{A}$, we can build an adversary $\mathcal{B}$ with a non-negligible advantage against the $\textsc{ind-cca1}$ game. First, $\mathcal{B}$ samples a vector of bit-strings $\vec u_s,s_s,t_s$ from $\sem{\vec u,s,t}_\cmodel$, querying the decryption oracle whenever $\mathcal{B}$ needs to compute a subterm of the from $\dec(\_,\sk(\nonce))$. Remark that the syntactic side-conditions:
  \begin{mathpar}
    \nonce \tpos_{\pk(\cdot),\sk(\cdot)} \vec u,s,t

    \sk(\nonce) \tpos_{\dec(\_,\cdot)} \vec u,s,t
  \end{mathpar}
  guarantee that this is always possible. Afterward, $\mathcal{B}$ queries the left-or-right oracle with $(s_s,t_s)$ to get a value $a$. Here, we need the side-condition $\fresh{\enonce}{\vec u,s,t}$ to guarantee that the random value $\enonce$ has not been sampled by $\mathcal{B}$. Indeed, the value $\enonce$ is sampled by the challenger, and is not available to $\mathcal{B}$. If the challenger internal bit $b$ is $0$ then $\vec u_s,a$ has been sampled from $\sem{\vec u, \enc{s}{\pk(\nonce)}{\enonce}}_\cmodel$, and if the challenger internal bit is $1$ then $\vec u_s,a$ has been sampled from $\sem{\vec u, \enc{t}{\pk(\nonce)}{\enonce}}_\cmodel$:
  \[
    \vec u_s,a \
    \rla
    \begin{dcases*}
      \sem{\vec u, \enc{s}{\pk(\nonce)}{\enonce}}_\cmodel & if $b = 0$\\
      \sem{\vec u, \enc{t}{\pk(\nonce)}{\enonce}}_\cmodel & if $b = 1$
    \end{dcases*}
  \]
  Then $\mathcal{B}$ returns $\mathcal{A}(\vec u_s,a)$.  It is easy to check that the advantage of $\mathcal{B}$ against the $\textsc{ind-cca1}$ game is exactly the advantage of $\mathcal{A}$ against
  \(
    \vec u, \enc{s}{\pk(\nonce)}{\enonce}
    \sim
    \vec v, \enc{t}{\pk(\nonce)}{\enonce}
  \)
  This advantage is the quantity in Equation~\ref{eq:qeeqwiofhaifna}, which we assumed non-negligible. Hence $\mathcal{B}$ is a winning adversary against the \textsc{ind-cca1} game. Contradiction.
\end{proof}

\paragraph{The $\ccao$ Axioms}
We now define the set of axioms $\ccao$, which is more convenient to use than $\ccao_s$:
\begin{definition}
  We let $\ccao$ be the set of axioms:
  \[
    \begin{array}[c]{c}
      \infer[\ccao]
      {\vec u, \enc{s}{\pk(\nonce)}{\enonce},\length(s)
        \sim
        \vec v, \enc{t}{\pk(\nonce')}{\enonce'},\length(s)}
      {\vec u \sim \vec v
        \quad&\quad
        \length(s) \peq \length(t)}
    \end{array}
    \qquad \qquad
    \text{ when }
    \begin{dcases}
      \fresh{\enonce,\enonce'}{\vec u,\vec v,s,t}\\
      \vec u \equiv \pk(\nonce),\_ \;\wedge\;\vec v \equiv \pk(\nonce'),\_\\
      \nonce \tpos_{\pk(\cdot),\sk(\cdot)} \vec u,s
      \;\wedge\; \sk(\nonce) \tpos_{\dec(\_,\cdot)} \vec u,s\\
      \nonce' \tpos_{\pk(\cdot),\sk(\cdot)} \vec v,t
      \;\wedge\; \sk(\nonce') \tpos_{\dec(\_,\cdot)} \vec v,t
    \end{dcases}
  \]
\end{definition}
We now state the following soundness theorem:
\begin{proposition}
  \label{prop:ccao-norm-valid}
  The $\ccao$ axioms are valid in any computational model where $(\enc{\_}{\_}{\_},\dec(\_,\_),\pk(\_),\sk(\_))$ is interpreted as a $\textsc{ind-cca1}$ secure encryption scheme.
\end{proposition}

\begin{proof}
  The proof relies on the transitivity axiom $\trans$ and the $\ccao_s$ axioms, which are valid are valid in  any computational model where $(\enc{\_}{\_}{\_},\dec(\_,\_),\pk(\_),\sk(\_))$ is interpreted as a $\textsc{ind-cca1}$ secure encryption scheme using Proposition~\ref{prop:ccaos-valid}.
  \[
    \infer[\trans]
    {
      \vec u, \enc{s}{\pk(\nonce)}{\enonce}
      \sim
      \vec v, \enc{t}{\pk(\nonce')}{\enonce'}
    }{
      \infer[\ccao_s]{
        \vec u, \enc{s}{\pk(\nonce)}{\enonce}
        \sim
        \vec u, \enc{0^{\length(s)}}{\pk(\nonce)}{\enonce}
      }{
        \unary{
          \length(s) = \length(0^{\length(s)})
        }
      }
      \;&\;
      \vec u, \enc{0^{\length(s)}}{\pk(\nonce)}{\enonce}
      \sim
      \vec v, \enc{0^{\length(t)}}{\pk(\nonce)}{\enonce}
      \;&\;
      \infer[\ccao_s]{
        \vec v, \enc{0^{\length(t)}}{\pk(\nonce)}{\enonce}
        \sim
        \vec v, \enc{t}{\pk(\nonce')}{\enonce'}
      }{
        \unary{
          \length(t) = \length(0^{\length(t)})
        }
      }
    }
  \]
  And:
  \[
    \infer[\fa^3]{
      \vec u, \enc{0^{\length(s)}}{\pk(\nonce)}{\enonce}
      \sim
      \vec v, \enc{0^{\length(t)}}{\pk(\nonce)}{\enonce}
    }{
      \infer[R]{
        \vec u, \length(s),\pk(\nonce),\enonce
        \sim
        \vec v, \length(t),\pk(\nonce),\enonce
      }{
        \infer[\dup]{
          \vec u, \length(s),\pk(\nonce),\enonce
          \sim
          \vec v, \length(s),\pk(\nonce),\enonce
        }{
          \infer[\ax{Fresh}]{
            \vec u, \length(s),\enonce
            \sim
            \vec v, \length(s),\enonce
          }{
            \unary{
              \vec u, \length(s)
              \sim
              \vec v, \length(s)
            }
          }
        }
        \;&\;
        \unary{
          \length(s)
          =
          \length(t)}
      }
    }
    \qedhere
  \]
\end{proof}


\subsection{$\textsc{prf-mac}$ Axioms}
\label{app:subsection-joints-prf}
\begin{definition}[$\prfass$ Function\cite{DBLP:books/cu/Goldreich2001,DBLP:journals/jacm/GoldreichGM86}]
  Let $H(\cdot,\cdot): \{0,1\}^*\times\{0,1\}^\eta \rightarrow \{0,1\}^\eta$ be a keyed hash functions. The function $H$  is a \emph{Pseudo Random Function} if, for any PPTM adversary $\cal A$ with access to an oracle ${\cal O}_f$:
  \[ |\Pr(k: \; {\cal A}^{{\cal O}_{H(\cdot,k)}}(1^\eta)=1)
    -
    \Pr(g:\; {\cal A}^{{\cal O}_{g(\cdot)}}(1^\eta)=1)|\]
  is negligible.
  Where
  \begin{itemize}
  \item $k$ is drawn uniformly in $\{0,1\}^\eta$.
  \item $g$ is drawn uniformly in the set of all functions from $\{0,1\}^*$ to $\{0,1\}^\eta$.
  \end{itemize}
\end{definition}
The authors of~\cite{DBLP:conf/csfw/ComonK17} already gave axioms for this property (and proved soundness). We recall their axiom schema below, using our notations:
\[
  \begin{array}[c]{c}
    \infer
    {
      \begin{alignedat}{2}
        &&&\textstyle\vec u,
        \ite{\bigvee_{i \in I} \eq{m}{m_i}}{\zero}{H(m,\key)}\\
        &\sim\;\;&&\textstyle
        \vec u,
        \ite{\bigvee_{i \in I} \eq{m}{m_i}}{\zero}{\nonce}
      \end{alignedat}
    }
    {
    }
  \end{array}
  \text{ when }
  \begin{dcases}
    \fresh{\nonce}{\vec u,m}\\
    \key \tpos_{H(\_,\cdot)} \vec u,m\\
    \{m_i \mid i \in I\} = \{u \mid H(u,\key) \in \st(\vec u,m)\}\\
    \forall u,v.\,\text{if } H(u,v) \in \st(\vec u,m) \text{ then } v \equiv \key
  \end{dcases}
\]
We simplify this axiom schema by dropping the last syntactical requirement. Indeed, it is not necessary to require that every occurrence of $H$ in $\vec u,m$ uses the key $\key$. We prove that this is valid. 
\begin{proposition}
  \label{prop:prf-simpl-less-requ}
  The following set of axioms is valid in any computational model where the $H$ is interpreted as a $\prfass$ function:
  \[
    \begin{array}[c]{c}
      \infer
      {
        \begin{alignedat}{2}
          &&&\textstyle\vec u,
          \ite{\bigvee_{i \in I} \eq{m}{m_i}}{\zero}{H(m,\key)}\\
          &\sim\;\;&&\textstyle
          \vec u,
          \ite{\bigvee_{i \in I} \eq{m}{m_i}}{\zero}{\nonce}
        \end{alignedat}
      }{}
    \end{array}
    \text{ when }
    \begin{dcases}
      \fresh{\nonce}{\vec u,m}\\
      \key \tpos_{H(\_,\cdot)} \vec u,m\\
      \{m_i \mid i \in I\} = \{u \mid H(u,\key) \in \st(\vec u,m)\}
    \end{dcases}
  \]
\end{proposition}

\begin{proof}
  We consider an instance of the axiom schema:
  \[
    \infer
    {
      \begin{alignedat}{2}
        &&&\textstyle\vec u,
        \ite{\bigvee_{i \in I} \eq{m}{m_i}}{\zero}{H(m,\key)}\\
        &\sim\;\;&&\textstyle
        \vec u,
        \ite{\bigvee_{i \in I} \eq{m}{m_i}}{\zero}{\nonce}
      \end{alignedat}
    }{}
  \]
  Let $\vec h \equiv (H(m_i,\key))_{i \in I}$ and $\vec v[],b[]$ be contexts such that $\vec v[\vec h] \equiv \vec u$ and $b[\vec h] \equiv \bigvee_{i \in I} \eq{m}{m_i}$ and such that $\key \not \in \st(\vec v, b)$. Let $\cmodel^0$ be a computational model and $\mathcal{A}$ be an adversary. We need to show that:
  \[
    \Pr\left(
      \mathcal{A}\left(
        \lrsem{\vec v[\vec h], \cond{b[\vec h]}{H(m,\key)}}_{\cmodel^0}
        \right) = 1
    \right)
    \approx
    \Pr\left(
      \mathcal{A}\left(
        \lrsem{\vec v[\vec h], \cond{b[\vec h]}{\nonce}}_{\cmodel^0}
      \right) = 1
    \right)
  \]
  Let $\cmodel$ be an extension of $\cmodel^0$ where we added two function symbols $g,g'$ which are interpreted as random functions. Then we know that it is sufficient to show that:
  \[
    \Pr\left(
      \mathcal{A}\left(
        \lrsem{\vec v[\vec h], \cond{b[\vec h]}{H(m,\key)}}_{\cmodel}
      \right) = 1
    \right)
    \approx
    \Pr\left(
      \mathcal{A}\left(
        \lrsem{\vec v[\vec h], \cond{b[\vec h]}{\nonce}}_{\cmodel}
      \right) = 1
    \right)
    \numberthis\label{eq:fdughofuqeriowr}
  \]
  Let $\vec r \equiv (g(m_i))_{i \in I} $. It is straightforward to check that, thanks to the $\prfass$ assumption of $H$, we have:
  \[
    \Pr\left(
      \mathcal{A}\left(
        \lrsem{\vec v[\vec h], \cond{b[\vec h]}{H(m,\key)}}_{\cmodel}
      \right) = 1
    \right)
    \approx
    \Pr\left(
      \mathcal{A}\left(
        \lrsem{\vec v[\vec r\,], \cond{b[\vec r\,]}{g(m)}}_{\cmodel}
      \right) = 1
    \right)
  \]
  Moreover, using the fact that the subterm $g(m)$ is guarded by $b[\vec r\,]$, we know that, except for a negligible number of samplings, $m$ is never queried to the random function $g$ except once, in $\cond{b[\vec r\,]}{g(m)}$. It follows that we can safely replace the last call to $g(m)$ by a call to $g'(m)$, which yields:
  \[
    \Pr\left(
      \mathcal{A}\left(
        \lrsem{\vec v[\vec r\,], \cond{b[\vec r\,]}{g(m)}}_{\cmodel}
      \right) = 1
    \right)
    \approx
    \Pr\left(
      \mathcal{A}\left(
        \lrsem{\vec v[\vec r\,], \cond{b[\vec r\,]}{g'(m)}}_{\cmodel}
      \right) = 1
    \right)
  \]
  Now, using again the $\prfass$ property of $H$, we know that:
  \[
    \Pr\left(
      \mathcal{A}\left(\lrsem{\vec v[\vec r\,], \cond{b[\vec r\,]}{g'(m)}}_{\cmodel}\right) = 1
    \right)
    \approx
    \Pr\left(
      \mathcal{A}\left(\lrsem{\vec v[\vec h], \cond{b[\vec h]}{g'(m)}}_{\cmodel}\right) = 1
    \right)
  \]
  Finally, since $g'$ appears only once in $\vec v[\vec h], \cond{b[\vec h]}{g'(m)}$, we can replace $g'(m)$ by a fresh nonce. Hence:
  \[
    \Pr\left(
      \mathcal{A}\left(\lrsem{\vec v[\vec h], \cond{b[\vec h]}{g'(m)}}_{\cmodel}\right) = 1
    \right)
    \approx
    \Pr\left(
      \mathcal{A}\left(\lrsem{\vec v[\vec h], \cond{b[\vec h]}{\nonce}}_{\cmodel}\right) = 1
    \right)
  \]
  This concludes the proof of \eqref{eq:fdughofuqeriowr}.
\end{proof}

This can be extended to have a finite family of functions being jointly $\prfass$. A finite family of functions $H_1(\cdot,k),\dots,H_n(\cdot,k)$ are jointly $\prfass$ if they are jointly computationally indistinguishable from random functions. Formally:
\begin{definition}[Jointly $\prfass$ Functions]
  Let $H_1(\cdot,\cdot),\dots,H_n(\cdot,\cdot)$ be a finite family of keyed hash functions such that for every $1 \le i \le n$, $H_i(\cdot,\cdot) : \{0,1\}^*\times\{0,1\}^\eta \rightarrow \{0,1\}^\eta$. The functions $H_1,\dots,H_n$ are \emph{Jointly Pseudo Random Functions} if, for any PPTM adversary $\cal A$ with access to oracles ${\cal O}_{f_1},\dots,{\cal O}_{f_n}$:
  \[ |\Pr(k: \; {\cal A}^{{\cal O}_{H_1(\cdot,k)},\dots,{\cal O}_{H_n(\cdot,k)}}(1^\eta)=1)
    -
    \Pr(g_1,\dots,g_n:\; {\cal A}^{{\cal O}_{g_1(\cdot)},\dots,{\cal O}_{g_n(\cdot)}}(1^\eta)=1)|\]
  is negligible.
  Where
  \begin{itemize}
  \item $k$ is drawn uniformly in $\{0,1\}^\eta$.
  \item $g_1,\dots,g_n$ are drawn uniformly in the set of all functions from $\{0,1\}^*$ to $\{0,1\}^\eta$.
  \end{itemize}
\end{definition}

\begin{remark}
  It is easy to build a family $H_1,\dots,H_n$ of jointly pseudo random functions from a pseudo random function $H(\cdot,\cdot)$. First, let $(\ttag_i(\cdot))_{1 \le i \le n}$ be a set of tagging functions. We require that these functions are unambiguous, i.e. for all bit-strings $u,v$ and $i \ne j$ we must have $\ttag_i(u) \ne \ttag_j(v)$. Then for every $1 \le i \le n$, we let $H_i(x,y) = H(\ttag_i(x),y)$. It is straightforward to show that if $H$ is a $\prfass$ then $H_1,\dots,H_n$ are jointly $\prfass$.
\end{remark}

Now, we translate this property for $\owsym$ and $\rowsym$ (resp. $\macsym^1$--$\,\macsym^5$) in the logic.

\begin{definition}
  We let $\setmac_{\mkey}^j(u)$ be the set of $\macsym^j$ terms under key $\mkey$ in $u$:
  \[
    \setmac_{\mkey}^j(u) = \{ m \mid \mac{m}{\mkey}{j} \in \st(u)\}
  \]
\end{definition}

\begin{definition}
  \label{def:prf-mac-ax}
  For every $1 \le j \le 5$, we let $\prfmac^j$ be the set of axioms:
  \[
    \begin{array}[c]{c}
      \infer[\prfmac^j]
      {
        \begin{alignedat}[c]{2}
          &&&\textstyle\vec u,
          \ite{\bigvee_{i \in I} \eq{m}{m_i}}{\zero}{\mac{m}{\mkey}{j}}\\
          &\sim\;\;&&\textstyle
          \vec u,
          \ite{\bigvee_{i \in I} \eq{m}{m_i}}{\zero}{\nonce}
        \end{alignedat}
      }{}
    \end{array}
    \text{ when }
    \begin{dcases}
      \fresh{\nonce}{\vec u,m}\\
      \mkey \tpos_{\mac{\_}{\cdot}{\_}} \vec u,m\\
      \{m_i \mid i \in I\} = \setmac_{\mkey}^j(\vec u,m)
    \end{dcases}
  \]
\end{definition}

\begin{definition}
  Let $g \in \{\owsym,\rowsym\}$. We let $\setprf^{g}_{\key}(u)$ be the set of $g$ terms under key $\key$ in $u$:
  \[
    \setprf^{g}_{\key}(u) = \{ m \mid g_\key(m) \in \st(u)\}
  \]
\end{definition}

\begin{definition}
  \label{def:prf-f-fr-ax}
  For every $g \in \{\owsym,\rowsym\}$, we let $\prfass$-$g$ be the set of axioms:
  \[
    \begin{array}[c]{c}
      \infer[\text{$\prfass$-g}]
      {
        \begin{alignedat}[c]{2}
          &&&\textstyle\vec u,
          \ite{\bigvee_{i \in I} \eq{m}{m_i}}{\zero}{g_\key(m)}\\
          &\sim\;\;&&\textstyle
          \vec u,
          \ite{\bigvee_{i \in I} \eq{m}{m_i}}{\zero}{\nonce}
        \end{alignedat}
      }{}
    \end{array}
    \text{ when }
    \begin{dcases}
      \fresh{\nonce}{\vec u,m}\\
      \key \tpos_{\owsym_{\cdot}(\_),\rowsym_{\cdot}(\_)} \vec u,m\\
      \{m_i \mid i \in I\} = \setprf^g_{\key}(\vec u,m)
    \end{dcases}
  \]
\end{definition}

\begin{proposition}
  \label{prop:prf-mult-norm-valid}
  The $(\prfmac^j)$ (resp. $\prff$ and $\prffr$) axiom schemes are valid in any computational model where the $(\macsym^j)$ (resp. $\owsym$ and $\rowsym$) function symbols are interpreted as jointly $\prfass$ functions.
\end{proposition}

\begin{proof}
  This is an extension of the axiom schema given in Proposition~\ref{prop:prf-simpl-less-requ}. The soundness proof follows the same step, by replacing every call $H_j(\_,\key)$ with a call to a random function $g_i(\_)$, where $(g_i)_{1 \le i \le n}$ are independent random functions. We omit the details.
\end{proof}

\begin{remark}
  If we have a valid instance of $\prfmac^j$:
  \[
    \infer[\prfmac^j]
    {
      \begin{alignedat}[c]{2}
        &&&\textstyle\vec u,
        \ite{\bigvee_{i \in I} \eq{m}{m_i}}{\zero}{\mac{m}{\mkey}{j}}\\
        &\sim\;\;&&\textstyle
        \vec u,
        \ite{\bigvee_{i \in I} \eq{m}{m_i}}{\zero}{\nonce}
      \end{alignedat}
    }{}
  \]
  then using transitivity we know that:
  \[
    \infer[\trans]
    {
      \vec u,
      \ite{\bigvee_{i \in I} \eq{m}{m_i}}{\zero}{\mac{m}{\mkey}{j}}
      \sim
      \vec v
    }{
      \infer[\prfmac^j]{
        \begin{alignedat}{2}
          &&&\textstyle\vec u,
          \ite{\bigvee_{i \in I} \eq{m}{m_i}}{\zero}{\mac{m}{\mkey}{j}}\\
          &\sim\;\;&&\textstyle
          \vec u,
          \ite{\bigvee_{i \in I} \eq{m}{m_i}}{\zero}{\nonce}
        \end{alignedat}
      }{}
      \quad&\quad
      \vec u,
      \ite{\bigvee_{i \in I} \eq{m}{m_i}}{\zero}{\nonce}
      \sim
      \vec v
    }
  \]
  Therefore the following axiom schema is admissible using $\prfmac^j + \trans$:
  \[
    \begin{array}[c]{c}\infer
      {
        \vec u,
        \ite{\bigvee_{i \in I} \eq{m}{m_i}}{\zero}{\mac{m}{\mkey}{j}}
        \sim
        \vec v
      }{
        \vec u,
        \ite{\bigvee_{i \in I} \eq{m}{m_i}}{\zero}{\nonce}
        \sim
        \vec v
      }
    \end{array}
    \text{ when }
    \begin{dcases}
      \fresh{\nonce}{\vec u,m}\\
      \mkey \tpos_{\mac{\_}{\cdot}{\_}} \vec u,m\\
      \{m_i \mid i \in I\} = \setmac_{\mkey}^j(\vec u,m)
    \end{dcases}
  \]
  We will prefer the axiom schema above over the axiom schema given in Definition~\ref{def:prf-mac-ax}. By a notation abuse, we refer also to the axiom above as $\prfmac^j$. The same remark applies to $\prff$ and $\prffr$.
\end{remark}

\subsection{$\textsc{euf-mac}$ Axioms}

\paragraph{The \textsc{simp-euf-mac} Axioms}
\begin{definition}
  We let $\setmac_{\mkey}(u)$ be the set of $\macsym$ terms under key $\mkey$ in $u$:
  \[
    \setmac_{\mkey}(u) = \{ m \mid \macsym_{\mkey}(m) \in \st(u)\}
  \]
\end{definition}

\begin{definition}
  A function $\macsym$ is \textsc{euf-mac} secure if for every PPTM $\mathcal{A}$, the following quantity is negligible in $\eta$:
  \[
    \Pr\left(\mkey \la \{0,1\}^\eta :
      (m,\sigma) \la \mathcal{A}^{\mathcal{O}_{\macsym}(\mkey)}(1^\eta),
      m \text{ not queried to } \mathcal{O}_{\macsym}(\mkey)
      \text{ and }
      \sigma = \mac{m}{\mkey}{}\right)
  \]
\end{definition}

This can be modeled using the following axioms:
\begin{definition}
  We let \textsc{simp-euf-mac} be the set of axioms:
  \[
    \unary{s = \macsym_{\mkey}(m) \ra
      \bigvee_{u \in S} s = \macsym_{\mkey}(u)}
    \quad\text{ when }
    \begin{cases}
      \mkey \tpos_{\macsym_{\cdot}(\_)} s,m\\
      S = \setmac_{\mkey}(s,m)
    \end{cases}
    \tag{\textsc{simp-euf-mac}}
  \]
\end{definition}

\begin{proposition}
  The \textsc{simp-euf-mac} axioms are valid in any computational model where the $\macsym$ function is interpreted as an \textsc{euf-mac} secure function.
\end{proposition}

The proof is similar to the proof of Proposition~\ref{prop:ccaos-valid}. We omit the details.

\paragraph{The \textsc{euf-mac} Axioms}
It is well-know that if a function $H$ is a $\prfass$ then $H$ is $\textsc{euf-mac}$ secure. We give here the counterpart of this result for a family of functions $H,H_1,\dots,H_l$ which are jointly $\prfass$.

If $H,H_1,\dots,H_l$ which are jointly $\prfass$, then no adversary can forge a mac of $H(\cdot,\mkey)$, even if the adversary has oracle access to $H(\cdot,\mkey),H_1(\cdot,\mkey),\dots,H_l(\cdot,\mkey)$. First, we define what it means for a function to be \textsc{euf-mac} secure with a key jointly used by other functions:
\begin{definition}
  A function $H$ is \textsc{euf-mac} secure with a key jointly used by $H_1,\dots,H_l$ if for every PPTM $\mathcal{A}$, the following quantity is negligible in $\eta$:
  \[
    \Pr\left(\mkey \la \{0,1\}^\eta :
      (m,\sigma) \la
      \mathcal{A}^{\mathcal{O}_{H(\cdot,\mkey)},
        \mathcal{O}_{H_1(\cdot,\mkey)},
        \dots,
        \mathcal{O}_{H_l(\cdot,\mkey)}}(1^\eta),
      m \text{ not queried to } \mathcal{O}_{H(\cdot,\mkey)}
      \text{ and }
      \sigma = H(m,\mkey)\right)
  \]
\end{definition}

\begin{proposition}
  \label{prop:prf-implies-mac-joints-app}
  If $H,H_1,\dots,H_l$ are jointly $\prfass$ then $H$ is \textsc{euf-mac} secure with a key jointly used by $H_1,\dots,H_l$.
\end{proposition}

\begin{proof}
  The proof is almost the same than the proof showing that if a function $H$ is a $\prfass$ then $H$ is $\textsc{euf-mac}$ secure, and is by reduction. If $H$ is not \textsc{euf-mac} secure with a key jointly used by $H_1,\dots,H_l$ then there exists an adversary $\mathcal{A}$ winning the corresponding game with a non-negligible probability. It is simple to build from $\mathcal{A}$ an adversary $\mathcal{B}$ against the joint $\prfass$ property of $H,H_1,\dots,H_l$.

  First, $\mathcal{B}$ runs the adversary $\mathcal{A}$, forwarding and logging its the oracle calls. Eventually, $\mathcal{A}$ returns a pair $(m,\sigma)$. Then, $\mathcal{B}$ queries the first oracle on $m$, which returns a value $\sigma'$. Finally, $\mathcal{B}$ returns $1$ if and only if $\mathcal{A}$ never queried the first oracle on $m$ and $\sigma' = \sigma$. Then:
  \begin{itemize}
  \item If $\mathcal{B}$ is interacting with the oracles $\mathcal{O}_{H(\cdot,\mkey)},\mathcal{O}_{H_1(\cdot,\mkey)}, \dots, \mathcal{O}_{H_l(\cdot,\mkey)}$, its probability of returning $1$ is exactly the advantage of $\mathcal{A}$ against the \textsc{euf-mac} game with key jointly used.
  \item If $\mathcal{B}$ is interacting with the oracles $\mathcal{O}_{g(\cdot)},\mathcal{O}_{g_1(\cdot)}, \dots, \mathcal{O}_{g_l(\cdot)}$ where $g,g_1,\dots,g_l$ are random functions, then its probability of returning $1$ is the probability of having $g(m) = \sigma$ knowing that $m$ was never queried to $g$. Since $g$ is a random function, this is less than $1/2^\eta$.
  \end{itemize}
  Since $\mathcal{A}$ has a non-negligible advantage against the \textsc{euf-mac} game with key jointly used, we deduce that $\mathcal{B}$ has a non-negligible advantage against the joint \prff game.
\end{proof}

We translate this cryptographic property in the logic to obtain the sets of axioms $(\textsc{euf-mac}^j)_{1 \le j \le 5}$. 
\begin{definition}
  We let $\textsc{euf-mac}^j$ be the set of axioms:
  \[
    \unary{s = \mac{m}{\mkey}{j} \ra
      \bigvee_{u \in S} s = \mac{u}{\mkey}{j}}
    \quad\text{ when }
    \begin{cases}
      \mkey \tpos_{\mac{\_}{\cdot}{\_}} s,m\\
      S = \setmac_{\mkey}^j(s,m)
    \end{cases}
    \tag{$\textsc{euf-mac}^j$}
  \]
\end{definition}

\begin{proposition}
  \label{prop:euf-mac-valid-app}
  The $\textsc{euf-mac}^j$ axioms are valid in any computational model where the $\macsym^j$ function is interpreted as a \textsc{euf-mac} secure function with a key jointly used by the interpretations of $(\macsym^i)_{1\le i \le 5, i \ne j}$.
\end{proposition}
Remark that it is easy to prove Proposition~\ref{prop:euf-mac-valid} using Proposition~\ref{prop:prf-implies-mac-joints-app} and Proposition~\ref{prop:euf-mac-valid-app}.

\begin{proof}[Proof of Proposition~\ref{prop:euf-mac-valid-app}]
  The proof is straightforward and has the same structure than the proof of Proposition~\ref{prop:ccaos-valid}. 

  We assume that there is a computational model $\cmodel$ where the $\macsym^j$ function is interpreted as a \textsc{euf-mac} secure function with a key jointly used by the interpretations of $(\macsym^i)_{1\le i \le 5, i \ne j}$. Moreover, we assume that there is an instance:
  \[
    \unary{s = \mac{m}{\mkey}{j} \ra
      \bigvee_{u \in S} s = \mac{u}{\mkey}{j}}
  \]
  of the $\textsc{euf-mac}^j$ which is not valid in $\cmodel$, and such that:
  \begin{mathpar}
    \mkey \tpos_{\mac{\_}{\cdot}{\_}} s,m

    S = \setmac_{\mkey}^j(s,m)
  \end{mathpar}
  Therefore we know that the following quantity is not negligible in $\eta$:
  \[
    \textstyle
    \Pr\left(\rho_1,\rho_2:
      \sem{s = \mac{m}{\mkey}{j}}_{\rho_1,\rho_2}^\eta
      \wedge
      \neg \sem{\bigvee_{u \in S} s = \mac{u}{\mkey}{j}}_{\rho_1,\rho_2}^\eta
    \right)
  \]
  Or equivalently the following quantity is not negligible in $\eta$:
  \[
    \textstyle
    \Pr\left(\rho_1,\rho_2:
      \sem{s}_{\rho_1,\rho_2}^\eta = \sem{\mac{m}{\mkey}{j}}_{\rho_1,\rho_2}^\eta
      \wedge
      \bigwedge_{u \in S} \sem{s}_{\rho_1,\rho_2}^\eta \ne
      \sem{\mac{u}{\mkey}{j}}_{\rho_1,\rho_2}^\eta
    \right)
    \numberthis\label{eq:fofhdsogfhqwirjiod}
  \]
  Using $\cmodel$ we can build a adversary $\mathcal{A}$ against the \textsc{euf-mac} game with key jointly used. The adversary $\mathcal{A}$ simply samples two values $a_s,a_{\macsym}$ from $\sem{s}_{\rho_1,\rho_2}^\eta$ and $\sem{\mac{m}{\mkey}{j}}_{\rho_1,\rho_2}^\eta$ by sampling a value from all the subterms of $s$ and $m$ in a bottom-up fashion. The adversary calls the $(\mathcal{O}_{\macsym^i})_{1 \le i \le 5}$ whenever he need to sample a value from a subterm of the form $\mac{\_}{\mkey}{i}$. Remark that the side-condition $\mkey \tpos_{\mac{\_}{\cdot}{\_}} s,m$ ensures that this is always possible. Then $\mathcal{A}$ returns $a_s,a_{\macsym}$. One can check that the advantage of $\mathcal{A}$ against the \textsc{euf-mac} game with key jointly used by $(\macsym^i)_{1\le i \le 5, i \ne j}$ is exactly the quantity in \ref{eq:fofhdsogfhqwirjiod}. It follows that $\mathcal{A}$ has a non-negligible probability of winning the game. Contradiction.
\end{proof}

\subsection{\textsc{cr} Axioms}
We recall the definition of Collision-Resistance:
\begin{definition}
  A function $H$ is \textsc{cr} secure if for every PPTM $\mathcal{A}$, the following quantity is negligible in $\eta$:
  \[
    \Pr\left(\mkey \la \{0,1\}^\eta :
      (m_1,m_2) \la \mathcal{A}^{\mathcal{O}_{H(\cdot,\mkey)}}(1^\eta),
      m_1 \ne m_2 \text{ and } H(m_1,\mkey) = H(m_2,\mkey)
    \right)
  \]
\end{definition}
As for unforgeability, we generalize this to allow the key $\mkey$ to be jointly used by others functions. Formally:
\begin{definition}
  A function $H$ is \textsc{cr} secure with a key jointly used by $H_1,\dots,H_l$ if for every PPTM $\mathcal{A}$, the following quantity is negligible in $\eta$:
  \[
    \Pr\left(\mkey \la \{0,1\}^\eta :
      (m_1,m_2) \la
      \mathcal{A}^{\mathcal{O}_{H(\cdot,\mkey)},
        \mathcal{O}_{H_1(\cdot,\mkey)},
        \dots,
        \mathcal{O}_{H_l(\cdot,\mkey)}}(1^\eta),
      m_1 \ne m_2 \text{ and } H(m_1,\mkey) = H(m_2,\mkey)
    \right)
  \]
\end{definition}
It is well-known that a \textsc{euf-mac} secure function is also \textsc{cr} secure. Similarly we have that:
\begin{proposition}
  \label{prop:euf-joints-implies-cr-joints}
  If $H$ is \textsc{euf-mac} secure with a key jointly used by $H_1,\dots,H_l$ then $H$ is \textsc{cr} secure with a key jointly used by $H_1,\dots,H_l$.
\end{proposition}

\begin{proof}
  We can easily build an adversary $\mathcal{B}$ against the \textsc{euf-mac} game with a key jointly from any adversary $\mathcal{A}$ against the \textsc{cr} game with a key jointly such that $\mathcal{A}$ and $\mathcal{B}$ have the same advantage against their respective games. The result follows.
\end{proof}

We translate this game in the logic as follows:
\begin{definition}
  We let $\textsc{cr}^j$ be the set of axioms:
  \[
    \begin{array}[c]{c}
      \unary{
        \mac{m_1}{\mkey}{j} = \mac{m_2}{\mkey}{j}
        \ra m_1 = m_2}
    \end{array}
    \quad\text{ when }
    \mkey \tpos_{\mac{\_}{\cdot}{\_}} m_1,m_2
    \tag{$\textsc{cr}^j$}
  \]
\end{definition}

\begin{proposition}
  \label{prop:cr-ax-valid-app}
  The $\textsc{cr}^j$ axioms are valid in any computational model where the $\macsym^j$ function is interpreted as a \textsc{cr} secure function with a key jointly used by the interpretations of $(\macsym^i)_{1\le i \le 5, i \ne j}$.
\end{proposition}

\begin{proof}
The proof works exactly like the proof of Proposition~\ref{prop:ccaos-valid} and Proposition~\ref{prop:euf-mac-valid-app}. We omit the details.
\end{proof}




\subsection{Cryptographic Axioms}

\begin{definition}
  We let $\axioms_{\textsf{crypto}}$ be the set of cryptographic axioms:
  \[
    \axioms_{\textsf{crypto}}
    =
    \ccao
    \cup \left(\prfmac^j\right)_{1 \le j \le 5}
    \cup\prff
    \cup\prffr
    \cup\left(\textsc{euf-mac}^j\right)_{1 \le j \le 5}
    \cup \left(\textsc{cr}^j\right)_{1 \le j \le 5}
  \]
\end{definition}

\begin{proposition}
  The axioms in $\axioms_{\textsf{crypto}}$ are valid in any computational model where the asymmetric encryption $\enc{\_}{\_}{\_}$ is \textsc{ind-cca1} secure and $\owsym$ and $\rowsym$ (resp. $\macsym^1$--$\,\macsym^5$) satisfy jointly the \prff assumption.
\end{proposition}

\begin{proof}
  This is a direct consequence of the Propositions~\ref{prop:ccao-norm-valid}, \ref{prop:prf-mult-norm-valid}, \ref{prop:euf-mac-valid-app} and \ref{prop:cr-ax-valid-app}.
\end{proof}

\subsection{Structural and Implementation Axioms}
\label{section-app:struct-axs}

\input{axiom-figure}

We present the structural axioms $\axioms_{\textsf{struct}}$, which are given in Fig.~\ref{figure:axioms}. All these axioms have been introduced in the literature (e.g. see \cite{Bana:2014:CCS:2660267.2660276,DBLP:conf/csfw/ComonK17}). Still, we informally describe them: the axioms $\sym,\refl$ and $\trans$ states that computational indistinguishability is an equivalence relation; the $\perm$ axiom is used to change the order of the terms using a permutation $\pi$ on both side of $\sim$; $\dup$ is used to remove duplicate; $\restr$ allows to strengthen the goal; $\fa$ states that to show that two function applications are indistinguishable, it is sufficient to show that their arguments are indistinguishable; the axiom $R$ allow to rewrite any occurrence of $s$ into $t$ if we can show that $s = t$; $\ax{Fresh}$ allows to remove a random sampling appearing if it is not used; $\ax{$\oplus$-\textsf{indep}}$ is the optimistic sampling rule (see~\cite{DBLP:conf/lpar/BartheDKLL10}); and finally, $\cs$ states that to show that two $\symite$ are indistinguishable, it is sufficient to show that their $\textsf{then}$ branches and $\textsf{else}$ branches are indistinguishable, when giving the value of the branching conditional to the adversary.

We then have the following soundness result:
\begin{proposition}
  The axioms in $\axioms_{\textsf{struct}}$ are valid in any computational model.
\end{proposition}
\begin{proof}
  The soundness proofs can be found in~\cite{Bana:2014:CCS:2660267.2660276,DBLP:conf/csfw/ComonK17}.
\end{proof}

We can now define the set of axioms $\axioms$:
\begin{definition}
  We let $\axioms$ be the set inference rules:
  \[
    \axioms =
    \axioms_{\textsf{struct}} \cup
    \axioms_{\textsf{impl}} \cup
    \axioms_{\textsf{crypto}}
  \]
\end{definition}

\begin{definition}
  We let $\simp$ denote a sequence of applications of $R, \fa$ and $\dup$, i.e.:
  \[
    \begin{gathered}[c]
      \infer[\simp]{\vec u \sim \vec v}{\vec s \sim \vec t}
    \end{gathered}
    \quad\text{ when }\quad
    \begin{gathered}[c]
      \infer[(R+\fa+\dup)^*]{\vec u \sim \vec v}{\vec s \sim \vec t}
    \end{gathered}
  \]
\end{definition}

\paragraph*{Implementation Axioms}
We describe the implementation axioms $\axioms_{\textsf{impl}}$ given in Fig.~\ref{fig:axioms-add}. The set of implementation axioms is the union of the following sets of axioms:
\begin{itemize}
\item The set $\axioms_{\textsf{ite}}$ of equalities satisfied by the $\symite$ function symbols.
\item The set $\axioms_{\textsf{eq}}$ of axioms satisfied by the equality function symbol $\eq{\_}{\_}$. This includes functional properties of projections and decryption, equality properties such as reflexivity and dis-equalities.
\item The set $\axioms_{\textsf{len}}$ of axioms satisfied by the length function $\length(\_)$.
\item The set $\axioms_{\textsf{inj}}$ of injectivity axioms.
\item The set $\axioms_{\sqn}$ of axioms satisfied by the $\range(\_,\_)$ function on sequence numbers.
\end{itemize}

\subsection{$\textsc{p-euf-mac}_s$ Axioms}
We can refine the unforgeability axioms $\textsc{euf-mac}^j$ using a finite partition of the outcomes.
\begin{definition}
  A finite family of conditionals $(b_i)_{i \in I}$ is a valid $\cs$ partition if:
  \begin{gather*}
    \textstyle
    \left(\bigvee_i b_i \wedge \bigwedge_{i \ne j} b_i \ne b_j\right) \peq \true
  \end{gather*}
\end{definition}
We can have a more precise axiom, by considering a valid $\cs$ partition $(b_i)_{i \in I}$ and applying the $\textsc{euf-mac}^j$ axiom once for each element of the partition.
\begin{definition}
  We let $\textsc{p-euf-mac}_s^j$ be the set of axioms:
  \[
    \unary{s = \mac{m}{\mkey}{j} \ra
      \bigvee_{i \in I}
      b_i \wedge
      \bigvee_{u \in S_i} s = \mac{u}{\mkey}{j}}
    \text{ when }
    \begin{cases}
      \mkey \tpos_{\mac{\_}{\cdot}{\_}} s,m\\
      (b_i)_{i \in I} \text{ is a valid $\cs$ partition}\\
      \text{There exists } (s_i,m_i)_{i \in I} \text{ s.t. for every } i \in I\\
      \qquad\cond{b_i}{s_i} \peq \cond{b_i}{s}
      \wedge \cond{b_i}{m_i} \peq \cond{b_i}{m} \\
      \qquad S_i = \setmac_{\mkey}^j(s_i,m_i)
    \end{cases}
    \tag{$\textsc{p-euf-mac}^j_s$}
  \]
\end{definition}

\begin{proposition}
  \label{prop:app-euf-mac-no-s}
  The $\textsc{p-euf-mac}^j_s$ axioms are logical consequences of the axioms $\axioms$.
\end{proposition}

\begin{proof}
  The proof is pretty straightforward:
  \begin{alignat*}{2}
    s = \mac{m}{\mkey}{j}
    &\;\ra\;\;&
    \bigvee_{i \in I}
    \left(
      b_i \wedge s = \mac{m}{\mkey}{j}
    \right)
    \tag{Since $(b_i)_{i \in I}$ is a valid $\cs$ partition}\\
    &\;\ra\;\;&
    \bigvee_{i \in I}
    \left(
      b_i \wedge s_i = \mac{m_i}{\mkey}{j}
    \right)\displaybreak[0]\\
    &\;\ra\;\;&
    \bigvee_{i \in I}
    b_i \wedge
    \bigvee_{u \in S_i} s_i = \mac{u}{\mkey}{j}
    \tag{Using $\textsc{euf-mac}^j$ for every $i \in I$}\\
    &\;\ra\;\;&
    \bigvee_{i \in I}
    b_i \wedge
    \bigvee_{u \in S_i} s = \mac{u}{\mkey}{j}
  \end{alignat*}
\end{proof}

\subsection{$\textsc{p-euf-mac}$ Axioms}

We can further refine the unforgeability axioms, by noticing that macs appearing only in boolean conditionals can be ignored.

\begin{definition}
  For every term $u$, we let $\sst(u)$ be the set of subterms of $u$ appearing outside a conditional:
  \begin{mathpar}
    \sst(\ite{b}{u}{v}) = \{\sst(\ite{b}{u}{v})\} \cup \sst(u) \cup \sst(v)
    
    \sst(f(\vec u)) = \{f(\vec u)\} \cup \bigcup_{u \in \vec u} \sst(u)
    \text{ when } f \ne \symite
  \end{mathpar}
\end{definition}

\begin{definition}
  We let $\ssetmac_{\mkey}^j(u)$ be the set of mac-ed terms under key $\mkey$ and tag $j$ in $u$ appearing outside a conditional:
  \[
    \ssetmac_{\mkey}^j(u) = \{ m \mid \mac{m}{\mkey}{j} \in \sst(u)\}
  \]
\end{definition}

\begin{definition}
  We let $\textsc{p-euf-mac}^j$ be the set of axioms:   
  \[
    \unary{s = \mac{m}{\mkey}{j} \ra
      \bigvee_{i \in I}
      b_i \wedge
      \bigvee_{u \in S_i} s = \mac{u}{\mkey}{j}}
    \text{ when }
    \begin{cases}
      \mkey \tpos_{\mac{\_}{\cdot}{\_}} s,m\\
      (b_i)_{i \in I} \text{ is a valid $\cs$ partition}\\
      \text{There exists } (s_i,m_i)_{i \in I} \text{ s.t. for every } i \in I\\
      \qquad\cond{b_i}{s_i} \peq \cond{b_i}{s}
      \wedge \cond{b_i}{m_i} \peq \cond{b_i}{m} \\
      \qquad S_i = \ssetmac_{\mkey}^j(s_i,m_i)
    \end{cases}
    \tag{$\textsc{p-euf-mac}^j$}
  \]
\end{definition}

\begin{proposition}
  \label{prop:p-comp-euf-mac-valid-app}
  The $\textsc{p-euf-mac}^j$ axioms are logical consequences of the axioms $\axioms$.
\end{proposition}

\begin{proof}
  First, we are going to show that the following axioms are consequences of the axioms $\axioms$:
  \[
    \unary{s = \mac{m}{\mkey}{j} \ra
      \bigvee_{u \in S} s = \mac{u}{\mkey}{j}}
    \quad\text{ when }
    \begin{cases}
      \mkey \tpos_{\mac{\_}{\cdot}{\_}} s,m\\
      S \equiv \ssetmac_{\mkey}^j(s,m)
    \end{cases}
    \numberthis\label{eq:asifjioruqweq}
  \]
  Assuming the axioms above are valid, it is easy to conclude by repeating the proof of Proposition~\ref{prop:app-euf-mac-no-s}, but using the axiom above instead of $\textsc{euf-mac}^j$.

  To show that the axioms in~\eqref{eq:asifjioruqweq} are valid, we are going to pull out all conditionals using the properties of the $\symite$ function symbols. This yields a term of the form $C[\vec \beta \diamond (\vec e)]$ where $\vec e$ themselves of the form $s' = \mac{u'}{\mkey}{j}$. We then apply the $\textsc{euf-mac}^j$ axioms to every $e \in \vec e$. Finally, we rewrite back the conditionals.

  To be able to do this last step, we need, when we pull out the conditionals, to remember which conditional appeared where. We do this by replacing a conditional $b$ with either $\true_b$ or $\false_b$, where the $b$ lower-script is a label that we attach to the term.

  This motivates the following definition: for every boolean term $b$, we let $\lval_b = \{\true_b,\false_b\}$. We extend this to vector of conditionals by having $\lval_{u_0,\dots,u_l} = \lval_{u_0} \times \dots \times \lval_{u_l}$. Basically, for every vector of conditionals $\vec \beta$, choosing a vector of terms $\vec \nu \in \lval_{\vec \beta}$ correspond to choosing a valuation of $\vec \beta$.

  We can start showing the validity of \eqref{eq:asifjioruqweq}. Let $\vec \beta$ be the set of conditionals appearing in $s,m$, and $C$ be an if-context such that:
  \[
    \left(
      s = \mac{m}{\mkey}{j}
    \right)
    \lra
    \left(
      C\left[
        \vec \beta
        \diamond
        \left(
          s[\vec \nu/ \vec \beta] = \mac{m[\vec \nu/ \vec \beta]}{\mkey}{j}
        \right)_{\vec \nu \in \lval_{\vec \beta}}
      \right]
    \right)
  \]
  where $t[\vec u/\vec v]$ denotes the substitution of every occurrence of $\vec v$ by $\vec u$ in $t$. Then:
  \begin{alignat*}{2}
    s = \mac{m}{\mkey}{j}
    &\;\ra\;\;&&
    \left(
      C\left[
        \vec \beta
        \diamond
        \left(
          s[\vec \nu/ \vec \beta] = \mac{m[\vec \nu/ \vec \beta]}{\mkey}{j}
        \right)_{\vec \nu \in \lval_{\vec \beta}}
      \right]
    \right)\\
    \intertext{For every $\vec \nu \in \lval_{\vec \beta}$,
      let $S_{\vec \nu} = \setmac_{\mkey}^j
      (s[\vec \nu/ \vec \beta],m[\vec \nu/ \vec \beta])$. By applying $\textsc{euf-mac}^j$}
    &\;\ra\;\;&&
    \left(
      C\left[
        \vec \beta
        \diamond
        \left(
          \bigvee_{u \in S_{\vec \nu}}
          s[\vec \nu/ \vec \beta] = \mac{u}{\mkey}{j}
        \right)_{\vec \nu \in \lval_{\vec \beta}}
      \right]
    \right)    
  \end{alignat*}
  Since any conditional of $s[\vec \nu/ \vec \beta]$ or $m[\vec \nu/ \vec \beta]$ is of the form $\true_x$ or $\false_x$ for some label $x$, we know that:
  \[
    S_{\vec \nu} =
    \setmac_{\mkey}^j
    (s[\vec \nu/ \vec \beta],m[\vec \nu/ \vec \beta])
    =
    \ssetmac_{\mkey}^j
    (s[\vec \nu/ \vec \beta],m[\vec \nu/ \vec \beta])
  \]
  Moreover, we can check that:
  \[
    \ssetmac_{\mkey}^j
    (s[\vec \nu/ \vec \beta],m[\vec \nu/ \vec \beta])
    =
    \left(
      \ssetmac_{\mkey}^j(s,m)
    \right)[\vec \nu/ \vec \beta]
  \]
  Let $S = \ssetmac_{\mkey}^j(s,m)$. Hence:
  \begin{alignat*}{2}
    \left(
      C\left[
        \vec \beta
        \diamond
        \left(
          \bigvee_{u \in S_{\vec \nu}}
          s[\vec \nu/ \vec \beta] = \mac{u}{\mkey}{j}
        \right)_{\vec \nu \in \lval_{\vec \beta}}
      \right]
    \right)
    &\;\ra\;\;&&
    \left(
      C\left[
        \vec \beta
        \diamond
        \left(
          \bigvee_{u \in S[\vec \nu/ \vec \beta]}
          s[\vec \nu/ \vec \beta] = \mac{u}{\mkey}{j}
        \right)_{\vec \nu \in \lval_{\vec \beta}}
      \right]
    \right)\displaybreak[0]\\
    &\;\ra\;\;&&
    \left(
      C\left[
        \vec \beta
        \diamond
        \left(
          \left(
            \bigvee_{u \in S}
            s = \mac{u}{\mkey}{j}
          \right)
          [\vec \nu/ \vec \beta]
        \right)_{\vec \nu \in \lval_{\vec \beta}}
      \right]
    \right)\\
    &\;\ra\;\;&&
    \bigvee_{u \in S}
    s = \mac{u}{\mkey}{j}
  \end{alignat*}
  This concludes this proof.
\end{proof}

\subsection{Additional Axioms}
We present additional axioms, and show that they are logical consequences of the axioms $\axioms$.

An if-context is a context build using only the $\symite$ function symbol, and hole variables. We also require that no $\symite$ function symbol appears in a conditional position of another $\symite$ function symbol. Moreover, we split the hole variables in two disjoint sets: the set $\vec x$ of hole variables appearing in a conditional position, and the set $\vec y$ of variables appearing in leave position. Formally:
\begin{definition}
  For all distinct variables $\vec x, \vec y$, an if-context $D[]_{\vec x\diamond \vec y}$ is a context in $\mathcal{T}(\ite{\_}{\_}{\_},\{[]_z \mid z \in \vec x \cup \vec y\})$ such that for all position $p$, $D_{|p} \equiv \ite{b}{u}{v}$ implies:
  \begin{itemize}
  \item $b \in \{ []_z\mid z \in \vec x\}$
  \item $u,v \not \in  \{ []_z\mid z \in \vec x\}$
  \end{itemize}
\end{definition}

\begin{definition}
  For every nonces $\nonce_0,\dots,\nonce_l$, for every ground terms $\vec u$, we let $\fresh{\nonce_0,\dots,\nonce_l}{\vec u}$ holds if and only if for every $0 \le i \le l$, $\nonce_i \not \in \st(\vec u)$.
\end{definition}

\paragraph{The \textsf{indep-branch} Axioms}
This is useful to define the $\textsf{indep-branch}$ axiom. Let $\vec u$, $\vec b$ be ground terms, $C$ an if-context and $\nonce,(\nonce_i)_{i \in I}$ nonces. If $\nonce,(\nonce_i)_{i \in I}$ are distinct and such that $\fresh{\nonce,(\nonce_i)_{i \in I}}{\vec u,\vec b,C[]}$. Then the following inference rule is an instance of the $\textsf{indep-branch}$ axiom:
\[
  \infer[\textsf{indep-branch}]
  {\vec u, C\left[\vec b \diamond (\nonce_i)_{i \in I}\right] \sim \vec u,\nonce}{}
\]

\begin{proposition}
  The $\textsf{indep-branch}$ axioms are a consequence of the $\axioms$ axioms.
\end{proposition}

\begin{proof}
  TO prove this, we first introduce the if-context $C$ on the right to match the shape of the left side. We then split the proof using $\cs$, and conclude by applying $\ax{Fresh}$. This yields the derivation:
  \[
    \infer[R]{
      \vec u, C\left[\vec b \diamond (\nonce_i)_{i \in I}\right] \sim
      \vec u,\nonce
    }{
      \infer[\cs^*]{
        \vec u, C\left[\vec b \diamond (\nonce_i)_{i \in I}\right] \sim
        \vec u, C\left[\vec b \diamond (\nonce)_{i \in I}\right]
      }{
        \infer[\ax{Fresh}]{
          \forall i \in I,
          \vec u, \vec b, \nonce_i \sim
          \vec u, \vec b, \nonce
        }{}
      }
    }\qedhere
  \]
\end{proof}

\paragraph{Function Application Under Context}
It is often convenient to apply the $\fa$ axiom under an if-context $C$. Formally, let $\vec v,\vec b,(u_{i,j})_{i \in I, 1 \le j \le n},(u'_{i,j})_{i \in I, 1 \le j \le n}$ be terms and $C$ an if-context. Then the following inference rule is an instance of the $\fac$ axiom:
\[
  \infer[\fac]{
    \vec v, C\left[\vec b \diamond
      \left(
        f((u_{i,j})_{1 \le j \le n})
        \right)_{i \in I}\right]
    \sim
    \vec v', C\left[\vec b' \diamond
      \left(f((u'_{i,j})_{1 \le j \le n})
        \right)_{i \in I}\right]
  }{
    \vec v,
    \left(
      C\left[\vec b \diamond (u_{i,j})_{i \in I}\right]
    \right)_{1 \le j \le n}
    \sim
    \vec v',
    \left(
      C\left[\vec b' \diamond (u'_{i,j})_{i \in I}\right]
    \right)_{1 \le j \le n}
  }
\]

\begin{proposition}
  The $\fac$ axioms are a consequence of the $\axioms$ axioms.
\end{proposition}

\begin{proof}
  First, we pull the $f$ function outside of the if-context $C$ using the homomorphism properties if $\symite$. Finally we apply the $\fa$ axiom. This yields the derivation:
  \[
    \infer[R]{
      \vec v, C\left[\vec b \diamond
        \left(
          f((u_{i,j})_{1 \le j \le n})
        \right)_{i \in I}\right]
      \sim
      \vec v', C\left[\vec b' \diamond
        \left(f((u'_{i,j})_{1 \le j \le n})
        \right)_{i \in I}\right]
    }{
      \infer[\fa]{
        \vec v,
        f\left(
          C\left[\vec b \diamond (u_{i,j})_{i \in I}\right]
        \right)_{1 \le j \le n}
        \sim
        \vec v',
        f\left(
          C\left[\vec b' \diamond (u'_{i,j})_{i \in I}\right]
        \right)_{1 \le j \le n}
      }{
        \vec v,
        \left(
          C\left[\vec b \diamond (u_{i,j})_{i \in I}\right]
        \right)_{1 \le j \le n}
        \sim
        \vec v',
        \left(
          C\left[\vec b' \diamond (u'_{i,j})_{i \in I}\right]
        \right)_{1 \le j \le n}
      }
    }
    \qedhere
  \]
\end{proof}

\paragraph{Program Constants}
\begin{definition}
  We define the set $\cstdom$ to be the set of program constant, with includes the set of agent names $\iddom$ and the constants $\unknownid$ and $\fail$:
  \[
    \cstdom \; := \;
    \iddom \cup
    \left\{
      \bot,\unknownid,\fail, 0, 1
    \right\}
  \]
\end{definition}

\begin{proposition}
  \label{prop:enc-neq}
  For every term $u,v$ we have, the following axiom is a consequence of the axioms $\axioms$:
  \begin{equation*}
    \label{eq:enc-agent}
    \infer{
      \eq{\enc{u}{\pk(\nonce)}{\enonce}}
      {\enc{v}{\pk(\nonce)}{\enonce'}} \peq \false
    }{
      \length(u) \peq \length(v)
      \;\;&\;\;
      \length(u) \not \peq 0
    }
    \qquad \text{ when }
    \begin{dcases}
      \enonce \not \equiv \enonce'\\
      \fresh{\enonce,\enonce'}{u, v}\\
      \nonce \tpos_{\pk(\cdot),\sk(\cdot)} u, v
      \;\wedge\; \sk(\nonce) \tpos_{\dec(\_,\cdot)} u, v
    \end{dcases}
  \end{equation*}
\end{proposition}

\begin{proof}
  We give directly the derivation:
  \[
    \footnotesize
    \infer[\trans]{
      \eq{\enc{u}{\pk(\nonce)}{\enonce}}
      {\enc{v}{\pk(\nonce)}{\enonce'}} \peq \false
    }{
      \infer[\fa]{
        \eq{\enc{u}{\pk(\nonce)}{\enonce}}
        {\enc{v}{\pk(\nonce)}{\enonce'}} \peq
        \eq{\enc{u}{\pk(\nonce)}{\enonce}}
        {\enc{1^{\length(v)}}{\pk(\nonce)}{\enonce'}}
      }{
        \infer[\restr]{
          \enc{u}{\pk(\nonce)}{\enonce},
          \enc{v}{\pk(\nonce)}{\enonce'} \sim
          \enc{u}{\pk(\nonce)}{\enonce},
          \enc{1^{\length(v)}}{\pk(\nonce)}{\enonce'}
        }{
          \infer[\ccao]{
            \pk(\nonce),
            \enc{u}{\pk(\nonce)}{\enonce},
            \enc{v}{\pk(\nonce)}{\enonce'} \sim
            \pk(\nonce),
            \enc{u}{\pk(\nonce)}{\enonce},
            \enc{1^{\length(v)}}{\pk(\nonce)}{\enonce'}
          }{
            \infer[\refl]
            {
              \pk(\nonce),\enc{u}{\pk(\nonce)}{\enonce},\length(v) \sim
              \pk(\nonce),\enc{u}{\pk(\nonce)}{\enonce},\length(v)
            }{}
            \;\;&\;\;
            \infer[]{
              \length(v) \peq \length(1^{\length(v)})
            }{
              \infer[\refl]{\length(v) \peq \length(v)}{}
            }
          }
        }
      }
      &
      \eq{\enc{u}{\pk(\nonce)}{\enonce}}
      {\enc{1^{\length(v)}}{\pk(\nonce)}{\enonce'}} \peq \false
    }
  \]
  To show $\eq{\enc{u}{\pk(\nonce)}{\enonce}} {\enc{1^{\length(v)}}{\pk(\nonce)}{\enonce'}} \peq \false$, we apply again transitivity:
  \[
    \infer[\trans]{
      \eq{\enc{u}{\pk(\nonce)}{\enonce}}
      {\enc{1^{\length(v)}}{\pk(\nonce)}{\enonce'}} \peq \false
    }{
      \eq{\enc{u}{\pk(\nonce)}{\enonce}}
      {\enc{1^{\length(v)}}{\pk(\nonce)}{\enonce'}} \peq
      \eq{\enc{0^{\length(u)}}{\pk(\nonce)}{\enonce}}
      {\enc{1^{\length(v)}}{\pk(\nonce)}{\enonce'}}
      \;\;&\;\;
      \eq{\enc{0^{\length(u)}}{\pk(\nonce)}{\enonce}}
      {\enc{1^{\length(v)}}{\pk(\nonce)}{\enonce'}} \peq
      \false
    }
  \]
  Now, we give the derivation of the left premise:
  \[
    \infer[\fa]{
      \eq{\enc{u}{\pk(\nonce)}{\enonce}}
      {\enc{1^{\length(v)}}{\pk(\nonce)}{\enonce'}} \peq
      \eq{\enc{0^{\length(u)}}{\pk(\nonce)}{\enonce}}
      {\enc{1^{\length(v)}}{\pk(\nonce)}{\enonce'}}
    }{
      \infer[\restr]{
        \enc{u}{\pk(\nonce)}{\enonce}
        \enc{1^{\length(v)}}{\pk(\nonce)}{\enonce'} \sim
        \enc{0^{\length(u)}}{\pk(\nonce)}{\enonce}
        \enc{1^{\length(v)}}{\pk(\nonce)}{\enonce'}
      }{
        \infer[\ccao]{
          \pk(\nonce),
          \enc{u}{\pk(\nonce)}{\enonce}
          \enc{1^{\length(v)}}{\pk(\nonce)}{\enonce'} \sim
          \pk(\nonce),
          \enc{0^{\length(u)}}{\pk(\nonce)}{\enonce}
          \enc{1^{\length(v)}}{\pk(\nonce)}{\enonce'}
        }{
          \infer[\refl]
          {
            \pk(\nonce),
            \enc{1^{\length(v)}}{\pk(\nonce)}{\enonce'},
            \length(u) \sim
            \pk(\nonce),
            \enc{1^{\length(v)}}{\pk(\nonce)}{\enonce'},
            \length(u)
          }{}
          \;\;&\;\;
          \infer[]{
            \length(u) \peq \length(0^{\length(u)})
          }{
            \infer[\refl]{\length(u) \peq \length(u)}{}
          }
        }
      }
    }
  \]
And finally we prove the right premise $\eq{\enc{0^{\length(u)}}{\pk(\nonce)}{\enonce}} {\enc{1^{\length(v)}}{\pk(\nonce)}{\enonce'}} \peq\false$:
  \[
    \infer[\textsf{EQInj}(\enc{\cdot}{\_}{\_}) + R]{
      \eq{\enc{0^{\length(u)}}{\pk(\nonce)}{\enonce}}
      {\enc{1^{\length(v)}}{\pk(\nonce)}{\enonce'}} \peq
      \false
    }{
      \infer[\textsf{l-neq}]{
        \eq{0^{\length(u)}}{1^{\length(v)}} \peq \false
      }{
        \infer[\textsf{EQConst}]{\eq{0}{1} \peq \false}{}
        \;\;&\;\;
        \unary{\length(0) \not \peq 0}
        \;\;&\;\;
        \length(u) \not \peq 0
      }
    }
    \qedhere
  \]

\end{proof}

\FloatBarrier


%% file: axiom-figure.tex
\begin{figure}[t]
  \begin{mathpar}
    \infer[\sym]{\vec u \sim \vec v}{\vec v \sim \vec u}

    \infer[\refl]{\vec u \sim \vec u}{}

    \infer[\trans]{
      \vec u \sim \vec v
    }{
      \vec u \sim \vec w & \vec w \sim \vec v
    }

    \infer[\perm]{
      (x_i)_{i \le n} \sim (y_i)_{i \le n}
    }{
      (x_{\pi(i)})_{i \le n} \sim (y_{\pi(i)})_{i \le n}
    }

    \infer[\dup]{
      \vec u,t,t \sim \vec v,t',t'
    }{
      \vec u,t \sim \vec v,t'}

    \infer[\restr]{\vec u \sim \vec v}{\vec u,t \sim \vec v,t'}

    \begin{array}[c]{c}
      \infer[\fa]{f(\vec{x}),\vec{y} \sim f(\vec{x'}), \vec{y'}}
      {\vec{x},\vec{y}\sim \vec{x'}, \vec{y'}}
    \end{array}

    \begin{array}[c]{c}
      \infer[R]{
        \vec u, t \sim \vec v
      }{
        \vec u, s \sim \vec v
        \;&\;
        s \peq t
      }
    \end{array}

    \begin{array}[c]{c}
      \infer[\ax{Fresh}]{
        \vec u, \nonce \sim \vec v, \nonce
      }{
        \vec u \sim \vec v
      }
    \end{array}
    \text{when $\nonce \not \in \st(\vec u,\vec v)$}

    \begin{array}[c]{c}
      \infer[\oplus\textsf{-indep}]{
        \vec u, t \oplus \nonce \sim \vec v
      }{
        \vec u, \nonce \sim \vec v
      }
    \end{array}
    \text{when $\nonce \not \in \st(\vec u,t)$}

    \begin{array}[c]{c}
      \infer[\cs]{
        \vec w, (\ite{b}{u_i}{v_i})_i
        \sim
        \vec w',(\ite{b'}{u_i'}{v_i'})_i}
      {
        \vec w, b, (u_i)_i \sim \vec w', b', (u_i')_i \quad & \quad
        \vec w, b, (v_i)_i \sim \vec w', b', (v_i')_i
      }
    \end{array}
  \end{mathpar}
  \textbf{Conventions:} $\pi$ is a permutation of $\{1,\ldots, n\}$ and $f \in \sig$.
  \caption{\label{figure:axioms} The Structural Axioms $\axioms_{\textsf{struct}}$.}
\end{figure}

\begin{figure}[p]
  \textbullet\ \textit{The set $\axioms_{\textsf{ite}}$ of equality axioms related to the $\symite$ function symbol:}
  \begin{mathpar}
    \unary{f(\vec{u},\ite{b}{x}{y},\vec v) \peq
      \ite{b}{f(\vec{u},x,\vec v)}{f(\vec{u},y,\vec v)}}
    \text{ for any } f \in \ssig

    \unary{\ite{(\ite{b}{a}{c})}{x}{y} \peq
    \ite{b}{(\ite{a}{x}{y})}{(\ite{c}{x}{y})}}

    \unary{\ite{\true}{x}{y} \peq x}

    \unary{\ite{\false}{x}{y} \peq y}

    \unary{\ite{b}{x}{x} \peq x}

    \unary{\ite{b}{(\ite{b}{x}{y})}{z} \peq \ite{b}{x}{z}}

    \unary{\ite{b}{x}{(\ite{b}{y}{z})} \peq \ite{b}{x}{z}}

    \unary{\ite{b}{(\ite{a}{x}{y})}{z} \peq
    \ite{a}{(\ite{b}{x}{z})}{(\ite{b}{y}{z})}}

    \unary{\ite{b}{x}{(\ite{a}{y}{z})} \peq
    \ite{a}{(\ite{b}{x}{y})}{(\ite{b}{x}{z})}}
  \end{mathpar}

  \textbullet\ \textit{The set $\axioms_{\textsf{eq}}$ of equality and disequality axiom:}
  \begin{mathpar}
    \unary{\pi_i (\pair{x_1}{x_2}) \peq x_i} \text{ for } i \in \{1,2\}
      
    \unary{\pi_i (\triplet{x_1}{x_2}{x_3}) \peq x_i} \text{ for } i \in \{1,2,3\}
        
    \unary{\dec(\enc{x}{\pk(y)}{z},\sk(y)) \peq x}

    \unary{\eq{x}{x} \peq \true}

    \infer[\textsf{$\ne$-Const}]{\eq{\agent{A}}{\agent{B}} \peq \false}{}
    \begin{array}[c]{l}
      \text{ for every }\agent{A},\agent{B} \in \cstdom\\
      \text{ s.t. }
      \agent{A} \not \peq \agent{B}
    \end{array}
    
    \infer[\ax{EQIndep}]{\eq{t}{\nonce} \peq \false}{}
    \text{ if } \nonce \not \in \st(t)
  \end{mathpar}

  \textbullet\ \textit{The set $\axioms_{\textsf{len}}$ of length axioms:}
  \begin{mathpar}
    \infer[]{
      \length(\pair{u}{v}) \peq \length(\pair{s}{t})
    }{
      \length(u) \peq \length(s)
      \quad&\quad
      \length(v) \peq \length(t)
    }

    \unary{\length(\ID_1) \peq \length(\ID_2)}
    \text{ for every } \ID,\ID' \in \iddom

    \unary{\length(\sqnsuc(\sqnini_\ue^\ID)) \peq \length(\sqnini_\ue^\ID)}

    \unary{\length(\sqnini_\ue^{\ID_1}) \peq \length(\sqnini_\ue^{\ID_2})}

    \unary{\length(0^x) \peq x}

    \unary{\length(1^x) \peq x}

    \unary{\length(\sfx) \not \peq 0}
    \text{ when }\sfx \in \cstdom

    \infer{
      \length(\pair{u}{v}) \not \peq 0
    }{
      \length(u) \not \peq 0
    }

    \infer{
      \length(\pair{u}{v}) \not \peq 0
    }{
      \length(v) \not \peq 0
    }

    \infer[\textsf{l-neq}]{
      A^x \not \peq B^y
    }{
      A \not \peq B
      \;&\;
      \length(A) \not \peq 0
      \;&\;
      x \not \peq 0
    }
  \end{mathpar}
  
  \textbullet\ \textit{The set $\axioms_{\textsf{inj}}$ of injectivity axioms:}
    \begin{mathpar}


      \infer[\textsf{EQInj}(\pair{\cdot}{\_})]{
        \neg \eq{u}{s} \wedge \eq{\pair{u}{v}}{\pair{s}{t}} \peq \false}{}

      \infer[\textsf{EQInj}(\pair{\_}{\cdot})]
      {\neg \eq{v}{t} \wedge \eq{\pair{u}{v}}{\pair{s}{t}} \peq \false}{}

      \infer[\textsf{EQInj}(\enc{\cdot}{\_}{\_})]
      {\neg \eq{u}{v} \wedge
        \eq{\enc{u}{\pk(\nonce)}{\enonce}}
        {\enc{v}{\pk(\nonce')}{\enonce'}} \peq \false}{}
  \end{mathpar}

  \textbullet\ \textit{The set $\axioms_{\sqn}$ of sequence number axioms:}
  \begin{mathpar}
    \unary{\range{u}{v} \peq \eq{u}{v}}

    \unary{\sqnsuc(u) \peq u + 1}

    \infer[$\sqn$\textsf{-ini}]{\sqnini_\hn^\ID \le \sqnini_\ue^\ID}{}

    \unary{\phi[\vec u] \peq \true}
    \begin{array}[c]{l}
      \text{when }
      \vec u \text{ are ground terms}\\
      \text{and }
      \text{Th}(\mathbb{Z},0,+,-,=,\le) \models \phi[\vec x\,]
    \end{array}

  \end{mathpar}
  
  \caption{\label{fig:axioms-add} The Set of Axiom $\axioms_{\textsf{impl}} = \axioms_{\textsf{ite}} \cup \axioms_{\textsf{eq}} \cup \axioms_{\textsf{len}} \cup \axioms_{\textsf{inj}} \cup \axioms_{\sqn} $.}
\end{figure}


%% file: symbolic.tex
\section{Protocol}
\FloatBarrier

\subsection{Symbolic Protocol}
\label{appendix:symbolic-protocol}

\begin{figure}[t]
  \begin{center}
    \begin{tikzpicture}[>=stealth]
      \tikzset{mn/.style={fill=white,draw,
          rounded corners,minimum height=2em,minimum width=4em}};

      \node[right] at (0,1) {\underline{Transition System $\mathcal{Q}_\ue^\ID$:}};

      \node[mn,draw=none,anchor=west]
      (a) at (0,0) {\footnotesize$\mathcal{E}_\ID^{\le j-1}$};
      \path (a) -- ++ (2,0)
      node[mn] (b) {\footnotesize$\npuai{0}{\ID}{j}$};
      \path (b) -- ++ (2,0)
      node[mn] (c0) {\footnotesize$\npuai{1}{\ID}{j}$};
      \path (c0) -- ++ (2,0)
      node[mn] (c) {\footnotesize$\npuai{2}{\ID}{j}$};
      
      \path (c) -- ++ (2,-0.8)
      node[mn] (d) {\footnotesize$\fuai_\ID(j)$};

      \node[mn,draw=none,anchor=west]
      (e) at (0,-1.6) {\footnotesize$\mathcal{E}_\ID^{\le j-1}$};
      \path (e) -- ++ (2.5,0)
      node[mn] (f) {\footnotesize$\cuai_\ID(j,0)$};
      \path (f) -- ++ (2.5,0)
      node[mn] (g) {\footnotesize$\cuai_\ID(j,1)$};

      \node[mn,draw=none,anchor=west]
      (h) at (0,-3.2) {\footnotesize$\mathcal{E}_\ID^{\le j-1}$};
      \path (h) -- ++ (2.5,0)
      node[mn] (i) {\footnotesize$\newsession_\ID(j)$};

      \draw[->] (a) -- (b);
      \draw[->] (b) -- (c0);
      \draw[->] (c0) -- (c);
      \draw[->] (c) -- (d);

      \draw[->] (e) -- (f);
      \draw[->] (f) -- (g);
      \draw[->] (g) -- (d);

      \draw[->] (h) -- (i);

      \node[right] at (10.5,1) {\underline{Transition System $\mathcal{Q}_\hn^j$:}};

      \node[mn] (b) at (11.5,-0.8) {\footnotesize$\pnai(j,0)$};
      \path (b) -- ++ (2.5,0)
      node[mn] (c) {\footnotesize$\pnai(j,1)$};
      \path (c) -- ++ (2.5,-0.8)
      node[mn] (d) {\footnotesize$\fnai(j)$};

      \node[mn] (f) at (11.5,-2.4) {\footnotesize$\cnai(j,0)$};
      \path (f) -- ++ (2.5,0)
      node[mn] (g) {\footnotesize$\cnai(j,1)$};

      \draw[->] (b) -- (c);
      \draw[->] (c) -- (d);

      \draw[->] (f) -- (g);
      \draw[->] (g) -- (d);

      \draw[thick] (10,0.6) -- ++(0,-4);;
    \end{tikzpicture}
  \end{center}
  \textbf{Convention:} where $\mathcal{E}_\ID^{\le j} = \{ \npuai{i}{\ID}{j_0}, \cuai_\ID(j_0,i), \fuai_\ID(j_0), \newsession_\ID(j_0) \mid j_0 \le j\}$, the initial states of $\mathcal{Q}_\ue^\ID$ are $\npuai{1}{\ID}{0}$ and $\cuai_\ID(0,0)$, and the initial states of $\mathcal{Q}_\hn^j$ are $\pnai(j,0)$ and $\cnai(j,0)$. Every state of $\mathcal{Q}_\ue^\ID$ or $\mathcal{Q}_\hn^j$ is final.
  \caption{\label{fig:as-ts}The transition systems used to define valid symbolic traces.}
\end{figure}

In this section we formally define the symbolic traces of the $\faka$ protocol, as well as some functions and properties of these traces.

We recall that $\iddom = \{\ID_1,\dots,\ID_N\}$ is the set of identities used in the protocol. We split these identities between base identities, which are used by the normal protocol, and copies of the base identities, which we use to express the $\sigma_{\sunlink}$-unlinkability of the protocol. We have $B$ base identities $\agent{A}_1,\dots,\agent{A}_{B}$, and we let $\baseiddom$ be the set of base identities. Then, for every base identity $\agent{A}_i$, we have $C$ copies $\agent{A}_{i} = \agent{A}_{i,1},\dots,\agent{A}_{i,C}$ of $\agent{A}_i$. In total we use $N = B \times C$ distinct identities, and $\ID_1,\dots,\ID_N$ is an arbitrary enumeration of all the identities.

\paragraph{Valid Symbolic Trace}
We recall that an symbolic trace is a sequence of action identifiers, which symbolically represents calls from the adversary to the oracles. Remark that some sequence of action identifiers do not correspond to a valid execution of the protocol. E.g., since the session $\tue_{\ID}(j)$ cannot execute both the $\supi$ and the $\guti$ protocols, a \emph{valid symbolic trace} cannot contain both $\npuai{\_}{\ID}{j}$ and $\cuai_\ID(j,\_)$. Similarly, the $\thn$'s second message in the $\supi$ protocol cannot be sent before the first message, hence $\pnai(j,1)$ cannot appear before $\pnai(j,0)$ in $\tau$. Formally:
\begin{definition}
  Let $(\mathcal{Q}_\ue^{\ID})_{\ID \in \iddom}$ and $(\mathcal{Q}_\hn^j)_{j \in \mathbb{N}}$ be the automatas given in Fig.~\ref{fig:as-ts}. A symbolic trace $\tau = \ai_0,\dots,\ai_n$ is a \emph{valid} symbolic trace iff $\tau$ is an inter-leaving of the words $ w_{\ID^1},\dots,w_{\ID^N},w_\hn^0,\dots,w_\hn^l,\dots$ where:
  \begin{itemize}
  \item for every $1 \le j \le N$, $w_\ID^j$ is a run of $\mathcal{Q}_\ue^{\ID_j}$.
  \item for every $j \in \mathbb{N}$, $w_\hn^j$ is a run of $\mathcal{Q}_\hn^j$.
  \item for every $j \in \mathbb{N}$ and $1 \le i \le n$, if $\ai_i$ is a state of $\mathcal{Q}_\hn^j$ then there exists $i_0 < i$ such that $\ai_{i_0}$ is a state of~$\mathcal{Q}_\hn^{j-1}$
  \end{itemize}
  Furthermore, $\tau$ is said to be \emph{basic} if for all $1 \le j \le N$, if $w_\ID^j \ne \epsilon$ then $\ID_j$ is a base identity (i.e. $\ID_j \in \{\agent{A}_1,\dots,\agent{A}_{B}\}$).
\end{definition}

Since we only informally describe the $\faka$ protocol, we cannot formally prove that every $\tau \in \support(\runlink)$, $\tau$ is a valid symbolic trace. Instead, we put as an assumption that any implementation of the $\faka$ protocol must ensures that messages are processed as described in $\mathcal{Q}_\ue^\ID$ and $\mathcal{Q}_\hn^j$.
\begin{assumption}
  \label{ass:support-valid}
  For every $\tau \in \support(\runlink)$, $\tau$ is a valid symbolic trace
\end{assumption}

\paragraph{Modeling Unlinkability}
Given a symbolic trace $\tau \in \support(\runlink)$, there is a particular and unique symbolic trace $\utau$ which is the ``most anonymised trace'' corresponding to $\tau$. Intuitively, $\utau$ is the trace $\tau$ where we changed a user identity every time we could (i.e. every time $\ns_\ID(\_)$ appears). This is useful to prove that the $\fiveaka$ protocol is $\sigma_{\sunlink}$-unlinkable, as it reduces the number of cases we have to consider: we only need to show that we can derive $\cframe_{\tau} \sim \cframe_{\utau}$ for every $\tau \in \support(\runlink)$.
\begin{definition}
  Given an identity $\agent{A}_{b,c}$ where $c < C$, we let $\freshid(\agent{A}_{b,c}) = \agent{A}_{b,c + 1}$, and given a base identity $\agent{A}_{b,1}$ we let $\copyid(\agent{A}_{b,1}) = \{\agent{A}_{b,1} \mid 1 \le i \le C\}$.
\end{definition}

\begin{definition}
  We define some functions on symbolic traces:
  \begin{itemize}
  \item We let $\sthead$ be the function that, given a symbolic trace $\tau$, returns the last action in $\tau$ (or $\epsilon$ is $\tau$ is empty):
    \[
      \sthead(\tau) \;=\;
      \begin{dcases}
        \ai_n & \text{ if } \tau = \ai_0,\dots,\ai_n \text{ and } n \ge 0\\
        \epsilon & \text{ if } \tau = \epsilon
      \end{dcases}
    \]
  \item Given a symbolic trace $\tau$, we let $\potau$ be the restriction of $\popre$ to the set of strict prefixes of $\tau$, i.e. $\tautt \potau \taut$ iff $\tautt \popre \taut$ and $\taut \popre \tau$.
  \item We extend $\potau$ to symbolic actions as follows: we have $\ai \potau \taut$ (resp. $\taut \potau \ai$) iff there exists $\tautt$ such that $\sthead(\tautt) = \ai$ and $\tautt \potau \taut$ (resp. $\taut \potau \tautt$).
  \end{itemize}

\end{definition}

\begin{definition}
  Given a symbolic trace $\tau$ with less than $C$ actions $\newsession_\ID(\_)$ for every $\ID$, we define the symbolic trace $\ufresh{\tau}$ where each time we encounter a action $\newsession_\ID(j)$, we replace all subsequent action with agent $\ID$ by action with agent $\freshid(\ID)$:
  \[
    \ufresh{\tau} =
    \begin{cases}
      \newsession_{\nu\ID}(j),\ufresh{\tau_0[\nu\ID/\ID]} & \text{ when } \tau = \newsession_\ID(j),\tau_0 \text{ and } \nu\ID=\freshid(\ID)\\
      \ai,\ufresh{\tau_0} & \text{ when } \tau = \ai,\tau_0 \text{ and } \ai \not \in \{\newsession_\ID(j)\mid \ID \in \iddom, j \in \mathbb{N}\}
    \end{cases}
  \]
\end{definition}

One can easily check that $\runlink(\tau,\utau)$. Besides, remark that for every $(\tau_l,\tau_r) \in \runlink$ we have $\ufresh{\tau_l} = \ufresh{\tau_r}$. Moreover, $\sim$ is a transitive relation. Therefore, instead of proving that for every $\runlink(\tau_l,\tau_r)$ the formula $\cstate_{\tau_l}  \sim \cstate_{\tau_r}$ can be derived using $\axioms$, it is sufficient to show that for every $\tau \in \support(\runlink)$, we can derive $\cstate_{\tau} \sim \cstate_{\utau}$ using $\axioms$. Formally:
\begin{proposition}
  \label{prop:main-unlink-bc-app}
  The $\fiveaka$ protocol is $\sigma_\sunlink$-unlinkable in any computational model satisfying some axioms $\axioms$ if for every $\tau \in \support(\runlink)$, there is a derivation using $\axioms$ of $\cframe_{\tau} \sim \cframe_{\utau}$.
\end{proposition}

\begin{proof}
  Assume that for every $\tau \in \support(\runlink)$, we can derive using $\axioms$ the formula $\cframe_{\tau} \sim \cframe_{\utau}$. Then, using Proposition~\ref{prop:main-unlink-bc} we know that $\fiveaka$ protocol is $\sigma_\sunlink$-unlinkable in any computational model satisfying axioms $\axioms$ if for every $(\tau_l,\tau_r) \in \runlink$, we can derive $\cframe_{\tau_l} \sim \cframe_{\tau_r}$. If $\faka$ is  $\sigma_\sunlink$-unlinkable with $N$ identities then it is  $\sigma_\sunlink$-unlinkable with $N'$ identities if $N' \le N$. Therefore, w.l.o.g. we can always assume that we have more identities than actions $\ns_\ID(\_)$ in $\tau$. Hence $\utau$ is well-defined, and we know that $(\tau_l,\ufresh{\tau_l}) \in \runlink$ and $(\tau_r,\ufresh{\tau_r}) \in \runlink$. By hypothesis, we have derivations of $\cframe_{\tau_l} \sim \cframe_{\ufresh{\tau_l}}$ and $\cframe_{\tau_r} \sim \cframe_{\ufresh{\tau_r}}$. Since $\ufresh{\tau_l} = \ufresh{\tau_r}$, and using the transitivity and symmetry axioms  $\trans$ and $\sym$, we get a derivation of $\cframe_{\tau_l} \sim \cframe_{\tau_r}$. This concludes this proof.
\end{proof}

\begin{proposition}
  \label{prop:tau-tu-utau}
  If $\tau$ is a valid basic symbolic trace with less than $C$ actions $\newsession$ then $\ufresh{\tau}$ is a valid symbolic trace.
\end{proposition}

\begin{proof}
  The proof is straightforward by induction over $\tau$.
\end{proof}

\begin{definition}
  Given a basic trace $\tau$ and a basic identity $\ID = \agent{A}_{i,0}$, we let $\nu_\tau(\ID)$ be the identity $\agent{A}_{i,l}$ where $l$ is the number of occurrences of $\newsession_\ID(\_)$ in $\tau$.
\end{definition}

\begin{definition}
  Let $\tau$ be a symbolic trace of actions $\ai_0,\dots,\ai_n$. Then for all $0 \le i < n$, $\suc_\tau(\ai_i) = \ai_{i + 1}$.
\end{definition}

\begin{definition}
  We define the partial $\session$ function:
  \begin{gather*}
    \session_\hn(\ai) = j \text{ when }
    \ai = \textsc{x}(j,\_), \textsc{x} \in \{\pnai,\cnai,\fnai \}
  \end{gather*}
\end{definition}

\begin{definition}
  We let $\sessionstarted_j(\tau)$ be true if and only if there exists $\ai \in \tau$ s.t. $\session(\ai)=j$.
\end{definition}

\subsection{The $\faka$ Protocol}

To show that the $\faka$ protocol is $\sigma_\sunlink$-unlinkable, we need to know, for every identity $\ID \in \iddom$, if there was a successful $\supi$ session since the last $\ns_\ID(\_)$. To do this, we extend the set of variables $\vardom$ by adding a phantom variable $\sync_\ue^\ID$ for every $\ID \in \iddom$. We also extend the symbolic state updates of $\ns_\ID(\_)$ and $\npuai{2}{\ID}{j}$ as follows:
\begin{itemize}
\item For $\ai = {\newsession_\ID(j)}$:
  \[
    \upstate_\tau \;\equiv\;
    \begin{dcases}
      \success_\ue^\ID \mapsto \false\\
      \sync_\ue^\ID \mapsto \false
    \end{dcases}
  \]
\item For $\ai = {\npuai{2}{\ID}{j}}$:
  \[
    \upstate_\tau \;\equiv\;
    \begin{dcases}
      \eauth_\ue^\ID \mapsto
      \ite{\accept_\tau^\ID}
      {\instate_\tau(\bauth_\ue^\ID)}{\fail}\\
      \sync_\ue^\ID \mapsto \instate_\tau(\sync_\ue^\ID) \vee \accept_\tau^\ID
    \end{dcases}
  \]
\end{itemize}
Remark that the variable $\sync_\ue^\ID$ is read only to update its value. It is not used in the actual protocol. By consequence, the $\faka$ protocol is  $\sigma_\sunlink$-unlinkable if and only if the extended $\faka$ protocol is $\sigma_\sunlink$-unlinkable.

We now give the definition of the initial symbolic state $\cstate_\epsilon$, which we omitted in the body:
\begin{definition}
  \label{def:init-sigma-phi}
  The symbolic state $\cstate_\epsilon$ is the function from $\vardom$ to terms defined by having, for every $\ID \in \iddom$ and $j \in \mathbb{N}$:
  \begin{mathpar}
    \cstate_\epsilon(\sqn^\ID_\ue) \equiv \sqnini_\ue^\ID

    \cstate_\epsilon(\sqn^\ID_\hn) \equiv \sqnini_\hn^\ID

    \cstate_\epsilon(\suci^\ID_\scx) \equiv \unset

    \cstate_\epsilon(\eauth^\ID_\ue) \equiv \fail

    \cstate_\epsilon(\bauth^\ID_\ue) \equiv \fail

    \cstate_\epsilon(\eauth^j_\hn) \equiv \fail

    \cstate_\epsilon(\bauth^j_\hn) \equiv \fail
    
    \cstate_\epsilon(\uetsuccess^{\ID}) \equiv \false
    
    \cstate_\epsilon(\success^{\ID}_\ue) \equiv \false

    \cstate_\epsilon(\tsuccess_\hn^{\ID}) \equiv \false

    \cstate_\epsilon(\sync^{\ID}) \equiv \false
  \end{mathpar}
\end{definition}

\subsection{Invariants and Necessary Acceptance Conditions}
\label{subsection:inv-nec-cond}

\paragraph{Notations}
From now on, the set of axioms $\axioms$ is fixed, and we stop specify the set of axioms used: we say that we have a derivation of a formula $\phi$ to mean that $\phi$ can be deduced from $\axioms$. Furthermore, we say that $\phi$ holds when there is a derivation of $\phi$.

Moreover, we abuse notations and write $u = v$ instead of $u \peq v$. We can always disambiguate using the context: if we expect a term, then $u = v$ stands for the term $\eq{u}{v}$, whereas if a formula is expected then $u = v$ stands for $\eq{u}{v} \sim \true$. We extends this to any boolean term: if $b$ is a boolean term then we say that $b$ holds if we can show that $b \sim \true$ holds. For example, $\cstate_\tau(\sqn_\ue^\ID) \ge \cstate_\tau(\sqn_\hn^\ID)$ holds if we can show that $\Geq{\cstate_\tau(\sqn_\ue^\ID)}{\cstate_\tau(\sqn_\hn^\ID)} \sim \true$.

\paragraph{Properties}
We now start to state and prove properties of the $\faka$ protocol.
\begin{proposition}
  \label{prop:len-eq-sqn}
  For every valid symbolic trace $\tau$, for every $\ID_1,\ID_2 \in \iddom$, we have a derivation of:
  \[
    \unary{
      \length(\instate_{\tau}(\sqn_\ue^{\ID_1}))
      =
      \length(\instate_{\tau}(\sqn_\ue^{\ID_2}))}
  \]
\end{proposition}

\begin{proof}
  It is easy to show by induction over $\tau$ that for every $\ID \in \iddom$, there exists an if-context $C$, terms $\vec b$ and integers $(k_i)_i$ such that:
  \[
    \instate_{\tau}(\sqn_\ue^{\ID}) =
    C[\vec b \diamond (\sqnsuc^{k_i}(\sqnini_\ue^\ID))_i)]
  \]
  Therefore, let $C_1,C_2$, $\vec b_1,\vec b_2$ and $(k^1_i)_i,(k^2_j)_j$ be such that:
  \begin{mathpar}
    \instate_{\tau}(\sqn_\ue^{\ID_1}) =
    C_1[\vec b_1 \diamond (\sqnsuc^{k^1_i}(\sqnini_\ue^{\ID_1}))_i)]

    \instate_{\tau}(\sqn_\ue^{\ID_2})) =
    C_2[\vec b_2 \diamond (\sqnsuc^{k^2_j}(\sqnini_\ue^{\ID_2}))_j)]
  \end{mathpar}
  Moreover, it is trivial to show using the axioms in $\axioms_{\textsf{len}}$ that for every $i,i',j,j'$:
  \[
    \length(\sqnsuc^{k^1_i}(\sqnini_\ue^{\ID_1})) =
    \length(\sqnsuc^{k^1_{i'}}(\sqnini_\ue^{\ID_1})) =
    \length(\sqnsuc^{k^2_j}(\sqnini_\ue^{\ID_2})) =
    \length(\sqnsuc^{k^2_{j'}}(\sqnini_\ue^{\ID_2}))
  \]
  It is then easy, using $R$, to get a derivation of:
  \[
    \unary{\length(\instate_{\tau}(\sqn_\ue^{\ID_1}))
      =
      \length(\instate_{\tau}(\sqn_\ue^{\ID_2}))}
    \qedhere
  \]
\end{proof}

The following proposition states that $\nonce_\hn$ appears only in the $\thn$ public key $\pk(\nonce_\hn)$ and secret key $\sk(\nonce_\hn)$, and that for every $\ID \in \iddom$, the keys $\key^\ID$ and $\mkey^\ID$ appear only in key position in $\macsym^1$--$\,\macsym^5$. These properties will be useful to apply the cryptographic axioms later.
\begin{proposition}[Invariant \textsc{(inv-key)}]
  For all valid symbolic trace $\tau$, we have:
  \begin{alignat*}{2}
    &&&\nonce_\hn \tpos_{\pk(\cdot),\sk(\cdot)} \cframe_\tau
    \;\wedge\;
    \sk(\nonce_\hn) \tpos_{\dec(\_,\cdot)} \cframe_\tau\\
    &\forall 1 \le i \le N,\quad&&
    \mkey^{\ID_i} \tpos_{\mac{\_}{\cdot}{\_}} \cframe_\tau\\
    &\forall 1 \le i \le N,\quad&&
    \key^{\ID_i} \tpos_{\ow{\_}{\cdot},\row{\_}{\cdot}} \cframe_\tau
  \end{alignat*}
\end{proposition}

\begin{proof}
  The proof is straightforward by induction on $\tau$.  
\end{proof}

\begin{proposition}
  \label{prop:unset-prop}
  For every valid symbolic trace $\tau$, for every $\tautt \potau \taui$ and identity $\ID \in \iddom$, we have:
  \[
    \left(
      \cstate_\tautt(\suci_\ue^\ID) = \unset \wedge
      \bigwedge_{\taut = \_,\fuai_\ID(j_1)\atop{\tautt \potau \taut \potau \taui}}
      \neg \accept_{\taut}^{\ID}
    \right)
    \ra
    \instate_\taui(\suci_\ue^\ID) = \unset
  \]
\end{proposition}
\begin{proof}
  The proof is straightforward by induction on $\taui$.
\end{proof}

We now state several simple properties of our system.
\begin{proposition}
  \label{prop:invs}
  Let $\tau = \_,\ai$ be a valid symbolic trace, then:
  \begin{enumerate}
  \item \customlabel{a1}{\lpa{1}} If $\neg \sessionstarted_j(\tau)$ then $n^j \not \in \st(\cframe_\tau)$.
  \item \customlabel{a2}{\lpa{2}} For all $\tauo = \_,\npuai{2}{\ID}{j_0} \popreleq \tau$ and $\taut = \_,\npuai{2}{\ID}{j_1} \popreleq \tau$, if $\tauo \ne \taut$ then:
    \[
      \instate_\tauo(\sqn_\ue^\ID) \ne \instate_\taut(\sqn_\ue^\ID)
    \]
  \item \customlabel{a3}{\lpa{3}} For every $\tauo = \_, \npuai{2}{\ID}{j_0}$, $\taut = \_, \npuai{1}{\ID}{j_1}$ such that $\taut \potau \tauo$, if $j_0 \ne j_1$ then:
    \[
      \instate_\tauo(\sqn_\ue^\ID) \ne \sqnsuc(\instate_\taut(\sqn_\ue^\ID))
    \]
  \item \customlabel{a4}{\lpa{4}} For all $\ID_0 \ne \ID_1$,
    \[
      \left(
        \instate_\tau(\suci_\hn^{\ID_0}) \ne \unset
        \wedge
        \instate_\tau(\suci_\hn^{\ID_1}) \ne \unset
      \right)
      \;\ra\;
      \instate_\tau(\suci_\hn^{\ID_0}) \ne
      \instate_\tau(\suci_\hn^{\ID_1})
    \]
  \item \customlabel{a5}{\lpa{5}}, \customlabel{a6}{\lpa{6}}, \customlabel{a7}{\lpa{7}} If $\sthead(\tau) = \pnai(j,1), \cnai(j,0)$ or $\cnai(j,1)$, then for every $\ID_0 \ne \ID_1$,
    \[
      \big(\neg \accept_\tau^{\ID_0}\big) \vee \big(\neg \accept_\tau^{\ID_1}\big)
    \]
  \item \customlabel{a8}{\lpa{8}} For every $\ID \in \iddom, j \in \mathbb{N}$, $\instate_\tau(\eauth_\ue^\ID) = \nonce^j  \ra \instate_\tau(\bauth_\ue^\ID) = \nonce^j$.
  \end{enumerate}
\end{proposition}

\begin{proof}
  \begin{itemize}
    All these properties are simple to show:
  \item \ref{a1} is trivial by induction over $\tau$.

  \item \ref{a2} and \ref{a3} both follow from the fact that if $\tau = \_,\npuai{1}{\ID}{j}$ then $\cstate_\tau(\sqn_\ue^\ID) \equiv \sqnsuc(\instate_\tau(\sqn_\ue^\ID))$, and therefore $\cstate_\tau(\sqn_\ue^\ID) > \instate_\tau(\sqn_\ue^\ID)$.

  \item \ref{a5} and \ref{a7} follow easily from the unforgeability axioms $\textsc{euf-mac}$.

  \item To prove \ref{a4}, we first observe that for every $\ID \in \iddom$, we initially have $\cstate_\epsilon(\suci_\hn^{\ID}) \equiv \unset$, and that the only value we store in $\suci_\hn^\ID$ are $\unset$ or $\guti^i$ for some $i \in \mathbb{N}$. Therefore it is easy to show that for every $\taun \popre \tau$:
    \[
      \instate_\taun(\suci_\hn^{\ID}) \ne \unset
      \ra
      \bigvee_{i \in \mathbb{N}} \instate_\taun(\suci_\hn^{\ID}) = \suci^i
    \]
    Moreover, we can only store $\suci^i$ in $\suci_\hn^{\ID}$ at $\pnai(i,1)$ or $\cnai(i,1)$, and by validity $\tau$ cannot contain both $\pnai(i,1)$ and $\cnai(i,1)$. We conclude observing that we cannot have $\accept_\taun^{\ID_0}$ and $\accept_\taun^{\ID_1}$ if $\taun = \_,\pnai(i,1)$ or $\_,\cnai(i,1)$ using \ref{a5} and \ref{a7}. The result follows.

  \item \ref{a6} is a consequence of \ref{a4}.
    
  \item \ref{a8} follows from the fact that whenever a new session of the protocol is started, we reset both $\bauth_\ue^\ID$ and $\eauth_\ue^\ID$. Then $\eauth_\ue^\ID$ is either set to $\fail$ or to $\bauth_\ue^\ID$.
    \qedhere
  \end{itemize}
\end{proof}

We can now state and prove our first acceptance necessary conditions.
\begin{lemma}
  \label{lem:accept-charac}
  Let $\tau = \_,\ai$ be a valid symbolic trace, then:
  \begin{enumerate}
  \item \customlabel{acc1}{\lpacc{1}} If $\ai = \pnai(j,1)$, then for every $\ID$ we have:
    \[
      \accept_\tau^\ID \ra
      \bigvee_{
        \tauo = \_,\npuai{1}{\ID}{j_0} \popre \tau}
      \left(
        \pi_1(g(\inframe_\tau)) =
        \enc{\pair
          {\ID}
          {\instate_{\tauo}(\sqn_\ue^\ID)}}
        {\pk_\hn}{\enonce^{j_0}}
        \wedge
        g(\inframe_{\tauo}) = \nonce^j
      \right)
    \]
  \item \customlabel{acc2}{\lpacc{2}} If $\ai = \npuai{2}{\ID}{j}$. Let $\taut = \_, \npuai{1}{\ID}{j} \popre \tau$. Then:
    \begin{center}
      \begin{tikzpicture}
        [dn/.style={inner sep=0.2em,fill=black,shape=circle},
        sdn/.style={inner sep=0.15em,fill=white,draw,solid,shape=circle},
        sl/.style={decorate,decoration={snake,amplitude=1.6}},
        dl/.style={dashed},
        pin distance=0.5em,
        every pin edge/.style={thin}]

        \draw[thick] (0,0)
        node[left=1.3em] {$\tau:$}
        -- ++(0.5,0)
        node[dn,pin={above:{$\npuai{1}{\ID}{j}$}}] {}
        node[below=0.3em]{$\taut$}
        -- ++(2.5,0)
        node[dn,pin={above:{$\pnai(j_0,1)$}}] {}
        node[below=0.3em]{$\tauo$}
        -- ++(2.5,0)
        node[dn,pin={above:{$\npuai{2}{\ID}{j}$}}] {}
        node[below=0.3em]{$\tau$};
      \end{tikzpicture}
    \end{center}
    \[
      \accept_\tau^\ID \;\ra\;
      \bigvee_{\tauo = \_, \pnai(j_0,1)
        \atop{\taut \potau \tauo}}
      \accept_\tauo^\ID \;\wedge\;
      g(\inframe_\taut) = \nonce^{j_0} \;\wedge\;
      \pi_1(g(\inframe_\tauo)) =
      \enc{
        \spair{\ID}
        {\instate_\taut(\sqn_\ue^\ID)}}
      {\pk_\hn}{\enonce^j}
    \]

  \item \customlabel{acc3}{\lpacc{3}} If $\ai = \cuai_\ID(j,1)$ then:
    \[
      \accept_\tau^\ID \;\ra\;
      \bigvee_{\tauo = \_, \cnai(j_0,0)
        \atop{\tauo \popre \tau}}
      \left(
        \begin{alignedat}{1}
          &\accept_\tauo^\ID \wedge
          \pi_1(g(\inframe_\tau)) = \nonce^{j_0}
          \wedge
          \pi_2(g(\inframe_\tau)) =
          \instate_\tauo(\sqn_\hn^\ID) \oplus \ow{\nonce^{j_0}}{\key}\\
          &\wedge
          \instate_\tau(\suci_\ue^\ID) = \instate_\tauo(\suci_\hn^\ID)
        \end{alignedat}
      \right)
    \]
    
  \item \customlabel{acc4}{\lpacc{4}} If $\ai = \cnai(j,1)$ then:
    \[
      \accept_\tau^\ID
      \;\ra\;
      \bigvee_{\tauo = \_,\cuai_\ID(\_,1) \popre \tau}
      \accept^\ID_\tauo \wedge \pi_1(g(\inframe_\tauo)) = \nonce^j
    \]
  \end{enumerate}
\end{lemma}

\subsection{Proof of Lemma~\ref{lem:accept-charac}}

\subsubsection*{Proof of \ref{acc1}}
Let $\ai = \pnai(j,1)$ and $\mkey \equiv \mkey^i$. Recall that:
\[
  \accept_\tau^{\ID} \;\equiv\;
  \left(\begin{array}[c]{ll}
      &\eq
      {\pi_1(\dec(\pi_1(g(\inframe_\tau)),\sk_\hn))}{\ID}\\
      \wedge &
      \eq{\pi_2(g(\inframe_\tau))}
      {\mac
        {\spair
          {\pi_1(g(\inframe_\tau))}
          {\nonce^j}}
        {\mkey}{1}}
    \end{array}\right)
\]
We apply the $\textsc{euf-mac}^1$ axiom (invariant \textsc{(inv-key)} guarantees that the syntactic side-conditions hold):
\begin{alignat*}{2}
  \accept_\tau^{\ID}
  \;\;\ra\;\;&
  \pi_2(g(\inframe_\tau))
  = \mac
  {\spair
    {\pi_1(g(\inframe_\tau))}
    {\nonce^j}}
  {\mkey}{1} \\
  \;\;\ra\;\;&
  \left(
    \begin{alignedat}{2}
      &\bigvee_{\tauo = \_, \npuai{1}{{\ID}}{j_0} \popre \tau}
      \pi_2(g(\inframe_\tau)) =
      \mac{\spair
        {\enc{\pair{\ID}{\instate_{\tauo}(\sqn_\ue^\ID)}}{\pk_\hn}{\enonce^{j_0}}}
        {g(\inframe_{\tauo})}}
      {\mkey}{1}\\
      &\vee
      \bigvee_{\tauo = \_, \pnai(j_0,1) \popre \tau}
      \pi_2(g(\inframe_\tau)) =
      {\mac
        {\spair
          {\pi_1(g(\inframe_{\tauo}))}
          {\nonce^{j_0}}}
        {\mkey}{1}}
    \end{alignedat}
  \right)
\end{alignat*}
From the validity of $\tau$ we know that $j_0 \ne j$. Hence using the axiom $\ax{Fresh}$, we get that $\nonce^{j_0} \ne \nonce^{j}$. Using the right injectivity of the pair (axiom $\textsf{EQInj}(\pair{\_}{\cdot})$), we know that:
\[
  {\spair
    {\pi_1(g(\inframe_\tau))}
    {\nonce^j}} \ne
  {\spair
    {\pi_1(g(\inframe_{\tauo}))}
    {\nonce^{j_0}}}
\]
From the collision-resistance of $\macsym$ (axiom $\textsc{cr}^1$), we have:
\begin{equation*}
  \mac
  {\spair
    {\pi_1(g(\inframe_\tau))}
    {\nonce^j}}
  {\mkey}{1} =
  \mac
  {\spair
    {\pi_1(g(\inframe_{\tauo}))}
    {\nonce^{j_0}}}
  {\mkey}{1}
  \ra
  {\spair
    {\pi_1(g(\inframe_\tau))}
    {\nonce^j}} =
  {\spair
    {\pi_1(g(\inframe_{\tauo}))}
    {\nonce^{j_0}}}
\end{equation*}
Therefore:
\begin{equation}
  \label{eq:proof1}
  \pi_2(g(\inframe_\tau))
  = \mac
  {\spair
    {\pi_1(g(\inframe_\tau))}
    {\nonce^j}}
  {\mkey}{1}
  \;\;\ra\;\;
  \neg \left(
    \;\;\bigvee_{{\tauo = \_, \pnai(j_0,1) \popre \tau}}\;\;
    \pi_2(g(\inframe_\tau)) =
    {\mac
      {\spair
        {\pi_1(g(\inframe_{\tauo}))}
        {\nonce^{j_0}}}
      {\mkey}{1}}
  \right)
\end{equation}

From which it follows that:
\[
  \accept_\tau^{\ID}
  \;\;\ra\;\;
  \bigvee_{\tauo = \_\, \npuai{1}{{\ID}}{j_0} \popre \tau}
  \pi_2(g(\inframe_\tau)) =
  \mac{\spair
    {\enc{\pair{\ID}{\instate_{\tauo}(\sqn_\ue^\ID)}}{\pk_\hn}{\enonce^{j_0}}}
    {g(\inframe_{\tauo})}}
  {\mkey}{1}
\]
To conclude, we use $\textsc{cr}^1$, $\textsf{EQInj}(\pair{\_}{\cdot})$ and $\textsf{EQInj}(\pair{\cdot}{\_})$ to show that for all $\tauo = \_,\npuai{1}{{\ID}}{j_0} \popre \tau$:
\begin{equation*}
  \left(\begin{alignedat}{2}
      &\pi_2(g(\inframe_\tau))
      = \mac
      {\spair
        {\pi_1(g(\inframe_\tau))}
        {\nonce^j}}
      {\mkey}{1}\\
      \wedge\; &
      \pi_2(g(\inframe_\tau)) =
      \mac{\spair
        {\enc{\pair{\ID}{\instate_{\tauo}(\sqn_\ue^\ID)}}{\pk_\hn}{\enonce^{j_0}}}
        {g(\inframe_{\tauo})}}
      {\mkey}{1}
    \end{alignedat}\right)
  \ra
  \left(\begin{alignedat}{2}
      &\pi_1(g(\inframe_\tau)) =
      \enc{\pair{\ID}{\instate_{\tauo}(\sqn_\ue^\ID)}}{\pk_\hn}{\enonce^{j_0}}\\
      \wedge\;&
      \nonce^j =
      g(\inframe_{\tauo})
    \end{alignedat}\right)
\end{equation*}
This, together with \ref{eq:proof1}, concludes the proof.

\subsubsection*{Proof of \ref{acc2}}
\begin{center}
  \begin{tikzpicture}
    [dn/.style={inner sep=0.2em,fill=black,shape=circle},
    sdn/.style={inner sep=0.15em,fill=white,draw,solid,shape=circle},
    sl/.style={decorate,decoration={snake,amplitude=1.6}},
    dl/.style={dashed},
    pin distance=0.5em,
    every pin edge/.style={thin}]

    \draw[thick] (0,0)
    node[left=1.3em] {$\tau:$}
    -- ++(0.5,0)
    node[dn,pin={above:{$\npuai{1}{\ID}{j}$}}] {}
    node[below=0.3em]{$\taut$}
    -- ++(2.5,0)
    node[dn,pin={above:{$\pnai(j_0,1)$}}] {}
    node[below=0.3em]{$\tauo$}
    -- ++(2.5,0)
    node[dn,pin={above:{$\npuai{2}{\ID}{j}$}}] {}
    node[below=0.3em]{$\tau$};
  \end{tikzpicture}
\end{center}
If $\ai = \npuai{2}{\ID}{j}$. Let $\mkey$ be the $\macsym$ key corresponding to $\ID$, i.e. $\mkey \equiv \mkey^\ID$. Recall that:
\[
  \accept_\tau^{\ID} \;\equiv\;
  g(\inframe_\tau) =
  \mac{
    \pair{\instate_\tau(\bauth_\ue^\ID)}
    {\instate_\tau(\sqn_\ue^\ID)}
  }{\mkey}{2}
\]

\paragraph{Part 1}
We are going to apply the $\textsc{p-euf-mac}^2$ axiom. We let:
\[
  S = \{ \tauo \mid \tauo =\_, \pnai(j_0,1) \popre \tau\}
\]
and for all $S_0 \subseteq S$ we let:
\[
  b_{S_0} =
  \Big(
  \bigwedge_{\tauo \in S_0} \accept_\tauo^\ID
  \Big)
  \wedge
  \Big(
  \bigwedge_{\tauo \in \overline{S_0}} \neg \accept_\tauo^\ID
  \Big)
\]
Then $(b_{S_0})_{S_0 \subseteq S}$ is a valid $\cs$ partition. It is straightforward to check that for every $S_0 \subseteq S$, for every $\tauo = \_, \pnai(j_0,1) \popre \tau$, if $\tauo \in S_0$ then we can rewrite $\cond{b_{S_0}}{t_{\tauo}}$ into a term $\cond{b_{S_0}}{t_{\tauo}^{S_0}}$ by removing the branch corresponding to $\accept_\tauo^\ID$. Therefore:
\[
  \mac{\spair{\nonce^{j_0}}
    {\sqnsuc(\pi_2(\dec(\pi_1(g(\inframe_\tauo)),\sk_\hn)))}}
  {\mkey}{2}
  \in
  \setmac_{\mkey}^2(t^{S_0}_{\tauo})
  \text{ if and only if }
  \tauo \in S_0
\]
Hence by applying the $\textsc{p-euf-mac}^2$ axiom we get that:
\[
  \accept_\tau^{\ID}
  \ra
  \bigvee_{S_0 \subseteq S}
  \left(
    b_{S_0} \wedge
    \begin{alignedat}{2}
      &\bigvee_{\tauo \in S_0}
      g(\inframe_\tau) =
      \mac{
        \spair{\nonce^{j_0}}
        {\sqnsuc(\pi_2(\dec(\pi_1(g(\inframe_\tauo)),\sk_\hn)))}}
      {\mkey}{2}\\
      \vee& \;\;\bigvee_{\mathclap{\tauo = \_,\npuai{2}{\ID}{j_0} \popre \tau}}\;\;
      g(\inframe_\tau) =
      \mac{
        \pair{\instate_\tauo(\bauth_\ue^\ID)}
        {\instate_\tauo(\sqn_\ue^\ID)}
      }{\mkey}{2}
    \end{alignedat}
  \right)
\]
We have:
\begin{align*}
  \left(
    \begin{alignedat}{2}
      &&&g(\inframe_\tau) = \mac{
        \pair{\instate_\tau(\bauth_\ue^\ID)}
        {\instate_\tau(\sqn_\ue^\ID)}
      }{\mkey}{2}
      \\
      &\wedge\;&&
      g(\inframe_\tau) =
      \mac{
        \pair{\instate_\tauo(\bauth_\ue^\ID)}
        {\instate_\tauo(\sqn_\ue^\ID)}
      }{\mkey}{2}
    \end{alignedat}
  \right)
  &\ra&&
  \pair{\instate_\tau(\bauth_\ue^\ID)}
  {\instate_\tau(\sqn_\ue^\ID)} =
  \pair{\instate_\tauo(\bauth_\ue^\ID)}
  {\instate_\tauo(\sqn_\ue^\ID)}
  \tag{$\textsc{cr}^2$}\\
  &\ra&&
  \instate_\tau(\sqn_\ue^\ID) =
  \instate_\tauo(\sqn_\ue^\ID)
  \tag{$\textsf{EQInj}(\pair{\_}{\cdot})$}\\
  &\ra&&
  \false \tag{\textbf{P9}}
\end{align*}
Moreover, remark that for $S_0 = \emptyset$, we have:
\[
  \left(
    \bigvee_{\tauo \in S_0}
    g(\inframe_\tau) =
    \mac{
      \spair{\nonce^j}
      {\sqnsuc(\pi_2(\dec(\pi_1(g(\inframe_\tauo)),\sk_\hn)))}}
    {\mkey}{2}
  \right)
  =
  \false
\]
Putting everything together, we get that:
\begin{alignat*}{2}
  \accept_\tau^{\ID}
  &\;\;\ra\;\;&&
  \bigvee_{S_0 \subseteq S\
    \atop{S_0 \ne \emptyset}}
  \left(
    b_{S_0} \wedge
    \bigvee_{\tauo \in S_0}
    g(\inframe_\tau) =
    \mac{
      \spair{\nonce^j}
      {\sqnsuc(\pi_2(\dec(\pi_1(g(\inframe_\tauo)),\sk_\hn)))}}
    {\mkey}{2}
  \right)\\
  &\;\;\ra\;\;&&
  \bigvee_{S_0 \subseteq S\atop{S_0 \ne \emptyset}}
  \bigvee_{\tauo \in S_0}
  \accept_\tauo^\ID \;\wedge\;
  g(\inframe_\tau) =
  \mac{
    \spair{\nonce^j}
    {\sqnsuc(\pi_2(\dec(\pi_1(g(\inframe_\tauo)),\sk_\hn)))}}
  {\mkey}{2}\displaybreak[0]\\
  &\;\;\ra\;\;&&
  \bigvee_{\tauo = \_,\pnai(j_0,1)\atop{\tauo \popre \tau}}
  \accept_\tauo^\ID \;\wedge\;
  g(\inframe_\tau) =
  \mac{
    \spair{\nonce^j}
    {\sqnsuc(\pi_2(\dec(\pi_1(g(\inframe_\tauo)),\sk_\hn)))}}
  {\mkey}{2}\displaybreak[0]\\
  &\;\;\ra\;\;&&
  \bigvee_{\tauo = \_, \pnai(j_0,1)\atop{\tauo \popre \tau}}
  \accept_\tauo^\ID \wedge
  \pair{\instate_\tau(\bauth_\ue^\ID)}
  {\instate_\tau(\sqn_\ue^\ID)} =
  \pair{\nonce^j}
  {\sqnsuc(\pi_2(\dec(\pi_1(g(\inframe_\tauo)),\sk_\hn)))}
  \tag{$\textsc{cr}^2$}\\
  &\;\;\ra\;\;&&
  \bigvee_{\tauo = \_, \pnai(j_0,1)\atop{\tauo \popre \tau}}
  \begin{alignedat}[c]{1}
    &\accept_\tauo^\ID \wedge
    \instate_\tau(\bauth_\ue^\ID) =
    \nonce^j \\
    &\wedge\;
    \instate_\tau(\sqn_\ue^\ID) =
    \sqnsuc(\pi_2(\dec(\pi_1(g(\inframe_\tauo)),\sk_\hn)))
  \end{alignedat}
  \tag{$\textsf{EQInj}(\pair{\_}{\cdot})$
    and $\textsf{EQInj}(\pair{\cdot}{\_})$}
\end{alignat*}

\paragraph{Part 2}
It only remains to show that we can restrict ourselves to the $\tauo$ such that $\taut \potau \tauo$. Using \ref{acc1} we know that:
\begin{alignat*}{2}
  \accept_\tauo^\ID
  &\;\;\ra\;\;&&
  \bigvee_{
    \tau' = \_,\npuai{1}{\ID}{j'}
    \atop{\tau' \potau\tauo}}
  \left(
    \pi_1(g(\inframe_\tauo)) =
    \enc{
      \pair
      {\ID}
      {\instate_{\tau'}(\sqn_\ue^\ID)}}
    {\pk_\hn}{\enonce^{j'}}
    \wedge
    g(\inframe_{\tau'}) = \nonce^{j_0}
  \right)
\end{alignat*}
Let $\tau' = \_,\npuai{1}{\ID}{j'}$ such that $\tau' \potau\tauo$. We now show that if $j' \ne j$ then the tests fail, which proves the impossibility of replaying an old message here. Assume $j' \ne j$, then:
\begin{alignat*}{2}
  &&&
  \instate_\tau(\sqn_\ue^\ID) =
  \sqnsuc(\pi_2(\dec(\pi_1(g(\inframe_\tauo)),\sk_\hn)))
  \wedge
  \pi_1(g(\inframe_\tauo)) =
  \enc{
    \pair
    {\ID}
    {\instate_{\tau'}(\sqn_\ue^\ID)}}
  {\pk_\hn}{\enonce^{j'}}\\
  &\;\;\ra\;\;&&
  \instate_\tau(\sqn_\ue^\ID) =
  \sqnsuc(\instate_{\tau'}(\sqn_\ue^\ID))\\
  &\;\;\ra\;\;&&
  \false
  \tag{\textbf{P9b}}
\end{alignat*}
We deduce that:
\[
  \accept_\tau^\ID \;\ra\;
  \bigvee_{\tauo = \_,\pnai(j_0,1)
    \atop{\taut \potau \tauo}}
  \accept_\tauo^\ID \;\wedge\;
  g(\inframe_\taut) = \nonce^{j_0} \;\wedge\;
  \pi_1(g(\inframe_\tauo)) =
  \enc{
    \spair{\ID}
    {\instate_\taut(\sqn_\ue^\ID)}}
  {\pk_\hn}{\enonce^j}
\]

\subsubsection*{Proof of \ref{acc3}}
Let $\ai = \cuai_\ID(j,1)$ and $\key$ be the $\owsym$ key corresponding to $\ID$. We know that:
\[
  \accept_\tau^{\ID} \;\ra\;
  \pi_3(g(\inframe_\tau)) =
  \mac
  {\striplet
    {\pi_1(g(\inframe_\tau))}
    {\pi_2(g(\inframe_\tau)) \xor \ow{\pi_1(g(\inframe_\tau))}{\key}}
    {\instate_\tau(\suci_\ue^\ID)}}
  {\mkey}{3}
\]
We are going to apply the $\textsc{p-euf-mac}^3$ axiom. We let $S = S_\hn \cup S_\ue$, where:
\[
  S_\hn = \{ \tauo \mid \tauo = \_, \cnai(j_0,1)\popre \tau\}
  \qquad\qquad
  S_\ue = \{ \tauo \mid \tauo = \_, \cuai_\ID(j_0,1)\popre \tau\}
\]
and for all $S_0 \subseteq S$ we let:
\[
  b_{S_0} =
  \Big(
  \bigwedge_{\tauo \in S_0} \accept_\tauo^\ID
  \Big)
  \wedge
  \Big(
  \bigwedge_{\tauo \in \overline{S_0}} \neg \accept_\tauo^\ID
  \Big)
\]
Then $(b_{S_0})_{S_0 \subseteq S}$ is a valid $\cs$ partition. It is straightforward to check that for every $S_0 \subseteq S$, for every $\tauo = \_,\cnai(j_0,1) \popre \tau$, if $\tauo \in S$ then we can rewrite $\cond{b_{S_0}}{t_{\tauo}}$ into a term $\cond{b_{S_0}}{t_{\tauo}^{S_0}}$ by removing the branch corresponding to $\accept_\tauo^\ID$. Therefore:
\[
  \mac{\striplet
    {\nonce^{j_0}}
    {\instate_\tauo(\sqn_\hn^{\ID})}
    {\instate_\tauo(\suci_\hn^\ID)}}
  {\mkey}{3}
  \in \setmac_{\mkey}^3(t^{S_0}_{\tauo})
  \text{ if and only if }
  \tauo \in S_0
\]
Similarly for every  $\tauo = \_,\cuai_\ID(j_0,1) \popre \tau$, if $\tauo \in S$ then we can rewrite $\cond{b_{S_0}}{t_{\tauo}}$ as follows:
\[
  \cond{b_{S_0}}{t_{\tauo}} =
  \begin{cases}
    \cond{b_{S_0}}
    {\mac{\pi_1(g(\inframe_\tauo))}{\mkey}{4}} & \text{ if } \tauo \in S_0\\
    \cond{b_{S_0}}{\textsf{error}} & \text{ if } \tauo \in \overline{S_0}
  \end{cases}
\]
Hence by applying the $\textsc{p-euf-mac}^4$ axiom we get that:
\[
  \accept_\tau^{\ID}
  \ra
  \bigvee_{S_0 \subseteq S}
  \left(
    b_{S_0} \wedge
    \bigvee_{\tauo \in S_0 \cap S_\hn}
    \pi_3(g(\inframe_\tau)) =
    \mac{\striplet
      {\nonce^{j_0}}
      {\instate_\tauo(\sqn_\hn^{\ID})}
      {\instate_\tauo(\suci_\hn^\ID)}}
    {\mkey}{3}
  \right)
\]
By $\textsc{cr}^3$, $\textsf{EQInj}(\pair{\_}{\cdot})$ and $\textsf{EQInj}(\pair{\cdot}{\_})$ we have:
\begin{multline*}
  \left(
    \begin{alignedat}{2}
      &&&
      \pi_3(g(\inframe_\tau)) =
      \mac
      {\triplet
        {\pi_1(g(\inframe_\tau))}
        {\pi_2(g(\inframe_\tau)) \xor \ow{\pi_1(g(\inframe_\tau))}{\key}}
        {\instate_\tau(\suci_\ue^\ID)}}
      {\mkey}{3}
      \\
      &\wedge\;&&
      \pi_3(g(\inframe_\tau)) =
      \mac{\striplet
        {\nonce^{j_0}}
        {\instate_\tauo(\sqn_\hn^{\ID})}
        {\instate_\tauo(\suci_\hn^\ID)}}
      {\mkey}{3}
    \end{alignedat}
  \right)\;\;\ra\\
  \pi_1(g(\inframe_\tau)) =
  \nonce^{j_0} \wedge
  \pi_2(g(\inframe_\tau)) \xor \ow{\pi_1(g(\inframe_\tau))}{\key} =
  \instate_\tauo(\sqn_\hn^{\ID}) \wedge
  \instate_\tau(\suci_\ue^\ID) = \instate_\tauo(\suci_\hn^\ID)
\end{multline*}
Using the idempotence of the $\xor$ we know that:
\[
  \left(
    \pi_1(g(\inframe_\tau)) =
    \nonce^{j_0} \;\wedge\;
    \pi_2(g(\inframe_\tau)) \xor \ow{\pi_1(g(\inframe_\tau))}{\key} =
    \instate_\tauo(\sqn_\hn^{\ID})
  \right)
  \;\;\ra\;\;
  \pi_2(g(\inframe_\tau)) =
  \instate_\tauo(\sqn_\hn^{\ID})  \xor \ow{\nonce^{j_0}}{\key}
\]
Moreover, remark that if $S_0 \cap S_\hn = \emptyset$, we have:
\[
  \bigvee_{S_0 \subseteq S}
  \left(
    b_{S_0} \wedge
    \bigvee_{\tauo \in S_0 \cap S_\hn}
    \pi_3(g(\inframe_\tau)) =
    \mac{\striplet
      {\nonce^{j_0}}
      {\instate_\tauo(\sqn_\hn^{\ID})}
      {\instate_\tauo(\suci_\hn^\ID)}}
    {\mkey}{3}
  \right)
  =
  \false
\]
Putting everything together, we get that:
\begin{alignat*}{2}
  \accept_\tau^{\ID}
  &\;\;\ra\;\;&&
  \bigvee_{S_0 \subseteq S
    \atop{S_0 \cap S_\hn \ne \emptyset}}
  \left(
    b_{S_0} \wedge
    \bigvee_{\tauo \in S_0 \cap S_\hn}
    \pi_3(g(\inframe_\tau)) =
    \mac{\striplet
      {\nonce^{j_0}}
      {\instate_\tauo(\sqn_\hn^{\ID})}
      {\instate_\tauo(\suci_\hn^\ID)}}
    {\mkey}{3}
  \right)\\
  &\;\;\ra\;\;&&
  \bigvee_{S_0 \subseteq S
    \atop{S_0 \cap S_\hn \ne \emptyset}}
  \bigvee_{\tauo \in S_0 \cap S_\hn}
  \accept_\tauo^\ID \wedge
  \pi_3(g(\inframe_\tau)) =
  \mac{\striplet
    {\nonce^{j_0}}
    {\instate_\tauo(\sqn_\hn^{\ID})}
    {\instate_\tauo(\suci_\hn^\ID)}}
  {\mkey}{3}\displaybreak[0]\\
  &\;\;\ra\;\;&&
  \bigvee_{\tauo = \_,\cnai(j_0,0) \popre \tau}
  \accept_\tauo^\ID \wedge
  \pi_3(g(\inframe_\tau)) =
  \mac{\striplet
    {\nonce^{j_0}}
    {\instate_\tauo(\sqn_\hn^{\ID})}
    {\instate_\tauo(\suci_\hn^\ID)}}
  {\mkey}{3}\\
  &\;\;\ra\;\;&&
  \bigvee_{\tauo = \_,\cnai(j_0,0)\popre \tau}
  \accept_\tauo^\ID \wedge
  \pi_1(g(\inframe_\tau)) =
  \nonce^{j_0}
  \begin{alignedat}[c]{1}
    \wedge\;
    \pi_2(g(\inframe_\tau)) =
    \instate_\tauo(\sqn_\hn^{\ID})  \xor \ow{\nonce^{j_0}}{\key}\\
    \wedge\;
    \instate_\tau(\suci_\ue^\ID) = \instate_\tauo(\suci_\hn^\ID)
  \end{alignedat}
\end{alignat*}

\subsubsection*{Proof of \ref{acc4}}
We are going to apply the $\textsc{p-euf-mac}^4$ axiom. We let $S = \left\{ \tauo \mid \tauo = \_,\cuai_\ID(j_0,1) \popre \tau \right\}$, and for all $S_0 \subseteq S$ we let :
\[
  b_{S_0} =
  \Big(
  \bigwedge_{\tauo \in S_0} \accept_\tauo^\ID
  \Big)
  \wedge
  \Big(
  \bigwedge_{\tauo\in \overline{S_0}} \neg \accept_\tauo^\ID
  \Big)
\]
Then $(b_{S_0})_{S_0 \subseteq S}$ is a valid $\cs$ partition. It is straightforward to check that for every $S_0 \subseteq S$, for every $\tauo = \_,\cuai_\ID(j_0,1) \popre \tau$:
\[
  \cond{b_{S_0}}{t_\tauo} =
  \begin{cases}
    \cond{b_{S_0}}
    {\mac{\pi_1(g(\inframe_\tauo))}{\mkey}{4}} & \text{ if } \tauo \in S_0\\
    \cond{b_{S_0}}{\textsf{error}} & \text{ if } \tauo \in \overline{S_0}
  \end{cases}
\]
Hence by applying the $\textsc{p-euf-mac}^4$ axiom we get that:
\[
  g(\inframe_\tau) = \mac{\nonce^j}{\mkey}{4}
  \ra
  \bigvee_{S_0 \subseteq S}
  b_{S_0} \wedge
  \left(
    \begin{alignedat}{2}
      &\bigvee_{\tauo \in S_0}
      g(\inframe_\tau) = \mac{\pi_1(g(\inframe_\tauo))}{\mkey}{4}\\
      \vee& \bigvee_{\tauo = \_,\cnai(j_0,1) \popre \tau}
      g(\inframe_\tau) = \mac{\nonce^{j_0}}{\mkey}{4}
    \end{alignedat}
  \right)
\]
Applying the $\textsc{cr}^4$ axiom we get that:
\begin{alignat*}{2}
  g(\inframe_\tau) = \mac{\nonce^{j_0}}{\mkey}{4}
  \;\wedge\;
  g(\inframe_\tau) = \mac{\nonce^j}{\mkey}{4}
  &\ra\;\;&&
  \nonce^{j_0}
  = \nonce^j \tag{$\textsc{cr}^4$}\\
  &\ra\;\;&&
  \false \tag{\ax{EQIndep}}
\end{alignat*}
Moreover, remark that for $S_0 = \emptyset$, we have:
\[
  \textstyle\neg\Big(
    b_{S_0} \wedge \bigvee_{\tauo \in S_0}
    g(\inframe_\tau) = \mac{\pi_1(g(\inframe_\tauo))}{\mkey}{4}
  \Big)
\]
Putting everything together, we get that:
\[
  g(\inframe_\tau) = \mac{\nonce^j}{\mkey}{4}
  \ra
  \bigvee_{S_0 \subseteq S\atop{S_0 \ne \emptyset}}
  \Big(
    b_{S_0} \wedge
    \bigvee_{\tauo \in S_0}
    g(\inframe_\tau) = \mac{\pi_1(g(\inframe_\tauo))}{\mkey}{4}
  \Big)
\]
Let $S_0 \subseteq S$ with $S_0 \ne \emptyset$, and let $\tauo \in S_0$. Using the $\textsc{cr}^4$ axiom we know that:
\[
  g(\inframe_\tau) = \mac{\nonce^j}{\mkey}{4}
  \;\wedge\;
  g(\inframe_\tau) = \mac{\pi_1(g(\inframe_\tauo))}{\mkey}{4}
  \;\ra\;
  \pi_1(g(\inframe_\tauo)) = \nonce^j
\]
Therefore:
\begin{alignat*}{2}
  g(\inframe_\tau) = \mac{\nonce^j}{\mkey}{4}
  &\ra\;\;&&
  \bigvee_{S_0 \subseteq S\atop{S_0 \ne \emptyset}}
  \Big(
    b_{S_0} \wedge
    \bigvee_{\tauo \in S_0}
    g(\inframe_\tau) = \mac{\pi_1(g(\inframe_\tauo))}{\mkey}{4}
  \Big)\\
  &\ra\;\;&&
  \bigvee_{S_0 \subseteq S\atop{S_0 \ne \emptyset}}
  \Big(
    b_{S_0} \wedge
    \bigvee_{\tauo \in S_0}
    \pi_1(g(\inframe_\tauo)) = \nonce^j
  \Big)\displaybreak[0]\\
  \intertext{And using the fact that $b_{S_0} \ra \accept_\tauo^\ID$:}
  g(\inframe_\tau) = \mac{\nonce^j}{\mkey}{4}
  &\ra\;\;&&
  \bigvee_{\tauo = \_,\cuai_\ID(\_,1) \popre \tau}
  \accept_\tauo \wedge \pi_1(g(\inframe_\tauo)) = \nonce^j
\end{alignat*}
We conclude by observing that $\accept_\tau^\ID \ra g(\inframe_\tau) = \mac{\nonce^j}{\mkey}{4}$.

\subsection{Authentication of the User by the Network}
\label{app:subsection-auth-serv-net-proof}
We now prove that the $\faka$ protocol provides authentication of the user the network. Remark that the lemma below subsumes Lemma~\ref{lem:auth-serv-net-body}.
\begin{lemma}
  \label{lem:auth-serv-net}
  For all valid symbolic trace $\tau$, $\inframe_\tau$ guarantees authentication of the user by the network:
  \[
    \forall \ID \in \iddom, j \in \mathbb{N},\;\;
    \instate_\tau(\eauth_\hn^j) = \ID \;\ra\;
    \bigvee_{\tau' \popreleq \tau}
    \instate_{\tau'}(\bauth_\ue^\ID) = \nonce^j
  \]
  Moreover, if $\tau = \_,\cnai(j,1)$ then:
  \begin{alignat*}{2}
    \accept_\tau^\ID
    &\;\ra\;&&
    \bigvee_{\tauo = \_,\cuai_\ID(\_,1) \popre \tau}
    \cstate_\tauo(\bauth_\ue^\ID) = \nonce^j
  \end{alignat*}
\end{lemma}

\begin{proof}
  We prove this by induction on $\tau$. First, for $\tau = \epsilon$ we have that for every $\ID$ $\instate_\tau(\eauth_\hn^j) = \bot \ne \ID$. Therefore the property holds.

  Let $\tau = \tau_0,\ai$. Observe that for all $j_0$, if $\upstate_\tau(\eauth_\hn^{j_0}) = \bot$ and if the authentication property holds for $\inframe_{\tau_0}$:
  \[
    \forall \ID \in \iddom,\;\;
    \instate_{\tau_0}(\eauth_\hn^{j_0}) = \ID \;\ra\;
    \bigvee_{\tau' \popreleq \tau_0} \instate_{\tau'}(\bauth_\ue^\ID) =
    \nonce^{j_0}
  \]
  then it holds for $\tau$. Therefore we only need to show that authentication holds for $\ai = \pnai(j,1)$ with $j_0 = j$, and for $\ai = \cnai(j,1)$ with $j_0 = j$.
  \begin{itemize}
  \item \textbf{Case $\ai = \pnai(j,1)$:} Let $\ID \in \iddom$. Using \textsf{EQConst}, we know that $\instate_{\tau}(\eauth_\hn^{j_0}) = \ID \ra \accept_\tau^\ID$ is true. Using \ref{acc1} of Proposition~\ref{prop:invs}, we deduce that:
    \begin{equation}
      \label{eq:prop1}
      \instate_{\tau}(\eauth_\hn^{j}) = \ID \ra
      \bigvee_{\tauo = \_, \npuai{1}{\ID}{j_0} \popre \tau}
      g(\inframe_{\tauo}) = \nonce^j
    \end{equation}
    By validity of $\tau$, we know there exists $\tautt$ such that $\tautt = \_,\pnai(j,0) \popre \tau$. Let $\tauo \potau \tautt$. We have $\neg \sessionstarted_j(\ai_0)$, therefore using invariant \ref{a1} we get that $n^j \not \in \st(\inframe_{\tauo})$. It follows from axiom $\ax{EQIndep}$ that $\neg g(\inframe_{\tauo}) = \nonce^j$. Hence:
    \begin{equation}
      \label{eq:prop2}
      \bigvee_{\tauo = \_, \npuai{1}{\ID}{j_0} \popre \tau}
      g(\inframe_{\tauo}) = \nonce^j
      \;\;\leftrightarrow\;\;
      \bigvee_{\tauo = \_, \npuai{1}{\ID}{j_0}\atop{ \tautt \potau \tauo\popre \tau}}
      g(\inframe_{\tauo}) = \nonce^j
    \end{equation}
    Let $\tauo$ be such that $\tautt \potau \tauo \popre \tau$ and $\tauo = \_,\npuai{1}{\ID}{j_0}$. Since $\upstate_{\tauo}(\bauth_\ue^\ID) = g(\inframe_{\tauo})$, we know that $\cstate_{\tauo}(\bauth_\ue^\ID) = g(\inframe_{\tauo})$. Hence:
    \begin{equation*}
      g(\inframe_{\tauo}) = \nonce^j \ra
      \cstate_{\tauo}(\bauth_\ue^\ID) = \nonce^j
    \end{equation*}
    Which shows that:
    \begin{equation}
      \label{eq:prop2b}
      \bigvee_{
        \tauo = \_,\npuai{1}{\ID}{j_0}
        \atop{\tautt \potau \tauo \popre \tau}}
      g(\inframe_{\tauo}) = \nonce^j
      \ra
      \bigvee_{
        \tauo = \_,\npuai{1}{\ID}{j_0}
        \atop{\tautt \potau \tauo \popre \tau}}
      \cstate_{\tauo}(\bauth_\ue^\ID) = \nonce^j
    \end{equation}
    Since $\{\tauo \mid \tauo = \_,\npuai{1}{\ID}{j_0} \wedge \tautt \potau \tauo \popre \tau \}$ is a subset of $\{\taut \mid \taut \popre \tau\}$, we have:
    \begin{alignat*}{2}
      \bigvee_{
        \tauo = \_,\npuai{1}{\ID}{j_0}
        \atop{\tautt \potau \tauo \popre \tau}}
      \cstate_{\tauo}(\bauth_\ue^\ID) = \nonce^j
      \;&\ra\;\;&&
      \bigvee_{
        \tauo \popre \tau}
      \cstate_{\tauo}(\bauth_\ue^\ID) = \nonce^j\\
      &\ra\;\;&&
      \bigvee_{
        \tauo \popreleq \tau}
      \instate_{\tauo}(\bauth_\ue^\ID) = \nonce^j
      \numberthis\label{eq:prop3}
    \end{alignat*}
    We conclude the proof using \ref{eq:prop1}, \ref{eq:prop2}, \ref{eq:prop2b} and \ref{eq:prop3}.

  \item \textbf{Case $\ai = \cnai(j,1)$:} Using \ref{acc4}, we know that:
    \[
      \accept_\tau^\ID
      \;\ra\;
      \bigvee_{\tauo = \_,\cuai_\ID(\_,1) \popre \tau}
      \accept_\tauo \wedge \pi_1(g(\inframe_\tauo)) = \nonce^j
    \]
    Moreover, for every $\tauo = \_,\cuai_\ID(\_,1) \popre \tau$, we have:
    \[
      \accept_\tauo^\ID
      \;\wedge\;
      \pi_1(g(\inframe_\tauo)) =
      \nonce^j
      \;\ra\;
      \cstate_\tauo(\bauth_\ue^\ID) =
      \nonce^j
    \]
    Hence:
    \begin{alignat*}{2}
      \accept_\tau^\ID
      &\ra\;\;&&
      \bigvee_{\tauo = \_,\cuai_\ID(\_,1) \popre \tau}
      \accept_\tauo \wedge \pi_1(g(\inframe_\tauo)) = \nonce^j
      \displaybreak[0]\\
      &\ra\;\;&&
      \bigvee_{\tauo = \_,\cuai_\ID(\_,1) \popre \tau}
      \cstate_\tauo(\bauth_\ue^\ID) = \nonce^j\\
      &\ra\;\;&&
      \bigvee_{\tauo \popreleq \tau}
      \instate_\tauo(\bauth_\ue^\ID) = \nonce^j
    \end{alignat*}
    \qedhere
  \end{itemize}
\end{proof}

\subsection{Authentication of the Network by the User}
\label{app:subsection-auth-net-serv-proof}
We now prove that the $\faka$ protocols provides authentication of the network by the user. We actually prove the stronger result that for any valid symbolic trace $\tau$, if the authentication of $\tue_{\ID}$ succeeded at instant $\tau$ (i.e. $\instate_\tau(\eauth_\ue^\ID)$ is not $\fail$ or $\bot$), then there exists some $j \in \mathbb{N}$ such that $\tue_{\ID}$ authenticated the session $\thn(j0$.
\begin{lemma}
  \label{lem:auth-net}
  For all valid symbolic trace $\tau$, $\inframe_\tau$ guarantees authentication of the network by the user. For all $\ID \in \iddom$ and $j \in \mathbb{N}$, we define the formulas:
  \begin{gather*}
    \sucauth_\tau(\ID) \;\;\equiv\;\;
    \instate_\tau(\eauth_\ue^\ID) \ne \fail
    \wedge
    \instate_\tau(\eauth_\ue^\ID) \ne \bot\\
    \auth_\tau(\ID,j) \;\;\equiv\;\;
    \instate_\tau(\bauth_\hn^j) = \ID \wedge
    \nonce^j = \instate_\tau(\eauth_\ue^\ID)
  \end{gather*}
  Then we have:
  \[
    \textstyle
    \forall \ID \in \iddom,\;\;
    \sucauth_\tau(\ID)
    \;\ra\;
    \bigvee_{j \in \mathbb{N}}
    \auth_\tau(\ID,j)
  \]
\end{lemma}

\begin{proof}
  We prove this by induction on $\tau$. First, for $\tau = \epsilon$ we have that for every $j$, $\instate_\tau(\eauth_\hn^j) = \bot$. Therefore the property holds. Let $\tau = \tau_0,\ai$, assume that:
  \[
    \textstyle
    \forall \ID \in \iddom,\;\;
    \sucauth_{\tau_0}(\ID)
    \;\ra\;
    \bigvee_{j \in \mathbb{N}}
    \auth_{\tau_0}(\ID,j)
  \]
  If for every $j_0$ we have:
  \begin{equation}
    \label{eq:phcase0}
    \upstate_\tau(\bauth_\hn^{j_0}) = \bot
    \qquad\qquad
    \forall \ID \in \iddom,\;
    \upstate_\tau(\eauth_\ue^{\ID}) = \bot
  \end{equation}
  then we have authentication of the network by the user at $\tau$. Therefore we only need to show that authentication holds for $\tau$ in the cases where $\ai$ is equal to $\pnai(j,0)$, $\pnai(j,1)$, $\npuai{2}{\ID}{j}$, $\cnai(j,0)$ or $\cuai_\ID(j,1)$.
  \begin{itemize}
  \item \textbf{Case $\ai = \pnai(j,0)$.} we are going to show that for all $\ID \in \iddom, j_0 \in \mathbb{N}$, we have:
    \begin{equation}
      \label{eq:phcase1}
      \left(
        \sucauth_\tau(\ID)
        \;\wedge\; \auth_\tau(\ID,j_0)
      \right)
      \;\leftrightarrow\;
      \left(
        \sucauth_{\tau_0}(\ID)
        \;\wedge\; \auth_{\tau_0}(\ID,j_0)
      \right)
    \end{equation}
    which implies the wanted result. Let $\ID \in \iddom, j_0 \in \mathbb{N}$.
    \begin{itemize}
    \item In the case $j_0 \ne j$, using the validity of $\tau$, we know that $\pnai(j,1) \not \popre \tau$. This implies that $\instate_\tau(\bauth_\hn^{j}) = \bot \ne \ID$, which in turn implies that $\neg \auth_\tau(\ID,j)$. By a similar reasoning, we get that $\pnai(j,1) \not \in \tau_0$, therefore $\instate_{\tau_0}(\bauth_\hn^{j}) \ne \ID$, and by consequence $\neg \auth_{\tau_0}(\ID,j)$. This concludes the proof of the validity of \eqref{eq:phcase1}.

    \item In the case $j_0 = j$, since $\instate_\tau(\bauth_\hn^{j}) = \bot \ne \ID$ we know that $\auth_\tau(\ID,j) = \false$. Similarly $\auth_{\tau_0}(\ID,j) = \false$. Therefore \eqref{eq:phcase1} holds for $j_0 = j$.    \end{itemize}

  \item \textbf{Case $\ai = \pnai(j,1)$.} Here also we are going to show that \eqref{eq:phcase1} holds for all $j_0$. Let $\ID \in \iddom, j_0 \in \mathbb{N}$.
    \begin{itemize}
    \item If $j_0 \ne j$, we have $\upstate_\tau(\bauth_\hn^{j_0}) = \bot$ and $\upstate_\tau(\eauth_\ue^{\ID}) = \bot$. It follows that \eqref{eq:phcase1} holds.

    \item If $j_0 = j$, using the validity of $\tau$ we know that $\instate_{\tau_0}(\bauth_\hn^j) \equiv \bot$. From $\textsf{EQConst}$ it follows that $\instate_{\tau_0}(\bauth_\hn^j) \ne \ID$, and therefore $\auth_{\tau_0}(\ID,j) = \false$.

      To conclude this case, we only need to show that $(\sucauth_\tau(\ID) \wedge \auth_{\tau}(\ID,j)) = \false$.

      First, assume that there never was a call to $\npuai{2}{\ID}{\_}$, i.e. $\npuai{2}{\ID}{\_} \not \popre \tau_1$. Then $\instate_\tau(\eauth_\ue^\ID) \equiv \bot$, and therefore $\sucauth_\tau(\ID) = \false$.

      Otherwise, let $\tauo = \_,\npuai{2}{\ID}{j_0}$ be the latest call to $\npuai{2}{\ID}{\_}$, i.e. $\tauo \not \potau \npuai{2}{\ID}{\_}$. By validity of $\tau$, we know that there exists $\tautt$ such that $\tautt = \_,\npuai{1}{\ID}{j_0}$. We know that $\instate_\tau(\eauth_\ue^\ID) \equiv \upstate_\tauo(\eauth_\ue^\ID)$. Hence:
      \begin{alignat*}{2}
        \sucauth_\tau(\ID) &\;\;\ra\;\;&&
        \upstate_\tauo(\eauth_\ue^\ID) \ne \fail
        \tag{by definition of $\sucauth_\tau(\ID)$}\\
        &\;\;\ra\;\;&&
        \accept_\tauo^\ID
        \tag{by definition of $\upstate_\tauo(\eauth_\ue^\ID)$}\\
        &\;\;\ra\;\;&&
        \accept_\tauo^\ID \wedge
        \bigvee_{
          \taut = \_,\pnai(j_1,1)
          \atop{\tautt \potau \taut \potau \tauo}}
        \left(g(\inframe_\tautt) = \nonce^{j_1}\right)
        \tag{by \ref{acc2}}\displaybreak[0]\\
        &\;\;\ra\;\;&&
        \bigvee_{
          \taut = \_,\pnai(j_1,1)
          \atop{\tautt \potau \taut \potau \tauo}}
        \left(\upstate_\tauo(\eauth_\ue^\ID) = \nonce^{j_1}\right)\\
        &\;\;\ra\;\;&&
        \bigvee_{
          \taut = \_,\pnai(j_1,1)
          \atop{\tautt \potau \taut \potau \tauo}}
        \left(\instate_\tau(\eauth_\ue^\ID) = \nonce^{j_1}\right)
        \tag{since $\instate_\tau(\eauth_\ue^\ID) \equiv \upstate_\tauo(\eauth_\ue^\ID)$}
      \end{alignat*}
      Since $\tauo \popre \tau$ we know that for every $\taut = \_,\pnai(j_1,1) \in \{\taut \mid \tautt \popre \taut \potau \tauo\}$, $j_1 \ne j$, and therefore using $\ax{EQIndep}$ we know that $(\nonce^{j_1} = \nonce^j) \lra \false$. Therefore:
      \begin{alignat*}{2}
        \sucauth_\tau(\ID) \wedge \auth_{\tau}(\ID,j)
        &\;\;\ra\;\;&&
        \bigvee_{
          \taut =\_, \pnai(j_1,1)
          \atop{\tautt \potau \taut \popre \tauo}}
        \left(
          \instate_\tau(\eauth_\ue^\ID) = \nonce^{j_1} \wedge
          \nonce^{j_1} \ne \nonce^j \wedge
          \auth_{\tau}(\ID,j)
        \right)\\
        \intertext{And by definition of $\auth_{\tau}(\ID,j)$:}
        &\;\;\ra\;\;&&
        \bigvee_{
          \taut = \_,\pnai(j_1,1)
          \atop{\tautt \potau \taut \popre \tauo}}
        \left(
          \instate_\tau(\eauth_\ue^\ID) = \nonce^{j_1} \wedge
          \nonce^{j_1} \ne \nonce^j \wedge
          \instate_\tau(\eauth_\ue^\ID) = \nonce^j
        \right)\\
        &\;\;\ra\;\;&&
        \false
      \end{alignat*}
    \end{itemize}

  \item \textbf{Case $\ai = \npuai{2}{\ID}{j}$.} For all $\ID_0 \ne \ID$ and for all $j_0 \in \mathbb{N}$, it is trivial that:
    \begin{equation*}
      \left(
        \sucauth_\tau(\ID_0)
        \;\wedge\; \auth_\tau(\ID_0,j_0)
      \right)
      \;\leftrightarrow\;
      \left(
        \sucauth_{\tau_0}(\ID_0)
        \;\wedge\; \auth_{\tau_0}(\ID_0,j_0)
      \right)
    \end{equation*}
    Therefore we only need to show that:
    \[
      \textstyle
      \sucauth_\tau(\ID)
      \;\ra\;
      \bigvee_{j \in \mathbb{N}}
      \auth_\tau(\ID,j)
    \]
    First, we observe that:
    \begin{alignat*}{2}
      \sucauth_\tau(\ID)
      &\;\;\ra\;\;&&
      \accept_\tau^\ID\\
      &\;\;\ra\;\;&&
      \bigvee_{\tauo = \_,\pnai(j_0,1)
        \atop{\taut = \_,\npuai{1}{\ID}{j}
          \atop{\taut \potau \tauo}}}
      \left(
        \begin{array}[c]{l}
          \accept_\tauo^\ID \;\wedge\;
          g(\inframe_\taut) = \nonce^{j_0} \;\wedge\\
          \pi_1(g(\inframe_\tauo)) =
          \enc{
            \spair{\ID}
            {\instate_\taut(\sqn_\ue^\ID)}}
          {\pk_\hn}{\enonce^j}
        \end{array}
      \right)
      \tag{by \ref{acc2}}
    \end{alignat*}

    Let $\tauo = \_,\pnai(j_0,1)$, $\taut = \_,\npuai{1}{\ID}{j}$ such that $\taut \potau \tauo$. Let $\tautt = \_,\pnai(j_0,1)$, by validity of $\tau$ we know that $\tautt \potau \tauo$. Moreover, if $\taut \potau \tautt$ then by \ref{a1} we have $\nonce^{j_0} \not \in \st(\inframe_\taut)$, and therefore using $\ax{EQIndep}$ we obtain that $g(\inframe_\taut) \ne \nonce^{j_0}$. Hence:
    \begin{alignat*}{2}
      \sucauth_\tau(\ID)
      &\;\;\ra\;\;&&
      \bigvee_{
        \tauo = \_,\pnai(j_0,1)
        \atop{\taut = \_,\npuai{1}{\ID}{j}
          \atop{\tautt = \_,\pnai(j_0,0)
            \atop{\tautt \potau \taut \potau \tauo}}}}
      \left(
        \begin{array}[c]{l}
          \accept_\tau^\ID \;\wedge\;
          \accept_\tauo^\ID \;\wedge\;
          g(\inframe_\taut) = \nonce^{j_0}\;\wedge\\
          \pi_1(g(\inframe_\tauo)) =
          \enc{
            \spair{\ID}
            {\instate_\taut(\sqn_\ue^\ID)}}
          {\pk_\hn}{\enonce^j}
        \end{array}
      \right)
    \end{alignat*}
    Moreover, we know that:
    \begin{mathpar}
      \accept_\tauo^\ID \;\ra\;\upstate_\tauo(\bauth_\hn^{j_0}) = \ID

      \accept_\tau^\ID \;\ra\;
      \cstate_\taut(\eauth_\ue^\ID)  = \upstate_\taut(\bauth_\ue^\ID)
    \end{mathpar}
    We represent graphically all the information we have below:
    \begin{center}
      \begin{tikzpicture}
        [dn/.style={inner sep=0.2em,fill=black,shape=circle},
        sdn/.style={inner sep=0.15em,fill=white,draw,solid,shape=circle},
        sl/.style={decorate,decoration={snake,amplitude=1.6}},
        dl/.style={dashed},
        pin distance=1em,
        every pin edge/.style={thin}]

        \draw[thick] (0,0)
        node[left=1.3em] {$\tau:$}
        -- ++(0.5,0)
        node[dn,pin={above:{$\tautt = \_,\pnai(j_0,0)$}}]
        (a) {}
        -- ++(3.5,0)
        node[dn,pin={above:{$\taut = \_,\npuai{1}{\ID}{j}$}}]
        (b) {}
        -- ++(3.5,0)
        node[dn,pin={above:{$\tauo = \_,\pnai(j_0,1)$}}]
        (c) {}
        -- ++(5,0)
        node[dn,pin={above:{$\tau = \_,\npuai{2}{\ID}{j}$}}]
        (d) {};

        \draw[dl] (a) -- ++ (0,-3.3)
        node[sdn] (a1) {}
        node[left,xshift=-0.3em] {$\nonce^{j_0}$};

        \draw[dl] (b) -- ++ (0,-0.8)
        node[sdn] (b1) {}
        node[left] {
          $\begin{alignedat}{2}
            \upstate_\taut(\bauth_\ue^\ID) & = && g(\inframe_\taut)\\
            & = && \nonce^{j_0}
          \end{alignedat}$};

        \draw[dl] (c) -- ++ (0,-2.4)
        node[sdn] (c1) {}
        node[left,xshift=-0.3em] {$\upstate_\tauo(\bauth_\hn^{j_0}) = \ID$};

        \draw[dl] (d) -- ++ (0,-3.3)
        node[sdn] (d4) {};

        \path (d) -- ++ (0,-2.4)
        node[sdn] (d3) {};

        \path (d) -- ++ (0,-1.6)
        node[sdn] (d2) {}
        node[left] {
          $\begin{alignedat}{2}
            \upstate_\taut(\eauth_\ue^\ID) & = && \upstate_\taut(\bauth_\ue^\ID)\\
            & = && \nonce^{j_0}
          \end{alignedat}$};

        \path (d) -- ++ (0,-0.8)
        node[sdn] (d1) {};

        \draw[sl] (a1) -- (d4);
        \draw[sl] (b1) -- (d1);
        \draw[sl] (c1) -- (d3);
      \end{tikzpicture}
    \end{center}
    It follows that:
    \[
      \left(
        \begin{array}[c]{l}
          \accept_\tau^\ID \;\wedge\;
          \accept_\tauo^\ID \;\wedge\;
          g(\inframe_\taut) = \nonce^{j_0}\;\wedge\\
          \pi_1(g(\inframe_\tauo)) =
          \enc{
            \spair{\ID}
            {\instate_\taut(\sqn_\ue^\ID)}}
          {\pk_\hn}{\enonce^j}
        \end{array}
      \right)
      \;\;\ra\;\;
      \auth_\tau(\ID,j_0)
    \]
    Hence:
    \[
      \sucauth_\tau(\ID)
      \;\;\ra\;\;
      \bigvee_{
        \tauo = \_,\pnai(j_0,1)
        \atop{\taut = \_,\npuai{1}{\ID}{j}
          \atop{\tautt = \_,\pnai(j_0,0)
            \atop{\tautt \potau \taut \potau \tauo}}}}
      \auth_\tau(\ID,j_0)
      \;\;\ra\;\;
      \bigvee_{j_0 \in \mathbb{N}}
      \auth_\tau(\ID,j_0)
    \]

  \item \textbf{Case $\ai = \cnai(j,0)$.} For all $\ID \in \iddom$ and for all $j_0 \in \mathbb{N}$ such that $j_0 \ne j$ we have:
    \begin{mathpar}
      \sucauth_\tau(\ID) \equiv
      \sucauth_{\tau_0}(\ID)

      \auth_\tau(\ID,j_0) \equiv
      \auth_{\tau_0}(\ID,j_0)
    \end{mathpar}
    Hence:
    \begin{equation*}
      \left(
        \sucauth_\tau(\ID)
        \;\wedge\; \auth_\tau(\ID,j_0)
      \right)
      \;\leftrightarrow\;
      \left(
        \sucauth_{\tau_0}(\ID)
        \;\wedge\; \auth_{\tau_0}(\ID,j_0)
      \right)
    \end{equation*}
    It only remains the case $j_0 = j$. We know that $\instate_{\tau_0}(\bauth_\hn^j) \equiv \bot$, therefore $\sucauth_{\tau_0}(\ID,j) = \false$, which in turn implies that:
    \[
      \left(
        \sucauth_{\tau_0}(\ID)
        \;\wedge\; \auth_{\tau_0}(\ID,j)
      \right)
      = \false
    \]
    Moreover:
    \begin{equation*}
      \auth_\tau(\ID,j)
      \;\ra\;
      \nonce^j = \instate_\tau(\eauth_\ue^\ID)
      \;\ra\;
      \nonce^j = \instate_\tau(\eauth_\ue^\ID)
    \end{equation*}
    Using \ref{a1} it is easy to show that $\nonce^j \not \in \st(\instate_\tau(\eauth_\ue^\ID))$, therefore we have $\neg \auth_\tau(\ID,j)$. This concludes this case.

  \item \textbf{Case $\ai = \cuai_\ID(j,1)$.} For all $\ID_0 \ne \ID$ and for all $j_0 \in \mathbb{N}$, it is trivial that:
    \begin{equation*}
      \left(
        \sucauth_\tau(\ID_0)
        \;\wedge\; \auth_\tau(\ID_0,j_0)
      \right)
      \;\leftrightarrow\;
      \left(
        \sucauth_{\tau_0}(\ID_0)
        \;\wedge\; \auth_{\tau_0}(\ID_0,j_0)
      \right)
    \end{equation*}
    Therefore we only need to show that:
    \[
      \textstyle
      \sucauth_\tau(\ID)
      \;\ra\;
      \bigvee_{i \in \mathbb{N}}
      \auth_\tau(\ID,i)
    \]
    Let $\key \equiv \key^\ID$. We observe that:
    \begin{alignat*}{2}
      \sucauth_\tau(\ID)
      &\;\;\ra\;\;&&
      \instate_\tau(\eauth_\ue^\ID) \ne \fail\\
      &\;\;\ra\;\;&&
      \accept_\tau^\ID\\
      &\;\;\ra\;\;&&
      \bigvee_{\tauo = \_, \cnai(j_0,0) \popre \tau}
      \begin{array}[c]{l}
        \accept_\tauo^\ID \;\wedge\;
        \pi_1(g(\inframe_\tau)) = \nonce^{j_0} \;\wedge\\
        \pi_2(g(\inframe_\tau)) =
        \instate_\tauo(\sqn_\hn^\ID) \oplus \ow{\nonce^{j_0}}{\key}
      \end{array}
      \tag{by \ref{acc3}}
    \end{alignat*}
    Let $\tauo = \cnai(j_0,0)$ such that $\tauo \potau \tau$. Then:
    \[
      \left(\pi_1(g(\inframe_\tau)) = \nonce^{j_0} \wedge \accept_\tau^\ID \right)
      \;\ra\;
      \instate_\tau(\eauth_\ue^\ID) = \nonce^{j_0}
    \]
    Moreover using \ref{a7} we know that $\accept^\ID_\tauo \ra \instate_\tauo(\bauth_\hn^j) = \ID$. Using the validity of $\tau$, we can easily show that for all $\tauo \potau \tau'$ we have $\upstate_{\tau'}(\bauth_\hn^j) \equiv \bot$. We deduce that $\accept^\ID_\tauo \ra \instate_\tauo(\bauth_\hn^j) = \ID$. Hence:
    \begin{align*}
      \sucauth_\tau(\ID)
      &&&\;\;\ra\;\;&&
      \bigvee_{\tauo = \_,\cnai(j_0,0) \popre \tau}
      \auth_\tau(\ID,j_0)
      &\;\;\ra\;\;&&
      \bigvee_{\tauo \popre \tau}
      \auth_\tau(\ID,j_0)
    \end{align*}
  \end{itemize}
\end{proof}

\subsection{Proof of Lemma~\ref{lem:auth-net-body}}
We give the proof of Lemma~\ref{lem:auth-net-body}, which relies on Lemma~\ref{lem:auth-net}.

\begin{proof}
  Let $\tau$ be a valid symbolic trace. First, observe that $\instate_\tau(\eauth_\ue^\ID) = \nonce^j$ implies that $\instate_\tau(\eauth_\ue^\ID) \ne \fail$ and that $\instate_\tau(\eauth_\ue^\ID) \ne \bot$.
  Using the remark above and Lemma~\ref{lem:auth-net} we get that:
  \begin{alignat*}{2}
    \instate_\tau(\eauth_\ue^\ID) = \nonce^j
    &\;\ra\;\;&&
    \instate_\tau(\eauth_\ue^\ID) \ne \fail \wedge
    \instate_\tau(\eauth_\ue^\ID) \ne \bot \\
    &\;\ra\;\;&&
    \sucauth_\tau(\ID)\\
    &\;\ra\;\;&&
    \bigvee_{j \in \mathbb{N}}
    \auth_\tau(\ID,j)
    \tag{By Lemma~\ref{lem:auth-net}}\\
    &\;\ra\;\;&&
    \instate_\tau(\bauth_\hn^j) = \ID
    \tag{Since $(\nonce^j = \instate_\tau(\eauth_\ue^\ID) \wedge  \nonce^{j'} = \instate_\tau(\eauth_\ue^\ID)) = \false$ if $j \ne j'$.}\\
    &\;\ra\;\;&&    \bigvee_{\tau' \popreleq \tau}\;
    \instate_{\tau'}(\bauth_\hn^j) = \ID
  \end{alignat*}
\end{proof}


\subsection{Injective Authentication of the Network by the User}
We actually can show that the authentication of the network by the user is \emph{injective}.
\begin{lemma}
  \label{lem:inj-auth-net}
  For all valid symbolic trace $\tau$, $\inframe_\tau$ guarantees \emph{injective} authentication of the network by the user. For all $\ID \in \iddom$ and $j \in \mathbb{N}$, we define the formula:
  \begin{gather*}
    \textstyle
    \injauth_\tau(\ID,j) \;\;\equiv\;\;
    \auth_\tau(\ID,j)
    \;\wedge\; \bigwedge_{i \ne j}
    \neg\auth_\tau(\ID,i)
  \end{gather*}
  Then we have:
  \[
    \textstyle
    \forall \ID \in \iddom,\;\;
    \sucauth_\tau(\ID)
    \;\ra\;
    \bigvee_{j \in \mathbb{N}}
    \injauth_\tau(\ID,j)
  \]
\end{lemma}

\begin{proof}
  First, we show that for $\ID \in \iddom$ and $i_0,i_1 \in \mathbb{N}$ with $i_0 \ne i_1$:
  \begin{equation}
    \sucauth_\tau(\ID) \ra
    \left(
      \neg \auth_\tau(\ID,i_0)
      \vee
      \neg \auth_\tau(\ID,i_1)
    \right)\label{eq:inj1aa}
  \end{equation}
  Indeed:
  \begin{alignat*}{2}
    &&&\sucauth_\tau(\ID) \wedge
    \auth_\tau(\ID,i_0)
    \wedge
    \auth_\tau(\ID,i_1)\\
    &\;\ra\;&&
    \sucauth_\tau(\ID) \wedge
    \instate_\tau(\nonce^{i_0}_\hn) = \instate_\tau(\eauth_\ue^\ID) \wedge
    \instate_\tau(\nonce^{i_1}_\hn) = \instate_\tau(\eauth_\ue^\ID)\\
    &\;\ra\;&&
    \sucauth_\tau(\ID) \wedge
    \begin{dcases}
      \nonce^{i_0} = \instate_\tau(\eauth_\ue^\ID) \wedge
      \nonce^{i_1} = \instate_\tau(\eauth_\ue^\ID) & \text{ if } \pnai(i_0,0) \in \tau \text{ and } \pnai(i_1,0) \in \tau\\
      \nonce^{i_0} = \instate_\tau(\eauth_\ue^\ID) \wedge
      \bot = \instate_\tau(\eauth_\ue^\ID) & \text{ if } \pnai(i_0,0) \in \tau \text{ and } \pnai(i_1,0) \not \in \tau\\
      \bot = \instate_\tau(\eauth_\ue^\ID) \wedge
      \nonce^{i_1} = \instate_\tau(\eauth_\ue^\ID) & \text{ if } \pnai(i_0,0) \not \in \tau \text{ and } \pnai(i_1,0) \in \tau\\
      \bot = \instate_\tau(\eauth_\ue^\ID) \wedge
      \bot = \instate_\tau(\eauth_\ue^\ID) & \text{ if } \pnai(i_0,0) \not \in \tau \text{ and } \pnai(i_1,0) \not \in \tau
    \end{dcases}
  \end{alignat*}
  Using $\ax{EQIndep}$, we know that $\nonce^{i_1} \ne \nonce^{i_0}$. Therefore:
  \[
    \left(
      \sucauth_\tau(\ID) \wedge
      \nonce^{i_0} = \instate_\tau(\eauth_\ue^\ID) \wedge
      \nonce^{i_1} = \instate_\tau(\eauth_\ue^\ID)
    \right)
    \ra \false
  \]
  Since $\sucauth_\tau(\ID) \ra \instate_\tau(\eauth_\ue^\ID) \ne \bot$, we know that:
  \[
    \left(\sucauth_\tau(\ID) \wedge \bot = \instate_\tau(\eauth_\ue^\ID) \right) \ra \false
  \]
  And therefore:
  \begin{gather*}
    \left(
      \sucauth_\tau(\ID) \wedge
      \nonce^{i_0} = \instate_\tau(\eauth_\ue^\ID) \wedge
      \bot = \instate_\tau(\eauth_\ue^\ID)
    \right) \ra \false\\
    \left(
      \sucauth_\tau(\ID) \wedge
      \bot = \instate_\tau(\eauth_\ue^\ID)\wedge
      \nonce^{i_1} = \instate_\tau(\eauth_\ue^\ID)
    \right) \ra \false\\
    \left(
      \sucauth_\tau(\ID) \wedge
      \bot = \instate_\tau(\eauth_\ue^\ID) \wedge
      \bot = \instate_\tau(\eauth_\ue^\ID)
    \right) \ra \false
  \end{gather*}
  This concludes the proof of \eqref{eq:inj1aa}. From Lemma~\ref{lem:auth-net} we know that:
  \[
    \textstyle
    \forall \ID \in \iddom,\;\;
    \sucauth_\tau(\ID)
    \;\ra\;
    \bigvee_{j \in \mathbb{N}}
    \auth_\tau(\ID,j)
  \]
  Moreover, using \eqref{eq:inj1aa} we have that for every $\ID \in \iddom, j \in \mathbb{N}$:
  \[
    \textstyle
    \sucauth_\tau(\ID)
    \wedge
    \auth_\tau(\ID,j)
    \;\ra\;
    \bigvee_{i \ne j}
    \neg\auth_\tau(\ID,i)
  \]
  We deduce that:
  \[
    \textstyle
    \forall \ID \in \iddom,\;\;
    \sucauth_\tau(\ID)
    \;\ra\;
    \bigvee_{j \in \mathbb{N}}
    \injauth_\tau(\ID,j)
    \qedhere
  \]
\end{proof}

\begin{proposition}
  \label{prop:injauth-charac}
  For every valid symbolic trace $\tau$, for every $j_o \in \mathbb{N}$:
  \[
    \injauth_\tau(\ID,j_0)
    \;\lra\;
    \nonce^{j_0} = \instate_\tau(\eauth_\ue^\ID)
  \]
\end{proposition}

\begin{proof}
  To do this we show both directions. The first direction is trivial:
  \[
    \injauth_\tau(\ID,j_0)
    \;\ra\;
    \auth_\tau(\ID,j_0)
    \;\ra\;
    \left(\nonce^{j_0} = \instate_\tau(\eauth_\ue^\ID)\right)
  \]
  We now prove the converse direction:
  \begin{alignat*}{2}
    \left(\nonce^{j_0} = \instate_\tau(\eauth_\ue^\ID)\right)
    &\;\ra\;&&
    \sucauth_\tau(\ID)
    \tag{Using $\ax{EQIndep}$}\\
    &\;\ra\;&&
    \bigvee_{j_1 \in \mathbb{N}}
    \injauth_\tau(\ID,j_1)
    \tag{Lemma~\ref{lem:inj-auth-net}}
  \end{alignat*}
  We conclude by observing that for every $j_1 \ne j_0$,
  \begin{alignat*}{2}
    \left(\nonce^{j_0} = \instate_\tau(\eauth_\ue^\ID)\right)
    \wedge
    \injauth_\tau(\ID,j_1)
    &\;\ra\;&&
    \left(\nonce^{j_0} = \instate_\tau(\eauth_\ue^\ID)\right)
    \wedge
    \auth_\tau(\ID,j_1)\\
    &\;\ra\;&&
    \left(\nonce^{j_0} = \instate_\tau(\eauth_\ue^\ID)\right)
    \wedge
    \left(\nonce^{j_1} = \instate_\tau(\eauth_\ue^\ID)\right)\\
    &\;\ra\;&&
    \false
    \tag{Using $\ax{EQIndep}$}
  \end{alignat*}
\end{proof}


%% file: acceptance.tex
\section{Acceptance Characterizations}
\label{section:acc-charac-app-first}

In this section, we prove necessary and sufficient conditions for a message to be accepted by the user or the network. This section is organized as follow: we start by showing some properties of the $\faka$ protocol, which we then use to show a first set of acceptance characterizations; then, using these, we show that the temporary identity $\suci_\ue^\ID$ is concealed until the subscriber starts of session of the $\suci$ sub-protocol; finally, using the $\suci$ concealment property, we show stronger acceptance characterizations.

\subsection{First Characterizations}

\begin{proposition}
  \label{prop:equiv-props-ini}
  For every valid symbolic trace $\tau = \_,\ai$ and identity $\ID$ we have:
  \begin{itemize}
  \item \customlabel{b5}{\lpb{1}} For every $\tauo \popreleq \taut \popreleq \tau$, $\cstate_\tauo(\sqn_\scx^\ID) \le \cstate_\taut(\sqn_\scx^\ID)$.

  \item \customlabel{b3}{\lpb{2}} If $\ai = \fuai_\ID(j)$ then for every and $j_0 \in \mathbb{N}$, if $\fnai(j_0) \potau \newsession_\ID(\_)$ then:
    \[
      \cstate_\tau(\eauth_\hn^{j_0}) \ne \unknownid \ra
      \neg\injauth_\tau(\ID,j_0)      
    \]
  \end{itemize}
\end{proposition}

\begin{proof}
  Let $\tau = \_,\ai$ be valid symbolic trace and $\ID \in \iddom$. We prove \ref{b5} and \ref{b3}:
  \begin{itemize}
  \item \ref{b5}. This is straightforward by induction over
    $\taut$.

  \item \ref{b3}. Let $\taux = \_,\fnai(j_0) \popre \tau$. We do a case disjunction on the protocol used by the user for authentication:
    \begin{itemize}
    \item If there exists $\taut = \_,\cuai_\ID(j,1) \popre \tau$. We know that there exists $\taun \popre \taux$ with $\taun = \_,\pnai(j_0,1)$ or $\_,\cnai(j_0,1)$.

      Assume that $\taun = \_,\pnai(j_0,1)$. We know that $\injauth_{\tau}(\ID,j_0) \ra \accept_\taut^\ID$, and by applying \ref{acc3}:
      \begin{alignat*}{2}
        \injauth_{\tau}(\ID,j_0)
        &\;\ra\;\;&&
        \bigvee_{\tautt = \_, \cnai(j_2,0)
          \atop{\tautt \popre \taut}}
        \cstate_\taut(\eauth_\ue^\ID) = \nonce^{j_2}\\
        &\;\ra\;\;&&
        \cstate_\taut(\eauth_\ue^\ID) \ne \nonce^{j_0}
        \tag{Since for every $\tautt = \_, \cnai(j_2,0) \popre \taut$,
          $j_2 \ne j_0$}\\
        &\;\ra\;\;&&
        \false
      \end{alignat*}
      Which is what we wanted.

      Now, assume that $\taun = \_,\cnai(j_0,1)$. Observe that $\cstate_\taun(\eauth_\hn^{j_0}) \ne \fail$ and that $\cstate_\tau(\eauth_\hn^{j_0}) = \cstate_\taun(\eauth_\hn^{j_0})$. Moreover, it is straightforward to check that for every valid symbolic trace $\tau'$:
      \[
        \left(
          \injauth_{\tau'}(\ID,j_0) \wedge
          \cstate_{\tau'}(\eauth_\hn^{j_0}) \ne \unknownid \wedge
          \cstate_{\tau'}(\eauth_\hn^{j_0}) \ne \fail
        \right)
        \ra
        \cstate_{\tau'}(\eauth_\hn^{j_0}) = \cstate_{\tau'}(\bauth_\hn^{j_0})
      \]
      Hence we deduce that:
      \[
        \left(
          \injauth_\tau(\ID,j_0) \wedge
          \cstate_\tau(\eauth_\hn^{j_0}) \ne \unknownid
        \right)
        \ra
        \cstate_\tau(\eauth_\hn^{j_0}) = \cstate_\tau(\bauth_\hn^{j_0})
      \]
      Since $\injauth_\tau(\ID,j_0) \ra \cstate_\tau(\bauth_\hn^{j_0}) = \ID$, we get that:
      \[
        \left(
          \injauth_\tau(\ID,j_0)  \wedge
          \cstate_\tau(\eauth_\hn^{j_0}) \ne \unknownid
        \right)
        \ra
        \cstate_\tau(\eauth_\hn^{j_0}) = \ID
      \]
      Moreover, $\cstate_\tau(\eauth_\hn^{j_0}) = \ID \ra \accept_\taun^\ID$. Using \ref{acc4} on $\taun$:
      \[
        \accept_\taun^\ID
        \;\ra\;
        \bigvee_{\taui = \_,\cuai_\ID(j_i,1) \popre \taun}
        \accept^\ID_{\taui} \wedge \pi_1(g(\inframe_{\taui})) = \nonce^{j_0}
      \]
      Let $\tauo = \cnai(j_0,0)$ and $\taui = \_,\cuai_\ID(j_i,1) \popre \taun$. Observe that $\taui \ne \taut$. Using \ref{acc3}, we can check that:
      \[
        \accept^\ID_{\taui} \wedge \pi_1(g(\inframe_{\taui})) = \nonce^{j_0}
        \ra
        \range{\instate_\taui(\sqn_\ue^\ID)}{\instate_\tauo(\sqn_\hn^\ID)}
      \]
      Recall that $\injauth_{\tau}(\ID,j_0) \ra \accept_\taut^\ID$. Moreover, $\injauth_{\tau}(\ID,j_0) \ra \pi_1(g(\inframe_{\taut})) = \nonce^{j_0}$. Hence using \ref{acc3} again we get: 
      \[
        \accept^\ID_{\taut} \wedge \pi_1(g(\inframe_{\taut})) = \nonce^{j_0}
        \ra
        \range{\instate_\taut(\sqn_\ue^\ID)}{\instate_\tauo(\sqn_\hn^\ID)}
      \]
      Putting everything together:
      \begin{alignat*}{2}
        \left(
          \injauth_\tau(\ID,j_0) \wedge
          \cstate_\tau(\eauth_\hn^{j_0}) \ne \unknownid
        \right)
        &\;\ra\;\;&&
        \left(
          \begin{alignedat}[c]{2}
            &&&
            {\instate_\taui(\sqn_\ue^\ID)} =
            {\instate_\tauo(\sqn_\hn^\ID)} \\
            &\wedge \;\;&&
            {\instate_\taut(\sqn_\ue^\ID)} =
            {\instate_\tauo(\sqn_\hn^\ID)}
          \end{alignedat}  
        \right)\\
        &\;\ra\;\;&&
        {\instate_\taut(\sqn_\ue^\ID)} =
        {\instate_\taui(\sqn_\ue^\ID)}
      \end{alignat*}
      Finally, $\accept_\taui^\ID \ra \instate_\taui(\sqn_\ue^\ID) < \cstate_\taui(\sqn_\ue^\ID)$, and using \ref{b5} we know that $ \cstate_\taui(\sqn_\ue^\ID) \le {\instate_\taut(\sqn_\ue^\ID)}$. We deduce that:
      \begin{alignat*}{2}
        \left(
          \injauth_\tau(\ID,j_0) \wedge
          \cstate_\tau(\eauth_\hn^{j_0}) \ne \unknownid
        \right)
        &\;\ra\;\;&&
        {\instate_\taut(\sqn_\ue^\ID)} =
        {\instate_\taui(\sqn_\ue^\ID)} <
        {\instate_\taut(\sqn_\ue^\ID)}\\
        &\;\ra\;\;&&
        \false
      \end{alignat*}
      This concludes this case. We summarize graphically this proof below:
      \begin{center}
        \begin{tikzpicture}
          [dn/.style={inner sep=0.2em,fill=black,shape=circle},
          sdn/.style={inner sep=0.15em,fill=white,draw,solid,shape=circle},
          sl/.style={decorate,decoration={snake,amplitude=1.6}},
          dl/.style={dashed},
          pin distance=0.5em,
          every pin edge/.style={thin}]

          \draw[thick] (0,0)
          node[left=1.3em] {$\tau:$}
          -- ++(0.5,0)
          node[dn,pin={above:{$\cnai(j_0,0)$}}]
          (a) {}
          node[below,yshift=-0.3em,name=a0] {$\tauo$}
          -- ++(3,0)
          node[dn,pin={above:{$\cuai_\ID(j_i,1)$}}]
          (b) {}
          node[below,yshift=-0.3em,name=b0] {$\taui$}
          -- ++(3,0)
          node[dn]
          (bp) {}
          node[below,yshift=-0.3em,name=bp0] {}
          -- ++(2,0)
          node[dn,pin={above:{$\cnai(j_0,1)$}}]
          (c) {}
          node[below,yshift=-0.3em,name=c0] {$\taun$}
          -- ++(1.5,0)
          node[dn,pin={above:{$\fnai(j_0)$}}]
          (d) {}
          node[below,yshift=-0.3em] {$\taux$}
          -- ++(1.5,0)
          node[dn,pin={above:{$\ns_\ID(\_)$}}]
          (e) {}
          -- ++(1.5,0)
          node[dn,pin={above:{$\cuai_\ID(j,1)$}}]
          (f) {}
          node[below,yshift=-0.3em,name=f0] {$\taut$}
          -- ++(0.5,0);

          \draw[thin,dashed] (a0) -- ++(0,-0.5) -| (b0)
          {[draw=none] -- ++(0,-0.5)} -| (c0);
          
          \draw[thin,dotted] (a0) {[draw=none] -- ++(0,-0.5)}
          -- ++(0,-0.5) -| (f0);

          \path (a) -- ++ (0,-2)
          node (a2) {$\instate_\tauo(\sqn_\hn^\ID)$};

          \path (b) -- ++ (0,-2)
          -- ++ (0,-1.3)
          node (b2) {$\instate_\taui(\sqn_\ue^\ID)$};

          \path (bp) -- ++ (0,-2)
          -- ++ (0,-1.3)
          node (bp2) {$\cstate_\taui(\sqn_\ue^\ID)$};

          \path (f) -- ++ (0,-2)
          -- ++ (0,-1.3)
          node (f2) {$\instate_\taut(\sqn_\ue^\ID)$};

          \draw (a2) -- (b2) node[midway,below,sloped]{$=$};
          \draw (a2) to[bend left=6] node[midway,above,sloped]{$=$} (f2);
          \draw (b2) -- (bp2) node[midway,below]{$<$}
          (bp2) -- (f2) node[midway,below]{$\le$};
        \end{tikzpicture}
      \end{center}

    \item If there exists $\taut = \_,\npuai{2}{\ID}{j} \popre \tau$. Let $\tauttt = \_,\npuai{1}{\ID}{j} \popre \taut$, we know that $\taux \popre \tauttt$. Remark that $\injauth_\tau(\ID,j_0) \ra \accept_\taut^\ID$, and using \ref{acc2} we easily get that:
      \[
        \accept_\taut^\ID \ra
        \bigvee_{\tautt = \_, \pnai(j_2,1)
          \atop{\tauttt \popre \tautt \popre \taut}}
        \instate_\taut(\eauth_\ue^\ID) = \nonce^{j_2}
      \]
      Since no $\ID$ action occurred between $\taut$ and $\tau$, we have $\instate_\taut(\eauth_\ue^\ID) = \cstate_\tau(\eauth_\ue^\ID)$. Moreover, $\injauth_\tau(\ID,j_0) \ra \cstate_\tau(\eauth_\ue^\ID) = \nonce^{j_0}$. Finally, for every $\tautt = \_, \pnai(j_2,1)$ such that $\tauttt \popre \tautt \popre \taut$, since $\taux \popre \tauttt$ we know that $j_2 \ne j_0$. It follows that:
      \[
        \injauth_\tau(\ID,j_0) \ra
        \left(
          \textstyle
          \bigvee_{\tautt = \_, \pnai(j_2,1)
            \atop{\tauttt \popre \tautt \popre \taut}}
          \nonce^{j_0} = \nonce^{j_2}
        \right)
        \ra
        \false
      \]
      This concludes this case.
    \end{itemize}
  \end{itemize}
\end{proof}

We now prove a first acceptance characterization:
\begin{lemma}
  \label{lem:equiv-accept-ini}
  For every valid symbolic trace $\tau = \_,\ai$ and identity $\ID$ we have:
  \begin{itemize}
  \item \customlabel{equ1}{\lpequ{1}} If $\ai = \fuai_\ID(j)$. For every $\taut = \_,\fnai(j_0) \popre \tau$, we let:
    \[
      \futr{\tau}{\taut} \;\equiv\;
      \left(\begin{alignedat}{2}
          &\injauth_\tau(\ID,j_0)
          \wedge\instate_\tau(\eauth_\hn^{j_0}) \ne \unknownid\\
          \wedge\;& \pi_1(g(\inframe_\tau)) =
          \suci^{j_0} \xor \row{\nonce^{j_0}}{\key}
          \wedge\; \pi_2(g(\inframe_\tau)) =
          \mac{\spair
            {\suci^{j_0}}
            {\nonce^{j_0}}}
          {\mkey}{5}
        \end{alignedat}\right)
    \]
    Then:
    \begin{alignat*}{3}
      \accept_\tau^\ID
      &\;\;\leftrightarrow\;\;&&
      \bigvee_{\taut = \_,\fnai(j_0) \popre \tau
        \atop{\taut \not \potau \newsession_\ID(\_)}}
      \futr{\tau}{\taut}
    \end{alignat*}
  \end{itemize}
\end{lemma}

\begin{proof}
  Using Lemma~\ref{lem:inj-auth-net} we know that:
  \[
    \sucauth_\tau(\ID)
    \;\ra\;
    \bigvee_{j_0 \in \mathbb{N}}
    \injauth_\tau(\ID,j_0)
  \]
  Let $\key \equiv \key_\ID$ and $\mkey \equiv \mkey^{\ID}$. Since:
  \[
    \accept_\tau^\ID \;\equiv\; \sucauth_\tau(\ID)
    \;\wedge\;
    \underbrace{
      \pi_2(g(\inframe_\tau)) =
      \mac
      {\spair
        {\pi_1(g(\inframe_\tau)) \xor \row{\instate_\tau(\eauth_\ue^{\ID})}{\key}}
        {\instate_\tau(\eauth_\ue^{\ID})}}
      {\mkey}{5}}_{\textsf{EQMac}}
  \]
  And since $\injauth_\tau(\ID,j_0) \;\ra\;\sucauth_\tau(\ID)$ we have:
  \begin{alignat*}{2}
    \accept_\tau^\ID
    &\;\;\leftrightarrow\;\;&&
    \bigvee_{j_0 \in \mathbb{N}} \injauth_\tau(\ID,j_0) \wedge \textsf{EQMac}\\
    &\;\;\leftrightarrow\;\;&&
    \bigvee_{j_0 \in \mathbb{N}} \injauth_\tau(\ID,j_0) \wedge
    \pi_2(g(\inframe_\tau)) =
    \mac
    {\spair
      {\pi_1(g(\inframe_\tau)) \xor \row{\nonce^{j_0}}{\key}}
      {\nonce^{j_0}}}
    {\mkey}{5}
  \end{alignat*}
  Using the $\textsc{p-euf-mac}^5$ and $\textsc{cr}^5$ axioms, it is easy to show that for every $j_0 \in \mathbb{N}$:
  \[
    \pi_2(g(\inframe_\tau)) =
    \mac
    {\spair
      {\pi_1(g(\inframe_\tau)) \xor \row{\nonce^{j_0}}{\key}}
      {\nonce^{j_0}}}
    {\mkey}{5}
    \;\ra\;
    \begin{dcases*}
      \left(
        \begin{alignedat}{2}
          &&&\pi_1(g(\inframe_\tau)) \xor \row{\nonce^{j_0}}{\key} = \suci^{j_0}\\
          &\wedge\,&&\instate_\tau(\eauth_\hn^{j_0}) \ne \unknownid
        \end{alignedat}
      \right)
      & if $\fnai(j_0) \in \tau$\\
      \false & otherwise
    \end{dcases*}
  \]
  Hence:
  \begin{alignat*}{2}
    \accept_\tau^\ID
    &\;\;\leftrightarrow\;\;&&
    \bigvee_{\tauo = \_,\fnai(j_0) \popre \tau}
    \left(\begin{alignedat}{2}
        &\injauth_\tau(\ID,j_0) \wedge
        \instate_\tau(\eauth_\hn^{j_0}) \ne \unknownid\\
        \wedge\;& \pi_1(g(\inframe_\tau)) =
        \suci^{j_0} \xor \row{\nonce^{j_0}}{\key} \wedge
        \pi_2(g(\inframe_\tau)) =
        \mac
        {\spair
          {\suci^{j_0}}
          {\nonce^{j_0}}}
        {\mkey}{5}
      \end{alignedat}\right)\\
    \intertext{By \ref{b3}:}
    \accept_\tau^\ID
    &\;\;\leftrightarrow\;\;&&
    \bigvee_{\tauo = \_,\fnai(j_0) \popre \tau
      \atop{\tauo \not \potau \newsession_\ID(\_)}}
    \left(\begin{alignedat}{2}
        &\injauth_\tau(\ID,j_0) \wedge
        \instate_\tau(\eauth_\hn^{j_0}) \ne \unknownid\\
        \wedge\;& \pi_1(g(\inframe_\tau)) =
        \suci^{j_0} \xor \row{\nonce^{j_0}}{\key} \wedge \pi_2(g(\inframe_\tau)) =
        \mac
        {\spair
          {\suci^{j_0}}
          {\nonce^{j_0}}}
        {\mkey}{5}
      \end{alignedat}\right)
  \end{alignat*}
  Which concludes this proof.
\end{proof}

We show the following additional properties:
\begin{proposition}
  \label{prop:equiv-props}
  For every valid symbolic trace $\tau = \_,\ai$ and identity $\ID$ we have:
  \begin{itemize}
  \item \customlabel{b1}{\lpb{3}} $\cstate_\tau(\success_\ue^\ID) \;\ra\; \cstate_\tau(\suci_\ue^\ID) \ne \unset$.
  \item \customlabel{b2}{\lpb{4}} For every $\tautt \potau \taut$:
    \[
      \cstate_\tautt(\sqn_\hn^\ID) < \instate_\taut(\sqn_\hn^\ID)
      \;\ra\;
      \qquad\bigvee_{\mathclap{
          \tautt \potau \taux \potau \taut
          \atop{\taux = \_,\cnai(j_x,0)
            ,\_,\cnai(j_x,1)
            \text{ or }  \_,\pnai(j_x,1)}}}\qquad
      \instate_\taut(\tsuccess_\hn^\ID) = \nonce^{j_x}
    \]
  \item \customlabel{b4}{\lpb{5}} $\cstate_\tau(\sqn_\hn^\ID) \le \cstate_\tau(\sqn_\ue^\ID)$.
  \item \customlabel{b6}{\lpb{6}} For every $\tauo \potau \taut$ such that $\tauo = \_,\ns_\ID(\_)$ or $\epsilon$, and such that $\tauo \not \potau \newsession_\ID(\_)$, we have:
    \begin{mathpar}
      \cstate_\taut(\sync_\ue^\ID) \;\ra\;
      \cstate_\taut(\sqn_\hn^\ID) > \cstate_\tauo(\sqn_\ue^\ID)
    \end{mathpar}

  \item \customlabel{b7}{\lpb{7}} If for all $\tau' \popreleq \tau$ such that $\tau' \not \potau \newsession_\ID(\_)$ we have $\tau' \ne \_, \fuai_\ID(\_)$, then:
    \[
      \cstate(\success_\ue^\ID)
      \;\ra\;
      \false
    \]

  \end{itemize}
\end{proposition}

\begin{proof}
  We give the proof of the properties \ref{b1} to \ref{b7}.
  \begin{itemize}
  \item \ref{b1}. We show this by induction over $\tau$. If $\tau = \epsilon$, we know from Definition~\ref{def:init-sigma-phi} that $\cstate_\epsilon(\success^{\ID}_\ue) \equiv \false$ and $\cstate_\epsilon(\suci^\ID_\scx) \equiv \unset$.  Therefore the property holds. Let $\tau = \tauo,\ai$, assume by induction that the property holds for $\tauo$. If $\ai$ is different from $\cuai_\ID(j,0),\npuai{1}{\ID}{j}$ and $\fuai(j)$ then $\upstate_\tau(\success_\ue^\ID) \equiv \upstate_\tau(\suci_\ue^\ID) \equiv \bot$, in which case we conclude immediately by induction hypothesis. We have three cases remaining:
    \begin{itemize}
    \item If $\ai = \cuai_\ID(j,0)$ or $\ai = \npuai{1}{\ID}{j}$ then $\upstate_\tau(\suci_\ue^\ID) \equiv \false$. Therefore the property holds.
    \item If $\ai = \fuai(j)$, using \ref{equ1} we can check that:
      \begin{equation*}
        \accept_\tau^\ID
        \;\ra\;
        \bigvee_{\taut = \_,\fnai(j_0) \popre \tau
          \atop{\taut \not \potau \newsession_\ID(\_)}}
        \left(\cstate_\tau(\suci_\ue^\ID) = \suci^{j_0}\right)
        \;\ra\;
        \cstate_\tau(\suci_\ue^\ID) \ne \unset
      \end{equation*}
      We conclude by observing that $\cstate_\epsilon(\success^{\,\ID}_\ue) \equiv \accept_\tau^\ID$.
    \end{itemize}
    
  \item \ref{b2}. We prove this directly. Intuitively, this holds because if $\cstate_\tautt(\sqn_\hn^\ID) < \instate_\taut(\sqn_\hn^\ID)$ then we know that $\sqn_\hn^\ID$ was updated between $\tautt$ and $\taut$. Moreover, if such an update occurs at $\taux = \_,\pnai(j_x,1)$ or $\cnai(j_x,1)$ then $\tsuccess_\hn^\ID$ has to be equal to $\nonce^{j_x}$ after the update. Finally, the fact that $\tsuccess_\hn^\ID$ is equal to $\nonce^{j_x}$ for some $\taux$ between $\tautt$ and $\taut$ with $\taux = \_,\cnai(j_x,0)$,  $\_,\cnai(j_x,1)$ or $\_,\pnai(j_x,1)$ is an invariant of the protocol. Now we give the formal proof.

    First, we remark that $\sqn_\hn^\ID$ is updated only at $\pnai(\_,1)$ and $\cnai(\_,1)$. Moreover, each update either left $\sqn_\hn^\ID$ unchanged or increments it by one. Finally, it is updated at $\taux \popre \tau$ if and only if $\incaccept_\taux^\ID$ holds. If follows that:
    \[
      \cstate_\tautt(\sqn_\hn^\ID) < \instate_\taut(\sqn_\hn^\ID)
      \;\ra\;
      \qquad\bigvee_{\mathclap{
          \tautt \potau \taux \potau \taut
          \atop{\taux = ,\_,\cnai(j_x,1)
            \text{ or }  \_,\pnai(j_x,1)}}}\qquad
      \incaccept_{\taux}^\ID
    \]
    We know that for every $\tautt \potau \taux \potau \taut$, if:
    \begin{itemize}
    \item $\taux = ,\_,\cnai(j_x,1)$ then $\incaccept_{\taux}^\ID \lra \cstate_\taux(\session_\hn^\ID) = \nonce^{j_x}$.
    \item $\taux = \_,\pnai(j_x,1)$ then since $\incaccept_{\taux}^\ID \equiv \incaccept_{\taux}^\ID \wedge \instate_\taux(\session_\hn^\ID) = \nonce^{j_x}$, we know that $\incaccept_{\taux}^\ID \ra \instate_\taux(\session_\hn^\ID) = \nonce^{j_x}$. Besides, since $\session_\hn^\ID$ is not updated at $\pnai(j_x,1)$ we deduce that $\incaccept_{\taux}^\ID \ra \cstate_\taux(\session_\hn^\ID) = \nonce^{j_x}$.
    \end{itemize}
    Hence:
    \[
      \cstate_\tautt(\sqn_\hn^\ID) < \instate_\taut(\sqn_\hn^\ID)
      \;\ra\;
      \qquad\bigvee_{\mathclap{
          \tautt \potau \taux \popreleq \taut
          \atop{\taux = ,\_,\cnai(j_x,1)
            \text{ or }  \_,\pnai(j_x,1)}}}\qquad
      \cstate_\taux(\session_\hn^\ID) = \nonce^{j_x}
      \numberthis\label{eq:fiqejfcnlcaprqww}
    \]
    Let $\tautt \potau \taux \potau \taut$ such that $\taux = ,\_,\cnai(j_x,1)$ or $\_,\pnai(j_x,1)$. Now, we prove by induction over $\tau'$ such that $\taux \popreleq \tau' \popre \taut$ that:
    \[
      \cstate_\taux(\session_\hn^\ID) = \nonce^{j_x}
      \;\ra\;
      \qquad\bigvee_{\mathclap{
          \taux \popreleq \taun \popreleq \tau'
          \atop{\taun = \_,\cnai(j_n,0)
            ,\_,\cnai(j_n,1)
            \text{ or }  \_,\pnai(j_n,1)}}}\qquad
      \cstate_{\tau'}(\tsuccess_\hn^\ID) = \nonce^{j_n}
    \]
    For $\tau' = \taux$ this is obvious. For the inductive case, we do a disjunction over the final action of $\tau'$. If $\tsuccess_\hn^\ID$ is not updated then we conclude by induction, otherwise we are in one of the following case:
    \begin{itemize}
    \item If $\tau' = \_,\cnai(j',0)$ then we do a case disjunction on $\accept_{\tau'}^\ID$:
      \[
        \neg\accept_{\tau'}^\ID
        \ra
        \cstate_{\tau'}(\tsuccess_\hn^\ID) =
        \instate_{\tau'}(\tsuccess_\hn^\ID)
        \numberthis\label{eq:fdfhfqeirjepo}
      \]
      Hence:
      \begin{alignat*}{2}
        \neg\accept_{\tau'}^\ID \wedge \cstate_\taux(\session_\hn^\ID) = \nonce^{j_x}
        &\;\ra\;\;&&
        \qquad\bigvee_{\mathclap{
            \taux \popreleq \taun \popre \tau'
            \atop{\taun = \_,\cnai(j_n,0)
              ,\_,\cnai(j_n,1)
              \text{ or }  \_,\pnai(j_n,1)}}}\qquad
        \instate_{\tau'}(\tsuccess_\hn^\ID) = \nonce^{j_n}
        \tag{by induction hypothesis}\\
        &\;\ra\;\;&&
        \qquad\bigvee_{\mathclap{
            \taux \popreleq \taun \popre \tau'
            \atop{\taun = \_,\cnai(j_n,0)
              ,\_,\cnai(j_n,1)
              \text{ or }  \_,\pnai(j_n,1)}}}\qquad
        \cstate_{\tau'}(\tsuccess_\hn^\ID) = \nonce^{j_n}
        \tag{Using \eqref{eq:fdfhfqeirjepo}}
        \displaybreak[0]\\
        &\;\ra\;\;&&
        \qquad\bigvee_{\mathclap{
            \taux \popreleq \taun \popreleq \tau'
            \atop{\taun = \_,\cnai(j_n,0)
              ,\_,\cnai(j_n,1)
              \text{ or }  \_,\pnai(j_n,1)}}}\qquad
        \cstate_{\tau'}(\tsuccess_\hn^\ID) = \nonce^{j_n}
        \tag{Relaxing the condition $\taun \popre \tau'$}
      \end{alignat*}
      Moreover,
      \[
        \accept_{\tau'}^\ID
        \ra
        \left(\cstate_{\tau'}(\tsuccess_\hn^\ID) =
          \nonce^{j'}\right)
        \;\ra\;
        \qquad\bigvee_{\mathclap{
            \taux \popreleq \taun \popreleq \tau'
            \atop{\taun = \_,\cnai(j_n,0)
              ,\_,\cnai(j_n,1)
              \text{ or }  \_,\pnai(j_n,1)}}}\qquad
        \cstate_{\tau'}(\tsuccess_\hn^\ID) = \nonce^{j_n}
      \]
      This concludes this case.
    \item If $\taun = \_,\pnai(j_n,1)$ then the proof is the same than in the previous case, but doing a case disjunction over $\incaccept_{\tau'}^\ID$.
    \end{itemize}
    Let $\tauo'$ be such that $\taut = \tauo',\_$. By applying the induction hypothesis to $\tauo'$, we get:
    \[
      \cstate_\taux(\session_\hn^\ID) = \nonce^{j_x}
      \;\ra\;
      \qquad\bigvee_{\mathclap{
          \taux \popreleq \taun \popreleq \tauo'
          \atop{\taun = \_,\cnai(j_n,0)
            ,\_,\cnai(j_n,1)
            \text{ or }  \_,\pnai(j_n,1)}}}\qquad
      \cstate_{\tauo'}(\tsuccess_\hn^\ID) = \nonce^{j_n}
      \;\ra\;
      \qquad\bigvee_{\mathclap{
          \taux \popreleq \taun \popre \taut
          \atop{\taun = \_,\cnai(j_n,0)
            ,\_,\cnai(j_n,1)
            \text{ or }  \_,\pnai(j_n,1)}}}\qquad
      \instate_{\taut}(\tsuccess_\hn^\ID) = \nonce^{j_n}
    \]
    We conclude using \eqref{eq:fiqejfcnlcaprqww} and the property above.    

  \item \ref{b4}. We prove this by induction over $\tau$. For $\tau = \epsilon$, from Definition~\ref{def:init-sigma-phi} we know that $\cstate_\epsilon(\sqn^{\ID}_\ue) \equiv \sqnini_\ue^\ID$ and $\cstate_\epsilon(\sqn^{\ID}_\hn) \equiv \sqnini_\hn^\ID$. Using the axiom $\sqn$-\textsf{ini}, we know that $\sqnini_\hn^\ID \le \sqnini_\ue^\ID$.

    For the inductive case, we let $\tau = \tauo,\ai$ and assume that the property holds for $\tauo$. We have three cases:
    \begin{itemize}
    \item If when executing the action $\ai$ the value $\sqn_\hn^\ID$ is not updated. Using \ref{b5} we know that $\cstate_\tau(\sqn^{\ID}_\ue) \ge \cstate_\tauo(\sqn^{\ID}_\ue)$, and we conclude by applying the induction hypothesis.
    \item If $\ai = \pnai(j,1)$, then we do a case disjunction on $\incaccept_\tau^\ID$. If it is true then: 
      \begin{alignat*}{2}
        \incaccept_\tau^\ID
        &\;\ra\;\;&&
        \bigvee_{
          \tauo = \_,\npuai{1}{\ID}{j_0} \popre \tau}
        \cstate_\tau(\sqn_\hn^\ID) = \instate_\tauo(\sqn_\ue^\ID)
        \tag{Using \ref{acc1}}\displaybreak[0]\\
        &\;\ra\;\;&&
        \bigvee_{
          \tauo = \_,\npuai{1}{\ID}{j_0} \popre \tau}
        \cstate_\tau(\sqn_\hn^\ID) = \instate_\tauo(\sqn_\ue^\ID)
        \wedge \instate_\tauo(\sqn_\ue^\ID) \le \cstate_\tau(\sqn_\ue^\ID)
        \tag{Using \ref{b5}}\\
        &\;\ra\;\;&&
        \cstate_\tau(\sqn_\hn^\ID) \le \cstate_\tau(\sqn_\ue^\ID)
      \end{alignat*}
      If $\incaccept_\tau^\ID$ is false then $\neg \incaccept_\tau^\ID \ra \cstate_\tau(\sqn_\hn^\ID) = \instate_\tau(\sqn_\hn^\ID)$, and we conclude by applying the induction hypothesis.

    \item If $\ai = \cnai(j,1)$, then we do a case disjunction on $\incaccept_\tau^\ID$. First we handle the case where it is true. Let $\tautt = \_,\cnai(j,0) \popre \tau$. We know that $\incaccept_\tau^\ID \ra \instate_\tau(\tsuccess_\hn^\ID) = \nonce^j$. Moreover:
      \begin{alignat*}{2}
        \instate_\tau(\tsuccess_\hn^\ID) = \nonce^j
        &\;\ra\;\;&&
        \qquad\bigwedge_{\mathclap{
            \tautt \popre \taut \popre \tau
            \atop{\taut = \_,\cnai(j_x,0)
              ,\_,\cnai(j_x,1)
              \text{ or }  \_,\pnai(j_x,1)}}}\qquad
        \instate_\tau(\tsuccess_\hn^\ID) \ne \nonce^{j_x}\displaybreak[0]\\
        &\;\ra\;\;&&
        \cstate_\tautt(\sqn_\hn^\ID) \le \instate_\tau(\sqn_\hn^\ID)
        \tag{Using the contrapositive of \ref{b2}}\\
        &\;\ra\;\;&&
        \cstate_\tautt(\sqn_\hn^\ID) = \instate_\tau(\sqn_\hn^\ID)
        \tag{Using \ref{b5}}
      \end{alignat*}
      We know that $\incaccept_\tau^\ID \ra \accept_\tau^\ID$. Moreover, using \ref{acc3} and \ref{acc4}, we can check that:
      \[
        \accept_\tau^\ID
        \;\ra\;
        \bigvee_{\taut = \_,\cuai_\ID(\_,1)\atop{\tautt \popre \taut \popre \tau}}
        \instate_\taut(\sqn_\ue^\ID) = \instate_\tautt(\sqn_\hn^\ID)
      \]
      Moreover, $\accept_\tau^\ID \ra \instate_\taut(\sqn_\ue^\ID) < \cstate_\taut(\sqn_\ue^\ID)$, and using \ref{b5} we know that $\cstate_\taut(\sqn_\ue^\ID) \le \cstate_\tau(\sqn_\ue^\ID)$. Finally, $\incaccept_\tau^\ID \ra \cstate_\tau(\sqn_\hn^\ID) = \instate_\tau(\sqn_\hn^\ID) + 1$. Putting everything together:
      \[
        \incaccept_\tau^\ID \ra
        \cstate_\tau(\sqn_\hn^\ID) \le \cstate_\tau(\sqn_\ue^\ID)
      \]
      Which is what we wanted. We summarize graphically this proof below:
      \begin{center}
        \begin{tikzpicture}
          [dn/.style={inner sep=0.2em,fill=black,shape=circle},
          sdn/.style={inner sep=0.15em,fill=white,draw,solid,shape=circle},
          sl/.style={decorate,decoration={snake,amplitude=1.6}},
          dl/.style={dashed},
          pin distance=0.5em,
          every pin edge/.style={thin}]

          \draw[thick] (0,0)
          node[left=1.3em] {$\tau:$}
          -- ++(0.5,0)
          node[dn,pin={above:{$\cnai(j,0)$}}]
          (a) {}
          node[below,yshift=-0.3em] {$\tautt$}
          -- ++(3,0)
          node[dn,pin={above:{$\cuai_\ID(\_,1)$}}]
          (b) {}
          node[below,yshift=-0.3em] {$\taut$}
          -- ++(3,0)
          node[dn]
          (bo) {}
          -- ++(3,0)
          node[dn,pin={above:{$\cnai(j,1)$}}]
          (c) {}
          node[below,yshift=-0.3em] {$\tau$}
          -- ++(3,0)
          node[dn]
          (cp) {}
          -- ++(0.5,0);

          \path (a) -- ++ (0,-1)
          node (a2) {$\instate_\tautt(\sqn_\hn^\ID)$};

          \path (b) -- ++ (0,-1)
          -- ++ (0,-1.3)
          node (b2) {$\instate_\taut(\sqn_\ue^\ID)$};

          \path (bp) -- ++ (0,-1)
          -- ++ (0,-1.3)
          node (bp2) {$\cstate_\taut(\sqn_\ue^\ID)$};

          \path (c) -- ++ (0,-1)
          node (c1) {$\instate_\tau(\sqn_\hn^\ID)$};

          \path (cp) -- ++ (0,-1)
          node (cp1) {$\cstate_\tau(\sqn_\hn^\ID)$}
          -- ++ (0,-1.3)
          node (cp2) {$\cstate_\tau(\sqn_\ue^\ID)$};

          \draw (a2) to node[midway,above]{$=$} (c1) 
          (c1) -- (cp1) node[midway,above]{$+1$}
          (a2) -- (b2) node[midway,below,sloped]{$=$}
          (b2) -- (bp2) node[midway,below]{$+1$}
          (bp2) -- (cp2) node[midway,below]{$\le$};
        \end{tikzpicture}
      \end{center}
      If $\incaccept_\tau^\ID$ is false then $\neg \incaccept_\tau^\ID \ra \cstate_\tau(\sqn_\hn^\ID) = \instate_\tau(\sqn_\hn^\ID)$, and we conclude by applying the induction hypothesis.
    \end{itemize}

  \item \ref{b6}. First, observe that:
    \[
      \cstate_\taut(\sync_\ue^\ID) \ra
      \bigvee_{\taun = \_,\npuai{2}{\ID}{j}
        \atop{\tauo \popre \taun \popre \taut}}
      \accept_\taun^\ID
      \numberthis\label{eq:fasdadasdsalfhafjsapf}
    \]
    Let $\taun = \_,\npuai{2}{\ID}{j}$ such that $\tauo \popre \taun \popre \taut$. Let $\taui = \_,\npuai{1}{\ID}{j}$ such that $\taui \popre \taun$. We know that $\taui \popre \tauo$. We apply \ref{acc2}:
    \[
      \accept_\taun^\ID \;\ra\;
      \bigvee_{\taux = \_, \pnai(j_x,1)
        \atop{\taui \popre \taux \popre \taun}}
      \accept_\taux^\ID \;\wedge\;
      g(\inframe_\taui) = \nonce^{j_x} \;\wedge\;
      \pi_1(g(\inframe_\taux)) =
      \enc{
        \spair{\ID}
        {\instate_\taui(\sqn_\ue^\ID)}}
      {\pk_\hn}{\enonce^j}
      \numberthis\label{eq:fslfjidfasjfiapfj}
    \]
    Let $\taux = \_, \pnai(j_x,1)$ such that $\taui \popre \taux \popre \taun$. Using \ref{b5}, we get that $\cstate_\tauo(\sqn_\ue^\ID) \le \instate_\taui(\sqn_\ue^\ID)$ and that  $\cstate_\taux(\sqn_\hn^\ID) \le \cstate_\taut(\sqn_\hn^\ID)$. There are two cases, depending on whether we have $\incaccept_\taux^\ID$.
    \begin{itemize}
    \item We know that $\incaccept_\taux^\ID \ra \cstate_\taux(\sqn_\hn^\ID) = \instate_\taui(\sqn_\ue^\ID) + 1 > \instate_\taui(\sqn_\ue^\ID)$. Putting everything together, we get that:
      \[
        \accept_\taun^\ID \wedge
        \incaccept_\taux^\ID
        \ra
        \cstate_\tauo(\sqn_\ue^\ID) < \cstate_\taut(\sqn_\hn^\ID)
        \numberthis\label{eq:fdfhvifajsifpeuipasras}
      \]

    \item We know that:
      \[
        \accept_\taux^\ID \wedge \neg \incaccept_\taux^\ID
        \wedge
        \pi_1(g(\inframe_\taux)) =
        \enc{
          \spair{\ID}
          {\instate_\taui(\sqn_\ue^\ID)}}
        {\pk_\hn}{\enonce^j}
        \ra \instate_\taui(\sqn_\ue^\ID) < \instate_\taux(\sqn_\hn^\ID)
      \]
      Moreover, $\neg \incaccept_\taux^\ID \ra \instate_\taux(\sqn_\hn^\ID) = \cstate_\taux(\sqn_\hn^\ID)$. We recall that  $\cstate_\tauo(\sqn_\ue^\ID) \le \instate_\taui(\sqn_\ue^\ID)$ and that  $\cstate_\taux(\sqn_\hn^\ID) \le \cstate_\taut(\sqn_\hn^\ID)$. Therefore:
      \[
        \accept_\taux^\ID \wedge \neg \incaccept_\taux^\ID
        \wedge
        \pi_1(g(\inframe_\taux)) =
        \enc{
          \spair{\ID}
          {\instate_\taui(\sqn_\ue^\ID)}}
        {\pk_\hn}{\enonce^j}
        \ra \cstate_\tauo(\sqn_\ue^\ID) < \cstate_\taut(\sqn_\hn^\ID)
        \numberthis\label{eq:fdfhvifajsifpeuipasras0}
      \]
    \end{itemize}
    Using \eqref{eq:fslfjidfasjfiapfj}, \eqref{eq:fdfhvifajsifpeuipasras} and \eqref{eq:fdfhvifajsifpeuipasras0} we get that $\accept_\taun^\ID \ra \cstate_\tauo(\sqn_\ue^\ID) < \cstate_\taut(\sqn_\hn^\ID)$. We summarize this graphically below:
    \begin{center}
      \begin{tikzpicture}
        [dn/.style={inner sep=0.2em,fill=black,shape=circle},
        sdn/.style={inner sep=0.15em,fill=white,draw,solid,shape=circle},
        sl/.style={decorate,decoration={snake,amplitude=1.6}},
        dl/.style={dashed},
        pin distance=0.5em,
        every pin edge/.style={thin}]

        \draw[thick] (0,0)
        node[left=1.3em] {$\tau:$}
        -- ++(0.5,0)
        node[dn,pin={above,align=left:{$\ns_\ID(\_)$\\ or $\epsilon$}}]
        (a) {}
        node[below,yshift=-0.3em,name=a0] {$\tauo$}
        -- ++(3,0)
        node[dn,pin={above:{$\npuai{1}{\ID}{j}$}}]
        (b) {}
        node[below,yshift=-0.3em,name=b0] {$\taui$}
        -- ++(3,0)
        node[dn,pin={above:{$\pnai(j_x,1)$}}]
        (c) {}
        node[below,yshift=-0.3em,name=c0] {$\taux$}
        -- ++(3,0)
        node[dn,pin={above:{$\npuai{2}{\ID}{j}$}}]
        (d) {}
        node[below,yshift=-0.3em,name=d0] {$\taun$}
        -- ++(3,0)
        node[dn]
        (f) {}
        node[below,yshift=-0.3em,name=f0] {$\taut$}
        -- ++(0.5,0);

        \draw[thin,dashed] (b0) -- ++(0,-0.5) -| (c0)
        {[draw=none] -- ++(0,-0.5)} -| (d0);
        
        \path (a) -- ++ (0,-1.3)
        -- ++ (0,-1.3)
        node (a2) {$\cstate_\tauo(\sqn_\ue^\ID)$};

        \path (b) -- ++ (0,-1.3)
        -- ++ (0,-1.3)
        node (b2) {$\instate_\taui(\sqn_\ue^\ID)$};

        \path (c) -- ++ (0,-1.3)
        node (c2) {$\cstate_\taux(\sqn_\hn^\ID)$};

        \path (f) -- ++ (0,-1.3)
        node (f2) {$\cstate_\taut(\sqn_\hn^\ID)$};

        \draw (a2) -- (b2) node[midway,above]{$\le$}
        (b2) -- (c2) node[midway,above,sloped]{$<$}
        (c2) -- (f2) node[midway,below]{$\le$};
      \end{tikzpicture}
    \end{center}

    Since this is true for all $\taun = \_,\npuai{2}{\ID}{j}$ such that $\tauo \popre \taun \popre \taut$, we deduce from \eqref{eq:fasdadasdsalfhafjsapf} that
    \[
      \cstate_\taut(\sync_\ue^\ID) \ra
      \cstate_\tauo(\sqn_\ue^\ID) < \cstate_\taut(\sqn_\hn^\ID)
    \]
    Which concludes this proof.
    
  \item \ref{b7}. Let $\tau_\ns = \epsilon$ or $\ns_\ID(\_)$ be such that $\tau_\ns \popreleq \tau$ and $\tau_\ns \not \potau \ns_\ID(\_)$. We show by induction over $\tau'$ with $\tau_\ns \popreleq \tau' \popreleq \tau$ that $\cstate_{\tau'}(\success^{\ID}_\ue) \equiv \false$.

    For $\tau' = \tau_\ns$, this is true using from Definition~\ref{def:init-sigma-phi} if if $\tau_\ns = \epsilon$, and from the protocol term definitions if $\tau_\ns = \ns_\ID(\_)$. The inductive case is straightforward.
    \qedhere
  \end{itemize}
  
\end{proof}

We can now state the following acceptance characterization properties.
\begin{lemma}
  \label{lem:equiv-accept}
  For every valid symbolic trace $\tau = \_,\ai$ and identity $\ID$ we have:
  \begin{itemize}
  \item \customlabel{equ2}{\lpequ{2}} If $\ai = \npuai{2}{\ID}{j}$. Let $\tautt = \_\npuai{1}{\ID}{j}$ such that $\tautt \popre \tau$. Also let:
    \begin{gather*}
      \supitr{\tautt,\tau}{\taut} \;\equiv\;
      \left(
        \begin{alignedat}{2}
          &&&
          g(\inframe_\tau) =
          \mac{\spair
            {\nonce^{j_1}}
            {\sqnsuc(\instate_\tautt(\sqn_\ue^\ID))}}
          {\mkey^\ID}{2}
          \;\wedge\;
          g(\inframe_\tautt) = \nonce^{j_1} \\
          &\wedge\;&&
          \pi_1(g(\inframe_\taut)) =
          \enc{\spair{\ID}{\instate_\tautt(\sqn_\ue^\ID)}}{\pk_\hn}{\enonce^j}
        \end{alignedat}
      \right)
    \end{gather*}
    Then:
    \[
      \accept_\tau^\ID
      \;\leftrightarrow\;
      \bigvee_{\taut = \_,\pnai(j_1,1)\atop{\tautt \potau \taut}}
      \supitr{\tautt,\tau}{\taut}
    \]

  \item \customlabel{equ3}{\lpequ{3}} If $\ai = \pnai(j,1)$. Then:
    \begin{alignat*}{2}
      \accept_\tau^\ID
      &\;\lra\;\;&&
      \bigvee_{
        \taut = \_, \npuai{1}{\ID}{j_1}
        \atop{\taut \popre \tau}}
      \left(
        \begin{alignedat}{2}
          &&&g(\inframe_{\taut}) = \nonce^j
          \wedge
          \pi_1(g(\inframe_\tau)) =
          \enc{\spair
            {\ID}
            {\instate_{\taut}(\sqn_\ue^\ID)}}
          {\pk_\hn}{\enonce^{j_1}}\\
          &\wedge\;&&
          \pi_2(g(\inframe_\tau)) =
          {\mac{\spair
              {\enc{\spair
                  {\ID}
                  {\instate_{\taut}(\sqn_\ue^\ID)}}
                {\pk_\hn}{\enonce^{j_1}}}
              {g(\inframe_{\taut})}}
            {\mkey^\ID}{1}}
        \end{alignedat}
      \right)\\
      &\;\lra\;\;&&
      \bigvee_{
        \taut = \_, \npuai{1}{\ID}{j_1}
        \atop{\taut \popre \tau}}
      \left(
        g(\inframe_{\taut}) = \nonce^j
        \wedge
        g(\inframe_\tau) = t_\taut
      \right)
    \end{alignat*}

  \item \customlabel{equ4}{\lpequ{4}} If $\ai = \cuai_\ID(j,1)$. For every $\taut = \_, \cnai(j_0,0)$ such that $\taut \popre \tau$, we let:
    \begin{alignat*}{2}
      \ctr{\tau}{\taut}
      &\equiv\;\;&&
      \left(
        \begin{alignedat}{1}
          &\pi_3(g(\inframe_\tau)) =
          \mac{\striplet
            {\nonce^{j_0}}
            {\instate_\taut(\sqn_\hn^\ID)}
            {\instate_\tau(\suci_\ue^\ID)}}
          {\mkey}{3} \wedge
          \instate_\tau(\uetsuccess^{\ID}) \\
          &\wedge\range{\instate_\tau(\sqn_\ue^\ID)}
          {\instate_\taut(\sqn_\hn^\ID)}\wedge
          g(\inframe_\taut) = \instate_\taut(\suci_\hn^{\ID}) \wedge
          \pi_1(g(\inframe_\tau)) = \nonce^{j_0}\\
          & \wedge
          \pi_2(g(\inframe_\tau)) =
          \instate_\taut(\sqn_\hn^\ID) \oplus \ow{\nonce^{j_0}}{\key}
          \wedge
          \instate_\tau(\suci_\ue^\ID) = \instate_\taut(\suci_\hn^\ID)
        \end{alignedat}
      \right)
    \end{alignat*}
    Then:
    \begin{mathpar}
      \left(
        \ctr{\tau}{\taut}
        \ra
        \accept_\taut^\ID
      \right)
      _{\taut = \_, \cnai(j_0,0)
        \atop{\taut \popre \tau}}

      \accept_\tau^\ID
      \;\lra\;
      \bigvee_{\taut = \_, \cnai(j_0,0)
        \atop{\taut \popre \tau}}
      \ctr{\tau}{\taut}
    \end{mathpar}

  \item \customlabel{equ5}{\lpequ{5}} If $\ai = \cnai(j,1)$. Let $\taut = \_,\cnai(j,0)$ such that $\taut \popre \tau$, and let $\ID \in \iddom$. Then:
    \begin{equation*}
      \accept_\tau^\ID
      \;\leftrightarrow\;\;
      \bigvee
      _{\taui = \_,\cuai_\ID(j_i,1)
        \atop{\taut \potau \taui}}
      \ctr{\taui}{\taut} \wedge
      g(\inframe_\tau) = \mac{\nonce^j}{\mkey^\ID}{4}
    \end{equation*}
  \end{itemize}
\end{lemma}

\subsection{Proof of Lemma~\ref{lem:equiv-accept}}

\subsubsection*{Proof of $\ref{equ2}$}
Using \ref{acc2} we know that:
\begin{alignat*}{2}
  \accept_\tau^\ID
  &\;\leftrightarrow\;&&
  \bigvee_{\taut = \_,\pnai(j_1,1)
    \atop{\tautt \potau \taut}}
  \accept_\tau^\ID
  \;\wedge\;
  g(\inframe_\tautt) = \nonce^{j_1}
  \;\wedge\;
  \pi_1(g(\inframe_\taut)) =
  \enc{\spair{\ID}{\instate_\tautt(\sqn_\ue^\ID)}}{\pk_\hn}{\enonce^j}
  \\
  &\;\leftrightarrow\;&&
  \bigvee_{\taut = \_,\pnai(j_1,1)
    \atop{\tautt \potau \taut}}
  \left(
    \begin{alignedat}{2}
      &&&
      g(\inframe_\tau) =
      \mac{\spair{\nonce^{j_1}}{\instate_\tau(\sqn_\ue^\ID)}}{\mkey^\ID}{2}
      \;\wedge\;
      g(\inframe_\tautt) = \nonce^{j_1}
      \\&\wedge\;&&
      \pi_1(g(\inframe_\taut)) =
      \enc{\spair{\ID}{\instate_\tautt(\sqn_\ue^\ID)}}{\pk_\hn}{\enonce^j}
    \end{alignedat}
  \right)
  \\
  \intertext{Since $\instate_\tau(\sqn_\ue^\ID)\;\equiv\; \sqnsuc(\instate_\tautt(\sqn_\ue^\ID))$:}
  &\;\leftrightarrow\;&&
  \bigvee_{\taut = \_,\pnai(j_1,1)
    \atop{\tautt \potau \taut}}
  \left(
    \begin{alignedat}{2}
      &&&
      g(\inframe_\tau) =
      \mac{\spair
        {\nonce^{j_1}}
        {\sqnsuc(\instate_\tautt(\sqn_\ue^\ID))}}
      {\mkey^\ID}{2}
      \;\wedge\;
      g(\inframe_\tautt) = \nonce^{j_1} \\
      &\wedge\;&&
      \pi_1(g(\inframe_\taut)) =
      \enc{\spair{\ID}{\instate_\tautt(\sqn_\ue^\ID)}}{\pk_\hn}{\enonce^j}
    \end{alignedat}
  \right)\\
  &\;\leftrightarrow\;&&
  \bigvee_{\taut = \_,\pnai(j_1,1)
    \atop{\tautt \potau \taut}}
  \supitr{\tautt,\tau}{\taut}
\end{alignat*}

\subsubsection*{Proof of $\ref{equ3}$}
Using \ref{acc1} it is easy to check that:
\begin{alignat*}{2}
  \accept_\tau^\ID
  &\;\leftrightarrow\;\;&&
  \bigvee_{
    \taut = \_,\npuai{1}{\ID}{j_1} \popre \tau}
  \left(
    \begin{alignedat}{2}
      &&&\dotuline{g(\inframe_{\taut}) = \nonce^j}
      \wedge
      \uwave{\pi_1(g(\inframe_\tau)) =
        \enc{\spair
          {\ID}
          {\instate_{\taut}(\sqn_\ue^\ID)}}
        {\pk_\hn}{\enonce^{j_1}}}\\
      &\wedge\;&&
      \pi_2(g(\inframe_\tau)) =
      {\mac{\spair
          {\uwave{\pi_1(g(\inframe_\tau))}}
          {\dotuline{\;\nonce^j}}}
        {\mkey^\ID}{1}}
    \end{alignedat}
  \right)\\
  \intertext{Which can be rewritten as follows (we identify above, using waves and dots, which equalities are used, and which terms are rewritten):}
  &\;\leftrightarrow\;\;&&
  \bigvee_{
    \taut = \_,\npuai{1}{\ID}{j_1} \popre \tau}
  \left(
    \begin{alignedat}{2}
      &&&g(\inframe_{\taut}) = \nonce^j
      \wedge
      \pi_1(g(\inframe_\tau)) =
      \enc{\spair
        {\ID}
        {\instate_{\taut}(\sqn_\ue^\ID)}}
      {\pk_\hn}{\enonce^{j_1}}\\
      &\wedge\;&&
      \pi_2(g(\inframe_\tau)) =
      {\mac{\spair
          {\enc{\spair
              {\ID}
              {\instate_{\taut}(\sqn_\ue^\ID)}}
            {\pk_\hn}{\enonce^{j_1}}}
          {g(\inframe_{\taut})}}
        {\mkey^\ID}{1}}
    \end{alignedat}
  \right)
\end{alignat*}
First, observe that:
\begin{mathpar}
  \enc{\spair
    {\ID}
    {\instate_{\taut}(\sqn_\ue^\ID)}}
  {\pk_\hn}{\enonce^{j_1}} = \pi_1(t_\taut)

  {\mac{\spair
      {\enc{\spair
          {\ID}
          {\instate_{\taut}(\sqn_\ue^\ID)}}
        {\pk_\hn}{\enonce^{j_1}}}
      {g(\inframe_{\taut})}}
    {\mkey^\ID}{1}} = = \pi_2(t_\taut)
\end{mathpar}
We conclude easily using the injectivity of the pair.

\subsubsection*{Proof of \ref{equ4}}
\noindent Using \ref{acc3} we know that:
\begin{alignat*}{2}
  \accept_\tau^\ID
  &\lra\;\;&&
  \bigvee_{\taut = \_, \cnai(j_0,0)
    \atop{\taut \popre \tau}}
  \left(
    \begin{alignedat}{1}
      &\accept_\tau^\ID \wedge
      \accept_\taut^\ID \wedge
      \pi_1(g(\inframe_\tau)) = \nonce^{j_0} \\
      &\wedge
      \pi_2(g(\inframe_\tau)) =
      \instate_\taut(\sqn_\hn^\ID) \oplus \ow{\nonce^{j_0}}{\key}
      \wedge
      \instate_\tau(\suci_\ue^\ID) = \instate_\taut(\suci_\hn^\ID)
    \end{alignedat}
  \right)\\
  \intertext{Inlining the definition of $\accept_\taut^\ID$:}
  &\lra\;\;&&
  \bigvee_{\taut = \_, \cnai(j_0,0)
    \atop{\taut \popre \tau}}
  \left(
    \begin{alignedat}{1}
      &\accept_\tau^\ID \wedge
      g(\inframe_\taut) = \instate_\taut(\suci_\hn^{\ID})
      \wedge
      \instate_\taut(\suci_\hn^{\ID}) \ne \unset \wedge
      \pi_1(g(\inframe_\tau)) = \nonce^{j_0} \\
      &\wedge
      \pi_2(g(\inframe_\tau)) =
      \instate_\taut(\sqn_\hn^\ID) \oplus \ow{\nonce^{j_0}}{\key}
      \wedge
      \instate_\tau(\suci_\ue^\ID) = \instate_\taut(\suci_\hn^\ID)
    \end{alignedat}
  \right)\displaybreak[0]\\
  \intertext{Inlining the definition of $\accept_\tau^\ID$:}
  &\lra\;\;&&
  \bigvee_{\taut = \_, \cnai(j_0,0)
    \atop{\taut \popre \tau}}
  \left(
    \begin{alignedat}{1}
      &\pi_3(g(\inframe_\tau)) =
      \mac{\striplet
        {\uline{\pi_1(g(\inframe_\tau))}}
        {\pi_2(g(\inframe_\tau)) \xor \ow{\uline{\pi_1(g(\inframe_\tau))}}{\key}}
        {\instate_\tau(\suci_\ue^\ID)}}
      {\mkey}{3}\\
      &\wedge
      \instate_\tau(\uetsuccess^{\ID}) \wedge
      \range{\instate_\tau(\sqn_\ue^\ID)}
      {\pi_2(g(\inframe_\tau)) \xor \ow{\uline{\pi_1(g(\inframe_\tau))}}{\key}}\\
      &g(\inframe_\taut) = \instate_\taut(\suci_\hn^{\ID})
      \wedge
      \instate_\taut(\suci_\hn^{\ID}) \ne \unset \wedge
      \uline{\pi_1(g(\inframe_\tau)) = \nonce^{j_0}} \\
      &\wedge
      \pi_2(g(\inframe_\tau)) =
      \instate_\taut(\sqn_\hn^\ID) \oplus \ow{\nonce^{j_0}}{\key}
      \wedge
      \instate_\tau(\suci_\ue^\ID) = \instate_\taut(\suci_\hn^\ID)
    \end{alignedat}
  \right)\displaybreak[0]\\
  \intertext{We rewrite $\pi_1(g(\inframe_\tau))$ into $\nonce^{j_0}$:}
  &\lra\;\;&&
  \bigvee_{\taut = \_, \cnai(j_0,0)
    \atop{\taut \popre \tau}}
  \left(
    \begin{alignedat}{1}
      &\pi_3(g(\inframe_\tau)) =
      \mac{\striplet
        {\uline{\nonce^{j_0}}}
        {\dashuline{\pi_2(g(\inframe_\tau)) \xor \ow{\uline{\nonce^{j_0}}}{\key}}}
        {\instate_\tau(\suci_\ue^\ID)}}
      {\mkey}{3}\\
      &\wedge
      \instate_\tau(\uetsuccess^{\ID}) \wedge
      \range{\instate_\tau(\sqn_\ue^\ID)}
      {\dashuline{\pi_2(g(\inframe_\tau)) \xor \ow{\uline{\nonce^{j_0}}}{\key}}}\\
      &g(\inframe_\taut) = \instate_\taut(\suci_\hn^{\ID})
      \wedge
      \instate_\taut(\suci_\hn^{\ID}) \ne \unset \wedge
      \uline{\pi_1(g(\inframe_\tau)) = \nonce^{j_0}} \\
      &\wedge
      \dashuline{\pi_2(g(\inframe_\tau)) =
        \instate_\taut(\sqn_\hn^\ID) \oplus \ow{\nonce^{j_0}}{\key}}
      \wedge
      \instate_\tau(\suci_\ue^\ID) = \instate_\taut(\suci_\hn^\ID)
    \end{alignedat}
  \right)\displaybreak[0]\\
  \intertext{We rewrite $\pi_2(g(\inframe_\tau)) \xor \ow{\nonce^{j_0}}{\key}$ into $\instate_\taut(\sqn_\hn^\ID)$:}
  &\lra\;\;&&
  \bigvee_{\taut = \_, \cnai(j_0,0)
    \atop{\taut \popre \tau}}
  \left(
    \begin{alignedat}{1}
      &\pi_3(g(\inframe_\tau)) =
      \mac{\striplet
        {\nonce^{j_0}}
        {\dashuline{\instate_\taut(\sqn_\hn^\ID)}}
        {\instate_\tau(\suci_\ue^\ID)}}
      {\mkey}{3}\\
      &\wedge
      \instate_\tau(\uetsuccess^{\ID}) \wedge
      \range{\instate_\tau(\sqn_\ue^\ID)}
      {\dashuline{\instate_\taut(\sqn_\hn^\ID)}}\\
      &\wedge g(\inframe_\taut) = \instate_\taut(\suci_\hn^{\ID})
      \wedge
      \instate_\taut(\suci_\hn^{\ID}) \ne \unset \wedge
      \pi_1(g(\inframe_\tau)) = \nonce^{j_0} \\
      &\wedge
      \dashuline{\pi_2(g(\inframe_\tau)) =
        \instate_\taut(\sqn_\hn^\ID) \oplus \ow{\nonce^{j_0}}{\key}}
      \wedge
      \instate_\tau(\suci_\ue^\ID) = \instate_\taut(\suci_\hn^\ID)
    \end{alignedat}
  \right)
  \numberthis\label{eq:dsjghwioeterierqwpo}
\end{alignat*}
Let $\tautt = \_,\cuai_\ID(j_0,0) \popre \tau$. By validity of $\tau$, there are no user $\ID$ actions between $\tautt$ and $\tau$, and therefore it is easy to check that $\instate_\tau(\uetsuccess^\ID) \ra \instate_\tautt(\success_\ue^\ID)$, and that $\instate_\tau(\suci_\ue^\ID) = \instate_\tautt(\suci_\ue^\ID)$. Moreover, using  \ref{b1} we know that $\instate_\tautt(\success_\ue^\ID) \ra \instate_\tautt(\suci_\ue^\ID) \ne \unset$. Therefore $\instate_\tau(\uetsuccess^\ID) \ra \instate_\tau(\suci_\ue^\ID) \ne \unset$. It follows that:
\[
  \left(
    \instate_\tau(\suci_\ue^\ID) = \instate_\taut(\suci_\hn^\ID) \wedge
    \instate_\tau(\uetsuccess^\ID)
  \right)
  \;\ra\;
  \instate_\taut(\suci_\hn^{\ID}) \ne \unset
\]
Hence we can simplify \eqref{eq:dsjghwioeterierqwpo} by removing $\instate_\taut(\suci_\hn^{\ID}) \ne \unset$. This yields:
\begin{alignat*}{2}
  \accept_\tau^\ID
  &\lra\;\;&&
  \bigvee_{\taut = \_, \cnai(j_0,0)
    \atop{\taut \popre \tau}}
  \left(
    \begin{alignedat}{1}
      &\pi_3(g(\inframe_\tau)) =
      \mac{\striplet
        {\nonce^{j_0}}
        {\instate_\taut(\sqn_\hn^\ID)}
        {\instate_\tau(\suci_\ue^\ID)}}
      {\mkey}{3} \wedge
      \instate_\tau(\uetsuccess^{\ID}) \\
      &\wedge\range{\instate_\tau(\sqn_\ue^\ID)}
      {\instate_\taut(\sqn_\hn^\ID)}\wedge
      g(\inframe_\taut) = \instate_\taut(\suci_\hn^{\ID})\wedge
      \pi_1(g(\inframe_\tau)) = \nonce^{j_0} \\
      &\wedge
      \pi_2(g(\inframe_\tau)) =
      \instate_\taut(\sqn_\hn^\ID) \oplus \ow{\nonce^{j_0}}{\key}
      \wedge
      \instate_\tau(\suci_\ue^\ID) = \instate_\taut(\suci_\hn^\ID)
    \end{alignedat}
  \right)\\
  &\lra\;\;&&
  \bigvee_{\taut = \_, \cnai(j_0,0)
    \atop{\taut \popre \tau}}
  \ctr{\tau}{\taut}
\end{alignat*}
Finally, it is easy to check that for every $\taut = \_, \cnai(j_0,0)$ such that $\taut \popre \tau$, we have $\ctr{\tau}{\taut}\ra \accept_\taut^\ID$.

\subsubsection*{Proof of \ref{equ5}}
Using \ref{acc4} we know that:
\[
  \accept_\tau^\ID
  \;\lra\;
  \bigvee_{\taui = \_,\cuai_\ID(j_i,1) \popre \tau}
  \accept_\tau^\ID \wedge
  \accept_\taui \wedge
  \pi_1(g(\inframe_\taui)) = \nonce^j
\]
Moreover, using \ref{equ4} we know that:
\begin{alignat*}{2}
  \accept_\tau^\ID
  &\lra\;\;&&
  \bigvee_{\taui = \_,\cuai_\ID(j_i,1) \popre \tau
    \atop{\tautt = \_, \cnai(j_2,0) \popre \taui}}
  \accept_\tau^\ID \wedge
  \ctr{\taui}{\tautt} \wedge
  \pi_1(g(\inframe_\taui)) = \nonce^j
\end{alignat*}
Let $\tautt = \_, \cnai(j_2,0) \popre \taui$. Then we know that $\ctr{\taui}{\tautt} \ra \pi_1(g(\inframe_\taui)) = \nonce^{j_2}$. Therefore using $\ax{EQIndep}$ we know that if $j_2 \ne j$:
\[
  \left(
    \ctr{\taui}{\tautt} \wedge
    \pi_1(g(\inframe_\taui)) = \nonce^j
  \right)
  \ra
  \left(
    \pi_1(g(\inframe_\taui)) = \nonce^{j_2} \wedge
    \pi_1(g(\inframe_\taui)) = \nonce^j
  \right)
  \ra
  \false
\]
Hence:
\begin{alignat*}{2}
  \accept_\tau^\ID
  &\lra\;\;&&
  \bigvee_{\taui = \_,\cuai_\ID(j_i,1)
    \atop{\taut \potau \taui}}
  \accept_\tau^\ID \wedge
  \ctr{\taui}{\taut} \wedge
  \pi_1(g(\inframe_\taui)) = \nonce^j\\
  \intertext{Since $\ctr{\taui}{\taut} \ra \pi_1(g(\inframe_\taui)) = \nonce^{j}$:}
  &\lra\;\;&&
  \bigvee_{\taui = \_,\cuai_\ID(j_i,1)
    \atop{\taut \potau \taui}}
  \accept_\tau^\ID \wedge
  \ctr{\taui}{\taut} \displaybreak[0]\\
  \intertext{We inline the definition of $\accept_\tau^\ID$:}
  &\lra\;\;&&
  \bigvee_{\taui = \_,\cuai_\ID(j_i,1)
    \atop{\taut \potau \taui}}
  g(\inframe_\tau) = \mac{\nonce^j}{\mkey^{\ID}}{4} \wedge
  \instate_\tau(\bauth_\hn^j) = \ID \wedge
  \ctr{\taui}{\taut}
\end{alignat*}
Using \ref{equ4}, we know that for every $\taut = \_, \cnai(j_0,0)$ such that $\taut \popre \tau$, $\ctr{\tau}{\taut} \ra \accept_\taut^\ID$. Moreover, using \ref{a6} we know that $\accept_\taut^\ID \ra \instate_\taut(\bauth_\hn^j) = \ID$. Besides, $\instate_\taut(\bauth_\hn^j) = \ID \ra \instate_\tau(\bauth_\hn^j) = \ID$. Hence $\ctr{\tau}{\taut} \ra \instate_\tau(\bauth_\hn^j) = \ID$. By consequence:
\begin{alignat*}{2}
  \accept_\tau^\ID
  &\lra\;\;&&
  \bigvee_{\taui = \_,\cuai_\ID(j_i,1)
    \atop{\taut \potau \taui}}
  g(\inframe_\tau) = \mac{\nonce^j}{\mkey^{\ID}}{4} \wedge
  \ctr{\taui}{\taut}
\end{alignat*}

\subsection{$\suci_\ue^\ID$ Concealment}

\begin{lemma}
  \label{lem:suci-conceal}
  Let $\tau$ be a valid symbolic trace and $\IDx \in \iddom$. For every $\taua = \_,\cnai(j_a,1)$ or $\taua = \_,\pnai(j_a,1)$ such that $\taua \popreleq \tau$, and for every $\taub = \npuai{1}{\IDx}{j_i}$ or $\taub = \cuai_\IDx(j_i,1)$ such that $\taub \popre \taua$, if:
  \[
    \left\{
      \taut \mid \taub \potau \taut
    \right\}
    \cap
    \left\{
      \npuai{\_}{\IDx}{j},
      \cuai_\IDx(j,\_),
      \fuai_\IDx(j)
      \mid j \in \mathbb{N}
    \right\}
    \subseteq
    \left\{
      \npuai{2}{\IDx}{j_i},
      \fuai_\IDx(j_i)
    \right\}
  \]
  Then there exists a derivation of:
  \begin{alignat*}{2}
    &&&
    \incaccept_\taua^\IDx \wedge
    \cstate_\taub(\bauth_\ue^\IDx) = \nonce^{j_a} \wedge
    \accept_\taub^\IDx
    \ra
    g(\inframe_\tau) \ne \suci^{j_a}
  \end{alignat*}
\end{lemma}

\subsubsection*{Proof of Lemma~\ref{lem:suci-conceal}}
Let $\inleak_\tau$ be the vector of terms containing:
\begin{itemize}
\item $\inleak_\tauo$ if $\tau = \tauo,\ai_0$ and $\tau \popre \taua$.
\item The term $\beta$.
\item All the keys except $\key^\IDx$, $\mkey^\IDx$ and the asymmetric secret key $\sk_\hn$.
\item All elements of $\instate_\tau$ (in an arbitrary order) except:
  \begin{itemize}
  \item All the user $\IDx$ values, i.e. for every $\scx$ we have $\instate_\tau(\scx_\ue^\IDx) \not \in \inleak_\tau$.
  \item The network's $\suci$ value of user $\IDx$, i.e. $\instate_\tau(\suci_\hn^\IDx) \not \in \inleak_\tau$.
  \end{itemize}
\end{itemize}
Let:
\[
  \beta \equiv
  \incaccept_\taua^\IDx \wedge
  \cstate_\taub(\bauth_\ue^\IDx) = \nonce^{j_a} \wedge
  \accept_\taub^\IDx
\]
Let $\suci$ be a fresh name. We show by induction on $\taut$ in $\taua \popreleq \taut \popre \tau$ that there are derivations of:
\begin{gather*}
  \lrpcond{\beta}{
    \cframe_\taut,\cleak_\taut,
    \suci^{j_a}}
  \;\sim\;
  \lrpcond{\beta}{
    \cframe_\taut,\cleak_\taut,
    \suci}\\
  \beta \ra \cstate_\taut(\suci_\hn^\IDx) = \suci^{j_a}
  \numberthis\label{eq:ewigqhiqerjqweqwp}
\end{gather*}
We depict the situation below:
\begin{center}
  \begin{tikzpicture}
    [dn/.style={inner sep=0.2em,fill=black,shape=circle},
    sdn/.style={inner sep=0.15em,fill=white,draw,solid,shape=circle},
    sl/.style={decorate,decoration={snake,amplitude=1.6}},
    dl/.style={dashed},
    pin distance=0.5em,
    every pin edge/.style={thin}]

    \draw[thick] (0,0)
    node[left=1.3em] {$\tau:$}
    -- ++(0.5,0)
    node[dn,pin={above,align=left:{$\cuai_\IDx(j_i,1)$\\or $\npuai{1}{\IDx}{j_i}$}}]
    (b) {}
    node[below,yshift=-0.3em,name=d0] {$\taub$}
    -- ++(3,0)
    node[dn,pin={above,align=left:{$\cnai(j_a,1)$\\or $\pnai(j_a,1)$}}]
    (c) {}
    node[below,yshift=-0.3em,name=d1] {$\taua$}
    -- ++(3,0)
    node[dn]
    (e) {}
    node[below,yshift=-0.3em,name=d2] {$\taut$}
    -- ++(3,0)
    node[dn]
    (e) {}
    node[below,yshift=-0.3em,name=d2] {$\tau$};

    \draw[thin,dotted] (d0) -- ++(0,-0.5) -| (d1);
  \end{tikzpicture}
\end{center}

\paragraph{Case $\taut = \taua$}
First, $\beta \ra \incaccept_\taua^\IDx$, and $\incaccept_\taua^\IDx \ra \cstate_\taua(\suci_\hn^\IDx) = \suci^{j_a}$. Therefore we know that:
\[
  \beta \ra \cstate_\taua(\suci_\hn^\IDx) = \suci^{j_a}
\]
Then, we observe from the definition of $\cleak_\taua$ that $\suci^{j_a} \not \in \st(\cleak_\taua)$ (since $\cstate_\taua(\suci_\hn^\IDx)$ is \emph{not} in $\cleak_\taua$). Moreover $\suci^{j_a}$ does not appear in $\inframe_\taua$ and $t_\taua$. Besides, $\suci$ is a fresh name. By consequence we can apply the $\ax{Fresh}$ axiom, and then conclude using $\refl$:
\[
  \infer[\ax{Fresh}]{
    \lrpcond{\beta}{
      \inframe_\taut,
      \inleak_\taut,
      \suci^{j_a}}
    \;\sim\;
    \lrpcond{\beta}{
      \inframe_\taut,
      \inleak_\taut,
      \suci}
  }{
    \infer[\refl]{
      \lrpcond{\beta}{
        \inframe_\taut,
        \inleak_\taut}
      \;\sim\;
      \lrpcond{\beta}{
        \inframe_\taut,
        \inleak_\taut}
    }{}
  }
\]

\paragraph{Case $\taua \popre \taut$}
Let $\ai$ be such that $\taut = \_,\ai$. Assume by induction that:
\begin{gather}
  \lrpcond{\beta}{
    \inframe_\taut,\inleak_\taut,
    \suci^{j_a}}
  \;\sim\;
  \lrpcond{\beta}{
    \inframe_\taut,\inleak_\taut,
    \suci}
  \label{eq:fiujqpurqwrfvvs}\\
  \beta \ra \instate_\taut(\suci_\hn^\IDx) = \suci^{j_a}
  \label{eq:wfeiruqerqwp}
\end{gather}

\paragraph{Part 1}
First, we show that:
\[
  \beta \ra \cstate_\taut(\suci_\hn^\IDx) = \suci^{j_a}
\]
Since we know that \eqref{eq:ewigqhiqerjqweqwp} holds, we just need to look at the case $\ai$ that updates $\suci_\hn^\IDx$ to conclude:
\begin{itemize}
\item If $\ai = \cnai(j,0)$. Using \eqref{eq:fiujqpurqwrfvvs}, we know that $g(\inframe_\taut) \ne \suci^{j_a}$. Hence using \eqref{eq:wfeiruqerqwp} we know that:
  \[
    \beta \ra \instate_\taut(\suci_\hn^\IDx) \ne g(\inframe_\taut)
  \]
  Which shows that $\beta \ra \neg \accept_\taut^\IDx$. This concludes this case.

\item If $\ai = \pnai(j,1)$. Since $\taua = \cnai(j_1,1)$ or $\pnai(j_1,1)$, we know by validity of $\tau$ that $j_a \ne j$. Using \ref{equ3} we know that:
  \begin{alignat*}{2}
    \accept_\taut^\IDx
    &\;\ra\;\;&&
    \bigvee_{
      \taun = \_, \npuai{1}{\ID}{j_n}
      \atop{\taun \popre \taut}}
    g(\inframe_{\taun}) = \nonce^{j}
    \numberthis\label{qe:faitqiothqriqrup}
  \end{alignat*}
  Since $j_a \ne j$ we know that $\nonce^{j} \ne \nonce^{j_a}$. Moreover:
  \[
    \cstate_\taub(\bauth_\ue^\IDx) = \nonce^{j_a} \wedge
    \accept_\taub^\IDx \ra
    g(\inframe_{\taub}) = \nonce^{j_a}
  \]
  Hence $\beta \ra g(\inframe_{\taub}) \ne \nonce^{j}$. Moreover, for every $\tau'$ between $\taub$ and $\taut$, we know that $\tau' \ne \npuai{1}{\IDx}{\_}$. Therefore we know that:
  \begin{alignat*}{2}
    \beta \wedge \accept_\taut^\IDx
    &\;\ra\;\;&&
    \bigvee_{
      \taun = \_, \npuai{1}{\ID}{j_n}
      \atop{\taun \popre \taub}}
    g(\inframe_{\taun}) = \nonce^{j}
    \wedge
    \pi_1(g(\inframe_\taut)) =
    \enc{\spair
      {\IDx}
      {\instate_{\taun}(\sqn_\ue^\IDx)}}
    {\pk_\hn}{\enonce^{j_n}}
  \end{alignat*}
  Let $\taun = \_, \npuai{1}{\ID}{j_n}$ such that $\taun \popre \taub$. We know that:
  \[
    \beta \ra
    \cstate_\taua(\sqn_\hn^\IDx) =
    \cstate_\taub(\sqn_\ue^\IDx) =
    \sqnsuc(\instate_\taub(\sqn_\ue^\IDx))
  \]
  And that:
  \begin{mathpar}
    \cstate_\taua(\sqn_\hn^\IDx) \le \instate_\taut(\sqn_\hn^\IDx)

    \instate_\taun(\sqn_\ue^\IDx) \le \instate_\taub(\sqn_\ue^\IDx)
  \end{mathpar}
  Graphically:
  \begin{center}
    \begin{tikzpicture}
      [dn/.style={inner sep=0.2em,fill=black,shape=circle},
      sdn/.style={inner sep=0.15em,fill=white,draw,solid,shape=circle},
      sl/.style={decorate,decoration={snake,amplitude=1.6}},
      dl/.style={dashed},
      pin distance=0.5em,
      every pin edge/.style={thin}]

      \draw[thick] (0,0)
      node[left=1.3em] {$\tau:$}
      -- ++(0.5,0)
      node[dn,pin={above:{$\npuai{1}{\IDx}{j_n}$}}]
      (a) {}
      node[below,yshift=-0.3em,name=d] {$\taun$}
      -- ++(3,0)
      node[dn,pin={above,align=left:{$\cuai_\IDx(j_i,1)$\\or $\npuai{1}{\IDx}{j_i}$}}]
      (b) {}
      node[below,yshift=-0.3em,name=d0] {$\taub$}
      -- ++(3,0)
      node[dn]
      (b0) {}
      -- ++(3,0)
      node[dn,pin={above,align=left:{$\cnai(j_a,1)$\\or $\pnai(j_a,1)$}}]
      (c) {}
      node[below,yshift=-0.3em,name=d1] {$\taua$}
      -- ++(3,0)
      node[dn,pin={above:{$\pnai(j,1)$}}]
      (e) {}
      node[below,yshift=-0.3em,name=d2] {$\taut$}
      -- ++(0.5,0);

      \draw[thin,dashed] (d) -- ++(0,-1) -| (d2);

      \draw[thin,dotted] (d0) -- ++(0,-0.5) -| (d1);

      \path (a) -- ++ (0,-2)
      -- ++ (0,-1.3)
      node (a2) {$\instate_\taun(\sqn_\ue^\ID)$};

      \path (b) -- ++ (0,-2)
      -- ++ (0,-1.3)
      node (b2) {$\instate_\taub(\sqn_\ue^\ID)$};

      \path (b0) -- ++ (0,-2)
      -- ++ (0,-1.3)
      node (b02) {$\cstate_\taub(\sqn_\ue^\ID)$};

      \path (c) -- ++ (0,-2)
      node (c1) {$\cstate_\taua(\sqn_\hn^\ID)$};

      \path (e) -- ++ (0,-2)
      node (e1) {$\instate_\taut(\sqn_\hn^\ID)$};

      \draw (a2) -- (b2) node[midway,below]{$\le$}
      (b2) -- (b02) node[midway,below,sloped]{$+1$}
      (b02) -- (c1) node[midway,sloped,above]{$=$}
      (c1) -- (e1) node[midway,above]{$\le$};
    \end{tikzpicture}
  \end{center}
  We deduce that:
  \[
    \beta \wedge
    \accept_\taut^\IDx \wedge
    g(\inframe_{\taun}) = \nonce^{j}
    \ra
    \instate_\taut(\sqn_\hn^\IDx) >
    \instate_\taun(\sqn_\ue^\IDx)
  \]
  Moreover:
  \[
    \left(
      \beta \wedge
      \incaccept_\taut^\IDx \wedge
      g(\inframe_{\taun}) = \nonce^{j}\wedge
      \pi_1(g(\inframe_\taut)) =
      \enc{\spair
        {\IDx}
        {\instate_{\taun}(\sqn_\ue^\IDx)}}
      {\pk_\hn}{\enonce^{j_n}}
    \right)
    \ra
    \instate_\taut(\sqn_\hn^\IDx) \le
    \instate_\taun(\sqn_\ue^\IDx)
  \]
  Hence:
  \[
    \left(
      \beta \wedge
      \accept_\taut^\IDx \wedge
      g(\inframe_{\taun}) = \nonce^{j} \wedge
      \wedge
      \pi_1(g(\inframe_\taut)) =
      \enc{\spair
        {\IDx}
        {\instate_{\taun}(\sqn_\ue^\IDx)}}
      {\pk_\hn}{\enonce^{j_n}}
    \right)
    \ra
    \neg \incaccept_\taut^\IDx
  \]
  Using \eqref{qe:faitqiothqriqrup}, this shows that:
  \[
    \beta \wedge
    \accept_\taut^\IDx
    \ra
    \neg \incaccept_\taut^\IDx
  \]
  This concludes this proof.

\item If $\ai = \cnai(j,1)$. Since $\taua = \cnai(j_1,1)$ or $\pnai(j_1,1)$, we know by validity of $\tau$ that $j_a \ne j$. From the induction hypothesis we know that:
  \begin{alignat*}{2}
    \beta \ra \instate_\taut(\suci_\hn^\IDx) = \suci^{j_a}
  \end{alignat*}
  It is easy to check that:
  \begin{alignat*}{2}
    \instate_\taut(\suci_\hn^\IDx) = \suci^{j_a}
    \ra
    \instate_\taut(\tsuccess_\hn^\IDx) = \nonce^{j_a}
  \end{alignat*}
  Hence:
  \begin{alignat*}{2}
    \beta
    &\ra\;\;&&
    \instate_\taut(\tsuccess_\hn^\IDx) = \nonce^{j_a}\\
    &\ra\;\;&&
    \instate_\taut(\tsuccess_\hn^\IDx) \ne \nonce^{j}
    \tag{Since $j \ne j_a$}\\
    &\ra\;\;&&
    \neg \incaccept_\taut^\IDx\\
    &\ra\;\;&&
    \cstate_\taut(\suci_\hn^\IDx) =
    \instate_\taut(\suci_\hn^\IDx) =
    \suci^{j_a}
  \end{alignat*}
  Which concludes this case.
\end{itemize}

\paragraph{Part 2}
We now show that:
\begin{gather*}
  \lrpcond{\beta}{
    \cframe_\taut,\cleak_\taut,
    \suci^{j_a}}
  \;\sim\;
  \lrpcond{\beta_\taut}{
    \cframe_\taut,\cleak_\taut,
    \suci}
\end{gather*}

We do a case disjunction on $\ai$. We only details the case where $\ai$ is a symbolic action of user $\ID$, with $\ID \ne \IDx$, and the case where $\ai = \fnai(j_a)$. All the other cases are similar to these two cases, and their proof will only be sketched.
\begin{itemize}
\item If $\ai$ is a symbolic action of user $\ID$, with $\ID \ne \IDx$, then for every $u \in \cleak_\taut \backslash \inleak_\taut$ (resp. $u \equiv t_\taut$) we show that there exists a many-hole context $C_u$ such that $u \equiv C_u[\inframe_\taut,\inleak_\taut]$ and $C_u$ does not contain any $\nonce$.

  We only detail the case $\ai = \fuai_\ID(j)$. First, observe that:
  \begin{gather*}
    \accept_\taut^\ID \;\equiv\;
    \left(\begin{array}[c]{ll}
        &\eq
        {\pi_2(g(\uline{\inframe_\taut}))}
        {\mac
          {\spair
            {\pi_1(g(\uline{\inframe_\taut}))
              \xor \row{\uline{\instate_\taut(\eauth_\ue^{\ID})}}{\uline{\key}}}
            {\uline{\instate_\taut(\eauth_\ue^{\ID})}}}
          {\uline{\mkey}}{5}}\\
        \wedge &
        \neg \eq{\uline{\instate_\taut(\eauth_\ue^{\ID})}}{\fail}
        \;\wedge\;
        \neg \eq{\uline{\instate_\taut(\eauth_\ue^{\ID})}}{\bot}
      \end{array}\right)
  \end{gather*}
  All the underlined subterms are in $\inframe_\taut,\inleak_\taut$, therefore there exists a context $C_\accept$ such that $\accept_\taut^\ID \equiv C_\accept[\inframe_\taut,\inleak_\taut]$.  Remark that $\cleak_\taut \backslash \inleak_\taut = \{\instate_\taut(\success_\ue^{\ID}), \instate_\taut(\suci_\ue^{\ID})\}$. Moreover:
  \begin{mathpar}
    t_\taut \;\equiv\;
    \ite
    {\accept_\taut^\ID}
    {\textsf{ok}}
    { \textsf{error}}

    \instate_\taut(\success_\ue^{\ID}) \;\equiv\; \accept_\taut^\ID

    \instate_\taut(\suci_\ue^{\ID}) \;\equiv\;
    \ite
    {\accept_\taut^\ID}
    {\pi_1(g(\uline{\inframe_\taut}))
      \xor
      \row{\uline{\instate_\taut(\eauth_\ue^{\ID})}}{\uline{\key}}}
    { \unset}
  \end{mathpar}
  Using the fact that all the underlined subterms are in $\inframe_\taut,\inleak_\taut$, and using $C_\accept$ it is easy to build the wanted contexts.

  We then conclude using the $\fa$ rule under context, the $\dup$ rule and the induction hypothesis:
  \[
    \infer[R]{
      \lrpcond{\beta}{
        \cframe_\taut,\cleak_\taut,
        \suci^{j_a}}
      \;\sim\;
      \lrpcond{\beta}{
        \cframe_\taut,\cleak_\taut,
        \suci}
    }{
      \infer[(\fac+\dup)^*]{
        \begin{alignedat}{2}
          &&&\lrpcond{\beta}{
            \inframe_\taut,\inleak_\taut,
            \suci^{j_a},
            (C_u[\inframe_\taut,\inleak_\taut])
            _{u \in \{t_\taut,\cleak_\taut \backslash \inleak_\taut\}}}\\
          &\sim\;\;&&
          \lrpcond{\beta}{
            \inframe_\taut,\inleak_\taut,
            \suci,
            (C_u[\inframe_\taut,\inleak_\taut])
            _{u \in \{t_\taut,\cleak_\taut \backslash \inleak_\taut\}}}
        \end{alignedat}
      }{
        \lrpcond{\beta}{
          \inframe_\taut,\inleak_\taut,
          \suci^{j_a}}
        \;\sim\;
        \lrpcond{\beta}{
          \inframe_\taut,\inleak_\taut,
          \suci}
      }
    }
  \]

\item If $\ai = \fnai(j_a)$. It is easy to check that:
  \[
    \instate_\taua(\eauth_\hn^\IDx) \ne \IDx
    \ra
    \instate_\taua(\suci_\hn^\IDx) \ne \suci^{j_a}
    \ra
    \instate_\tau(\suci_\hn^\IDx) \ne \suci^{j_a}
  \]
  Therefore using the induction property \eqref{eq:wfeiruqerqwp} we deduce that $\beta \ra \instate_\taua(\eauth_\hn^\IDx) = \IDx$. Moreover by validity of $\tau$, there are no session $j_a$ network actions between $\taua$ and $\taut$. It follows that $\instate_\taua(\eauth_\hn^\IDx) = \IDx \ra \instate_\taut(\eauth_\hn^\IDx) = \IDx$. Hence:
  \[
    t_\taut =
    \pair
    {\suci^{j_a} \oplus \row{\nonce^{j_a}}{\key^{\IDx}}}
    {\mac{\spair
        {\suci^{j_a}}
        {\nonce^{j_a}}}
      {\mkey^{\IDx}}{5}}
  \]
  Observe that:
  \[
    \lrpcond{\beta}{
      \cframe_\taut,\cleak_\taut,
      \suci^{j_a}}
    =
    \lrpcond{\beta}{
      \inframe_\taut,
      \pair
      {\suci^{j_a} \oplus \row{\nonce^{j_a}}{\key^{\IDx}}}
      {\mac{\spair
          {\suci^{j_a}}
          {\nonce^{j_a}}}
        {\mkey^{\IDx}}{5}},
      \inleak_\taut,
      \suci^{j_a}}
  \]
  We are now going to apply the $\prff$ axiom on the left to replace $\suci^{j_a} \oplus \row{\nonce^{j_a}}{\key^{\IDx}}$ with $\suci^{j_a} \oplus \nonce_{\textsf{f}}$ where $\nonce_{\textsf{f}}$ is a fresh nonce. For every $\tautt = \_,\fuai_\ID(\_) \popre \taut$, we use \ref{equ1} to replace every occurrences of $\accept_\tautt$ in $\inframe_\taut,\inleak_\taut,\beta$ with:
  \[
    \gamma_\tautt \equiv
    \bigvee_{\tauttt = \_,\fnai(\_) \popre \tautt
      \atop{\tauttt \not \popre_{\tautt} \newsession_\ID(\_)}}
    \futr{\tautt}{\tauttt}
  \]
  which yields the terms $\inframep_\taut,\inleakp_\taut,\beta'$. We can check that:
  \[
    \setprf_{\key^{\IDx}}^{\rowsym}(\gamma_\tautt) \subseteq
    \left\{
      \nonce^p
      \mid \exists \tau' = \_,\fnai(p) \popre \taut
    \right\}
  \]
  And that:
  \[
    \setprf_{\key^{\IDx}}^{\rowsym}(\inframep_\taut,\inleakp_\taut) =
    \left\{
      \nonce^p
      \mid \exists \tau' = \_,\fnai(p) \popre \taut
    \right\}
  \]
  Therefore we can apply the \prff axiom as wanted: first we replace $\inframe_\taut,\inleak_\taut,\beta$ by $\inframep_\taut,\inleakp_\taut,\beta'$ using rule $R$; then we apply the \prff axiom; and finally we rewrite all $\gamma_\tautt$ back into $\accept_\tautt^\IDx$.  Then, we use the $\oplus\textsf{-indep}$ axiom to replace $\suci^{j_a} \oplus \nonce_{\textsf{f}}$ with a fresh nonce $\nonce_{\textsf{f}}'$. This yield the derivation:
  \[
    \infer[R]{
      \lrpcond{\beta}{
        \cframe_\taut,\cleak_\taut,
        \suci^{j_a}}
      \;\sim\;
      \lrpcond{\beta}{
        \cframe_\taut,\cleak_\taut,
        \suci}
    }{
      \infer[\prff]{
        \lrpcond{\beta'}{
          \inframep_\taut,
          \pair
          {\suci^{j_a} \oplus \row{\nonce^{j_a}}{\key^{\IDx}}}
          {\mac{\spair
              {\suci^{j_a}}
              {\nonce^{j_a}}}
            {\mkey^{\IDx}}{5}},
          \inleakp_\taut,
          \suci^{j_a}}
        \;\sim\;
        \lrpcond{\beta}{
          \cframe_\taut,\cleak_\taut,
          \suci}
      }{
        \infer[R]{
          \lrpcond{\beta'}{
            \inframep_\taut,
            \pair
            {\suci^{j_a} \oplus \nonce_{\textsf{f}}}
            {\mac{\spair
                {\suci^{j_a}}
                {\nonce^{j_a}}}
              {\mkey^{\IDx}}{5}},
            \inleakp_\taut,
            \suci^{j_a}}
          \;\sim\;
          \lrpcond{\beta}{
            \cframe_\taut,\cleak_\taut,
            \suci}
        }{
          \infer[R]{
            \lrpcond{\beta}{
              \inframe_\taut,
              \pair
              {\suci^{j_a} \oplus \nonce_{\textsf{f}}}
              {\mac{\spair
                  {\suci^{j_a}}
                  {\nonce^{j_a}}}
                {\mkey^{\IDx}}{5}},
              \inleak_\taut,
              \suci^{j_a}}
            \;\sim\;
            \lrpcond{\beta}{
              \cframe_\taut,\cleak_\taut,
              \suci}
          }{
            \infer[\oplus\textsf{-indep}]{
              \lrpcond{\beta}{
                \inframe_\taut,
                \pair
                {\suci^{j_a} \oplus \nonce_{\textsf{f}}}
                {\mac{\spair
                    {\suci^{j_a}}
                    {\nonce^{j_a}}}
                  {\mkey^{\IDx}}{5}},
                \inleak_\taut,
                \suci^{j_a}}
              \;\sim\;
              \lrpcond{\beta}{
                \cframe_\taut,\cleak_\taut,
                \suci}
            }{
              \lrpcond{\beta}{
                \inframe_\taut,
                \pair
                {\nonce'_{\textsf{f}}}
                {\mac{\spair
                    {\suci^{j_a}}
                    {\nonce^{j_a}}}
                  {\mkey^{\IDx}}{5}},
                \inleak_\taut,
                \suci^{j_a}}
              \;\sim\;
              \lrpcond{\beta}{
                \cframe_\taut,\cleak_\taut,
                \suci}
            }
          }
        }
      }
    }
  \]
  We do a similar reasoning to replace ${\mac{\spair{\suci^{j_a}}{\nonce^{j_a}}}{\mkey^{\IDx}}{5}}$ with a fresh nonce $\nonce''_{\textsf{f}}$ using $\prfmac^5$ axiom (we omit the details):
  \[
    \infer[(R+\prfmac^5)^*]{
      \lrpcond{\beta}{
        \inframe_\taut,
        \pair
        {\nonce'_{\textsf{f}}}
        {\mac{\spair{\suci^{j_a}}{\nonce^{j_a}}}{\mkey^{\IDx}}{5}},
        \inleak_\taut,
        \suci^{j_a}}
      \;\sim\;
      \lrpcond{\beta}{
        \cframe_\taut,\cleak_\taut,
        \suci}
    }{
      \lrpcond{\beta}{
        \inframe_\taut,
        \pair
        {\nonce'_{\textsf{f}}}
        {\nonce''_{\textsf{f}}},
        \inleak_\taut,
        \suci^{j_a}}
      \;\sim\;
      \lrpcond{\beta}{
        \cframe_\taut,\cleak_\taut,
        \suci}
    }
  \]
  We then do the same thing on the right side, and use the $\fa$ axiom under context
  \[
    \infer[(R+\prfmac^5 + \prff + \oplus\textsf{-indep})^*]{
      \lrpcond{\beta}{
        \inframe_\taut,
        \pair
        {\nonce'_{\textsf{f}}}
        {\nonce''_{\textsf{f}}},
        \inleak_\taut,
        \suci^{j_a}}
      \;\sim\;
      \lrpcond{\beta}{
        \cframe_\taut,\cleak_\taut,
        \suci}
    }{
      \infer[\fac]{
        \lrpcond{\beta}{
          \inframe_\taut,
          \pair
          {\nonce'_{\textsf{f}}}
          {\nonce''_{\textsf{f}}},
          \inleak_\taut,
          \suci^{j_a}}
        \;\sim\;
        \lrpcond{\beta}{
          \inframe_\taut,
          \pair
          {\nonce'_{\textsf{f}}}
          {\nonce''_{\textsf{f}}},
          \inleak_\taut,
          \suci}
      }{
        \lrpcond{\beta}{
          \inframe_\taut,
          {\nonce'_{\textsf{f}}},
          {\nonce''_{\textsf{f}}},
          \inleak_\taut,
          \suci^{j_a}}
        \;\sim\;
        \lrpcond{\beta}{
          \inframe_\taut,
          {\nonce'_{\textsf{f}}},
          {\nonce''_{\textsf{f}}},
          \inleak_\taut,
          \suci}
      }
    }
  \]
  Using the fact that $\beta \in \inleak_\taut$, we have:
  \[
    \infer[\simp]{
      \lrpcond{\beta}{
        \inframe_\taut,
        {\nonce'_{\textsf{f}}},
        {\nonce''_{\textsf{f}}},
        \inleak_\taut,
        \suci^{j_a}}
      \;\sim\;
      \lrpcond{\beta}{
        \inframe_\taut,
        {\nonce'_{\textsf{f}}},
        {\nonce''_{\textsf{f}}},
        \inleak_\taut,
        \suci}
    }{
      \lrpcond{\beta}{
        \inframe_\taut,
        \inleak_\taut,
        \suci^{j_a}},
      {\nonce'_{\textsf{f}}},
      {\nonce''_{\textsf{f}}},
      \;\sim\;
      \lrpcond{\beta}{
        \inframe_\taut,
        \inleak_\taut,
        \suci},
      {\nonce'_{\textsf{f}}},
      {\nonce''_{\textsf{f}}},
    }
  \]
  We then conclude using \ax{Fresh}:
  \[
    \infer[\ax{Fresh}^2]{
      \lrpcond{\beta}{
        \inframe_\taut,
        \inleak_\taut,
        \suci^{j_a}},
      {\nonce'_{\textsf{f}}},
      {\nonce''_{\textsf{f}}}
      \;\sim\;
      \lrpcond{\beta}{
        \inframe_\taut,
        \inleak_\taut,
        \suci},
      {\nonce'_{\textsf{f}}},
      {\nonce''_{\textsf{f}}}
    }{
      \lrpcond{\beta}{
        \inframe_\taut,
        \inleak_\taut,
        \suci^{j_a}}
      \;\sim\;
      \lrpcond{\beta}{
        \inframe_\taut,
        \inleak_\taut,
        \suci}
    }
  \]
\end{itemize}
We now sketch the proof of the induction property for the remaining cases:
\begin{itemize}
\item If $\ai = \fnai(j)$ with $j \ne j_a$. First, we can decompose $t_\taut$ into terms of $\inframe_\taut,\inleak_\taut$, except for the term:
  \[
    \pair
    {\suci^{j} \oplus \row{\nonce^{j}}{\key^{\IDx}}}
    {\mac{\spair
        {\suci^{j}}
        {\nonce^{j}}}
      {\mkey^{\IDx}}{5}}
  \]
  The rest of the proof goes as in case $\ai = \fnai(j_a)$. On both side, we do the following:
  \begin{itemize}
  \item We apply the \prff axiom to replace $\suci^{j} \oplus \row{\nonce^{j}}{\key^{\IDx}}$ with $\suci^{j} \oplus \nonce_{\textsf{f}}$ where $\nonce_{\textsf{f}}$ is a fresh nonce.
  \item We use the $\oplus\textsf{-indep}$ axiom to replace $\suci^{j} \oplus \nonce_{\textsf{f}}$ with a fresh nonce $\nonce_{\textsf{f}}'$
  \item We apply the $\prfmac^5$ axiom to replace ${\mac{\spair{\suci^{j}}{\nonce^{j}}}{\mkey^{\IDx}}{5}}$ with a fresh nonce $\nonce''_{\textsf{f}}$.
  \end{itemize}
  Finally we use \ax{Fresh} to get rid of the introduced nonces $\nonce'_{\textsf{f}}$ and $\nonce''_{\textsf{f}}$.

\item If $\ai = \cnai(j,0)$. Using the induction hypothesis we know that $\beta \ra \neg \accept_\taut^\IDx$. We can therefore rewrite all occurrences of $\accept_\taut^\IDx$ into $\false$ under the condition $\beta$. This removes all occurrences of $\instate_\taut(\suci_\hn^\IDx)$ in $\cleak_\taut \backslash \inleak_\taut$ and $t_\taut$. We can then decompose the resulting terms into terms of $\inframe_\taut,\inleak_\taut$.

\item If $\ai = \cnai(j,1)$. We can decompose $\cleak_\taut \backslash \inleak_\taut$ and $t_\taut$ into terms of $\inframe_\taut,\inleak_\taut$, except for the term $\mac{\nonce^j}{\mkey^{\IDx}}{4}$. We get rid of this term using the $\prfmac^4$ axiom.

\item If $\ai = \pnai(j,0)$. This is trivial using $\ax{Fresh}$.

\item If $\ai = \pnai(j,1)$. We use \ref{equ3} to rewrite all occurrences of $\accept_\taut^\IDx$ in $\cleak_\taut \backslash \inleak_\taut$ and $t_\taut$. We can then decompose the resulting terms into terms of $\inframe_\taut,\inleak_\taut$, except for the term:
  \[
    \mac{\spair
      {\nonce^j}
      {\sqnsuc(\pi_2(\dec(\pi_1(g(\inframe_\taut)),\sk_\hn)))}}
    {\mkey^{\IDx}}{2}
  \]
  We get rid of this term using the $\prfmac^2$ axiom.

\item If $\ai$ is a symbolic action of user $\ID$, with $\ID = \IDx$, then either $\ai = \npuai{2}{\IDx}{j_i}$ or $\ai = \fuai_\IDx(j_i)$.
  \begin{itemize}
  \item If $\ai = \npuai{2}{\IDx}{j_i}$, then we show using \ref{equ2} that:
    \[
      \beta \ra
      \left(
        \accept_\taut^\IDx
        \lra
        g(\inframe_\taut) = t_\taua
      \right)
    \]
    Therefore we can rewrite $\accept_\taut^\IDx$ into $g(\inframe_\taut) = t_\taua$ under $\beta$ in $t_\taut$. The resulting term can be easily decomposed into terms of $\inframe_\taut,\inleak_\taut$.

  \item $\ai = \fuai_\IDx(j_i)$. We do a similar reasoning, but using \ref{equ1} instead of \ref{equ2}. We omit the details.
  \end{itemize}
\end{itemize}

\subsection{Stronger Characterizations}
Using the $\guti$ concealment lemma, we can show the following stronger version of \ref{acc3}:
\begin{lemma}
  For every valid symbolic trace $\tau = \_,\ai$ and identity $\ID$ we have:
  \begin{itemize}
  \item \customlabel{sacc1}{\lpsacc{1}} If $\ai = \cuai_\ID(j,1)$. Let $\taut = \_, \cuai_\ID(j,0)$ such that $\taut \popre \tau$, and let $\key \equiv \key^\ID$. Then:
    \begin{center}
      \begin{tikzpicture}
        [dn/.style={inner sep=0.2em,fill=black,shape=circle},
        sdn/.style={inner sep=0.15em,fill=white,draw,solid,shape=circle},
        sl/.style={decorate,decoration={snake,amplitude=1.6}},
        dl/.style={dashed},
        pin distance=0.5em,
        every pin edge/.style={thin}]

        \draw[thick] (0,0)
        node[left=1.3em] {$\tau:$}
        -- ++(0.5,0)
        node[dn,pin={above:{$\cuai_\ID(j,0)$}}] {}
        node[below=0.3em]{$\taut$}
        -- ++(2.5,0)
        node[dn,pin={above:{$\cnai(j_1,0)$}}] {}
        node[below=0.3em]{$\tauo$}
        -- ++(2.5,0)
        node[dn,pin={above:{$\cuai_\ID(j,1)$}}] {}
        node[below=0.3em]{$\tau$};
      \end{tikzpicture}
    \end{center}
    \[
      \accept_\tau^\ID \;\ra\;
      \bigvee_{\tauo = \_, \cnai(j_0,0)
        \atop{\taut \potau \tauo}}
      \left(
        \begin{alignedat}{2}
          &&&\accept_\tauo^\ID \;\wedge\;
          g(\inframe_\tauo) = \instate_\taut(\suci_\ue^\ID)\;\wedge\;
          \pi_1(g(\inframe_\tau)) = \nonce^{j_0} \\
          &\wedge\;\;&&
          \pi_2(g(\inframe_\tau)) =
          \instate_\tauo(\sqn_\hn^\ID) \oplus \ow{\nonce^{j_0}}{\key}
          \wedge
          \instate_\tau(\suci_\ue^\ID) = \instate_\tauo(\suci_\hn^\ID)
        \end{alignedat}
      \right)
    \]
  \end{itemize}
\end{lemma}

\begin{proof}
  First, by applying \ref{acc3} we get that:
  \begin{alignat*}{2}
    \accept_\tau^\ID &\;\ra\;\;&&
    \bigvee_{\tauo = \_, \cnai(j_0,0)
      \atop{\tauo \popre \tau}}
    \left(
      \begin{alignedat}{1}
        &\accept_\tauo^\ID \wedge
        \pi_1(g(\inframe_\tau)) = \nonce^{j_0}
        \wedge
        \pi_2(g(\inframe_\tau)) =
        \instate_\tauo(\sqn_\hn^\ID) \oplus \ow{\nonce^{j_0}}{\key}\\
        &\wedge
        \instate_\tau(\suci_\ue^\ID) = \instate_\tauo(\suci_\hn^\ID)
      \end{alignedat}
    \right)
    \numberthis\label{eq:fdilfajfeijeqripurqw}
  \end{alignat*}
  We have $\accept_\tau^\ID \ra \instate_\tau(\uetsuccess^\ID)$, and $\instate_\tau(\uetsuccess^\ID) \ra \instate_\taut(\success_\ue^\ID)$. Let $\tauo = \_, \cnai(j_0,0)$, we know that $\accept_\tauo^\ID \ra \instate_\tauo(\suci_\hn^\ID) \ne \unset$. Therefore:
  \begin{alignat*}{2}
    \accept_\tau^\ID &\;\ra\;\;&&
    \bigvee_{\tauo = \_, \cnai(j_0,0)
      \atop{\tauo \popre \tau}}
    \left(
      \begin{alignedat}{1}
        &\accept_\tauo^\ID \wedge
        \pi_1(g(\inframe_\tau)) = \nonce^{j_0}
        \wedge
        \pi_2(g(\inframe_\tau)) =
        \instate_\tauo(\sqn_\hn^\ID) \oplus \ow{\nonce^{j_0}}{\key}\\
        &\wedge
        \instate_\tau(\suci_\ue^\ID) = \instate_\tauo(\suci_\hn^\ID)
        \wedge
        \instate_\tau(\suci_\ue^\ID) \ne \unset
        \wedge
        \instate_\taut(\success_\ue^\ID)
      \end{alignedat}
    \right)
  \end{alignat*}
  To conclude, we need to get a contradiction if $\tauo \popre \taut$. Therefore, we assume that $\tauo \popre \taut$. If there does not exists any $\tautt$ such that $\tautt = \_,\fuai_\ID(j_i) \popre \taut$, then it is easy to show that $\instate_\tau(\suci_\ue^\ID) = \unset$. In that case, from the equation above we get that $\neg\accept_\tau^\ID$, which concludes this case.

  Therefore, let $\tautt$ be maximal w.r.t $\popre$ such that $\tautt = \_,\fuai_\ID(j_i) \popre \taut$. We have $\tautt \not \potau \fuai_\ID(\_)$. Assume that there exists a user $\ID$ action between $\tautt$ and $\taut$. It is easy to show by induction (over $\tau'$ in $\tautt \popre \tau' \popreleq \taut$ that, since there are no $\fuai_\ID(\_)$ action between $\tautt$ and $\taut$, we have $\neg\instate_\taut(\success_\ue^\ID)$. This implies $\neg\accept_\tau^\ID$, which concludes this case.

  Therefore we can safely assume that there are no user $\ID$ actions between $\tautt$ and $\taut$. We deduce that $\instate_\taut(\success_\ue^\ID) \ra \accept_\tautt^\ID$. Hence $\accept_\tau^\ID \ra \accept_\tautt^\ID$. By applying \ref{equ1} to $\tautt$, we know that:
  \begin{alignat*}{3}
    \accept_\tau^\ID
    &\;\;\ra\;\;&&
    \bigvee_{\taua = \_,\fnai(j_a) \popre \tautt
      \atop{\taua \not \potau \newsession_\ID(\_)}}
    \futr{\tautt}{\taua}
    \numberthis\label{eq:dihafijdapofapfas}
  \end{alignat*}
  We recall that:
  \[
    \futr{\tautt}{\taua} \;\equiv\;
    \left(\begin{alignedat}{2}
        &\injauth_\tautt(\ID,j_a)
        \wedge\instate_\tautt(\eauth_\hn^{j_a}) \ne \unknownid\\
        \wedge\;& \pi_1(g(\inframe_\tautt)) =
        \suci^{j_a} \xor \row{\nonce^{j_a}}{\key}
        \wedge\; \pi_2(g(\inframe_\tautt)) =
        \mac{\spair
          {\suci^{j_a}}
          {\nonce^{j_a}}}
        {\mkey}{5}
      \end{alignedat}\right)
  \]
  Let $\taua = \_,\fnai(j_a) \popre \tautt$ such that $\taua \not \potau \newsession_\ID(\_)$. We know that there exists $\taun = \_,\pnai(j_a,1)$ and $\taun = \_,\cnai(j_a,1)$ such that $\taun \popre \taua$, and that $\futr{\tautt}{\taua} \ra \accept_\taun^\ID$. Let $\taui = \_,\npuai{1}{\ID}{j_i}$ or $\_,\cuai_\ID(j_i,1)$ such that $\taui \popre \tautt$. If $\taun \popre \taui$, we can show using \ref{acc1} if $\taun = \_,\pnai(j_a,1)$ or \ref{acc4} if $\taun = \_,\pnai(j_a,1)$ that we have $\neg \futr{\tautt}{\taua}$. Therefore, we assume that $\taui \popre \taun$. We depict the situation below:
  \begin{center}
    \begin{tikzpicture}
      [dn/.style={inner sep=0.2em,fill=black,shape=circle},
      sdn/.style={inner sep=0.15em,fill=white,draw,solid,shape=circle},
      sl/.style={decorate,decoration={snake,amplitude=1.6}},
      dl/.style={dashed},
      pin distance=0.5em,
      every pin edge/.style={thin}]

      \draw[thick] (0,0)
      node[left=1.3em] {$\tau:$}
      -- ++(0.5,0)
      node[dn,pin={above,align=left:{$\npuai{1}{\ID}{j_i}$\\ or $\cuai_\ID(j_i,1)$}}]
      (a) {}
      node[below,yshift=-0.3em,name=a0] {$\taui$}
      -- ++(2.5,0)
      node[dn,pin={above,align=left:{$\pnai(j_a,1)$\\ or $\cnai(j_a,1)$}}]
      (b) {}
      node[below,yshift=-0.3em,name=b0] {$\taun$}
      -- ++(2.5,0)
      node[dn,pin={above:{$\fnai(j_a)$}}]
      (c) {}
      node[below,yshift=-0.3em,name=c0] {$\taua$}
      -- ++(2.5,0)
      node[dn,pin={above:{$\fuai_\ID(j_i)$}}]
      (d) {}
      node[below,yshift=-0.3em,name=d0] {$\tautt$}
      -- ++(2.5,0)
      node[dn,pin={above:{$\cuai_\ID(j,0)$}}]
      (e) {}
      node[below,yshift=-0.3em,name=e0] {$\taut$}
      -- ++(2.5,0)
      node[dn,pin={above:{$\cuai_\ID(j,1)$}}]
      (f) {}
      node[below,yshift=-0.3em,name=f0] {$\tau$}
      -- ++(0.5,0);

      \draw[thin,dashed] (a0) -- ++(0,-0.5) -| (b0)
      {[draw=none] -- ++(0,-0.5)} -| (c0)
      {[draw=none] -- ++(0,-0.5)} -| (d0);

      \draw[thin,dotted] (e0) -- ++(0,-0.5) -| (f0);
    \end{tikzpicture}
  \end{center}
  We can check that $\futr{\tautt}{\taua} \ra \cstate_\tautt(\suci_\ue^\ID) = \suci^{j_a}$. Moreover, since there are no user $\ID$ actions between $\tautt$ and $\taut$ or between $\taut$ and $\tau$, $\cstate_\tautt(\suci_\ue^\ID) = \instate_\tau(\suci_\ue^\ID)$. From \eqref{eq:fdilfajfeijeqripurqw}, we know that $\accept_\tau^\ID \ra \instate_\tau(\suci_\ue^\ID) = \instate_\tauo(\suci_\hn^\ID)$. It follows that:
  \[
    \accept_\tau^\ID \wedge \futr{\tautt}{\taua}
    \ra
    \instate_\tauo(\suci_\hn^\ID) = \suci^{j_a}
    \numberthis\label{eq:falhiforfawiorjapsidsa}
  \]
  If $\tauo \popre \taun$, then it is easy to check that $\instate_\tauo(\suci_\hn^\ID) \ne \suci^{j_a}$. Therefore we have $\neg (\accept_\tau^\ID \wedge \futr{\tautt}{\taua})$.

  Now, we assume that $\taun \popre \tauo$. Recall that we assumed $\tauo \popre \taut$. Our goal is to apply the $\suci$ concealment lemma (Lemma~\ref{lem:suci-conceal}) to $\tauo$ get a contradiction. We can check that the following hypothesis of Lemma~\ref{lem:suci-conceal} is true:
  \[
    \left\{
      \tau' \mid \taui \popre_{\tauo} \taub
    \right\}
    \cap
    \left\{
      \npuai{\_}{\ID}{j},
      \cuai_\ID(j,\_),
      \fuai_\ID(j)
      \mid j \in \mathbb{N}
    \right\}
    \subseteq
    \left\{
      \npuai{2}{\ID}{j_i},
      \fuai_\ID(j_i)
    \right\}
  \]
  We deduce that:
  \begin{alignat*}{2}
    \incaccept_\taun^\ID \wedge
    \cstate_\taui(\bauth_\ue^\ID) = \nonce^{j_a} \wedge
    \accept_\taui^\IDx
    \ra
    g(\inframe_\tauo) \ne \suci^{j_a}
    \numberthis\label{eq:gsdighsfijapf}
  \end{alignat*}
  We know that:
  \[
    \futr{\tautt}{\taua} \ra
    \accept_\taui^\ID \wedge
    \cstate_\taui(\bauth_\ue^\ID) = \nonce^{j_a}
    \numberthis\label{eq:aifjafiajpfoasf}
  \]
  Moreover, $\neg \incaccept_\taun^\ID \ra \cstate_\taun(\suci_\hn^\ID) \ne \guti^{j_a}$. It is then straightforward to check that $\neg \incaccept_\taun^\ID \ra \cstate_\tauo(\suci_\hn^\ID) \ne \guti^{j_a}$. Therefore, using \eqref{eq:falhiforfawiorjapsidsa} we get that:
  \[
    \accept_\tau^\ID \wedge \futr{\tautt}{\taua}
    \wedge
    \neg \incaccept_\taun^\ID
    \ra
    \left(
      \instate_\tauo(\suci_\hn^\ID) = \suci^{j_a}
      \wedge
      \instate_\tauo(\suci_\hn^\ID) \ne \suci^{j_a}
    \right)
    \ra
    \false
  \]
  Hence $ \accept_\tau^\ID \wedge \futr{\tautt}{\taua} \ra \incaccept_\taun^\ID$. Therefore using~\eqref{eq:gsdighsfijapf} and~\eqref{eq:aifjafiajpfoasf}, we get:
  \begin{alignat*}{2}
    \accept_\tau^\ID \wedge \futr{\tautt}{\taua}
    \ra
    g(\inframe_\tauo) \ne \suci^{j_a}
    \numberthis\label{eq:adfkafjaifjaafasrfaf}
  \end{alignat*}
  We have $\accept_\tauo^\ID \ra g(\inframe_\tauo) = \instate_\tauo(\suci_\hn^\ID)$. We get from this, \eqref{eq:falhiforfawiorjapsidsa} and \eqref{eq:adfkafjaifjaafasrfaf} that:
  \[
    \accept_\tau^\ID \wedge \futr{\tautt}{\taua} \wedge \accept_\tauo^\ID
    \ra
    \false
  \]
  We showed that this holds for every $\taua = \_,\fnai(j_a) \popre \tautt$. We deduce from \eqref{eq:dihafijdapofapfas} that:
  \[
    \accept_\tau^\ID \wedge \accept_\tauo^\ID
    \ra
    \false
  \]
  Since we have this for every $\tauo \popre \taut$, we can rewrite \eqref{eq:fdilfajfeijeqripurqw} to get:
  \begin{alignat*}{2}
    \accept_\tau^\ID &\;\ra\;\;&&
    \bigvee_{\tauo = \_, \cnai(j_0,0)
      \atop{\taut \popre \tauo \popre \tau}}
    \left(
      \begin{alignedat}{1}
        &\accept_\tauo^\ID \wedge
        \pi_1(g(\inframe_\tau)) = \nonce^{j_0}
        \wedge
        \pi_2(g(\inframe_\tau)) =
        \instate_\tauo(\sqn_\hn^\ID) \oplus \ow{\nonce^{j_0}}{\key}\\
        &\wedge
        \instate_\tau(\suci_\ue^\ID) = \instate_\tauo(\suci_\hn^\ID)
      \end{alignedat}
    \right)
    \numberthis\label{eq:sidhifjasapsf}
  \end{alignat*}
  To conclude, we observe that $\accept_\tau^\ID \wedge \futr{\tautt}{\taua} \ra \instate_\taut(\suci_\ue^\ID) = \suci^{j_a}$. We recall that $\accept_\tauo^\ID \ra g(\inframe_\tauo) = \instate_\tauo(\suci_\hn^\ID)$. We conclude using \eqref{eq:falhiforfawiorjapsidsa} that:
  \[
    \accept_\tau^\ID \wedge \futr{\tautt}{\taua} \ra
    \instate_\taut(\suci_\ue^\ID)
    =
    g(\inframe_\tauo)
  \]
  Since this holds for every $\taua = \_,\fnai(j_a) \popre \tautt$, we deduce from \eqref{eq:dihafijdapofapfas} that:
  \[
    \accept_\tau^\ID \wedge \accept_\tauo^\ID
    \ra
    \instate_\taut(\suci_\ue^\ID)
    =
    g(\inframe_\tauo)
  \]
  Hence using \eqref{eq:sidhifjasapsf} we get:
  \begin{alignat*}{2}
    \accept_\tau^\ID &\;\ra\;\;&&
    \bigvee_{\tauo = \_, \cnai(j_0,0)
      \atop{\taut \popre \tauo \popre \tau}}
    \left(
      \begin{alignedat}{1}
        &\accept_\tauo^\ID \wedge
        \pi_1(g(\inframe_\tau)) = \nonce^{j_0}
        \wedge
        \pi_2(g(\inframe_\tau)) =
        \instate_\tauo(\sqn_\hn^\ID) \oplus \ow{\nonce^{j_0}}{\key}\\
        &\wedge
        \instate_\tau(\suci_\ue^\ID) = \instate_\tauo(\suci_\hn^\ID)
        \wedge
        \instate_\taut(\suci_\ue^\ID)
        =
        g(\inframe_\tauo)
      \end{alignedat}
    \right)
  \end{alignat*}
  This concludes this proof.
\end{proof}

We can now prove the following strong acceptance characterization properties:
\begin{lemma}
  \label{lem:strong-charac}
  For every valid symbolic trace $\tau = \_,\ai$ and identity $\ID$ we have:
  \begin{itemize}
  \item \customlabel{sequ1}{\lpsequ{1}} If $\ai = \fuai_\ID(j)$. For every $\taut = \_,\fnai(j_0) \popre \tau$, if we let $\tautt = \_,\cuai_\ID(j,0)$ or $\_,\npuai{1}{\ID}{j}$ then:
    \begin{alignat*}{2}
      \accept_\tau^\ID
      &\;\;\leftrightarrow\;\;&
      \bigvee_{\tautt \potau \taut = \_,\fnai(j_0)}
      \futr{\tau}{\taut}
    \end{alignat*}

  \item \customlabel{sequ2}{\lpsequ{2}} If $\ai = \cuai_\ID(j,1)$. Let $\tautt = \_,\cuai_\ID(j,0)$ such that $\tautt \popre \tau$. Then for every $\taut$ such that $\taut = \_,\cnai(j_1,0)$ and $\tautt \potau \taut$, we let:
    \begin{alignat*}{2}
      \ptr{\tautt,\tau}{\taut}
      &\;\equiv\;\;&&
      \left(
        \begin{alignedat}{2}
          &&&
          \pi_1(g(\inframe_\tau)) = \nonce^{j_1}
          \;\wedge\;
          \pi_2(g(\inframe_\tau)) =
          \instate_\taut(\sqn_\hn^\ID) \oplus \ow{\nonce^{j_1}}{\key^\ID} \\
          &\wedge\;\;&&
          \pi_3(g(\inframe_\tau)) =
          \mac{\striplet
            {\nonce^{j_1}}
            {\instate_\taut(\sqn_\hn^\ID)}
            {\instate_\tautt(\suci_\ue^\ID)}}
          {\mkey^\ID}{3}\\
          &\wedge\;\;&&
          g(\inframe_\taut) = \instate_\tautt(\suci_\ue^\ID)
          \;\wedge\;
          \instate_\tautt(\suci_\ue^\ID) = \instate_\taut(\suci_\hn^\ID)
          \;\wedge\;
          \instate_\tautt(\success_\ue^\ID)\\
          &\wedge\;\;&&
          \range{\instate_\tau(\sqn_\ue^\ID)}{\instate_\taut(\sqn_\hn^\ID)}
        \end{alignedat}
      \right)
    \end{alignat*}
    Then we have:
    \begin{mathpar}
      \left(
        \ptr{\tautt,\tau}{\taut} \;\ra\; \accept_\tau^\ID \wedge \accept_\taut^\ID
      \right)
      _{\taut = \_,\cnai(j_1,0)
        \atop{\tautt \potau \taut}}

      \accept_\tau^\ID
      \leftrightarrow
      \bigvee
      _{\taut = \_,\cnai(j_1,0)
        \atop{\tautt \potau \taut}}
      \ptr{\tautt,\tau}{\taut}
    \end{mathpar}

  \item \customlabel{sequ3}{\lpsequ{3}} If $\ai = \cnai(j,1)$. Let $\taut = \_,\cnai(j,0)$ such that $\taut \popre \tau$. Let $\ID \in \iddom$ and $\taui,\tautt$ be such that $\taui = \_,\cuai_\ID(j_i,1)$, $\tautt = \_,\cuai_\ID(j_i,0)$ and $\tautt \potau \taut \potau \taui$. Let:
    \begin{alignat*}{2}
      \ftr{\tautt,\taui}{\taut,\tau}
      &\;\equiv\;\;&&
      \left(
        \ptr{\tautt,\taui}{\taut} \wedge
        g(\inframe_\tau) = \mac{\nonce^j}{\mkey^\ID}{4}
      \right)
    \end{alignat*}
    Then we have:
    \begin{mathpar}
      \left(
        \ftr{\tautt,\taui}{\taut,\tau} \;\ra\;
        \accept_\tau^\ID \wedge
        \accept_\taui^\ID \wedge
        \accept_\taut^\ID
      \right)
      _{\tautt = \_,\cuai_\ID(j_i,0)
        \atop{\taui = \_,\cuai_\ID(j_i,1)
          \atop{\tautt \potau \taut \potau \taui}}}

      \accept_\tau^\ID
      \;\leftrightarrow\;\;
      \bigvee
      _{\tautt = \_,\cuai_\ID(j_i,0)
        \atop{\taui = \_,\cuai_\ID(j_i,1)
          \atop{\tautt \potau \taut \potau \taui}}}
      \ftr{\tautt,\taui}{\taut,\tau}
    \end{mathpar}

  \item \customlabel{sequ4}{\lpsequ{4}} If $\ai = \npuai{2}{\ID}{j}$ then for every $\taut = \_,\pnai(j_1,1)$ such that $\tautt \potau \taut$, we have:
    \[
      \left(
        \neg\instate_\tau(\sync_\ue^\ID) \wedge \supitr{\tautt,\tau}{\taut}
      \right)
      \;\ra\;
      \incaccept_\taut^\ID \wedge
      \instate_\tau(\sqn_\hn^\ID) - \cstate_\taut(\sqn_\hn^\ID) = \zero
    \]
    Moreover:
    \[
      \left(
        \neg\instate_\tau(\sync_\ue^\ID)\wedge
        \accept_\tau^\ID
      \right)
      \;\ra\;
      \cstate_\tau(\sqn_\ue^\ID) - \cstate_\tau(\sqn_\hn^\ID) = \zero
    \]
  \end{itemize}
\end{lemma}

\subsection{Proof of Lemma~\ref{lem:strong-charac}}

\subsubsection*{Proof of \ref{sequ1}}
First, we apply \ref{equ1}:
\begin{alignat*}{3}
  \accept_\tau^\ID
  &\;\;\leftrightarrow\;\;&&
  \bigvee_{\taut = \_,\fnai(j_0) \popre \tau
    \atop{\taut \not \potau \newsession_\ID(\_)}}
  \futr{\tau}{\taut}
\end{alignat*}
Let $\taut = \_,\fnai(j_0) \popre \tau$. Remark that if $\tautt \popre \taut$ then $\taut \not \potau \newsession_\ID(\_)$. Hence to conclude we just need to show that if $\taut \popre \tautt$ then $\neg\futr{\tau}{\taut}$.

Let $\taui = \_,\npuai{2}{\ID}{j}$ or $\_,\cuai_\ID(j,1)$ such that $\taui \popre \tau$. We do a case disjunction on $\taui$:
\begin{itemize}
\item If $\taui = \_,\npuai{2}{\ID}{j}$. We know that $\futr{\tau}{\taut} \ra \accept_\taui^\ID$, hence by applying \ref{acc2} to $\taui$:
  \[
    \futr{\tau}{\taut} \;\ra\;
    \bigvee_{\taux = \_, \pnai(j_x,1)
      \atop{\tautt \popre \taux \popre \taui}}
    \accept_\taux^\ID \;\wedge\;
    g(\inframe_\tautt) = \nonce^{j_x} \;\wedge\;
    \pi_1(g(\inframe_\taux)) =
    \enc{
      \spair{\ID}
      {\instate_\tautt(\sqn_\ue^\ID)}}
    {\pk_\hn}{\enonce^j}
  \]
  We know that $\futr{\tau}{\taut} \ra g(\inframe_\tautt) = \nonce^{j_0}$. We deduce that the main term of the disjunction above is false whenever $j_x \ne j_0$. Hence we have $\neg\futr{\tau}{\taut}$ if there does not exist any $\tauo$ such that $\tautt \popre \tauo \popre \taui$ and $\tauo = \_,\pnai(j_0,1)$.

  If $\taut \popre \tautt$ then we know that for every $\tauo$, if $\tauo = \_,\pnai(j_0,1) \popre \tau$ then $\tauo \popre \taut$, and by transitivity $\tauo \popre \tautt$. Hence there does not exist any $\tauo$ such that $\tautt \popre \tauo \popre \taui$ and $\tauo = \_,\pnai(j_0,1)$. We deduce that if $\taut \popre \tautt$ then $\neg\futr{\tau}{\taut}$ holds, which is what we wanted.
  
\item If $\taui = \_,\cuai_\ID(j,1)$. We know that $\futr{\tau}{\taut} \ra \accept_\taui^\ID$, hence by applying \ref{sacc1} to $\taui$:
  \[
    \futr{\tau}{\taut} \;\ra\;
    \bigvee_{\taux = \_, \cnai(j_x,0)
      \atop{\tautt \popre \taux \popre \taui}}
    \left(
      \begin{alignedat}{2}
        &&&\accept_\taux^\ID \;\wedge\;
        g(\inframe_\taux) = \instate_\tautt(\suci_\ue^\ID)\;\wedge\;
        \pi_1(g(\inframe_\taui)) = \nonce^{j_x} \\
        &\wedge\;\;&&
        \pi_2(g(\inframe_\taui)) =
        \instate_\taux(\sqn_\hn^\ID) \oplus \ow{\nonce^{j_x}}{\key}
        \wedge
        \instate_\taui(\suci_\ue^\ID) = \instate_\taux(\suci_\hn^\ID)
      \end{alignedat}
    \right)
  \]
  Similarly to what we did for $\taui = \_,\npuai{2}{\ID}{j_i}$, the main term above if false if $j_x \ne j_0$. Hence we have $\neg\futr{\tau}{\taut}$ if there does not exist any $\tauo$ such that $\tautt \popre \tauo \popre \taui$ and $\tauo = \_,\cnai(j_0,0)$. Since this is the case whenever $\taut \popre \tautt$, we deduce that if $\taut \popre \tautt$ then $\neg\futr{\tau}{\taut}$ holds. This concludes this case, and this proof.
\end{itemize}

\subsubsection*{Proof of \ref{sequ2}}
We start by repeating the proof of \ref{equ4}, but using \ref{sacc1} instead of \ref{acc3}. All the reasonings we did apply, only the set of $\taut$ the disjunction quantifies upon changes. We quantify over $\taut$ in $\{\taut \mid \taut = \_, \cnai(j_0,0) \wedge \tautt \potau \taut\}$ instead of $\{\taut \mid \taut = \_, \cnai(j_0,0) \wedge \taut \popre \tau\}$. We get that:
\begin{alignat*}{2}
  \accept_\tau^\ID
  &\lra\;\;&&
  \bigvee_{\taut = \_, \cnai(j_0,0)
    \atop{\tautt \potau \taut}}
  \left(
    \begin{alignedat}{1}
      &\pi_3(g(\inframe_\tau)) =
      \mac{\striplet
        {\nonce^{j_0}}
        {\instate_\taut(\sqn_\hn^\ID)}
        {\instate_\tau(\suci_\ue^\ID)}}
      {\mkey}{3} \wedge
      \instate_\tau(\uetsuccess^{\ID}) \\
      &\wedge\range{\instate_\tau(\sqn_\ue^\ID)}
      {\instate_\taut(\sqn_\hn^\ID)}\wedge
      g(\inframe_\taut) = \instate_\taut(\suci_\hn^{\ID})\wedge
      \pi_1(g(\inframe_\tau)) = \nonce^{j_0}\\
      & \wedge
      \pi_2(g(\inframe_\tau)) =
      \instate_\taut(\sqn_\hn^\ID) \oplus \ow{\nonce^{j_0}}{\key}
      \wedge
      \instate_\tau(\suci_\ue^\ID) = \instate_\taut(\suci_\hn^\ID)
    \end{alignedat}
  \right)
  \numberthis\label{eq:fdglhetuqeirqwupe}
\end{alignat*}
Since no user $\ID$ action occurs between $\tautt$ and $\tau$, we know that:
\begin{mathpar}
  \instate_\tau(\suci_\ue^\ID) = \instate_\tautt(\suci_\ue^\ID)

  \instate_\tau(\uetsuccess^\ID) \lra \instate_\tautt(\success_\ue^\ID)
\end{mathpar}
Using this, we can rewrite \eqref{eq:fdglhetuqeirqwupe} as follows (we underline the subterms where rewriting occurred):
\begin{alignat*}{2}
  \accept_\tau^\ID
  &\lra\;\;&&
  \bigvee_{\taut = \_, \cnai(j_0,0)
    \atop{\tautt \potau \taut}}
  \left(
    \begin{alignedat}{1}
      &\pi_3(g(\inframe_\tau)) =
      \mac{\striplet
        {\nonce^{j_0}}
        {\instate_\taut(\sqn_\hn^\ID)}
        {\uline{\instate_\tautt(\suci_\ue^\ID)}}}
      {\mkey}{3} \wedge
      \uline{\instate_\tautt(\success_\ue^\ID)}\\
      &\wedge\range{\instate_\tau(\sqn_\ue^\ID)}
      {\instate_\taut(\sqn_\hn^\ID)}\wedge
      g(\inframe_\taut) = \instate_\taut(\suci_\hn^{\ID})\\
      &\wedge\pi_1(g(\inframe_\tau)) = \nonce^{j_0} \wedge
      \pi_2(g(\inframe_\tau)) =
      \instate_\taut(\sqn_\hn^\ID) \oplus \ow{\nonce^{j_0}}{\key}
      \wedge
      \uline{\instate_\tautt(\suci_\ue^\ID)} = \instate_\taut(\suci_\hn^\ID)
    \end{alignedat}
  \right)\\\displaybreak[0]
  \intertext{We rewrite $\instate_\taut(\suci_\hn^\ID)$ into $\instate_\tautt(\suci_\ue^\ID)$:}
  &\lra\;\;&&
  \bigvee_{\taut = \_, \cnai(j_0,0)
    \atop{\tautt \potau \taut}}
  \left(
    \begin{alignedat}{1}
      &\pi_3(g(\inframe_\tau)) =
      \mac{\striplet
        {\nonce^{j_0}}
        {\instate_\taut(\sqn_\hn^\ID)}
        {\instate_\tautt(\suci_\ue^\ID)}}
      {\mkey}{3} \wedge
      \instate_\tautt(\success_\ue^\ID)\\
      &\wedge\range{\instate_\tau(\sqn_\ue^\ID)}
      {\instate_\taut(\sqn_\hn^\ID)}\wedge
      g(\inframe_\taut) = \dashuline{\instate_\tautt(\suci_\ue^\ID)}\\
      &\wedge\pi_1(g(\inframe_\tau)) = \nonce^{j_0} \wedge
      \pi_2(g(\inframe_\tau)) =
      \instate_\taut(\sqn_\hn^\ID) \oplus \ow{\nonce^{j_0}}{\key}
      \wedge
      \instate_\tautt(\suci_\ue^\ID) = \instate_\taut(\suci_\hn^\ID)
    \end{alignedat}
  \right)\\\displaybreak[0]
  \intertext{Finally we re-order the conjuncts:}
  &\lra\;\;&&
  \bigvee_{\taut = \_, \cnai(j_0,0)
    \atop{\tautt \potau \taut}}
  \left(
    \begin{alignedat}{2}
      &&&
      \pi_1(g(\inframe_\tau)) = \nonce^{j_1}
      \;\wedge\;
      \pi_2(g(\inframe_\tau)) =
      \instate_\taut(\sqn_\hn^\ID) \oplus \ow{\nonce^{j_1}}{\key^\ID} \\
      &\wedge\;\;&&
      \pi_3(g(\inframe_\tau)) =
      \mac{\striplet
        {\nonce^{j_1}}
        {\instate_\taut(\sqn_\hn^\ID)}
        {\instate_\tautt(\suci_\ue^\ID)}}
      {\mkey^\ID}{3}\\
      &\wedge\;\;&&
      g(\inframe_\taut) = \instate_\tautt(\suci_\ue^\ID)
      \;\wedge\;
      \instate_\tautt(\suci_\ue^\ID) = \instate_\taut(\suci_\hn^\ID)
      \;\wedge\;
      \instate_\tautt(\success_\ue^\ID)\\
      &\wedge\;\;&&
      \range{\instate_\tau(\sqn_\ue^\ID)}{\instate_\taut(\sqn_\hn^\ID)}
    \end{alignedat}
  \right)\\
  &\lra\;\;&&
  \bigvee_{\taut = \_, \cnai(j_0,0)
    \atop{\tautt \potau \taut}}
  \ptr{\tautt,\tau}{\taut}
\end{alignat*}
Checking that for every $\taut = \_,\cnai(j_1,0) \tautt \potau \taut$:
\[
  \left(
    \ptr{\tautt,\tau}{\taut} \;\ra\; \accept_\tau^\ID \wedge \accept_\taut^\ID
  \right)
\]
is straightforward.

\subsubsection*{Proof of \ref{sequ3}}
The proof that:
\[
  \accept_\tau^\ID
  \;\leftrightarrow\;\;
  \bigvee
  _{\tautt = \_,\cuai_\ID(j_i,0)
    \atop{\taui = \_,\cuai_\ID(j_i,1)
      \atop{\tautt \potau \taut \potau \taui}}}
  \ftr{\tautt,\taui}{\taut,\tau}
\]
is exactly the same than the proof of \ref{equ5}, but using \ref{sequ2} instead of \ref{equ4}.

Finally, it is straightforward to check that for every $\tautt = \_,\cuai_\ID(j_i,0)$, $\taui = \_,\cuai_\ID(j_i,1)$ such that $\tautt \potau \taut \potau \taui$ we have:
\[
  \left(
    \ftr{\tautt,\taui}{\taut,\tau} \;\ra\;
    \accept_\tau^\ID \wedge
    \accept_\taui^\ID \wedge
    \accept_\taut^\ID
  \right)
\]

\subsubsection*{Proof of \ref{sequ4}}
Let $\tautt = \_\npuai{1}{\ID}{j}$ such that $\tautt \popre \tau$.  Using \ref{equ2}, we know that:
\[
  \accept_\tau^\ID \;\lra\;
  \bigvee_{\taut = \_,\pnai(j_1,1)
    \atop{\tautt \potau \taut}}
  \supitr{\tautt,\tau}{\taut}
\]
Therefore to prove \ref{sequ4} it is sufficient to show that for every $\taut$ such that $\taut = \_, \pnai(j_1,1)$ and $\tautt \potau \taut$ we have:
\[
  \left(
    \neg\instate_\tau(\sync_\ue^\ID)
    \wedge \supitr{\tautt,\tau}{\taut}\right)
  \;\ra\;
  \incaccept_\taut^\ID \wedge
  \instate_\tau(\sqn_\hn^\ID) - \cstate_\taut(\sqn_\hn^\ID) = \zero \wedge
  \cstate_\tau(\sqn_\ue^\ID) - \cstate_\tau(\sqn_\hn^\ID) = \zero
\]
Hence let $\taut$ with $\taut = \_, \pnai(j_1,1)$ and $\tautt \potau \taut$.

\paragraph{Part 1}
First, we are going to show that:
\[
  \left(
    \neg\instate_\tau(\sync_\ue^\ID)
    \wedge \supitr{\tautt,\tau}{\taut}\right)
  \;\ra\;
  \cstate_{\taut}(\sqn_\hn^\ID) = \cstate_\tautt(\sqn_\ue^\ID)
  \numberthis\label{eq:diff-zero0}
\]
We know that $\incaccept_\taut^\ID \;\ra\; \cstate_{\taut}(\sqn_\hn^\ID) = \cstate_\tautt(\sqn_\ue^\ID)$, which is what we wanted. Hence it only remains to show \eqref{eq:diff-zero0} when $\neg\incaccept_\taut^\ID$.  Using \ref{b4} we know that $\cstate_{\taut}(\sqn_\hn^\ID) \le \cstate_\taut(\sqn_\ue^\ID)$. By validity of $\tau$ there are no user action between $\tautt$ and $\tau$, hence $\cstate_\tau(\sqn_\ue^\ID) = \instate_\tautt(\sqn_\ue^\ID)$. Observe that:
\[
  \left(
    \supitr{\tautt,\tau}{\taut} \wedge
    \neg \incaccept_\taut^\ID
  \right)
  \;\ra\;
  \instate_{\taut}(\sqn_\hn^\ID) > \instate_\tautt(\sqn_\ue^\ID)
\]
And:
\[
  \left(
    \supitr{\tautt,\tau}{\taut} \wedge
    \neg \incaccept_\taut^\ID
  \right)
  \;\ra\;
  \cstate_{\tautt}(\sqn_\ue^\ID) = \instate_\tautt(\sqn_\ue^\ID) + 1
\]
Graphically:
\begin{center}
  \begin{tikzpicture}
    [dn/.style={inner sep=0.2em,fill=black,shape=circle},
    sdn/.style={inner sep=0.15em,fill=white,draw,solid,shape=circle},
    sl/.style={decorate,decoration={snake,amplitude=1.6}},
    dl/.style={dashed},
    pin distance=0.5em,
    every pin edge/.style={thin}]

    \draw[thick] (0,0)
    node[left=1.3em] {$\tau:$}
    -- ++(0.5,0)
    node[dn,pin={above:{$\npuai{1}{\ID}{j}$}}]
    (b) {}
    node[below,yshift=-0.3em] {$\tautt$}
    -- ++(3,0)
    node[dn] (b0) {}
    -- ++(3,0)
    node[dn,pin={above:{$\pnai(j_1,1)$}}]
    (c) {}
    node[below,yshift=-0.3em] {$\taut$}
    -- ++(3,0)
    node[dn,pin={above:{$\npuai{2}{\ID}{j}$}}]
    (d) {}
    node[below,yshift=-0.3em] {$\tau$};

    \path (b) -- ++ (0,-1.1)
    -- ++ (0,-1)
    node[sdn] (sb2) {}
    node[below] (b2) {$\instate_{\tautt}(\sqn_\ue^\ID)$};

    \path (b0) -- ++ (0,-1.1)
    -- ++ (0,-1)
    node[sdn] (sb02) {}
    node[below] (b02) {$\cstate_{\tautt}(\sqn_\ue^\ID)$};

    \path (c) -- ++ (0,-1.1)
    node[sdn] (sc1) {}
    node[above right] (c1) {$\instate_\taut(\sqn_\hn^\ID)$}
    -- ++ (0,-1)
    node[sdn] (sc2) {}
    node[below right] (c2) {$\instate_\taut(\sqn_\ue^\ID)$};

    \draw (sb2) -- (sb02) node[midway,below]{$+1$};
    \draw (sb02) -- (sc2) node[midway,below]{$=$};
    \draw (sc1) -- (sc2) node[midway,above,sloped]{$\le$};
    \draw (sb2) -- (sc1) node[midway,above,sloped]{$<$};
  \end{tikzpicture}
\end{center}
We deduce that:
\begin{alignat*}{2}
  \left(
    \neg\instate_\tau(\sync_\ue^\ID)
    \wedge \supitr{\tautt,\tau}{\taut}\wedge
    \neg \incaccept_\taut^\ID
  \right)
  &\;\ra\;&&
  \instate_\tautt(\sqn_\ue^\ID) <
  \instate_{\taut}(\sqn_\hn^\ID) \le
  \instate_\tautt(\sqn_\ue^\ID) + 1\\
  &\;\ra\;&&
  \instate_{\taut}(\sqn_\hn^\ID) =
  \instate_\tautt(\sqn_\ue^\ID) + 1\\
  &\;\ra\;&&
  \cstate_{\taut}(\sqn_\hn^\ID) =
  \cstate_\tautt(\sqn_\ue^\ID)
  \numberthis\label{eq:vsdovshoaeq}
\end{alignat*}
Which is what we wanted to show.

\paragraph{Part 2}
We now show that:
\[
  \left(
    \neg\instate_\tau(\sync_\ue^\ID)
    \wedge \supitr{\tautt,\tau}{\taut}\right)
  \;\ra\;
  \cstate_{\taut}(\sqn_\hn^\ID) > \instate_\tautt(\sqn_\hn^\ID)
  \numberthis\label{eq:diff-strict}
\]
First, notice that:
\begin{alignat*}{2}
  \incaccept_\taut^\ID
  &\;\ra\;&&
  \cstate_{\taut}(\sqn_\hn^\ID) = \instate_{\taut}(\sqn_\hn^\ID) + 1\\
  &\;\ra\;&&
  \cstate_{\taut}(\sqn_\hn^\ID) > \instate_{\taut}(\sqn_\hn^\ID)\\
  &\;\ra\;&&
  \cstate_{\taut}(\sqn_\hn^\ID) > \instate_{\tautt}(\sqn_\hn^\ID)
  \tag{By \ref{b5}}
\end{alignat*}
Therefore we only need to prove:
\[
  \left(
    \neg\instate_\tau(\sync_\ue^\ID)
    \wedge \supitr{\tautt,\tau}{\taut}
    \wedge \neg \incaccept_\taut^\ID
  \right)
  \;\ra\;
  \cstate_{\taut}(\sqn_\hn^\ID) > \instate_\tautt(\sqn_\hn^\ID)
\]
Which is straightforward:
\begin{alignat*}{2}
  \left(
    \neg\instate_\tau(\sync_\ue^\ID)
    \wedge \supitr{\tautt,\tau}{\taut}\wedge
    \neg \incaccept_\taut^\ID
  \right)
  &\;\ra\;&&
  \cstate_{\taut}(\sqn_\hn^\ID) =
  \instate_\tautt(\sqn_\ue^\ID) + 1
  \tag{By \eqref{eq:vsdovshoaeq}}\\
  &\;\ra\;&&
  \cstate_{\taut}(\sqn_\hn^\ID) >
  \instate_\tautt(\sqn_\ue^\ID) \\
  &\;\ra\;&&
  \cstate_{\taut}(\sqn_\hn^\ID) >
  \cstate_\tautt(\sqn_\hn^\ID)
  \tag{By \ref{b4}}
\end{alignat*}
Which concludes the proof of \eqref{eq:diff-strict}.

\paragraph{Part 3}
We give the proof of:
\[
  \left(
    \neg\instate_\tau(\sync_\ue^\ID)
    \wedge \supitr{\tautt,\tau}{\taut}
  \right)
  \;\ra\;
  \cstate_\tau(\sqn_\hn^\ID) = \cstate_\taut(\sqn_\hn^\ID) \wedge
  \cstate_\tau(\sqn_\ue^\ID) - \cstate_\tau(\sqn_\hn^\ID) = \zero
  \numberthis\label{eq:diff-zero-42}
\]
By validity of $\tau$ we know that $\cstate_\tau(\sqn_\ue^\ID) = \cstate_\tautt(\sqn_\ue^\ID)$, therefore using \eqref{eq:diff-zero0} we know that:
\[
  \left(
    \neg\instate_\tau(\sync_\ue^\ID)
    \wedge \supitr{\tautt,\tau}{\taut}
  \right)
  \;\ra\;
  \cstate_\taut(\sqn_\hn^\ID) =
  \cstate_\tau(\sqn_\ue^\ID)
\]
To conclude, we need to show that $\sqn_\hn^\ID$ was kept unchanged since $\taut$, i.e. that $\neg\instate_\tau(\sync_\ue^\ID) \wedge \supitr{\tautt,\tau}{\taut}$ implies that $\cstate_\taut(\sqn_\hn^\ID) = \cstate_\tau(\sqn_\hn^\ID)$. This requires that no $\supi$ or $\suci$ network session incremented $\sqn_\hn^\ID$. Therefore we need to show the two following properties:
\begin{itemize}
\item \textbf{$\supi$:} For every $\taut \potau \taui$ such that $\taui = \_, \pnai(j_i,1)$:
  \[
    \left(
      \neg\instate_\tau(\sync_\ue^\ID)
      \wedge \supitr{\tautt,\tau}{\taut}
    \right)
    \;\ra\;
    \neg\incaccept_\taui^{\ID}
    \numberthis\label{eq:not-increasing-sqn}
  \]
\item \textbf{$\suci$:} For every $\taut \potau \taui$ such that $\taui = \_, \cnai(j_i,1)$:
  \[
    \left(
      \neg\instate_\tau(\sync_\ue^\ID)
      \wedge
      \supitr{\tautt,\tau}{\taut}
    \right)
    \;\ra\;
    \neg\incaccept_\taui^{\ID}
    \numberthis\label{eq:not-increasing-sqn2}
  \]
\end{itemize}
Assuming the two properties above, showing that \eqref{eq:diff-zero-42} holds is easy. First, using \eqref{eq:not-increasing-sqn} and \eqref{eq:not-increasing-sqn2} we know that:
\[
  \left(
    \neg\instate_\tau(\sync_\ue^\ID)
    \wedge \supitr{\tautt,\tau}{\taut}
  \right)
  \;\ra\;
  \cstate_\tau(\sqn_\hn^\ID) = \cstate_\taut(\sqn_\hn^\ID)
\]
We know that $\cstate_\tau(\sqn_\ue^\ID) = \instate_\tau(\sqn_\ue^\ID)$. We deduce that $\cstate_\tau(\sqn_\hn^\ID) = \cstate_\tau(\sqn_\ue^\ID)$, which concludes this case. We summarize graphically this proof below:
\begin{center}
  \begin{tikzpicture}
    [dn/.style={inner sep=0.2em,fill=black,shape=circle},
    sdn/.style={inner sep=0.15em,fill=white,draw,solid,shape=circle},
    sl/.style={decorate,decoration={snake,amplitude=1.6}},
    dl/.style={dashed},
    pin distance=0.5em,
    every pin edge/.style={thin}]

    \draw[thick] (0,0)
    node[left=1.3em] {$\tau:$}
    -- ++(0.5,0)
    node[dn,pin={above:{$\npuai{1}{\ID}{j}$}}]
    (a) {}
    node[below,yshift=-0.3em] {$\tautt$}
    -- ++(3,0)
    node[dn,pin={above:{$\pnai(j_1,1)$}}]
    (b) {}
    node[below,yshift=-0.3em] {$\taut$}
    -- ++(4,0)
    node[dn,pin={above:
      {$\pnai(j_i,1)$ or $\cnai(j_i,1)$}}]
    (c) {}
    node[below,yshift=-0.3em] {$\taui$}
    -- ++(4,0)
    node[dn,pin={above:{$\npuai{2}{\ID}{j}$}}]
    (d) {}
    node[below,yshift=-0.3em] {$\tau$};

    \path (a) -- ++ (0,-1.3)
    -- ++ (0,-1)
    node[sdn] (aa1) {}
    node[left,align=left] (a1) {$\cstate_\tautt(\sqn_\ue^\ID) =$\\
      $  \sqnsuc(\instate_\tautt(\sqn_\ue^\ID))$};

    \path (b) -- ++ (0,-1.3)
    node[sdn] (bb1) {}
    node[left,yshift=1.2em,align=left] (b1) {$\cstate_\taut(\sqn_\hn^\ID) =$\\
      $ \sqnsuc(\instate_\tautt(\sqn_\ue^\ID))$};

    \path (c) -- ++ (0,-1.3)
    node[sdn] (cc1) {}
    node[align=left,above] (c1) {$\cstate_\taui(\sqn_\hn^\ID) =
      \instate_\taui(\sqn_\hn^\ID)$};

    \path (d) -- ++ (0,-1.3)
    node[sdn] (dd1) {}
    node[right,align=left] (d1) {$\cstate_\tau(\sqn_\hn^\ID)$}
    -- ++ (0,-1)
    node[sdn] (dd2) {}
    node[right,align=left] (d2) {$\cstate_\tau(\sqn_\ue^\ID) =$\\
      $ \instate_\tau(\sqn_\ue^\ID)$};

    \draw (aa1) -- (dd2)
    node[midway,below] {$=$};
    \draw (aa1) -- (bb1)
    node[midway,below,sloped] {$=$}
    -- (cc1)
    node[midway,below] {$=$}
    -- (dd1)
    node[midway,below] {$=$};
  \end{tikzpicture}
\end{center}

\paragraph{Part 4 (Proof of  \eqref{eq:not-increasing-sqn})}
Let $\taut \potau \taui$ such that $\taui = \_, \pnai(j_i,1)$. Using \ref{acc1} we know that:
\[
  \accept_\taui^\ID \;\ra\;
  \bigvee_{\tau' = \_,\npuai{1}{\ID}{j'} \potau \taui}
  \left(
    \pi_1(g(\inframe_\taui)) =
    \enc{\pair
      {\ID}
      {\instate_{\tau'}(\sqn_\ue^\ID)}}
    {\pk_\hn}{\enonce^{j'}}
    \wedge
    g(\inframe_{\tau'}) = \nonce^{j_i}
  \right)
\]
We know that $\supitr{\tautt,\tau}{\taut} \ra g(\inframe_{\tautt}) = \nonce^{j_1} \ne \nonce^{j_i}$. Moreover from the validity of $\tau$ we know that for every $\tau''$ such that:
\[
  \tautt = \_,\npuai{1}{\ID}{j}
  \potau
  \tau'' = \_,\ai''
  \potau \tau = \_,\npuai{2}{\ID}{j}
\]
We have $\ai'' \ne \npuai{\_}{\ID}{\_}$. Hence:
\[
  \supitr{\tautt,\tau}{\taut} \wedge \accept_\taui^\ID \;\ra\;
  \bigvee_{\tau' = \_,\npuai{1}{\ID}{j'} \potau \tautt}
  \left(
    \pi_1(g(\inframe_\taui)) =
    \enc{\pair
      {\ID}
      {\instate_{\tau'}(\sqn_\ue^\ID)}}
    {\pk_\hn}{\enonce^{j'}}
    \wedge
    g(\inframe_{\tau'}) = \nonce^{j_i}
  \right)
\]
Which implies that:
\[
  \supitr{\tautt,\tau}{\taut} \wedge \incaccept_\taui^\ID \;\ra\;
  \bigvee_{\tau' = \_,\npuai{1}{\ID}{j'} \potau \tautt}
  \left(
    \cstate_{\taui}(\sqn_\hn^\ID) =
    \sqnsuc(\instate_{\tau'}(\sqn_\ue^\ID))
  \right)
\]
We recall \eqref{eq:diff-zero0}:
\[
  \left(
    \neg\instate_\tau(\sync_\ue^\ID)
    \wedge \supitr{\tautt,\tau}{\taut}
  \right)
  \;\ra\;
  \left(
    \cstate_\taut(\sqn_\hn^\ID) =
    \cstate_\tautt(\sqn_\ue^\ID)
  \right)
\]
Let $\tau' = \_,\npuai{1}{\ID}{j'} \potau \tautt$. We know using \ref{b5} that:
\begin{mathpar}
  \cstate_\taut(\sqn_\hn^\ID) \le
  \cstate_\taui(\sqn_\hn^\ID)

  \cstate_{\tau'}(\sqn_\ue^\ID) \le
  \cstate_\tautt(\sqn_\ue^\ID)
\end{mathpar}
Moreover using \ref{a2} we know that $\cstate_{\tau'}(\sqn_\ue^\ID) \ne \cstate_\tautt(\sqn_\ue^\ID)$, hence $\cstate_{\tau'}(\sqn_\ue^\ID) < \cstate_\tautt(\sqn_\ue^\ID)$.  This implies that:
\begin{alignat*}{2}
  &&&
  \neg\instate_\tau(\sync_\ue^\ID)\wedge
  \supitr{\tautt,\tau}{\taut} \wedge \incaccept_\taui^\ID\\
  &\;\;\ra\;\;&&
  \bigvee_{\tau' = \_,\npuai{1}{\ID}{j'} \potau \tautt}
  \left(
    \begin{alignedat}{4}
      &&&\cstate_{\tau'}(\sqn_\ue^\ID) <
      \cstate_\tautt(\sqn_\ue^\ID)
      &\wedge&
      \cstate_\tautt(\sqn_\ue^\ID) =
      \cstate_\taut(\sqn_\hn^\ID)\\
      &\wedge&&
      \cstate_\taut(\sqn_\hn^\ID) \le
      \cstate_\taui(\sqn_\hn^\ID)
      &\wedge&
      \cstate_{\taui}(\sqn_\hn^\ID) =
      \cstate_{\tau'}(\sqn_\ue^\ID)
    \end{alignedat}
  \right)\\
  &\;\;\ra\;\;&&
  \bigvee_{\tau' = \_,\npuai{1}{\ID}{j'} \potau \tautt}
  \left(
    \cstate_{\tau'}(\sqn_\ue^\ID) <
    \cstate_{\tau'}(\sqn_\ue^\ID)
  \right)\\
  &\;\;\ra\;\;&& \false
\end{alignat*}
Which concludes this proof. We summarize graphically below:
\begin{center}
  \begin{tikzpicture}
    [dn/.style={inner sep=0.2em,fill=black,shape=circle},
    sdn/.style={inner sep=0.15em,fill=white,draw,solid,shape=circle},
    sl/.style={decorate,decoration={snake,amplitude=1.6}},
    dl/.style={dashed},
    pin distance=0.5em,
    every pin edge/.style={thin}]

    \draw[thick] (0,0)
    node[left=1.3em] {$\tau:$}
    -- ++(0.5,0)
    node[dn,pin={above:{$\npuai{1}{\ID}{j'}$}}]
    (a) {}
    node[below,yshift=-0.3em] {$\tau'$}
    -- ++(3.5,0)
    node[dn,pin={above:{$\npuai{1}{\ID}{j}$}}]
    (b) {}
    node[below,yshift=-0.3em] {$\tautt$}
    -- ++(3.5,0)
    node[dn,pin={above:{$\pnai(j_1,1)$}}]
    (c) {}
    node[below,yshift=-0.3em] {$\taut$}
    -- ++(3.5,0)
    node[dn,pin={above:{$\pnai(j_i,1)$}}]
    (d) {}
    node[below,yshift=-0.3em] {$\taui$}
    -- ++(0.5,0);

    \path (a) -- ++ (0,-1)
    node[sdn] (aa1) {}
    node[left] (a1) {$\cstate_{\tau'}(\sqn_\ue^\ID)$};

    \path (b) -- ++ (0,-1)
    node[sdn] (bb1) {}
    node[above right] (b1) {$\cstate_\tautt(\sqn_\ue^\ID)$};

    \path (c) -- ++ (0,-1)
    -- ++ (0,-1.5)
    node[sdn] (cc1) {}
    node[below left] (c1) {$\cstate_\taut(\sqn_\hn^\ID)$};

    \path (d) -- ++ (0,-1)
    -- ++ (0,-1.5)
    node[sdn] (dd1) {}
    node[right] (d1) {$\cstate_\taui(\sqn_\hn^\ID)$};

    \draw (aa1) -- (bb1)
    node[midway,above] {$<$};
    \draw (aa1) ..controls +(-60:1.5) and +(120:1.5) .. (dd1)
    node[pos=0.3,below] {$=$};
    \draw (bb1) -- (cc1)
    node[pos=0.3,above,sloped] {$=$};
    \draw (cc1) -- (dd1)
    node[midway,below] {$\le$};
  \end{tikzpicture}
\end{center}

\paragraph{Part 5 (Proof of  \eqref{eq:not-increasing-sqn2})}
Let $\taut \potau \taui$ such that $\taui = \_, \cnai(j_i,1)$. Using Lemma~\ref{lem:auth-serv-net}, we know that:
\[
  \accept_\taui^\ID \;\ra\;
  \left(\instate_\taui(\eauth_\hn^j) = \ID\right)
  \;\ra\;
  \bigvee_{\tau' = \_, \cuai_\ID(\_,1) \atop{\tau' \potau \taui}}
  \cstate_{\tau'}(\bauth_\ue^\ID) = \nonce^{j_i}
\]
Since $\supitr{\tautt,\tau}{\taut} \ra g(\inframe_\tautt) = \nonce^{j_1}$, we know that $\supitr{\tautt,\tau}{\taut} \ra \cstate_\tautt(\bauth_\ue^\ID) = \nonce^{j_1}$. As we know that $\nonce^{j_1} \ne \nonce^{j_i}$, we deduce that $\supitr{\tautt,\tau}{\taut} \;\ra\; \cstate_\tautt(\bauth_\ue^\ID) \ne \nonce^{j_i}$. Moreover using the validity of $\tau$ we know that $\cstate_\taui(\bauth_\ue^\ID) = \cstate_\tautt(\bauth_\ue^\ID)$. Therefore:
\[
  \left(\supitr{\tautt,\tau}{\taut} \wedge \accept_\taui^\ID \right)
  \;\ra\;
  \bigvee_{\tau' = \_,\cuai_\ID(\_,1) \atop{\tau' \potau \tautt}}
  \cstate_{\tau'}(\bauth_\ue^\ID) = \nonce^{j_i}
\]
Let $\tau' = \_,\cuai_\ID(\_,1)$ with $\tau' \potau \tautt$. We know that $\cstate_{\tau'}(\bauth_\ue^\ID) = \nonce^{j_i}$ implies that $\cstate_{\tau'}(\bauth_\ue^\ID) \ne \fail$, and therefore $\accept_{\tau'}$ holds:
\[
  \left(\cstate_{\tau'}(\bauth_\ue^\ID) = \nonce^{j_i}\right)
  \;\ra\;
  \left(\cstate_{\tau'}(\bauth_\ue^\ID) \ne \fail\right)
  \;\ra\;
  \accept_{\tau'}
\]
By applying \ref{acc3} we know that:
\[
  \accept_{\tau'} \;\ra\;
  \bigvee_{\taui' = \_,\cnai(j_i',0) \potau \tau'}
  \pi_1(g(\inframe_{\tau'})) = \nonce^{j_i'}
\]
Since $\cond{\accept_{\tau'}}{\cstate_{\tau'}(\bauth_\ue^\ID) } = \cond{\accept_{\tau'}}{\pi_1(g(\inframe_{\tau'}))}$ we deduce:
\[
  \left(\cstate_{\tau'}(\bauth_\ue^\ID) = \nonce^{j_i} \right)
  \;\ra\;
  \false  \text{ if }\tau' \potau \cnai(j_i,0)
\]
Hence if $\tau' \potau \cnai(j_i,0)$ we know that $\neg \left(\supitr{\tautt,\tau}{\taut} \wedge \accept_\taui^\ID \right)$, which is what we wanted to show. Therefore let $\taui' = \_,\cnai(j_i,0)$, and assume $\taui' \potau \tau'$.   We summarize graphically this below:
\begin{center}
  \begin{tikzpicture}
    [dn/.style={inner sep=0.2em,fill=black,shape=circle},
    sdn/.style={inner sep=0.15em,fill=white,draw,solid,shape=circle},
    sl/.style={decorate,decoration={snake,amplitude=1.6}},
    dl/.style={dashed},
    pin distance=0.5em,
    every pin edge/.style={thin}]

    \draw[thick] (0,0)
    node[left=1.3em] {$\tau:$}
    -- ++(0.5,0)
    node[dn,pin={above:{$\cnai(j_i,0)$}}]
    (a) {}
    node[below,yshift=-0.3em] {$\taui'$}
    -- ++(2.5,0)
    node[dn,pin={above:{$\cuai_\ID(\_,1)$}}]
    (a) {}
    node[below,yshift=-0.3em] {$\tau'$}
    -- ++(2.5,0)
    node[dn,pin={above:{$\npuai{1}{\ID}{j}$}}]
    (e) {}
    node[below,yshift=-0.3em] {$\tautt$}
    -- ++(2.5,0)
    node[dn,pin={above:{$\pnai(j_1,1)$}}]
    (b) {}
    node[below,yshift=-0.3em] {$\taut$}
    -- ++(2.5,0)
    node[dn,pin={above:{$\cnai(j_i,1)$}}]
    (c) {}
    node[below,yshift=-0.3em] {$\taui$}
    -- ++(2.5,0)
    node[dn,pin={above:{$\npuai{2}{\ID}{j}$}}]
    (d) {}
    node[below,yshift=-0.3em] {$\tau$};



  \end{tikzpicture}
\end{center}
We recall \eqref{eq:diff-strict}:
\[
  \left(
    \neg\instate_\tau(\sync_\ue^\ID)
    \wedge \supitr{\tautt,\tau}{\taut}
  \right)
  \;\ra\;
  \left(
    \instate_\tautt(\sqn_\hn^\ID) <
    \cstate_\taut(\sqn_\hn^\ID)
  \right)
\]
Hence, using \ref{b2} we know that:
\begin{alignat*}{2}
  \left(
    \neg\instate_\tau(\sync_\ue^\ID)
    \wedge \supitr{\tautt,\tau}{\taut}
  \right)
  &\;\ra\;&&
  \bigvee_{\tautt \popreleq \taux \popreleq \taut
    \atop{\taux = \_,\cnai(j_x,0)
      \text{ or }  \_,\cnai(j_x,1)
      \text{ or }  \_,\pnai(j_x,1)}}
  \cstate_\taut(\tsuccess_\hn^\ID) = \nonce^{j_x}
\end{alignat*}
Since $\cnai(j_i,0) \potau \tautt$ and $\taut \potau \cnai(j_i,1)$:
\begin{alignat*}{2}
  \left(
    \neg\instate_\tau(\sync_\ue^\ID)
    \wedge \supitr{\tautt,\tau}{\taut}
  \right)
  &\;\ra\;&&
  \cstate_\taut(\tsuccess_\hn^\ID) \ne \nonce^{j_i}
\end{alignat*}
For every $\taut \popreleq \tau''$ we have:
\[
  \cstate_{\tau''}(\tsuccess_\hn^\ID) \;=\;
  \begin{dcases*}
    \ite{\incaccept^\ID_{\tau''}}
    {\nonce^{j''}}
    {\instate_{\tau''}(\tsuccess_\hn^\ID)} & if $\tau''= \_,\pnai(j'',1)$\\
    \ite{\accept^\ID_{\tau''}}
    {\nonce^{j''}}
    {\instate_{\tau''}(\tsuccess_\hn^\ID)} & if $\tau'' = \_,\cnai(j'',0)$\\
    \instate_{\tau''}(\tsuccess_\hn^\ID) & otherwise
  \end{dcases*}
\]
Since $\tau' \not \potau \cnai(j_i,0)$, we know that after having set $\cstate_{\tau''}(\tsuccess_\hn^\ID)$ to $\nonce^{j_1}$ at $\taut$, it can never be set to $\nonce^{j_i}$. Formally, we show by induction that:
\[
  \cstate_\taut(\tsuccess_\hn^\ID) \ne \nonce^{j_i}
  \;\ra\;
  \cstate_{\tau''}(\tsuccess_\hn^\ID) \ne \nonce^{j_i}
\]
We conclude by observing that $\instate_{\taui}(\tsuccess_\hn^\ID) \ne \nonce^{j_i} \ra \neg \incaccept_\taui^\ID$.

\paragraph{Part 6}
To conclude the proof of $\ref{sequ4}$, it only remains to show that:
\[
  \left(
    \neg\instate_\tau(\sync_\ue^\ID)
    \wedge \supitr{\tautt,\tau}{\taut}\right)
  \;\ra\;
  \incaccept_\taut^\ID
  \numberthis\label{eq:dsjvbsuowiofge}
\]
Since $\supitr{\tautt,\tau}{\taut} \ra \accept_\taut^\ID$, and since:
\[
  \left(\accept_\taut^\ID \wedge \neg \incaccept_\taut^\ID \right)
  \;\lra\;
  \instate_\taut(\sqn_\hn^\ID) > \instate_\tautt(\sqn_\ue^\ID)
\]
To show that \eqref{eq:dsjvbsuowiofge} holds, it is sufficient to show that:
\[
  \left(
    \neg\instate_\tau(\sync_\ue^\ID)
    \wedge \supitr{\tautt,\tau}{\taut}\right)
  \;\ra\;
  \instate_\taut(\sqn_\hn^\ID) \le \instate_\tautt(\sqn_\ue^\ID)
\]
We generalize this, and show by induction that for every $\taun$ such that $\tautt \popreleq \taun \potau \taut$, we have:
\[
  \left(
    \neg\instate_\tau(\sync_\ue^\ID)
    \wedge \supitr{\tautt,\tau}{\taut}\right)
  \;\ra\;
  \cstate_\taun(\sqn_\hn^\ID) \le \instate_\tautt(\sqn_\ue^\ID)
\]
\begin{itemize}
\item If $\taun = \tautt$, this is immediate using \ref{b4} and the fact that $\cstate_\taun(\sqn_\hn^\ID) = \instate_\taun(\sqn_\hn^\ID)$.
\item Let $\taun >_\tau \tautt$. By induction, assume that:
  \[
    \left(
      \neg\instate_\tau(\sync_\ue^\ID)
      \wedge \supitr{\tautt,\tau}{\taut}\right)
    \;\ra\;
    \instate_{\taun}(\sqn_\hn^\ID) \le \instate_\tautt(\sqn_\ue^\ID)
  \]
  We then have three cases:
  \begin{itemize}
  \item If $\taun \ne \_, \pnai(\_,1)$ and $\taun \ne \_, \cnai(\_,1)$, we know that $\cstate_{\taun}(\sqn_\hn^\ID) = \instate_{\taun}(\sqn_\hn^\ID)$, and we conclude directly using the induction hypothesis.
  \item If $\taun = \_,\pnai(j_n,1)$. Using \ref{equ3} we know that:
    \begin{alignat*}{2}
      \cstate_{\taun}(\sqn_\hn^\ID) \ne \instate_{\taun}(\sqn_\hn^\ID)
      &\ra\;\;&&
      \accept_\taun^\ID\\
      &\ra\;\;&&
      \bigvee_{
        \taux = \_, \npuai{1}{\ID}{j_x}
        \atop{\taux \potau \taun}}
      \underbrace{\left(
          \begin{alignedat}{2}
            &&&g(\inframe_{\taux}) = \nonce^{j_n}
            \wedge
            \pi_1(g(\inframe_\taun)) =
            \enc{\spair
              {\ID}
              {\instate_{\taux}(\sqn_\ue^\ID)}}
            {\pk_\hn}{\enonce^{j_n}}\\
            &\wedge\;&&
            \pi_2(g(\inframe_\taun)) =
            {\mac{\spair
                {\enc{\spair
                    {\ID}
                    {\instate_{\taux}(\sqn_\ue^\ID)}}
                  {\pk_\hn}{\enonce^{j_n}}}
                {g(\inframe_{\taux})}}
              {\mkey^\ID}{1}}
          \end{alignedat}
        \right)}
      _{\theta_\taux}
    \end{alignat*}
    Since $\taun \potau \taut$, we know that $j_n \ne j_1$. Moreover, $\supitr{\tautt,\tau}{\taut} \ra g(\inframe_{\tautt}) = \nonce^{j_1}$. By consequence:
    \[
      \left(\supitr{\tautt,\tau}{\taut} \wedge g(\inframe_{\tautt}) = \nonce^{j_n} \right) \ra \false
    \]
    Which shows that $\left(\supitr{\tautt,\tau}{\taut} \wedge \theta_\tautt\right)\ra \false$. Hence:
    \begin{alignat*}{2}
      \supitr{\tautt,\tau}{\taut} \wedge
      \cstate_{\taun}(\sqn_\hn^\ID) \ne \instate_{\taun}(\sqn_\hn^\ID)
      &\ra\;\;&&
      \bigvee_{
        \taux = \_, \npuai{1}{\ID}{j_x}
        \atop{\taux \potau \tautt}}
      \theta_\taux
    \end{alignat*}
    Observe that for every $\taux = \_, \npuai{1}{\ID}{j_x}$ such that $\taux \potau \tautt$:
    \begin{alignat*}{2}
      \theta_\taux
      &\ra\;\;&&
      \cstate_\taun(\sqn_\hn^\ID) \;=\;
      \ite{\instate_\taun(\sqn_\hn^\ID) \le \instate_\taux(\sqn_\ue^\ID)}
      {\instate_\taux(\sqn_\ue^\ID)}
      {\instate_\taun(\sqn_\hn^\ID)}
    \end{alignat*}
    Using \ref{b5}, we know that $\instate_\taux(\sqn_\ue^\ID) \le \instate_\tautt(\sqn_\ue^\ID)$. Therefore:
    \begin{gather*}
      \theta_\taux \;\ra\;
      \cstate_\taun(\sqn_\hn^\ID)
      \;\le\;
      \ite{\instate_\taun(\sqn_\hn^\ID) \le \instate_\taux(\sqn_\ue^\ID)}
      {\instate_\tautt(\sqn_\ue^\ID)}
      {\instate_\taun(\sqn_\hn^\ID)}
    \end{gather*}
    And using the induction hypothesis, we get that:
    \begin{gather*}
      \left(
        \neg\instate_\tau(\sync_\ue^\ID)
        \wedge \supitr{\tautt,\tau}{\taut}
        \wedge \theta_\taux
      \right)
      \;\ra\;
      \cstate_\taun(\sqn_\hn^\ID)
      \le \instate_\tautt(\sqn_\ue^\ID)
    \end{gather*}
    Hence:
    \begin{gather*}
      \left(
        \neg\instate_\tau(\sync_\ue^\ID) \wedge
        \supitr{\tautt,\tau}{\taut} \wedge
        \cstate_{\taun}(\sqn_\hn^\ID) \ne \instate_{\taun}(\sqn_\hn^\ID)
      \right)
      \;\ra\;
      \cstate_\taun(\sqn_\hn^\ID)
      \le \instate_\tautt(\sqn_\ue^\ID)
    \end{gather*}
    From which we deduce, using the induction hypothesis, that:
    \begin{gather*}
      \left(
        \neg\instate_\tau(\sync_\ue^\ID) \wedge
        \supitr{\tautt,\tau}{\taut}
      \right)
      \;\ra\;
      \cstate_\taun(\sqn_\hn^\ID)
      \le \instate_\tautt(\sqn_\ue^\ID)
    \end{gather*}

  \item If $\taun = \_,\cnai(j_n,1)$. Using \ref{sequ2}, we know that:
    \begin{alignat*}{2}
      \cstate_{\taun}(\sqn_\hn^\ID) \ne \instate_{\taun}(\sqn_\hn^\ID)
      &\ra\;\;&&
      \accept_\taun^\ID\\
      &\ra\;\;&&
      \bigvee
      _{\taux' = \_,\cuai_\ID(j_x,0)
        \atop{\taun' = \_,\cnai(j_n,0)
          \atop{\taux = \_,\cuai_\ID(j_x,1)
            \atop{\taux' \potau \taun' \potau \taux \potau \taun}}}}
      \ftr{\taux',\taux}{\taun',\taun}
    \end{alignat*}
    Let $\taux = \_,\cuai_\ID(j_x,1)$, $\taun' = \_,\cnai(j_n,0)$, $\taux' = \_,\cuai_\ID(j_x,0)$ be such that $\taux' \potau \taun' \potau \taux \potau \taun$. One can check that:
    \begin{alignat*}{2}
      \incaccept_\taun^\ID\;
      &\ra\;\;&&
      \bigwedge_{\taux \potau \taui \potau \taun}
      \neg \incaccept_{\taui}^\ID\\
      &\ra\;\;&&
      \cstate_{\taun'}(\sqn_\hn^\ID) =
      \instate_\taun(\sqn_\hn^\ID)
    \end{alignat*}
    Moreover, since:
    \[
      \incaccept_\taun^\ID\;\ra\;
      \instate_\taux(\sqn_\ue^\ID)
      \;=\;
      \cstate_{\taun'}(\sqn_\hn^\ID)
    \]
    We deduce that:
    \[
      \ftr{\taux',\taux}{\taun',\taun} \;\ra\;
      \cstate_\taun(\sqn_\hn^\ID) =
      \ite{\incaccept_\taun^\ID}
      {\sqnsuc(\instate_\taux(\sqn_\ue^\ID))}
      {\instate_\taun(\sqn_\hn^\ID)}
    \]
    By validity of $\tau$, we know that $j_x \ne j$ and that $\taux \potau \tautt$.
    Therefore using \ref{b5} we know that $\cstate_\taux(\sqn_\ue^\ID) \le \instate_\tautt(\sqn_\ue^\ID)$. Moreover $\cstate_\taux(\sqn_\ue^\ID) = \sqnsuc(instate_\taux(\sqn_\ue^\ID))$. Hence:
    \[
      \ftr{\taux',\taux}{\taun',\taun} \;\ra\;
      \cstate_\taun(\sqn_\hn^\ID) \le
      \ite{\incaccept_\taun^\ID}
      {\instate_\tautt(\sqn_\ue^\ID)}
      {\instate_\taun(\sqn_\hn^\ID)}
    \]
    And using the induction hypothesis, we get that:
    \begin{gather*}
      \left(
        \neg\instate_\tau(\sync_\ue^\ID)
        \wedge \supitr{\tautt,\tau}{\taut}
        \wedge \ftr{\taux',\taux}{\taun',\taun}
      \right)
      \;\ra\;
      \cstate_\taun(\sqn_\hn^\ID)
      \le \instate_\tautt(\sqn_\ue^\ID)
    \end{gather*}
    Hence:
    \begin{gather*}
      \left(
        \neg\instate_\tau(\sync_\ue^\ID) \wedge
        \supitr{\tautt,\tau}{\taut} \wedge
        \cstate_{\taun}(\sqn_\hn^\ID) \ne \instate_{\taun}(\sqn_\hn^\ID)
      \right)
      \;\ra\;
      \cstate_\taun(\sqn_\hn^\ID)
      \le \instate_\tautt(\sqn_\ue^\ID)
    \end{gather*}
    From which we deduce, using the induction hypothesis, that:
    \begin{gather*}
      \left(
        \neg\instate_\tau(\sync_\ue^\ID) \wedge
        \supitr{\tautt,\tau}{\taut}
      \right)
      \;\ra\;
      \cstate_\taun(\sqn_\hn^\ID)
      \le \instate_\tautt(\sqn_\ue^\ID)
    \end{gather*}
  \end{itemize}
\end{itemize}


%% file: unlinkability.tex
\section{Unlinkability}
\label{section:app-unlink}
The goal of this section is to prove unlinkability of the $\faka$ protocol. To do this, we need, for every valid basic symbolic trace $\tau$, to show that there exists a derivation of $\cframe_\tau \sim \cframe_{\ufresh{\tau}}$. We show this by induction on $\tau$.

\subsection{Resistance against de-synchronization attacks}
To show that the $\suci$ protocol guarantees unlinkability, we need the protocol the be resilient to de-synchronization attacks: for every agent $\ID$, the adversary should not be able to keep $\ID$ synchronized in the left protocol, while de-synchronizing $\nu_\tau(\ID)$ in the right protocol.

Therefore, we need the range check on the sequence number to hold on the left iff the range check hold on the right. More precisely, for every identity $\ID$ and the matching identity $\nu_{\tau}(\ID)$ on the right, the range checks on the sequence numbers should be indistinguishable:
\begin{equation}
  \range
  {\cstate_\tau(\sqn_\ue^\ID)}
  {\cstate_\tau(\sqn_\hn^\ID)}
  \;\;\sim\;\;
  \range
  {\cstate_{\utau}(\sqn_\ue^{\nu_{\tau}(\ID)})}
  {\cstate_{\utau}(\sqn_\hn^{\nu_{\tau}(\ID)})}
  \label{eq:range-ex}
\end{equation}
But the property above is not a invariant of the $\faka$ protocol for two reasons:
\begin{itemize}
\item First, knowing that the range checks are indistinguishable after a
  symbolic execution $\tau$ is not enough to show that they are
  indistinguishable after $\tau_1 = \tau,\ai$ (for some $\ai$). For
  example, take a model where $\range{u}{v}$ is implemented as a check
  that the difference between $u$ and $v$ lies in some interval:
  \[
    \sem{\range{u}{v}} \text{ if and only if }
    \sem{u} - \sem{v} \in \{0,\dots,D\}
  \]
  for some constant $D > 0$, and $\sqnsuc$ is implemented as an increment by one. Then, a priori, we may have:
  \begin{gather*}
    \sem{\cstate_\tau(\sqn_\ue^\ID)} -
    \sem{\cstate_\tau(\sqn_\hn^\ID)} = 0 \in \{0,\dots,D\}\\
    \sem{\cstate_{\utau}(\sqn_\ue^{\nu_{\tau}(\ID)})} -
    \sem{\cstate_{\utau}(\sqn_\hn^{\nu_{\tau}(\ID)})} = D \in \{0,\dots,D\}
  \end{gather*}

  While \eqref{eq:range-ex} holds for $\tau$, it does not hold for $\tau_1 = \tau,\npuai{1}{\ID}{j}$. Indeed, after executing $\npuai{1}{\ID}{j}$ we have:
  \begin{gather*}
    \sem{\cstate_{\tau_1}(\sqn_\ue^\ID)} -
    \sem{\cstate_{\tau_1}(\sqn_\hn^\ID)} = 1 \in \{0,\dots,D\}\\
    \sem{\cstate_{\ufresh{\tau_1}}(\sqn_\ue^{\nu_{\tau_1}(\ID)})} -
    \sem{\cstate_{\ufresh{\tau_1}}(\sqn_\hn^{\nu_{\tau_1}(\ID)})} = D + 1
    \not \in \{0,\dots,D\}
  \end{gather*}
  To avoid this, we require that $\range{\_}{\_}$ and $\sqnsuc(\_)$ are implemented as, respectively, an equality check and an integer by-one increment.

  Moreover, we strengthen the induction property to show that the difference between the sequence numbers are indistinguishable, i.e.:
  \begin{equation}
    \idiff
    {\cstate_\tau(\sqn_\ue^\ID)}
    {\cstate_\tau(\sqn_\hn^\ID)}
    \;\;\sim\;\;
    \idiff
    {\cstate_{\utau}(\sqn_\ue^{\nu_{\tau}(\ID)})}
    {\cstate_{\utau}(\sqn_\hn^{\nu_{\tau}(\ID)})}
    \label{eq:diff-inv}
  \end{equation}

\item Second, the property in \eqref{eq:diff-inv} actually does not always hold: after a $\newsession_\ID(\_)$ action, the agent $\ID$ and the network may be synchronized on the left (if, e.g., the $\supi$ protocol has just been successfully executed), but $\nu_{\tau}(\ID)$ is not synchronized with the network.

  Even though the property does not hold, there is no attack on unlinkability. Indeed a desynchronization attack would need the $\suci$ protocol to succeed on the left and fail on the right. But the $\suci$ protocol requires that a fresh $\suci$ has been established between $\ID$ (resp. $\nu_{\tau}(\ID)$) and the network. This can only be achieved through a honest execution of the $\supi$ protocol. As such a execution will re-synchronized the agent and the network sequence numbers \emph{on both side}, there is no attack.

  To model this, we extend the symbolic state with a new boolean variable, $\sync_\ue^\ID$, that records whether there was a successful execution of the $\supi$ protocol with agent $\ID$ since the last $\newsession_\ID(\_)$. This variable is only here for proof purposes, and is never used in the actual protocol. We can then state the synchronization invariant:
  \[
    \underbrace{\begin{alignedat}{2}
        \ite{\cstate_\tau(\sync_\ue^\ID)&}
        {\idiff
          {\cstate_\tau(\sqn_\ue^\ID)}
          {\cstate_\tau(\sqn_\hn^\ID)}\\ &}
        {\bot}
      \end{alignedat}}_{\syncdiff_\tau^\ID}
    \;\;\sim\;\;
    \underbrace{\begin{alignedat}{2}
        \ite{\cstate_\utau(\sync_\ue^{\nu_{\tau}(\ID)})&}
        {\idiff
          {\cstate_{\utau}(\sqn_\ue^{\nu_{\tau}(\ID)})}
          {\cstate_{\utau}(\sqn_\hn^{\nu_{\tau}(\ID)})}\\ &}
        {\bot}
      \end{alignedat}}_{\syncdiff_\utau^{\nu_\tau(\ID)}}
  \]
\end{itemize}

\subsection{Strengthened induction hypothesis}

\begin{definition}
  Let $L = (i_1,\dots,i_l)$ be a list of indexes, and $(b_i)_{i \in L}$, $(t_i)_{i \in L}$ two list of terms. Then:
  \begin{equation*}
    \switch{i \in L}{(b_i)_{i \in L}}{(m_i)_{i \in L}} \;\equiv\;
    \begin{dcases*}
      \ite{b_{i_1}}{m_{i_1}}{\switch{i \in L_0}{(b_i)_{i \in L_0}}{(m_i)_{i \in L_0}}}
      & when $L \ne \emptyset$ and $L_0 = (i_2,\dots,i_l)$\\
      \bot & otherwise
    \end{dcases*}
  \end{equation*}
  We will often abuse notation, and write $\switch{i \in L}{b_i}{m_i}$ instead of $\switch{i \in L}{(b_i)_{i \in L}}{(m_i)_{i \in L}}$.
\end{definition}

\begin{proposition}
  \label{prop:case}
  Let $L = (i_1,\dots,i_l)$ be a list of indexes, and $(b_i)_{i \in L}$, $(t_i)_{i \in L}$ two list of terms. If $(b_i)_{i \in L}$ is a $\cs$ partition, then for any permutation $\pi$ of $\{1,\dots,l\}$, if we let $L_\pi = (i_{\pi(1)},\dots,i_{\pi{l}})$ then:
  \[
    \switch{i \in L}{b_i}{m_i} \;=\;
    \switch{i \in L_\pi}{b_i}{m_i}
  \]
  When that is the case, we write $\switch{i \in \{i_1,\dots,i_l\}}{b_i}{m_i}$ (i.e. we use a set notation instead of list notation).
\end{proposition}

\begin{proof}
  The proof is straightforward by induction over $|L|$.
\end{proof}

If $(b_i)_{i \in L}$ is such that $(\bigvee_{i \in L} b_i) = \true$ then the case where all tests fail and we return $\bot$ never happens. This motivates the introduction of a second definition.

\begin{definition}
  Let $L = (i_1,\dots,i_l)$ be a list of indexes with $l\ge 1$, and $(b_i)_{i \in L}$, $(t_i)_{i \in L}$ two list of terms. Then:
  \begin{equation*}
    \sswitch{i \in L}{(b_i)_{i \in L}}{(m_i)_{i \in L}} \;\equiv\;
    \begin{dcases*}
      \ite{b_{i_1}}{m_{i_1}}{\switch{i \in L_0}{(b_i)_{i \in L_0}}{(m_i)_{i \in L_0}}}
      & when $L_0 = (i_2,\dots,i_l)$ and $l \ge 1$\\
      m_1 & if $l = 1$
    \end{dcases*}
  \end{equation*}
\end{definition}

\begin{proposition}
  For every list of terms $(b_i)_{i \in L}$ and $(t_i)_{i \in L}$, if  $(\bigvee_{i \in L} b_i) = \true$ then:
  \[
    \switch{i \in L}{b_i}{m_i} = \sswitch{i \in L}{b_i}{m_i}
  \]
\end{proposition}

\begin{proof}
  We omit the proof.
\end{proof}

\begin{definition}
  \label{def:app-ind-gen}
  Let $\tau = \ai_0,\dots,\ai_n$ be a valid basic symbolic trace. Then $\reveal_{\tau}$ is a list of elements of the form $u \sim v$, where $u,v$ are terms, representing the information that can be safely leaked to the adversary. Let $\ai = \ai_n$. Then $\reveal_{\tau}$ contains exactly the following list of elements:
  \begin{enumerate}
  \item All the elements from $\reveal_{\tauo}$, where $\tauo = \ai_0,\dots,\ai_{n-1}$.
  \item \label{item:ref} For every base identity $\ID$, let:
    \[
      \msuci^\ID_\tau \;\equiv\;
      \cond{\cstate_\tau(\success_\ue^{\ID})}
      {\cstate_\tau(\suci_\ue^\ID)}
    \]
    We then have the following synchronization invariants.
    \begin{mathpar}
      \cstate_\tau(\success_\ue^{\ID})
      \;\sim\;
      \cstate_\utau(\success_\ue^{\nu_\tau(\ID)})

      \msuci^\ID_\tau
      \;\sim\;
      \msuci^{\nu_\tau(\ID)}_\utau

      \cstate_\tau(\sync_\ue^{\ID})
      \;\sim\;
      \cstate_\utau(\sync_\ue^{\nu_\tau(\ID)})

      \syncdiff_\tau^{\ID}
      \;\sim\;
      \syncdiff_\utau^{\nu_\tau(\ID)}

      \length(\spair{\ID}{\instate_{\tau}(\sqnini_\ue^\ID)})
      \;\sim\;
      \length(\spair{\ID}{\instate_{\tau}(\sqnini_\ue^\ID)})
    \end{mathpar}

  \item If $\ai \ne \newsession_{\_}(\_)$ then for every base identity $\ID$:
    \[
      \cstate_\tau(\sqn_\ue^\ID) -
      \instate_\tau(\sqn_\ue^\ID)
      \;\;\sim\;\;
      \cstate_\utau(\sqn_\ue^{\nu_\tau(\ID)}) -
      \instate_\utau(\sqn_\ue^{\nu_\tau(\ID)})
    \]

  \item If $\ai = \cuai_\ID(j,0)$, then:
    \[
      \cstate_\tau(\uetsuccess^{\ID})
      \;\sim\;
      \cstate_\utau(\uetsuccess^{\nu_\tau(\ID)})
    \]

  \item If $\ai = \npuai{1}{\ID}{j}$, then:
    \begin{alignat*}{2}
      \enc{\spair
        {\ID}
        {\instate_{\tau}(\sqn_\ue^\ID)}}
      {\pk_\hn}{\enonce^{j}}
      & \;\;\sim\;\;&&
      \enc{\spair
        {\nu_\tau(\ID)}
        {\instate_{\tau}(\sqn_\ue^{\nu_\tau(\ID)})}}
      {\pk_\hn}{\enonce^{j}}\\
      \mac{\spair
        {\enc{\spair
            {\ID}
            {\instate_{\tau}(\sqn_\ue^\ID)}}
          {\pk_\hn}{\enonce^{j}}}
        {g(\inframe_{\tau})}}
      {\mkey^\ID}{1}
      & \;\;\sim\;\;&&
      \mac{\spair
        {\enc{\spair
            {\nu_\tau(\ID)}
            {\instate_{\utau}(\sqn_\ue^{\nu_\tau(\ID)})}}
          {\pk_\hn}{\enonce^{j}}}
        {g(\inframe_{\utau})}}
      {\mkey^{\nu_\tau(\ID)}}{1}
    \end{alignat*}

  \item If $\ai = \npuai{2}{\ID}{\_}$, $\cuai_\ID(\_,1)$ or $\fuai_\ID(\_)$:
    \begin{alignat*}{2}
      \cstate_\tau(\eauth_\ue^\ID)
      & \;\;\sim\;\;&&
      \cstate_\utau(\eauth_\ue^{\nu_\tau(\ID)})
    \end{alignat*}

  \item If $\cuai_\ID(j,1)$ then for every $\taut = \_,\cnai(j_0,0)$ such that $\cuai_\ID(j,0) \potau \taut$:
    \[
      \mac{\nonce^{j_0}}{\mkey^\ID}{4}
      \;\;\sim\;\;
      \mac{\nonce^{j_0}}{\mkey^{\nu_\tau(\ID)}}{4}
    \]

  \item If $\ai = \pnai(j,1)$ then for every base identity $\ID$, for every $\taut = \_, \npuai{1}{\ID}{j_1} \popre \tau$ such that $\taut \not \potau \ns_\ID(\_)$ we have:
    \begin{alignat*}{2}
      \mac{\spair
        {\nonce^{j}}
        {\sqnsuc(\instate_\taut(\sqn_\ue^\ID))}}
      {\mkey^\ID}{2}
      & \;\;\sim\;\;&&
      \mac{\spair
        {\nonce^{j}}
        {\sqnsuc(\instate_\utaut(\sqn_\ue^{\nu_\tau(\ID)}))}}
      {\mkey^{\nu_\tau(\ID)}}{2}
    \end{alignat*}

  \item If $\ai = \pnai(j,1)$ or $\ai = \cnai(j,1)$, for all base identity $\ID$, we let:
    \begin{alignat*}{2}
      \neauth_\tau(\ID,j) &\;\;\equiv\;\;&&
      \eq{\cstate_\tau(\eauth_\hn^{j})}{\ID}\\
      \uneauth_\utau(\ID,j) &\;\;\equiv\;\;&&
      \bigvee_{\uID\in \copyid(\ID)}\eq{\cstate_\utau(\eauth_\hn^{j})}{\uID}
    \end{alignat*}
    Then we ask that:
    \begin{alignat*}{2}
      \neauth_\tau(\ID,j)
      &\;\;\sim\;\;&&
      \uneauth_\utau(\ID,j)
    \end{alignat*}

  \item If $\ai = \fnai(j)$ for every base identity $\ID$ we let
    $\{\uID_1,\dots,\uID_{l_\ID}\} = \copyid(\ID)$. We define:
    \begin{alignat*}{2}
      \tsuci_\tau(\ID,j)
      &\;\;\equiv\;\;&&
      \suci^j
      \oplus \row{\nonce^{j}}{\key^{\ID}}\\
      \utsuci_\utau(\ID,j)
      &\;\;\equiv\;\;&&
      \sswitch{1 \le i \le l_\ID}
      {\eq{\cstate_\utau(\eauth_\hn^{j})}{\uID_i}}
      {\suci^j\oplus\row{\nonce^{j}}{\key^{\uID_i}}}\\
      \tmac_\tau(\ID,j)
      &\;\;\equiv\;\;&&
      \mac{\spair{\suci^j}{\nonce^j}}{\mkey^{\ID}}{5}\\
      \utmac_\utau(\ID,j)
      &\;\;\equiv\;\;&&
      \sswitch{1 \le i \le l_\ID}
      {\eq{\cstate_\utau(\eauth_\hn^{j})}{\uID_i}}
      {\mac{\spair{\suci^j}{\nonce^j}}{\mkey^{\uID_i}}{5}}
    \end{alignat*}

    Then we ask that:
    \begin{alignat*}{2}
      \suci^j
      & \;\;\sim\;\;&&
      \suci^j\\
      \cond{\neauth_\tau(\ID,j)}
      {\left(\tsuci_\tau(\ID,j)\right)}
      & \;\;\sim\;\;&&
      \cond{\uneauth_\utau(\ID,j)}
      {\left(\utsuci_\utau(\ID,j)\right)}\\
      \cond{\neauth_\tau(\ID,j)}
      {\left(\tmac_\tau(\ID,j)\right)}
      & \;\;\sim\;\;&&
      \cond{\uneauth_\utau(\ID,j)}
      {\left(\utmac_\utau(\ID,j)\right)}
    \end{alignat*}
  \end{enumerate}
  Let $(u_i \sim v_i)_{i \in I}$ be such that $\reveal_{\tau} = (u_i \sim v_i)_{i \in I}$. Then we let $\lreveal_{\tau} = (u_i)_{i \in I}$ be the list of left elements of $\reveal_{\tau}$, and $\rreveal_{\tau} = (v_i)_{i \in I}$ list of left elements of $\reveal_{\tau}$ (in the same order).
\end{definition}

\begin{proposition}
  \label{prop:derivation}
  For every basic valid symbolic trace $\tau = \_,\ai$:
  \begin{itemize}
  \item \customlabel{der1}{\lpder{1}} For every base identity, for every $\taut$ such that $\taut \popre \tau$ and $\taut \not \potau \newsession_\ID(\_)$, there exist derivations using only $\fa$ and $\dup$ of:
    \begin{gather*}
      \begin{gathered}[c]
        \infer[\simp]{
          \begin{alignedat}{2}
            &&&\lreveal_{\tauo},
            \instate_\tau(\sync_\ue^\ID) \wedge
            \instate_\tau(\sqn_\hn^\ID) <
            \instate_\taut(\sqn_\ue^\ID)\\
            &\sim\;\;&&
            \rreveal_{\tauo},
            \instate_\utau(\sync_\ue^{\nu_\tau(\ID)}) \wedge
            \instate_\utau(\sqn_\hn^{\nu_\tau(\ID)}) <
            \instate_\utaut(\sqn_\ue^{\nu_\tau(\ID)})
          \end{alignedat}
        }{
          \inframe_\tau, \lreveal_{\tauo}
          \sim
          \inframe_\utau,\rreveal_{\tauo}
        }
      \end{gathered}
      \label{eq:deriv:u-n}\\[1em]
      \begin{gathered}
        \infer[\simp]{
          \begin{alignedat}{2}
            &&&\lreveal_{\tauo},
            \instate_\taut(\sync_\ue^\ID) \wedge
            \instate_\taut(\sqn_\hn^\ID) <
            \instate_\tau(\sqn_\ue^\ID)\\
            &\sim\;\;&&
            \rreveal_{\tauo},
            \instate_\utaut(\sync_\ue^{\nu_\tau(\ID)}) \wedge
            \instate_\utaut(\sqn_\hn^{\nu_\tau(\ID)}) <
            \instate_\utau(\sqn_\ue^{\nu_\tau(\ID)})
          \end{alignedat}
        }{
          \inframe_\tau, \lreveal_{\tauo}
          \sim
          \inframe_\utau,\rreveal_{\tauo}
        }
      \end{gathered}
      \label{eq:deriv:n-u}
    \end{gather*}
  \item \customlabel{der2}{\lpder{2}} If $\ai = \fuai_\ID(j)$. For every $\ID \in \iddom$, for every $\taut = \_,\fnai(j_0) \popre \tau$ such that $\taut \not \potau \ns_\ID(\_)$:
    \begin{itemize}
    \item We have $\utaut = \_,\fnai(j_0)$, $\utau = \_,\fuai_{\nu_\tau(\ID)}(j)$, $\utaut \poutau \utau$ and $\utaut \not \poutau \ns_{\nu_\tau(\ID)}(\_)$. Therefore, $\futr{\utau}{\utaut}$ is well-defined.
    \item There is a derivation of the form:
      \[
        \infer[\simp]{
          \inframe_\tau,
          \lreveal_{\tauo},
          \futr{\tau}{\taut}
          \;\sim\;
          \inframe_\utau,
          \rreveal_{\tauo},
          \futr{\utau}{\utaut}
        }{
          \inframe_\tau, \lreveal_{\tauo}
          \sim
          \inframe_\utau,\rreveal_{\tauo}
        }
      \]
    \end{itemize}

  \item \customlabel{der3}{\lpder{3}} If $\ai = \cuai_\ID(j,1)$. For every $\taut = \_,\cnai(j_1,0)$, $\tautt = \_,\cuai_\ID(j,0)$ such that $\tautt \potau \taut$:
    \begin{center}
      \begin{tikzpicture}
        [dn/.style={inner sep=0.2em,fill=black,shape=circle},
        sdn/.style={inner sep=0.15em,fill=white,draw,solid,shape=circle},
        sl/.style={decorate,decoration={snake,amplitude=1.6}},
        dl/.style={dashed},
        pin distance=0.5em,
        every pin edge/.style={thin}]

        \draw[thick] (0,0)
        node[left=1.3em] {$\tau:$}
        -- ++(0.5,0)
        node[dn,pin={above:{$\cuai_\ID(j,0)$}}] {}
        node[below=0.3em]{$\tautt$}
        -- ++(2.5,0)
        node[dn,pin={above:{$\cnai(j_1,0)$}}] {}
        node[below=0.3em]{$\taut$}
        -- ++(2.5,0)
        node[dn,pin={above:{$\cuai_\ID(j,1)$}}] {}
        node[below=0.3em]{$\tau$};
      \end{tikzpicture}
    \end{center}
    \begin{itemize}
    \item We have $\utautt = \_,\cuai_{\nu_\tau(\ID)}(j,0)$, $\utaut = \_,\cuai_{\nu_\tau(\ID)}(j,1)$ and $\utautt \poutau \utaut \poutau \utau$. Therefore, $\ptr{\utautt,\utau}{\utaut}$ is well-defined.
    \item There is a derivation of the form:
      \[
        \infer[\simp]{
          \inframe_\tau,
          \lreveal_{\tauo},
          \ptr{\tautt,\tau}{\taut}
          \;\sim\;\
          \inframe_\utau,
          \rreveal_{\tauo},
          \ptr{\utautt,\utau}{\utaut}
        }{
          \inframe_\tau, \lreveal_{\tauo}
          \sim
          \inframe_\utau,\rreveal_{\tauo}
        }
      \]
    \end{itemize}

  \item \customlabel{der4}{\lpder{4}}  If $\ai = \cnai(j,1)$. For every $\ID \in \iddom$, $\taui = \_,\cuai_\ID(j_i,1)$, $\taut = \_,\cnai(j,0)$, $\tautt = \_,\cuai_\ID(j_i,0)$ such that $\tautt \potau \taut \potau \taui$:
    \begin{center}
      \begin{tikzpicture}
        [dn/.style={inner sep=0.2em,fill=black,shape=circle},
        sdn/.style={inner sep=0.15em,fill=white,draw,solid,shape=circle},
        sl/.style={decorate,decoration={snake,amplitude=1.6}},
        dl/.style={dashed},
        pin distance=0.5em,
        every pin edge/.style={thin}]

        \draw[thick] (0,0)
        node[left=1.3em] {$\tau:$}
        -- ++(0.5,0)
        node[dn,pin={above:{$\cuai_\ID(j_i,0)$}}] {}
        node[below=0.3em]{$\tautt$}
        -- ++(2.5,0)
        node[dn,pin={above:{$\cnai(j,0)$}}] {}
        node[below=0.3em]{$\taut$}
        -- ++(2.5,0)
        node[dn,pin={above:{$\cuai_\ID(j_i,1)$}}] {}
        node[below=0.3em]{$\taui$}
        -- ++(2.5,0)
        node[dn,pin={above:{$\cnai(j,1) $}}] {}
        node[below=0.3em]{$\tau$};
      \end{tikzpicture}
    \end{center}
    \begin{itemize}
    \item We have $\utautt = \_,\cuai_{\nu_\taut(\ID)}(j_i,0)$, $\utaui = \_,\cuai_{\nu_\taut(\ID)}(j_i,1)$ and $\utautt \poutau \utaut \poutau \utaui \poutau \utau$. Therefore, $\ftr{\utautt,\utaui}{\utaut,\utau}$ is well-defined.
    \item There is a derivation of the form:
      \[
        \infer[\simp]{
          \inframe_\tau,
          \lreveal_{\tauo},
          \ftr{\tautt,\taui}{\taut,\tau}
          \sim\;\;
          \inframe_\utau,
          \rreveal_{\tauo},
          \ftr{\utautt,\utaui}{\utaut,\utau}
        }{
          \inframe_\tau, \lreveal_{\tauo}
          \sim
          \inframe_\utau,\rreveal_{\tauo}
        }
      \]
    \end{itemize}

  \end{itemize}
\end{proposition}

The proof is given in Section~\ref{sec:derivations}

\begin{lemma}
  \label{lem:unlink}
  For all valid basic symbolic trace $\tau$ with at most $C$ actions $\newsession$, there exists a derivation of:
  \[
    \cframe_\tau,\lreveal_{\tau}\sim \cframe_{\ufresh{\tau}},\rreveal_{\tau}
  \]
\end{lemma}
The proof is given in Section~\ref{sec:unlink-proof}.

Using this lemma, we can prove our main theorem, which we recall below:
\begin{theorem*}[Theorem~\ref{thm:main-unlink}]
  The $\fiveaka$ protocol is $\sigma_\sunlink$-unlinkable for an arbitrary number of agents and sessions when the asymmetric encryption $\enc{\_}{\_}{\_}$ is \textsc{ind-cca1} secure and $\owsym$ and $\rowsym$ (resp. $\macsym^1$--$\,\macsym^5$) satisfy jointly the $\prfass$ assumption.
\end{theorem*}

\begin{proof}
  Using Proposition~\ref{prop:main-unlink-bc-app}, we only need to show that for every $\tau \in \support(\runlink)$, there is a derivation of $\cframe_{\tau} \sim \cframe_{\utau}$ using $\axioms$. Moreover, using Assumption~\ref{ass:support-valid} we know that for every $\tau \in \support(\runlink)$, $\tau$ is a valid symbolic trace. Therefore, it is sufficient to prove that for every valid symbolic trace $\tau$, we have a derivation using $\axioms$ of $\cframe_{\tau}  \sim \cframe_{\utau}$. Using Lemma~\ref{lem:unlink}, we know that we have a derivation of $\cframe_\tau,\lreveal_{\tau}\sim \cframe_{\ufresh{\tau}},\rreveal_{\tau}$. We conclude using the $\restr$ rule:
  \[
    \infer[\restr]{
      \cframe_\tau\sim \cframe_{\ufresh{\tau}}
    }{
      \cframe_\tau,\lreveal_{\tau}
      \sim
      \cframe_{\ufresh{\tau}},\rreveal_{\tau}
    }
    \qedhere
  \]
\end{proof}

\input{derivation-proof}

\input{proof}


%% file: derivation-proof.tex
\newpage
\section{Proof of Proposition~\ref{prop:derivation}}
\label{sec:derivations}
\FloatBarrier
\subsection*{Proof of \ref{der1}}
We have two cases:
\begin{itemize}
\item either there exists $l$ such that $\newsession_\ID(l) \popre \tau$ and $\newsession_\ID(l) \not\potau \newsession_\ID(\_)$. In that case we have $\newsession_\ID(l) \potau \taut$.
\item or for every $i$, $\newsession_\ID(i) \not\potau \taut$.
\end{itemize}
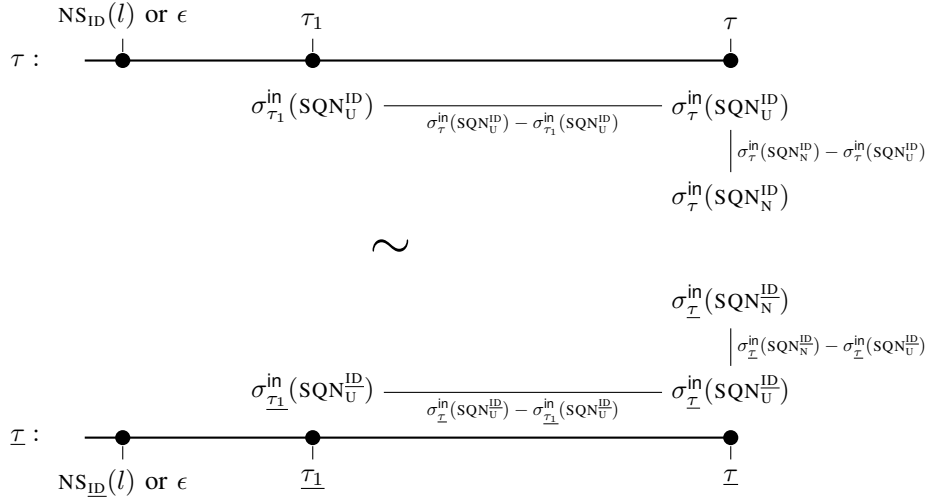
\begin{figure}[t]
  \begin{center}
    \begin{tikzpicture}
      [dn/.style={inner sep=0.2em,fill=black,shape=circle},
      sdn/.style={inner sep=0.15em,fill=white,draw,solid,shape=circle},
      sl/.style={decorate,decoration={snake,amplitude=1.6}},
      dl/.style={dashed},
      pin distance=0.5em,
      every pin edge/.style={thin}]

      \draw[thick] (0,0)
      node[left=1.3em] {$\tau:$}
      -- ++(0.5,0)
      node[dn,pin={above:{$\newsession_\ID(l)$ or $\epsilon$}}] 
      (a) {}
      -- ++(2.5,0)
      node[dn,pin={above:{$\taut$}}] 
      (b) {}
      -- ++(5.5,0)
      node[dn,pin={above:{$\tau$}}] 
      (c) {};

      \path (b) -- ++ (0,-0.6)
      node (b1) {$\instate_\taut(\sqn_\ue^\ID)$};

      \path (c) -- ++ (0,-0.6)
      node (c1) {$\instate_\tau(\sqn_\ue^\ID)$}
      -- ++ (0,-1.2)
      node (c2) {$\instate_\tau(\sqn_\hn^\ID)$};
      
      \draw (c1) -- (c2) node[midway,right,transform shape, scale=0.7]
      {$\instate_\tau(\sqn_\hn^\ID) - \instate_\tau(\sqn_\ue^\ID)$};

      \draw (b1) -- (c1) node[midway,below,transform shape, scale=0.7]
      {$\instate_\tau(\sqn_\ue^\ID) - \instate_\taut(\sqn_\ue^\ID)$};

      \draw[thick] (0,-5)
      node[left=1.3em] {$\utau:$}
      -- ++(0.5,0)
      node[dn,pin={below:{$\newsession_{\uID}(l)$ or $\epsilon$}}] 
      (a) {}
      -- ++(2.5,0)
      node[dn,pin={below:{$\utaut$}}] 
      (bb) {}
      -- ++(5.5,0)
      node[dn,pin={below:{$\utau$}}] 
      (c) {};

      \path (bb) -- ++ (0,0.6)
      node (b1) {$\instate_\utaut(\sqn_\ue^\uID)$};

      \path (c) -- ++ (0,0.6)
      node (c1) {$\instate_\utau(\sqn_\ue^\uID)$}
      -- ++ (0,1.2)
      node (c2) {$\instate_\utau(\sqn_\hn^\uID)$};
      
      \draw (c1) -- (c2) node[midway,right,transform shape, scale=0.7]
      {$\instate_\utau(\sqn_\hn^\uID) - \instate_\utau(\sqn_\ue^\uID)$};

      \draw (b1) -- (c1) node[midway,below,transform shape, scale=0.7]
      {$\instate_\utau(\sqn_\ue^\uID) - \instate_\utaut(\sqn_\ue^\uID)$};

      \path (b) -- (bb) node[midway,right=0.5cm,transform shape,scale=2]
      {$\mathbf{\sim}$};
    \end{tikzpicture}
  \end{center}

  \caption{First Graphical Representation for the Proof of \ref{der1}}
  \label{fig:adksjdasldkjas}
\end{figure}
Let $\uID = \nu_\tau(\ID)$. We summarize the situation graphically in Fig.~\ref{fig:adksjdasldkjas}. In both case, for every $\taut \popreleq \tau' \popre \tau$ we have:
\begin{gather*}
  \left(
    \cstate_{\tau'}(\sqn_\ue^\ID) - \instate_{\tau'}(\sqn_\ue^\ID),
    \cstate_{\utau'}(\sqn_\ue^{\uID})
    - \instate_{\utau'}(\sqn_\ue^{\uID})
  \right) \in
  \reveal_{\tauo}\\
  \left(
    \cond{\instate_\tau(\sync_\ue^\ID)}
    {\left(\instate_\tau(\sqn_\hn^\ID) -
        \instate_\tau(\sqn_\ue^\ID)\right)},
    \cond{\instate_\utau(\sync_\ue^{\uID})}
    {\left(\instate_\utau(\sqn_\hn^{\uID}) -
        \instate_\utau(\sqn_\ue^{\uID})\right)}
  \right) \in
  \reveal_{\tauo}
\end{gather*}
We know that:
\begin{gather*}
  \instate_\tau(\sqn_\ue^\ID) - \instate_\taut(\sqn_\ue^\ID)
  \;=\;
  \cstate_\tauo(\sqn_\ue^\ID) - \instate_\taut(\sqn_\ue^\ID)
  \;=\;
  \sum_{\taut \popreleq \tau'}
  \cstate_{\tau'}(\sqn_\ue^\ID) - \instate_{\tau'}(\sqn_\ue^\ID)
\end{gather*}
And:
\begin{multline*}
  \left(
    \instate_\tau(\sync_\ue^\ID) \wedge 
    \instate_\tau(\sqn_\hn^\ID) <
    \instate_\taut(\sqn_\ue^\ID)
  \right)
  \;\lra\;\\
  \instate_\tau(\sync_\ue^\ID) \wedge 
  \left(
    \left(\instate_\tau(\sqn_\ue^\ID) - \instate_\taut(\sqn_\ue^\ID)\right) +
    \cond{\instate_\tau(\sync_\ue^\ID)}
    {\left(\instate_\tau(\sqn_\hn^\ID) - \instate_\tau(\sqn_\ue^\ID)\right)}
    < \zero\right)
\end{multline*}
Similarly:
\begin{multline*}
  \instate_\utau(\sqn_\ue^{\uID}) -
  \instate_\utaut(\sqn_\ue^{\uID})
  =
  \cstate_\utauo(\sqn_\ue^{\uID}) -
  \instate_\utaut(\sqn_\ue^{\uID})
  =
  \sum_{\utaut\, \popreleq\, \utau'\, \popreleq\, \utauo}
  \cstate_{\utau'}(\sqn_\ue^{\uID}) -
  \instate_{\utau'}(\sqn_\ue^{\uID})\\
  =
  \sum_{\taut \popreleq \tau'}
  \cstate_{\utau'}(\sqn_\ue^{\uID}) -
  \instate_{\utau'}(\sqn_\ue^{\uID})
\end{multline*}
And:
\begin{multline*}
  \left(
    \instate_\utau(\sync_\ue^\uID) \wedge 
    \instate_\utau(\sqn_\hn^\uID) <
    \instate_\utaut(\sqn_\ue^\uID)
  \right)
  \;\lra\;\\
  \instate_\utau(\sync_\ue^\uID) \wedge 
  \left(
    \left(
      \left(\instate_\utau(\sqn_\ue^\uID) - \instate_\utaut(\sqn_\ue^\uID)\right) +
      \cond{\instate_\utau(\sync_\ue^\uID)}
      {\left(\instate_\utau(\sqn_\hn^\uID) - \instate_\utau(\sqn_\ue^\uID)\right)}
    \right) < \zero
  \right)
\end{multline*}
Putting everything together, we get the following derivation:
\[
  \infer[\simp]{
    \begin{alignedat}{2}
      &&&\lreveal_{\tauo},
      \instate_\taut(\sync_\ue^\ID) \wedge 
      \instate_\tau(\sqn_\hn^\ID) <
      \instate_\taut(\sqn_\ue^\ID)\\
      &\sim\;\;&&
      \rreveal_{\tauo},
      \instate_\utaut(\sync_\ue^{\uID}) \wedge 
      \instate_\utau(\sqn_\hn^{\uID}) <
      \instate_\utaut(\sqn_\ue^{\uID})
    \end{alignedat}
  }{
    \infer[\dup^*]{
      \begin{alignedat}{2}
        &&&\lreveal_{\tauo},
        \instate_\tau(\sync_\ue^\ID),
        \cond{\instate_\tau(\sync_\ue^\ID)}
        {\left(\instate_\tau(\sqn_\hn^\ID) - \instate_\tau(\sqn_\ue^\ID)\right)},
        \left(
          \cstate_{\tau'}(\sqn_\ue^\ID) - \instate_{\tau'}(\sqn_\ue^\ID),
        \right)_{\taut \popreleq \tau'}
        \\
        &\sim&&
        \rreveal_{\tauo},
        \instate_\utau(\sync_\ue^\uID),
        \cond{\instate_\utau(\sync_\ue^\uID)}
        {\left(\instate_\utau(\sqn_\hn^\uID) - \instate_\utau(\sqn_\ue^\uID)\right)},
        \left(
          \cstate_{\utau'}(\sqn_\ue^{\uID}) - 
          \instate_{\utau'}(\sqn_\ue^{\uID}),
        \right)_{\taut \popreleq \tau'}          
      \end{alignedat}
    }{
      \lreveal_{\tauo}
      \;\sim\;
      \rreveal_{\tauo}
    }
  }
\]
The derivation of \eqref{eq:deriv:n-u} is very similar. We omit the details, and only give the graphical representation of the situation in Fig.~\ref{fig:dsaldasdklasdjmlaskd}.

\subsection*{Proof of \ref{der3}}
\begin{center}
  \begin{tikzpicture}
    [dn/.style={inner sep=0.2em,fill=black,shape=circle},
    sdn/.style={inner sep=0.15em,fill=white,draw,solid,shape=circle},
    sl/.style={decorate,decoration={snake,amplitude=1.6}},
    dl/.style={dashed},
    pin distance=0.5em,
    every pin edge/.style={thin}]

    \draw[thick] (0,0)
    node[left=1.3em] {$\tau:$}
    -- ++(0.5,0)
    node[dn,pin={above:{$\cuai_\ID(j,0)$}}] {}
    node[below=0.3em]{$\tautt$}
    -- ++(2.5,0)
    node[dn,pin={above:{$\cnai(j_1,0)$}}] {}
    node[below=0.3em]{$\taut$}
    -- ++(2.5,0)
    node[dn,pin={above:{$\cuai_\ID(j,1)$}}] {}
    node[below=0.3em]{$\tau$};
  \end{tikzpicture}
\end{center}

\noindent Recall that:
\begin{alignat*}{2}
  \ptr{\tautt,\tau}{\taut}
  &\;\equiv\;\;&&
  \left(
    \begin{alignedat}[c]{2}
      &&&
      \pi_1(g(\inframe_\tau)) = \nonce^{j_1}
      \;\wedge\;
      \pi_2(g(\inframe_\tau)) =
      \dotuline{\instate_\taut(\sqn_\hn^\ID)
        \oplus \ow{\nonce^{j_1}}{\key^\ID}} \\
      &\wedge\;\;&&
      \pi_3(g(\inframe_\tau)) =
      \dotuline{\mac{\striplet
          {\nonce^{j_1}}
          {\instate_\taut(\sqn_\hn^\ID)}
          {\instate_\tautt(\suci_\ue^\ID)}}
        {\mkey^\ID}{3}}\\
      &\wedge\;\;&&
      g(\inframe_\taut) = \dashuline{\instate_\tautt(\suci_\ue^\ID)}
      \;\wedge\;
      \uwave{\instate_\tautt(\suci_\ue^\ID)
        = \instate_\taut(\suci_\hn^\ID)}
      \;\wedge\;
      \dashuline{\instate_\tautt(\success_\ue^\ID)}\\
      &\wedge\;\;&&
      \dashuline{\range{\instate_\tau(\sqn_\ue^\ID)}
        {\instate_\taut(\sqn_\hn^\ID)}}
    \end{alignedat}
  \right)
  \numberthis\label{eq:sdvnbiogfjdpiafjpof}
\end{alignat*}
Since $\tau$ is valid, we know that for every $\tau'$, if $\tautt \potau \tau'$ then $\tau' \ne \newsession_\ID(\_)$. It follows that $\utautt = \_,\cuai_{\nu_\tau(\ID)}(j,0)$ and $\utau = \_,\cuai_{\nu_\tau(\ID)}(j,1)$. The fact that $\utautt \poutau \utaut$ is then straightforward. Letting $\uID = \nu_\tau(\ID)$, we can then check that:
\begin{alignat*}{2}
  \ptr{\utautt,\utau}{\utaut}
  &\;\equiv\;\;&&
  \left(
    \begin{alignedat}{2}
      &&&
      \pi_1(g(\inframe_\utau)) = \nonce^{j_1}
      \;\wedge\;
      \pi_2(g(\inframe_\utau)) =
      \dotuline{\instate_\utaut(\sqn_\hn^{\uID})
        \oplus \ow{\nonce^{j_1}}{\key^{\uID}}} \\
      &\wedge\;\;&&
      \pi_3(g(\inframe_\utau)) =
      \dotuline{
        \mac{\striplet
          {\nonce^{j_1}}
          {\instate_\utaut(\sqn_\hn^{\uID})}
          {\instate_\utautt(\suci_\ue^{\uID})}}
        {\mkey^{\uID}}{3}}\\
      &\wedge\;\;&&
      g(\inframe_\utaut) = \dashuline{\instate_\utautt(\suci_\ue^{\uID})}
      \;\wedge\;
      \uwave{\instate_\utautt(\suci_\ue^{\uID}) =
        \instate_\utaut(\suci_\hn^{\uID})}
      \;\wedge\;
      \dashuline{\instate_\utautt(\success_\ue^{\uID})}\\
      &\wedge\;\;&&
      \dashuline{\range{\instate_\utau(\sqn_\ue^{\uID})}
        {\instate_\utaut(\sqn_\hn^{\uID})}}
    \end{alignedat}
  \right)
  \numberthis\label{eq:sdijvp9eqrqurqptuafjpda}
\end{alignat*}
We have two cases.
\paragraph{Case 1}
Assume that for all $\tau' \potau \taut$ such that $\tau' \not \potau \newsession_\ID(\_)$ we have $\tau' \ne \_, \fuai_\ID(\_)$.

Then we know that for all $\utau' <_\utau \utaut$ such that $\utau' \not <_\utau \newsession_{\nu_\tau(\ID)}(\_)$ we have $\utau' \ne \_, \fuai_{\nu_\tau(\ID)}(\_)$. Therefore using $\ref{b7}$ twice we get:
\begin{mathpar}
  \ptr{\tautt,\tau}{\taut} \;\ra\; \false

  \ptr{\utautt,\utau}{\utaut} \;\ra\;\false
\end{mathpar}

Therefore we have a trivial derivation:
\[
  \begin{gathered}[c]
    \infer[R]{
      \inframe_\tau,
      \lreveal_{\tauo},
      \ptr{\tautt,\tau}{\taut}
      \;\sim\;
      \inframe_\utau,
      \rreveal_{\tauo},
      \ptr{\utautt,\utau}{\utaut}
    }{
      \infer[\fa]{
        \inframe_\tau, \lreveal_{\tauo},\false
        \sim
        \inframe_\utau,\rreveal_{\tauo},\false
      }{
        \inframe_\tau, \lreveal_{\tauo}
        \sim
        \inframe_\utau,\rreveal_{\tauo}
      }
    }
  \end{gathered}
  \numberthis\label{eq:dsuovhsoghhifgs}
\]

\paragraph{Case 2}
Assume that there exists $\tauttt = \_,\fuai_\ID(j_0)$ such that $\tauttt \potau \taut$, $\tauttt \not \potau \newsession_\ID(\_)$ and $\tauttt \not \potau \fuai_\ID(\_)$. Then $\utauttt = \_,\fuai_{\nu_\tau(\ID)}(\_)$, $\utauttt <_\utau \utaut$, $\utauttt \not <_\utau \newsession_{\nu_\tau(\ID)}(\_)$ and $\utauttt \not <_\utau \fuai_{\nu_\tau(\ID)}(\_)$.

First, we show that $j_0 \ne j$: assume that $j_0 = j$, then we know that $\tau \potau \tauttt$, which is absurd. Therefore $j_0 \ne j$. Using the validity of $\tau$, we know that $\tauttt$ cannot occur between $\tautt = \_,\cuai_\ID(j,0)$ and $\tau = \_,\cuai_\ID(j,0)$. Hence $\tauttt \potau \tautt$.

Let $\tau_\ns$ be the latest $\ns_\ID(\_)$, if it exists, or $\epsilon$ otherwise: $\tau_\ns = \_,\ns_\ID(\_)$ or $\epsilon$ and $\tau_\ns \not \potau \ns_\ID(\_)$. Let $\taux$ be $\_,\cuai_\ID(j_0,0)$ or $\_,\npuai{1}{\ID}{j_0}$ be the beginning of the $\ue$ session associated to $\tauttt$. We know that $\tau_\ns \potau \taux \potau \tauttt$.

We know that $\ptr{\tautt,\tau}{\taut} \ra \instate_\tautt(\success_\ue^\ID)$. As $\tauttt \not \potau \fuai_\ID(\_)$, we know that there are no $\fuai_\ID(\_)$ action between $\tauttt$ and $\tautt$. If there exists a action by user $\ID$ between $\tauttt$ and $\tautt$, then we have either $\tauttt \potau \npuai{1}{\ID}{\_} \potau \tautt$ or $\tauttt \potau \cuai_\ID(\_,0) \potau \tautt$. In both case, $\success_\ue^\ID$ is set to $\false$, and cannot be set back to something else without a $\fuai_\ID(\_)$ action. It follows that if there exists a user action between $\tauttt$ and $\tautt$ then $\neg \instate_\tautt(\success_\ue^\ID)$. Using the same reasoning we have $\neg \instate_\utautt(\success_\ue^\uID)$ if there exists a user action between $\tauttt$ and $\tautt$. Hence in that case the derivation \eqref{eq:dsuovhsoghhifgs} works.

By consequence we now assume that:
\[
  \left\{ \_,\cuai_\ID(\_),\_,\npuai{\_}{\ID}{\_}1,\fuai_\ID(\_) \right\}
  \cap \left\{ \tau' \mid \tauttt \potau \tau' \potau \tautt \right\} = \emptyset
  \numberthis\label{eq:dsvuovsibvsibw}
\]
It follows that $\neg \accept_\tauttt^\ID \ra \neg \instate_\tautt(\success_\ue^\ID)$, hence $\ptr{\tautt,\tau}{\taut} \ra \accept_\tauttt^\ID$. We also deduce from \eqref{eq:dsvuovsibvsibw} that $\cstate_\tauttt(\suci_\ue^\ID) \equiv \instate_\tautt(\suci_\ue^\ID)$. Applying \ref{sequ1}, we know that:
\begin{alignat*}{3}
  \accept_\tauttt^\ID
  &\;\;\leftrightarrow\;\;&
  \bigvee_{\taux \potau \taua \,=\, \_,\fnai(j_a) \potau \tauttt}
  \futr{\tauttt}{\taua}
\end{alignat*}
Therefore:
\[
  \ptr{\tautt,\tau}{\taut} \;\lra\;
  \bigvee_{\taux \potau \taua \,=\, \_,\fnai(j_a) \potau \tauttt}
  \futr{\tauttt}{\taua} \wedge \ptr{\tautt,\tau}{\taut}
\]
Similarly, we show that $\cstate_\utauttt(\suci_\ue^\uID) \equiv \instate_\utautt(\suci_\ue^\uID)$ and that:
\[
  \ptr{\utautt,\utau}{\utaut} \;\lra\;
  \bigvee_{\taux \potau \taua \,=\, \_,\fnai(j_a) \potau \tauttt}
  \futr{\utauttt}{\utaua} \wedge \ptr{\utautt,\utau}{\utaut}
\]
We can start building the wanted derivation:
\[
  \infer[R]{
    \inframe_\tau,
    \lreveal_{\tauo},
    \ptr{\tautt,\tau}{\taut}
    \;\sim\;
    \inframe_\utau,
    \rreveal_{\tauo},
    \ptr{\utautt,\utau}{\utaut}
  }{
    \infer[\fa^*]{
      \begin{alignedat}{2}
        &&&\inframe_\tau, \lreveal_{\tauo},
        \bigvee_{\taux \potau \taua \,=\, \_,\fnai(j_a) \potau \tauttt}
        \futr{\tauttt}{\taua} \wedge \ptr{\tautt,\tau}{\taut}\\
        &\sim\;\;&&
        \inframe_\utau,\rreveal_{\tauo},
        \bigvee_{\taux \potau \taua \,=\, \_,\fnai(j_a) \potau \tauttt}
        \futr{\utauttt}{\utaua} \wedge \ptr{\utautt,\utau}{\utaut}
      \end{alignedat}
    }{
      \begin{alignedat}{2}
        &&&\inframe_\tau, \lreveal_{\tauo},
        \left(
          \futr{\tauttt}{\taua} \wedge \ptr{\tautt,\tau}{\taut}
        \right)_{\taux \potau \taua \,=\, \_,\fnai(j_a) \potau \tauttt}\\
        &\sim\;\;&&
        \inframe_\utau,\rreveal_{\tauo},
        \left(
          \futr{\utauttt}{\utaua} \wedge \ptr{\utautt,\utau}{\utaut}
        \right)_{\taux \potau \taua \,=\, \_,\fnai(j_a) \potau \tauttt}
      \end{alignedat}
    }
  }
\]
\begin{figure}[t]
\begin{center}
  \begin{tikzpicture}
    [dn/.style={inner sep=0.2em,fill=black,shape=circle},
    sdn/.style={inner sep=0.15em,fill=white,draw,solid,shape=circle},
    sl/.style={decorate,decoration={snake,amplitude=1.6}},
    dl/.style={dashed},
    pin distance=0.5em,
    every pin edge/.style={thin}]

    \draw[thick] (0,0)
    node[left=1.3em] {$\tau:$}
    -- ++(0.5,0)
    node[dn,pin={above:{$\newsession_\ID(l)$ or $\epsilon$}}] 
    (a) {}
    -- ++(2.5,0)
    node[dn,pin={above:{$\taut$}}] 
    (b) {}
    -- ++(5.5,0)
    node[dn,pin={above:{$\tau$}}] 
    (c) {};

    \path (b) -- ++ (0,-0.6)
    node (b1) {$\instate_\taut(\sqn_\ue^\ID)$}
    -- ++ (0,-1.2)
    node (b2) {$\instate_\taut(\sqn_\hn^\ID)$};

    \path (c) -- ++ (0,-0.6)
    node (c1) {$\instate_\tau(\sqn_\ue^\ID)$};
    
    \draw (b1) -- (b2) node[midway,left,transform shape, scale=0.7]
    {$\instate_\taut(\sqn_\hn^\ID) - \instate_\taut(\sqn_\ue^\ID)$};

    \draw (b1) -- (c1) node[midway,below,transform shape, scale=0.7]
    {$\instate_\tau(\sqn_\ue^\ID) - \instate_\taut(\sqn_\ue^\ID)$};

    \draw[thick] (0,-5)
    node[left=1.3em] {$\utau:$}
    -- ++(0.5,0)
    node[dn,pin={below:{$\newsession_{\uID}(l)$ or $\epsilon$}}] 
    (a) {}
    -- ++(2.5,0)
    node[dn,pin={below:{$\utaut$}}] 
    (bb) {}
    -- ++(5.5,0)
    node[dn,pin={below:{$\utau$}}] 
    (c) {};

    \path (bb) -- ++ (0,0.6)
    node (b1) {$\instate_\utaut(\sqn_\ue^\uID)$}
    -- ++ (0,1.2)
    node (b2) {$\instate_\utaut(\sqn_\hn^\uID)$};

    \path (c) -- ++ (0,0.6)
    node (c1) {$\instate_\utau(\sqn_\ue^\uID)$};
    
    \draw (b1) -- (b2) node[midway,left,transform shape, scale=0.7]
    {$\instate_\utaut(\sqn_\hn^\uID) - \instate_\utaut(\sqn_\ue^\uID)$};

    \draw (b1) -- (c1) node[midway,below,transform shape, scale=0.7]
    {$\instate_\utau(\sqn_\ue^\uID) - \instate_\utaut(\sqn_\ue^\uID)$};

    \path (b) -- (bb) node[midway,right=0.5cm,transform shape,scale=2]
    {$\mathbf{\sim}$};
  \end{tikzpicture}
\end{center}

\caption{Second Graphical Representation for the Proof of \ref{der1}}
\label{fig:dsaldasdklasdjmlaskd}
\end{figure}
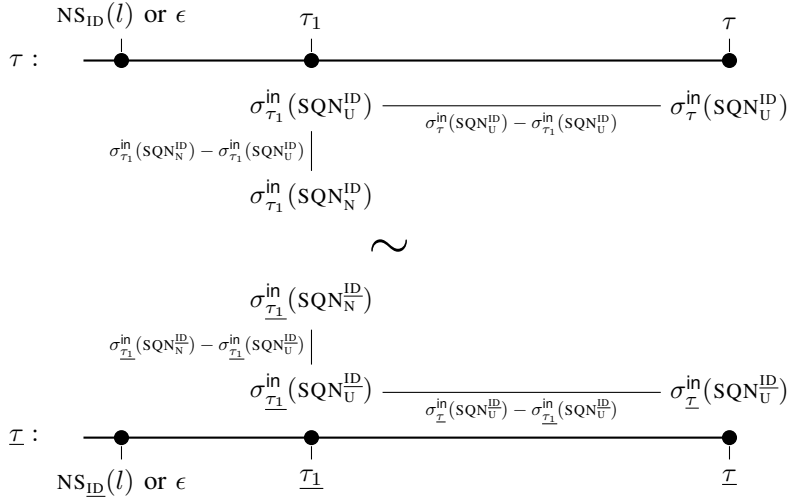
Let $\taua = \_,\fnai(j_a)$ be such that $\taux \potau \taua \potau \tauttt$. Let $\taub$ be $\_,\cnai(j_a,1)$ or $\_,\pnai(j_a,1)$ such that $\taub \potau \taua$. To conclude, we just need to build a derivation of:
\[
  \inframe_\tau, \lreveal_{\tauo},
  \futr{\tauttt}{\taua} \wedge \ptr{\tautt,\tau}{\taut}
  \;\sim\;\;
  \inframe_\utau,\rreveal_{\tauo},
  \futr{\utauttt}{\utaua} \wedge \ptr{\utautt,\utau}{\utaut}
\]
The proof consist in rewriting $\futr{\tauttt}{\taua} \wedge \ptr{\tautt,\tau}{\taut}$ and $\futr{\utauttt}{\utaua} \wedge \ptr{\utautt,\utau}{\utaut}$ such that they can be decomposed (using $\fa$) into corresponding parts appearing in $\reveal_{\tauo}$. We do this piece by piece: the  waved underlined part first, the dotted underlined and the dashed underlined part. We represent graphically the protocols executions below:
\begin{center}
  \begin{tikzpicture}
    [dn/.style={inner sep=0.2em,fill=black,shape=circle},
    sdn/.style={inner sep=0.15em,fill=white,draw,solid,shape=circle},
    sl/.style={decorate,decoration={snake,amplitude=1.6}},
    dl/.style={dashed},
    pin distance=0.5em,
    every pin edge/.style={thin}]

    \draw[thick] (0,0)
    -- ++(0.5,0)
    node[dn,pin={above,align=left:{$\ns_\ID(\_)$\\or $\epsilon$}}] {}
    node[below=0.3em,name=a]{$\tau_\ns$}
    -- ++(2,0)
    node[dn,pin={above,align=left:{$\cuai_\ID(j_0,0)$\\or $\npuai{1}{\ID}{j_0}$}}] {}
    node[below=0.3em,name=b]{$\taux$}
    -- ++(2.5,0)
    node[dn,pin={above,align=left:{$\cnai(j_a,1)$\\ or $\pnai(j_a,1)$}}] {}
    node[below=0.3em,name=c0]{$\taub$}
    -- ++(2.5,0)
    node[dn,pin={above:{$\fnai(j_a)$}}] {}
    node[below=0.3em,name=c]{$\taua$}
    -- ++(2,0)
    node[dn,pin={above:{$\fuai_\ID(j_0)$}}] {}
    node[below=0.3em,name=d]{$\tauttt$}
    -- ++(2,0)
    node[dn,pin={above:{$\cuai_\ID(j,0)$}}] {}
    node[below=0.3em,name=e]{$\tautt$}
    -- ++(2,0)
    node[dn,pin={above:{$\cnai(j_1,0)$}}] {}
    node[below=0.3em,name=f]{$\taut$}
    -- ++(2,0)
    node[dn,pin={above:{$ \cuai_\ID(j,1) $}}] {}
    node[below=0.3em,name=g]{$\tau$};

    \draw[thin,dashed] (b) -- ++(0,-0.5) -| (c0)
    {[draw=none] -- ++(0,-0.5)} -| (c)
    {[draw=none] -- ++(0,-0.5)} -| (d);
    \draw[thin,densely dotted] (e) -- ++(0,-0.5) -| (f)
    {[draw=none] -- ++(0,-0.5)} -| (g);
  \end{tikzpicture}
\end{center}

\paragraph{Part 1 (Waves)}
We are going to give a derivation of:
\[
  \inframe_\tau,
  \lreveal_{\tauo},
  \futr{\tauttt}{\taua}\wedge
  \instate_\tautt(\suci_\ue^\ID) =
  \instate_\taut(\suci_\hn^\ID)
  \;\sim\;
  \inframe_\utau,
  \rreveal_{\tauo},
  \futr{\utauttt}{\utaua}\wedge
  \instate_\utautt(\suci_\ue^\uID) =
  \instate_\utaut(\suci_\hn^\uID)
\]
Recall that $\cstate_\tauttt(\suci_\ue^\ID) \equiv \instate_\tautt(\suci_\ue^\ID)$ and $\cstate_\utauttt(\suci_\ue^\uID) \equiv \instate_\utautt(\suci_\ue^\uID)$. Therefore it is sufficient to give a derivation of:
\[
  \inframe_\tau,
  \lreveal_{\tauo},
  \futr{\tauttt}{\taua} \wedge
  \cstate_\tauttt(\suci_\ue^\ID) =
  \instate_\taut(\suci_\hn^\ID)
  \;\sim\;
  \inframe_\utau,
  \rreveal_{\tauo},
  \futr{\utauttt}{\utaua} \wedge
  \cstate_\utauttt(\suci_\ue^\uID) =
  \instate_\utaut(\suci_\hn^\uID)
\]
We know that:
\[
  \cond{\futr{\tauttt}{\taua}}{\cstate_\tauttt(\suci_\ue^\ID)} =
  \cond{\futr{\tauttt}{\taua}}{\suci^{j_a}}
\]
Hence:
\[
  \left(
    \futr{\tauttt}{\taua} \wedge
    \cstate_\tauttt(\suci_\ue^\ID) =
    \instate_\taut(\suci_\hn^\ID)
  \right)
  \;\lra\;
  \left(
    \futr{\tauttt}{\taua} \wedge
    \instate_\taut(\suci_\hn^\ID) = \suci^{j_a}
  \right)
\]
Intuitively, the only way we can have $\instate_\taut(\suci_\hn^\ID) = \suci^{j_a}$ is:
\begin{itemize}
\item if the $\supi$ or $\suci$ network session $j_a$ accepts with the increasing sequence number condition.
\item and if $\instate_\taut(\suci_\hn^\ID)$ was not over-written between $\taub$ and $\taut$.
\end{itemize}
It is actually straightforward to show by induction that:
\[
  \instate_\taut(\suci_\hn^\ID) \ne \suci^{j_a}
  \;\lra\;
  \left(
    \neg \incaccept_\taub^\ID
    \vee
    \bigvee_{\tau'=\_,\cnai(j',1)
      \atop{\text{or }\tau'=\_,\pnai(j',1)
        \atop{\taub \potau \tau' \potau \taut}}}
    \incaccept_{\tau'}^\ID
    \vee
    \bigvee_{\tau'=\_,\cnai(j',0)\atop{\taub \potau \tau' \potau \taut}}
    \accept_{\tau'}^\ID
  \right)
\]
Hence:
\begin{alignat*}{2}
  &&&
  \futr{\tauttt}{\taua} \wedge
  \cstate_\tauttt(\suci_\ue^\ID) =  \instate_\taut(\suci_\hn^\ID)\\
  &\lra\;\;&&
  \futr{\tauttt}{\taua} \wedge
  \incaccept_\taub^\ID
  \wedge
  \bigwedge_{\tau'=\_,\cnai(j',1)
    \atop{\text{or }\tau'=\_,\pnai(j',1)
      \atop{\taub \potau \tau' \potau \taut}}}
  \neg\incaccept_{\tau'}^\ID
  \wedge
  \bigwedge_{\tau'=\_,\cnai(j',0)\atop{\taub \potau \tau' \potau \taut}}
  \neg\accept_{\tau'}^\ID\\
  &\lra\;\;&&
  \futr{\tauttt}{\taua} \wedge
  \incaccept_\taub^\ID
  \wedge
  \bigwedge_{\tau'=\_,\cnai(j',1)
    \atop{\text{or }\tau'=\_,\pnai(j',1)
      \atop{\taub \potau \tau' \potau \taut}}}
  \neg\incaccept_{\tau'}^\ID
  \wedge
  \bigwedge_{\tau'=\_,\cnai(j',0)\atop{\taub \potau \tau' \potau \taut}}
  g(\inframe_{\tau'}) \ne \suci^{j_a}  
\end{alignat*}
For every $\taun = \_,\cnai(\_,1)$ or $\_,\pnai(\_,1)$, we know that $\sqn_\hn^\ID$ is incremented at $\taun$ if and only if $\incaccept_\taun^\ID$ is true. Therefore:
\begin{alignat*}{2}
  \incaccept_\taun^\ID\;\;
  &\lra\;\;&&
  \instate_\taun(\sqn_\hn^\ID) < \cstate_\taun(\sqn_\hn^\ID)
\end{alignat*}
Using the fact that $\instate_\taun(\sqn_\ue^\ID) = \cstate_\taun(\sqn_\ue^\ID)$, we can rewrite this as:
\begin{alignat*}{2}
  \incaccept_\taun^\ID\;\;
  &\lra\;\;&&
  \instate_\taun(\sqn_\hn^\ID) - \instate_\taun(\sqn_\ue^\ID) <
  \cstate_\taun(\sqn_\hn^\ID) - \cstate_\taun(\sqn_\ue^\ID)
\end{alignat*}
Using this remark we can show that:
\begin{alignat*}{2}
  &&& \futr{\tauttt}{\taua} \wedge
  \cstate_\tauttt(\suci_\ue^\ID) =  \instate_\taut(\suci_\hn^\ID)\\
  &\lra\;\;&&
  \futr{\tauttt}{\taua} \wedge
  \left(
    \begin{alignedat}{2}
      &&&\instate_\taub(\sqn_\hn^\ID) - \instate_\taub(\sqn_\ue^\ID)\\
      &<\;\;&&
      \cstate_\taub(\sqn_\hn^\ID) - \cstate_\taub(\sqn_\ue^\ID)
    \end{alignedat}
  \right)
  \wedge
  \left(
    \begin{alignedat}{2}
      &&&\cstate_\taub(\sqn_\hn^\ID) - \cstate_\taub(\sqn_\ue^\ID)\\
      &=\;\;&&
      \instate_\taut(\sqn_\hn^\ID) - \instate_\taut(\sqn_\ue^\ID)
    \end{alignedat}
  \right)
  \wedge
  \;\;\bigwedge_{
    \mathclap{\tau'=\_,\cnai(j',0)}
    \atop{\mathclap{\taub \potau \tau' \potau \taut}}}\;\;
  g(\inframe_{\tau'}) \ne \suci^{j_a}
  \numberthis\label{eq:vdsuovhwioqirjq}
\end{alignat*}
Doing exactly the same reasoning, we show that:
\begin{alignat*}{2}
  &&&
  \futr{\utauttt}{\utaua} \wedge
  \cstate_\utauttt(\suci_\ue^\uID) =  \instate_\utaut(\suci_\hn^\uID)\\
  &\lra\;\;&&
  \futr{\utauttt}{\utaua} \wedge
  \left(
    \begin{alignedat}{2}
      &&&\instate_\utaub(\sqn_\hn^\uID) - \instate_\utaub(\sqn_\ue^\uID)\\
      &<\;\;&&
      \cstate_\utaub(\sqn_\hn^\uID) - \cstate_\utaub(\sqn_\ue^\uID)
    \end{alignedat}
  \right)
  \wedge
  \left(
    \begin{alignedat}{2}
      &&&\cstate_\utaub(\sqn_\hn^\uID) - \cstate_\utaub(\sqn_\ue^\uID)\\
      &=\;\;&&
      \instate_\utaut(\sqn_\hn^\uID) - \instate_\utaut(\sqn_\ue^\uID)
    \end{alignedat}
  \right)
  \wedge
  \;\;\bigwedge_{
    \mathclap{\tau'=\_,\cnai(j',0)}
    \atop{\mathclap{\taub \potau \tau' \potau \taut}}}\;\;
  g(\inframe_{\utau'}) \ne \suci^{j_a}
  \numberthis\label{eq:qeufohvivspiovad}
\end{alignat*}
We introduce some notation that will be used later: for every symbolic trace $\tau = \tauo,\ai$ and identity $\ID$, we let $\syncdiffin_\tau^\ID \equiv \syncdiff_\tauo^\ID$.

We now split the proof in two, depending on whether $\instate_\taub(\sync_\ue^\ID)$ is true or false. Let $\psi \equiv \futr{\tauttt}{\taua}\wedge \instate_\tautt(\suci_\ue^\ID) = \instate_\taut(\suci_\hn^\ID)$ and $\upsi \equiv \futr{\utauttt}{\utaua}\wedge\instate_\utautt(\suci_\ue^\uID) = \instate_\utaut(\suci_\hn^\uID)$. Using the fact that:
\[
  \left(\instate_\taub(\sync_\ue^\ID),\instate_\utaub(\sync_\ue^\uID)\right) \in \reveal_{\tauo}
\]
We can build the derivation:
\begin{equation*}
  \infer[\simp]{
    \inframe_\tau,
    \lreveal_{\tauo},
    \psi
    \;\sim\;
    \inframe_\utau,
    \rreveal_{\tauo},
    \upsi
  }{
    \infer[\dup]{
      \begin{alignedat}{2}
        &&&\inframe_\tau,
        \lreveal_{\tauo},
        \instate_\taub(\sync_\ue^\ID),
        \instate_\taub(\sync_\ue^\ID) \wedge
        \psi,
        \neg \instate_\taub(\sync_\ue^\ID) \wedge
        \psi\\
        &\sim\;\;&&
        \inframe_\utau,
        \rreveal_{\tauo},
        \instate_\utaub(\sync_\ue^\uID),
        \instate_\utaub(\sync_\ue^\uID)\wedge
        \upsi,
        \neg \instate_\utaub(\sync_\ue^\uID)\wedge
        \upsi
      \end{alignedat}
    }{
      \inframe_\tau,
      \lreveal_{\tauo},
      \instate_\taub(\sync_\ue^\ID) \wedge
      \psi,
      \neg \instate_\taub(\sync_\ue^\ID) \wedge
      \psi
      \;\sim\;
      \inframe_\utau,
      \rreveal_{\tauo},
      \instate_\utaub(\sync_\ue^\uID)\wedge
      \upsi,
      \neg \instate_\utaub(\sync_\ue^\uID)\wedge
      \upsi
    }
  }
\end{equation*}
We now build a derivation of $\inframe_\tau, \lreveal_{\tauo}, \instate_\taub(\sync_\ue^\ID) \wedge \psi$ and one for $\inframe_\tau, \lreveal_{\tauo}, \neg \instate_\taub(\sync_\ue^\ID) \wedge \psi$:
\begin{itemize}
\item Using the fact that we have $\instate_\taub(\sync_\ue^\ID)$ and \eqref{eq:vdsuovhwioqirjq}, we know that:
  \begin{alignat*}{2}
    &&& \instate_\taub(\sync_\ue^\ID) \wedge
    \futr{\tauttt}{\taua} \wedge
    \cstate_\tauttt(\suci_\ue^\ID) =  \instate_\taut(\suci_\hn^\ID)\\
    &\lra\;\;&&
    \instate_\taub(\sync_\ue^\ID) \wedge
    \futr{\tauttt}{\taua} \wedge
    \left(
      \begin{alignedat}{2}
        &&&\syncdiffin_\taub^\ID\\
        &<\;\;&&
        \syncdiff_\taub^\ID
      \end{alignedat}
    \right)
    \wedge
    \left(
      \begin{alignedat}{2}
        &&&\syncdiff_\taub^\ID\\
        &=\;\;&&
        \syncdiffin_\taut^\ID
      \end{alignedat}
    \right)
    \wedge
    \;\;\bigwedge_{
      \mathclap{\tau'=\_,\cnai(j',0)}
      \atop{\mathclap{\taub \potau \tau' \potau \taut}}}\;\;
    g(\inframe_{\tau'}) \ne \suci^{j_a}
  \end{alignat*}
  Similarly, using \eqref{eq:qeufohvivspiovad} we get:
  \begin{alignat*}{2}
    &&&\instate_\utaub(\sync_\ue^\uID) \wedge
    \futr{\utauttt}{\utaua} \wedge
    \cstate_\utauttt(\suci_\ue^\uID) =  \instate_\utaut(\suci_\hn^\uID)\\
    &\lra\;\;&&
    \instate_\utaub(\sync_\ue^\uID) \wedge
    \futr{\utauttt}{\utaua} \wedge
    \left(
      \begin{alignedat}{2}
        &&&\syncdiffin_\utaub^\uID\\
        &<\;\;&&
        \syncdiff_\utaub^\uID
      \end{alignedat}
    \right)
    \wedge
    \left(
      \begin{alignedat}{2}
        &&&\syncdiff_\utaub^\uID\\
        &=\;\;&&
        \syncdiffin_\utaut^\uID
      \end{alignedat}
    \right)
    \wedge
    \;\;\bigwedge_{
      \mathclap{\tau'=\_,\cnai(j',0)}
      \atop{\mathclap{\taub \potau \tau' \potau \taut}}}\;\;
    g(\inframe_{\utau'}) \ne \suci^{j_a}
  \end{alignat*}
  Moreover, we know that:
  \begin{mathpar}
    \left(
      \left(
        \suci^{j_a},\suci^{j_a}
      \right) \in \reveal_{\tauo}
    \right)_{\tau'=\_,\cnai(j',0)
      \atop{\taub \potau \tau' \potau \taut}}
    
    \left(
      \syncdiffin_\taut^\ID,
      \syncdiffin_\utaut^\uID
    \right) \in \reveal_{\tauo}

    \left(
      \syncdiffin_\taub^\ID,
      \syncdiffin_\utaub^\uID
    \right) \in \reveal_{\tauo}

    \left(
      \syncdiff_\taub^\ID,
      \syncdiff_\utaub^\uID
    \right) \in \reveal_{\tauo}

    \left(
      \instate_\taub(\sync_\ue^\ID),
      \instate_\utaub(\sync_\ue^\uID)
    \right) \in \reveal_{\tauo}    
  \end{mathpar}
  And using \ref{der2}, we know that we have a derivation of:
  \[
    \infer[\simp]{
      \inframe_\tau,
      \lreveal_{\tauo},
      \futr{\tauttt}{\taua}
      \sim\;\;
      \inframe_\utau,
      \rreveal_{\tauo},
      \futr{\tauttt}{\taua}
    }{
      \inframe_\tau, \lreveal_{\tauo}
      \sim
      \inframe_\utau,\rreveal_{\tauo}
    }
  \]
  Using this, we can rewrite $\instate_\taub(\sync_\ue^\ID) \wedge \psi$ and $\instate_\utaub(\sync_\ue^\uID)\wedge \upsi$ as two terms that decompose, using $\fa$, into matching part of $\reveal_{\tauo}$. By consequence we can build the following derivation:
  \[
    \begin{gathered}[c]
      \infer[\simp]{
        \inframe_\tau,
        \lreveal_{\tauo},
        \instate_\taub(\sync_\ue^\ID) \wedge
        \psi
        \;\sim\;
        \inframe_\utau,
        \rreveal_{\tauo},
        \instate_\utaub(\sync_\ue^\uID)\wedge
        \upsi
      }{
        \inframe_\tau,
        \lreveal_{\tauo}
        \psi
        \;\sim\;
        \inframe_\utau,
        \rreveal_{\tauo}
      }
    \end{gathered}
    \numberthis\label{eq:vnsduvuiogquigw}
  \]

\item We now focus on the case where we have $\neg \instate_\taub(\sync_\ue^\ID)$.

  First, assume that $\taub = \_,\cnai(j_a,1)$. In that case, we know that $\futr{\tauttt}{\taua} \ra \accept_\taub^\ID$. Since $\accept_\taub^\ID \ra \instate_\taub(\sync_\ue^\ID)$, we get that $(\neg \instate_\taub(\sync_\ue^\ID) \wedge \psi) \lra \false$. Similarly we have $(\neg \instate_\utaub(\sync_\ue^\uID) \wedge \upsi) \lra \false$. By consequence, we have a trivial derivation:
  \[
    \infer[\simp]{
      \inframe_\tau,
      \lreveal_{\tauo},
      \neg \instate_\taub(\sync_\ue^\ID) \wedge
      \psi
      \;\sim\;
      \inframe_\utau,
      \rreveal_{\tauo},
      \neg \instate_\utaub(\sync_\ue^\uID)\wedge
      \upsi 
    }{
      \infer[\fa]{
        \inframe_\tau,
        \lreveal_{\tauo},
        \false
        \;\sim\;
        \inframe_\utau,
        \rreveal_{\tauo},
        \false
      }{
        \inframe_\tau,
        \lreveal_{\tauo}
        \;\sim\;
        \inframe_\utau,
        \rreveal_{\tauo}
      }
    }
  \]

  Now assume that $\taub = \_, \pnai(j_a,1)$. Since $\tauttt = \_,\fuai_\ID(j_0) \popre \tau$, we know by validity of $\tau$ there there exists $\tau' = \_,\npuai{2}{\ID}{j_0}$ or $\_,\cuai_\ID(j_0,1)$ such that $\tau' \potau \tauttt$. It is straightforward to check that if $\tau' = \_,\cuai_\ID(j_0,1)$ then since $\taub = \_, \pnai(j_a,1)$ we have $\futr{\tauttt}{\taua} \lra \false$ and $\futr{\utauttt}{\utaua} \lra \false$. Building the wanted derivation is then trivial.

  Therefore assume that $\tau' = \_,\npuai{2}{\ID}{j_0}$. Observe that $\futr{\tauttt}{\taua} \ra \accept^\ID_{\tau'}$. We have two cases:
  \begin{itemize}
  \item Assume $\tau' \potau \taub$. Using \ref{equ2}, we know that:
    \begin{alignat*}{2}
      \accept_{\tau'}^\ID
      &\;\ra\;\;&&
      \bigvee_{\taun = \_,\pnai(j_n,1)\atop{\taux \potau \taun \potau \tau'}}
      \supitr{\taux,\tau'}{\taun}\\
      &\;\ra\;\;&&
      \bigvee_{\taun = \_,\pnai(j_n,1)\atop{\taux \potau \taun \potau \tau'}}
      g(\inframe_\taux) = \nonce^{j_n}\\
      &\;\ra\;\;&&
      \bigvee_{\taun = \_,\pnai(j_n,1)\atop{\taux \potau \taun \potau \tau'}}
      \cstate_\taux^\ID(\bauth_\ue^\ID) = \nonce^{j_n}\\
      &\;\ra\;\;&&
      \cstate_\taux^\ID(\bauth_\ue^\ID) \ne \nonce^{j_a}
      \tag{Since $\tau' \potau \taub$}
    \end{alignat*}
    Moreover:
    \[
      \futr{\tauttt}{\taua} \ra \cstate_{\tau'}^\ID(\eauth_\ue^\ID) = \nonce^{j_a}
      \ra \cstate_{\taux}^\ID(\bauth_\ue^\ID) = \nonce^{j_a}
    \]
    Therefore $\futr{\tauttt}{\taua} \ra \false$. Similarly we can show that $\futr{\utauttt}{\utaua} \ra \false$. It is then easy build the wanted derivation.

  \item Assume $\taub \potau \tau'$. We summarize graphically the situation below:
    \begin{center}
      \begin{tikzpicture}
        [dn/.style={inner sep=0.2em,fill=black,shape=circle},
        sdn/.style={inner sep=0.15em,fill=white,draw,solid,shape=circle},
        sl/.style={decorate,decoration={snake,amplitude=1.6}},
        dl/.style={dashed},
        pin distance=0.5em,
        every pin edge/.style={thin}]

        \draw[thick] (0,0)
        -- ++(0.5,0)
        node[dn,pin={above,align=left:{$\ns_\ID(\_)$\\or $\epsilon$}}] {}
        node[below=0.3em,name=a]{$\tau_\ns$}
        -- ++(2,0)
        node[dn,pin={above:{$\npuai{1}{\ID}{j_0}$}}] {}
        node[below=0.3em,name=b]{$\taux$}
        -- ++(2.5,0)
        node[dn,pin={above:{$\pnai(j_a,1)$}}] {}
        node[below=0.3em,name=c0]{$\taub$}
        -- ++(2.5,0)
        node[dn,pin={above:{$\npuai{2}{\ID}{j_0}$}}] {}
        node[below=0.3em,name=c]{$\tau'$}
        -- ++(2,0)
        node[dn,pin={above:{$\fuai_\ID(j_0)$}}] {}
        node[below=0.3em,name=d]{$\tauttt$}
        -- ++(2,0)
        node[dn,pin={above:{$\cnai(j_1,0)$}}] {}
        node[below=0.3em,name=f]{$\taut$}
        -- ++(2,0)
        node[dn,pin={above:{$ \cuai_\ID(j,1) $}}] {}
        node[below=0.3em,name=g]{$\tau$};

        \draw[thin,dashed] (b) -- ++(0,-0.5) node (spe) {} -| (c0)
        {[draw=none] -- ++(0,-0.5)} -| (c)
        let \p1 = (c) in
        let \p2 = (spe) in
        (\x1,\y2) -| (d);
      \end{tikzpicture}
    \end{center}
    First, since there are no $\ID$ actions between $\taub$ and $\tau'$, we know that $\neg \instate_\taub(\sync_\ue^\ID) \ra \neg \instate_{\tau'}(\sync_\ue^\ID)$. Recall that $\futr{\tauttt}{\taua} \ra \accept^\ID_{\tau'}$. Using \ref{equ2}, it is simple to check that $\futr{\tauttt}{\taua} \wedge \accept^\ID_{\tau'} \ra \supitr{\taux,\tau'}{\taub}$. Therefore:
    \begin{alignat*}{2}
      \left(
        \neg \instate_\taub(\sync_\ue^\ID) \wedge
        \futr{\tauttt}{\taua}
      \right)
      \;&\ra\;\;&&
      \neg \instate_{\tau'}(\sync_\ue^\ID) \wedge
      \accept^\ID_{\tau'}\\
      \;&\ra\;\;&&
      \incaccept_\taub^\ID
      \begin{alignedat}[t]{2}
        &\;\wedge\;&&
        \instate_{\tau'}(\sqn_\hn^\ID) - \cstate_{\taub}(\sqn_\hn^\ID) = \zero\\
        &\;\wedge\;&&
        \cstate_{\tau'}(\sqn_\ue^\ID) - \cstate_{\tau'}(\sqn_\hn^\ID) = \zero
      \end{alignedat}
      \tag{Using \ref{sequ4}}
    \end{alignat*}
    Using again the fact that there are no $\ID$ actions between $\taub$ and $\tau'$, we know that $\instate_\taub(\sqn_\ue^\ID) \equiv \instate_{\tau'}(\sync_\ue^\ID)$. Moreover $\instate_{\tau'}(\sync_\ue^\ID) \equiv \cstate_{\tau'}(\sync_\ue^\ID)$, therefore $\instate_\taub(\sqn_\ue^\ID) = \cstate_{\tau'}(\sync_\ue^\ID)$. Similarly, we know that $\cstate_{\tau'}(\sqn_\hn^\ID) \equiv \instate_{\tau'}(\sqn_\hn^\ID)$. Summarizing:
    \begin{center}
      \begin{tikzpicture}
        [dn/.style={inner sep=0.2em,fill=black,shape=circle},
        sdn/.style={inner sep=0.15em,fill=white,draw,solid,shape=circle},
        sl/.style={decorate,decoration={snake,amplitude=1.6}},
        dl/.style={dashed},
        pin distance=0.5em,
        every pin edge/.style={thin}]

        \draw[thick] (0,0)
        -- ++(0.5,0)
        node[dn,pin={above,align=left:{$\ns_\ID(\_)$\\or $\epsilon$}}] {}
        node[below=0.3em,name=a]{$\tau_\ns$}
        -- ++(2,0)
        node[dn,pin={above:{$\npuai{1}{\ID}{j_0}$}}] {}
        node[below=0.3em,name=b]{$\taux$}
        -- ++(2.5,0)
        node[dn,pin={above:{$\pnai(j_a,1)$}}] {}
        node[below=0.3em,name=c0]{$\taub$}
        -- ++(2.5,0)
        node[dn,pin={above:{$\npuai{2}{\ID}{j_0}$}}] {}
        node[below=0.3em,name=c]{$\tau'$}
        -- ++(2,0)
        node[dn,pin={above:{$\fuai_\ID(j_0)$}}] {}
        node[below=0.3em,name=d]{$\tauttt$}
        -- ++(2,0)
        node[dn,pin={above:{$\cnai(j_1,0)$}}] {}
        node[below=0.3em,name=f]{$\taut$}
        -- ++(2,0)
        node[dn,pin={above:{$ \cuai_\ID(j,1) $}}] {}
        node[below=0.3em,name=g]{$\tau$};

        \draw[thin,dashed] (b) -- ++(0,-0.5) node (spe) {} -| (c0)
        {[draw=none] -- ++(0,-0.5)} -| (c)
        let \p1 = (c) in
        let \p2 = (spe) in
        (\x1,\y2) -| (d);

        \path (c0) -- ++ (0,-1.3)
        node (c01) {$\instate_{\taub}(\sqn_\hn^\ID)$}
        -- ++ (0,-1)
        node (c02) {$\instate_{\taub}(\sqn_\ue^\ID)$};
        
        \path (c) -- ++ (0,-1.3)
        node (c1) {$\cstate_{\tau'}(\sqn_\hn^\ID)$}
        -- ++ (0,-1)
        node (c2) {$\cstate_{\tau'}(\sqn_\ue^\ID)$};
        
        \draw (c01) -- (c1) node[midway,above,sloped]{$=$}
        (c1) -- (c2) node[midway,above,sloped]{$=$}
        (c02) -- (c2) node[midway,above,sloped]{$=$};
      \end{tikzpicture}
    \end{center}
    Using the fact that we have $\neg\instate_\taub(\sync_\ue^\ID)$ and \eqref{eq:vdsuovhwioqirjq}, we know that:
    \begin{alignat*}{2}
      &&& \neg\instate_\taub(\sync_\ue^\ID) \wedge
      \futr{\tauttt}{\taua} \wedge
      \cstate_\tauttt(\suci_\ue^\ID) =  \instate_\taut(\suci_\hn^\ID)\\
      &\lra\;\;&&
      \neg\instate_\taub(\sync_\ue^\ID) \wedge
      \futr{\tauttt}{\taua} \wedge
      \incaccept_\taub^\ID
      \wedge
      \left(
        \begin{alignedat}{2}
          &&&\cstate_{\tau'}(\sqn_\hn^\ID) - \cstate_{\tau'}(\sqn_\ue^\ID)\\
          &=\;\;&&
          \instate_\taut(\sqn_\hn^\ID) - \instate_\taut(\sqn_\ue^\ID)
        \end{alignedat}
      \right)
      \wedge
      \;\;\bigwedge_{
        \mathclap{\tau'=\_,\cnai(j',0)}
        \atop{\mathclap{\taub \potau \tau' \potau \taut}}}\;\;
      g(\inframe_{\tau'}) \ne \suci^{j_a}
    \end{alignat*}
    Besides, $\accept_{\tau'}^\ID \ra \cstate_{\tau'}(\sync_\ue^\ID)$, and since $\tau' \potau \taut$ we know that $\cstate_{\tau'}(\sync_\ue^\ID) \ra \instate_{\taut}(\sync_\ue^\ID)$. Hence:
    \begin{alignat*}{2}
      &&& \neg\instate_\taub(\sync_\ue^\ID) \wedge
      \futr{\tauttt}{\taua} \wedge
      \cstate_\tauttt(\suci_\ue^\ID) =  \instate_\taut(\suci_\hn^\ID)\\
      &\lra\;\;&&
      \neg\instate_\taub(\sync_\ue^\ID) \wedge
      \futr{\tauttt}{\taua} \wedge
      \incaccept_\taub^\ID
      \wedge
      \syncdiff_{\tau'}^\ID =
      \syncdiffin_\taut^\ID
      \wedge
      \;\;\;\;\bigwedge_{
        \mathclap{\tau'=\_,\cnai(j',0)}
        \atop{\mathclap{\taub \potau \tau' \potau \taut}}}\;\;\;\;
      g(\inframe_{\tau'}) \ne \suci^{j_a}
    \end{alignat*}
    Similarly we have:
    \begin{alignat*}{2}
      &&& \neg\instate_\utaub(\sync_\ue^\uID) \wedge
      \futr{\utauttt}{\utaua} \wedge
      \cstate_\utauttt(\suci_\ue^\uID) =  \instate_\utaut(\suci_\hn^\uID)\\
      &\lra\;\;&&
      \neg\instate_\utaub(\sync_\ue^\uID) \wedge
      \futr{\utauttt}{\utaua} \wedge
      \incaccept_\utaub^\uID
      \wedge
      \syncdiff_{\utau'}^\uID =
      \syncdiffin_\utaut^\uID
      \wedge
      \;\;\;\;\bigwedge_{
        \mathclap{\tau'=\_,\cnai(j',0)}
        \atop{\mathclap{\taub \potau \tau' \potau \taut}}}\;\;\;\;
      g(\inframe_{\utau'}) \ne \suci^{j_a}
    \end{alignat*}
    And using \ref{der2}, we know that we have a derivation of:
    \[
      \infer[\simp]{
        \inframe_\tau,
        \lreveal_{\tauo},
        \futr{\tauttt}{\taua}
        \sim\;\;
        \inframe_\utau,
        \rreveal_{\tauo},
        \futr{\tauttt}{\taua}
      }{
        \inframe_\tau, \lreveal_{\tauo}
        \sim
        \inframe_\utau,\rreveal_{\tauo}
      }
    \]
    Moreover, we know that:
    \begin{mathpar}
      \left(
        \left(
          \suci^{j_a},\suci^{j_a}
        \right) \in \reveal_{\tauo}
      \right)_{\tau'=\_,\cnai(j',0)
        \atop{\taub \potau \tau' \potau \taut}}
      
      \left(
        \syncdiffin_\taut^\ID,
        \syncdiffin_\utaut^\uID
      \right) \in \reveal_{\tauo}

      \left(
        \syncdiff_{\tau}^\ID,
        \syncdiff_{\utau'}^\uID
      \right) \in \reveal_{\tauo}

      \left(
        \instate_\taub(\sync_\ue^\ID),
        \instate_\utaub(\sync_\ue^\uID)
      \right) \in \reveal_{\tauo}    
    \end{mathpar}
    Similarly to what we did in \eqref{eq:vnsduvuiogquigw}, we can rewrite $\neg\instate_\taub(\sync_\ue^\ID) \wedge \psi$ and $\neg\instate_\utaub(\sync_\ue^\uID)\wedge \upsi$ as two terms that decompose, using $\fa$, into matching part of $\reveal_{\tauo}$. By consequence we can build the following derivation:
    \[
      \begin{gathered}[c]
        \infer[\simp]{
          \inframe_\tau,
          \lreveal_{\tauo},
          \neg\instate_\taub(\sync_\ue^\ID) \wedge
          \psi
          \;\sim\;
          \inframe_\utau,
          \rreveal_{\tauo},
          \neg\instate_\utaub(\sync_\ue^\uID)\wedge
          \upsi
        }{
          \inframe_\tau,
          \lreveal_{\tauo}
          \psi
          \;\sim\;
          \inframe_\utau,
          \rreveal_{\tauo}
        }
      \end{gathered}
    \]
  \end{itemize}
\end{itemize}

\paragraph{Part 2 (Dots)}

Using \ref{sequ2} we know that $\ptr{\tautt,\tau}{\taut} \ra \accept_\taut^{\ID}$. Therefore, using \ref{a6} we get that $\ptr{\tautt,\tau}{\taut} \ra \neg\accept_\taut^{\ID'}$ for every $\ID' \ne \ID$. It follows that $\ptr{\tautt,\tau}{\taut} \ra t_\taut = \textsf{msg}_\taut^\ID$, and therefore:
\[
  \ptr{\tautt,\tau}{\taut}\;\ra\;
  \pi_2(t_\taut) =
  \instate_\taut(\sqn_\hn^{\ID})
  \oplus \ow{\nonce^{j_1}}{\key^{\ID}}
  \numberthis\label{eq:gwouhsdifjpoifjrqepr}
\]
And:
\[
  \ptr{\tautt,\tau}{\taut}\;\ra\;
  \pi_3(t_\taut) =
  \mac{\striplet
  {\nonce^{j_1}}
  {\instate_\taut(\sqn_\hn^\ID)}
  {\instate_\taut(\suci_\ue^\ID)}}
  {\mkey^\ID}{3}
\]
Moreover, since no action from agent $\ID$ occurs between $\tautt$ and $\taut$, we know that $\instate_\taut(\suci_\ue^{\ID}) = \instate_\tautt(\suci_\ue^{\ID})$. Hence:
\[
  \ptr{\tautt,\tau}{\taut}\;\ra\;
  \pi_3(t_\taut) =
  \mac{\striplet
  {\nonce^{j_1}}
  {\instate_\taut(\sqn_\hn^\ID)}
  {\instate_\tautt(\suci_\ue^\ID)}}
  {\mkey^\ID}{3}
  \numberthis\label{eq:fdjslvieruqwrirupq}
\]
Therefore using \eqref{eq:gwouhsdifjpoifjrqepr} and \eqref{eq:fdjslvieruqwrirupq} we can rewrite $\ptr{\tautt,\tau}{\taut}$ as follows:
\begin{alignat*}{2}
  \ptr{\tautt,\tau}{\taut}
  &\;=\;\;&&
  \left(
    \begin{alignedat}[c]{2}
      &&&
      \pi_1(g(\inframe_\tau)) = \nonce^{j_1}
      \;\wedge\;
      \pi_2(g(\inframe_\tau)) =
      \dotuline{\pi_2(t_\taut)} 
      \;\wedge\;
      \pi_3(g(\inframe_\tau)) =
      \dotuline{\pi_3(t_\taut)}\\
      &\wedge\;\;&&
      g(\inframe_\taut) = \dashuline{\instate_\tautt(\suci_\ue^\ID)}
      \;\wedge\;
      \uwave{\instate_\tautt(\suci_\ue^\ID)
        = \instate_\taut(\suci_\hn^\ID)}
      \;\wedge\;
      \dashuline{\instate_\tautt(\success_\ue^\ID)}\\
      &\wedge\;\;&&
      \dashuline{\range{\instate_\tau(\sqn_\ue^\ID)}
        {\instate_\taut(\sqn_\hn^\ID)}}
    \end{alignedat}
  \right)
\end{alignat*}
By a similar reasoning we rewrite $\ptr{\utautt,\utau}{\utaut}$ as follows:
\begin{alignat*}{2}
  \ptr{\utautt,\utau}{\utaut}
  &\;\equiv\;\;&&
  \left(
    \begin{alignedat}{2}
      &&&
      \pi_1(g(\inframe_\utau)) = \nonce^{j_1}
      \;\wedge\;
      \pi_2(g(\inframe_\utau)) =
      \dotuline{\pi_2(t_\utaut)}
      \;\wedge\;
      \pi_3(g(\inframe_\utau)) =
      \dotuline{\pi_3(t_\utaut)}\\
      &\wedge\;\;&&
      g(\inframe_\utaut) = \dashuline{\instate_\utautt(\suci_\ue^{\uID})}
      \;\wedge\;
      \uwave{\instate_\utautt(\suci_\ue^{\uID}) =
        \instate_\utaut(\suci_\hn^{\uID})}
      \;\wedge\;
      \dashuline{\instate_\utautt(\success_\ue^{\uID})}\\
      &\wedge\;\;&&
      \dashuline{\range{\instate_\utau(\sqn_\ue^{\uID})}
        {\instate_\utaut(\sqn_\hn^{\uID})}}
    \end{alignedat}
  \right)
\end{alignat*}

\paragraph{Part 3 (Dash)}
Since $\ptr{\tautt,\tau}{\taut} \ra \instate_\tautt(\success_\ue^\ID)$ we know that:
\[
  \ptr{\tautt,\tau}{\taut}
  \;\ra\;
  \instate_\tautt(\suci_\ue^\ID)
  \;=\;
  \msuci_\tau^\ID
\]
Besides, as $\instate_\tautt(\success_\ue^\ID) \ra \instate_\tautt(\sync_\ue^\ID)$, and since $\instate_\tautt(\success_\ue^\ID) \ra \instate_\taut(\success_\ue^\ID)$ (because $\tautt \potau \taut$ and $\tautt \not \potau \ns_\ID(\_)$), we know that:
\[
  \ptr{\tautt,\tau}{\taut} \ra
  \left(
    \range{\instate_\tau(\sqn_\ue^\ID)}
    {\instate_\taut(\sqn_\hn^\ID)}
    \lra
    \left(
      \instate_\taut(\success_\ue^\ID) \wedge
      \instate_\tau(\sqn_\ue^\ID) =
      \instate_\taut(\sqn_\hn^\ID)
    \right)
  \right)
\]
Similarly we have:
\begin{gather*}
  \ptr{\utautt,\utau}{\utaut}
  \;\ra\;
  \instate_\utautt(\suci_\ue^{\uID})
  \;=\;
  \msuci_\utau^{\uID}\\
  \ptr{\utautt,\utau}{\utaut} \ra
  \left(
    \range{\instate_\utau(\sqn_\ue^\uID)}
    {\instate_\utaut(\sqn_\hn^\uID)}
    \lra
    \left(
      \instate_\utaut(\success_\ue^\uID) \wedge
      \instate_\utau(\sqn_\ue^\uID) =
      \instate_\utaut(\sqn_\hn^\uID)
    \right)
  \right)
\end{gather*}
Moreover:
\begin{mathpar}
  \left(
    \msuci_\tau^\ID
    \;\;\sim\;\;
    \msuci_\utau^{\uID}
  \right)
  \in \reveal_{\tauo}

  \left(
    \instate_\tautt(\success_\ue^\ID)
    \;\;\sim\;\;
    \instate_\utautt(\success_\ue^{\uID})
  \right)
  \in \reveal_{\tauo}
\end{mathpar}
Finally, using \ref{der1}, we know that we have a derivation of:
\[
  \infer[\fa^*]{
    \lreveal_{\tauo},
    \instate_\taut(\success_\ue^\ID) \wedge
    \instate_\tau(\sqn_\ue^\ID) =
    \instate_\taut(\sqn_\hn^\ID)  
    \;\sim\;
    \rreveal_{\tauo},
    \instate_\utaut(\success_\ue^\uID) \wedge
    \instate_\utau(\sqn_\ue^\uID) =
    \instate_\utaut(\sqn_\hn^\uID)
  }{
    \lreveal_{\tauo}
    \;\sim\;
    \rreveal_{\tauo}
  }
\]

\paragraph{Part 4 (conclusion)}
To conclude, we combine the derivations of {Part~1}, {Part~2} and {Part~3}.

\subsection*{Proof of \ref{der4}}
\begin{center}
  \begin{tikzpicture}
    [dn/.style={inner sep=0.2em,fill=black,shape=circle},
    sdn/.style={inner sep=0.15em,fill=white,draw,solid,shape=circle},
    sl/.style={decorate,decoration={snake,amplitude=1.6}},
    dl/.style={dashed},
    pin distance=0.5em,
    every pin edge/.style={thin}]

    \draw[thick] (0,0)
    node[left=1.3em] {$\tau:$}
    -- ++(0.5,0)
    node[dn,pin={above:{$\cuai_\ID(j_i,0)$}}] {}
    node[below=0.3em]{$\tautt$}
    -- ++(2.5,0)
    node[dn,pin={above:{$\cnai(j,0)$}}] {}
    node[below=0.3em]{$\taut$}
    -- ++(2.5,0)
    node[dn,pin={above:{$\cuai_\ID(j_i,1)$}}] {}
    node[below=0.3em]{$\taui$}
    -- ++(2.5,0)
    node[dn,pin={above:{$\cnai(j,1) $}}] {}
    node[below=0.3em]{$\tau$};
  \end{tikzpicture}
\end{center}
Recall that:
\begin{alignat*}{2}
  \ftr{\tautt,\taui}{\taut,\tau}
  &\;\equiv\;\;&&
  \left(
    \ptr{\tautt,\taui}{\taut} \wedge
    g(\inframe_\tau) = \mac{\nonce^j}{\mkey^\ID}{4}
  \right)
\end{alignat*}
The fact that $\utautt = \_,\cuai_{\nu_\taut(\ID)}(j_i,0)$, $\utaui = \_,\cuai_{\nu_\taut(\ID)}(j_i,1)$ and $\utautt <_\utau \utaut <_\utau \utaui$ is straightforward from~$\ref{der3}$. It is easy to check that:
\begin{alignat*}{2}
  \ftr{\utautt,\utaui}{\utaut,\utau}
  &\;\equiv\;\;&&
  \left(
    \ptr{\utautt,\utaui}{\utaut} \wedge
    g(\inframe_\utau) = \mac{\nonce^j}{\mkey^{\nu_\taut(\ID)}}{4}
  \right)
\end{alignat*}
Moreover:
\[
  \left(
    \mac{\nonce^j}{\mkey^\ID}{4},
    \mac{\nonce^j}{\mkey^{\nu_\taut(\ID)}}{4}
  \right) \in
  \reveal_{\tauo}
\]
And, using \ref{der3}, we know that there exists a derivation using only $\fa$ and $\dup$ of:
\[
  \inframe_\tau, \lreveal_{\tauo}
  \sim
  \inframe_\utau,\rreveal_{\tauo}
  \;\ra\;
  \left(
    \begin{alignedat}{4}
      &&&
      \inframe_\tau,&&
      \lreveal_{\tauo},&&
      \ptr{\tautt,\tau}{\taut}\\
      &\sim\;\;&&
      \inframe_\utau,&&
      \rreveal_{\tauo},&&
      \ptr{\utautt,\utau}{\utaut}
    \end{alignedat}
  \right)
\]
It is therefore easy to built the wanted derivation using only $\fa$ and $\dup$.

\subsection*{Proof of \ref{der2}}
We recall that:
\begin{gather*}
  \futr{\tau}{\taut} \;\equiv\;
  \left(\begin{alignedat}{2}
      &\injauth_\tau(\ID,j_0) 
      \wedge\instate_\tau(\eauth_\hn^{j_0}) \ne \unknownid\\
      \wedge\;& \pi_1(g(\inframe_\tau)) =
      \suci^{j_0} \xor \row{\nonce^{j_0}}{\key}
      \wedge\; \pi_2(g(\inframe_\tau)) =
      \mac{\spair
        {\suci^{j_0}}
        {\nonce^{j_0}}}
      {\mkey}{5}
    \end{alignedat}\right)\\
  \futr{\utau}{\utaut} \;\equiv\;
  \left(\begin{alignedat}{2}
      &\injauth_\utau({\nu_\tau(\ID)},j_0) 
      \wedge\instate_\utau(\eauth_\hn^{j_0}) \ne \unknownid\\
      \wedge\;& \pi_1(g(\inframe_\utau)) =
      \suci^{j_0} \xor \row{\nonce^{j_0}}{\key}
      \wedge\; \pi_2(g(\inframe_\utau)) =
      \mac{\spair
        {\suci^{j_0}}
        {\nonce^{j_0}}}
      {\mkey}{5}
    \end{alignedat}\right)\\
\end{gather*}
Let $j_0 \in \mathbb{N}$. Using Proposition~\ref{prop:injauth-charac} on $\tau$, we know that:
\[
  \injauth_\tau(\ID,j_0)
  \;\lra\;
  \nonce^{j_0} = \instate_\tau(\eauth_\ue^\ID)
  \numberthis\label{eq:i2-3}
\]
Similarly, using Proposition~\ref{prop:injauth-charac} on $\utau$ we have:
\[
  \injauth_\utau(\nu_\tau(\ID),j_0)
  \;\lra\;
  \nonce^{j_0} = \instate_\utau(\eauth_\ue^{\nu_\tau(\ID)})
  \numberthis\label{eq:i2-4}
\]
Let $\tauo$ be such that $\tau = \tauo,\ai$. It is straightforward to check that for any $n \in \mathbb{N}$:
\begin{gather*}
  \underbrace{
    \left(
      \cstate_\tauo(\eauth_\hn^{j_0}) = \unknownid
    \right)}
  _{\textsf{unk}}
  \;\lra\;
  \bigwedge_{1 \le i \le B}
  \neg \neauth_\tau(\agent{A}_i,j_0)\\
  \underbrace{
    \left(
      \cstate_\utauo(\eauth_\hn^{j_0}) = \unknownid
    \right)}
  _{\ufresh{\textsf{unk}}}
  \;\lra\;
  \bigwedge_{1 \le i \le B}
  \neg \uneauth_\utau(\agent{A}_i,j_0)
\end{gather*}
Since for all $1 \le i \le B$:
\[
  (\neauth_\tau(\agent{A}_i,j_0)
  \sim
  \uneauth_\utau(\agent{A}_i,j_0))
  \in \reveal_{\tauo}
\]
and since $\futr{\tau}{\taut} \wedge \textsf{unk} \;\ra\;\false$ and $\futr{\utau}{\utaut} \wedge \ufresh{\textsf{unk}} \;\ra\;\false$, we deduce that:
\[
  \infer[\fa^*]{
    \inframe_\tau,\lreveal_{\tauo},
    \futr{\tau}{\taut}
    \sim
    \inframe_\utau,\rreveal_{\tauo},
    \futr{\utau}{\utaut}
  }{
    \infer[R]{
      \begin{alignedat}{2}
        &&&\inframe_\tau,\lreveal_{\tauo},
        \textsf{unk},
        \futr{\tau}{\taut} \wedge \textsf{unk},
        \futr{\tau}{\taut} \wedge \neg\textsf{unk}\\
        &\sim&&
        \inframe_\utau,\rreveal_{\tauo},
        \ufresh{\textsf{unk}},
        \futr{\utau}{\utaut} \wedge \ufresh{\textsf{unk}},
        \futr{\utau}{\utaut} \wedge \neg \ufresh{\textsf{unk}}
      \end{alignedat}
    }{
      \infer[\dup^*]{
        \begin{alignedat}{2}
          &&&\inframe_\tau,\lreveal_{\tauo},
          \textsf{unk},
          \false,
          \futr{\tau}{\taut}\wedge \neg\textsf{unk}\\
          &\sim&&
          \inframe_\utau,\rreveal_{\tauo},
          \ufresh{\textsf{unk}},
          \false,
          \futr{\utau}{\utaut} \wedge \neg \ufresh{\textsf{unk}}
        \end{alignedat}
      }{
        \inframe_\tau,\lreveal_{\tauo},
        b_{j_i}\wedge \neg\textsf{unk}
        \sim
        \inframe_\utau,\rreveal_{\tauo},
        \futr{\utau}{\utaut} \wedge \neg \ufresh{\textsf{unk}}
      }
    }
  }
\]
From the definitions, we get that:
\[
  \left(\instate_\tau(\bauth_\hn^{j_0}) = \ID \right)
  \;\ra\;
  \left(
    \instate_\tau(\eauth_\hn^{j_0}) = \ID
    \vee
    \instate_\tau(\eauth_\hn^{j_0}) = \unknownid
  \right)
\]
Therefore:
\[
  \left(\futr{\tau}{\taut} \wedge \neg \textsf{unk}\right)
  \;\;\ra\;\;
  \instate_\tau(\eauth_\hn^{j_0}) = \ID
  \;\;\ra\;\;
  \neauth_\tau(\ID,j_0)
\]
Moreover:
\begin{equation*}
  \left(
    \neauth_\tau(\ID,j_0) \;\;\ra\;\;
    \left(
      \begin{alignedat}{2}
        & \suci^{j_0} \xor \row{\nonce^{j_0}}{{\key}} =
        \cond{\neauth_\tau(\ID,j_0)}
        {\tsuci_\tau(\ID,j_0)}\\
        \wedge\;& 
        \mac{\spair
          {\suci^{j_0}}
          {\nonce^{j_0}}}
        {{\mkey}}{5} =
        \cond{\neauth_\tau(\ID,j_0)}
        {\tmac_\tau(\ID,j_0)}
      \end{alignedat}\right)
  \right)
\end{equation*}
Using \eqref{eq:i2-3} and the observations above, we can rewrite $\futr{\tau}{\taut}\wedge \neg\textsf{unk}$ as follows:
\begin{alignat*}{3}
  \futr{\tau}{\taut}\wedge \neg\textsf{unk}
  &\;\;=\;\;&&
  \left(\begin{alignedat}{2}
      &\nonce^{j_0} = \instate_\tau(\eauth_\ue^\ID) 
      \;\wedge\;\neg\textsf{unk}\\
      \wedge\;& \pi_1(g(\inframe_\tau)) =
      \cond{\neauth_\tau(\ID,j_0)}
      {\tsuci_\tau(\ID,j_0)}\\
      \wedge\;& \pi_2(g(\inframe_\tau)) =
      \cond{\neauth_\tau(\ID,j_0)}
      {\tmac_\tau(\ID,j_0)}
    \end{alignedat}\right)
\end{alignat*}
Similarly, using \eqref{eq:i2-4}, we can rewrite $\futr{\utau}{\utaut} \wedge \neg \ufresh{\textsf{unk}}$ as follows:
\begin{alignat*}{2}
  \futr{\utau}{\utaut} \wedge \neg \ufresh{\textsf{unk}}
  &\;\;=\;\;&&
  \left(\begin{alignedat}{2}
      &\nonce^{j_0} = \instate_\utau(\eauth_\ue^{\nu_\tau(\ID)})
      \;\wedge\;\neg \ufresh{\textsf{unk}}\\
      \wedge\;& \pi_1(g(\inframe_\utau)) =
      \cond{\uneauth_\utau(\ID,j_0)}
      {\utsuci_\utau(\ID,j_0)}\\
      \wedge\;& \pi_2(g(\inframe_\utau)) =
      \cond{\uneauth_\utau(\ID,j_0)}
      {\utmac_\utau(\ID,j_0)}
    \end{alignedat}\right)
\end{alignat*}
We can now conclude the proof:
\[
  \infer[R+\fa^*]{
    \inframe_\tau,\lreveal_{\tauo},
    \futr{\tau}{\taut}\wedge \neg\textsf{unk}
    \sim
    \inframe_\utau,\rreveal_{\tauo},
    \futr{\utau}{\utaut} \wedge \neg \ufresh{\textsf{unk}}
  }{
    \infer[\dup^*]{
      \begin{alignedat}{2}
        &&&\inframe_\tau,\lreveal_{\tauo},
        \cstate_\tau(\eauth_\ue^\ID),
        \left(
          \begin{alignedat}{2}
            &\nonce^{j_0},\neg\textsf{unk},
            \suci^{j_0},\\
            & \cond{\neauth_\tau(\ID,j_0)}
            {\tsuci_\tau(\ID,j_0)},\\
            & \cond{\neauth_\tau(\ID,j_0)}
            {\tmac_\tau(\ID,j_0)}
          \end{alignedat}
        \right)\\
        &\sim&&
        \inframe_\utau,\rreveal_{\tauo},
        \cstate_\utau(\eauth_\ue^{\nu_\tau(\ID)}),
        \left(
          \begin{alignedat}{2}
            &\nonce^{j_0},\neg \ufresh{\textsf{unk}},
            \suci^{j_0},\\
            &\cond{\uneauth_\utau(\ID,j_0)}
            {\utsuci_\utau(\ID,j_0)},\\
            & \cond{\uneauth_\utau(\ID,j_0)}
            {\utmac_\utau(\ID,j_0)}
          \end{alignedat}
        \right)
      \end{alignedat}
    }{
      \inframe_\tau,\lreveal_{\tauo}
      \sim
      \inframe_\utau,\rreveal_{\tauo}
    }
  }
\]


%% file: proof.tex
\newpage
\section{Proof of Lemma~\ref{lem:unlink}}
\label{sec:unlink-proof}

The proof is by induction over $\tau$. For $\tau = \epsilon$, we just need to check that the elements from Item~\ref{item:ref} of Definition~\ref{def:app-ind-gen} are indistinguishable, which is obvious from the definition of $\cstate_\epsilon$ in Definition~\ref{def:init-sigma-phi}.

We now show the inductive case: let $\tau = \ai_0,\dots,\ai_n$ be a valid basic symbolic trace with at most $C$ actions $\newsession$, and let $\uain{0},\dots,\uain{n}$ be such that $\utau = \uain{0},\dots,\uain{n}$.  Also let $\tauo = \ai_0,\dots,\ai_{n-1}$ and $\utauo = \uain{0},\dots,\uain{n-1}$. We assume by induction that there exists a derivation of:
\[
  \inframe_\tau, \lreveal_{\tauo}
  \sim
  \inframe_\utau,\rreveal_{\tauo}
\]
We do a case disjunction on the value of $\ai$.

\subsection{Case $\ai = \newsession_\ID(j)$}
We know that $\uai = \newsession_{\nu_\utau(\ID)}(j)$ and $\nu_\tau(\ID) = \freshid(\nu_{\tauo}(\ID))$. Moreover, $\cframe_\tau \equiv \inframe_{\tau}$ and $\cframe_\utau \equiv \inframe_{\utau}$. Hence $\lreveal_{\tau}$ and $\lreveal_{\tauo}$ coincide everywhere except on:
\begin{mathpar}
  \cstate_\tau(\success_\ue^\ID) \sim
  \cstate_\utau(\success_\ue^{\nu_\tau(\ID)})

  \syncdiff_\tau^{\ID} \sim \syncdiff_\utau^{\nu_\tau(\ID)}

  \msuci^{\ID}_\tau \sim \msuci^{\nu_\tau(\ID)}_\utau
\end{mathpar}
We can easily conclude with the following derivation:
\begin{equation*}
  \infer[R]{
    \begin{alignedat}{2}
      &&&\inframe_\tau, \lreveal_{\tauo},
      \cstate_\tau(\success_\ue^\ID),
      \msuci^{\ID}_\tau,
      \syncdiff_\tau^{\ID}\\
      &\;\sim\;\;&&
      \inframe_\utau,\rreveal_{\tauo},
      \cstate_\utau(\success_\ue^{\nu_\tau(\ID)}),
      \msuci^{\nu_\tau(\ID)}_\utau,
      \syncdiff_\utau^{\nu_\tau(\ID)}
    \end{alignedat}}
  {
    \infer[\simp]{
      \inframe_\tau, \lreveal_{\tauo},
      \false
      ,\bot,\false
      \sim
      \inframe_\utau,\rreveal_{\tauo},
      \false
      ,\bot,\false
    }{
      \inframe_\tau, \lreveal_{\tauo}
      \sim
      \inframe_\utau,\rreveal_{\tauo}
    }
  }
\end{equation*}

\subsection{Case $\ai = \pnai(j,0)$}
We know that $\uai = \pnai(j,0)$. Here $\lreveal_{\tau}$ and $\lreveal_{\tauo}$ coincides completely. Using invariant \ref{a1} we know that $\nonce^j \not \in \st(\inframe_\tau)$, and $\nonce^j \not \in \st(\cframe_{\utauo})$. Therefore we conclude this case easily using the axiom \ax{Fresh}:
\[
  \infer[\ax{Fresh}]{
    \inframe_\tau,\lreveal_{\tauo},\nonce^j
    \sim
    \inframe_\utau,\rreveal_{\tauo},\nonce^j
  }{
    \inframe_\tau,\lreveal_{\tauo}
    \sim
    \inframe_\utau,\rreveal_{\tauo}
  }
\]

\subsection{Case $\ai = \npuai{1}{\ID}{j}$}
We know that $\uai = \npuai{1}{{\nu_\tau(\ID)}}{j}$. Here $\lreveal_{\tau}$ and $\lreveal_{\tauo}$ coincides everywhere except on the pairs:
\begin{mathpar}
  \cstate_\tau(\success_\ue^{\,\ID})
  \;\sim\;
  \cstate_\utau(\success_\ue^{\,\nu_\tau(\ID)})

  \msuci^\ID_\tau
  \;\sim\;
  \msuci^{\nu_\tau(\ID)}_\utau

  \syncdiff_\tau^\ID
  \;\;\sim\;\;
  \syncdiff_\utau^{\nu_\tau(\ID)}

  \idiff
  {\cstate_\tau(\sqn_\ue^\ID)}
  {\instate_\tau(\sqn_\ue^\ID)}
  \;\;\sim\;\;
  \idiff
  {\cstate_\utau(\sqn_\ue^{\nu_\tau(\ID)})}
  {\instate_\utau(\sqn_\ue^{\nu_\tau(\ID)})}

  \left(
    \enc{\spair
      {\ID}
      {\instate_{\tau}(\sqn_\ue^\ID)}}
    {\pk_\hn}{\enonce^{j}}
    \;\sim\;
    \enc{\spair
      {\nu_\tau(\ID)}
      {\instate_{\tau}(\sqn_\ue^{\nu_\tau(\ID)})}}
    {\pk_\hn}{\enonce^{j}}
  \right)

  \left(
    \mac{\spair
      {\enc{\spair
          {\ID}
          {\instate_{\tau}(\sqn_\ue^\ID)}}
        {\pk_\hn}{\enonce^{j}}}
      {g(\inframe_{\tau})}}
    {\mkey^\ID}{1}
    \;\sim\;
    \mac{\spair
      {\enc{\spair
          {\nu_\tau(\ID)}
          {\instate_{\utau}(\sqn_\ue^{\nu_\tau(\ID)})}}
        {\pk_\hn}{\enonce^{j}}}
      {g(\inframe_{\utau})}}
    {\mkey^{\nu_\tau(\ID)}}{1}
  \right)
\end{mathpar}

\paragraph{Part 1}
We know that $\cstate_\tau(\success_\ue^{\,\ID}) \equiv \cstate_\utau(\success_\ue^{\,\nu_\tau(\ID)}) \equiv \false$. We deduce that $\msuci^\ID_\tau = \msuci^{\nu_\tau(\ID)}_\utau = \bot$. It follows that we have the derivation:
\[
  \begin{gathered}[c]
    \infer[R]{
      \inframe_\tau,\lreveal_{\tauo},
      \cstate_\tau(\success_\ue^{\,\ID}),
      \msuci^\ID_\tau
      \sim
      \inframe_\utau,\rreveal_{\tauo},
      \cstate_\utau(\success_\ue^{\,\nu_\tau(\ID)}),
      \msuci^{\nu_\tau(\ID)}_\utau
    }{
      \infer[\fa^*]{
        \inframe_\tau,\lreveal_{\tauo},
        \false,
        \bot
        \sim
        \inframe_\utau,\rreveal_{\tauo},
        \false,
        \bot
      }{
        \inframe_\tau,\lreveal_{\tauo}
        \sim
        \inframe_\utau,\rreveal_{\tauo}
      }
    }
  \end{gathered}
  \numberthis\label{eq:ffhdjqoiejqirueprapdo}
\]
  
\paragraph{Part 2}
We have:
\begin{alignat*}{4}
  &\idiff
  {\cstate_\tau(\sqn_\ue^\ID)}
  {\instate_\tau(\sqn_\ue^\ID)} &\;\;=\;\;&
  \idiff
  {\sqnsuc(\instate_\tau(\sqn_\ue^\ID))}
  {\instate_\tau(\sqn_\ue^\ID)} &\;\;=\;\;&
  \one\\
  &\idiff
  {\cstate_\utau(\sqn_\ue^{\nu_\tau(\ID)})}
  {\instate_\utau(\sqn_\ue^{\nu_\tau(\ID)})} &\;\;=\;\;&
  \idiff
  {\sqnsuc(\instate_\utau(\sqn_\ue^{\nu_\tau(\ID)}))}
  {\instate_\utau(\sqn_\ue^{\nu_\tau(\ID)})} &\;\;=\;\;&
  \one
\end{alignat*}
And:
\begin{alignat*}{2}
  \syncdiff_\tau^\ID
  &\;\;=\;\;&&
  \lrpcond{\cstate_\tau(\sync_\ue^\ID)}
  {\idiff
    {\cstate_\tau(\sqn_\ue^\ID)}
    {\cstate_\tau(\sqn_\hn^\ID)}}\displaybreak[1]\\
  &\;\;=\;\;&&
  \lrpcond{\instate_\tau(\sync_\ue^\ID)}
  {\idiff
    {\sqnsuc(\instate_\tau(\sqn_\ue^\ID))}
    {\instate_\tau(\sqn_\hn^\ID)}}\\
  &\;\;=\;\;&&
  \lrpcond{\instate_\tau(\sync_\ue^\ID)}
  {\sqnsuc(\syncdiff_\tauo^\ID)}
\end{alignat*}
Similarly, $\syncdiff_\utau^{\nu_\tau(\ID)} = \lrpcond{\instate_\utau(\sync_\ue^{\nu_\tau(\ID)})}{\sqnsuc(\syncdiff_\utauo^{\nu_\tau(\ID)})}$. Hence we have the derivation:
\[
  \begin{gathered}
    \infer[\simp]{
      \begin{alignedat}{4}
        &&&\inframe_\tau,\lreveal_{\tauo},\;&&
        \syncdiff_\tau^\ID,\;&&
        \idiff
        {\cstate_\tau(\sqn_\ue^\ID)}
        {\instate_\tau(\sqn_\ue^\ID)}\\
        &\sim\;\;&&
        \inframe_\utau,\rreveal_{\tauo},\;&&
        \syncdiff_\utau^{\nu_\tau(\ID)},\;&&
        \idiff
        {\cstate_\utau(\sqn_\ue^{\nu_\tau(\ID)})}
        {\instate_\utau(\sqn_\ue^{\nu_\tau(\ID)})}
      \end{alignedat}
    }{
      \infer[\dup^*]{
        \inframe_\tau,\lreveal_{\tauo},
        \instate_\tau(\sync_\ue^\ID),
        \syncdiff_\tauo^\ID
        \sim
        \inframe_\utau,\rreveal_{\tauo},
        \instate_\utau(\sync_\ue^{\nu_\tau(\ID)}),
        \syncdiff_\utauo^{\nu_\tau(\ID)}
      }{
        \inframe_\tau,\lreveal_{\tauo}
        \sim
        \inframe_\utau,\rreveal_{\tauo}
      }
    }
  \end{gathered}
  \numberthis\label{eq:lmdopadc}
\]

\paragraph{Part 3}
Let $s_l \equiv \length(\spair{\ID}{\instate_{\tau}(\sqn_\ue^\ID)})$. Using the $\ccao$ axiom we directly have that:
\[
  \begin{gathered}
    \infer[\ccao]{
      \inframe_\tau,\lreveal_{\tauo},
      \enc{\spair
        {\ID}
        {\instate_{\tau}(\sqn_\ue^\ID)}}
      {\pk_\hn}{\enonce^{j}}
      \;\sim\;
      \inframe_\utau,\rreveal_{\tauo},
      \enc{\spair
        {\nu_\tau(\ID)}
        {\instate_{\tau}(\sqn_\ue^{\nu_\tau(\ID)})}}
      {\pk_\hn}{\enonce^{j}}
    }{
      \inframe_\tau,\lreveal_{\tauo},s_l
      \sim
      \inframe_\utau,\rreveal_{\tauo},s_l
      \;&\;
      \infer{
        \length(\spair
        {\ID}
        {\instate_{\tau}(\sqn_\ue^\ID)})
        =
        \length(\spair
        {\nu_\tau(\ID)}
        {\instate_{\tau}(\sqn_\ue^{\nu_\tau(\ID)})})
      }{
        \unary{\length(\ID) =
          \length(\nu_\tau(\ID))}
        \;&\;
        \length(\instate_{\tau}(\sqn_\ue^\ID))
        =
        \length(\instate_{\tau}(\sqn_\ue^{\nu_\tau(\ID)}))
      }
    }
  \end{gathered}
  \numberthis\label{eq:lmdopadc2}
\]
Moreover, using Proposition~\ref{prop:len-eq-sqn}, we know that:
\[
  \unary{
    \length(\instate_{\tau}(\sqn_\ue^\ID))
    =
    \length(\instate_{\tau}(\sqn_\ue^{\nu_\tau(\ID)}))}
\]
Similarly, we can show that $s_l = \length(\spair{\ID}{\instate_{\tau}(\sqnini_\ue^\ID)})$. Since:
\[
  \left(
    \length(\spair{\ID}{\instate_{\tau}(\sqnini_\ue^\ID)}),
    \length(\spair{\ID}{\instate_{\tau}(\sqnini_\ue^\ID)})
  \right)
  \in
  \reveal_{\tauo}
\]
we know that:
\[
  \infer[R+\dup]{
    \inframe_\tau,\lreveal_{\tauo},s_l
    \sim
    \inframe_\utau,\rreveal_{\tauo},s_l
  }{
    \inframe_\tau,\lreveal_{\tauo}
    \sim
    \inframe_\utau,\rreveal_{\tauo}
  }
\]
This completes the derivation in~\ref{eq:lmdopadc2}.

\paragraph{Part 4}
To conclude, it only remains to deal with the $\macsym^1$ terms. We start by computing $\setmac_{\mkey^\ID}^1$:
\begin{alignat*}{4}
  \setmac_{\mkey^\ID}^1(\inframe_\tau,\lreveal_{\tauo})
  &\;\;=\;\;&&&&
  \left\{
    \spair
    {\enc{\spair{\ID}{\instate_\taut(\sqn_\ue^\ID)}}{\pk_\hn}{\enonce^{j_1}}}
    {g(\inframe_\taut)} \mid
    \taut = \_,\npuai{1}{\ID}{j_1} \popre \tau
  \right\}\\
  &&&\cup\;\;&&
  \left\{
    \spair
    {\pi_1(g(\inframe_\taut))}
    {\nonce^{j_1}} \mid
    \taut = \_,\pnai(j_1,1) \popre \tau
  \right\}
\end{alignat*}
We want to get rid of the second set above: using $\ref{equ3}$, we know that for every $\taut = \_,\pnai(j_1,1) \popre \tau$:
\begin{alignat*}{2}
  \accept_\tau^\ID
  &\;\lra\;\;&&
  \bigvee_{
    \tautt = \_,\npuai{1}{\ID}{j_2}
    \atop{ \tautt \potau \taut}}
  \left(
    \begin{alignedat}{2}
      &&&g(\inframe_{\tautt}) = \nonce^j
      \wedge
      \pi_1(g(\inframe_\taut)) =
      \enc{\spair
        {\ID}
        {\instate_{\tautt}(\sqn_\ue^\ID)}}
      {\pk_\hn}{\enonce^{j_2}}\\
      &\wedge\;&&
      \pi_2(g(\inframe_\taut)) =
      {\mac{\spair
          {\enc{\spair
              {\ID}
              {\instate_{\tautt}(\sqn_\ue^\ID)}}
            {\pk_\hn}{\enonce^{j_2}}}
          {g(\inframe_{\tautt})}}
        {\mkey^\ID}{1}}
    \end{alignedat}
  \right)
  \numberthis\label{eq:dacxzc}
\end{alignat*}
We let $\Psi'$ be the vector of terms $\inframe_\tau,\lreveal_{\tauo}$ where we replaced every occurrence of $\accept_\taut^\ID$ (where $\taut = \_,\pnai(j_1,1) \popre \tau$) by the equivalent term from \eqref{eq:dacxzc}. We can check that we have:
\begin{alignat*}{4}
  \setmac_{\mkey^\ID}^1(\Psi')
  &\;\;=\;\;&&&&
  \left\{
    \spair
    {\enc{\spair{\ID}{\instate_\taut(\sqn_\ue^\ID)}}{\pk_\hn}{\enonce^{j_1}}}
    {g(\inframe_\taut)} \mid
    \taut = \_,\npuai{1}{\ID}{j_1} \popre \tau
  \right\}
\end{alignat*}
For every $\taut = \_,\npuai{1}{\ID}{j_1} \popre \tau$ it is easy to show using Proposition~\ref{prop:len-eq-sqn} that :
\[
  \length(\spair{\ID}{\instate_\tau(\sqn_\ue^\ID)}) =
  \length(\spair{\ID}{\instate_\taut(\sqn_\ue^\ID)})
\]
Moreover, using the axioms in $\axioms_{\textsf{len}}$ we know that $\length(\spair{\ID}{\instate_\tau(\sqn_\ue^\ID)}) \ne 0$. Therefore, using Proposition~\ref{prop:enc-neq} we get that we have:
\[
  \enc{\spair{\ID}{\instate_\tau(\sqn_\ue^\ID)}}{\pk_\hn}{\enonce^{j}} \ne
  \enc{\spair{\ID}{\instate_\taut(\sqn_\ue^\ID)}}{\pk_\hn}{\enonce^{j_1}}
\]
Hence by left injectivity of $\pair{\cdot}{\_}$:
\[
  \spair
  {\enc{\spair
      {\ID}
      {\instate_{\tau}(\sqn_\ue^\ID)}}
    {\pk_\hn}{\enonce^{j}}}
  {g(\inframe_{\tau})}
  \ne
  \spair
  {\enc{\spair{\ID}{\instate_\taut(\sqn_\ue^\ID)}}{\pk_\hn}{\enonce^{j_1}}}
  {g(\inframe_\taut)}
\]
It follows that we can apply the $\prfmac^1$ axiom to replace the following term by a fresh nonce $\nonce$:
\[
  \mac{\spair
    {\enc{\spair
        {\ID}
        {\instate_{\tau}(\sqn_\ue^\ID)}}
      {\pk_\hn}{\enonce^{j}}}
    {g(\inframe_{\tau})}}
  {\mkey^\ID}{1}
\]
We then rewrite every occurrence of the right-hand side of \eqref{eq:dacxzc} into $\accept_\taut^\ID$ (where $\taut = \_,\pnai(j_1,1) \popre \tau$). This yields the derivation:
\[
  \infer[\prfmac^1]{
    \begin{alignedat}{3}
      &&&\inframe_\tau,\lreveal_{\tauo},\;&&
      \mac{\spair
        {\enc{\spair
            {\ID}
            {\instate_{\tau}(\sqn_\ue^\ID)}}
          {\pk_\hn}{\enonce^{j}}}
        {g(\inframe_{\tau})}}
      {\mkey^\ID}{1}\\
      &\sim\;\;&&
      \inframe_\utau,\rreveal_{\tauo},\;&&
      \mac{\spair
        {\enc{\spair
            {\nu_\tau(\ID)}
            {\instate_{\utau}(\sqn_\ue^{\nu_\tau(\ID)})}}
          {\pk_\hn}{\enonce^{j}}}
        {g(\inframe_{\utau})}}
      {\mkey^{\nu_\tau(\ID)}}{1}
    \end{alignedat}
  }{
    \inframe_\tau,\lreveal_{\tauo},
    \nonce
    \;\sim\;
    \inframe_\utau,\rreveal_{\tauo}
    \mac{\spair
      {\enc{\spair
          {\nu_\tau(\ID)}
          {\instate_{\utau}(\sqn_\ue^{\nu_\tau(\ID)})}}
        {\pk_\hn}{\enonce^{j}}}
      {g(\inframe_{\utau})}}
    {\mkey^{\nu_\tau(\ID)}}{1}
  }
\]
We then do the same on the right side (we omit the details), and conclude using \ax{Fresh}:
\[
  \infer[\prfmac^1]{
    \inframe_\tau,\lreveal_{\tauo},
    \nonce
    \;\sim\;
    \inframe_\utau,\rreveal_{\tauo}
    \mac{\spair
      {\enc{\spair
          {\nu_\tau(\ID)}
          {\instate_{\utau}(\sqn_\ue^{\nu_\tau(\ID)})}}
        {\pk_\hn}{\enonce^{j}}}
      {g(\inframe_{\utau})}}
    {\mkey^{\nu_\tau(\ID)}}{1}
  }{
    \infer[\ax{Fresh}]{
      \inframe_\tau,\lreveal_{\tauo},
      \nonce
      \;\sim\;
      \inframe_\utau,\rreveal_{\tauo},
      \nonce
    }{
      \inframe_\tau,\lreveal_{\tauo}
      \;\sim\;
      \inframe_\utau,\rreveal_{\tauo}
    }
  }
\]
We conclude the proof by combining the derivation above with the derivations in \eqref{eq:ffhdjqoiejqirueprapdo}, \eqref{eq:lmdopadc} and \eqref{eq:lmdopadc2}, and by using the induction hypothesis.

\subsection{Case $\ai = \pnai(j,1)$}
We know that $\uai = \pnai(j,1)$. For every base identity $\ID$, let $M_\ID$ be the set:
\[
  M_\ID \;=\;
  \left\{
    \tautt \mid
    \tautt = \_,\npuai{1}{\ID}{j_1} \popre \tau \wedge
    \forall \taut \text{ s.t. } \taut \potau \taut\,
    \taut \ne \_,\newsession_\ID(\_)
  \right\}
\]
Here $\lreveal_{\tau}$ and $\lreveal_{\tauo}$ coincides everywhere except on the following pairs:
\begin{mathpar}
  \left(
    \syncdiff_\tau^\ID
    \;\;\sim\;\;
    \syncdiff_\utau^{\nu_\tau(\ID)}
  \right)_{\ID \in \baseiddom}

  \left(
    \neauth_\tau(\ID,j)
    \;\;\sim\;\;
    \uneauth_\utau(\ID,j)
  \right)_{\ID \in \baseiddom}

  \left(
    \mac{\spair
      {\nonce^{j}}
      {\sqnsuc(\instate_\tautt(\sqn_\ue^\ID))}}
    {\mkey^\ID}{2}
    \;\;\sim\;\;
    \mac{\spair
      {\nonce^{j}}
      {\sqnsuc(\instate_\utautt(\sqn_\ue^{\nu_\tau(\ID)}))}}
    {\mkey^{\nu_\tau(\ID)}}{2}
  \right)_{\tautt \in M_\ID, \ID \in \baseiddom}
\end{mathpar}
\paragraph{Part 1}
Let $\ID$ be a base identity. We consider all the new sessions started with identity $\ID$ in $\tau$:
\[
  \left\{
    \newsession_\ID(0),\dots,\newsession_\ID(l_\ID)
  \right\}
  =
  \left\{
    \newsession_\ID(i) \mid \newsession_\ID(i) \in \tau
  \right\}
\]
This induce a partition of symbolic actions in $\tau$ for identity $\ID$.
Indeed, let $k$ be such that $\ID = \agent{A}_{k,0}$, and for every $-1 \le i \le l_\ID$, let $\uID_i = \agent{A}_{k,i+1}$. Then we define, for every $-1 \le i \le l_\ID$:
\begin{equation*}
  T_\ID^i \;=\;
  \left\{ \taut \mid
    \taut = \_,\npuai{1}{\ID}{j_1} \wedge
    \begin{dcases*}
      \newsession_\ID(i) \potau \taut \potau \newsession_\ID(i+1)
      & if $1 \le i < l_\ID$\\
      \taut \potau \newsession_\ID(0)
      & if $i = -1$\\
      \newsession_\ID(l) \potau \taut \popre \tau
      & if $i = l$\\
    \end{dcases*}
  \right\}
\end{equation*}
And $T_\ID = \{ \taut \mid \taut = \_,\npuai{1}{\ID}{j_1} \wedge \taut \popre \tau \}$. We have $T_\ID = \biguplus_{-1 \le i \le l_\ID} T_\ID^i$, and for every $-1 \le i \le l_\ID$:
\begin{equation*}
  \forall \taut \in T_\ID^i,\,\nu_\taut(\ID) = \uID_i
  \quad\text{ and }\quad
  T_\ID^i \;=\;
  \left\{ \taut \mid
    \utaut = \_,\npuai{1}{{\uID_i}}{j_1} \wedge
    \utaut \popre \utaut
  \right\}
\end{equation*}

\paragraph{Part 2}
Using \ref{equ3} we know that:
\begin{alignat*}{2}
  \accept_\tau^\ID
  &\;\lra\;\;&&
  \bigvee_{
    \taut = \_,\npuai{1}{\ID}{j_1} \in T_\ID}
  \underbrace{\left(
      \begin{alignedat}{2}
        &&&g(\inframe_{\taut}) = \nonce^j
        \wedge
        \pi_1(g(\inframe_\tau)) =
        \enc{\spair
          {\ID}
          {\instate_{\taut}(\sqn_\ue^\ID)}}
        {\pk_\hn}{\enonce^{j_1}}\\
        &\wedge\;&&
        \pi_2(g(\inframe_\tau)) =
        {\mac{\spair
            {\enc{\spair
                {\ID}
                {\instate_{\taut}(\sqn_\ue^\ID)}}
              {\pk_\hn}{\enonce^{j_1}}}
            {g(\inframe_{\taut})}}
          {\mkey^\ID}{1}}
      \end{alignedat}
    \right)}_{b_\taut^\ID}
  \numberthis\label{eq:newuniq}
\end{alignat*}
For all $\taut \in T_\ID$, we let $b^\ID_\taut$ be the main term of the disjunction above.

Similarly, using \ref{equ3} on $\utau$, which is a valid symbolic frame, we have that for every $-1 \le i \le l_\ID$:
\begin{alignat*}{2}
  \accept_\utau^{\uID_i}
  &\;\lra\;\;&&
  \bigvee
  _{\taut = \_,\npuai{1}{\ID}{j_1} \in T_\ID^i}
  \underbrace{
    \left(
      \begin{alignedat}{2}
        &&&g(\inframe_{\utaut}) = \nonce^j
        \wedge
        \pi_1(g(\inframe_\utau)) =
        \enc{\spair
          {\uID_i}
          {\instate_{\utaut}(\sqn_\ue^{\uID_i})}}
        {\pk_\hn}{\enonce^{j_1}}\\
        &\wedge\;&&
        \pi_2(g(\inframe_\utau)) =
        {\mac{\spair
            {\enc{\spair
                {\uID_i}
                {\instate_{\utaut}(\sqn_\ue^{\uID_i})}}
              {\pk_\hn}{\enonce^{j_1}}}
            {g(\inframe_{\utaut})}}
          {\mkey^{\uID_i}}{1}}
      \end{alignedat}
    \right)
  }_{\ufresh{b}_\utaut^{\uID_i}}
  \numberthis\label{eq:newuniq2}
\end{alignat*}
Moreover, if we let $\{\uID_{l_\ID+1},\dots,\uID_{m}\}$ be such that:
\[
  \copyid(\ID) = \{\uID_{0},\dots,\uID_{l_\ID}\}
  \uplus  \{\uID_{l_\ID+1},\dots,\uID_{m}\}
\]
Then, for all $i > l_\ID$, we have $\accept_\utau^{\uID_i} \lra \false$. Therefore, using \ref{a5}, we can show that:
\[
  \uneauth_\utau^\ID \;\lra\;
  \bigvee_{-1 \le i \le l} \accept_\utau^{\uID_i}
  \numberthis\label{eq:newuniq3}
\]

\paragraph{Part 3}
For every $\taut,\tautt \in T_\ID$ such that $\taut \ne \tautt$, $\taut = \_,\npuai{1}{\ID}{j_1}$ and $\tautt = \_,\npuai{1}{\ID}{j_2}$, using Proposition~\ref{prop:enc-neq} and \ref{prop:len-eq-sqn} we can show that:
\begin{alignat*}{2}
  b^\ID_\taut \wedge b^\ID_\tautt
  &\;\ra\;&&
  \enc{\spair
    {\ID}
    {\instate_{\taut}(\sqn_\ue^\ID)}}
  {\pk_\hn}{\enonce^{j_1}}
  =
  \enc{\spair
    {\ID}
    {\instate_{\tautt}(\sqn_\ue^\ID)}}
  {\pk_\hn}{\enonce^{j_2}}\\
  &\;\ra\;&&
  \false
  \numberthis\label{eq:ferdff}
\end{alignat*}
Similarly, for every $\taut,\tautt \in T^{\uID_i}_\ID$ such that $\taut \ne \tautt$, $\taut = \_,\npuai{1}{\ID}{j_1}$ and $\tautt = \_,\npuai{1}{\ID}{j_2}$, using Proposition~\ref{prop:enc-neq} and \ref{prop:len-eq-sqn} we have that:
\begin{alignat*}{2}
  \ufresh{b}_\utaut^{\uID_i} \wedge \ufresh{b}_\utautt^{\uID_i}
  &\;\ra\;&&
  \enc{\spair
    {{\uID_i}}
    {\instate_{\utaut}(\sqn_\ue^{\uID_i})}}
  {\pk_\hn}{\enonce^{j_1}}
  =
  \enc{\spair
    {{\uID_i}}
    {\instate_{\utautt}(\sqn_\ue^{\uID_i})}}
  {\pk_\hn}{\enonce^{j_2}}\\
  &\;\ra\;&&
  \false
  \numberthis\label{eq:ferdff2}
\end{alignat*}
Moreover, since for all identities $\ID_1 \ne \ID_2$, we have $\eq{\ID_1}{\ID_2} = \false$ we know that:
\[
  \left(\accept_\tau^{\ID_1} \wedge \accept_\tau^{\ID_2}\right)
  = \false
  \qquad\qquad
  \left(\accept_\utau^{\ID_1} \wedge \accept_\utau^{\ID_2}\right)
  = \false
\]
And for all non base identity $\ID$, using \ref{acc1} we know that $\accept_\tau^\ID \lra \false$. We deduce that:
\begin{equation*}
  \left(
    \left(
      (b^\ID_\taut)_{\taut \in T_\ID}
    \right)_{\ID \in \baseiddom},
    \underbrace{{\textstyle \bigwedge_{\ID  \in \baseiddom}} \neg \accept_\tau^\ID}
    _{b_\textsf{unk}}
  \right)
  \qquad\text{ and }\qquad
  \left(
    \left(
      (\ufresh{b}^{\uID_i}_\utaut)_{\taut \in T^{\uID_i}_\ID\atop{-1 \le i \le l_\ID}}
    \right)_{\ID  \in \baseiddom},
    \underbrace{{\textstyle\bigwedge_{\ID \in \iddom}} \neg \accept_\tau^\ID}
    _{\ufresh{b_\textsf{unk}}}
  \right)
\end{equation*}
are $\cs$ partitions. Besides, for all $\taut \in T_\ID$ we have:
\[
  \lrpcond{b^\ID_\taut}{t_\tau \;=\;
    \mac{\spair{\nonce^j}{\sqnsuc(\instate_\taut(\sqn_\ue^\ID)}}
    {\mkey^\ID}{2}}
  \qquad\text{ and }\qquad
  \lrpcond{b_{\textsf{unk}}}{t_\tau = \unknownid}
\]
From Proposition~\ref{prop:case} we deduce:
\[
  \begin{gathered}[c]
    t_\tau \;=\;
    \begin{alignedat}[t]{2}
      &\ite{\neg b_{\textsf{unk}}}{
        \switch{\taut \in T_\ID\atop{\ID \in \baseiddom}}
        {b^\ID_\taut}
        {\mac{\spair{\nonce^j}{\sqnsuc(\instate_\taut(\sqn_\ue^\ID)}}
          {\mkey^\ID}{2}}\\ &}{\unknownid}
    \end{alignedat}
  \end{gathered}
  \numberthis\label{eq:pnai1-switch}
\]
Similarly, for every $-1 \le i \le l$, for every $\taut \in T_i^\ID$:
\[
  \lrpcond{\ufresh{b}_\utaut^{\uID_i}}{t_\utau \;=\;
    \mac{\spair{\nonce^j}{\sqnsuc(\instate_\utaut(\sqn_\ue^{\uID_i})}}
    {\mkey^{\uID_i}}{2}}
  \qquad\text{ and }\qquad
  \lrpcond{\ufresh{b_{\textsf{unk}}}}{t_\utau = \unknownid}
\]
Again, from Proposition~\ref{prop:case} we deduce:
\[
  t_\utau \;=\;
  \begin{alignedat}[t]{2}
    &\ite{\neg \ufresh{b_{\textsf{unk}}}}{
      \switch{\taut \in T^i_\ID
        \atop{-1\le i \le l_\ID\atop{\ID \in \baseiddom}}}
      {\ufresh{b}^{\uID_i}_\utaut}
      {\mac{\spair{\nonce^j}{\sqnsuc(\instate_\utaut(\sqn_\ue^{\uID_i})}}
        {\mkey^{\uID_i}}{2}}\\ &}{\unknownid}
  \end{alignedat}
\]
Since $T_\ID \;=\;\biguplus_{-1 \le i \le l_\ID} T_\ID^i$, and since $\forall \taut \in T^i_\ID,\,\uID_i = \nu_\taut(\ID)$, we know that:
\[
  \begin{gathered}[c]
    t_\utau \;=\;
    \begin{alignedat}[t]{2}
      &\ite{\neg \ufresh{b_{\textsf{unk}}}}{
        \switch{\taut \in T_\ID\atop{\ID \in \baseiddom}}
        {\ufresh{b}^{\nu_\taut(\ID)}_\utaut}
        {\mac{\spair{\nonce^j}{\sqnsuc(\instate_\utaut(\sqn_\ue^{\nu_\taut(\ID)})}}
          {\mkey^{\nu_\taut(\ID)}}{2}}\\ &}{\unknownid}
    \end{alignedat}
  \end{gathered}
  \numberthis\label{eq:pnai1-switch1}
\]
\paragraph{Part 4}
We are going to show that for every $\ID \in \baseiddom$, for every $\taut = \npuai{1}{{\ID}}{j_1} \in T_\ID$, there is a derivation~of:
\[
  \Phi_\taut \;\; \equiv \;\;
  \inframe_\tau, \lreveal_{\tauo},b_\taut^\ID
  \sim
  \inframe_\utau,\rreveal_{\tauo},\ufresh{b}^{\uID_i}_\utaut
\]
For this, we rewrite $b_\taut^\ID$ and $\ufresh{b}^{\uID_i}_\utaut$ using, respectively, \eqref{eq:newuniq} and \eqref{eq:newuniq2}. First, remark that:
\[
  \infer[\dup^*]{
    \inframe_{\tau},
    \inframe_{\taut}
    \;\;\sim\;\;
    \inframe_{\utau},
    \inframe_{\utaut}
  }{
    \inframe_{\tau}\;\;\sim\;\;\inframe_{\utau}
  }
\]
And that the following pairs of terms are in $\reveal_{\tauo}$:
\begin{mathpar}
  (\nonce^j,\nonce^j)

  \left(
    \enc{\spair
      {\ID}
      {\instate_{\taut}(\sqn_\ue^\ID)}}
    {\pk_\hn}{\enonce^{j_1}},
    \enc{\spair
      {\nu_\taut(\ID)}
      {\instate_{\utaut}(\sqn_\ue^{\nu_\taut(\ID)})}}
    {\pk_\hn}{\enonce^{j_1}}
  \right)

  \left(
    \mac{\spair
      {\enc{\spair
          {\ID}
          {\instate_{\taut}(\sqn_\ue^\ID)}}
        {\pk_\hn}{\enonce^{j_1}}}
      {g(\inframe_{\taut})}}
    {\mkey^\ID}{1},
    \mac{\spair
      {\enc{\spair
          {\nu_\taut(\ID)}
          {\instate_{\utaut}(\sqn_\ue^{\nu_\taut(\ID)})}}
        {\pk_\hn}{\enonce^{j_1}}}
      {g(\inframe_{\utaut})}}
    {\mkey^{\nu_\taut(\ID)}}{1}
  \right)
\end{mathpar}
Therefore:
\[
  \begin{gathered}[c]
    \infer[\simp]{
      \inframe_\tau, \lreveal_{\tauo},b_\taut^\ID
      \sim
      \inframe_\utau,\rreveal_{\tauo},\ufresh{b}^{\uID_i}_\utaut
    }{
      \inframe_\tau, \lreveal_{\tauo}
      \sim
      \inframe_\utau,\rreveal_{\tauo}
    }
  \end{gathered}
  \numberthis\label{eq:pnai1-deriv0}
\]
Combining this with \eqref{eq:newuniq}, \eqref{eq:newuniq2} and \eqref{eq:newuniq3}, we have:
\[
  \begin{gathered}[c]
    \infer[\simp]{
      \inframe_\tau, \lreveal_{\tauo},\neauth_\tau^\ID
      \sim
      \inframe_\utau,\rreveal_{\tauo},\uneauth_\utau^\ID
    }{
      \infer[\simp]{
        \inframe_\tau, \lreveal_{\tauo},
        \left(b^{\ID}_\utaut\right)_{\taut \in T_\ID}
        \sim
        \inframe_\utau,\rreveal_{\tauo},
        \left(\ufresh{b}^{\uID_i}_\utaut\right)_{\taut \in T_\ID}
      }{
        \inframe_\tau, \lreveal_{\tauo}
        \sim
        \inframe_\utau,\rreveal_{\tauo}
      }
    }
  \end{gathered}
  \numberthis\label{eq:pnai1-deriv4}
\]
And:
\[
  \begin{gathered}[c]
    \infer[\simp]{
      \inframe_\tau, \lreveal_{\tauo},b_{\textsf{unk}}
      \sim
      \inframe_\utau,\rreveal_{\tauo},\ufresh{b_{\textsf{unk}}}
    }{
      \infer[\simp]{
        \inframe_\tau, \lreveal_{\tauo},
        \left(b^{\ID}_\utaut\right)_{\taut \in T_\ID,\ID \in \baseiddom}
        \sim
        \inframe_\utau,\rreveal_{\tauo},
        \left(\ufresh{b}^{\uID_i}_\utaut\right)_{\taut \in T_\ID,\ID \in \baseiddom}
      }{
        \inframe_\tau, \lreveal_{\tauo}
        \sim
        \inframe_\utau,\rreveal_{\tauo}
      }
    }
  \end{gathered}
  \numberthis\label{eq:pnai1-deriv5}
\]
We can now prove that $t_\tau \sim t_\utau$. First we rewrite $t_\tau$ and $t_\utau$ using, respectively, \eqref{eq:pnai1-switch} and \eqref{eq:pnai1-switch1}. Then we split the proof with $\fa$, and combine it with \eqref{eq:pnai1-deriv0} and \eqref{eq:pnai1-deriv5}. This yields:
\[
  \begin{gathered}[c]
    \infer[\simp]{
      \inframe_\tau, \lreveal_{\tauo},t_\tau
      \sim
      \inframe_\utau,\rreveal_{\tauo},t_\utau
    }{
      \infer{
        \begin{alignedat}{4}
          &&&\inframe_\tau, \lreveal_{\tauo},&&
          b_{\textsf{unk}},&&
          \left(
            b^\ID_\taut,
            \mac{\spair{\nonce^j}
              {\sqnsuc(\instate_\taut(\sqn_\ue^\ID))}}
            {\mkey^\ID}{2}
          \right)_{\taut \in T_\ID,\ID \in \baseiddom}\\
          &\sim\;\;&&
          \inframe_\utau,\rreveal_{\tauo},&&
          \ufresh{b_{\textsf{unk}}},&&
          \left(
            \ufresh{b}^{\nu_\taut(\ID)}_\utaut,
            \mac{\spair{\nonce^j}
              {\sqnsuc(\instate_\utaut(\sqn_\ue^{\nu_\taut(\ID)}))}}
            {\mkey^{\nu_\taut(\ID)}}{2}
          \right)_{\taut \in T_\ID,\ID \in \baseiddom}
        \end{alignedat}
      }{
        \begin{alignedat}{3}
          &&&\inframe_\tau, \lreveal_{\tauo},&&
          \left(
            \mac{\spair{\nonce^j}{\sqnsuc(\instate_\taut(\sqn_\ue^\ID))}}
            {\mkey^\ID}{2}
          \right)_{\taut \in T_\ID,\ID \in \baseiddom}\\
          &\sim\;\;&&
          \inframe_\utau,\rreveal_{\tauo},&&
          \left(
            \mac{\spair{\nonce^j}
              {\sqnsuc(\instate_\utaut(\sqn_\ue^{\nu_\taut(\ID)}))}}
            {\mkey^{\nu_\taut(\ID)}}{2}
          \right)_{\taut \in T_\ID,\ID \in \baseiddom}
        \end{alignedat}
      }
    }
  \end{gathered}
  \numberthis\label{eq:pnai1-deriv6}
\]
Notice that for every $\ID \in \baseiddom$, $M_\ID = T_\ID^{l_\ID}$. Therefore the $\macsym$ part in $\reveal_{\tau} \backslash \reveal_{\tauo}$ appears in the derivation above, i.e.:
\[
  \begin{alignedat}[c]{2}
    &&&\left(
      \mac{\spair
        {\nonce^{j}}
        {\sqnsuc(\instate_\tautt(\sqn_\ue^\ID))}}
      {\mkey^\ID}{2},
      \mac{\spair
        {\nonce^{j}}
        {\sqnsuc(\instate_\utautt(\sqn_\ue^{\nu_\tau(\ID)}))}}
      {\mkey^{\nu_\tau(\ID)}}{2}
    \right)_{\tautt \in M_\ID, \ID \in \baseiddom}\\
    &\subseteq\quad&&
    \left(
      \mac{\spair{\nonce^j}{\sqnsuc(\instate_\taut(\sqn_\ue^\ID))}}
      {\mkey^\ID}{2},
      \mac{\spair{\nonce^j}
        {\sqnsuc(\instate_\utaut(\sqn_\ue^{\nu_\taut(\ID)}))}}
      {\mkey^{\nu_\taut(\ID)}}{2}
    \right)_{\taut \in T_\ID,\ID \in \baseiddom}
  \end{alignedat}
  \numberthis\label{eq:pnai1-deriv10}
\]

\paragraph{Part 5}
Let $\ID \in \baseiddom$. Our goal is to apply the $\prfmac^2$ hypothesis to $\mac{\spair{\nonce^j} {\sqnsuc(\instate_\taut(\sqn_\ue^\ID))}} {\mkey^\ID}{2}$ simultaneously for every $\taut \in T_\ID$ in:
\[
  \Psi \;\equiv\;
  \inframe_\tau, \lreveal_{\tauo},
  \left(
    \mac{\spair{\nonce^j}{\sqnsuc(\instate_\taut(\sqn_\ue^\ID))}}
    {\mkey^\ID}{2}
  \right)_{\taut \in T_\ID,\ID \in \baseiddom}
\]
Using \ref{equ2} we know that for every $\newsession_\ID(l_\ID) \potau \taui = \_,\npuai{2}{\ID}{j_i}$:
\begin{alignat*}{2}
  \accept_\taui^\ID
  &\;\lra\;&&
  \bigvee
  _{\taut = \_,\pnai(j_1,1)
    \atop{\tautt = \_,\npuai{1}{\ID}{j_i}
      \atop{\tautt \potau \taut \popre \tau}}}
  g(\inframe_\tau) =
  \mac{\spair
    {\nonce^{j_1}}
    {\sqnsuc(\instate_\tautt(\sqn_\ue^\ID))}}
  {\mkey^\ID}{2}
  \;\wedge\;
  g(\inframe_\tautt) = \nonce^{j_1}
  \numberthis\label{eq:rewrite-accept}
\end{alignat*}
Let $\Psi'$ be the formula obtained from $\Psi$ by rewriting every $\accept_\taui^\ID$ s.t. $\newsession_\ID(l_\ID) \potau \taui = \_,\npuai{2}{\ID}{j_i}$ using the equation above. Then we can check that for every $\taut \in T_\ID$, there is only one occurrence of $\mac{\spair{\nonce^j}{\sqnsuc(\instate_\taut(\sqn_\ue^\ID))}}{\mkey^\ID}{2}$ in $\Psi'$. Moreover:
 \begin{alignat*}{3}
  \lefteqn{\setmac_{\ID}^{2}\left(\Psi'\right)
    \backslash\{
    \spair{\nonce^j}
    {\sqnsuc(\instate_\taut(\sqn_\ue^\ID))}
    \} \;\;=}\\
  &\qquad\qquad&&&&
  \left\{
    \spair{\nonce^j}
    {\sqnsuc(\instate_\tautt(\sqn_\ue^\ID))}
    \mid
    \tautt \in T_\ID \wedge \taut \ne \tautt
  \right\}\\
  &&&\cup\;&&
  \left\{
    \spair
    {\nonce^{j_0}}
    {\sqnsuc(\pi_2(\dec(\pi_1(g(\inframe_\taui)),\sk_\hn)))}
    \mid
    \taui = \_,\pnai(j_0,1) \popre \tau
  \right\}
\end{alignat*}
To apply the $\prfmac^2$ axioms, it is sufficient to show that for every element $u$ in the set above, we have $(\spair{\nonce^j}{\sqnsuc(\instate_\taut(\sqn_\ue^\ID))} \ne u$:
\begin{itemize}
\item Using $\ref{a2}$ we know that for every
  $\taut,\tautt \in T_\ID$, if $\taut \ne \tautt$ then
  $\instate_\tautt(\sqn_\ue^\ID)) \ne
  \instate_\tautt(\sqn_\ue^\ID))$. Therefore:
  \[
    \spair{\nonce^j}
    {\sqnsuc(\instate_\taut(\sqn_\ue^\ID))} \;\ne\;
    \spair{\nonce^j}
    {\sqnsuc(\instate_\tautt(\sqn_\ue^\ID))}
  \]
\item for every $\taui = \_,\pnai(j_0,1) \popre \tau$, we have $j_0 < j$, hence $\nonce^{j_0} \ne \nonce^j$ and by consequence:
  \[
    \spair{\nonce^j}
    {\sqnsuc(\instate_\taut(\sqn_\ue^\ID))}
    \;\ne\;
    \spair
    {\nonce^{j_0}}
    {\sqnsuc(\pi_2(\dec(\pi_1(g(\inframe_\taui)),\sk_\hn)))}
  \]
\end{itemize}
We can conclude: we rewrite $\Psi$ into $\Psi'$; we apply $\prfmac^2$ for every $\taut \in T_\ID$, replacing $\mac{\spair{\nonce^j}{\sqnsuc(\instate_\taut(\sqn_\ue^\ID))}}{\mkey^\ID}{2}$ by a fresh nonce $\nonce^{j,\taut}$; and we rewrite any term of the form \eqref{eq:rewrite-accept} back into $\accept_\taui^\ID$. Doing this for every base identity $\ID \in \baseiddom$, this yields:
\[
  \infer[(\prfmac^2)^*]{
    \begin{alignedat}{3}
      &&&\inframe_\tau, \lreveal_{\tauo},&&
      \left(
        \mac{\spair{\nonce^j}{\sqnsuc(\instate_\taut(\sqn_\ue^\ID))}}
        {\mkey^\ID}{2}
      \right)_{\taut \in T_\ID,\ID \in \baseiddom}\\
      &\sim\;\;&&
      \inframe_\utau,\rreveal_{\tauo},&&
      \left(
        \mac{\spair{\nonce^j}
          {\sqnsuc(\instate_\utaut(\sqn_\ue^{\nu_\taut(\ID)}))}}
        {\mkey^{\nu_\taut(\ID)}}{2}
      \right)_{\taut \in T_\ID,\ID \in \baseiddom}
    \end{alignedat}
  }{
    \begin{alignedat}{3}
      &&&\inframe_\tau, \lreveal_{\tauo},&&
      \left(
        \nonce^{j,\taut}
      \right)_{\taut \in T_\ID,\ID \in \baseiddom}\\
      &\sim\;\;&&
      \inframe_\utau,\rreveal_{\tauo},&&
      \left(
        \mac{\spair{\nonce^j}
          {\sqnsuc(\instate_\utaut(\sqn_\ue^{\nu_\taut(\ID)}))}}
        {\mkey^{\nu_\taut(\ID)}}{2}
      \right)_{\taut \in T_\ID,\ID \in \baseiddom}
    \end{alignedat}
  }
\]
We then do the same thing to replace, for every base identity $\ID$ and $\taut \in T_\ID$, the mac $\mac{\spair{\nonce^j} {\sqnsuc(\instate_\utaut(\sqn_\ue^{\nu_\taut(\ID)}))}} {\mkey^{\nu_\taut(\ID)}}{2}$ by the nonce $\nonce^{j,\taut}$ in the formula:
\[
  \ufresh{\Psi} \;\equiv\;
  \inframe_\utau,\rreveal_{\tauo},
  \left(
    \mac{\spair{\nonce^j}
      {\sqnsuc(\instate_\utaut(\sqn_\ue^{\nu_\taut(\ID)}))}}
    {\mkey^{\nu_\taut(\ID)}}{2}
  \right)_{\taut \in T_\ID,\ID \in \baseiddom}
\]
The proof is similar, we omit to check the details. This yields:
\[
  \infer[(\prfmac^2)^*]{
    \begin{alignedat}{6}
      &&&\inframe_\tau, \lreveal_{\tauo},&&
      \left(
        \nonce^{j,\taut}
      \right)_{\taut \in T_\ID,\ID \in \baseiddom}
      &\sim\;\;&&
      \inframe_\utau,\rreveal_{\tauo},&&
      \left(
        \mac{\spair{\nonce^j}
          {\sqnsuc(\instate_\utaut(\sqn_\ue^{\nu_\taut(\ID)}))}}
        {\mkey^{\nu_\taut(\ID)}}{2}
      \right)_{\taut \in T_\ID,\ID \in \baseiddom}
    \end{alignedat}
  }{
    \infer[\ax{Fresh}^*]
    {\begin{alignedat}{3}
        &&&\inframe_\tau, \lreveal_{\tauo},&&
        \left(
          \nonce^{j,\taut}
        \right)_{\taut \in T_\ID,\ID \in \baseiddom}
        &\sim\;\;&&
        \inframe_\utau,\rreveal_{\tauo},&&
        \left(
          \nonce^{j,\taut}
        \right)_{\taut \in T_\ID,\ID \in \baseiddom}
      \end{alignedat}
    }{
      \inframe_\tau, \lreveal_{\tauo}
      \;\;\sim\;\;
      \inframe_\utau,\rreveal_{\tauo}
    }
  }
\]
Combining this with \eqref{eq:pnai1-deriv6}, we get:
\[
  \begin{gathered}[c]
    \infer[\simp^*]{
      \begin{alignedat}{2}
        \inframe_\tau, \lreveal_{\tauo},t_\tau
        \sim
        \inframe_\utau,\rreveal_{\tauo},t_\utau
      \end{alignedat}
    }{
      \infer*{}{
        \infer[]{}{
          \inframe_\tau,\lreveal_{\tauo}
          \;\sim\;
          \inframe_\utau,\rreveal_{\tauo}
        }
      }
    }
    \numberthis\label{eq:pnai1-deriv7}
  \end{gathered}
\]

\paragraph{Part 6}
We now handle the $\syncdiff_\tau^\ID \sim \syncdiff_\utau^{\nu_\tau(\ID)}$ part. We first handle the case where $\instate_\tau(\sync_\ue^\ID)$ is false. Observe that $\instate_\tau(\sync_\ue^\ID) = \instate_\tauo(\sync_\ue^\ID)$, $\instate_\utau(\sync_\ue^{\nu_\tau(\ID)}) = \instate_\utauo(\sync_\ue^{\nu_\tau(\ID)})$ and that
$(\instate_\tauo(\sync_\ue^\ID), \instate_\utauo(\sync_\ue^{\nu_\tau(\ID)}))\in \reveal_{\tauo}$. Moreover:
\[
  \cond{\neg\instate_\tau(\sync_\ue^\ID)}{\syncdiff_\tau^\ID}
  \; = \; \bot
  \qquad\qquad
  \cond{\neg\instate_\utau(\sync_\ue^{\nu_\tau(\ID)})}
  {\syncdiff_\utau^{\nu_\tau(\ID)}} \;=\; \bot
\]
Hence:
\[
  \begin{alignedat}[c]{2}
    \infer[\fa^*]{
      \lreveal_{\tauo},
      \syncdiff_\tau^\ID
      \;\sim\;
      \rreveal_{\tauo},
      \syncdiff_\utau^{\nu_\tau(\ID)}
    }{
      \infer[\simp]{
        \begin{alignedat}{5}
          &&&\lreveal_{\tauo},&&
          \instate_\tau(\sync_\ue^\ID),&&
          \cond{\instate_\tau(\sync_\ue^\ID)}{\syncdiff_\tau^\ID},&&
          \cond{\neg\instate_\tau(\sync_\ue^\ID)}{\syncdiff_\tau^\ID}\\
          &\sim\;\;&&
          \rreveal_{\tauo},&&
          \instate_\utau(\sync_\ue^{\nu_\tau(\ID)}),&&
          \cond{\instate_\utau(\sync_\ue^{\nu_\tau(\ID)})}
          {\syncdiff_\utau^{\nu_\tau(\ID)}},&&
          \cond{\neg\instate_\utau(\sync_\ue^{\nu_\tau(\ID)})}
          {\syncdiff_\utau^{\nu_\tau(\ID)}}
        \end{alignedat}
      }{
        \lreveal_{\tauo},
        \cond{\instate_\tau(\sync_\ue^\ID)}{\syncdiff_\tau^\ID}
        \;\sim\;
        \rreveal_{\tauo},
        \cond{\instate_\utau(\sync_\ue^{\nu_\tau(\ID)})}
        {\syncdiff_\utau^{\nu_\tau(\ID)}}
      }
    }
  \end{alignedat}
  \numberthis\label{eq:pnai1-deriv1}
\]
Therefore we can focus on the case where $\instate_\tau(\sync_\ue^{\ID})$ is true. For all $\ID \in \baseiddom$, we let:
\[
  \incsqn_\tau^\ID \;\equiv\;
  \Geq{\pi_2(\dec(\pi_1(g(\inframe_\tau)),\sk^{\ID}_\hn))}
  {\instate_\tau(\sqn^{\ID}_\hn)}
\]
Then:
\[
  \cond{\instate_\tau(\sync_\ue^\ID)}{\syncdiff_\tau^\ID} \;=\;
  \begin{alignedat}{2}
    \lrswitch{\taut \in T_\ID}{b^\ID_\taut}
    {
      \begin{alignedat}{2}
        \ite{
          \left(\begin{alignedat}[c]{2}
              &&&\instate_\tau(\sync_\ue^\ID)
              \\ &  \wedge\, &&\incsqn_\tau^\ID
            \end{alignedat}\right)}
        {&\idiff{\instate_\tau(\sqn_\ue^\ID)}
          {\sqnsuc(\instate_\tau(\sqn_\hn^\ID))}\\ &}
        {\cond{\instate_\tau(\sync_\ue^\ID)}
          {\syncdiff_\tauo^\ID}}
      \end{alignedat}
    }
  \end{alignedat}
  \numberthis\label{eq:sync-diff-expr}
\]
And:
\begin{multline*}
  \cond{\instate_\utau(\sync_\ue^{\nu_\tau(\ID)})}
  {\syncdiff_\utau^{\nu_\tau(\ID)}} \;=\;\\
  \begin{alignedat}{2}
    \lrswitch{\taut \in T^{\uID_{l_\ID}}_\ID}{\ufresh{b}^{\nu_\tau(\ID)}_\utaut}
    {
      \begin{alignedat}{2}
        \ite{
          \left(\begin{alignedat}[c]{2}
              &&&\instate_\utau(\sync_\ue^{\nu_\tau(\ID)})
              \\ &  \wedge\, &&\incsqn_\utau^{\nu_\tau(\ID)}
            \end{alignedat}\right)}
        {&\idiff{\instate_\utau(\sqn_\ue^{\nu_\tau(\ID)})}
          {\sqnsuc(\instate_\utau(\sqn_\hn^{\nu_\tau(\ID)}))}\\ &}
        {\cond{\instate_\utau(\sync_\ue^{\nu_\tau(\ID)})}
          {\syncdiff_\utauo^{\nu_\tau(\ID)}}}
      \end{alignedat}
    }
  \end{alignedat}
  \numberthis\label{eq:sync-diff-expr2}
\end{multline*}
Take $\taut \in T_\ID$, and let $\taui$ be such that $\taui = \_,\newsession_\ID(l_\ID)$ and $\taui \popre \tau$. We have two cases:
\begin{itemize}
\item If $\taut \potau \newsession_\ID(l_\ID)$, then using \ref{b5} and \ref{b6}, we know that $\instate_\taut(\sqn_\ue^\ID) \le \instate_\taui(\sqn_\ue^\ID)$ and that $ \instate_\taut(\sync_\ue^\ID) \ra \instate_\tau(\sqn_\hn^\ID) > \instate_\taui(\sqn_\ue^\ID)$. We summarize this below:
  \begin{center}
    \begin{tikzpicture}
      [dn/.style={inner sep=0.2em,fill=black,shape=circle},
      sdn/.style={inner sep=0.15em,fill=white,draw,solid,shape=circle},
      sl/.style={decorate,decoration={snake,amplitude=1.6}},
      dl/.style={dashed},
      pin distance=0.5em,
      every pin edge/.style={thin}]

      \draw[thick] (0,0)
      node[left=1.3em] {$\tau:$}
      -- ++(0.5,0)
      node[dn,pin={above:{$\npuai{1}{\ID}{j_1}$}}]
      (a) {}
      node[below,yshift=-0.3em] {$\taut$}
      -- ++(3,0)
      node[dn,pin={above:{$\newsession_\ID(l_\ID)$}}]
      (b) {}
      node[below,yshift=-0.3em] {$\taui$}
      -- ++(3,0)
      node[dn,pin={above:{$\pnai(j,1)$}}]
      (c) {}
      node[below,yshift=-0.3em] {$\tau$};

      \path (a) -- ++ (0,-1)
      node (a1) {$\instate_\taut(\sqn_\ue^\ID)$};

      \path (b) -- ++ (0,-1)
      node (b1) {$\instate_\taui(\sqn_\ue^\ID)$};

      \path (c) -- ++ (0,-1)
      -- ++ (0,-1)
      node (c2) {$\instate_\tau(\sqn_\hn^\ID)$};

      \draw (a1) -- (b1) node[midway,above]{$\le$};

      \draw (b1) -- (c2) node[midway,above]{$<$};
    \end{tikzpicture}
  \end{center}
  Hence $\neg( b^\ID_\taut \wedge \instate_\tau(\sync_\ue^\ID) \wedge \incsqn_\tau^\ID)$.
  
  Now we look at the right protocol: since $\taut \potau \newsession_\ID(l_\ID)$, we know that $\nu_\taut(\ID) = \uID_{l_\ID - p}$ for some $p > 0$. Hence $\nu_\taut(\ID) \ne \uID_{l_\ID} = \nu_\tau(\ID)$, which implies that:
  \[
    \ufresh{b}^{\nu_\taut(\ID)}_\utaut
    \;\ra\;
    \accept_\utau^{\nu_\taut(\ID)}
    \;\ra\;
    \neg \accept_\utau^{\nu_\tau(\ID)}
    \;\ra\;
    \bigwedge_{\tautt \in T_\ID^{l_\ID}} \neg \ufresh{b}_\utautt^{\nu_\tau(\ID)}
  \]
  We deduce that:
  \begin{gather*}
    \cond{b^\ID_\taut\wedge\instate_\tau(\sync_\ue^\ID)}{\syncdiff_\tau^\ID}
    \;=\;
    \cond{b^\ID_\taut\wedge\instate_\tau(\sync_\ue^\ID)}{\syncdiff_\tauo^\ID}\\
    \cond{\ufresh{b}^{\nu_\taut(\ID)}_\utaut
      \wedge\instate_\utau(\sync_\ue^{\nu_\tau(\ID)})}
    {\syncdiff_\utau^{\nu_\tau(\ID)}}
    \;=\;
    \cond{\ufresh{b}^{\nu_\taut(\ID)}_\utaut
      \wedge\instate_\utau(\sync_\ue^{\nu_\tau(\ID)})}
    {\syncdiff_\utauo^{\nu_\tau(\ID)}}
  \end{gather*}
  Since $( \syncdiff_\tauo^\ID, \syncdiff_\utauo^{\nu_\tau(\ID)}) \in \reveal_{\tauo}$, we have:
  \[
    \begin{gathered}[c]
      \infer[\fa^*]{
        \begin{alignedat}{1}
          \lreveal_{\tauo},
          \cond{b^\ID_\taut\wedge\instate_\tau(\sync_\ue^\ID)}
          {\syncdiff_\tau^\ID}
          \;\sim\;
          \rreveal_{\tauo},
          \cond{\ufresh{b}^{\nu_\taut(\ID)}_\utaut
            \wedge\instate_\utau(\sync_\ue^{\nu_\tau(\ID)})}
          {\syncdiff_\utau^{\nu_\tau(\ID)}}
        \end{alignedat}
      }{
        \infer[\dup^*]{
          \begin{alignedat}{1}
            \lreveal_{\tauo},
            b^\ID_\taut,\instate_\tau(\sync_\ue^\ID),
            \syncdiff_\tauo^\ID
            \;\sim\;
            \rreveal_{\tauo},
            \ufresh{b}^{\nu_\taut(\ID)}_\utaut,
            \instate_\utau(\sync_\ue^{\nu_\tau(\ID)}),
            \syncdiff_\utauo^{\nu_\tau(\ID)}
          \end{alignedat}
        }{
          \lreveal_{\tauo},
          b^\ID_\taut
          \;\sim\;
          \rreveal_{\tauo},
          \ufresh{b}^{\nu_\taut(\ID)}_\utaut
        }
      }
    \end{gathered}
  \]
  Combining this with \eqref{eq:pnai1-deriv0}, we can get rid of $b^\ID_\taut \sim\ufresh{b}^{\nu_\taut(\ID)}_\utaut$:
  \[
    \begin{gathered}[c]
      \infer[\fa^*]{
        \begin{alignedat}{2}
          &&&\inframe_\tau,\lreveal_{\tauo},
          \cond{b^\ID_\taut\wedge\instate_\tau(\sync_\ue^\ID)}
          {\syncdiff_\tau^\ID}\\
          &\sim\;\;&&
          \inframe_\utau,\rreveal_{\tauo},
          \cond{\ufresh{b}^{\nu_\taut(\ID)}_\utaut
            \wedge\instate_\utau(\sync_\ue^{\nu_\tau(\ID)})}
          {\syncdiff_\utau^{\nu_\tau(\ID)}}
        \end{alignedat}
      }{
        \infer*{}{
          \infer{}{
            \inframe_\tau,\lreveal_{\tauo}
            \;\sim\;
            \inframe_\utau,\rreveal_{\tauo}
          }{}
        }
      }
    \end{gathered}
    \numberthis\label{eq:pnai1-deriv2}
  \]

\item If $\taut \not \potau \newsession_\ID(l_\ID)$, then $\nu_\taut(\ID) = \nu_\tau(\ID)$. Let $\uID = \nu_\tau(\ID)$, and using \eqref{eq:sync-diff-expr} and \eqref{eq:sync-diff-expr2} we get that:
  \begin{alignat*}{4}
    &\cond{b^\ID_\taut\wedge\instate_\tau(\sync_\ue^\ID)}
    {\syncdiff_\tau^\ID}
    &\;= &&&
    \lrpcond{b^\ID_\taut\wedge\instate_\tau(\sync_\ue^\ID)}
    {\idiff{\instate_\tau(\sqn_\ue^\ID)}
      {\sqnsuc(\instate_\tau(\sqn_\hn^\ID))}}\\
    &&&& +\;&
    \ite{
      b^\ID_\taut\wedge
      \instate_\tau(\sync_\ue^\ID) \wedge
      \incsqn_\tau^\ID}
    {\mone}{\zero}\\
    &\cond{\ufresh{b}^{\uID}_\utaut
      \wedge\instate_\utau(\sync_\ue^{\uID})}
    {\syncdiff_\utau^{\uID}}
    &\;= &&&
    \lrpcond{\ufresh{b}^{\uID}_\utaut
      \wedge\instate_\utau(\sync_\ue^{\uID})}
    {\idiff{\instate_\utau(\sqn_\ue^{\uID})}
      {\sqnsuc(\instate_\utau(\sqn_\hn^{\uID}))}}\\
    &&&& +\;&
    \ite{
      \ufresh{b}^{\uID}_\utaut \wedge
      \instate_\utau(\sync_\ue^{\uID})\wedge
      \incsqn_\utau^{\uID}}
    {\mone}{\zero}
  \end{alignat*}
  Hence using \eqref{eq:pnai1-deriv0} we get:
  \[
    \begin{gathered}[c]
      \infer[\fa^*]{
        \inframe_\tau,\lreveal_{\tauo},
        \cond{b^\ID_\taut\wedge\instate_\tau(\sync_\ue^\ID)}
        {\syncdiff_\tau^\ID}
        \;\sim\;
        \inframe_\utau,\rreveal_{\tauo},
        \cond{\ufresh{b}^{\uID}_\utaut
          \wedge\instate_\utau(\sync_\ue^{\uID})}
        {\syncdiff_\utau^{\uID}}
      }{
        \infer[\dup]{
          \begin{alignedat}{2}
            &&&
            \inframe_\tau,\lreveal_{\tauo},
            b^\ID_\taut,\instate_\tau(\sync_\ue^\ID),
            {\idiff{\instate_\tau(\sqn_\ue^\ID)}
              {\instate_\tau(\sqn_\hn^\ID)}},
            b^\ID_\taut\wedge
            \instate_\tau(\sync_\ue^\ID) \wedge
            \incsqn_\tau^\ID\\
            &\sim\;\;&&
            \inframe_\utau,\rreveal_{\tauo},
            \ufresh{b}^{\uID}_\utaut,
            \instate_\utau(\sync_\ue^{\uID}),
            {\idiff{\instate_\utau(\sqn_\ue^{\uID})}
              {\instate_\utau(\sqn_\hn^{\uID})}},
            \ufresh{b}^{\uID}_\utaut \wedge
            \instate_\utau(\sync_\ue^{\uID})\wedge
            \incsqn_\utau^{\uID}
          \end{alignedat}
        }{
          \inframe_\tau,\lreveal_{\tauo},
          b^\ID_\taut\wedge
          \instate_\tau(\sync_\ue^\ID) \wedge
          \incsqn_\tau^\ID
          \;\sim\;
          \inframe_\utau,\rreveal_{\tauo},
          \ufresh{b}^{\uID}_\utaut \wedge
          \instate_\utau(\sync_\ue^{\uID})\wedge
          \incsqn_\utau^{\uID}
        }
      }
    \end{gathered}
    \numberthis\label{eq:pnai1-deriv30}
  \]
  We split the proof in two, depending on whether $\instate_\taut(\sync_\ue^\ID)$ is true or not.
  \begin{itemize}
  \item If it is true, this is simple:
    \begin{gather*}
      \left(
        \instate_\taut(\sync_\ue^\ID) \wedge
        b^\ID_\taut\wedge
        \instate_\tau(\sync_\ue^\ID) \wedge
        \incsqn_\tau^{\ID}
      \right)
      \lra
      \left(
        b^\ID_\taut\wedge
        \instate_\taut(\sync_\ue^\ID) \wedge
        \instate_\taut(\sqn_\ue^\ID) <
        \instate_\tau(\sqn_\hn^\ID)
      \right)\\
      \left(
        \instate_\utaut(\sync_\ue^{\uID}) \wedge
        \ufresh{b}^{\uID}_\utaut \wedge
        \instate_\utau(\sync_\ue^{\uID})\wedge
        \incsqn_\utau^{\uID}
      \right)
      \lra
      \left(
        \ufresh{b}^{\uID}_\utaut \wedge
        \instate_\utaut(\sync_\ue^{\uID})\wedge
        \instate_\utaut(\sqn_\ue^{\uID}) <
        \instate_\utau(\sqn_\hn^{\uID})
      \right)
    \end{gather*}
    Hence using \eqref{eq:pnai1-deriv0} we get:
    \[
      \infer[R]{
        \begin{alignedat}{2}
          &&&\inframe_\tau,\lreveal_{\tauo},
          \instate_\taut(\sync_\ue^\ID) \wedge
          b^\ID_\taut\wedge
          \instate_\tau(\sync_\ue^\ID) \wedge
          \incsqn_\tau^\ID\\
          &\sim\;\;&&
          \inframe_\utau,\rreveal_{\tauo},
          \instate_\utaut(\sync_\ue^{\uID}) \wedge
          \ufresh{b}^{\uID}_\utaut \wedge
          \instate_\utau(\sync_\ue^{\uID})\wedge
          \incsqn_\utau^{\uID}
        \end{alignedat}
      }{
        \infer[\simp]{
          \begin{alignedat}{2}
            &&&\inframe_\tau,\lreveal_{\tauo},
            b^\ID_\taut\wedge
            \instate_\taut(\sync_\ue^\ID) \wedge
            \instate_\taut(\sqn_\ue^\ID) <
            \instate_\tau(\sqn_\hn^\ID)\\
            &\sim\;\;&&
            \inframe_\utau,\rreveal_{\tauo},
            \ufresh{b}^{\uID}_\utaut \wedge
            \instate_\utaut(\sync_\ue^{\uID})\wedge
            \instate_\utaut(\sqn_\ue^{\uID}) <
            \instate_\utau(\sqn_\hn^{\uID})
          \end{alignedat}
        }{
          \begin{alignedat}{2}
            &&&\inframe_\tau,\lreveal_{\tauo},
            \instate_\taut(\sync_\ue^\ID) \wedge
            \instate_\taut(\sqn_\ue^\ID) <
            \instate_\tau(\sqn_\hn^\ID)\\
            &\sim\;\;&&
            \inframe_\utau,\rreveal_{\tauo},
            \instate_\utaut(\sync_\ue^{\uID})\wedge
            \instate_\utaut(\sqn_\ue^{\uID}) <
            \instate_\utau(\sqn_\hn^{\uID})
          \end{alignedat}
        }
      }
    \]
    We conclude the case $\instate_\taut(\sync_\ue^{\ID})$ using \ref{der1}:
    \[
      \begin{gathered}
        \infer[\simp]{
          \begin{alignedat}{2}
            &&&\inframe_\tau,\lreveal_{\tauo},
            \instate_\taut(\sync_\ue^\ID) \wedge
            \instate_\taut(\sqn_\ue^\ID) <
            \instate_\tau(\sqn_\hn^\ID)\\
            &\sim\;\;&&
            \inframe_\utau,\rreveal_{\tauo},
            \instate_\utaut(\sync_\ue^{\uID})\wedge
            \instate_\utaut(\sqn_\ue^{\uID}) <
            \instate_\utau(\sqn_\hn^{\uID})
          \end{alignedat}
        }{
          \lreveal_{\tauo}
          \;\sim\;
          \rreveal_{\tauo}
        }
      \end{gathered}
      \numberthis\label{eq:gasdjnlsfhglfzzgfnlg}
    \]

  \item If $\sync_\ue^\ID$ is false at $\taut$ and true at $\tau$, then we know that there is an instant $\taut \popreleq \taua$ such that $\neg\instate_\taua(\sync_\ue^\ID) \wedge \instate_\taua(\sync_\ue^\ID)$. Since $\sync_\ue^\ID$ is only updated at instant $\npuai{\_}{\ID}{\_}$ and $\ns_\ID(\_)$, and since $\taut \not \potau \ns_\ID(\_)$, the only possibilities are $\taua$ of the form $\_,\npuai{2}{\ID}{j_a}$. In that case, we must have $\accept_\taua^\ID$. Formally, it is straightforward to show by induction that:
    \[
      \left(
        b^\ID_\taut\wedge
        \neg \instate_\taut(\sync_\ue^\ID) \wedge
        \instate_\tau(\sync_\ue^\ID)
      \right)
      \;\ra\;
      \bigvee_{\taua = \_,\npuai{2}{\ID}{j_a} \atop{\taut \potau \taua}}
      \neg \instate_\taua(\sync_\ue^\ID) \wedge
      \accept_\taua^\ID
      \numberthis\label{eq:fojvsqroadsjaspcv}
    \]
    Using \ref{sequ4}, we know that:
    \[
      \accept_\taua^\ID \wedge \neg \instate_\taua(\sync_\ue^\ID)
      \;\ra\;
      \cstate_\taua(\sqn_\ue^\ID) = \cstate_\taua(\sqn_\hn^\ID)
    \]
    We know that $\cstate_\taua(\sqn_\ue^\ID) = \instate_\taua(\sqn_\ue^\ID)$ and $\cstate_\taua(\sqn_\hn^\ID) = \instate_\taua(\sqn_\hn^\ID)$. Moreover using \ref{b5} we have:
    \begin{mathpar}
      \cstate_\taut(\sqn_\ue^\ID) \le \cstate_\taua(\sqn_\ue^\ID)

      \cstate_\taua(\sqn_\hn^\ID) \le \instate_\tau(\sqn_\hn^\ID)
    \end{mathpar}
    Finally, we know that $\cstate_\taut(\sqn_\ue^\ID) = \instate_\taut(\sqn_\ue^\ID) + 1$, and therefore $\cstate_\taut(\sqn_\ue^\ID) > \instate_\taut(\sqn_\ue^\ID)$. We summarize this graphically:
    \begin{center}
      \begin{tikzpicture}
        [dn/.style={inner sep=0.2em,fill=black,shape=circle},
        sdn/.style={inner sep=0.15em,fill=white,draw,solid,shape=circle},
        sl/.style={decorate,decoration={snake,amplitude=1.6}},
        dl/.style={dashed},
        pin distance=0.5em,
        every pin edge/.style={thin}]

        \draw[thick] (0,0)
        node[left=1.3em] {$\tau:$}
        -- ++(0.5,0)
        node[dn,pin={above,align=left:{$\ns_\ID(\_)$\\or $\epsilon$}}]
        (a) {}
        node[below,yshift=-0.3em] {$\taui$}
        -- ++(3,0)
        node[dn,pin={above:{$\npuai{1}{\ID}{j_1}$}}]
        (b) {}
        node[below,yshift=-0.3em,name=b0] {$\taut$}
        -- ++(3,0)
        node[dn,pin={above:{$\npuai{2}{\ID}{j_a}$}}]
        (c) {}
        node[below,yshift=-0.3em] {$\taua$}
        -- ++(3,0)
        node[dn,pin={above:{$\pnai(j,1)$}}]
        (d) {}
        node[below,yshift=-0.3em,name=d0] {$\tau$};

        \draw[thin,dashed] (b0) -- ++(0,-0.5) -| (d0);

        \path (b) -- ++ (0,-1.8)
        -- ++ (0,-1.3)
        node (b2) {$\instate_\taut(\sqn_\ue^\ID)$};

        \path (c) -- ++ (0,-1.8)
        node (c1) {$\instate_\taua(\sqn_\hn^\ID)$}
        -- ++ (0,-1.3)
        node (c2) {$\instate_\taua(\sqn_\ue^\ID)$};

        \path (d) -- ++ (0,-1.8)
        node (d1) {$\instate_\tau(\sqn_\hn^\ID)$};

        \draw (b2) -- (c2) node[midway,below]{$<$}
        (c2) -- (c1) node[midway,below,sloped]{$=$}
        (c1) -- (d1) node[midway,above]{$\le$};
      \end{tikzpicture}
    \end{center}
    Therefore:
    \[
      \left(
        \neg \instate_\taua(\sync_\ue^\ID) \wedge
        \accept_\taua^\ID
      \right)
      \;\ra\;
      \instate_\taut(\sync_\ue^\ID) <
      \instate_\taua(\sync_\hn^\ID)
    \]
    Hence we deduce from \eqref{eq:fojvsqroadsjaspcv} that:
    \[
      \left(
        b^\ID_\taut\wedge
        \neg \instate_\taut(\sync_\ue^\ID) \wedge
        \instate_\tau(\sync_\ue^\ID)
      \right)
      \;\ra\;
      \incsqn_\tau^{\ID}
    \]
    Similarly, we show that:
    \[
      \left(
        \ufresh{b}^{\uID}_\utaut \wedge
        \neg \instate_\utaut(\sync_\ue^{\uID}) \wedge
        \instate_\utau(\sync_\ue^{\uID})
      \right)
      \;\ra\;
      \incsqn_\utau^{\uID}
    \]
    Hence using \eqref{eq:pnai1-deriv0} we get:
    \[
      \begin{gathered}
        \infer[R]{
          \begin{alignedat}{2}
            &&&\inframe_\tau,\lreveal_{\tauo},
            \neg\instate_\taut(\sync_\ue^\ID) \wedge
            b^\ID_\taut\wedge
            \instate_\tau(\sync_\ue^\ID) \wedge
            \incsqn_\tau^\ID\\
            &\sim\;\;&&
            \inframe_\utau,\rreveal_{\tauo},
            \neg\instate_\utaut(\sync_\ue^{\uID}) \wedge
            \ufresh{b}^{\uID}_\utaut \wedge
            \instate_\utau(\sync_\ue^{\uID})\wedge
            \incsqn_\utau^{\uID}
          \end{alignedat}
        }{
          \infer[\simp]{
            \begin{alignedat}{2}
              &&&\inframe_\tau,\lreveal_{\tauo},
              \neg\instate_\taut(\sync_\ue^\ID) \wedge
              b^\ID_\taut\wedge
              \instate_\tau(\sync_\ue^\ID)\\
              &\sim\;\;&&
              \inframe_\utau,\rreveal_{\tauo},
              \neg\instate_\utaut(\sync_\ue^{\uID}) \wedge
              \ufresh{b}^{\uID}_\utaut \wedge
              \instate_\utau(\sync_\ue^{\uID})
            \end{alignedat}
          }{
            \infer[\dup^*]{
              \begin{alignedat}{2}
                \inframe_\tau,\lreveal_{\tauo},
                \instate_\taut(\sync_\ue^\ID),
                b^\ID_\taut,
                \instate_\tau(\sync_\ue^\ID)
                \;\sim\;
                \inframe_\utau,\rreveal_{\tauo},
                \instate_\utaut(\sync_\ue^{\uID}),
                \ufresh{b}^{\uID}_\utaut,
                \instate_\utau(\sync_\ue^{\uID})
              \end{alignedat}
            }{
              \begin{alignedat}{2}
                \inframe_\tau,\lreveal_{\tauo}
                \;\sim\;
                \inframe_\utau,\rreveal_{\tauo}
              \end{alignedat}
            }
          }
        }
      \end{gathered}
      \numberthis\label{eq:sdbasdgoharhagauv}
    \]
  \end{itemize}
  Combining \eqref{eq:gasdjnlsfhglfzzgfnlg}, \eqref{eq:sdbasdgoharhagauv} with \eqref{eq:pnai1-deriv0} and \eqref{eq:pnai1-deriv30}, it is easy to build a derivation of the form:
  \[
    \begin{gathered}[c]
      \infer[\fa^*]{
        \begin{alignedat}{2}
          &&&\inframe_\tau,\lreveal_{\tauo},
          \cond{b^\ID_\taut\wedge\instate_\tau(\sync_\ue^\ID)}
          {\syncdiff_\tau^\ID}\\
          &\sim\;\;&&
          \inframe_\utau,\rreveal_{\tauo},
          \cond{\ufresh{b}^{\nu_\taut(\ID)}_\utaut
            \wedge\instate_\utau(\sync_\ue^{\uID})}
          {\syncdiff_\utau^{\uID}}
        \end{alignedat}
      }{
        \infer*{}{
          \infer{}{
            \inframe_\tau,\lreveal_{\tauo}
            \;\sim\;
            \inframe_\utau,\rreveal_{\tauo}
          }{}
        }
      }
    \end{gathered}
    \numberthis\label{eq:pnai1-deriv8}
  \]

\end{itemize}

\paragraph{Part 7} Now it only remains to put everything together. First combining \eqref{eq:pnai1-deriv0}, \eqref{eq:pnai1-deriv2} and \eqref{eq:pnai1-deriv8}, we get:
\[
  \infer[\fa^*]{
    \inframe_\tau,
    \lreveal_{\tauo},
    \cond{\instate_\tau(\sync_\ue^\ID)}{\syncdiff_\tau^\ID}
    \;\sim\;
    \inframe_\utau,
    \rreveal_{\tauo},
    \cond{\instate_\utau(\sync_\ue^{\uID})}
    {\syncdiff_\utau^{\uID}}
  }{
    \infer{
      \begin{alignedat}{4}
        &&&\inframe_\tau,&&
        \lreveal_{\tauo},&&
        \left(
          b_\taut^\ID,
          \cond{\instate_\tau(\sync_\ue^\ID)\wedge b_\taut^\ID}
          {\syncdiff_\tau^\ID}
        \right)_{\taut \in T_\ID}\\
        &\sim\;\;&&
        \inframe_\utau,&&
        \rreveal_{\tauo},&&
        \left(
          \ufresh{b}_\utaut^{\nu_\taut(\ID)},
          \cond{\instate_\utau(\sync_\ue^{\uID})\wedge
            \ufresh{b}_\utaut^{\nu_\taut(\ID)}}
          {\syncdiff_\utau^{\uID}}
        \right)_{\taut \in T_\ID}
      \end{alignedat}}{
      \infer*{}{
        \infer{}{
          \inframe_\tau,
          \lreveal_{\tauo}
          \;\;\sim\;\;
          \inframe_\utau,
          \rreveal_{\tauo}
        }
      }
    }
  }
\]
Combine with \eqref{eq:pnai1-deriv1}, this yields:
\[
  \begin{gathered}[c]
    \infer[\fa^*]{
      \inframe_\tau,
      \lreveal_{\tauo},
      \syncdiff_\tau^\ID
      \;\sim\;
      \inframe_\utau,
      \rreveal_{\tauo},
      \syncdiff_\utau^{\uID}
    }{
      \infer*{}{
        \infer{}{
          \inframe_\tau,
          \lreveal_{\tauo}
          \;\;\sim\;\;
          \inframe_\utau,
          \rreveal_{\tauo}
        }
      }
    }
  \end{gathered}
  \numberthis\label{eq:pnai1-deriv9}
\]
We conclude the proof of this case by combining \eqref{eq:pnai1-deriv4}, \eqref{eq:pnai1-deriv7}, and \eqref{eq:pnai1-deriv9} (recall that the $\macsym$ in $\reveal_{\tau} \backslash \reveal_{\tauo}$ where handled in \eqref{eq:pnai1-deriv10}).

\subsection{Case $\ai = \npuai{2}{\ID}{j}$}
We know that $\uai = \npuai{2}{\nu_\tau(\ID)}{j}$. Here $\lreveal_{\tau}$ and $\lreveal_{\tauo}$ coincides everywhere except on the pairs:
\begin{mathpar}
  \syncdiff_\tau^\ID
  \;\;\sim\;\;
  \syncdiff_\utau^{\nu_\tau(\ID)}

  \cstate_\tau(\eauth_\ue^\ID)
  \;\;\sim\;\;
  \cstate_\utau(\eauth_\ue^{\nu_\tau(\ID)})

  \cstate_\tau(\sync_\ue^\ID)
  \;\;\sim\;\;
  \cstate_\utau(\sync_\ue^{\nu_\tau(\ID)})
\end{mathpar}
Therefore we are looking for a derivation of:
\begin{equation}
  \Phi \;\equiv\;
  \begin{alignedat}{2}
    &&&\inframe_\tau,\lreveal_{\tauo},
    \syncdiff_\tau^\ID,
    \cstate_\tau(\eauth_\ue^\ID),
    \cstate_\tau(\sync_\ue^{\ID}),
    \accept_\tau^\ID\\
    &\;\;\sim\;\;&&
    \inframe_\utau,\rreveal_{\tauo},
    \syncdiff_\utau^{\nu_\tau(\ID)},
    \cstate_\utau(\eauth_\ue^{\nu_\tau(\ID)}),
    \cstate_\utau(\sync_\ue^{\nu_\tau(\ID)}),
    \accept_\utau^{\nu_\tau(\ID)}
  \end{alignedat}
  \label{eq:puaij1}
\end{equation}
Let $\tautt = \_,\npuai{1}{\ID}{j} \popre \tau$. We know that $\tautt \not \potau \newsession_\ID(\_)$, and therefore $\utautt = \_, \npuai{1}{\nu_\tau(\ID)}{j}$. We also know that:
\[
  \instate_\tau(\bauth_\ue^\ID)
  \equiv \cstate_\tautt(\bauth_\ue^\ID)
  \equiv g(\inframe_\tautt)
  \qquad\qquad
  \instate_\utau(\bauth_\ue^{\nu_\tau(\ID)})
  \equiv \cstate_\utautt(\bauth_\ue^{\nu_\tau(\ID)})
  \equiv g(\inframe_\utautt)
\]
We summarize this graphically:
\begin{center}
  \begin{tikzpicture}
    [dn/.style={inner sep=0.2em,fill=black,shape=circle},
    sdn/.style={inner sep=0.15em,fill=white,draw,solid,shape=circle},
    sl/.style={decorate,decoration={snake,amplitude=1.6}},
    dl/.style={dashed},
    pin distance=0.5em,
    every pin edge/.style={thin}]

    \draw[thick] (0,0)
    node[left=1.3em] {$\tau:$}
    -- ++(0.5,0)
    node[dn,pin={above:{$\tautt = \_,\npuai{1}{\ID}{j}$}}]
    (b) {}
    -- ++(4.5,0)
    node[dn,pin={above:{$\tau = \_\npuai{2}{\ID}{j}$}}]
    (c) {};

    \path (b) -- ++ (0,-0.6)
    node[sdn] (b1) {}
    node[left,scale=0.9, transform shape]
    {$\cstate_\tautt(\bauth_\ue^\ID) = g(\inframe_\tautt)$};

    \path (c) -- ++ (0,-0.6)
    node[sdn] (c1) {};

    \draw[thick] (0,-3)
    node[left=1.3em] {$\utau:$}
    -- ++(0.5,0)
    node[dn,pin={below:{$\utautt = \_,\npuai{1}{{\nu_\tau(\ID)}}{j}$}}]
    (bb) {}
    -- ++(4.5,0)
    node[dn,pin={below:{$\utau = \_,\npuai{2}{\nu_\tau(\ID)}{j}$}}]
    (cc) {};

    \path (bb) -- ++ (0,+0.6)
    node[sdn] (bb1) {}
    node[left,scale=0.9, transform shape]
    {$\cstate_\utautt(\bauth_\ue^{\nu_\tau(\ID)}) =
      g(\inframe_\utautt)$};

    \path (cc) -- ++ (0,+0.6)
    node[sdn] (cc1) {};

    \draw[sl] (b1) -- (c1);
    \draw[sl] (bb1) -- (cc1);
    \path (b) -- (c) node[midway] (d) {};
    \path (bb) -- (cc) node[midway] (dd) {};
    \path (d) -- (dd) node[midway,transform shape,scale=2]
    {$\mathbf{\sim}$};
  \end{tikzpicture}
\end{center}

Hence we can start deconstructing the terms using $\fa$ and simplifying with $\dup$:
\[
  \infer[\simp]{
    \Phi
  }{
    \infer[\simp]{
      \begin{alignedat}{2}
        &&&\inframe_\tau,\lreveal_{\tauo},
        \syncdiff_\utau^{\ID},
        \accept_\tau^\ID,
        g(\cframe_\tautt)\\
        &\sim\;\;&&
        \inframe_\utau,\rreveal_{\tauo},
        \syncdiff_\utau^{\nu_\tau(\ID)},
        \accept_\utau^{\nu_\tau(\ID)},
        g(\cframe_\utautt)
      \end{alignedat}
    }{
      \inframe_\tau,\lreveal_{\tauo},
      \syncdiff_\tau^{\ID},
      \accept_\tau^\ID\\
      \;\;\sim\;\;
      \inframe_\utau,\rreveal_{\tauo},
      \syncdiff_\utau^{\nu_\tau(\ID)},
      \accept_\utau^{\nu_\tau(\ID)}
    }
  }
\]
\paragraph{Part 1}
We now focus on $\accept_\tau^\ID$. Let:
\[
  T = \{ \taut \mid
  \taut = \_, \pnai(j_1,1) \wedge \tautt \potau \taut \popre \tau\}
\]
Using \ref{equ2} we know that:
\begin{alignat*}{2}
  \accept_\tau^\ID
  &\;\lra\;&&
  \bigvee_{\taut = \_,\pnai(j_1,1) \in T}
  \underbrace{\left(
      \begin{alignedat}{2}
        &&&
        g(\inframe_\tau) =
        \mac{\spair
          {\nonce^{j_1}}
          {\sqnsuc(\instate_\tautt(\sqn_\ue^\ID))}}
        {\mkey^\ID}{2}
        \;\wedge\;
        g(\inframe_\tautt) = \nonce^{j_1} \\
        &\wedge\;&&
        \pi_1(g(\inframe_\taut)) =
        \enc{\spair{\ID}{\instate_\tautt(\sqn_\ue^\ID)}}{\pk_\hn}{\enonce^j}
      \end{alignedat}
    \right)}_{\supitr{\tautt,\tau}{\taut}}
  \numberthis\label{eq:p10b}
\end{alignat*}
Using again \ref{equ2} on $\utau$ (which is a valid symbolic trace) we also have:
\begin{alignat*}{2}
  \accept_\utau^{\nu_\tau(\ID)}
  &\;\lra\;&&
  \bigvee_{\taut = \_,\pnai(j_1,1) \in T}
  \underbrace{\left(
      \begin{alignedat}{2}
        &&&
        g(\inframe_\utau) =
        \mac{\spair
          {\nonce^{j_1}}
          {\sqnsuc(\instate_\utautt(\sqn_\ue^{\nu_\tau(\ID)}))}}
        {\mkey^{\nu_\tau(\ID)}}{2}
        \;\wedge\;
        g(\inframe_\utautt) = \nonce^{j_1} \\
        &\wedge\;&&
        \pi_1(g(\inframe_\utaut)) =
        \enc{\spair
          {\ID^{\nu_\tau(\ID)}}
          {\instate_\utautt(\sqn_\ue^{\nu_\tau(\ID)})}}
        {\pk_\hn}{\enonce^j}
      \end{alignedat}
    \right)}_{\supitr{\utautt,\utau}{\utaut}}
\end{alignat*}

\paragraph{Part 2}
We focus on $\syncdiff_\tau^\ID$. First we get rid of the case where $\instate_\tau(\sync_\ue^\ID)$ is true. Indeed, we have:
\begin{mathpar}
  \cond{\instate_\tau(\sync_\ue^\ID)}
  {\syncdiff_\tau^{\ID}}
  \; = \;
  \cond{\instate_\tau(\sync_\ue^\ID)}
  {\sqnsuc(\syncdiff_\tauo^{\ID})}

  \cond{\instate_\utau(\sync_\ue^{\nu_\tau(\ID)})}
  {\syncdiff_\utau^{\nu_\tau(\ID)}}
  \; = \;
  \cond{\instate_\utau(\sync_\ue^{\nu_\tau(\ID)})}
  {\sqnsuc(\syncdiff_\utauo^{\nu_\tau(\ID)})}
\end{mathpar}
And:
\begin{mathpar}
  \left(
    \syncdiff_\tauo^{\ID},\syncdiff_\utauo^{\nu_\tau(\ID)}
  \right)
  \in \reveal_{\tauo}

  \left(
    \instate_\tau(\sync_\ue^\ID),
    \instate_\utau(\sync_\ue^{\nu_\tau(\ID)})
  \right)
  \in \reveal_{\tauo}
\end{mathpar}
Therefore:
\[
  \infer[\simp]{
    \inframe_\tau,\lreveal_{\tauo},
    \syncdiff_\tau^{\ID}
    \;\;\sim\;\;
    \inframe_\utau,\rreveal_{\tauo},
    \syncdiff_\utau^{\nu_\tau(\ID)}
  }{
    \inframe_\tau,\lreveal_{\tauo},
    \cond{\neg\instate_\tau(\sync_\ue^\ID)}
    {\syncdiff_\tau^{\ID}}
    \;\;\sim\;\;
    \inframe_\utau,\rreveal_{\tauo},
    \cond{\neg\instate_\utau(\sync_\ue^{\nu_\tau(\ID)})}
    {\syncdiff_\utau^{\nu_\tau(\ID)}}
  }
\]
Similarly:
\begin{mathpar}
  \cond{\neg \instate_\tau(\sync_\ue^\ID)
    \wedge \neg \accept_\tau^\ID}
  {\syncdiff_\tau^{\ID}}
  \; = \;
  \bot

  \cond{\neg \instate_\utau(\sync_\ue^{\nu_\tau(\ID)})
    \wedge \neg \accept_\utau^{\nu_\tau(\ID)}}
  {\syncdiff_\utau^{\nu_\tau(\ID)}}
  \; = \;
  \bot
\end{mathpar}
Hence we can go one step further:
\[
  \begin{gathered}
    \infer[\simp]{
      \begin{alignedat}{4}
        &&&\inframe_\tau,\lreveal_{\tauo},&&
        \accept_\tau^\ID,&&
        {\syncdiff_\tau^{\ID}}\\
        &\sim\;\;&&
        \inframe_\utau,\rreveal_{\tauo},&&
        \accept_\utau^{\nu_\tau(\ID)},&&
        {\syncdiff_\utau^{\nu_\tau(\ID)}}
      \end{alignedat}
    }{
      \infer[\simp^*]{
        \begin{alignedat}{4}
          &&&\inframe_\tau,\lreveal_{\tauo},&&
          \accept_\tau^\ID,&&
          \cond{\neg\instate_\tau(\sync_\ue^\ID)
            \wedge \accept_\tau^\ID}
          {\syncdiff_\tau^{\ID}}\\
          &\sim\;\;&&
          \inframe_\utau,\rreveal_{\tauo},&&
          \accept_\utau^{\nu_\tau(\ID)},&&
          \cond{\neg\instate_\utau(\sync_\ue^{\nu_\tau(\ID)})
            \wedge \accept_\utau^{\nu_\tau(\ID)}}
          {\syncdiff_\utau^{\nu_\tau(\ID)}}
        \end{alignedat}
      }{
        \begin{alignedat}{4}
          &&&\inframe_\tau,\lreveal_{\tauo},&&
          \left(\supitr{\tautt,\tau}{\taut},
            \cond{\neg\instate_\tau(\sync_\ue^\ID)
              \wedge b_\taut}
            {\syncdiff_\tau^{\ID}}\right)_{\taut \in T}\\
          &\sim\;\;&&
          \inframe_\utau,\rreveal_{\tauo},&&
          \left(\supitr{\utautt,\utau}{\utaut},
            \cond{\neg\instate_\utau(\sync_\ue^{\nu_\tau(\ID)})
              \wedge \ufresh{b}_\utaut,}
            {\syncdiff_\utau^{\nu_\tau(\ID)}}
          \right)_{\taut \in T}
        \end{alignedat}
      }
    }
  \end{gathered}
  \numberthis\label{eq:adasdsads}
\]

\paragraph{Part 3}
Using \ref{sequ4} twice, we know that for every $\taut \in T$:
\[
  \left(
    \neg\instate_\tau(\sync_\ue^\ID)
    \wedge \supitr{\tautt,\tau}{\taut}
  \right)
  \;\ra\;
  \syncdiff_\tau^{\ID} = \zero
  \numberthis\label{eq:diff-zero}
\]
And that:
\[
  \left(
    \neg\instate_\utau(\sync_\ue^{\nu_\tau(\ID)})
    \wedge\supitr{\utautt,\utau}{\utaut}
  \right)
  \;\ra\;
  \syncdiff_\utau^{\nu_\tau(\ID)} = \zero
  \numberthis\label{eq:diff-zero1}
\]
Using \eqref{eq:diff-zero} and \eqref{eq:diff-zero1}, we can extend the derivation in \eqref{eq:adasdsads}:
\[
  \begin{gathered}
    \infer[\simp]{
      \begin{alignedat}{4}
        &&&\inframe_\tau,\lreveal_{\tauo},&&
        \accept_\tau^\ID,&&
        {\syncdiff_\tau^{\ID}}\\
        &\sim\;\;&&
        \inframe_\utau,\rreveal_{\tauo},&&
        \accept_\utau^{\nu_\tau(\ID)},\;&&
        {\syncdiff_\utau^{\nu_\tau(\ID)}}
      \end{alignedat}
    }{
      \infer[\simp]{
        \begin{alignedat}{4}
          &&&\inframe_\tau,\lreveal_{\tauo},&&
          \left(\supitr{\tautt,\tau}{\taut},
            \cond{\neg\instate_\tau(\sync_\ue^\ID)
              \wedge \supitr{\tautt,\tau}{\taut}}
            {\zero}\right)_{\taut \in T}\\
          &\sim\;\;&&
          \inframe_\utau,\rreveal_{\tauo},&&
          \left(\supitr{\utautt,\utau}{\utaut},
            \cond{\neg\instate_\utau(\sync_\ue^{\nu_\tau(\ID)})
              \wedge \supitr{\utautt,\utau}{\utaut}}
            {\zero}
          \right)_{\taut \in T}
        \end{alignedat}
      }{
        \inframe_\tau,\lreveal_{\tauo},
        \left(\supitr{\tautt,\tau}{\taut}\right)_{\taut \in T}
        \;\;\sim\;\;
        \inframe_\utau,\rreveal_{\tauo},
        \left(\supitr{\utautt,\utau}{\utaut}\right)_{\taut \in T}
      }
    }
  \end{gathered}
  \numberthis\label{eq:qeamvocvm}
\]
We can check that for all $\taut = \_,\pnai(j_1,1) \in T$, since $\tautt \potau \taut$ and $\tautt \not \potau \newsession_\ID(\_)$ we have that:
\begin{mathpar}
  \left(
    \mac{\spair
      {\nonce^{j_1}}
      {\sqnsuc(\instate_\tautt(\sqn_\ue^\ID))}}
    {\mkey^\ID}{2},
    \mac{\spair
      {\nonce^{j_1}}
      {\sqnsuc(\instate_\utautt(\sqn_\ue^{\nu_\tau(\ID)}))}}
    {\mkey^{\nu_\tau(\ID)}}{2}
  \right) \in \reveal_{\tauo}

  \left(
    \nonce^{j_1},
    \nonce^{j_1}
  \right) \in \reveal_{\tauo}

  \left(
    \enc{\spair
      {\ID}
      {\instate_\tautt(\sqn_\ue^\ID)}}
    {\pk_\hn}{\enonce^j},
    \enc{\spair
      {\ID^{\nu_\tau(\ID)}}
      {\instate_\utautt(\sqn_\ue^{\nu_\tau(\ID)})}}
    {\pk_\hn}{\enonce^j}
  \right) \in \reveal_{\tauo}
\end{mathpar}
We can complete the derivation in \eqref{eq:qeamvocvm}: first, for every $\taut \in T$, we deconstruct  $b_\taut \sim \ufresh{b}_\utaut$ with $\fa$; and then, we absorb the subterms into $\reveal_{\tauo}$ using rule $\dup$ (which is sound using the remark above). This yields:
\[
  \infer[\simp]{
    \begin{alignedat}{4}
      &&&\inframe_\tau,\lreveal_{\tauo},&&
      \accept_\tau^\ID,&&
      {\syncdiff_\tau^{\ID}}\\
      &\sim\;\;&&
      \inframe_\utau,\rreveal_{\tauo},&&
      \accept_\utau^{\nu_\tau(\ID)},&&
      {\syncdiff_\utau^{\nu_\tau(\ID)}}
    \end{alignedat}
  }{
    \inframe_\tau,\lreveal_{\tauo}
    \;\;\sim\;\;
    \inframe_\utau,\rreveal_{\tauo}
  }
\]
Finally we conclude using the induction hypothesis.

\subsection{Case $\ai = \fnai(j)$}
We know that $\ai = \fnai(j)$. Here $\lreveal_{\tau}$ and $\lreveal_{\tauo}$ coincides everywhere except on the following pairs: for every base identity $\ID$:
\begin{alignat*}{2}
  \suci^j
  & \;\;\sim\;\;&&
  \suci^j\\
  \cond{\neauth_\tau(\ID,j)}
  {\left(\tsuci_\tau(\ID,j)\right)}
  & \;\;\sim\;\;&&
  \cond{\uneauth_\utau(\ID,j)}
  {\left(\utsuci_\utau(\ID,j)\right)}\\
  \cond{\neauth_\tau(\ID,j)}
  {\left(\tmac_\tau(\ID,j)\right)}
  & \;\;\sim\;\;&&
  \cond{\uneauth_\utau(\ID,j)}
  {\left(\utmac_\utau(\ID,j)\right)}
\end{alignat*}
\paragraph{Part 1}
Let $\ID \in \iddom$. Using Lemma~\ref{lem:auth-serv-net}, we know that:
\[
  \cstate_\tau(\eauth_\hn^j) = \ID \;\ra\;
  \bigvee_{\tau' \popreleq \tau}
  \cstate_{\tau'}(\bauth_\ue^\ID) = \nonce^j
\]
Let $\tau' \popreleq \tau$. If $\ID$ is not a base identity we know that $\cstate_{\tau'}(\bauth_\ue^\ID) \equiv \bot$, and therefore:
\[
  \neg\left(
    \cstate_{\tau'}(\bauth_\ue^\ID) = \nonce^j
  \right)
\]
It follows that $\eq{\cstate_\tau(\eauth^j_\hn)}{\ID} = \false$. We can then check that:
\begin{equation*}
  t_\tau \;=\;
  \begin{alignedat}[c]{2}
    &\ite{\neauth_\tau(\agent{A}_{1},j)}
    {\\ & \;\;\pair
      {\tsuci_\tau(\agent{A}_{1},j)}
      {\tmac_\tau(\agent{A}_{1},j)}\\
      & \!\!}
    {\ite{\neauth_\tau(\agent{A}_{2},j)}
      {\\ & \;\;\pair
        {\tsuci_\tau(\agent{A}_{2},j)}
        {\tmac_\tau(\agent{A}_{2},j)}\\
        & \qquad \cdots \\ &\!\! }
      {\unknownid}}
  \end{alignedat}
  \qquad\qquad
  t_\utau \;=\;
  \begin{alignedat}[c]{2}
    &\ite{\uneauth_\utau(\agent{A}_{1},j)}
    {\\ & \;\;\spair
      {\utsuci_\utau(\agent{A}_{1},j)}
      {\utmac_\utau(\agent{A}_{1},j)}\\
      & \!\!}
    {\ite{\uneauth_\utau(\agent{A}_{2},j)}
      {\\ & \;\;\spair
        {\utsuci_\utau(\agent{A}_{2},j)}
        {\utmac_\utau(\agent{A}_{2},j)}\\
        & \qquad \cdots \\ &\!\! }
      {\unknownid}}
  \end{alignedat}
\end{equation*}
Using the $\fa$ axiom, we can split $t_\tau$ and $t_\utau$ as follows:
\[
  \infer[\fa^*]{t_\tau \sim t_\utau}{
    \begin{alignedat}{4}
      &\big(
      &\neauth_\tau(\agent{A}_{i},j)&,\;&
      \cond{\neauth_\tau(\agent{A}_{i},j)}
      {\tsuci_\tau(\agent{A}_{i},j)}&,\;&
      \cond{\neauth_\tau(\agent{A}_{i},j)}
      {\tmac_\tau(\agent{A}_{i},j)}
      \big)_{i \le B}\\
      \;\sim\;&
      \big(
      &\uneauth_\utau(\agent{A}_{i},j)&,\;&
      \cond{\uneauth_\utau(\agent{A}_{i},j)}
      {\utsuci_\utau(\agent{A}_{i},j)}&,\;&
      \cond{\uneauth_\utau(\agent{A}_{i},j)}
      {\utmac_\utau(\agent{A}_{i},j)}
      \big)_{i \le B}
    \end{alignedat}
  }
\]
Since:
\[
  \left(\neauth_\tau(\agent{A}_{i},j),\uneauth_\utau(\agent{A}_{i},j)\right)
  \in
  \reveal_{\tauo}
\]
We just need to prove that there is a derivation of:
\[
  \begin{alignedat}{4}
    &&&\inframe_\tau,\lreveal_{\tauo}&,\;&
    \big(\cond{\neauth_\tau(\agent{A}_{i},j)}
    {\tsuci_\tau(\agent{A}_{i},j)}&,\;&
    \cond{\neauth_\tau(\agent{A}_{i},j)}
    {\tmac_\tau(\agent{A}_{i},j)}
    \big)_{i \le B}\\
    &\;\sim\;&&\inframe_\utau,\rreveal_{\tauo}&,\;&
    \big(\cond{\uneauth_\utau(\agent{A}_{i},j)}
    {\utsuci_\utau(\agent{A}_{i},j)}&,\;&
    \cond{\uneauth_\utau(\agent{A}_{i},j)}
    {\utmac_\utau(\agent{A}_{i},j)}
    \big)_{i \le B}
  \end{alignedat}
\]
Assume that we have a proof of
\[
  \begin{alignedat}[c]{2}
    &&&\inframe_\tau,\lreveal_{\tauo},\;
    \big(\cond{\neauth_\tau(\agent{A}_{i},j)}
    {\tsuci_\tau(\agent{A}_{i},j)},\;
    \cond{\neauth_\tau(\agent{A}_{i},j)}
    {\tmac_\tau(\agent{A}_{i},j)}
    \big)_{i \le B}\\
    &\;\sim\;&&
    \inframe_\tau,\lreveal_{\tauo},\;
    \big(\nonce_{i,j},\;
    \nonce_{i,j}'
    \big)_{i \le B}\\
  \end{alignedat}
  \numberthis\label{eq:myeqaligned}
\]
And:
\[
  \begin{alignedat}[c]{2}
    &&&
    \inframe_\utau,\rreveal_{\tauo},\;
    \big(\nonce_{i,j},\;
    \nonce_{i,j}'
    \big)_{i \le B} \\
    &\;\sim\;&&
    \inframe_\utau,\rreveal_{\tauo},\;
    \big(\cond{\uneauth_\utau(\agent{A}_{i},j)}
    {\utsuci_\utau(\agent{A}_{i},j)},\;
    \cond{\uneauth_\utau(\agent{A}_{i},j)}
    {\utmac_\utau(\agent{A}_{i},j)}
    \big)_{i \le B}
  \end{alignedat}
  \numberthis\label{eq:myeqaligned2}
\]
Where for all $\{\nonce_{i,j}, \nonce_{i,j}' \mid 1 \le i \le B\}$ are fresh distinct nonces. Since:
\[
  \infer[\ax{Fresh}]{
    \inframe_\tau,\lreveal_{\tauo},\;
    \big(\nonce_{i,j},\;
    \nonce_{i,j}'
    \big)_{i \le B}
    \;\sim\;
    \inframe_\utau,\rreveal_{\tauo},\;
    \big(\nonce_{i,j},\;
    \nonce_{i,j}'
    \big)_{i \le B}
  }{
    \inframe_\tau,\lreveal_{\tauo}
    \;\sim\;
    \inframe_\utau,\rreveal_{\tauo}
  }
\]
We can conclude by induction.

\paragraph{Part 2}
It only remains to give derivations of the formulas in Eq.~\eqref{eq:myeqaligned} and Eq.~\eqref{eq:myeqaligned2}. We only give the proof for Eq.~\eqref{eq:myeqaligned2}, and we omit the derivation of Eq.~\eqref{eq:myeqaligned} (as it is similar, and simpler).

Instead of doing the proof simultaneously for all $i$ in $\{1,\dots,B\}$, we give the proof for a single $i$. We let the reader check that the syntactic side-conditions necessary for the derivations for $i$ and $i'$, with $i \ne i'$, are compatible. Therefore the derivations can be sequentially composed, which yield the full proof.

Let $1 \le i \le B$. By transitivity, we only have to show that:
\[
  \begin{alignedat}[c]{2}
    &&&
    \inframe_\utau,\rreveal_{\tauo},\;
    \nonce_{i,j},\;
    \nonce_{i,j}'\\
    &\sim\;\;&&
    \inframe_\utau,\rreveal_{\tauo},\;
    \nonce_{i,j},\;
    \cond{\uneauth_\utau(\agent{A}_{i},j)}
    {\utmac_\utau(\agent{A}_{i},j)}
  \end{alignedat}
  \numberthis\label{eq:myeqaligned4b}
\]
And:
\[
  \begin{alignedat}[c]{2}
    &&&\inframe_\utau,\rreveal_{\tauo},\;
    \nonce_{i,j},\;
    \cond{\uneauth_\utau(\agent{A}_{i},j)}
    {\utmac_\utau(\agent{A}_{i},j)}\\
    &\sim\;\;&&
    \inframe_\utau,\rreveal_{\tauo},\;
    \cond{\uneauth_\utau(\agent{A}_{i},j)}
    {\utsuci_\utau(\agent{A}_{i},j)},\;
    \cond{\uneauth_\utau(\agent{A}_{i},j)}
    {\utmac_\utau(\agent{A}_{i},j)}
  \end{alignedat}
  \numberthis\label{eq:myeqaligned4}
\]

\paragraph{Derivation of Formula~\eqref{eq:myeqaligned4}}
Let $\{\uID_1,\dots,\uID_l\} = \copyid(\ID_i)$. We define, for every $0 \le y \le l$, the partially randomized terms $\utsuci_\tau^y(\ID_i,j)$:
\begin{alignat*}{2}
  \utsuci_\utau^y(\ID_i,j)
  &\;\;\equiv\;\;&&
  \begin{alignedat}[t]{2}
    &\ite{
      \eq{\cstate_\utau(\eauth_\hn^{j})}{\uID_1}
    }{\nonce_{i,j}^1
      \\ & \qquad \cdots \\ &\!\! }
    {\ite{
        \eq{\cstate_\utau(\eauth_\hn^{j})}{\uID_{y-1}}
      }{\nonce_{i,j}^{y-1}\\ & \!\!}
      {\ite{
          \eq{\cstate_\utau(\eauth_\hn^{j})}{\uID_y}
        }{
          \suci^j
          \oplus \row{\nonce^{j}}{\key^{\uID_y}}
          \\ & \qquad \cdots \\ &\!\! }
        {
          \suci^j
          \oplus \row{\nonce^{j}}{\key^{\uID_l}}
        }
      }
    }
  \end{alignedat}
\end{alignat*}
Remark that:
\begin{alignat*}{2}
  \cond{\uneauth_\utau(\agent{A}_{i},j)}{\utsuci_\utau^0(\ID_i,j)}
  &\;\;=\;\;&& \cond{\uneauth_\utau(\agent{A}_{i},j)}
  {\utsuci_\utau(\agent{A}_{i},j)}
\end{alignat*}
And that:
\[
  \infer[\textsf{indep-branch}]
  {\begin{alignedat}{3}
      &&&\inframe_\utau,\rreveal_{\tauo},\;
      \nonce_{i,j},&&
      \cond{\uneauth_\utau(\agent{A}_{i},j)}
      {\utmac_\utau(\agent{A}_{i},j)}\\
      &\;\sim\;&&
      \inframe_\utau,\rreveal_{\tauo},\;
      \cond{\uneauth_\utau(\agent{A}_{i},j)}{\utsuci_\utau^l(\ID_i,j)},\;&&
      \cond{\uneauth_\utau(\agent{A}_{i},j)}
      {\utmac_\utau(\agent{A}_{i},j)}
    \end{alignedat}}{}
\]

Hence by transitivity, to prove that there exists a derivation of Formula~\eqref{eq:myeqaligned4} it is sufficient to prove that, for every $0 < y \le l$, that we have a derivation of $\cframe_{y-1} \sim \cframe_y$, where:
\begin{alignat*}{2}
  \cframe_{y-1}&\;\;\equiv\;\;&&
  \inframe_\utau,\rreveal_{\tauo},\;
  \cond{\uneauth_\utau(\agent{A}_{i},j)}
  {\utsuci_\utau^{y-1}(\ID_i,j)},\;
  \cond{\uneauth_\utau(\agent{A}_{i},j)}
  {\utmac_\utau(\agent{A}_{i},j)}\\
  \cframe_{y}&\;\;\equiv\;\;&&
  \inframe_\utau,\rreveal_{\tauo},\;
  \cond{\uneauth_\utau(\agent{A}_{i},j)}
  {\utsuci_\utau^{y}(\ID_i,j)},\;
  \cond{\uneauth_\utau(\agent{A}_{i},j)}
  {\utmac_\utau(\agent{A}_{i},j)}
\end{alignat*}
Let $1 \le y \le B$, we are going to give a derivation of $\cframe_{y-1} \sim \cframe_y$. This is done in two times:
\begin{itemize}
\item First, we are going to use the $\prffr$ axiom applied to $\rowsym$, with key $\key^{\uID_{y}}$, to replace $\suci^j\oplus \row{\nonce^{j}}{\key^{\uID_y}}$ with $\suci^j\oplus \nonce''^y_{i,j}$ (where $\nonce''^y_{i,j}$ is a fresh nonce).

  First, observe that there is only one occurrence of $\row{\nonce^{j}}{\key^{\uID_y}}$ in $\cframe_{y-1}$ (and none in $\cframe_y$). Moreover:
  \begin{alignat*}{3}
    \setprf_{\key^{\uID_{y}}}^{\rowsym}\left(\cframe_{y-1},\cframe_{y}\right)
    \backslash\{
    \nonce^{j}\}
    &\;\;=\;\;&&&&
    \left\{
      \instate_{\taut}(\eauth_\ue^{\ID})
      \mid \taut = \_,\fuai_{\uID_{y}}(p) \popre \tau
    \right\}\\
    &&&\cup\;&&
    \left\{
      \nonce^p
      \mid \taut = \_,\fnai(p) \popre \tau
    \right\}
  \end{alignat*}
  Let $\taut = \_, \fnai(p) \popre \tau$. We know that $p \ne j$, and therefore that $(\nonce^p = \nonce^j) = \false$. We deduce that:
  \[
    \row{\nonce^{j}}{\key^{\uID_y}}
    \;=\;
    \lrcond{\bigwedge
      _{\taut = \_, \fnai(p) \popre \tau}
      \nonce^p \ne \nonce^j}
    {\row{\nonce^{j}}{\key^{\uID_y}}}
  \]
  But we still need guards for $\instate_{\taut}(\eauth_\ue^{\ID}) = \nonce^j$, for every $\taut = \_, \fuai_{\key^{\uID_{y}}}(p) \popre \tau$. The problem is that it is not true that $(\instate_{\taut}(\eauth_\ue^{\ID}) = \nonce^j) = \false$. We solve this problem by rewriting $\cframe_{y-1}$ (resp. $\cframe_{y}$) into the vector of terms $\cframe_{y-1}'$ (resp. $\cframe_{y}'$) obtained by replacing (recursively) any occurrence of $\accept_{\taut}^{\key^{\uID_y}}$ with:
  \[
    \bigvee_{\tauo = \_\fnai(j_0) \popre \taut
      \atop{\tauo \not \popre_{\taut} \newsession_{\uID_y}(\_)}}
    \left(\begin{alignedat}{2}
        &\injauth_{\taut}({\uID_y},j_0)
        \wedge\instate_{\taut}(\eauth_\hn^{j_0}) \ne \unknownid\\
        \wedge\;& \pi_1(g(\inframe_{\taut})) =
        \suci^{j_0} \xor \row{\nonce^{j_0}}{\key^{\uID_y}}
        \wedge\; \pi_2(g(\inframe_{\taut})) =
        \mac
        {\spair
          {\suci^{j_0}}
          {\nonce^{j_0}}}
        {\mkey^{\uID_y}}{5}
      \end{alignedat}\right)
    \numberthis\label{eq:mylabel}
  \]
  Which is sound using \ref{equ1}. We then have:
  \begin{alignat*}{3}
    \setprf_{\key^{\uID_{y}}}^{\rowsym}\left(\cframe'\right) &\;\;=\;\;&&&&
    \left\{
      \nonce^p
      \mid \taut = \_,\fnai(p) \popre \tau
    \right\}
  \end{alignat*}
  Therefore we can apply the $\prffr$ axioms as wanted: first we replace $\cframe_{y-1}$ and $\cframe_{y}$ by $\cframe'_{y-1}$ and $\cframe'_{y}$ using rule $R$; then we apply the $\prffr$ axiom; and finally we rewrite any term of the form \eqref{eq:mylabel} back into $\accept_{\taut}^{\key^{\uID_y}}$.

\item Then, we use the $\oplus\textsf{-indep}$ axiom to replace $\suci^j\oplus \nonce''^y_{i,j}$ with $\nonce^y_{i,j}$.
\end{itemize}

\paragraph{Derivation of Formula~\eqref{eq:myeqaligned4b}}
We use the same proof technique. We define, for every $0 \le y \le l$, the partially randomized terms $\utmac_\tau^y(\ID_i,j)$:
\begin{alignat*}{2}
  \utmac_\utau^y(\ID_i,j)
  &\;\;\equiv\;\;&&
  \begin{alignedat}[t]{2}
    &\ite{
      \eq{\cstate_\utau(\eauth_\hn^{j})}{\uID_1}
    }{\nonce_{i,j}'^1
      \\ & \qquad \cdots \\ &\!\! }
    {\ite{
        \eq{\cstate_\utau(\eauth_\hn^{j})}{\uID_{y-1}}
      }{\nonce_{i,j}'^{y-1}\\ & \!\!}
      {\ite{
          \eq{\cstate_\utau(\eauth_\hn^{j})}{\uID_y}
        }{
          \mac{\spair{\suci^j}{\nonce^j}}{\mkey^{\uID_y}}{5}
          \\ & \qquad \cdots \\ &\!\! }
        {
          \mac{\spair{\suci^j}{\nonce^j}}{\mkey^{\uID_l}}{5}
        }
      }
    }
  \end{alignedat}
\end{alignat*}
Remark that:
\begin{alignat*}{2}
  \cond{\uneauth_\utau(\agent{A}_{i},j)}{\utmac_\utau^0(\ID_i,j)}
  &\;\;=\;\;&& \cond{\uneauth_\utau(\agent{A}_{i},j)}
  {\utmac_\utau(\agent{A}_{i},j)}
\end{alignat*}
And that:
\[
  \infer[\textsf{indep-branch}]
  {\begin{alignedat}{3}
      &&&\inframe_\utau,\rreveal_{\tauo},\;
      \nonce_{i,j},&&
      \nonce'_{i,j}\\
      &\;\sim\;&&
      \inframe_\utau,\rreveal_{\tauo},\;
      \nonce_{i,j},\;&&
      \cond{\uneauth_\utau(\agent{A}_{i},j)}
      {\utmac^l_\utau(\agent{A}_{i},j)}
    \end{alignedat}}{}
\]

Hence by transitivity, to prove that there exists a derivation of Formula~\eqref{eq:myeqaligned4b} it is sufficient to prove that, for every $0 < y \le l$, that we have a derivation of $\psi_{y-1} \sim \psi_y$, where:
\begin{alignat*}{2}
  \psi_{y-1}&\;\;\equiv\;\;&&
  \psi_\utauo,\rreveal_{\tauo},\;
  \nonce_{i,j},\;
  \cond{\uneauth_\utau(\agent{A}_{i},j)}
  {\utmac_\utau^{y-1}(\ID_i,j)}\\
  \psi_{y}&\;\;\equiv\;\;&&
  \psi_\utauo,\rreveal_{\tauo},\;
  \nonce_{i,j},\;
  \cond{\uneauth_\utau(\agent{A}_{i},j)}
  {\utmac_\utau^{y}(\ID_i,j)}
\end{alignat*}
Let $1 \le y \le B$, we are going to give a derivation of $\psi_{y-1} \sim \psi_y$. For this, we are going to use the $\prfmac^5$ axiom with key $\mkey^{\uID_{y}}$, to replace $\mac{\spair{\suci^j}{\nonce^j}}{\mkey^{\uID_y}}{5}$ with a fresh nonce $\tilde{\nonce}^y_{i,j}$.

First, observe that there is only one occurrence of $\mac{\spair{\suci^j}{\nonce^j}}{\mkey^{\uID_y}}{5}$ in $\psi_{y-1}$ (and none in $\psi_y$). Moreover:
\begin{multline*}
  \setmac_{\key^{\uID_{y}}}^{5}\left(\psi_{y-1},\psi_{y}\right)
  \backslash\left\{
    \spair{\suci^j}{\nonce^j}
  \right\}
  \;\;=\;\;\\
  \begin{alignedat}{2}
    &&&\left\{
      \spair
      {\suci^{p}}
      {\nonce^p}
      \mid \taut = \_, \fnai(p) \popre \tau
    \right\}\\
    &\cup\;&&
    \left\{
      \spair
      {\pi_1(g(\inframe_{\taut})) \xor
        \row{\instate_{\taut}(\eauth_\ue^{{\key^{\uID_{y}}}})}{\key}}
      {\instate_{\taut}(\eauth_\ue^{{\key^{\uID_{y}}}})}
      \mid \taut = \_, \fnai(p) \popre \tau
    \right\}
  \end{alignedat}
\end{multline*}
Let $\taut = \_, \fnai(p) \popre \tau$. Since $\suci^j$ is a fresh nonce, we know using $\ax{EQIndep}$ and the injectivity of the pair function that:
\begin{gather*}
  \left(
    \spair{\suci^j}{\nonce^j} =
    \spair{\suci^{p}}{\nonce^p}
  \right) = \false\\
  \left(
    \spair{\suci^j}{\nonce^j} =
    \spair
    {\pi_1(g(\inframe_{\taut})) \xor
      \row{\instate_{\taut}(\eauth_\ue^{{\key^{\uID_{y}}}})}{\key}}
    {\instate_{\taut}(\eauth_\ue^{{\key^{\uID_{y}}}})}
  \right) = \false
\end{gather*}
Therefore we can directly apply the $\prfmac^5$ axiom, which concludes this case.

\subsection{Case $\ai = \fuai_\ID(j)$}
We know that $\uai = \fuai_{\nu_\tau(\ID)}(j)$. Here $\lreveal_{\tau}$ and $\lreveal_{\tauo}$ coincides everywhere except on the pairs:
\begin{alignat*}{2}
  \cstate_\tau(\success_\ue^\ID)
  &\;\;\sim\;\;&&
  \cstate_\utau(\success_\ue^{\nu_\tau(\ID)})\\
  \underbrace{\begin{alignedat}{2}
      \ite{\cstate_\tau(\success_\ue^\ID)&}
      {\cstate_\tau(\suci_\ue^\ID)\\ &}
      {\bot}
    \end{alignedat}}_{\msuci_\tau^\ID}
  &\;\;\sim\;\;&&
  \underbrace{\begin{alignedat}{2}
      \ite{\cstate_\utau(\success_\ue^{\nu_\tau(\ID)})&}
      {\cstate_\utau(\suci_\ue^{\nu_\tau(\ID)})\\ &}
      {\bot}
    \end{alignedat}}_{\msuci_\utau^{\nu_\tau(\ID)}}
\end{alignat*}
Moreover, we also need to show that:
\[
  \accept_\tau^\ID \sim \accept_\utau^{\nu_\tau(\ID)}
\]
Recall that $\tau = \tauo,\ai$ and $\utau = \utauo,\uai$, and that:
\begin{alignat*}{2}
  \cstate_\tau(\success_\ue^{\ID})
  &\;\equiv\;&& \accept_\tau^\ID\\
  \cstate_\tau(\suci_\ue^{\ID}) &\;\equiv\;&&
  \begin{alignedat}[t]{2}
    &\ite
    {\accept_\tau^\ID}
    {\pi_1(g(\inframe_\tau))
      \xor
      \row{\cstate_\tauo(\eauth_\ue^{\ID})}{\key}\\ & \!\!}
    {\unset}
  \end{alignedat}\\\displaybreak[0]
  \cstate_\utau(\success_\ue^{{\nu_\tau(\ID)}})
  &\;\equiv\;&&
  \accept_\utau^{\nu_\tau(\ID)}\\
  \cstate_\utau(\suci_\ue^{{\nu_\tau(\ID)}})
  &\;\equiv\;&&
  \begin{alignedat}[t]{2}
    &\ite
    {\accept_\utau^{\nu_\tau(\ID)}}
    {\pi_1(g(\inframe_\utau))
      \xor
      \row{\cstate_\utauo(\eauth_\ue^{{\nu_\tau(\ID)}})}{\key}\\ & \!\!}
    {\unset}
  \end{alignedat}
\end{alignat*}
Therefore we want a proof of:
\[
  \inframe_\tau,\lreveal_{\tauo},
  \accept_\tau^\ID,\msuci_\tau^\ID
  \sim
  \inframe_\utau,\rreveal_{\tauo},
  \accept_\utau^{\nu_\tau(\ID)},\msuci_\utau^{\nu_\tau(\ID)}
  \numberthis\label{eq:fuai1}
\]
Using \ref{equ1}, we know that:
\begin{alignat*}{2}
  \accept_\tau^\ID
  &\;\;\lra\;\;&&
  \bigvee_{\taut = \_,\fnai(j_0) \popre \tau
    \atop{\taut \not \potau \newsession_\ID(\_)}}
  \futr{\tau}{\taut}
  \numberthis\label{eq:i2-1}
\end{alignat*}
We know that $\utau = \utauo,\fuai_{\nu_\tau(\ID)}(j)$ is a valid symbolic trace. Using \ref{equ1} again, we know that:
\begin{alignat*}{2}
  \accept_\utau^{\nu_\tau(\ID)}
  &\;\;\lra\;\;&&
  \bigvee_{\taut = \_,\fnai(j_0) \popre \tau
    \atop{\taut \not \potau \newsession_\ID(\_)}}
  \futr{\utau}{\utaut}
  \numberthis\label{eq:i2-2}
\end{alignat*}
Let:
\begin{gather*}
  \{j_0,\dots,j_l\} =
  \{i \mid \tau' = \_,\fnai(i) \popre \tau \wedge
  \tau' \not \potau \newsession_\ID(\_)\}
\end{gather*}
One can check that:
\begin{gather*}
  \{j_0,\dots,j_l\} =
  \{i \mid \tau' = \_,\fnai(i) \popre \utau
  \wedge \tau' \not \poutau \newsession_{\nu_\tau(\ID)}(\_)\}
\end{gather*}
For all $0 \le i \le l$, let $\tau_{j_i}$ be such that $\tau_{j_i} = \_,\fnai(j_i) \popre \tau$. One can check that:
\begin{equation*}
  \msuci_\tau^\ID =
  \begin{alignedat}[t]{2}
    &\ite{
      \futr{\tau}{\tau_{j_0}}
    }
    {\suci^{j_0}\\ & \!\!}
    {
      \ite{
        \futr{\tau}{\tau_{j_1}}
      }
      { \suci^{j_1} \\ & \qquad \cdots \\ &\!\! }
      { \suci^{j_l}}
    }
  \end{alignedat}
  \qquad\qquad
  \msuci_\utau^{\nu_\tau(\ID)} =
  \begin{alignedat}[t]{2}
    &\ite{
      \futr{\utau}{\ufresh{\tau_{j_0}}}
    }
    {\suci^{j_0}\\ & \!\!}
    {
      \ite{
        \futr{\utau}{\ufresh{\tau_{j_1}}}
      }
      { \suci^{j_1} \\ & \qquad \cdots \\ &\!\! }
      { \suci^{j_l}}
    }
  \end{alignedat}
\end{equation*}
We can now start giving a derivation of \eqref{eq:fuai1}:
\[
  \infer[\fa^*]{
    \inframe_\tau,\lreveal_{\tauo},
    \accept_\tau^\ID,\msuci_\tau^\ID
    \sim
    \inframe_\utau,\rreveal_{\tauo},
    \accept_\utau^{\nu_\tau(\ID)},\msuci_\utau^{\nu_\tau(\ID)}
  }{
    \infer[\dup^*]{
      \inframe_\tau,\lreveal_{\tauo},
      \left(\futr{\tau}{\tau_{j_i}}\right)_{i \le l},
      \left(\suci^{j_i}\right)_{i \le l}
      \sim
      \inframe_\utau,\rreveal_{\tauo},
      \left(\futr{\utau}{\ufresh{\tau_{j_i}}}\right)_{i \le l},
      \left(\suci^{j_i}\right)_{i \le l}
    }{
      \inframe_\tau,\lreveal_{\tauo},
      \left(\futr{\tau}{\tau_{j_i}}\right)_{i \le l}
      \sim
      \inframe_\utau,\rreveal_{\tauo},
      \left(\futr{\utau}{\ufresh{\tau_{j_i}}}\right)_{i \le l}
    }}
\]
Since for all $1 \le i \le l$, $(\suci^{j_i} \sim \suci^{j_i}) \in \reveal_{\tauo}$. Finally, we conclude using \ref{der2} for every $0 \le i \le l$:
\[
  \infer[\fa^*]{
    \inframe_\tau,\lreveal_{\tauo},
    \left(\futr{\tau}{\tau_{j_i}}\right)_{i \le l}
    \sim
    \inframe_\utau,\rreveal_{\tauo},
    \left(\futr{\utau}{\ufresh{\tau_{j_i}}}\right)_{i \le l}
  }{
    \inframe_\tau,\lreveal_{\tauo}
    \sim
    \inframe_\utau,\rreveal_{\tauo}
  }
\]

\subsection{Case $\ai = \cuai_\ID(j,0)$}
We know that $\uai = \cuai_{\nu_\tau(\ID)}(j,0)$. Let $\uID = \nu_\tau(\ID)$. Here $\lreveal_{\tau}$ and $\lreveal_{\tauo}$ coincides everywhere except on:
\begin{mathpar}
  \cstate_\tau(\success_\ue^\ID)
  \;\sim\;
  \cstate_\utau(\success_\ue^{\uID})

  \cstate_\tau(\uetsuccess^\ID)
  \;\sim\;
  \cstate_\utau(\uetsuccess^{\uID})

  \msuci^\ID_\tau
  \;\sim\;
  \msuci^{\uID}_\utau
\end{mathpar}
Handling these is completely trivial because:
\begin{mathpar}
  \cstate_\tau(\success_\ue^\ID) \equiv \false

  \cstate_\utau(\success_\ue^{\uID}) \equiv \false

  \cstate_\tau(\uetsuccess^\ID) \equiv
  \instate_\tau(\success_\ue^\ID)

  \cstate_\utau(\uetsuccess^{\uID}) \equiv
  \instate_\utau(\success_\ue^{\uID})

  \msuci^\ID_\tau \equiv \bot

  \msuci^{\uID}_\utau \equiv \bot
\end{mathpar}
And $(\instate_\tau(\success_\ue^\ID),\instate_\utau(\success_\ue^{\uID})) \in \reveal_{\tauo}$. Finally, we conclude by observing that:
\begin{mathpar}
  t_\tau = \ite{\instate_\tau(\success_\ue^\ID)}
  {\msuci^\ID_\tau}{\nosuci}

  t_\utau = \ite{\instate_\utau(\success_\ue^\uID)}
  {\msuci^\uID_\utau}{\nosuci}
\end{mathpar}
Hence:
\[
  \begin{gathered}[c]
    \infer[\simp]{
      \inframe_\tau,
      \lreveal_{\tauo},
      t_\tau
      \;\sim\;\
      \inframe_\utau,
      \rreveal_{\tauo},
      t_\utau
    }{
      \infer[\dup^*]{
        \inframe_\tau, \lreveal_{\tauo},
        \instate_\tau(\success_\ue^\ID),
        \msuci^\ID_\tau,
        \nosuci
        \;\sim\;
        \inframe_\utau,\rreveal_{\tauo},
        \instate_\utau(\success_\ue^\uID),
        \msuci^\uID_\utau,
        \nosuci
      }{
        \inframe_\tau, \lreveal_{\tauo}
        \;\sim\;
        \inframe_\utau,\rreveal_{\tauo},
        \instate_\utau(\success_\ue^\uID)
      }
    }
  \end{gathered}
\]

\subsection{Case $\ai = \cnai(j,0)$}
We know that $\uai = \cnai(j,0)$. Using \ref{a6}, we know that for every $\ID \ne \ID'$, $\neg\accept_\tau^\ID \lra \neg \accept_\tau^{\ID'}$. Therefore the answer from the network does not depend on the order in which we make the $\accept_\tau^\ID$ tests. Formally, the following list of conditionals is a $\cs$ partition:
\[
  \left(
    \left(\accept_\tau^\ID\right)_{\ID \in \iddom},
    \bigwedge_{\ID \in \iddom}\neg \accept_\tau^\ID
  \right)
\]
To get a uniform notation, we let $\accept_\tau^{\IDdum} \equiv \bigwedge_{\ID \in \iddom}\neg \accept_\tau^\ID$, and $\extiddom = \iddom \cup \{\IDdum\}$. Hence using Proposition~\ref{prop:case} we get that:
\begin{gather*}
  t_\tau \;=\;
  \switch{\ID \in \extiddom}{\accept_\tau^\ID}{\msg_\tau^\ID}
\end{gather*}
We are now going to show that for every $\ID \in \extiddom$, the term $\msg_\tau^\ID$ can be replaced by $\triplet{\nonce^j}{\nonce_\ID^\oplus}{\nonce_\ID^\macsym}$ (where $(\nonce_\ID^\oplus)_{\ID \in \extiddom}$ and $(\nonce_\ID^\macsym)_{\ID \in \extiddom}$ are fresh distinct nonces). We will then conclude easily using the $\ax{fresh}$ axiom.

Let $\ID_1,\dots,\ID_l$ be an arbitrary enumeration of $\extiddom$. For every $1 \le n \le l$, and for every $\ID_i \in \{ \ID_1,\dots,\ID_l\}$, we let:
\[
  \rmsg_n^{\ID_i} \equiv
  \begin{dcases*}
    \triplet{\nonce^j}{\nonce_{\ID_i}^\oplus}{\nonce_{\ID_i}^\macsym} & if $i \le n$\\
    \rmsg_\tau^{\ID_i} & if $i > n$
  \end{dcases*}
\]
And we let $t_n$ be the term $t_\tau$ where the subterms $\msg_\tau^\ID$ have been replaced by $\triplet{\nonce^j}{\nonce_\ID^\oplus}{\nonce_\ID^\macsym}$ for the first $n$ identities:
\[
  t_n \equiv
  \switch{\ID \in \extiddom}{\accept_\tau^\ID}{\rmsg_n^\ID}
\]
We can check that $t_0 \equiv t_\tau$.

\paragraph{Part 1} We now show that for every $1 \le n \le l$, we have a derivation of:
\[
  \inframe_\tau,\lreveal_{\tauo},t_{n-1}
  \;\sim\;
  \inframe_\tau,\lreveal_{\tauo},t_{n}
  \numberthis\label{eq:fgsghafiadsssdfg}
\]
Let $n$ be in $\{1,\dots,l\}$. Let $\ID = \ID_n$, $\key = \key^\ID$ and $\mkey = \mkey^\ID$. We are going to apply $\prff$ axiom with key $\key$ to replace $\ow{\key}{\nonce^j}$ by $\nonce_{\ID}$, where $\nonce_{\ID}$ is a fresh nonce. Recall that:
\[
  \msg_\tau^{\ID} \;\equiv\;
  \striplet{\nonce^j}
  {\underbrace{\instate_\tau(\sqn_\hn^{\ID})}_{u_\sqn} \oplus \ow{\nonce^j}{\key^{\ID}}}
  {\underbrace{\mac{\striplet
        {\nonce^j}
        {\instate_\tau(\sqn_\hn^{\ID})}
        {\instate_\tau(\suci_\hn^{\ID})}}
      {\mkey^{\ID}}{3}}_{u_\macsym}}
\]
We let $\psi$ be the context with one hole (which has only one occurrence) such that:
\[
  \psi[\triplet
  {\nonce^j}
  {u_\sqn \oplus \ow{\nonce^j}{\key^{\ID}}}
  {u_\macsym}]
  \equiv
  \inframe_\tau,\lreveal_{\tauo},t_{n - 1}
  \qquad\qquad
  \psi[\triplet{\nonce^j}{\nonce_\ID^\oplus}{\nonce_\ID^\macsym}]
  \equiv
  \inframe_\tau,\lreveal_{\tauo},t_{n}
\]
Let $\psi_0[] \equiv \psi[\triplet{\nonce^j}{u_\sqn \oplus []}{u_\macsym}]$. Notice that:
\begin{alignat*}{3}
  \setprf_{\key}^{\owsym}
  \left(
    \psi_0[]
  \right)
  &\;\;=\;\;&&&&
  \left\{
    \pi_1(\inframe_\taut)
    \mid \taut = \_, \cuai_{\ID}(p,1) \popre \tau
  \right\}\\
  &&&\cup\;&&
  \left\{
    \nonce^p
    \mid \taut = \_, \cnai(p) \popre \tau
  \right\}
\end{alignat*}
We want to get rid of the sub-terms of the form $\ow{\pi_1(\inframe_\taut)}{\key}$, for any $\taut$ such that $\taut = \_, \cuai_{\ID}(p,1) \popre \tau$. To do this, for every $\taut = \_, \cuai_{\ID}(p,1) \popre \tau$, we let $\tauttt = \_,\cuai_\ID(j_p,0) \popre \tau$, and we apply \ref{sequ2} to rewrite all occurrence of $\accept_\taut^\ID$ in $\psi_0$ using:
\[
  \accept_\taut^\ID
  \leftrightarrow
  \bigvee
  _{\tautt = \_,\cnai(j_1,0)
    \atop{\tauttt <_\taut \tautt <_\taut \taut}}
  \left(\ptr{\tauttt,\taut}{\tautt}\right)
  \numberthis\label{eq:ddsuovweiegrjqw9}
\]
This yields a vector of terms $\psi_0'[]$ with one hole. It is easy to check that:
\begin{alignat*}{3}
  \setprf_{\key}^{\owsym}
  \left(
    \psi_0'[]
  \right)
  &\;\;=\;\;&&&&
  \left\{
    \nonce^p
    \mid \taut = \_, \cnai(p) \popre \tau
  \right\}
\end{alignat*}
By validity of $\tau$, we know that for every $\taut = \_, \cnai(p) \popre \tau$, we have $p \ne j$. Therefore using \ax{fresh} we have $(\nonce^j = \nonce^P)\lra \false$. It follows that we can apply the $\prff$ axiom in $\psi'_0[\ow{\key}{\nonce^j}]$, replacing $\ow{\key}{\nonce^j}$ by $\nonce_{\ID}$, which yields $\psi'_0[\nonce_{\ID}]$. We then rewrite any term of the form in \eqref{eq:ddsuovweiegrjqw9} back into $\accept_\taut^\ID$, obtaining $\psi_0[\nonce_{\ID}] \equiv \psi[\triplet{\nonce^j}{u_\sqn \oplus \nonce_\ID}{u_\macsym}]$. We then use $\oplus\textsf{-indep}$ to replace $u_\sqn \oplus \nonce_\ID$ by $\nonce_\ID^\oplus$.
\[
  \infer[R]{
    \inframe_\tau,\lreveal_{\tauo},t_{n-1}
    \;\sim\;
    \inframe_\tau,\lreveal_{\tauo},t_{n}
  }{
    \infer[R]{
      \psi[\triplet
      {\nonce^j}
      {u_\sqn \oplus \ow{\nonce^j}{\key^{\ID}}}
      {u_\macsym}]
      \sim
      \psi[\triplet{\nonce^j}{\nonce_\ID^\oplus}{\nonce_\ID^\macsym}]
    }{
      \infer[\prff]{
        \psi_0'[\ow{\key}{\nonce^j}]
        \sim
        \psi[\triplet{\nonce^j}{\nonce_\ID^\oplus}{\nonce_\ID^\macsym}]
      }{
        \infer[R]{
          \psi_0'[\nonce_\ID]
          \sim
          \psi[\triplet{\nonce^j}{\nonce_\ID^\oplus}{\nonce_\ID^\macsym}]
        }{
          \infer[\oplus\textsf{-indep}]{
            \psi[\triplet
            {\nonce^j}
            {u_\sqn \oplus \nonce_\ID}
            {u_\macsym}]
            \sim
            \psi[\triplet{\nonce^j}{\nonce_\ID^\oplus}{\nonce_\ID^\macsym}]
          }{
            \psi[\triplet
            {\nonce^j}
            {\nonce_\ID^\oplus}
            {u_\macsym}]
            \sim
            \psi[\triplet{\nonce^j}{\nonce_\ID^\oplus}{\nonce_\ID^\macsym}]
          }
        }
      }
    }
  }
\]
We now the same thing with $u_\macsym$, applying $\prfmac^3$ axiom to replace $u_\macsym$ by $\nonce_\ID^\macsym$. The proof is similar to the one we just did for $\prff$, and we omit the details. We then conclude using $\refl$. This yields:
\[
  \infer{
    \psi[\triplet
    {\nonce^j}{\nonce_\ID^\oplus}{u_\macsym}]
    \sim
    \psi[\triplet{\nonce^j}{\nonce_\ID^\oplus}{\nonce_\ID^\macsym}]
  }{
    \infer*{}{
      \infer[\refl]{}{
        \psi[\triplet
        {\nonce^j}{\nonce_\ID^\oplus}{\nonce_\ID^\macsym}]
        \sim
        \psi[\triplet{\nonce^j}{\nonce_\ID^\oplus}{\nonce_\ID^\macsym}]
      }
    }
  }
\]

\paragraph{Part 2}
Using the fact that $t_0 \equiv t_\tau$ and \eqref{eq:fgsghafiadsssdfg}, and using the transitivity axiom, we can build a derivation of:
\[
  \inframe_\tau,\lreveal_{\tauo},t_\tau
  \;\sim\;
  \inframe_\tau,\lreveal_{\tauo},t_{l}
\]
Moreover, using the \textsf{indep-branch} axiom we know that:
\[
  \infer[\textsf{indep-branch}]{
    \inframe_\tau,\lreveal_{\tauo},t_{l}
    \;\sim\;
    \inframe_\tau,\lreveal_{\tauo},\nonce}{}
\]
where $\nonce$ is a fresh nonce. Using transitivity again, we get a derivation of:
\[
  \inframe_\tau,\lreveal_{\tauo},t_\tau
  \;\sim\;
  \inframe_\tau,\lreveal_{\tauo},\nonce
  \numberthis\label{eqref:fgjdgofqerfqjporkapsd}
\]
Repeating everything we did in \textbf{Part~1}, we can show that we have a derivation of:
\[
  \inframe_\utau,\rreveal_{\tauo},\nonce'
  \;\sim\;
  \inframe_\utau,\rreveal_{\tauo},t_\utau
  \numberthis\label{eqref:gsdhifjaafdfd}
\]
where $\nonce'$ is a fresh nonce. We then conclude using the transitivity and \ax{Fresh}:
\[
  \infer[\trans]{
    \inframe_\tau,\lreveal_{\tauo},t_\tau
    \;\sim\;
    \inframe_\utau,\rreveal_{\tauo},t_\utau
  }{
    \infer{
      \begin{alignedat}{2}
        &&&\inframe_\tau,\lreveal_{\tauo},t_\tau\\
        &\sim\;\;&&
        \inframe_\tau,\lreveal_{\tauo},\nonce
      \end{alignedat}
    }{\eqref{eqref:fgjdgofqerfqjporkapsd}}
    &
    \infer[\ax{Fresh}]{
      \inframe_\tau,\lreveal_{\tauo},\nonce
      \;\sim\;
      \inframe_\utau,\rreveal_{\tauo},\nonce'
    }{
      \inframe_\tau,\lreveal_{\tauo}
      \;\sim\;
      \inframe_\utau,\rreveal_{\tauo}
    }
    &
    \infer{
      \begin{alignedat}{2}
        &&&\inframe_\utau,\rreveal_{\tauo},\nonce'\\
        &\sim\;\;&&
        \inframe_\utau,\rreveal_{\tauo},t_\utau
      \end{alignedat}
    }{\eqref{eqref:gsdhifjaafdfd}}
  }
\]

\subsection{Case $\ai = \cuai_\ID(j,1)$}
We know that $\uai = \cuai_{\nu_\tau(\ID)}(j,1)$. Let $\uID = \nu_\tau(\ID)$. By validity of $\tau$, we know that there exists $\tautt = \_,\cuai_\ID(j,0)$ such that $\tautt \popre \tau$. Here $\lreveal_{\tau}$ and $\lreveal_{\tauo}$ coincides everywhere except on:
\begin{mathpar}
  \cstate_\tau(\sqn_\ue^\ID) -
  \instate_\tau(\sqn_\ue^\ID)
  \;\sim\;
  \cstate_\utau(\sqn_\ue^{\uID}) -
  \instate_\utau(\sqn_\ue^{\uID})

  \cstate_\tau(\eauth_\ue^\ID)
  \;\sim\;
  \cstate_\utau(\eauth_\ue^{\uID})

  \left(
    \mac{\nonce^{j_0}}{\mkey^\ID}{4}
    \;\sim\;
    \mac{\nonce^{j_0}}{\mkey^{\uID}}{4}
  \right)_{\taut = \_,\cnai(j_0,0)\atop{\tautt \potau \taut}}
\end{mathpar}
First, using \ref{sequ2} twice we know that:
\begin{mathpar}
  \accept_\tau^\ID
  \leftrightarrow
  \bigvee
  _{\taut = \_,\cnai(j_1,0)
    \atop{\tautt \potau \taut}}
  \ptr{\tautt,\tau}{\taut}

  \accept_\utau^\uID
  \leftrightarrow
  \bigvee
  _{\taut = \_,\cnai(j_1,0)
    \atop{\tautt \potau \taut}}
  \ptr{\utautt,\utau}{\utaut}
\end{mathpar}
Using \ref{der3} we know that for every $\taut = \_,\cnai(j_1,0)$ such that $\tautt \potau \taut$ we have a derivation:
\[
  \begin{gathered}[c]
    \infer[\simp]{
      \inframe_\tau,
      \lreveal_{\tauo},
      \ptr{\tautt,\tau}{\taut}
      \;\sim\;\
      \inframe_\utau,
      \rreveal_{\tauo},
      \ptr{\utautt,\utau}{\utaut}
    }{
      \inframe_\tau, \lreveal_{\tauo}
      \;\sim\;
      \inframe_\utau,\rreveal_{\tauo}
    }
  \end{gathered}
  \numberthis\label{eq:sdjuofheiafjagoae}
\]
Therefore we can build the following derivation:
\[
  \begin{gathered}[c]
    \infer[\simp]{
      \inframe_\tau,
      \lreveal_{\tauo},
      \accept_\tau^\ID
      \;\sim\;\
      \inframe_\utau,
      \rreveal_{\tauo},
      \accept_\utau^\uID
    }{
      \infer[\simp]{
        \inframe_\tau, \lreveal_{\tauo},
        \left(
          \ptr{\tautt,\tau}{\taut}
        \right)
        _{\taut = \_,\cnai(j_1,0)
          \atop{\tautt \potau \taut}}
        \;\sim\;
        \inframe_\utau,\rreveal_{\tauo},
        \left(
          \ptr{\utautt,\utau}{\utaut}
        \right)
        _{\taut = \_,\cnai(j_1,0)
          \atop{\tautt \potau \taut}}
      }{
        \inframe_\tau, \lreveal_{\tauo}
        \;\sim\;
        \inframe_\utau,\rreveal_{\tauo}
      }
    }
  \end{gathered}
  \numberthis\label{eq:dfivjqirqurqwpiqwoeiqwpox}
\]

\paragraph{Part 1}
We can check that for every $\taut = \_,\cnai(j_1,0)$ such that $\tautt \potau \taut $:
\begin{gather*}
  \ptr{\tautt,\tau}{\taut} \ra
  \cstate_\tau(\eauth_\ue^\ID) = \nonce^{j_1}
  \qquad\qquad
  \ptr{\utautt,\utau}{\utaut} \ra
  \cstate_\utau(\eauth_\ue^{\uID}) = \nonce^{j_1}\\
  \neg \accept_\tau^\ID \ra
  \cstate_\tau(\eauth_\ue^\ID) = \fail
  \qquad\qquad
  \neg \accept_\utau^\uID \ra
  \cstate_\utau(\eauth_\ue^{\uID}) = \fail
\end{gather*}
And $(\nonce^{j_1},\nonce^{j_1}) \in \reveal_{\tauo}$. Therefore we can decompose $\cstate_\tau(\eauth_\ue^\ID)$ and $\cstate_\utau(\eauth_\ue^{\uID})$ using $\fa$ and get rid of the resulting terms using \eqref{eq:sdjuofheiafjagoae} and \eqref{eq:dfivjqirqurqwpiqwoeiqwpox}:
\[
  \begin{gathered}[c]
    \infer[R]{
      \inframe_\tau,
      \lreveal_{\tauo},
      \cstate_\tau(\eauth_\ue^\ID)
      \;\sim\;\
      \inframe_\utau,
      \rreveal_{\tauo},
      \cstate_\utau(\eauth_\ue^{\uID})
    }{
      \infer[\simp]{
        \begin{alignedat}{2}
          &&&\inframe_\tau, \lreveal_{\tauo},
          \ite{\accept_\tau^\ID}{
            \switch{{\taut = \_,\cnai(j_1,0)
                \atop{\tautt \potau \taut}}}
            {\ptr{\tautt,\tau}{\taut}}
            {\nonce^{j_1}}}
          {\fail}\\
          &\sim\;\;&&
          \inframe_\utau,\rreveal_{\tauo},
          \ite{\accept_\utau^\uID}{
            \switch{{\taut = \_,\cnai(j_1,0)
                \atop{\tautt \potau \taut}}}
            {\ptr{\utautt,\utau}{\utaut}}
            {\nonce^{j_1}}}
          {\fail}
        \end{alignedat}
      }{
        \infer[\simp]{
          \begin{alignedat}{2}
            &&&\inframe_\tau, \lreveal_{\tauo},
            \accept_\tau^\ID,
            \left(
              \ptr{\tautt,\tau}{\taut},\nonce^{j_1}
            \right)
            _{\taut = \_,\cnai(j_1,0)
              \atop{\tautt \potau \taut}},
            \fail\\
            &\sim\;\;&&
            \inframe_\utau,\rreveal_{\tauo},
            \accept_\utau^\uID,
            \left(
              \ptr{\utautt,\utau}{\utaut},\nonce^{j_1}
            \right)
            _{\taut = \_,\cnai(j_1,0)
              \atop{\tautt \potau \taut}},
            \fail
          \end{alignedat}
        }{
          \inframe_\tau, \lreveal_{\tauo}
          \;\sim\;
          \inframe_\utau,\rreveal_{\tauo}
        }
      }
    }
  \end{gathered}
  \numberthis\label{eq:gsduiroqruqiri8}
\]

\paragraph{Part 2}
Observe that for every $\taut = \_,\cnai(j_1,0)$ such that $\tautt \potau \taut$:
\begin{mathpar}
  \ptr{\tautt,\tau}{\taut} \ra
  \cstate_\tau(\sqn_\ue^\ID) -
  \instate_\tau(\sqn_\ue^\ID) =
  \one

  \ptr{\utautt,\utau}{\utaut} \ra
  \cstate_\utau(\sqn_\ue^{\uID}) -
  \instate_\utau(\sqn_\ue^{\uID}) =
  \one

  \neg \accept_\tau^\ID \ra
  \cstate_\tau(\sqn_\ue^\ID) -
  \instate_\tau(\sqn_\ue^\ID) =
  \zero

  \neg \accept_\utau^\uID \ra
  \cstate_\utau(\sqn_\ue^{\uID}) -
  \instate_\utau(\sqn_\ue^{\uID}) =
  \zero
\end{mathpar}
It is then easy to adapt the derivation in \eqref{eq:gsduiroqruqiri8} to get a derivation of (we omit the details):
\[
  \begin{gathered}[c]
    \infer[\simp]{
      \inframe_\tau,
      \lreveal_{\tauo},
      \cstate_\tau(\sqn_\ue^\ID) -
      \instate_\tau(\sqn_\ue^\ID)
      \;\sim\;\
      \inframe_\utau,
      \rreveal_{\tauo},
      \cstate_\utau(\sqn_\ue^{\uID}) -
      \instate_\utau(\sqn_\ue^{\uID})
    }{
      \inframe_\tau, \lreveal_{\tauo}
      \;\sim\;
      \inframe_\utau,\rreveal_{\tauo}
    }
  \end{gathered}
  \numberthis\label{eq:oipiewtjiqewfapcd}
\]

\paragraph{Part 3}
We finally take care of $t_\tau$ and the $\macsym^4$ terms. First, we check that for every $\taut = \_,\cnai(j_1,0)$ such that $\tautt \potau \taut$:
\begin{gather*}
  \ptr{\tautt,\tau}{\taut} \ra
  t_\tau =
  \mac{\nonce^{j_0}}{\mkey^\ID}{4}
  \qquad\qquad
  \ptr{\utautt,\utau}{\utaut} \ra
  t_\utau =
  \mac{\nonce^{j_0}}{\mkey^{\uID}}{4}\\
  \neg \accept_\tau^\ID \ra
  t_\tau = \textsf{error}
  \qquad\qquad
  \neg \accept_\utau^\uID \ra
  t_\utau = \textsf{error}
\end{gather*}
Similarly to what we did in \eqref{eq:gsduiroqruqiri8}, we decompose $t_\tau$ and $t_\utau$ using \eqref{eq:sdjuofheiafjagoae} and \eqref{eq:dfivjqirqurqwpiqwoeiqwpox}. Omitting the detail of the derivation, this yield:
\[
  \begin{gathered}[c]
    \infer[\simp]{
      \inframe_\tau,
      \lreveal_{\tauo},
      t_\tau
      \;\sim\;\
      \inframe_\utau,
      \rreveal_{\tauo},
      t_\utau
    }{
      \inframe_\tau, \lreveal_{\tauo},
      \left(
        \mac{\nonce^{j_0}}{\mkey^\ID}{4}
      \right)_{\taut = \_,\cnai(j_0,0)\atop{\tautt \potau \taut}}
      \;\sim\;
      \inframe_\utau,\rreveal_{\tauo},
      \left(
        \mac{\nonce^{j_0}}{\mkey^{\uID}}{4}
      \right)_{\taut = \_,\cnai(j_0,0)\atop{\tautt \potau \taut}}
    }
  \end{gathered}
\]
Observe that the $\macsym^4$ terms here are exactly the $\macsym^4$ terms in $\lreveal_{\tau} \backslash \lreveal_{\tauo}$. To conclude this proof, it only remains to give a derivation of:
\[
  \inframe_\tau, \lreveal_{\tauo},
  \left(
    \mac{\nonce^{j_0}}{\mkey^\ID}{4}
  \right)_{\taut = \_,\cnai(j_0,0)\atop{\tautt \potau \taut}}
  \;\sim\;
  \inframe_\utau,\rreveal_{\tauo},
  \left(
    \mac{\nonce^{j_0}}{\mkey^{\uID}}{4}
  \right)_{\taut = \_,\cnai(j_0,0)\atop{\tautt \potau \taut}}
\]
For every $\taut = \_,\cnai(j_1,0)$ such that $\tautt \potau \taut$, we are going to apply the $\prfmac^4$ axiom with key $\mkey^\ID$ to replace $\mac{\nonce^{j_0}}{\mkey^\ID}{4}$ by a fresh nonce $\nonce_{\taut}$. Let $\psi \equiv \inframe_\tau, \lreveal_{\tauo}$, observe that:
\begin{alignat*}{3}
  \setmac_{\ID}^{4}\left(\psi\right)\;\;=
  &&&&&
  \left\{
    \pi_1(g(\inframe_\taua))
    \mid
    \taua = \_,\cuai_\ID(j_a,1) \popre \tau
  \right\}\\
  &&&\cup\;&&
  \left\{
    \nonce^{j_n}
    \mid
    \taun = \_,\cnai(j_n,1) \popre \tau
  \right\}
\end{alignat*}
Let:
\[
  \mathcal{N} =
  \left\{
    \nonce^{j_0}
    \mid
    \taut = \_,\cnai(j_0,0) \wedge \tautt \potau \taut
  \right\}
\]
Our goal is to rewrite $\psi$ into a vector of terms $\psi_1$ such that $\setmac_{\ID}^{4}\left(\psi_1\right) \cap \mathcal{N} = \emptyset$. This will allow us to apply the $\prfmac^4$ axiom. We are going to rewrite $\psi$ as follows:
\begin{itemize}
\item Let $\taua = \_,\cuai_\ID(j_a,1) \popre \tau$. By validity of $\tau$, we know that $\taua \potau \tautt$, and that there exists $\taub = \_,\cuai_\ID(j_a,0) \potau \taua$. Using \ref{sequ2}, we know that:
  \[
    \accept_\taua^\ID
    \leftrightarrow
    \bigvee
    _{\taux = \_,\cnai(j_x,0)
      \atop{\taub \potau \taux \potau \taua}}
    \ptr{\taub,\taua}{\taux}
  \]
  We let $\alpha_\taua^\ID$ be the right-hand side of the equation above. Using this, we can check that:
  \[
    t_\taua \; = \;
    \ite{\alpha_\taua^\ID}
    {\switch{\taux = \_,\cnai(j_x,0)
        \atop{\taub \potau \taux \potau \taua}}
      {\ptr{\taub,\taua}{\taux}}
      {\mac{\nonce^{j_x}}{\mkey^\ID}{4}}}
    {\textsf{error}}
  \]
  Let $\kappa_\taua^\ID$ be the right-hand side of the equation above. For every $\taux = \_,\cnai(j_x,0)$ such that $\tautt \potau \taut$, we have $\nonce^{j_x} \in \setmac_{\ID}^{4}(\alpha_\taua^\ID,\kappa_\taua^\ID)$ if and only if $\taub \potau \taux \potau \taua$. Therefore:
  \begin{alignat*}{2}
    &&&\setmac_{\ID}^{4}\left(\alpha_\taua^\ID,\kappa_\taua^\ID\right)
    \cap
    \mathcal{N}\\
    &=\;\;&&
    \left\{
      \nonce^{j_x}
      \mid
      \taux = \_,\cnai(j_x,0) \wedge
      \taub \potau \taux \potau \taua
    \right\}
    \cap
    \left\{
      \nonce^{j_0}
      \mid
      \taut = \_,\cnai(j_0,0) \wedge \tautt \potau \taut
    \right\}\\
    &=\;\;&&
    \left\{
      \nonce^{j_x}
      \mid
      \taux = \_,\cnai(j_x,0) \wedge
      \taub \potau \taux \potau \taua
      \wedge
      \tautt \potau \taux
    \right\}
  \end{alignat*}
  By validity of $\tau$, we know that $\taua \potau \tautt$. This implies that whenever $\taub \potau \taux \potau \taua$ and $\tautt \potau \taux$, we have $\taux \potau \tautt \potau \taux$. Hence:
  \begin{equation*}
    \setmac_{\ID}^{4}\left(\alpha_\taua^\ID,\kappa_\taua^\ID\right)
    \cap
    \mathcal{N}
    \;=\;
    \emptyset
    \numberthis\label{eq:fnsdjhioafjqepir}
  \end{equation*}
  Let $\psi_0$ be $\psi$ in which we replace, for every $\taua = \_,\cuai_\ID(j_a,1) \popre \tau$, any occurrence of $\accept_\taua^\ID$ and $t_\taua$ by, respectively, $\alpha_\taua^\ID$ and $\kappa_\taua^\ID$. We then have:
\begin{alignat*}{3}
  \setmac_{\ID}^{4}\left(\psi_0\right)\;=\;
  \left\{
    \nonce^{j_n}
    \mid
    \taun = \_,\cnai(j_n,1) \popre \tau
  \right\}
  \;\cup\;
  \bigcup_{
    \taua = \_,\cuai_\ID(j_a,1)
    \atop{\taua \popre \tau}}
  \setmac_{\ID}^{4}\left(\alpha_\taua^\ID,\kappa_\taua^\ID\right)
\end{alignat*}
And using \eqref{eq:fnsdjhioafjqepir}, we know that:
\[
  \setmac_{\ID}^{4}\left(\psi_0\right)
  \cap
  \mathcal{N} =
  \left\{
    \nonce^{j_n}
    \mid
    \taun = \_,\cnai(j_n,1) \popre \tau
  \right\}
  \numberthis\label{eq:gisgieruqpwruqwp}
\]
\item Let $\taun = \_,\cnai(j_n,1)$ and $\taun' = \_,\cnai(j_n,0)$ such that $\taun' \potau \taun$. Using \ref{sequ3}, we know that:
  \[
    \accept_\taun^\ID
    \;\leftrightarrow\;\;
    \bigvee
    _{\taui' = \_,\cuai_\ID(j_i,0)
      \atop{\taui = \_,\cuai_\ID(j_i,1)
        \atop{\taui' \potau \taun' \potau \taui \potau \taun}}}
    \ftr{\taui',\taui}{\taun',\taun}
  \]
  Let $\lambda_\taun^\ID$ be the right-hand side of the equation above. We can check that $\nonce^{j_n} \in \setmac_{\ID}^{4}(\lambda_\taun^\ID)$ if and only if there exists $\taui' = \_,\cuai_\ID(j_i,0)$ and $\taui = \_,\cuai_\ID(j_i,1)$ such that $\taui' \potau \taun' \potau \taui \potau \taun$.  Since $\taui \popre \tau$, we know that $j_i \ne j$. Therefore $\taui \potau \tautt$, and we can show that:
  \begin{equation*}
    \setmac_{\ID}^{4}\left(\lambda_\taun^\ID\right)
    \cap
    \mathcal{N}
    \;=\;
    \emptyset
    \numberthis\label{eq:sfjlghwifqeirqjsd}
  \end{equation*}
  Let $\psi_1$ be $\psi_0$ in which we replace, for every $\taun = \_,\cnai(j_n,1)$ and $\taun' = \_,\cnai(j_n,0)$ such that $\taun' \potau \taun$, any occurrence of $\accept_\taun^\ID$ by $\lambda_\taua^\ID$. Using \eqref{eq:gisgieruqpwruqwp} and \eqref{eq:sfjlghwifqeirqjsd}, we can check that:
  \[
    \setmac_{\ID}^{4}\left(\psi_1\right)
    \cap
    \mathcal{N}
    \;=\;
    \emptyset
  \]
  Which is what we wanted to show.
\end{itemize}

\paragraph{Part 4}
Let $\taut = \_,\cnai(j_0,0)$ be such that $\tautt \potau \taut$. For every $\taut' = \_,\cnai(j'_0,0)$ be such that $\tautt \potau \taut'$, if $j_0' \ne j_0$ then $(\nonce^{j_0'} = \nonce^{j_0} ) \lra \false$. Moreover, since $\setmac_{\ID}^{4}\left(\psi_1\right) \cap \mathcal{N} = \emptyset$, we know that for every $\nonce \in \setmac_{\ID}^{4}\left(\psi_1\right)$, $(\nonce = \nonce^{j_0} ) \lra \false$.

We can therefore apply simultaneously the $\prfmac^4$ axiom with key $\mkey^\ID$ for every $\taut = \_,\cnai(j_0,0)$ be such that $\tautt \potau \taut$, to replace $\mac{\nonce^{j_0}}{\mkey^\ID}{4}$ by a fresh nonce $\nonce_{\taut}$. We then rewrite back $\psi_1$ into $\psi$. This yield the derivation:
\[
  \infer[R]{
    \inframe_\tau, \lreveal_{\tauo},
    \left(
      \mac{\nonce^{j_0}}{\mkey^\ID}{4}
    \right)_{\taut = \_,\cnai(j_0,0)\atop{\tautt \potau \taut}}
    \;\sim\;
    \zeta
  }{
    \infer[\prfmac^4]{
      \psi_1,
      \left(
        \mac{\nonce^{j_0}}{\mkey^\ID}{4}
      \right)_{\taut = \_,\cnai(j_0,0)\atop{\tautt \potau \taut}}
      \;\sim\;
      \zeta
    }{
      \infer[R]{
        \psi_1,
        \left(
          \nonce_{\taut}
        \right)_{\taut = \_,\cnai(j_0,0)\atop{\tautt \potau \taut}}
        \;\sim\;
        \zeta
      }{
        \inframe_\tau, \lreveal_{\tauo},
        \left(
          \nonce_{\taut}
        \right)_{\taut = \_,\cnai(j_0,0)\atop{\tautt \potau \taut}}
        \;\sim\;
        \zeta
      }
    }
  }
\]
where:
\[
  \zeta \equiv
  \inframe_\utau,\rreveal_{\tauo},
  \left((\mac{\nonce^{j_0}}{\mkey^{\uID}}{4}\right)
  _{\taut = \_,\cnai(j_0,0)\atop{\tautt \potau \taut}}
\]
Observe that we never used the fact that $\tau$ was a \emph{basic} trace of actions above, but only the fact that $\tau$ is a \emph{valid} trace of actions. Therefore the same reasoning applies to $\zeta$, which allows us, for every $\taut = \_,\cnai(j_0,0)$ be such that $\tautt \potau \taut$, to replace $\mac{\nonce^{j_0}}{\mkey^\uID}{4}$ by a fresh nonce $\nonce_{\taut}'$. We conclude using \ax{fresh}. We get:
\[
  \infer[R+\prfmac^4]{
    \inframe_\tau, \lreveal_{\tauo},
    \left(
      \nonce_{\taut}
    \right)_{\taut = \_,\cnai(j_0,0)\atop{\tautt \potau \taut}}
    \;\sim\;
    \inframe_\utau,\rreveal_{\tauo},
    \left((\mac{\nonce^{j_0}}{\mkey^{\uID}}{4}\right)
    _{\taut = \_,\cnai(j_0,0)\atop{\tautt \potau \taut}}
  }{
    \infer[\ax{fresh}]{
      \inframe_\tau, \lreveal_{\tauo},
      \left(
        \nonce_{\taut}
      \right)_{\taut = \_,\cnai(j_0,0)\atop{\tautt \potau \taut}}
      \;\sim\;
      \inframe_\utau,\rreveal_{\tauo},
      \left(
        \nonce_\taut'
      \right)
      _{\taut = \_,\cnai(j_0,0)\atop{\tautt \potau \taut}}
    }{
      \inframe_\tau, \lreveal_{\tauo}
      \;\sim\;
      \inframe_\utau,\rreveal_{\tauo}
    }
  }
\]
Which concludes this proof.

\subsection{Case $\ai = \cnai(j,1)$}
We know that $\uai = \cnai(j,1)$. Here $\lreveal_{\tau}$ and $\lreveal_{\tauo}$ coincides everywhere except on:
\begin{mathpar}
  \neauth_\tau(\ID,j)
  \;\sim\;
  \uneauth_\utau(\ID,j)

  \syncdiff_\tau^\ID
  \;\sim\;
  \syncdiff_\utau^{\nu_\tau(\ID)}
\end{mathpar}
Let $\ID \in \baseiddom$, $\taui = \_,\cuai_\ID(j_i,1)$, $\taut = \_,\cnai(j,0)$, $\tautt = \_,\cuai_\ID(j_i,0)$ such that $\tautt \potau \taut \potau \taui$:
\begin{center}
  \begin{tikzpicture}
    [dn/.style={inner sep=0.2em,fill=black,shape=circle},
    sdn/.style={inner sep=0.15em,fill=white,draw,solid,shape=circle},
    sl/.style={decorate,decoration={snake,amplitude=1.6}},
    dl/.style={dashed},
    pin distance=0.5em,
    every pin edge/.style={thin}]

    \draw[thick] (0,0)
    node[left=1.3em] {$\tau:$}
    -- ++(0.5,0)
    node[dn,pin={above:{$\cuai_\ID(j_i,0)$}}] {}
    node[below=0.3em]{$\tautt$}
    -- ++(2.5,0)
    node[dn,pin={above:{$\cnai(j,0)$}}] {}
    node[below=0.3em]{$\taut$}
    -- ++(2.5,0)
    node[dn,pin={above:{$\cuai_\ID(j_i,1)$}}] {}
    node[below=0.3em]{$\taui$}
    -- ++(2.5,0)
    node[dn,pin={above:{$\ai = \cnai(j,1) $}}] {}
    node[below=0.3em]{$\tau$};
  \end{tikzpicture}
\end{center}
Let $\sff \equiv \ftr{\tautt,\taui}{\taut,\tau}$ and $\usff \equiv \ftr{\utautt,\utaui}{\utaut,\utau}$. Using \ref{der4} we know that we have the following derivation:
\[
  \begin{gathered}[c]
    \infer[\simp]{
      \inframe_\tau,
      \lreveal_{\tauo},
      \sff
      \;\sim\;
      \inframe_\utau,
      \rreveal_{\tauo},
      \usff
    }{
      \inframe_\tau, \lreveal_{\tauo}
      \sim
      \inframe_\utau,\rreveal_{\tauo}
    }
  \end{gathered}
  \numberthis\label{eq:ascvjkxo}
\]
Since $\sff \ra \accept_\tau^\ID$, we have:
\[
  \cond{\sff \wedge \instate_\tau(\sync_\ue^\ID)}
  {\syncdiff_\tau^\ID} =
  \cond{\sff \wedge \instate_\tau(\sync_\ue^\ID)}{
    \left(\begin{alignedat}{2}
        \ite{\instate_\tau(\tsuccess_\hn^\ID) = \nonce^j}
        {& \sqnsuc(\syncdiff_\tauo^\ID)\\}
        {& \syncdiff_\tauo^\ID}
      \end{alignedat}\right)}
\]

\paragraph{Case 1}
Assume that $\taui = \_,\cuai_\ID(j_i,1) \potau \newsession_\ID(\_)$. Let $\tau_\ns = \_,\ns_\ID(j_\newsession)$ be the latest session reset in $\tau$, i.e. $\tau_\ns <_\tau\tau$ and $\tau_\ns \not \potau \newsession_\ID(\_)$. We show by induction that for every $\tau'$ such that $\tau_\ns \popreleq \tau'$ we have:
\[
  \left(
    f \wedge \instate_\tau(\tsuccess_\hn^\ID) = \nonce^j
  \right)
  \;\ra\;
  \cstate_{\tau_\ns}(\sqn_\hn^\ID) =
  \cstate_{\tau'}(\sqn_\hn^\ID)
  \numberthis\label{eq:fsdhkvzxclka}
\]
Let $\tau'$ be such that $\tau_\ns \popreleq \tau'$:
\begin{itemize}
\item If $\tau' = \tau_\ns$ then the property trivially holds.
\item If $\tau_\ns \potau \tau'$. The only cases where $\sqn_\hn^\ID$ is updated are $\pnai(j',1)$ and $\cnai(j',1)$:
  \begin{itemize}
  \item If $\tau' = \_,\pnai(j',1)$: since $\tau = \cnai(j,1)$ we know by validity of $\tau$ that $j' \ne j$. Therefore:
    \[
      \incaccept_{\tau'}^\ID
      \;\ra\;
      \left(\cstate_{\tau'}(\tsuccess_\hn^\ID) = \nonce^{j'}\right)
      \;\ra\;
      \left(\cstate_{\tau'}(\tsuccess_\hn^\ID) \ne \nonce^{j}\right)
      \;\ra\;
      \left(\instate_\tau(\tsuccess_\hn^\ID) \ne \nonce^j\right)
    \]
    It follows that:
    \begin{alignat*}{2}
      \left(\instate_\tau(\tsuccess_\hn^\ID) = \nonce^j\right)
      \;\ra\;
      \neg \incaccept_{\tau'}^\ID
      \;\ra\;
      \left(
        \instate_{\tau'}(\sqn_\hn^\ID) =
        \cstate_{\tau'}(\sqn_\hn^\ID)
      \right)
    \end{alignat*}
    And we conclude by applying the induction hypothesis.
  \item If $\tau' = \_,\cnai(j',1)$: since $\tau = \cnai(j,1)$ and $\tau' \popre \tau$, we know that $j' \ne j$ (by validity of $\tau$). Therefore:
    \begin{alignat*}{2}
      \left(\instate_\tau(\tsuccess_\hn^\ID) = \nonce^j\right)
      \;\ra\;
      \neg \incaccept_{\tau'}^\ID
      \;\ra\;
      \left(
        \instate_{\tau'}(\sqn_\hn^\ID) =
        \cstate_{\tau'}(\sqn_\hn^\ID)
      \right)
    \end{alignat*}
    And we conclude by applying the induction hypothesis.
  \end{itemize}
\end{itemize}
This concludes the proof of \eqref{eq:fsdhkvzxclka}.

We then prove by induction over $\tau'$, for $\newsession_\ID(j_\newsession) \popreleq \tau' \popreleq \tau$ we have:
\[
  \left(
    f \wedge \instate_\tau(\tsuccess_\hn^\ID) = \nonce^j
  \right)
  \;\ra\;
  \neg \cstate_{\tau'}(\sync_\ue^\ID)
  \numberthis\label{eq:vzxclka}
\]
Let $\ai'$ be such that $\tau' = \_,\ai'$.
\begin{itemize}
\item The case $\ai' = \newsession_\ID(j_\newsession)$ is trivial since we then have $\cstate_{\tau'}(\sync_\ue^\ID) = \false$.

\item If $\ai' \ne \npuai{2}{\ID}{\_}$, then since $\newsession(j_\newsession) \not \potau \newsession(\_)$ we know that $\ai' \ne \newsession(\_)$. Hence $\upstate_{\tau'}(\sync_\ue^\ID) = \bot$, which implies $\cstate_{\tau'}(\sync_\ue^\ID) \equiv \instate_{\tau'}(\sync_\ue^\ID)$. By induction hypothesis we know that:
  \[
    \left(
      f \wedge \instate_\tau(\tsuccess_\hn^\ID) = \nonce^j
    \right)
    \;\ra\;
    \neg \instate_{\tau'}(\sync_\ue^\ID)
  \]
  which concludes this case.

\item If $\ai' = \npuai{2}{\ID}{j'}$.  Let $\tau''' = \_,\npuai{1}{\ID}{j'} <_\tau$. By validity of $\tau$ we know that $\tau_\ns \potau \tau'''$. Using \ref{equ2} we know that:
  \[
    \accept_{\tau'}^\ID
    \;\lra\;
    \bigvee_{\tau'' = \_,\pnai(j'',1)
      \atop{\tau''' \potau \tau'' \potau \tau'}}
    \supitr{\tau''',\tau'}{\tau''}
  \]
  And using \ref{sequ4}:
  \[
    \left(
      \neg\instate_{\tau'}(\sync_\ue^\ID)
      \wedge \supitr{\tau''',\tau'}{\tau''}
    \right)
    \;\ra\;
    \idiff{\cstate_{\tau'}(\sqn_\ue^\ID)}{\cstate_{\tau'}(\sqn_\hn^\ID)} = \zero
  \]
  Using \eqref{eq:fsdhkvzxclka}, we know that:
  \[
    \left(
      f \wedge \instate_\tau(\tsuccess_\hn^\ID) = \nonce^j
    \right)
    \;\ra\;
    \left(
      \cstate_{\tau_\ns}(\sqn_\hn^\ID) =
      \instate_{\tau'''}(\sqn_\hn^\ID) \wedge
      \cstate_{\tau_\ns}(\sqn_\hn^\ID) =
      \cstate_{\tau'}(\sqn_\hn^\ID)
    \right)
  \]
  Therefore:
  \[
    \left(
      f \wedge \instate_\tau(\tsuccess_\hn^\ID) = \nonce^j
    \right)
    \;\ra\;
    \left(
      \instate_{\tau'''}(\sqn_\hn^\ID) =
      \cstate_{\tau'}(\sqn_\hn^\ID)
    \right)
  \]
  Using \ref{b4} we know that $\instate_{\tau'''}(\sqn_\hn^\ID) \le \instate_{\tau'''}(\sqn_\ue^\ID)$, and by \ref{b5} we know that $\cstate_{\tau'''}(\sqn_\hn^\ID) \le \cstate_{\tau'}(\sqn_\ue^\ID)$. Moreover $\cstate_{\tau'''}(\sqn_\hn^\ID) = \sqnsuc(\instate_{\tau'''}(\sqn_\hn^\ID)) < \instate_{\tau'''}(\sqn_\hn^\ID)$. We summarize all of this graphically below:
  \begin{center}
    \begin{tikzpicture}
      [dn/.style={inner sep=0.2em,fill=black,shape=circle},
      sdn/.style={inner sep=0.15em,fill=white,draw,solid,shape=circle},
      sl/.style={decorate,decoration={snake,amplitude=1.6}},
      dl/.style={dashed},
      pin distance=0.5em,
      every pin edge/.style={thin}]

      \draw[thick] (0,0)
      node[left=1.3em] {$\tau:$}
      -- ++(0.5,0)
      node[below=0.3em] {$\tau_\ns$}
      node[dn,pin={above:{$\ns_\ID(j_\ns)$}}]
      (a) {}
      -- ++(3,0)
      node[below=0.3em] {$\tau'''$}
      node[dn,pin={above:{$\npuai{1}{\ID}{j'}$}}]
      (b) {}
      -- ++(3,0)
      node[dn,pin={above:{}}]
      (c) {}
      -- ++(3,0)
      node[below=0.3em] {$\tau''$}
      node[dn,pin={above:{$\pnai(j'',1)$}}]
      (d) {}
      -- ++(3,0)
      node[below=0.3em] {$\tau'$}
      node[dn,pin={above:{$\npuai{2}{\ID}{j'}$}}]
      (e) {}
      -- ++(0.5,0);

      \path (b) -- ++ (0,-0.8)
      node (b1) {$\instate_{\tau'''}(\sqn_\hn^\ID)$}
      -- ++ (0,-1)
      node (b2) {$\instate_{\tau'''}(\sqn_\ue^\ID)$};

      \path (c) -- ++ (0,-0.8)
      -- ++ (0,-1)
      node (c2) {$\cstate_{\tau'''}(\sqn_\ue^\ID)$};

      \path (e) -- ++ (0,-0.8)
      node (e1) {$\cstate_{\tau'}(\sqn_\hn^\ID)$}
      -- ++ (0,-1)
      node (e2) {$\cstate_{\tau'}(\sqn_\ue^\ID)$};

      \draw (e1) -- (e2) node[midway,above,sloped]{$=$};
      \draw (b1) -- (e1) node[midway,above,sloped]{$=$};
      \draw (b1) -- (b2) node[midway,above,sloped]{$\le$};
      \draw (b2) -- (c2) node[midway,above,sloped]{$<$};
      \draw (c2) -- (e2) node[midway,above,sloped]{$\le$};
    \end{tikzpicture}
  \end{center}
  Putting everything together we get that:
  \[
    \left(
      \sff \wedge
      \neg\instate_{\tau'}(\sync_\ue^\ID)
      \wedge \supitr{\tau''',\tau'}{\tau''}
    \right)
    \;\ra\;
    \left(\cstate_{\tau'}(\sqn_\ue^\ID) < \cstate_{\tau'}(\sqn_\ue^\ID)\right)
    \;\ra\;
    \false
  \]
  We deduce that:
  \begin{alignat*}{2}
    \left(
      \sff\wedge \neg\instate_{\tau'}(\sync_\ue^\ID) \wedge \accept_{\tau'}^\ID
    \right)
    &\;\ra\;\;&&
    \bigvee_{\tau'' = \_,\pnai(j'',1)
      \atop{\npuai{1}{\ID}{j'} \potau \tau'' \potau \tau'}}
    \left(
      \sff\wedge \neg\instate_{\tau'}(\sync_\ue^\ID) \wedge
      \supitr{\_,\tau'}{\tau''}
    \right)\\
    &\;\ra\;\;&&
    \false
  \end{alignat*}
  Moreover, using the induction hypothesis we know that:
  \[
    \left(
      f \wedge \instate_\tau(\tsuccess_\hn^\ID) = \nonce^j
    \right)
    \;\ra\;
    \neg \instate_{\tau'}(\sync_\ue^\ID)
  \]
  Therefore:
  \begin{alignat*}{3}
    \left(
      f \wedge \instate_\tau(\tsuccess_\hn^\ID) = \nonce^j
    \right)
    &\;\ra\;\;&
    \neg \accept_{\tau'}
    &\;\ra\;\;&
    \neg \cstate_{\tau'}(\sync_\ue^\ID)
  \end{alignat*}
\end{itemize}
This concludes the proof by induction of~\eqref{eq:vzxclka}. Using~\eqref{eq:vzxclka} we get that:
\[
  \cond{\sff}{\syncdiff_\tau^\ID} =
  \cond{\sff \wedge \instate_\tau(\sync_\ue^\ID)}{\syncdiff_\tauo^\ID}
\]
We know that $\usff \ra \accept_\utau^{\nu_\tautt(\ID)}$. Moreover, $\nu_\tautt(\ID) \ne \nu_\tau(\ID)$, hence using \ref{a5} we know that $\usff \ra \neg \accept_\utau^{\nu_\tau(\ID)}$. Hence:
\[
  \cond{\usff}
  {\syncdiff_\utau^{\nu_\tau(\ID)}} =
  \cond{\usff \wedge \instate_\utau(\sync_\ue^{\nu_\tau(\ID)})}
  {\syncdiff_\utauo^{\nu_\tau(\ID)}}
\]
Using the derivation in~\eqref{eq:ascvjkxo} and the fact that:
\begin{mathpar}
  \left(
    \instate_\tau(\sync_\ue^\ID),
    \instate_\utau(\sync_\ue^{\nu_\tau(\ID)})
  \right) \in \reveal_{\tauo}

  \left(
    \syncdiff_\tau^\ID,
    \syncdiff_\utauo^{\nu_\tau(\ID)}
  \right) \in \reveal_{\tauo}
\end{mathpar}
We can build the derivation:
\[
  \begin{gathered}[c]
    \infer[\simp]{
      \begin{alignedat}{2}
        \inframe_\tau,
        \lreveal_{\tauo},
        \cond{\sff}
        {\syncdiff_\tau^\ID}
        \;\sim\;
        \inframe_\utau,
        \rreveal_{\tauo},
        \cond{\usff}
        {\syncdiff_\utau^{\nu_\tau(\ID)}}
      \end{alignedat}
    }{
      \infer[\simp]{
        \begin{alignedat}{6}
          &&&\inframe_\tau,&&
          \lreveal_{\tauo},&&
          \sff,&&
          \instate_\tau(\sync_\ue^\ID),&&
          \syncdiff_\tauo^\ID\\
          &\sim\;\;&&
          \inframe_\utau,\;&&
          \rreveal_{\tauo},\;&&
          \usff,\;&&
          \instate_\utau(\sync_\ue^{\nu_\tau(\ID)}),\;&&
          \syncdiff_\utauo^{\nu_\tau(\ID)}
        \end{alignedat}
      }{
        \inframe_\tau, \lreveal_{\tauo}
        \sim
        \inframe_\utau,\rreveal_{\tauo}
      }
    }
  \end{gathered}
  \numberthis\label{eq:vcofeojgsdpifjajn}
\]

\paragraph{Case 2}
Assume that $\taui = \_,\cuai_\ID(j_i,1) \not\potau \newsession_\ID(\_)$. We introduce the term $\theta_\pnai$ (resp. $\theta_\cnai$) which states that no $\supi$ (resp. $\suci$) network session has been initiated which $\ID$ between $\taut$ and $\tau$:
\begin{mathpar}
  \theta_\pnai \;\;\equiv\;\;
  \bigwedge_{\tau' = \_,\pnai(\_,1) \atop{\taut \potau \tau'}}
  \neg\incaccept_{\tau'}^\ID

  \theta_\cnai \;\;\equiv\;\;
  \bigwedge_{\tau' = \cnai(\_,0) \atop{\taut \potau \tau'}}
  \neg\accept_{\tau'}^\ID
\end{mathpar}
It is easy to show that:
\[
  \left(\sff \wedge \instate_\tau(\tsuccess_\hn^\ID) = \nonce^j\right)
  \;\lra\;
  \left(\sff \wedge \theta_\pnai \wedge \theta_\cnai\right)
\]
We are now going to show that for every $\taut \popreleq \tau'$, $P(\tau')$ holds where $P(\tau')$ is the term:
\begin{equation}
  P(\tau') \;\;\equiv\;\;
  \left(\sff \wedge \theta_\pnai\right) \;\ra\;
  \left(
    \begin{array}[c]{l}
      \cstate_{\tau'}(\suci_\hn^\ID) = \unset\\
      \wedge \; \cstate_{\tau'}(\tsuccess_\hn^\ID) = \nonce^j
    \end{array}
    \wedge
    \bigwedge_{\taut \potau \tau'' \popreleq \tau'\atop{\tau'' = \cnai(\_,0)}}
    \neg\accept_{\tau''}^\ID
  \right)
  \label{eq:asdvjzxlv}
\end{equation}
Since $\sff \ra \accept_\taut$, we know that $\sff \ra\cstate_\taut(\suci_\hn^\ID) = \unset$. This shows that $P(\suc_\tau(\taut))$ holds. Let $\taut \popreleq \tau'_0$, and assume $P(\tau'_0)$ holds by induction. Let $\tau' = \suc_\tau(\tau'_0)$. We have four cases:
\begin{itemize}
\item If $\tau' \not \in \{\cnai(\_,0),\cnai(\_,1),\pnai(\_,1)\}$ then $P(\tau') \equiv P(\tau'_0)$, which concludes this case.

\item If $\tau' = \cnai(\_,0)$, then using the induction hypothesis $P(\tau'_0)$ we know that $\sff \wedge \theta_\pnai \ra \instate_{\tau'}(\suci_\hn^\ID) = \unset$. Therefore $\sff \wedge \theta_\pnai \ra \neg \accept_{\tau'}^\ID$. We know that $\sff \wedge \theta_\pnai \ra \instate_{\tau'}(\tsuccess_\hn^\ID) = \nonce^j$. We conclude by observing that:
  \[
    \left(
      \neg \accept_{\tau'}^\ID \wedge
      \instate_{\tau'}(\suci_\hn^\ID) = \unset \wedge
      \instate_{\tau'}(\tsuccess_\hn^\ID) = \nonce^j
    \right)
    \ra
    \left(
      \cstate_{\tau'}(\suci_\hn^\ID) = \unset \wedge
      \cstate_{\tau'}(\tsuccess_\hn^\ID) = \nonce^j
    \right)
  \]

\item If $\tau' = \cnai(j',1)$. Since $\tau' \popre \tau$, we know that $j \ne j'$. Therefore $\instate_{\tau'}(\tsuccess_\hn^\ID) = \nonce^j \ra \instate_{\tau'}(\tsuccess_\hn^\ID) \ne \nonce^{j'}$. We deduce that $\sff \wedge \theta_\pnai \ra \neg \accept_{\tau'}^\ID$. This concludes this case.

\item If $\tau' = \_,\pnai(\_,1)$. We know that:
  \[
    \sff \wedge \theta_\pnai \ra
    \neg \incaccept_{\tau'}^\ID
  \]
  We then directly conclude using the facts that:
  \begin{mathpar}
    \neg \incaccept_{\tau'}^\ID \ra
    \cstate_{\tau'}(\tsuccess_\hn^\ID) =
    \instate_{\tau'}(\tsuccess_\hn^\ID)

    \neg \incaccept_{\tau'}^\ID \ra
    \cstate_{\tau'}(\suci_\hn^\ID) =
    \instate_{\tau'}(\suci_\hn^\ID)
  \end{mathpar}
\end{itemize}
By applying \eqref{eq:asdvjzxlv} at instant $\tauo$, we get that:
\[
  \left(\sff \wedge \instate_\tau(\tsuccess_\hn^\ID) = \nonce^j\right)
  \;\lra\;
  \left(\sff \wedge \theta_\pnai \wedge \theta_\cnai\right)
  \;\lra\;
  \left(\sff \wedge \theta_\pnai\right)
  \numberthis\label{eq:vioiqperqfa}
\]

\paragraph{Part 1}
Let $\tau' = \_,\pnai(j',1)$, with $\taut \potau \tau'$. Let $\tau'_0 = \pnai(j',0)$. Using \ref{equ3} we know that:
\[
  \accept_{\tau'}^\ID
  \;\lra\;
  \bigvee_{\taua = \_,\npuai{1}{\ID}{j_a}
    \atop{\tau'_0 \potau \taua \potau \tau'}}
  \underbrace{\left(
      \begin{alignedat}{2}
        &&&g(\inframe_{\taua}) = \nonce^{j'}
        \wedge
        \pi_1(g(\inframe_{\tau'})) =
        \enc{\spair
          {\ID}
          {\instate_{\taua}(\sqn_\ue^\ID)}}
        {\pk_\hn}{\enonce^{j_a}}\\
        &\wedge\;&&
        \pi_2(g(\inframe_{\tau'})) =
        {\mac{\spair
            {\enc{\spair
                {\ID}
                {\instate_{\taua}(\sqn_\ue^\ID)}}
              {\pk_\hn}{\enonce^{j_a}}}
            {g(\inframe_{\taua})}}
          {\mkey^\ID}{1}}
      \end{alignedat}
    \right)}_{\lambda_\taua^{\tau'}}
  \numberthis\label{eq:avcxvschiaa}
\]
We define:
\[
  \tau_\newsession =
  \begin{dcases*}
    \newsession_\ID(j_\newsession) & if there exists $j_\newsession$ s.t. $\newsession_\ID(j_\newsession) <_\tau\tau$ and $\newsession_\ID(j_\newsession) \not \potau \newsession_\ID(\_)$.    \\
    \epsilon & otherwise
  \end{dcases*}
\]
Let $\taua = \_,\npuai{1}{\ID}{j_a}$ such that $\tau'_0 \potau \taua \potau \tau'$. Since $\taui = \_,\cuai_\ID(j_i,1) \not\potau \newsession_\ID(\_)$, we have only three interleavings possible: $\taut \potau \taua$, $\taut \potau \taua$ or $\tau_\newsession \potau \taua \potau \tautt$. First, we are going to show that in the first two cases we have:
\[
  \sff \wedge \lambda_\taua^{\tau'}
  \;\ra\;
  \neg \incaccept_{\tau'}^\ID
\]
\begin{itemize}
\item If $\taua \potau \tau_\newsession$, we have the following interleaving:
  \begin{center}
    \begin{tikzpicture}
      [dn/.style={inner sep=0.2em,fill=black,shape=circle},
      sdn/.style={inner sep=0.15em,fill=white,draw,solid,shape=circle},
      sl/.style={decorate,decoration={snake,amplitude=1.6}},
      dl/.style={dashed},
      pin distance=0.5em,
      every pin edge/.style={thin}]

      \draw[thick] (0,0)
      node[left=1.3em] {$\tau:$}
      -- ++(0.5,0)
      node[dn,pin={above:{$\npuai{1}{\ID}{j_a}$}}] {}
      node[below=0.3em,name=a0]{$\taua$}
      -- ++(2,0)
      node[dn,pin={above:{$\newsession_\ID(\_)$}}] {}
      node[below=0.3em]{$\tau_\newsession$}
      -- ++(2,0)
      node[dn,pin={above:{$\cuai_\ID(j_i,0)$}}] {}
      node[below=0.3em,name=b0]{$\tautt$}
      -- ++(2,0)
      node[dn,pin={above:{$\cnai(j,0)$}}] {}
      node[below=0.3em,name=b1]{$\taut$}
      -- ++(2,0)
      node[dn,pin={above:{$\pnai(j',1)$}}] {}
      node[below=0.3em,name=a1]{$\tau'$}
      -- ++(2,0)
      node[dn,pin={above:{$ \cnai(j,1) $}}] {}
      node[below=0.3em,name=b2]{$\tau$};

      \draw[thin,dashed] (a0) -- ++(0,-0.5) -| (a1);
      \draw[thin,densely dotted] (b0) -- ++(0,-0.8) -| (b1);
      \path (b1) -- ++(0,-0.8) node[draw=none,inner sep=0] (b1p){};
      \draw[thin,densely dotted] (b1p) -| (b2);
    \end{tikzpicture}
  \end{center}
  By definition of $\incaccept_{\tau'}^\ID$, and using the fact that $\lambda_\taua^{\tau'} \ra \accept_{\tau'}^\ID$ we know that:
  \[
    \left(
      \lambda_\taua^{\tau'} \wedge \incaccept_{\tau'}^\ID
    \right)
    \;\ra\;
    \instate_{\tau'}(\sqn_\hn^\ID) \le \instate_\taua(\sqn_\ue^\ID)
  \]
  To conclude this case, we only need to show that:
  \[
    \left(
      \lambda_\taua^{\tau'} \wedge \incaccept_{\tau'}^\ID
    \right)
    \;\ra\;
    \instate_\taua(\sqn_\ue^\ID) < \instate_{\tau'}(\sqn_\hn^\ID)
    \numberthis\label{qe:azxzxof}
  \]
  From which we obtain directly a contradiction, implies that:
  \[
    \sff \wedge \lambda_\taua^{\tau'}
    \;\ra\;
    \neg \incaccept_{\tau'}^\ID
    \qquad\qquad\text{ when $\taua \potau \tau_\newsession$}
    \numberthis\label{eq:rutwdfsdfvsjov}
  \]
  The proof of \eqref{qe:azxzxof} is straightforward using \ref{b5} and \ref{b6}, we just give the proof graphically below:
  \begin{center}
    \begin{tikzpicture}
      [dn/.style={inner sep=0.2em,fill=black,shape=circle},
      sdn/.style={inner sep=0.15em,fill=white,draw,solid,shape=circle},
      sl/.style={decorate,decoration={snake,amplitude=1.6}},
      dl/.style={dashed},
      pin distance=0.5em,
      every pin edge/.style={thin}]

      \draw[thick] (0,0)
      node[left=1.3em] {$\tau:$}
      -- ++(0.5,0)
      node[below=0.3em] {$\taua$}
      node[dn,pin={above:{$\npuai{1}{\ID}{j_a}$}}]
      (a) {}
      -- ++(3,0)
      node[below=0.3em] {$\tau_\newsession$}
      node[dn,pin={above:{$\newsession_\ID(j_\newsession)$}}]
      (b) {}
      -- ++(3,0)
      node[below=0.3em] {$\tau'$}
      node[dn,pin={above:{$\pnai(j',1)$}}]
      (c) {}
      -- ++(0.5,0);

      \path (a) -- ++ (0,-0.8)
      node (a1) {$\instate_\taua(\sqn_\ue^\ID)$};

      \path (b) -- ++ (0,-0.8)
      node (b1) {$\instate_{\tau_\newsession}(\sqn_\ue^\ID)$};

      \path (c) -- ++ (0,-0.8)
      -- ++ (0,-1)
      node (c2) {$\instate_{\tau'}(\sqn_\hn^\ID)$};

      \draw (a1) -- (b1) node[midway,above]{$\le$};
      \draw (b1) -- (c2) node[midway,above]{$<$};
    \end{tikzpicture}
  \end{center}

\item If $\tau_\newsession \potau \taua \potau \tautt$, we have the following interleaving:
  \begin{center}
    \begin{tikzpicture}
      [dn/.style={inner sep=0.2em,fill=black,shape=circle},
      sdn/.style={inner sep=0.15em,fill=white,draw,solid,shape=circle},
      sl/.style={decorate,decoration={snake,amplitude=1.6}},
      dl/.style={dashed},
      pin distance=0.5em,
      every pin edge/.style={thin}]

      \draw[thick] (0,0)
      node[left=1.3em] {$\tau:$}
      -- ++(0.5,0)
      node[dn,pin={above,align=left:{$\newsession_\ID(\_)$\\or $\epsilon$}}] {}
      node[below=0.3em]{$\tau_\newsession$}
      -- ++(2,0)
      node[dn,pin={above:{$\npuai{1}{\ID}{j_a}$}}] {}
      node[below=0.3em,name=a0]{$\taua$}
      -- ++(2,0)
      node[dn,pin={above:{$\cuai_\ID(j_i,0)$}}] {}
      node[below=0.3em,name=b0]{$\tautt$}
      -- ++(2,0)
      node[dn,pin={above:{$\cnai(j,0)$}}] {}
      node[below=0.3em,name=b1]{$\taut$}
      -- ++(2,0)
      node[dn,pin={above:{$\pnai(j',1)$}}] {}
      node[below=0.3em,name=a1]{$\tau'$}
      -- ++(2,0)
      node[dn,pin={above:{$\cnai(j,1) $}}] {}
      node[below=0.3em,name=b2]{$\tau$};

      \draw[thin,dashed] (a0) -- ++(0,-0.5) -| (a1);
      \draw[thin,densely dotted] (b0) -- ++(0,-0.8) -| (b1);
      \path (b1) -- ++(0,-0.8) node[draw=none,inner sep=0] (b1p){};
      \draw[thin,densely dotted] (b1p) -| (b2);
    \end{tikzpicture}
  \end{center}
  We know that $\lambda_\taua^{\tau'} \ra \cstate_\taua(\suci_\ue^\ID) = \unset$, and that $\sff \ra \instate_\tautt(\sync_\ue^\ID)$. By \ref{b1}, we get $\sff \ra \instate_\tautt(\suci_\ue^\ID) \ne \unset$. This means that $\suci_\ue^\ID$ is unset at $\taua$, but set at $\tautt$. Therefore there was a successful run of the protocol ($\supi$ or $\suci$) between $\taua$ and $\tautt$. More precisely, using Proposition~\ref{prop:unset-prop} we have:
  \begin{alignat*}{2}
    \sff \wedge \lambda_\taua^{\tau'}
    &\;\;\ra\;\;&&
    \left(
      \cstate_\taua(\suci_\ue^\ID) = \unset \wedge
      \instate_\tautt(\suci_\ue^\ID) \ne \unset
    \right)\\
    &\;\;\ra\;\;&&
    \bigvee_{\tau'' = \_,\fuai_\ID(j'')\atop{\taua \potau \tau'' \potau \tautt}}
    \accept_{\tau''}^\ID
    \numberthis\label{eq:nkvjva}
  \end{alignat*}
  Let $\tau'' = \_,\fuai_\ID(j'')$ such that $\taua \potau \tau'' \potau \tautt$. We then have two cases:
  \begin{itemize}
  \item Assume $j'' = j_a$. In order to have $\accept_{\tau''}^\ID$, we need the $\supi$ or $\suci$ session $j''$ to have been executed before $\tau''$. Intuitively, this cannot happen if $j'' = j_a$ because the user session $j_a$ is interacting with the network session $j'$, and $\tau'' \potau \tau'$. Formally, using the fact that $j'' = j_a$ we are going to show that:
    \[
      \left(
        \lambda_\taua^{\tau'} \wedge \accept_{\tau''}^\ID
      \right)
      \ra \false
      \numberthis\label{eq:sdnjksdf}
    \]
    First, by \ref{equ1} we know that:
    \begin{alignat*}{2}
      \accept_{\tau''}^\ID
      &\;\;\ra\;\;&&
      \bigvee_{\fnai(j_x) \not <_{\tau''} \newsession_\ID(\_)}
      \injauth_{\tau''}(\ID,j_x) \\
      &\;\;\ra\;\;&&
      \bigvee_{\fnai(j_x) \not <_{\tau''} \newsession_\ID(\_)}
      \left(
        \instate_{\tau''}(\bauth_\hn^{j_x}) = \ID \wedge
        \instate_{\tau''}(\eauth_\ue^{\ID}) = \nonce^{j_x}
      \right)\\
      \intertext{By \ref{a8} we get:}
      &\;\;\ra\;\;&&
      \bigvee_{\fnai(j_x) \not <_{\tau''} \newsession_\ID(\_)}
      \left(
        \instate_{\tau''}(\bauth_\hn^{j_x}) = \ID \wedge
        \instate_{\tau''}(\bauth_\ue^{\ID}) = \nonce^{j_x}
      \right)
      \numberthis\label{nkcvzioe}
    \end{alignat*}
    We know that $\lambda_\taua^{\tau'} \ra \instate_{\taua}(\bauth_\ue^{\ID}) = \nonce^{j'}$. Moreover, using the validity of $\tau$ we know that $\bauth_\ue^{\ID}$ is not updated between $\taua$ and $\tau''$, therefore $\lambda_\taua^{\tau'} \ra \instate_{\tau''}(\bauth_\ue^{\ID}) = \nonce^{j'}$. Putting this together with \eqref{nkcvzioe}, and using the fact that:
    \[
      \left(
        \instate_{\tau''}(\bauth_\ue^{\ID}) = \nonce^{j_x} \wedge
        \instate_{\tau''}(\bauth_\ue^{\ID}) = \nonce^{j'}
      \right) \ra \false \text{ if } j_x \ne j'
    \]
    We get:
    \begin{alignat*}{2}
      \accept_{\tau''}^\ID \wedge \lambda_\taua^{\tau'}
      &\;\;\ra\;\;&&
      \instate_{\tau''}(\bauth_\hn^{j'}) = \ID \wedge
      \instate_{\tau''}(\bauth_\ue^{\ID}) = \nonce^{j'}\displaybreak[1]\\
      \intertext{Since $\tau'' \potau \tau'$, we know that $\instate_{\tau''}(\bauth_\hn^{j'}) = \bot$. This yields a contradiction:}
      &\;\;\ra\;\;&&
      \instate_{\tau''}(\bauth_\hn^{j'}) = \ID \wedge
      \instate_{\tau''}(\bauth_\hn^{j'}) = \bot\\
      &\;\;\ra\;\;&&
      \false
    \end{alignat*}
    Which concludes the proof of \eqref{eq:sdnjksdf}.

  \item Assume $j'' \ne j_a$. Intuitively, we know that $\accept_{\tau''}^\ID$ implies that $\sqn^\ID_\ue$ and $\sqn^\ID_\hn$ have been incremented and synchronized between $\taua$ and $\tau'$. Therefore we know that the test $\incaccept_{\tau'}^\ID$ fails. Formally, we show that:
    \begin{alignat*}{2}
      \accept_{\tau''}^\ID
      &\;\;\ra\;\;&&
      \cstate_\taua(\sqn_\ue^\ID) < \instate_{\tau''}(\sqn_\hn^\ID)
      \numberthis\label{eq:fdsfjaoifjasodaspoas}
    \end{alignat*}
    We give the outline of the proof. First, we apply \ref{sequ1} to $\tau''$. Then, we take $\tau''_0 = \_,\fnai(j_e) \popre \tau''$. We let $\tau''_1 = \_,\pnai(j_e,1)$ or $\_,\cnai(j_e,1)$ such that $\tau''_1 \popre \tau''_0$, and we do a case disjunction on $\tau''_1$:
    \begin{itemize}
    \item If $\tau''_1 = \_,\pnai(j_e,1)$, then we use \ref{sequ4} on it, and we show that $\cstate_\taua(\sqn_\ue^\ID) < \instate_{\tau''}(\sqn_\hn^\ID)$ by doing a case disjunction on $\incaccept_{\tau''_1}^\ID$.
    \item If $\tau''_1 = \_,\cnai(j_e,1)$, then we use \ref{sequ2} on it, and we show that $\cstate_\taua(\sqn_\ue^\ID) < \instate_{\tau''}(\sqn_\hn^\ID)$ using \ref{b2}
    \end{itemize}
    We omit the details.
    
    Using \ref{b5} we know that $\instate_{\tau''}(\sqn_\hn^\ID) \le \instate_{\tau'}(\sqn_\hn^\ID)$ and $\instate_\taua(\sqn_\ue^\ID) \le \cstate_\taua(\sqn_\ue^\ID)$. Hence, we deduce from \eqref{eq:fdsfjaoifjasodaspoas} that:
    \begin{alignat*}{2}
      \accept_{\tau''}^\ID
      &\;\;\ra\;\;&&
      \instate_\taua(\sqn_\ue^\ID) < \instate_{\tau'}(\sqn_\hn^\ID)
      \numberthis\label{eq:fsdsdsdfjkvzd}
    \end{alignat*}

    Moreover, by definition of $\incaccept_{\tau'}^\ID$, and using the fact that $\lambda_\taua^{\tau'} \ra \accept_{\tau'}^\ID$ we know that:
    \[
      \left(
        \lambda_\taua^{\tau'} \wedge \incaccept_{\tau'}^\ID
      \right)
      \;\ra\;
      \instate_{\tau'}(\sqn_\hn^\ID) \le \instate_\taua(\sqn_\ue^\ID)
      \numberthis\label{eq:sdfjkvzd}
    \]
    Putting  \eqref{eq:fsdsdsdfjkvzd} and \eqref{eq:sdfjkvzd} together:
    \[
      \left(
        \lambda_\taua^{\tau'} \wedge \incaccept_{\tau'}^\ID \wedge \accept_{\tau''}^\ID
      \right)
      \;\ra\;
      \left(
        \instate_{\tau'}(\sqn_\hn^\ID) \le \instate_\taua(\sqn_\ue^\ID)
        \wedge
        \instate_\taua(\sqn_\ue^\ID) < \instate_{\tau'}(\sqn_\hn^\ID)
      \right)
      \;\ra\;
      \false
    \]
    Hence:
    \[
      \left(
        \lambda_\taua^{\tau'} \wedge \accept_{\tau''}^\ID
      \right)
      \;\ra\;
      \neg \incaccept_{\tau'}^\ID
      \numberthis\label{eq:vcjcsdfjkvzd}
    \]
  \end{itemize}
  From \eqref{eq:nkvjva}, \eqref{eq:sdnjksdf} and \eqref{eq:vcjcsdfjkvzd} we deduce that:
  \begin{alignat*}{2}
    \sff \wedge \lambda_\taua^{\tau'}
    &\;\;\ra\;\;&&
    \bigvee_{\tau'' = \_,\fuai_\ID(j'')\atop{\taua \potau \tau'' \potau \tautt}}
    \sff \wedge \lambda_\taua^{\tau'}\wedge\accept_{\tau''}^\ID\\
    &\;\;\ra\;\;&&
    \bigvee_{\tau'' = \_,\fuai_\ID(j'')\atop{\taua \potau \tau'' \potau \tautt}}
    \neg \incaccept_{\tau'}^\ID\\\intertext{Hence:}
    \sff \wedge \lambda_\taua^{\tau'}
    &\;\;\ra\;\;&&
    \neg \incaccept_{\tau'}^\ID
    \qquad\qquad \text{ when $\tau_\newsession \potau \taua \potau \tautt$}
    \numberthis\label{eq:pqwqokdp}
  \end{alignat*}
\end{itemize}

\paragraph{Part 2}
Using \eqref{eq:rutwdfsdfvsjov} and \eqref{eq:pqwqokdp}, we know that we can focus on the (partial) $\supi$ sessions that started after $\taui$, i.e. the sessions with transcript of the from $\lambda_\taua^{\tau'}$, where $\taua = \_,\npuai{1}{\ID}{j_a}$, $\tau' = \_,\pnai(j',1)$ and $\taui \potau \taua \potau \tau'$. Formally, we have:
\begin{alignat*}{2}
  \left(\sff \wedge \theta_\pnai\right)
  &\;\;\lra\;\;&&
  \bigwedge_{\tau' = \_,\pnai(\_,1) \atop{\taut \potau \tau'}}
  \neg\incaccept_{\tau'}^\ID\\
  &\;\;\lra\;\;&&
  \bigwedge_{\tau' = \_,\pnai(\_,1)\atop{\taut \potau \tau'}}
  \left(
    \left(\sff \wedge \accept_{\tau'}^\ID\right)
    \ra
    \neg\incaccept_{\tau'}^\ID
  \right)\\\displaybreak[1]
  &\;\;\lra\;\;&&
  \bigwedge_{\tau' = \_,\pnai(j',1)
    \atop{\tau'_0 = \_,\pnai(j',0)
      \atop{\taua = \_,\npuai{1}{\ID}{j_a}
        \atop{\taut \potau \tau'
          \atop{\tau'_0 \potau \taua \potau \tau'}}}}}
  \left(
    \left(\sff \wedge \lambda_\taua^{\tau'}\right)
    \ra
    \neg\incaccept_{\tau'}^\ID
  \right)
  \tag{By \eqref{eq:avcxvschiaa}}\\\displaybreak[1]
  &\;\;\lra\;\;&&
  \bigwedge_{
    \taua = \_,\npuai{1}{\ID}{j_a}
    \atop{\tau' = \_,\pnai(j',1)
      \atop{\taui \potau \taua \potau \tau'}}}
  \left(
    \left(
      \sff \wedge \lambda_\taua^{\tau'}
    \right)
    \;\ra\;
    \neg\incaccept_{\tau'}^\ID
  \right)
  \tag{By \eqref{eq:rutwdfsdfvsjov} and \eqref{eq:pqwqokdp}}
\end{alignat*}
We represent graphically the shape of the interleavings that we need to consider:
\begin{center}
  \begin{tikzpicture}
    [dn/.style={inner sep=0.2em,fill=black,shape=circle},
    sdn/.style={inner sep=0.15em,fill=white,draw,solid,shape=circle},
    sl/.style={decorate,decoration={snake,amplitude=1.6}},
    dl/.style={dashed},
    pin distance=0.5em,
    every pin edge/.style={thin}]

    \draw[thick] (0,0)
    node[left=1.3em] {$\tau:$}
    -- ++(0.5,0)
    node[dn,pin={above,align=left:{$\newsession_\ID(\_)$\\or $\epsilon$}}] {}
    node[below=0.3em]{$\tau_\newsession$}
    -- ++(2,0)
    node[dn,pin={above:{$\cuai_\ID(j_i,0)$}}] {}
    node[below=0.3em,name=b0]{$\tautt$}
    -- ++(2,0)
    node[dn,pin={above:{$\cnai(j,0)$}}] {}
    node[below=0.3em,name=b1]{$\taut$}
    -- ++(2,0)
    node[dn,pin={above:{$\cuai_\ID(j_i,1)$}}] {}
    node[below=0.3em,name=b9]{$\taui$}
    -- ++(2,0)
    node[dn,pin={above:{$\npuai{1}{\ID}{j_a}$}}] {}
    node[below=0.3em,name=a0]{$\taua$}
    -- ++(2,0)
    node[dn,pin={above:{$\pnai(j',1)$}}] {}
    node[below=0.3em,name=a1]{$\tau'$}
    -- ++(2,0)
    node[dn,pin={above:{$ \cnai(j,1) $}}] {}
    node[below=0.3em,name=b2]{$\tau$};

    \draw[thin,dashed] (a0) -- ++(0,-0.5) -| (a1);
    \draw[thin,densely dotted] (b0) -- ++(0,-0.8) -| (b1);
    \path (b1) -- ++(0,-0.8) node[draw=none,inner sep=0] (b1p){};
    \draw[thin,densely dotted] (b1p) -| (b9);
    \path (b9) -- ++(0,-0.8) node[draw=none,inner sep=0] (b9p){};
    \draw[thin,densely dotted] (b9p) -| (b2);
  \end{tikzpicture}
\end{center}

\paragraph{Part 3}
We are now going to show that if at least one partial $\supi$ session that started after $\taui$ accepts (i.e. $\sff \wedge \lambda_\taua^{\tau'}$ holds), then we have $\instate_\tau(\tsuccess_\hn^\ID) \ne \nonce^j$. First, from what we showed in \textbf{Part 2}, and using \eqref{eq:vioiqperqfa} we know that:
\begin{alignat*}{2}
  \neg \left(\sff \wedge \instate_\tau(\tsuccess_\hn^\ID) = \nonce^j\right)
  &\;\;\lra\;\;&&
  \bigvee_{
    \taua = \_,\npuai{1}{\ID}{j_a}
    \atop{\tau' = \_,\pnai(j',1)
      \atop{\taui \potau \taua \potau \tau'}}}
  \sff \wedge \lambda_\taua^{\tau'} \wedge \incaccept_{\tau'}^\ID \\
  &\;\;\ra\;\;&&
  \bigvee_{
    \taua = \_,\npuai{1}{\ID}{j_a}
    \atop{\tau' = \_,\pnai(j',1)
      \atop{\taui \potau \taua \potau \tau'}}}
  \sff \wedge \lambda_\taua^{\tau'}
\end{alignat*}
In a first time, assume that for every $\taua = \_,\npuai{1}{\ID}{j_a}$ and $\tau' = \_,\pnai(j',1)$ such that $\taui \potau \taua \potau \tau'$ we have:
\[
  \left(
    \sff \wedge \lambda_\taua^{\tau'} \wedge \neg \incaccept_{\tau'}^\ID
  \right)
  \;\ra\;
  \instate_\tau(\tsuccess_\hn^\ID) \ne \nonce^j
  \numberthis\label{eq:dffovsnsiofqwerefascs}
\]
Then we know that:
\[
  \bigvee_{
    \taua = \_,\npuai{1}{\ID}{j_a}
    \atop{\tau' = \_,\pnai(j',1)
      \atop{\taui \potau \taua \potau \tau'}}}
  \sff \wedge \lambda_\taua^{\tau'}
  \;\ra\;
  \neg \left(\sff \wedge \instate_\tau(\tsuccess_\hn^\ID) = \nonce^j\right)
\]
Therefore:
\[
  \neg \left(\sff \wedge \instate_\tau(\tsuccess_\hn^\ID) = \nonce^j\right)
  \;\lra\;
  \bigvee_{
    \taua = \_,\npuai{1}{\ID}{j_a}
    \atop{\tau' = \_,\pnai(j',1)
      \atop{\taui \potau \taua \potau \tau'}}}
  \sff \wedge \lambda_\taua^{\tau'}
  \numberthis\label{eq:dffovsnsiofqwere}
\]
We now give the proof of \eqref{eq:dffovsnsiofqwerefascs}. Let $\taua = \_,\npuai{1}{\ID}{j_a}$ and $\tau' = \_,\pnai(j',1)$ such that $\taui \potau \taua \potau \tau'$. We know that:
\[
  \left(
    \lambda_\taua^{\tau'} \wedge \neg \incaccept_{\tau'}^\ID
  \right)
  \;\ra\;
  \instate_{\taua}(\sqn_\ue^\ID) < \instate_{\tau'}(\sqn_\hn^\ID)
\]
And that:
\[
  \sff \;\ra\;
  \instate_{\taut}(\sqn_\hn^\ID) = \instate_{\taui}(\sqn_\ue^\ID)
\]
Moreover by \ref{b5} we know that $\instate_{\taui}(\sqn_\ue^\ID) \le \instate_{\taua}(\sqn_\ue^\ID)$. We summarize this graphically:
\begin{center}
  \begin{tikzpicture}
    [dn/.style={inner sep=0.2em,fill=black,shape=circle},
    sdn/.style={inner sep=0.15em,fill=white,draw,solid,shape=circle},
    sl/.style={decorate,decoration={snake,amplitude=1.6}},
    dl/.style={dashed},
    pin distance=0.5em,
    every pin edge/.style={thin}]

    \draw[thick] (0,0)
    node[left=1.3em] {$\tau:$}
    -- ++(0.5,0)
    node[below=0.3em] (a) {$\tautt$}
    node[dn,pin={above:{$\cuai_\ID(j_i,0)$}}] {}
    -- ++(3,0)
    node[below=0.3em] (b) {$\taut$}
    node[dn,pin={above:{$\cnai(j,0)$}}] {}
    -- ++(3,0)
    node[below=0.3em] (c) {$\taui$}
    node[dn,pin={above:{$\cuai_\ID(j_i,1)$}}] {}
    -- ++(3,0)
    node[below=0.3em] (d) {$\taua$}
    node[dn,pin={above:{$\npuai{1}{\ID}{j_a}$}}] {}
    -- ++(3,0)
    node[below=0.3em] (e) {$\tau'$}
    node[dn,pin={above:{$\pnai(j',1)$}}] {}
    -- ++(3,0)
    node[below=0.3em] (f) {$\tau$}
    node[dn,pin={above:{$\cnai(j,1)$}}] {};

    \path (c) -- ++ (0,-1.6)
    node (c1) {$\instate_\taui(\sqn_\ue^\ID)$}
    (d) -- ++ (0,-1.6)
    node (d1) {$\instate_\taua(\sqn_\ue^\ID)$}
    (b) -- ++ (0,-2.6)
    node (b1) {$\instate_\taut(\sqn_\hn^\ID)$}
    (e) -- ++ (0,-2.6)
    node (e1) {$\instate_{\tau'}(\sqn_\hn^\ID)$};

    \draw (b1) -- (c1) node[midway,above,sloped]{$=$}
    (c1) -- (d1) node[midway,above,sloped]{$\le$}
    (d1) -- (e1) node[midway,above,sloped]{$<$};

    \draw[thin,dashed] (a) -- ++(0,-0.8) -| (b)
    {[draw=none] -- ++(0,-0.8)} -| (c)
    {[draw=none] -- ++(0,-0.8)} -| (f);
    \draw[thin,densely dotted] (d) -- ++(0,-0.5) -| (e);
  \end{tikzpicture}
\end{center}
We deduce that:
\[
  \left(
    \sff \wedge \lambda_\taua^{\tau'} \wedge \neg \incaccept_{\tau'}^\ID
  \right)
  \;\ra\;
  \instate_{\taut}(\sqn_\hn^\ID) < \instate_{\tau'}(\sqn_\hn^\ID)
\]
Moreover:
\[
  \instate_{\taut}(\sqn_\hn^\ID) < \instate_{\tau'}(\sqn_\hn^\ID)
  \;\ra\;
  \Big(
  \bigvee_{\taux = \pnai(j_x,1)\atop{\taut \potau \taux \potau \tau'}}
  \incaccept_{\taux}^\ID
  \Big)
  \vee
  \Big(
  \bigvee_{\taux = \cnai(j_x,1)\atop{\taut \potau \taux \potau \tau'}}
  \incaccept_{\taux}^\ID
  \Big)
\]
For every $\taux = \pnai(j_x,1)$ such that  $\taut \potau \taux \potau \tau'$ we have $j_x \ne j$. Therefore:
\begin{alignat*}{2}
  \bigvee_{\taux = \pnai(j_x,1)\atop{\taut \potau \taux \potau \tau'}}
  \incaccept_{\taux}^\ID
  &\;\;\ra\;\;&&
  \bigvee_{\taux = \pnai(j_x,1)\atop{\taut \potau \taux \potau \tau'}}
  \cstate_\taux(\tsuccess_\hn^\ID) = \nonce^{j_x}\\
  &\;\;\ra\;\;&&
  \instate_\tau(\tsuccess_\hn^\ID) \ne \nonce^j
\end{alignat*}
And:
\begin{alignat*}{2}
  \bigvee_{\taux = \cnai(j_x,1)\atop{\taut \potau \taux \potau \tau'}}
  \incaccept_{\taux}^\ID
  &\;\;\ra\;\;&&
  \bigvee_{\taux = \cnai(j_x,1)\atop{\taut \potau \taux \potau \tau'}}
  \instate_\taux(\tsuccess_\hn^\ID) = \nonce^{j_x}\\
  &\;\;\ra\;\;&&
  \instate_\tau(\tsuccess_\hn^\ID) \ne \nonce^j
\end{alignat*}
This concludes the proof of \eqref{eq:dffovsnsiofqwerefascs}.

The proofs in \textbf{Part 1} to \textbf{3} only used the fact that $\tau$ is a valid symbolic trace. We never used the fact that $\tau$ is a basic trace. Therefore, carrying out the same proof, we can show that:
\[
  \neg \left(
    \usff \wedge \instate_\utau(\tsuccess_\hn^{\nu_\tau(\ID)}) = \nonce^j
  \right)
  \;\lra\;
  \bigvee_{
    \taua = \_,\npuai{1}{\ID}{j_a}
    \atop{\tau' = \_,\pnai(j',1)
      \atop{\taui \potau \taua \potau \tau'}}}
  \usff \wedge \lambda_\utaua^{\utau'}
  \numberthis\label{eq:dffovsnsfafsiofqwere}
\]

\paragraph{Part 4}
Let $\taua = \_,\npuai{1}{\ID}{j_a}$ and $\tau' = \_,\pnai(j',1)$ be such that $\taui \potau \taua \potau \tau'$. Observing that:
\begin{mathpar}
  \left(\nonce^{j'},\nonce^{j'}\right)
  \in \reveal_{\tauo}

  \left(
    \enc{\spair
      {\ID}
      {\instate_{\taua}(\sqn_\ue^\ID)}}
    {\pk_\hn}{\enonce^{j_a}},
    \enc{\spair
      {{\nu_\tau(\ID)}}
      {\instate_{\utaua}(\sqn_\ue^{\nu_\tau(\ID)})}}
    {\pk_\hn}{\enonce^{j_a}}
  \right)
  \in \reveal_{\tauo}

  \left(
    \mac{\spair
      {\enc{\spair
          {\ID}
          {\instate_{\taua}(\sqn_\ue^\ID)}}
        {\pk_\hn}{\enonce^{j_a}}}
      {g(\inframe_{\taua})}}
    {\mkey^\ID}{1},
    \mac{\spair
      {\enc{\spair
          {{\nu_\tau(\ID)}}
          {\instate_{\utaua}(\sqn_\ue^{\nu_\tau(\ID)})}}
        {\pk_\hn}{\enonce^{j_a}}}
      {g(\inframe_{\utaua})}}
    {\mkey^{\nu_\tau(\ID)}}{1}
  \right)
  \in \reveal_{\tauo}
\end{mathpar}
It is straightforward to show that we have a derivation of:
\[
  \begin{gathered}[c]
    \infer[\simp]{
      \inframe_\tau,
      \lreveal_{\tauo},
      \lambda_\taua^{\tau'}
      \;\sim\;
      \inframe_\utau,
      \rreveal_{\tauo},
      \lambda_\utaua^{\utau'}
    }{
      \inframe_\tau, \lreveal_{\tauo}
      \sim
      \inframe_\utau,\rreveal_{\tauo}
    }
  \end{gathered}
\]
Using \eqref{eq:dffovsnsiofqwere} and \eqref{eq:dffovsnsfafsiofqwere}, and combining the derivation above with the derivation in \eqref{eq:ascvjkxo}, we can build the following derivation:
\[
  \begin{gathered}[c]
    \infer[R]{
      \inframe_\tau,
      \lreveal_{\tauo},
      \sff \wedge \instate_\tau(\tsuccess_\hn^{\ID}) = \nonce^j
      \;\sim\;
      \inframe_\utau,
      \rreveal_{\tauo},
      \usff \wedge \instate_\utau(\tsuccess_\hn^{\nu_\tau(\ID)}) = \nonce^j
    }{
      \infer[(\dup,\fa)^*]{
        \inframe_\tau, \lreveal_{\tauo},
        \neg\left(
          \bigvee_{
            \taua = \_,\npuai{1}{\ID}{j_a}
            \atop{\tau' = \_,\pnai(j',1)
              \atop{\taui \potau \taua \potau \tau'}}}
          \sff \wedge \lambda_\taua^{\tau'}
        \right)
        \sim
        \inframe_\utau,\rreveal_{\tauo},
        \neg \left(
          \bigvee_{
            \taua = \_,\npuai{1}{\ID}{j_a}
            \atop{\tau' = \_,\pnai(j',1)
              \atop{\taui \potau \taua \potau \tau'}}}
          \usff \wedge \lambda_\utaua^{\utau'}
        \right)
      }{
        \inframe_\tau, \lreveal_{\tauo},
        \sim
        \inframe_\utau,\rreveal_{\tauo},
      }
    }
  \end{gathered}
  \numberthis\label{eq:vjkcqpweiqra}
\]
We know that:
\[
  \cond{\sff}{\syncdiff_\tau^\ID} =
  \begin{alignedat}{2}
    &\ite{\sff \wedge \instate_\tau(\sync_\ue^\ID)\wedge
      \instate_\tau(\tsuccess_\hn^\ID) = \nonce^j\\ &\quad}
    { \sqnsuc(\syncdiff_\tauo^\ID)\\ &\quad}
    { \syncdiff_\tauo^\ID}
  \end{alignedat}
\]
Similarly:
\[
  \cond{\usff}{\syncdiff_\utau^{\nu_\tau(\ID)}} =
  \begin{alignedat}{2}
    &\ite{\usff \wedge \instate_\utau(\sync_\ue^{\nu_\tau(\ID)})\wedge
      \instate_\utau(\tsuccess_\hn^{\nu_\tau(\ID)}) = \nonce^j\\ &\quad}
    { \sqnsuc(\syncdiff_\utauo^{\nu_\tau(\ID)})\\ &\quad}
    { \syncdiff_\utauo^{\nu_\tau(\ID)}}
  \end{alignedat}
\]
Hence, using \eqref{eq:vjkcqpweiqra} and the fact that:
\begin{mathpar}
  \left(
    \instate_\tau(\sync_\ue^\ID),
    \instate_\utau(\sync_\ue^{\nu_\tau(\ID)}
  \right) \in \reveal_{\tauo}

  \left(
    \syncdiff_\tauo^\ID,
    \syncdiff_\utauo^{\nu_\tau(\ID)}
  \right) \in \reveal_{\tauo}
\end{mathpar}
We have a derivation of:
\[
  \begin{gathered}[c]
    \infer[\simp]{
      \inframe_\tau,
      \lreveal_{\tauo},
      \cond{\sff}{\syncdiff_\tau^\ID}
      \;\sim\;
      \inframe_\utau,
      \rreveal_{\tauo},
      \cond{\usff}{\syncdiff_\utau^{\nu_\tau(\ID)}}
    }{
      \inframe_\tau, \lreveal_{\tauo}
      \sim
      \inframe_\utau,\rreveal_{\tauo}
    }
  \end{gathered}
  \numberthis\label{eq:dasdasdsadsdvjkcqpweiqra}
\]

\paragraph{Part 5}
Using \textbf{(J10)}, we know that:
\begin{alignat*}{2}
  \accept_\tau^\ID
  &\;\;\lra\;\;&&
  \bigvee_{
    \taui = \_,\cuai_\ID(j_i,1)
    \atop{\taut = \_,\cnai(j,0)
      \atop{\tautt = \_,\cuai_\ID(j_i,0)
        \atop{\tautt \potau \taut \potau \taui}}}}
  \left(\ftr{\tautt,\taui}{\taut,\tau}\right)\\
  \intertext{We split between the cases $\taui \potau \tau_\newsession$ and $\taui \not \potau \tau_\newsession$:}
  &\;\;\lra\;\;&&
  \bigvee_{
    \taui = \_,\cuai_\ID(j_i,1)
    \atop{\taut = \_,\cnai(j,0)
      \atop{\tautt = \_,\cuai_\ID(j_i,0)
        \atop{\mathclap{\tautt \potau \taut <_\tau
            \taui \potau \tau_\newsession}}}}}
  \left(\ftr{\tautt,\taui}{\taut,\tau}\right)
  \;\vee\;
  \bigvee_{
    \taui = \_,\cuai_\ID(j_i,1)
    \atop{\taut = \_,\cnai(j,0)
      \atop{\tautt = \_,\cuai_\ID(j_i,0)
        \atop{\mathclap{\tau_\newsession \potau \tautt <_\tau
            \taut \potau \taui}}}}}
  \left(\ftr{\tautt,\taui}{\taut,\tau}\right)
\end{alignat*}
If $\taui \potau \tau_\newsession$ then $\nu_\tautt(\ID) = \nu_\taui(\ID) \ne \nu_\tau(\ID)$, and if $\taui \not \potau \tau_\newsession$ then $\nu_\tautt(\ID) = \nu_\taui(\ID) = \nu_\tau(\ID)$. It follows, using \textbf{(J10)} on $\utau$, that:
\begin{mathpar}
  \bigvee_{\uID \in \copyid(\ID)\atop{\uID \ne \nu_\tau(\ID)}}
  \accept_\utau^\uID
  \;\lra\;
  \bigvee_{
    \taui = \_,\cuai_\uID(j_i,1)
    \atop{\taut = \_,\cnai(j,0)
      \atop{\tautt = \_,\cuai_\uID(j_i,0)
        \atop{\mathclap{\tautt \potau \taut <_\tau
            \taui \potau \tau_\newsession}}}}}
  \left(\ftr{\utautt,\utaui}{\utaut,\utau}\right)

  \accept_\utau^{\nu_\tau(\ID)}
  \;\lra\;
  \bigvee_{
    \taui = \_,\cuai_{\nu_\tau(\ID)}(j_i,1)
    \atop{\taut = \_,\cnai(j,0)
      \atop{\tautt = \_,\cuai_{\nu_\tau(\ID)}(j_i,0)
        \atop{\mathclap{\tau_\newsession <_\tau
            \tautt \potau \taut \potau \taui}}}}}
  \left(\ftr{\utautt,\utaui}{\utaut,\utau}\right)
\end{mathpar}

\noindent Hence, using \eqref{eq:vcofeojgsdpifjajn} if $\taui \potau \newsession_\ID$, and \eqref{eq:dasdasdsadsdvjkcqpweiqra} if $\taui \not\potau \newsession_\ID$, we can build the following derivation:
\[
  \begin{gathered}[c]
    \infer[\simp]{
      \inframe_\tau,
      \lreveal_{\tauo},
      \syncdiff_\tau^\ID
      \;\sim\;
      \inframe_\utau,
      \rreveal_{\tauo},
      \syncdiff_\utau^{\nu_\tau(\ID)}
    }{
      \inframe_\tau, \lreveal_{\tauo}
      \sim
      \inframe_\utau,\rreveal_{\tauo}
    }
  \end{gathered}
\]

\paragraph{Part 6}
Observe that:
\begin{mathpar}
  \neauth_\tau(\ID,j)
  \;\lra\;
  \accept_\tau^\ID

  \uneauth_\utau(\ID,j)
  \;\lra\;
  \bigvee_{\uID \in \copyid(\ID)}
  \accept_\utau^\uID
\end{mathpar}
We therefore easily obtain the derivation:
\[
  \begin{gathered}[c]
    \infer[]{
      \inframe_\tau,
      \lreveal_{\tauo},
      \neauth_\tau(\ID,j)
      \;\sim\;
      \inframe_\utau,
      \rreveal_{\tauo},
      \uneauth_\utau(\ID,j)
    }{
      \inframe_\tau, \lreveal_{\tauo}
      \sim
      \inframe_\utau,\rreveal_{\tauo}
    }
  \end{gathered}
\]
Finally, using \textbf{(J10)}, we know that:
\[
  \bigvee_{\ID \in \iddom} \accept_\tau^\ID
  \;\lra\;
  \bigvee_{\ID \in \baseiddom} \accept_\tau^\ID
  \;\lra\;
  \neauth_\tau^\ID
\]
Moreover:
\[
  \bigvee_{\ID \in \iddom} \accept_\utau^\ID
  \;\lra\;
  \bigvee_{\ID \in \baseiddom}
  \left(
    \bigvee_{\uID \in \copyid(\ID)}
    \accept_\utau^\uID
  \right)
  \;\lra\;
  \bigvee_{\ID \in \baseiddom}
  \uneauth_\utau(\ID,j)
\]
It follows that:
\[
  \begin{gathered}[c]
    \infer[\fa]{
      \inframe_\tau,
      \lreveal_{\tauo},
      t_\tau
      \;\sim\;
      \inframe_\utau,
      \rreveal_{\tauo},
      t_\utau
    }{
      \infer[R]{
        \inframe_\tau, \lreveal_{\tauo},
        \bigvee_{\ID \in \iddom}
        \accept_\tau^\ID
        \sim
        \inframe_\utau,\rreveal_{\tauo},
        \bigvee_{\ID \in \iddom}
        \accept_\utau^\ID
      }{
        \infer[\simp]{
          \inframe_\tau, \lreveal_{\tauo},
          \bigvee_{\ID \in \baseiddom}
          \neauth_\tau^\ID
          \sim
          \inframe_\utau,\rreveal_{\tauo},
          \bigvee_{\ID \in \baseiddom}
          \uneauth_\utau(\ID,j)
        }{
          \inframe_\tau, \lreveal_{\tauo}
          \sim
          \inframe_\utau,\rreveal_{\tauo}
        }
      }
    }
  \end{gathered}
\]
Which concludes this proof.
